\documentclass[10pt,oneside,b5paper,final]{memoir} 

\usepackage[export]{adjustbox}
\usepackage{ifthen}
\usepackage{accsupp}
\usepackage{scalerel}
\usepackage{mathtools}
\usepackage{caption}
\usepackage{float}
\usepackage{adjustbox}
\usepackage{amsbsy}
\usepackage{amscd}
\usepackage{amsfonts}
\usepackage{amsmath}
\usepackage{amssymb}
\usepackage{MnSymbol}
\usepackage{stmaryrd}
\SetSymbolFont{stmry}{bold}{U}{stmry}{m}{n}
\usepackage[english]{babel}
\usepackage{bussproofs}
\usepackage{color}
\usepackage{enumerate}
\usepackage{mathpartir}
\usepackage{wrapfig}
\usepackage{microtype}
\usepackage{tikz}
\usepackage{tkz-berge}
\usepackage{url}
\usepackage[utf8]{inputenc}
\usepackage{newunicodechar}
\usepackage{xspace}
\usepackage{marginnote}
\usepackage[section]{placeins}
\usepackage{textpos}
\usepackage{verbatimbox}
\definecolor{linkcolor}{rgb}{0.7,0.8,0.6}
\usepackage{xr-hyper,xr}
\usepackage{hyperref}
\hypersetup{linkbordercolor=linkcolor,urlbordercolor=linkcolor,unicode}
\usepackage{amsthm}
\usepackage{memhfixc}
\usepackage{cleveref}
\usepackage{subfig}

\newunicodechar{λ}{\m{\lambda}}
\newunicodechar{β}{\m{\beta}}
\newunicodechar{μ}{\m{\mu}}
\newunicodechar{∈}{\m{\in}}
\newunicodechar{α}{\m{\alpha}}
\newunicodechar{∉}{\m{\not\in}}
\newunicodechar{ℕ}{\m{\mathbb{N}}}
\newunicodechar{⇀}{\m{\rightharpoonup}}
\newunicodechar{∪}{\m{\cup}}
\newunicodechar{∘}{\m{\circ}}
\newunicodechar{≠}{\m{\neq}}
\newunicodechar{↦}{\m{\mapsto}}
\newunicodechar{∀}{\m{\forall}}
\newunicodechar{∃}{\m{\exists}}
\newunicodechar{⋅}{\m{\cdot}}
\newunicodechar{∧}{\m{\land}}
\newunicodechar{≤}{\m{\leq}}
\newunicodechar{≥}{\m{\geq}}
\newunicodechar{⇐}{\m{\Leftarrow}}
\newunicodechar{⇒}{\m{\Rightarrow}}
\newunicodechar{≡}{\m{\equiv}}
\newunicodechar{⇔}{\m{\Leftrightarrow}}
\newunicodechar{∨}{\m{\lor}}
\newunicodechar{∩}{\m{\cap}}
\newunicodechar{→}{\m{\rightarrow}}
\newunicodechar{←}{\m{\leftarrow}}

\let\oldlambda\lambda
\renewcommand\lambda{\m{\oldlambda}}
\let\oldalpha\alpha
\renewcommand\alpha{\m\oldalpha}
\let\oldbeta\beta
\renewcommand\beta{\m\oldbeta}
\let\oldmu\mu
\renewcommand\mu{\m{\oldmu}}
\let\oldomega\omega
\renewcommand\omega{\m\oldomega}
\let\oldtau\tau
\renewcommand\tau{\m\oldtau}
\let\oldSigma\Sigma
\renewcommand\Sigma{\m\oldSigma}
\let\oldDelta\Delta
\renewcommand\Delta{\m\oldDelta}
\let\oldbot\bot
\renewcommand\bot{\m\oldbot}
\let\oldiota\iota
\renewcommand\iota{\m\oldiota}
\let\oldtop\top
\renewcommand\top{\m\oldtop}
\let\oldpi\pi
\renewcommand\pi{\m\oldpi}
\let\paragraphsign\S
\makeatletter
\let\S\@undefined
\makeatother
\let\oldnot\not
\renewcommand\not[1]{\m{\oldnot{#1}}}
\let\oldsim\sim
\renewcommand\sim{\m{\oldsim}}
\let\oldvec\vec
\renewcommand\vec[1]{\m{\oldvec{#1}}}
\let\oldcheck\check
\renewcommand\check[1]{\m{\oldcheck{#1}}}
\let\oldCheck\Check
\renewcommand\Check[1]{\m{\oldCheck{#1}}}
\let\oldbar\bar
\renewcommand\bar[1]{\m{\oldbar{#1}}}
\let\oldwidehat\widehat
\renewcommand\widehat[1]{\m{\oldwidehat{#1}}}
\let\oldHat\Hat
\renewcommand\Hat[1]{\m{\oldHat{#1}}}
\let\oldBox\Box
\renewcommand\Box{\m{\oldBox}}
\let\olddots\dots
\renewcommand\dots{\m{\olddots}}
\let\oldvdots\vdots
\renewcommand\vdots{\m{\oldvdots}}
\let\oldddots\ddots
\renewcommand\ddots{\m{\oldddots}}
\let\oldmathcal\mathcal
\renewcommand\mathcal[1]{\m{\oldmathcal{#1}}}
\let\oldboldsymbol\boldsymbol
\renewcommand\boldsymbol[1]{\m{\oldboldsymbol{#1}}}

\let\originalleft\left
\let\originalright\right
\renewcommand{\left}{\mathopen{}\mathclose\bgroup\originalleft}
\renewcommand{\right}{\aftergroup\egroup\originalright}

\theoremstyle{definition}
\newtheorem{theorem}{Theorem}[section]
\newtheorem{lemma}[theorem]{Lemma}
\newtheorem{corollary}[theorem]{Corollary}
\newtheorem{proposition}[theorem]{Proposition}
\newtheorem{conjecture}[theorem]{Conjecture}
\newtheorem{para}[theorem]{\paragraphsign}

\newtheorem{example}[theorem]{Example}
\newtheorem{definition}[theorem]{Definition}
\newtheorem{notation}[theorem]{Notation}
\newtheorem{terminology}[theorem]{Terminology}

\newtheorem{remark}[theorem]{Remark}

\newcommand\sortname[1]{}

\newcommand\ifempty[3]{\ifthenelse{\equal{#1}{}}{#2}{#3}}

\DeclareRobustCommand\plaintext[2]{\BeginAccSupp{ActualText=#1}#2\EndAccSupp{}}
\DeclareRobustCommand\plainhext[2]{\BeginAccSupp{method=hex,unicode,ActualText=#1}#2\EndAccSupp{}}

\newcommand\m[1]{\ensuremath{#1}\texorpdfstring{\xspace}{}}
\newcommand\mrel[1]{\m{\mathrel{#1}}}

\newcommand\fun[1]{\m{\mathrm{#1}}} 
\newcommand\funsymbol[1]{\m{\mathsf{#1}}} 

\newcommand\newvar[2]{\newcommand{#1}{\m{#2}}}
\newcommand\renewvar[2]{\renewcommand{#1}{\m{#2}}}
\newcommand\deflabel[2]{\newcommand{#1}{\m{\mathrm{#2}}}}
\newcommand\newunary[2]{\newcommand{#1}[1]{\m{{#2}\protect\ifempty{##1}{}{\left(##1\right)}}}}
\newcommand\newunaryi[2]{\newcommand{#1}[2]{\m{{#2}_{##1}\protect\ifempty{##2}{}{\left(##2\right)}}}}
\newcommand\newunaryij[2]{\newcommand{#1}[3]{\m{{#2}_{##1,##2}\protect\ifempty{##3}{}{\left(##3\right)}}}}
\newcommand\newfunction[2]{\newunary{#1}{\fun{#2}}}
\newcommand\newfunctioni[2]{\newunaryi{#1}{\fun{#2}}}
\newcommand\newfunctionij[2]{\newunaryij{#1}{\fun{#2}}}
\newcommand\newfunsymbol[2]{\newcommand{#1}{\funsymbol{#2}}} 
\newcommand\newkeyword[2]{\newcommand{#1}{\m{\mathsf{#2}}}}
\newcommand\newsemantics[3]{\newcommand{#1}[1]{\m{#2\protect\ifempty{##1}{\cdot}{##1}#3}}}
\newcommand\newsemanticsi[3]{\newcommand{#1}[2]{\m{#2\protect\ifempty{##2}{\cdot}{##2}#3_{##1}}}}

\renewcommand\emph[1]{{\em #1}}

\newcommand\fig[1]{\includegraphics[scale=0.8]{figs/{{#1}}}}
\newcommand\graph[2]{\vcentered{\m{#1}:~~} \vcentered{\fig{#2}}\hfill}

\newcommand\vcentered[1]{\raisebox{-0.5\height}{#1}}
\newcommand\multilinebox[1]{%
		\begin{tabular}{@{}l@{}}%
			#1%
		\end{tabular}%
	}
\newcommand\eqAnnotation[1]{\text{\small #1}}
\newcommand\sideCondition[1]{\eqAnnotation{(#1)}}
\newcommand\mlSideCondition[1]{\text{\small\multilinebox{(#1)}}}

\newcommand\Tag[1]{\text{(#1)}}

\newcommand\hsep\hfill
\catcode`\&=\active
\newenvironment{hspread}{%
	\par\vspace{.4ex}\noindent\mbox{}\hfill%
		\newcommand&\hfill%
		\catcode`\&=\active%
		\let\oldnl\\%
		\renewcommand\\{\hfill\mbox{}\oldnl[3ex]\mbox{}\hfill}}
		{\hfill\mbox{}\vspace{1ex}\newline}
\catcode`\&=4

\newcommand\proofsystem[1]{\framebox{\begin{minipage}{\linewidth-0.7em}\vspace{0.3em}\begin{hspread}#1\end{hspread}\end{minipage}}}
\newenvironment{bprooftree}{\leavevmode\hbox\bgroup}{\DisplayProof\egroup}
\newcommand\axiom[1]{\AxiomC{\m{#1}}}
\newcommand\emptyAxiom{\axiom{\phantom{(I)}}}
\newcommand\unaryInf[1]{\UnaryInfC{\m{#1}}}
\newcommand\binaryInf[1]{\BinaryInfC{\m{#1}}}
\newcommand\ternaryInf[1]{\TrinaryInfC{\m{#1}}}

\newcommand\infLabel[1]{\RightLabel{\m{#1}}}
\newcommand\proofpar{\vspace{-1ex}\\\par}

\makeatletter
\def\namedlabel#1#2{\begingroup
    #2
    \def\@currentlabel{#2}
    \phantomsection\label{#1}\endgroup
}
\makeatother

\newcommand\commentOut[1]{}

\newcommand\graphsize[1]{\m{\length{#1}}}
\newcommand\termsize[1]{\m{\length{#1}}}
\renewcommand\implies{\m{\Rightarrow}}

\newcommand\tuple[1]{\m{\left\langle{#1}\right\rangle}}
\newcommand\pair[2]{\m{\tuple{#1, #2}}}
\newcommand\triple[3]{\m{\tuple{#1, #2, #3}}}
\newcommand\quadruple[4]{\m{\tuple{#1, #2, #3, #4}}}
\newfunction\domain{dom}
\newfunction\range{ran}
\newfunction\image{im}
\newcommand\defd[1]{\m{{#1}{\downarrow}}} 
\newfunction\freevars{FV}
\newfunction\bigO{O}
\newcommand\length[1]{\m{\left|{#1}\right|}}
\newfunction\vs{vs} 
\newfunction\Ter{Ter}
\newfunction\iTer{{Ter}^\infty}

\newcommand\states{\m{S}}
\newcommand\scoll{\m{\text{\small\textbar}\hspace{-0.73ex}\downarrow}} 
\newcommand\coll[1]{\m{{#1}\hspace{0.17ex}{\scoll}}} 

\newcommand\lambdacal{\boldsymbol{λ}}
\newcommand\lambdaletreccal{\m{\boldsymbol{λ}_\letrec}}
\newcommand\lambdabh{\m{λ_\bh}}
\newcommand\lambdabhcal{\m{\lambdacal_\bh}}

\newcommand\lambdaletrecprefixcal{\m{\smash{\boldsymbol{(\lambdaletrec)}}}}
\newcommand\lambdaprefixcal{\boldsymbol{(λ)}}
\newcommand\lambdaprefixposcal{\m{\boldsymbol{(λ)_\spos}}}

\newcommand\lambdaletrec{\m{λ_\letrec}}

\newcommand\adbmal{\m{\scalebox{-1}[1]{λ}}}

\let\iNoDot\i

\makeatletter
\let\a\@undefined
\let\i\@undefined
\let\j\@undefined
\let\o\@undefined
\let\r\@undefined
\let\v\@undefined
\let\u\@undefined
\let\w\@undefined
\makeatother

\newvar\A{A}
\newvar\B{B}
\renewvar\C{C}
\newvar\a{a}
\newvar\e{e}
\newvar\h{h} 
\newvar\i{i} 
\newvar\j{i} 
\newvar\f{f} 
\newvar\g{g} 
\newvar\n{n}
\newvar\o{o}
\newvar\p{p}
\newvar\q{q}
\newvar\r{r} 
\newvar\v{v} 
\newvar\w{w} 
\newvar\u{u} 
\newvar\x{x} 
\newvar\y{y} 
\renewvar\L{L} 
\newvar\M{M} 
\newvar\N{N} 
\renewvar\O{O} 
\renewvar\P{P} 
\renewvar\G{G} 
\newvar\V{V} 
\newvar\E{E} 

\newcommand\recprefix[1]{\m{\left({#1}\right)}} 
\newcommand\absprefix[1]{\m{\recprefix{\protect\ifempty{#1}{}{λ{#1}}}}} 
\newcommand\posprefix[3]{\m{\absprefix{#1}_{#2}^{#3}}} 

\newcommand\prefixsep{\;} 
\newcommand\prefixed[2]{\m{\absprefix{#1}\prefixsep{#2}}} 
\newcommand\recprefixed[2]{\m{\recprefix{#1}\prefixsep{#2}}} 
\newcommand\posprefixed[4]{\m{\posprefix{#1}{#2}{#3}\prefixsep{#4}}} 
\newcommand\annprefixed[3]{\m{\prefixed{#1}{#2\,:\,#3}}} 
\newcommand\labprefixed[3]{\m{#1\,:\,\prefixed{#2}{#3}}} 

\newfunction\arity{ar}

\newcommand\afunsym{\m{f}}

\newfunction\stopsymb{\m{\bullet}}

\newcommand\aconstname{\m{\mathsf{c}}}

\newcommand\CRSabs[2]{\m{[{#1}]\hspace*{1pt}{#2}}}
\newcommand\KahrsCRSabs[2]{\m{\{{#1}\}\hspace*{1pt}{#2}}}

\newcommand\arule{\m{\rho}}
\newcommand\brule{\m{\sigma}}

\newcommand\aCRS{\mathcal{C}}

\newcommand\acxthole{\m{\Box}}
\newcommand\sacxt{\m{C}}
\newcommand\acxt[1]{\m{\sacxt[#1]}}

\newfunsymbol\sabsCRS{abs}
\newcommand\absCRS[2]{\m{{\sabsCRS}\left({\CRSabs{#1}{#2}}\right)}}
\newfunsymbol\sappCRS{app}
\newcommand\appCRS[2]{\m{{\sappCRS}\left({#1,\,#2}\right)}}
\newcommand\appCRSbreak[2]{\m{\sappCRS\left(\begin{array}{l}{#1},\\\indent{#2}\end{array}\right)}}

\newcommand\sPrefixedCRS[1]{\m{\funsymbol{pre}_{#1}}}
\newcommand\prefixedCRS[3]{\m{\sPrefixedCRS{#1}\left(\ifempty{#2}{#3}{\CRSabs{#2}{#3}}\right)}}

\newcommand\unfCRS{\m{\textit{\textbf R}_{\hspace*{-0.5pt}\boldsymbol{\unfold}}}}
\newcommand\unfbhCRS{\m{\textit{\textbf R}_{\hspace*{-0.5pt}\boldsymbol{\unfoldbh}}}}

\newcommand\unfOne{\m{{\unfold_{\hspace{-0.8pt}\mathbf{1}}}}}
\newcommand\unfTwo{\m{{\unfold_{\hspace{-0.8pt}\mathbf{2}}}}}

\newcommand\sRegCRS{\m{\textit{\textbf{Reg}}}}
\newcommand\RegCRS{\m{\sRegCRS}}
\newcommand\stRegCRS{\m{\sRegCRS^{\boldsymbol{+}}}}
\newcommand\RegzeroCRS{\m{\sRegCRS^{\boldsymbol{-}}}}

\newcommand\RegposCRS{\m{\sRegCRS_{\textit{\textbf{pos}}}}}
\newcommand\stRegposCRS{\m{\sRegCRS_{\textit{\textbf{pos}}}^{\boldsymbol{+}}}}
\newcommand\RegposzeroCRS{\m{\sRegCRS_{\textit{\textbf{pos}}}^{\boldsymbol{-}}}}

\newcommand\RegARS{\m{\textit{Reg}}}
\newcommand\stRegARS{\m{\textit{Reg}^{+}}}
\newcommand\RegzeroARS{\m{\textit{Reg}^{-}}}

\newcommand\stRegARSlab{\m{\widehat{\stRegARS}}}

\newcommand\RegposARS{\m{\textit{Reg}_{\textit{pos}}}}
\newcommand\stRegposARS{\m{\textit{Reg}_{\textit{pos}}^+}}
\newcommand\RegposzeroARS{\m{\textit{Reg}_{\textit{pos}}^-}}

\newcommand\stParseCRS{\m{\textit{\textbf{Parse}}^{\boldsymbol{+}}}}
\newcommand\stParseUnfCRS{\m{\textit{\textbf{Parse}}_\unfold^{\boldsymbol{+}}}}
\newcommand\ParseCRS{\m{\textit{\textbf{Parse}}}}

\newcommand\RegletrecCRS{\m{\sRegCRS_\letrec}}
\newcommand\stRegletrecCRS{\m{\sRegCRS^{+}_\letrec}}

\newcommand\RegletrecARS{\m{\textit{Reg\/}_\letrec}}
\newcommand\stRegletrecARS{\m{\textit{Reg\/}^{+}_\letrec}}

\newcommand\wPrefix{\smash{(λ)}} 
\newcommand\twolines[2]{\m{\begin{array}{r}#1\\#2\end{array}}}

\newcommand\sig{\m{\Sigma}}

\newcommand\sigCRS{\m{\sig^λ}}
\newcommand\sigCRSbh{\m{\sig^λ_\bh}}
\newcommand\sigCRSletrec{\m{\sig^λ_\letrec}}
\newcommand\sigCRSPrefixed{\m{\sig^{\wPrefix}}}
\newcommand\sigCRSLetrecPrefixed{\m{\sig^{\wPrefix}_{\letrec}}}
\newcommand\sigCRSParse{\m{\sig^{\wPrefix}_{\sparse}}}

\newcommand\sigTG{\m{\sig^λ}}
\newcommand\sigTGi[1]{\m{\sig^λ_{\0_{#1}}}}
\newcommand\sigTGij[2]{\m{\sig^λ_{\0_{#1},\S_{#2}}}}
\newcommand\sigTGbh{\m{\sig^λ_\bh}}
\newcommand\sigTGSbh{\m{\sig^λ_{\S,\bh}}}

\newcommand\GeneratedSubARS[1]{\succsoford{#1}{\mred}}
\newcommand\GeneratedSubARSi[2]{\succsoford{#2}{\mred_{#1}}}

\newcommand\steps{\m{\Phi}}

\newfunction\src{src}
\newfunctioni\srci{src}
\newfunction\tgt{tgt}
\newfunctioni\tgti{tgt}

\newcommand\aARS{\mathcal{A}}
\newcommand\astep{\m{\phi}}

\newcommand\wDepth[2]{\m{{#1}_{#2}}}

\newcommand\soutgoingst[1]{\m{{#1}_{\text{out}}}}
\newcommand\sincomingst[1]{\m{{#1}_{\text{in}}}}
\newcommand\outgoingst[2]{\m{\soutgoingst{#1}\left({#2}\right)}}
\newcommand\incomingst[2]{\m{\sincomingst{#1}\left({#2}\right)}}

\newcommand\alabelling{\m{\mathbb{L}}}
\newcommand\aARSbisim{\mrel{\mathbb{B}}}

\newcommand\apath{\m{\pi}}

\newfunctioni\ST{ST}

\newcommand\arewseq{\m{\tau}}
\newcommand\brewseq{\m{\xi}}
\newcommand\crewseq{\m{\pi}}

\newcommand\red{\mrel{\to}}
\newcommand\mred{\mrel{\twoheadrightarrow}}
\newcommand\convred{\mrel{\leftarrow}}
\newcommand\convmred{\mrel{\twoheadleftarrow}}
\newcommand\parred{\mrel{\hspace{1ex}||\hspace{-2.8ex}\longrightarrow}}
\newcommand\thsp{-1.46ex}
\newcommand\threeheadrightarrow{\m{\plainhext{219200BB}{\twoheadrightarrow\hspace*\thsp→}}}
\newcommand\threeheadleftarrow{\m{\plainhext{00AB2190}{←\hspace*\thsp\twoheadleftarrow}}}
\newcommand\infred{\mrel{\threeheadrightarrow}}
\newcommand\convinfred{\mrel{\threeheadleftarrow}}
\newcommand\convomeganfred{\m{\convred^{!\,\omega}}}
\newcommand\eqred{\m{\red^=}}

\newcommand\morestepred{\m{\red^+}}
\newcommand\omegared{\m{\red^{\omega}}}
\newcommand\omeganfred{\m{\red^{!\,\omega}}}

\newcommand\nfred{\m{\red^{\scriptstyle !}}}
\newcommand\infnfred{\m{\infred^{!}}}

\newcommand\rewritingSequence[1]{\begin{alignat*}{2}#1\end{alignat*}}

\newcommand\rednull{& &~}
\newcommand\redcell[1]{& #1 {} &~}
\newcommand\termcell[1]{ & #1 \\}
\newcommand\step[2]{\redcell{#1} \termcell{#2}}
\newcommand\init[1]{\rednull \termcell{#1}}

\newcommand\unfold{\m{\triangledown}}
\newcommand\unfoldbh{\m{{{\triangledown\hspace{-.35ex}\raisebox{-0.6\height}{\bh}}\hspace{-0.75ex}}}}
\newcommand\decompose{\m{\Delta}} 

\newcommand\unfoldpre[1]{\m{\unfold\hspace{-.1em}.#1}}

\newcommand\rulep[1]{\m{\plainhext{03F1}{\varrho}^{#1}}}
\newcommand\rulebp[2]{\m{\plainhext{03F1}{\varrho}_{#1}^{#2}}}

\newcommand\unfrule[1]{\rulebp{\unfold}{#1}}
\newcommand\decrule[1]{\rulebp{\decompose}{#1}}
\newcommand\rulepos[1]{\m{\varrho_{pos}^{#1}}}
\newcommand\ruleref[1]{\m{(#1)}} 

\deflabel\reg{reg}
\deflabel\streg{reg^+}
\deflabel\regzero{reg^-}
\deflabel\nil{nil}
\deflabel\rec{rec}
\deflabel\merg{letrec}
\deflabel\tighten{tighten}
\deflabel\reduce{red}
\deflabel\del{del}
\deflabel\sscope{sc}
\deflabel\sparse{parse^+} 
\deflabel\sparseunf{parse^+_\unfold} 

\newcommand\extscope{scope\m{^+}}
\newcommand\extscopes{scopes\m{^+}}

\deflabel\eager{eag}
\deflabel\lazy{lazy}
\deflabel\fullybacklinked{fbl}

\newcommand\astrat{\m{\mathbb{S}}}
\newcommand\astratplus{\m{\mathbb{S}^+}}
\newcommand\eagStrat{\m{\astrat_{\eager}}}
\newcommand\lazyStrat{\m{\astrat_{\lazy}}}
\newcommand\eagStratPlus{\m{\astratplus_{\eager}}}
\newcommand\lazyStratPlus{\m{\astratplus_{\lazy}}}

\newcommand\ainst{\m{\iota}}
\newcommand\binst{\m{\kappa}}

\newcommand\Reg{\m{\mathbf{Reg}}}
\newcommand\stReg{\m{\mathbf{Reg}^{\boldsymbol{+}}}}
\newcommand\stRegzero{\m{\mathbf{Reg}_{\boldsymbol{0}}^{\boldsymbol{+}}}}

\newcommand\annstRegzero{\m{\mathbf{ann-Reg}_{\boldsymbol{0}}^{\boldsymbol{+}}}}

\newcommand\stRegeq{\m{\mathbf{Reg}^{\boldsymbol{+}}_{\boldsymbol{=}}}}

\newcommand\stReglab{\m{\widehat{\mathbf{Reg}^{\boldsymbol{+}}}}}

\newcommand\Regletrec{\m{\mathbf{Reg}_{\letrec}}}
\newcommand\stRegletrec{\m{\mathbf{Reg}^{\boldsymbol{+}}_{\letrec}}}
\newcommand\ugRegletrec{\m{\mathbf{ug-Reg}_{\letrec}}}
\newcommand\ugstRegletrec{\m{\mathbf{ug-Reg}^{\boldsymbol{+}}_{\letrec}}}
\newcommand\annstRegletrec{\m{\mathbf{ann-Reg}^{\boldsymbol{+}}_{\letrec}}}

\newcommand\siCRSPretermAlpha[1]{\m{\mathbf{A}^{\hspace{-.2em}\boldsymbol{\infty}}_{#1}}}
\newcommand\iCRSPretermAlpha[2]{\m{\mathbf{A}^{\hspace{-.2em}\boldsymbol{\infty}}_{#1}\left(#2\right)}}
\newcommand\Schroer{\mathcal{S}}
\newcommand\Kahrs{\mathcal{K}}

\newcommand\EqTer{\m{{\mathbf{EQ}}^{\boldsymbol{\infty}}}}
\newcommand\AlphaPreTer{\m{\mathbf{EQ}_{\boldsymbol{α}}^{\boldsymbol{\infty}}}}

\deflabel\FIX{FIX}
\deflabel\FIXletrec{\FIX_\letrec}
\deflabel\FIXletrecmin{\FIX^-_\letrec}

\newcommand\presup[2]{\m{\prescript{#1}{}{#2}}}

\newcommand\derivablein[2]{\m{\vdash_{#1}{#2}}}
\newcommand\sinfderivable{\presup{\infty}{\vdash}}
\newcommand\infderivablein[2]{\m{\sinfderivable_{#1}\hspace*{1.5pt}{#2}}}

\newcommand\ahotg{\mathcal{G}}

\newcommand\altg{\m{G}}

\newcommand\Deriv{\mathcal{D}}
\newcommand\Derivlab{\m{{\Deriv}^{\smash{(\mathrm{lb})}}}}
\newcommand\Derivann{\m{\hat{\Deriv}}}

\newfunsymbol\0{0}
\newfunsymbol\S{S}
\newcommand\indir{\m{\text{\small\textbar}}}
\newcommand\nodef{\m{\text{?}}}

\newcommand\parsei[1]{\m{\funsymbol{parse}^\fun{+}_{#1}}} 
\newcommand\parse[2]{\m{\parsei{#1}\left(#2\right)}} 

\newkeyword\Let{let}
\newkeyword\In{in}
\newkeyword\letrec{letrec}

\newcommand\letin[2]{\m{\Let\;{#1}\;\In\;{#2}}}
\newcommand\eqsep{,\,} 
\newcommand\abs[2]{\m{λ{#1}\protect\ifempty{#2}{}{.\,{#2}}}}
\newcommand\simpleabs[2]{\m{λ{#1}{.\,{#2}}}}
\newcommand\app[2]{\m{{#1}\;{#2}}}
\newcommand\appbreak[2]{\m{\begin{array}{l}\left(#1\right)\\~~\left(#2\right)\end{array}}}

\newcommand\infDeriv{\m{\Deriv^{\infty}}}

\newcommand\depth[1]{\m{\left|{#1}\right|}}

\newcommand\amarker{\m{u}}
\newcommand\bmarker{\m{v}}
\newcommand\cmarker{\m{w}}

\newcommand\mcdots{\m{\hspace*{1pt}\cdots\hspace*{2pt}}}

\newcommand\alignbreak{\m{\displaybreak[0]\\}}

\newcommand\aproj{\m{\pi}}

\newcommand\eqcl[2]{\m{{[{#1}]}_{#2}}}
\newcommand\eqclin[3]{\m{{[{#1}]}_{#2}^{#3}}}

\newcommand\srestrictto[2]{\m{{#1}\!\mid_{\,#2}}}

\newfunsymbol\LetCRS{let}
\newcommand\InCRS[1]{\m{\In_{#1}}}
\newcommand\letCRS[3]{\m{\LetCRS\left(\ifempty{#2}{\InCRS{#1}\left(#3\right)}{\CRSabs{#2}{\InCRS{#1}\left(#3\right)}}\right)}}
\newcommand\converse[1]{\m{{#1}^{\smile}}}
\newcommand\emptyword{\m{\epsilon}}
\newcommand\prefixcon{\m{\,}} 

\newcommand\fs[1]{\set{#1}}
\newcommand\starfs[1]{\m{\star\fs{#1}}}
\newcommand\stdprefix{\m{x_0^{v_0}\fs{\vec{f}_0} \dots x_n^{v_n}\fs{\vec{f}_n}}}

\newfunctioni\idon{id} 

\newcommand\pcup{\mrel{\vec{\cup}}}

\newfunction\lb{lb}

\newcommand\setcompr[2]{\m{\left\{{#1}~\middle|~{#2}\right\}}}
\DeclareMathOperator*{\comma}{\scalerel*{,}{\sum}}
\newcommand\listcompr[2]{\m{\comma_{#1}{#2}}}

\newcommand\arel{\m{R}}
\newcommand\brel{\m{S}}

\newcommand\sequence[2]{\m{\{{#1}\}_{#2}}}

\newcommand\enumsequence[1]{\m{\langle{#1}\rangle}}
\newcommand\slub{\m{\bigsqcup}}
\newcommand\sglb{\m{\bigsqcap}}
\newcommand\lub[1]{\m{\slub {#1}}}
\newcommand\glb[1]{\m{\sglb {#1}}}

\newcommand\subst[3]{\m{{#1}[{#2}:={#3}]}}

\newcommand\vecsub[2]{\m{\vec{#1}_{#2}}}
\newcommand\nmvec[3]{\m{{#1}_{#2} \dots {#1}_{#3}}}
\newcommand\nvec[2]{\m{\nmvec{#1}{1}{#2}}}
\newcommand\vecOneToN[1]{\m{\nvec{#1}{n}}}

\newcommand\sep[1]{\m{&#1&}}

\newcommand\Lbracket{\plainhext{27E6}{\llbracket}}
\newcommand\Rbracket{\plainhext{27E7}{\rrbracket}}
\newsemantics\unfsem{\Lbracket}{\Rbracket_{\hspace*{-0.1pt}{\lambdacal}}}
\newsemantics\unfbhsem{\Lbracket}{\Rbracket_{\hspace*{-0.1pt}{\lambdabhcal}}}
\newsemantics\graphsem{\Lbracket}{\Rbracket}

\newsemanticsi\graphsemC{\Lbracket}{\Rbracket}
\deflabel\noSsh{min}
\newsemanticsi\graphsemCmin{\Lbracket}{\Rbracket^\noSsh}

\newcommand\rulestranslambdaletreccaltolhotgs{\mathcal{R}}
\newcommand\rulestranslambdaletreccaltoltgs{\m{\mathcal{R}_{\S}}}

\newcommand\alphaequiv{\m{\equiv_{α}}}

\newcommand\SilentLTS[2]{\aLTS_{{#1},{#2}}}

\newcommand\sltg{\m{\mathcal{G}}} 
\newcommand\ltg[2]{\m{\sltg_{#1}\left(#2\right)}}
\newcommand\SilentLTG[3]{\m{\sltg_{{#1},{#2}}\left({#3}\right)}}

\newcommand\aLTS{\mathcal{L}}

\newcommand\atg{\m{G}}
\newcommand\alab{\m{a}}

\newcommand\binds{\mrel{\leftspoon}}
\newcommand\bindseq{\m{\mrel\leftspoon^{=}}}
\newcommand\captures{\mrel{\dashleftarrow}}
\newcommand\capturedby{\mrel{\dashrightarrow}}

\deflabel\spos{pos}

\newfunction\positions{Pos}

\newcommand\vecpositions{\m{\vec{ℕ}^*}}

\newcommand\set[1]{\m{\left\{{#1}\right\}}}
\newcommand\rootpos{\m{\epsilon}}

\newcommand\apreter{\m{s}}
\newcommand\bpreter{\m{t}}

\newcommand\succsoford[2]{\m{\left(#1 #2\right)}}

\newcommand\succsofordin[3]{\m{\left({#1}{#2}\right)^{#3}}}

\newfunction\bp{bp} 

\newfunction\readback{\mathit{rb}}

\newfunction\vertsof{V}
\newfunctioni\vertsiof{V}
\newfunction\lab{lab}
\newfunctioni\labi{lab}
\newfunction\args{args}
\newfunctioni\argsi{args}

\newcommand\bh{\m{\bullet}}

\newcommand\tgsucc{\mrel{\rightarrowtail}}
\newcommand\tgsuccstar{\m{\mrel{\rightarrowtail}^*}}

\newcommand\tgsuccis[2]{\m{\mrel{\presup{\smash{#2}}{\tgsucc}}_{#1}}}
\newcommand\tgsuccisstar[2]{\m{\left(\tgsuccis{#1}{#2}\right)^*}}

\newcommand\pathto{\m{\rightarrowtail^*}}

\newcommand\iso{\mrel{\sim}}

\newcommand\bisim{\mrel{
	\setbox0=\hbox{\kern-.1ex{\m{\leftrightarrow}}\kern-.1ex}
	\setbox1=\vbox{\hbox{\raise .1ex \box0}\hrule}
	\hbox{\kern.1ex\box1\kern.1ex}
}}

\newcommand\invfunbisim{\mrel{
	\setbox0=\hbox{\kern-.1ex{\m{\leftarrow}}\kern-.1ex}
	\setbox1=\vbox{\hbox{\raise .1ex \box0}\hrule}
	\hbox{\kern.1ex\box1\kern.1ex}
}}

\newcommand\funbisim{\mrel{
	\setbox0=\hbox{\kern-.1ex{\m{\rightarrow}}\kern-.1ex}
	\setbox1=\vbox{\hbox{\raise .1ex \box0}\hrule}
	\hbox{\kern.1ex\box1\kern.1ex}
}}

\newcommand\convfunbisim{\mrel{
	\setbox0=\hbox{\kern-.1ex{\m{\leftarrow}}\kern-.1ex}
	\setbox1=\vbox{\hbox{\raise .1ex \box0}\hrule}
	\hbox{\kern.1ex\box1\kern.1ex}
}}

\newcommand\scollC[1]{\m{{{\text{\small\textbar}}\hspace{-0.73ex}\downarrow}_{#1}}}
\newcommand\collC[2]{\m{{\scollC{#1}}\left({#2}\right)}}

\newcommand\abisim{\m{R}} 

\newfunction\scope{Sc}
\newfunction\scopemin{{Sc}^{-}}
\newfunctioni\scopei{Sc}

\newfunctioni\scopeof{scope}
\newfunctioni\extscopeof{scope^+\hspace{-1ex}}

\newfunction\binders{bds}

\newfunction\abspre{P} 
\newfunctioni\absprei{P}

\newcommand\aaphotg{\mathcal{G}}
\newfunction\aphotg{\aaphotg}
\newcommand\aaphotgiso{\boldsymbol{\mathcal{G}}}
\newcommand\atgiso{\m{\textbf{\textit{G}}}}

\newcommand\classlhotgs{\mathcal{H}}

\newcommand\wBPrefix{\smash{\boldsymbol{(λ)}}}

\newcommand\classltgs{\mathcal{T}} 
\newcommand\classltgsi[1]{\m{\mathcal{T}_{#1}^{λ}}}
\newcommand\classltgsij[2]{\m{\mathcal{T}_{#1,#2}^{λ}}}

\newcommand\classlhotgsi[1]{\m{\mathcal{H}^{λ}_{#1}}}
\newcommand\classaphotgsi[1]{\m{\mathcal{H}_{#1}^{\wPrefix}}}
\newcommand\classlhotgsisoi[1]{\m{{\boldsymbol{\mathcal{H}}}^{\boldsymbol{λ}}_{#1}}}
\newcommand\classaphotgsisoi[1]{\m{{\boldsymbol{\mathcal{H}}}_{{#1}}^{\wBPrefix}}}
\newcommand\classltgsisoij[2]{\m{{\boldsymbol{\mathcal{T}}}_{{#1},{#2}}^{\wBPrefix}\!}}

\newcommand\eag[1]{\presup{\smash{\eager}}{#1}}
\newcommand\fbl[1]{\presup{\smash{\fullybacklinked}}{#1}}

\newcommand\classtgssiglambdai[1]{\m{{\mathcal{T}}_{#1}}} 
\newcommand\classtgssiglambdaij[2]{\m{{\mathcal{T}}_{{#1},{#2}}}} 

\newfunctioni\lhotgstoaphotgsi{A}
\newfunctioni\aphotgstolhotgsi{B}
\newcommand\slhotgsisotoaphotgsisoi[1]{\m{{\textit{\textbf{A}}}_{#1}}}
\newcommand\saphotgsisotolhotgsisoi[1]{\m{{\textit{\textbf{B}}}_{#1}}}

\newfunction\lhotgstoltgs{\mathcal{HT}}
\newfunction\ltgstolhotgs{\mathcal{TH}}

\newfunctionij\aphotgstoltgsij{G} 
\newfunctionij\ltgstoaphotgsij{\mathcal{G}} 

\newfunctionij\aphotgsisotoltgsisoij{\textbf{G}}
\newfunctionij\ltgsisotoaphotgsisoij{\boldsymbol{\mathcal{G}}}

\newcommand\apre{\m{p}}
\newcommand\bpre{\m{q}}
\newcommand\cpre{\m{r}}
\newcommand\dpre{\m{s}}

\newcommand\aclass{\mathcal{K}}

\newfunction\tgsminover{TG^{-}}
\newfunction\tgsover{TG}

\newenvironment{linespacing}[1]{%
		\par%
		\edef\oldlineskip{\the\lineskip}%
		\edef\oldlineskiplimit{\the\lineskiplimit}%
		\setlength{\global\lineskip}{#1}%
		\setlength{\global\lineskiplimit}{#1}%
	}{%
		\par%
		\setlength{\global\lineskip}{\oldlineskip}%
		\setlength{\global\lineskiplimit}{\oldlineskiplimit}%
	}

\usetikzlibrary{calc,matrix,shapes,positioning,decorations.markings,arrows}

\tikzset{row sep=9mm, column sep=9mm, >=stealth}

\makeatletter \g@addto@macro\@floatboxreset\centering \makeatother

\tikzstyle{->>>} =
  [shorten >= 0.3mm,
   decoration =
     {markings,
      mark=at position -2.4mm with {\arrow{>}},
      mark=at position -1.2mm with {\arrow{>}},
      mark=at position 1 with {\arrow{>}}},
   postaction={decorate}]

\pgfdeclaredecoration{funbisim}{final}{
  \state{final}[width=\pgfdecoratedpathlength]{
    \draw[->]
      (0,\pgfdecorationsegmentamplitude+0.1mm) -- +(\pgfdecoratedpathlength,0);
    \draw[-]
      (0,-\pgfdecorationsegmentamplitude-0.1mm) -- +(\pgfdecoratedpathlength,0);}}
\tikzset{
  funbisim/.style={
    decoration={funbisim, amplitude=0.25ex},
    decorate,
    funbisim options/.style={#1}    
  }}

\pgfdeclaredecoration{bisim}{final}{
  \state{final}[width=\pgfdecoratedpathlength]{
    \draw[<->]
      (0,\pgfdecorationsegmentamplitude+0.1mm) -- +(\pgfdecoratedpathlength,0);
    \draw[-]
      (0,-\pgfdecorationsegmentamplitude-0.1mm) -- +(\pgfdecoratedpathlength,0);}}
\tikzset{
  bisim/.style={
    decoration={bisim, amplitude=0.25ex},
    decorate,
    bisim options/.style={#1}    
  }}

\makeatletter
\pgfdeclareshape{transbox}{
  \inheritsavedanchors[from=rectangle]
  \inheritanchor[from=rectangle]{center}
  \inheritanchor[from=rectangle]{north}
  \inheritanchor[from=rectangle]{south}
  \inheritanchor[from=rectangle]{west}
  \inheritanchor[from=rectangle]{east}
  \inheritbehindbackgroundpath[from=rectangle]
  \inheritbackgroundpath[from=rectangle]
  \inheritbeforebackgroundpath[from=rectangle]
  \inheritbehindforegroundpath[from=rectangle]
  \inheritforegroundpath[from=rectangle]
  \inheritbeforeforegroundpath[from=rectangle]
}
\makeatother

\newcommand\transpicture[1]{\vcentered{\begin{tikzpicture}[node distance=4mm] #1 \end{tikzpicture}}}

\newcommand\emptynode[2][]{\node[inner sep=1pt, #1](#2){}}

\newcommand\ltgnode[3][]{\node[#1,draw,shape=circle,inner sep=0.1mm,minimum size=5mm](#2){\large #3}}

\newcommand{\addPrefix}[3][]{\node[node distance=1mm,#1,left=of #2.north]{\m{\scriptstyle(#3)}}}
\newcommand{\addPos}[2]{\node[node distance=1.2mm,right=of #1.north]{\m{\scriptstyle #2}}}
\newcommand{\addPrefixswphantom}[3][]{\node[node distance=1mm,#1,right=of #2.south]{\m{\scriptstyle\phantom{(#3)}}}}

\newcommand{\translation}[3]{\translationif{#1}{#2}{#3}{}{}}
\newcommand{\rbedge}[4][]{\Edge[style={->,thin,#1},label={#2}](#3)(#4)}

\newcommand{\transname}[1]{\noindent\vcentered{#1:}}
\newcommand\transspace{\hspace{2ex}}
\newcommand{\translationif}[5]{\transname{#1}\transspace\transpicture{#2}\transspace\vcentered{$\overset{#4}{\underset{#5}{\implies}}$}\transspace\transpicture{#3}}
\newcommand\translationv[3]{\transname{#1}\m{\begin{array}{c}\transpicture{#2}\\[6ex]\m{\Downarrow}\\[3ex]\transpicture{#3}\end{array}}}
\newcommand\translationvif[4]{\transname{#1}\m{\begin{array}{c}\transpicture{#2}\\[6ex]\m{\Downarrow~~\sideCondition{#3}}\\[3ex]\transpicture{#4}\end{array}}}

\newenvironment{transenv}{\begin{linespacing}{5ex}\noindent\mbox{}\hfill}{\hfill\mbox{}\end{linespacing}}
\newcommand\transinit[1]{\transpicture{#1}}
\newcommand\transstep[2]{\hfill\mbox{} \mbox{}\nobreak\hfill\vcentered{\m{\implies_{#1}}}\nobreak\hfill\transpicture{#2}}

\newcommand\fusionDownarrow{\vcentered{\m{~\Downarrow}}}

\newcommand\fusionstepsimple[2]{\transpicture{#1}\vcentered{\m{\Longleftarrow}}\transpicture{#2}}

\newcommand\fusionstep[3]{\m{\begin{array}{c}\begin{tikzpicture}[node distance=4mm]{#1}\end{tikzpicture}\\[1ex]\fusionDownarrow\\[3ex]\begin{tikzpicture}[node distance=4mm]{#2}\end{tikzpicture}\end{array}\vcentered{\m{\Longleftarrow}}\transpicture{#3}}}

\newcommand\fusionstepclap[3]{\m{\begin{array}{c}\begin{tikzpicture}[node distance=4mm]{#1}\end{tikzpicture}\\[1ex]\fusionDownarrow\\[3ex]\begin{tikzpicture}[node distance=4mm]{#2}\end{tikzpicture}\end{array}\vcentered{\m{\mathclap{\Longleftarrow}}}\transpicture{#3}}}

\newcommand\fusiontwostep[4]{\m{\begin{array}{c}\begin{tikzpicture}[node distance=4mm]{#1}\end{tikzpicture}\\[1ex]\fusionDownarrow\\[3ex]\begin{tikzpicture}[node distance=4mm]{#2}\end{tikzpicture}\\[1ex]\fusionDownarrow\\[3ex]\begin{tikzpicture}[node distance=4mm]{#3}\end{tikzpicture}\end{array}\vcentered{\m{\Longleftarrow}}\transpicture{#4}}}

\newcommand\fusionthreestep[5]{\m{\begin{array}{c}\begin{tikzpicture}[node distance=4mm]{#1}\end{tikzpicture}\\[1ex]\fusionDownarrow\\[3ex]\begin{tikzpicture}[node distance=4mm]{#2}\end{tikzpicture}\\[1ex]\fusionDownarrow\\[3ex]\begin{tikzpicture}[node distance=4mm]{#3}\end{tikzpicture}\\[0ex]\fusionDownarrow\\[3ex]\begin{tikzpicture}[node distance=4mm]{#4}\end{tikzpicture}\end{array}\transspace\vcentered{\m{\Longleftarrow}}\transspace\transpicture{#5}}}

\crefformat{equation}{(#2#1#3)}
\crefname{enumi}{}{}
\creflabelformat{enumi}{(#2#1#3)}
\crefname{property}{property}{properties}
\creflabelformat{property}{(#2{P}#1#3)}
\crefname{Rb}{}{}
\creflabelformat{Rb}{(#2{Rb-}#1#3)}

\makeatletter
\newcommand\fs@myfs{\def\@fs@cfont{\bfseries}\let\@fs@capt\floatc@ruled
	\def\@fs@pre{}%
	\def\@fs@post{\kern2pt\hrule height.8pt depth0pt \relax}%
	\def\@fs@mid{\vspace{0.5\abovecaptionskip}\hrule\vspace{0.4\abovecaptionskip}\relax}%
	\let\@fs@iftopcapt\iffalse}
\makeatother
\floatstyle{myfs}
\restylefloat{figure}

\makeatletter
\let\xx@thm\@thm
\AtBeginDocument{\let\@thm\xx@thm}
\makeatother

\nouppercaseheads

\begin{document}

\copypagestyle{chapter}{plain}
\makeoddfoot{plain}
  {}
  {}
  {}

\frontmatter

{
\setlength{\parindent}{0cm}
\setlength{\parskip}{6ex}
\thispagestyle{empty}

\pretitle{\begin{center}Dissertation:\\\LARGE}
\title{\textsc{Unfolding Semantics of the Untyped\\λ-Calculus with \letrec}}
\author{Jan Rochel \texttt{<\href{mailto:jan@rochel.info}{jan@rochel.info}>}}
\date{defended on June 20, 2016\\online version, revision: \today}

\maketitle
\thispagestyle{empty}

\vfill
\begin{center}
\begin{minipage}{9cm}
Dedication: \textit{Gewidmet meiner ehemaligen Klavierlehrerin Eva-Maria Rieckert, die mir Musik auf eine Weise vermittelte, die mich bis heute über das Musizieren hinaus prägt.}
\end{minipage}
\end{center}




\newpage
\thispagestyle{empty}

\begin{center}
\begin{minipage}{10.9cm}
\begin{tabular}{lcl}
	Promotoren: & & Copromotor: \\\\
	Prof.dr.\ S.\ D.\ Swierstra& ~~~~~~~~~~~~~~~~~~~~~~~~~~~ & Dr.\ C.\ Grabmayer \\
   Prof.dr.\ V.\ van Oostrom
\end{tabular}
\end{minipage}

\textbf{Unfolding Semantics of the Untyped λ-Calculus with \letrec}
\\[3ex]
\textbf{Ontvouwingssemantiek van de ongetypeerde\\λ-calculus met \letrec}
\\(met een samenvatting in het Nederlands)

Proefschrift

\begin{minipage}{10cm}
ter verkrijging van de graad van doctor aan de Universiteit Utrecht op gezag van de rector magnificus, prof.dr.\ G.J.\ van der Zwaan, ingevolge het besluit van het college voor promoties in het openbaar te verdedigen op maandag
20 juni 2016 des middags te 2.30 uur
\end{minipage}

door

\textbf{Jan Rochel}

geboren op 20 juli 1984\\
te Bretten, Duitsland

\vfill

\noindent
\begin{minipage}{10.9cm}
Dit proefschrift werd (mede) mogelijk gemaakt met financiële steun van de Nederlandse Organisatie voor Wetenschappelijk Onderzoek (NWO).
\end{minipage}
\end{center}
}

\newpage

\tableofcontents

\chapter{Preface}

This thesis documents research which was carried out as part of an NWO\footnote{Nederlandse Organisatie voor Wetenschappelijk Onderzoek} research project under my promotors Vincent van Oostrom and Doaitse Swierstra titled `Realising Optimal Sharing'. The objective of the project was to investigate whether the theory of `optimal evaluation of the λ-calculus' could be used in practise to increase the execution efficiency for programs written in functional programming languages. As regards the original project goal my research was unsuccessful. This is not due to a lack of results but due to a distraction: when approaching the subject of `optimal sharing' -- which is `dynamic' in the sense that it concerns sharing that an evaluator maintains at run time -- more fundamental, still unresolved questions emerged concerning `static' sharing -- which is the sharing inherent in a given program definition. Therefore the presented results are not about optimal and therefore `dynamic' sharing but solely about `static' sharing.


The research was carried out in close collaboration with Clemens Grabmayer. All of the fundamental ideas and results are due to joint efforts and fruitful -- if sometimes fierce -- debate. We complemented each other splendidly: I profited much from Clemens Grabmayer's expertise with formal systems, while I myself could contribute my proficiency with functional programming languages and compiler construction; also I have worked out an implementation of our methods. The substance of this thesis is from three research papers \cite{grab:roch:expressibility,grab:roch:representations,grab:roch:maxsharing} published in the context of my doctoral research with Grabmayer and Rochel as authors, presented here in a more coherent narrative with supplementary explanations and examples. In order to do justice to the Clemens Grabmayer's substantial contribution, authors who wish to cite this thesis in their work are kindly advised to cite at least one of the papers alongside.

This document is structured according to these three papers: \cref{chap:introduction} introduces the \lambdaletrec-formalism and our rewriting system for unfolding \lambdaletrec-terms. There we also pose the problems that are resolved in the following three chapters, \cref{chap:expressibility}, \cref{chap:representations}, and \cref{chap:maxsharing}, each of which corresponds to one of these papers. Note, that the formalisms in this thesis deviate to varying degrees from their original form in the papers. The changes were required for the formalims to be consistent throughout the chapters.

\section{Abstract}
	In this thesis we investigate the relationship between finite terms in \lambdaletreccal, the λ-calculus with letrec, and the infinite λ-terms they express. We say that a \lambdaletrec-term \emph{expresses} a λ-term if the latter can be obtained as an infinite unfolding of the former. Unfolding is the process of substituting occurrences of function variables by the right-hand side of their definition. We consider the following questions:
\begin{enumerate}[(i)]
	\item How can we characterise those infinite λ-terms that are \lambdaletrec-expressible?
	\item given two \lambdaletrec-terms, how can we determine whether they have the same unfolding?
	\item given a \lambdaletrec-term, can we find a more compact version of the term with the same unfolding?
\end{enumerate}
To tackle these questions we introduce and study the following formalisms:
\begin{itemize}
	\item a rewriting system for unfolding \lambdaletrec-terms into λ-terms
	\item a rewriting system for `observing' λ-terms by dissecting their term structure
	\item higher-order and first-order graph formalisms together with translations between them as well as translations from and to \lambdaletreccal
\end{itemize}
We identify a first-order term graph formalism on which bisimulation preserves and reflects the unfolding semantics of \lambdaletreccal and which is closed under functional bisimulation. From this we derive efficient methods to determine whether two terms are equivalent under infinite unfolding and to compute the maximally shared form of a given \lambdaletrec-term.

\mainmatter

\setcounter{chapter}{-1}


\chapter{\lambdaletrec{} and Unfolding}\label{chap:introduction}

\begin{para}[abstract]
	This thesis concerns itself with terms in the λ-calculus with \letrec, specifically with \emph{unfolding} these terms. Unfolding refers to the process of substituting occurrences of \Let-bound variables by their definition. We define unfolding by means of a rewriting system. We study the properties of that rewriting system and build various formal systems on top of it to derive further results. These results include:
	\begin{description}
		\item[\cref{chap:expressibility}:] a characterisation of the infinite λ-terms that can be expressed finitely as \lambdaletrec-terms.
		\item[\cref{chap:representations}:] a graph representation for \lambdaletrec-expressible λ-terms.
		\item[\cref{chap:maxsharing}:] practical and efficient methods for transforming a \lambdaletrec-term into a maximally compact form; and deciding whether two \lambdaletrec-terms have the same unfolding.
	\end{description}
\end{para}

\begin{para}[required background]
	The reader is expected to have some proficiency with functional programming languages based on the λ-calculus \cite{chur:thec41,barendregt1984lambda} (preferably Haskell), which is most likely required to understand the presented results and their relevance. Furthermore, the formal systems used for reasoning in this thesis are mostly rewriting systems. Therefore, at least a basic background in term rewriting \cite{terese:2003} is assumed.
\end{para}

\begin{para}[chapter overview]
	This chapter gives an overview over \lambdaletrec and outlines the perspective from which we will study this calculus. We provide definitions and basic properties of the rewriting system by which we unfold \lambdaletrec-terms. Then we will pose the questions and problems to contextualise the results presented in the following chapters.
\end{para}

\section{Introduction}


\begin{para}[\lambdaletrec as an abstraction of functional programming languages]\label{lambdaletrec-as-an-abstraction}
	The λ-calculus is a formal system in computer science and logic for expressing computation. It is the model of computation at the core of functional programming languages. In this thesis we will look at one specific instance of the λ-calculus, namely the untyped λ-calculus extended by the \letrec-construct, or in short \lambdaletrec. In that form, it serves well as a minimalistic abstraction of functional programming languages like Haskell. While Haskell is a typed language, it is typically translated into a simplified form during the compilation process in which type information is discarded (\emph{type erasure}). Thus types can be regarded as auxiliary means for the programmer, and can be neglected when looking at the evaluation semantics of a type-checked program.
\end{para}

\begin{para}[the \letrec-construct]\label{the-letrec-construct}
	The \letrec-construct serves a number of purposes. By allowing us to bind subterms to variables it adds to the simple, untyped λ-calculus the means for:
	\begin{description}
		\item[modularisation:] Subterms with a specific purpose can be given a descriptive name and more easily be treated as entities of their own.
		\item[sharing:] Instead of repeating identical subterms at different locations, a subterm can be defined once and be referenced by its function definition multiple times.
		\item[cyclicity:] Function definitions can contain references to themselves, allowing for cyclic (and mutually cyclic) bindings.
	\end{description}
	While all the above can also be achieved by a collection of top-level bindings, the \letrec has one distinguishing characteristic, which is that function bindings can be defined at any position in the term. The scope of thus defined function bindings does not range over the entire program but can only be used underneath the position of their definition. This locality of definitions allows for a more structured approach to programming. Also, it can refer to λ-variables bound outside of the binding which results in fewer β-reduction steps during evaluation.
\end{para}

Before we concern ourselves with formal definitions let us first fix some notation and terminology.

\begin{notation}[\Let = \letrec]
	Throughout this thesis we write \Let to denote the \letrec-construct, as is done in Haskell.
\end{notation}

\begin{terminology}[\Let-expression]
	A term that starts with a \Let, thus a term that has the form \letin{B}{L}, is called a \emph{\Let-expression}.
\end{terminology}

\begin{terminology}[binding group]
	We call the collection of bindings \m{B} defined in a \Let-expression \letin{B}{L} a \emph{binding group}.
\end{terminology}

\begin{terminology}[function binding, function variable, \Let-bound variable]
	We call the equations of a \Let-expression \emph{function bindings} or simply \emph{bindings}. We call the variable on the left-hand side of a binding a \emph{\Let-bound variable}, or a \emph{function variable}.\footnote{While a term bound by \Let to a variable may well be constant (i.e.\ not a λ-abstraction) we still call such a binding a function binding and the \Let-bound variable a function variable.}
\end{terminology}

\begin{terminology}[body]
	The part \m{L} of the expression \letin{B}{L} is called the \emph{body} of the \Let-expression.
\end{terminology}

\begin{example}\label{ex:let-terminology}
	Consider the \lambdaletrec-term \letin{\fun{fix} = \abs{f}{\app{f}{(\app{\fun{fix}}{f})}}}{\fun{fix}}. The binding group of the \Let-expression consists of a single function binding \m{\fun{fix} = \abs{f}{\app{f}{(\app{\fun{fix}}{f})}}} which binds the term \abs{f}{\app{f}{(\app{\fun{fix}}{f})}} to the function variable \fun{fix}. The body of the \Let-expression consists of an occurrence of the function variable \fun{fix}.
\end{example}

\begin{remark}[Turing completeness, well-typedness, termination]
	Typed λ-calculi (with finite types) are strongly normalising, which means that every computation in such a calculus terminates. Therefore such calculi are not Turing complete. However, strong normalisation does not hold for typed λ-calculi that allow for cyclic definitions, such as \lambdaletrec. Therefore the \letrec can also be seen as a way to restore Turing completeness for a typed λ-calculus. Other approaches would be fixed-point combinators or top-level bindings. The \letrec offers the most convenience from a programmer's point of view. 
\end{remark}

\begin{para}[infinite unfolding]
	A \lambdaletrec-term \L can be seen as a finite representation of a (possibly) infinite λ-term \M, which we obtain by repeatedly substituting every occurrence of a function variable by the right-hand side of the corresponding binding. We call \M the \emph{infinite unfolding} of \L, and we write \m{\unfsem{\L} = \M}. A definition for \unfsem{} is given later (\cref{def:unf-mapping}). 
\end{para}

\begin{example}[infinite unfolding of \app{\fun{fix}}{\fun{id}}]\label{ex:unf_fix_id}
	Let us consider a naive implementation of the \fun{fix}-function applied to \abs{x}{x}, the identity function:
	\begin{gather*}
		\unfsem{\app{(\letin{\fun{fix} = \abs{f}{\app{f}{(\app{\fun{fix}}{f})}}}{\fun{fix}})}{(\abs{x}{x})}}\\ = \app{(\abs{f}{\app{f}{(\app{(\abs{f}{\app{f}{(\app{\dots}{f})}})}{f})}})}{(\abs{x}{x})}
	\end{gather*}
\end{example}

\begin{notation}[ellipsis: \dots]
	Note that (as in the example above) we will be using the ellipsis as an informal notation for `and so on' extensively. It occurs both in infinite terms and infinite rewriting sequences. Instead of providing a first-order formalisation of the ellipsis, we trust that the reader will find the context always sufficient to infer the shape of the entire term or rewriting sequence.
\end{notation}

\begin{para}[evaluation and unfolding]
	While semantics of \lambdaletrec-based programming languages can be defined via infinite unfolding, evaluation of programs on a computer cannot operate on infinite terms but must rely on a finite representation. The evaluation of such programs requires in addition to β-reduction a mechanism to unfold \Let-expressions. This is best modelled in a rewriting system that extends the λ-calculus by unfolding rules.
\end{para}

\begin{example}[leftmost-outermost evaluation of \app{\fun{fix}}{\fun{id}}]\label{ex:eval_fix_id}
	Let us evaluate a small example program to see how unfolding comes into play in the course of evaluating \lambdaletrec-terms:
	\[\app{(\letin{\fun{fix} = \abs{f}{\app{f}{(\app{\fun{fix}}{f})}}}{\fun{fix}})}{(\abs{x}{x})}\]
	This term has no (visible) β-redex as the \abs{f}{\dots}-abstraction is `blocked' by the surrounding \Let. In order to turn it into a proper β-redex we need to unfold its definition, which means essentially substituting the right-hand side of the binding for both occurrences of the function variable \fun{fix}:
	\[\mred_\unfold \quad \app{(\abs{f}{\app{f}{(\app{\letin{\fun{fix} = \abs{f}{\app{f}{(\app{\fun{fix}}{f})}}}{\fun{fix}})}{f}}})}{(\abs{x}{x})}\]
	In the formalisation of unfolding as a rewriting system (see \cref{sec:unfolding_informal}) this requires actually a number of steps, hence the many-step reduction \m{\mred_\unfold}. Now that we have a visible β-redex to contract, we can substitute \abs{x}{x} for \f:
	\[\red_β \quad \app{(\abs{x}{x})}{(\app{(\letin{\fun{fix} = \abs{f}{\app{f}{(\app{\fun{fix}}{f})}}}{\fun{fix}})}{(\abs{x}{x})})}\]
	Next, we apply the identity function and we arrive at the initial term and evaluation continues as above and will thus never terminate.
	\[\red_β \quad \app{(\letin{\fun{fix} = \abs{f}{\app{f}{(\app{\fun{fix}}{f})}}}{\fun{fix}})}{(\abs{x}{x})} \; \mred_\unfold \; \dots\]
\end{example}

\begin{remark}[\unfold]\label{rem:triangle}
	Using the triangle as a symbol for unfolding is inspired by graph rewriting systems like Lambdascope \cite{oost:looi:zwit:2004} where sharing is indicated by an explicit sharing node in the shape of a triangle, with multiple incoming edges at the top (shared occurrences) and one outgoing edge at the bottom (to the shared subgraph). An unsharing step would then `unzip' the shared subgraph node by node, with the triangle acting as the zipper foot.
\end{remark}

\begin{para}[mixing unfolding and β-reduction]
	A simple interpreter for \lambdaletrec proceeds along the lines of \cref{ex:eval_fix_id}, i.e.\ by interspersing β-reduction with unfolding steps in a combined rewriting system. Most of the scientific works around unfolding \lambdaletrec-terms are works that study evalutors that include unfolding rules with β-reduction (and possibly α-reduction) in a rewriting system.
\end{para}

\begin{para}[unfolding as a subject of study]
	This thesis however focusses entirely on the unfolding portion of the semantics of \lambdaletrec; β-reduction will henceforth play a marginal role at most.\footnote{This vaguely suggests possible future research: it could be a promising venture to develop modular semantics for \lambdaletreccal, i.e.\ `β-reduction modulo unfolding', and express existing evaluators in terms of this semantis in a modular manner.}
\end{para}

\begin{para}[outlook]
	In the following two sections we will define the term language for \lambdaletrec-terms as well as a rewriting system for unfolding terms in \lambdaletrec. First we give an informal account to introduce the notation that we will be using for examples. Afterwards we provide sound formalisations.
\end{para}

\section{λ-terms and \lambdaletrec-terms -- informal}\label{sec:terms_informal}

\begin{para}[overview]
	This section provides first-order notations for terms in the λ-calculus and the \lambdaletrec-calculus. Mind, that these are not formal definitions and are only used to convey an intuition for the issue at hand. Only in \cref{sec:terms_formal} the definitions are formalised in the CRS (Combinatory Reduction System) framework.
\end{para}

\begin{para}[set of λ-terms]
	Let \V be a set of variable names. The set of λ-terms is `coinductively' defined by the following grammar, where \m{x ∈ V}:
	\begin{equation*}
		\begin{array}{lllll}
			\Tag{term} & \L \sep{::=} \abs{x}{\L} & \Tag{abstraction} \\
			           &         \sep{ | } \app{\L}{\L} & \Tag{application} \\
			           &         \sep{ | } x                  & \Tag{variable} \\
		\end{array}
	\end{equation*}
\end{para}

\begin{para}[infinite terms]
	Mind, that we interpret this grammar coinductively (i.e.\ as a final coalgebra). Therefore finite as well as infinite terms arise from it. Since unfoldings of \lambdaletrec-terms are typically infinite, in this thesis we will more often than not deal with infinite λ-terms. \lambdaletrec-terms on the other hand will always be finite.
\end{para}

\begin{example}[a simple infinite λ-term]
	See \cref{ex:unf_fix_id}.
\end{example}

Adding a production for the \letrec-construct to the above grammar, we obtain a grammar for \lambdaletrec.

\begin{para}[set of \lambdaletrec-terms]
	Let \V be a set of variable names. The set of \lambdaletrec-terms is inductively defined by the following grammar, where \m{x, f_1,\dots,f_n ∈ V}:
	\begin{equation*}
		\begin{array}{lllll}
			\Tag{term} & \L \sep{::=} \abs{x}{\L}   & \Tag{abstraction} \\
			              &         \sep{ | } \app{\L}{\L} & \Tag{application} \\
			              &         \sep{ | } x                  & \Tag{variable}    \\
			              &   \sep{ | } \letin{B}{\L}                          & \Tag{letrec}    \\
			\Tag{binding group} & B \sep{::=} f_1=\L \eqsep \dots \eqsep f_n=\L & \Tag{bindings} \\
			                       &   \sep{   } (f_1,\dots,f_n~\text{all distinct})   \\
		\end{array}
	\end{equation*}
\end{para}


\section{Unfolding \lambdaletrec-terms -- informal}\label{sec:unfolding_informal}

\newcommand\rulearray[1]{\[\begin{array}{lll} #1 \end{array}\]}
On this grammar we will now develop a rewriting system to describe unfolding of \lambdaletrec-terms in an informal notation.

\begin{para}[substitution of function variables]
	First we need a rule to perform the actual unfolding, i.e.\ the substitution of a function variable occurrence by the right-hand side of its definition.
	\rulearray{
		\letin{B_1 \eqsep f = \L \eqsep B_2}{f} \sep{\red_\rec} \letin{B_1 \eqsep f = \L \eqsep B_2}{\L}
	}
	This rule is only applicable to a function binding with a function variable as its body.
\end{para}

\begin{para}[distributing function bindings]
	In case of a more complex body, we distribute the function binding over its constituents. The following two rules distribute function bindings over applications and abstractions.
	\rulearray{
		\letin{B}{\app{\L_0}{\L_1}} \sep{\red_@} \app{(\letin{B}{\L_0})}{(\letin{B}{\L_1})}
		\\[1ex]
		\letin{B}{\abs{x}{\L}} \sep{\red_λ} \abs{x}{\letin{B}{\L}}
	}
\end{para}

\begin{para}[merging function bindings]
	Two nested function bindings can be merged into one:
	\rulearray{
		\letin{B_0}{\letin{B_1}{\L}} \sep{\red_\merg} \letin{B_0,B_1}{\L}
	}
\end{para}

\begin{para}[name clashes, α-renaming]
	Note, that the rules above are all in informal notation. The actual definitions (\cref{sec:unfolding_formal}) are CRS rewriting rules. Thus, name clashes (as for instance two functions of the same name being defined in \m{B_0} as well as in \m{B_1} in the above rewriting rule) are not a problem that we need to concern ourselves with, as they are dealt with by the CRS formalism. When using first-order notation we will rename variables whenever necessary (or convenient).
\end{para}

\begin{para}[garbage collection]
	The above rules would suffice to distribute the function bindings to the corresponding function variable occurrences and unfold them. In order to obtain an unfolded λ-term without any residual function bindings, we include these garbage-collection rules:
	\rulearray{
		\letin{f_1 = \L_1 \dots f_n = \L_n}{\L} \sep{\red_\reduce} \letin{f_{j_1} = \L_{j_1} \dots f_{j_{n'}} = \L_{j_{n'}}}{\L}
		\\ & & \hspace*{-4.9cm}\sideCondition{if \m{f_{j_1},\dots,f_{j_{n'}}} are the function variables reachable from \L}
		\\[1ex]
		\letin{}{\L} \sep{\red_\nil} \L
	}
	The latter discards empty function bindings, while the former removes all function bindings from a function binding that are not `reachable'. We consider a function binding to be `reachable' if the corresponding function variable either occurs in the body of the \Let-expression or in any other of the function bindings that is `reachable'. The side condition, which ensures that only superfluous bindings are removed from the binding group is non-trivial and requires a reachability analysis because there might be mutually recursive unused function bindings.
\end{para}

The above rules define a rewriting system for unfolding \lambdaletrec-terms to (possibly infinite) λ-terms.

\begin{example}[unfolding derivation of \m{\fun{fix} = \abs{f}{\letin{r=\app{f}{r}}{r}}}]\label{ex:unfolding_deriv_fix}
	\rewritingSequence{
		\init              {\abs{f}{\letin{r=\app{f}{r}}{r}}}
		\step{\red_\rec}   {\abs{f}{\letin{r=\app{f}{r}}{\app{f}{r}}}}
		\step{\red_@}      {\abs{f}{\app{(\letin{r=\app{f}{r}}{f})}{(\letin{r=\app{f}{r}}{r})}}}
		\step{\red_\reduce}{\abs{f}{\app{(\letin{}{f})}{(\letin{r=\app{f}{r}}{r})}}}
		\step{\red_\nil}   {\abs{f}{\app{f}{(\letin{r=\app{f}{r}}{r})}}}
	}
	and therefore:
	\rewritingSequence{
		\init                 {\abs{f}{\letin{r=\app{f}{r}}{r}}}
		\step{\mred_\unfold}  {\abs{f}{\app{f}{(\letin{r=\app{f}{r}}{r})}}}
		\step{\mred_\unfold}  {\abs{f}{\app{f}{(\app{f}{(\letin{r=\app{f}{r}}{r})})}}}
		\step{\infred_\unfold}{\abs{f}{\app{f}{(\app{f}{(\app{f}{\dots})})}}}
	}
	We say that \fun{fix} \emph{unfolds to} \m{~\abs{f}{\app{f}{(\app{f}{(\app{f}{\dots})})}}~} and write
	\[\unfsem{\fun{fix}} = \abs{f}{\app{f}{(\app{f}{(\app{f}{\dots})})}}\]
\end{example}

\begin{para}[meaningless bindings]\label{meaningless_bindings}
	However, not every \lambdaletrec-term represents a λ-term. For instance the \lambdaletrec-term \L defined as \[\abs{x}{\letin{f = f}{\app{f}{x}}}\] has a meaningless function binding \m{f = f} that does not unfold to a λ-term. The rewriting rules above admit only the cyclic rewriting sequence \m{\L \red_\rec \L}. Therefore \m{\L ∉ \domain{\unfsem{}}}, which means that the unfolding semantics \unfsem{} based on these rules can only be partial. In order to obtain a total unfolding semantics, we include a constant symbol to signified that the term is undefined at the point of its occurrence. As is customary since \cite{ario:klop:1996} we use the `black hole' symbol \bh, which we include in an extended version of the grammars for λ-terms and \lambdaletrec-terms. The unfolding semantics of \L will then be \abs{x}{\app{\bh}{x}}. On the extended grammar we define two additional rules for turning meaningless bindings into black holes:
	\rulearray{
		\letin{B_1, f = g, B_2}{\L} \sep{\red_\tighten} \letin{\subst{B_1}{f}{g}, \subst{B_2}{f}{g}}{\subst{\L}{f}{g}}
		\\ && \hspace{-4cm} \sideCondition{if \g is defined in \m{B_1} or \m{B_2}} \\
		\letin{B_1, f = f, B_2}{\L} \sep{\red_\bh} \letin{\subst{B_1}{f}{\bh}, \subst{B_2}{f}{\bh}}{\subst{\L}{f}{\bh}}
	}
	The former rule inlines alias functions, which simplifies meaningless bindings to the form \m{f = f}, such that they can be turned into a black hole by the latter rule. Thus we obtain an extended unfolding semantics \unfbhsem{} which (in contrast to \unfsem{}) is defined for all \lambdaletrec-terms.
\end{para}

\begin{example}[\lambdaletrec-term with a meaningless binding]
	The rules for handling meaningless bindings allow us to reduce \L from above to a normal form, which means that \m{\L ∈ \domain{\unfbhsem{}}}, or in particular \m{\unfbhsem{\abs{x}{\letin{f = f}{\app{f}{x}}}} = \abs{x}{\app{\bh}{x}}}\ \ as is witnessed by the following rewriting sequence:
	\rewritingSequence{
		\init              {\abs{x}{\letin{f = f}{\app{f}{x}}}}
		\step{\red_@}      {\abs{x}{\app{(\letin{f = f}{f})}{(\letin{f = f}{x})}}}
		\step{\red_\reduce}{\abs{x}{\app{(\letin{f = f}{f})}{(\letin{}{x})}}}
		\step{\red_\nil}{\abs{x}{\app{(\letin{f = f}{f})}{x}}}
		\step{\red_\bh}    {\abs{x}{\app{\letin{}{\bh}}{x}}}
		\step{\red_\nil}   {\abs{x}{\app{\bh}{x}}}
	}
\end{example}

\begin{example}[meaninglessness due to mutual recursion]
	\letin{f = g \eqsep g = f}{f} is meaningless due to mutually recursive functions. With the aid of \m{\red_\tighten} the term can be reduced to a normal form:
	\rewritingSequence{
		\init               {\letin{f = g \eqsep g = f}{f}}
		\step{\red_\tighten}{\letin{f = g \eqsep g = g}{g}}
		\step{\red_\reduce} {\letin{g = g}{g}}
		\step{\red_\bh}     {\letin{}{\bh}}
		\step{\red_\nil}    {\bh}
	}
\end{example}

\begin{example}[meaninglessness due to nested mutual recursion]\label{ex:non-unfoldable}
	Let us consider another \lambdaletrec-term \L defined as \letin{f=\letin{g=f}{g}}{f}, which illustrates that meaninglessness is not always tied to a simple pattern. \m{\L ∉ \domain{\unfsem{}}}, but \m{\L ∈ \domain{\unfbhsem{}}} as is witnessed by the following rewriting sequence:
	\rewritingSequence{
		\init               {\letin{f = \letin{g = f}{g}}{f}}
		\step{\red_\rec}    {\letin{f = \letin{g = f}{f}}{f}}
		\step{\red_\reduce} {\letin{f = \letin{}{f}}{f}}
		\step{\red_\nil}    {\letin{f = f}{f}}
		\step{\red_\bh}     {\letin{}{\bh}}
		\step{\red_\nil}    {\bh}
	}
\end{example}

\begin{para}[\m{\red_\unfold} and \m{\red_\unfoldbh}]
	We define rewriting relations \m{\red_\unfold} and \m{\red_\unfoldbh} for unfolding \lambdaletrec-terms, where \m{\red_\tighten} and \m{\red_\bh} are included only in \m{\red_\unfoldbh}:
	\begin{align*}
		&\red_\unfold   &=~~~& \bigcup\setcompr{\red_\arule}{\arule ∈ \set{@,λ,\merg,\reduce,\nil}} \\
		&\red_\unfoldbh &=~~~& \bigcup\setcompr{\red_\arule}{\arule ∈ \set{@,λ,\merg,\reduce,\nil,\tighten,\bh}}
	\end{align*}
\end{para}


\begin{para}[necessity of \m{\red_\reduce} and \m{\red_\nil}]\label{rem:reduce:nil:motivation}
	The purpose of \m{\red_\reduce} together with \m{\red_\nil} is to prevent unbounded growth of binding groups during unfolding. Consider for instance the outermost rewrite sequence on the term \m{{\letin{f=\letin{g={\app{f}{g}}}{g}}{f}}} shown in \cref{fig:unbounded:growth:bindgroup}, where after the fourth rewriting step \m{g'} (the left one) becomes unreachable and could be removed by a \m{\red_\reduce}-step.
	\begin{figure}
		\begin{align*}
			&&& \letin{f=\letin{g={\app{f}{g}}}{g}}{f}
			\\
			& \m{\red_\rec} && \letin{f=\letin{g={\app{f}{g}}}{g}}{\letin{g'={\app{f}{g'}}}{g'}}
			\\
			& \m{\red_\merg} && \letin{
			\begin{array}{l}
				f=\letin{g={\app{f}{g}}}{g}
				\\
				g'={\app{f}{g'}}
			\end{array}
			}{g'}
			\\
			& \m{\red_\rec} && \letin{
			\begin{array}{l}
				f=\letin{g={\app{f}{g}}}{g}
				\\
				g'={\app{f}{g'}}
			\end{array}
			}{{\app{f}{g'}}}
			\\
			& \m{\red_@} && \app{\left(\letin{
			\begin{array}{l}
				f=\letin{g={\app{f}{g}}}{g}
				\\
				g'={\app{f}{g'}}
			\end{array}
			}{f}\right)}
			{\left(\letin{
			\begin{array}{l}
				f=\letin{g={\app{f}{g}}}{g}
				\\
				g'={\app{f}{g'}}
			\end{array}
			}{g'}\right)}
			\\
			& \m{\red_\rec} &&
			\appbreak{\letin{\begin{array}{l}
				f=\letin{g={\app{f}{g}}}{g}
				\\
				g'={\app{f}{g'}}
			\end{array}}
			{\letin{g''={\app{f}{g''}}}{g''}}}
			{\letin{
			\begin{array}{l}
				f=\letin{g={\app{f}{g}}}{g}
				\\
				g'={\app{f}{g'}}
			\end{array}
			}{g'}}
			\\
			& \m{\red_\merg} &&
			\app
			{\Biggl(\letin{
			\begin{array}{l}
				f=\letin{g={\app{f}{g}}}{g}
				\\
				g'={\app{f}{g'}}
				\\
				g''={\app{f}{g''}}
			\end{array}
			}{g''}\Biggr)}
			{\left(\letin{
			\begin{array}{l}
				f=\letin{g={\app{f}{g}}}{g}
				\\
				g'={\app{f}{g'}}
			\end{array}
			}{g'}\right)}
		\end{align*}
		\caption{Unbounded growth of binding groups indicated by the initial segment of an infinite \m{\red_\unfold}-rewrite-sequence without \m{\red_\reduce}-steps.}
		\label{fig:unbounded:growth:bindgroup}
	\end{figure}
	\par While restricting the size of binding groups during unfolding is a sensible constraint on the unfolding process, it is not strictly necessary to define the unfolding of a \lambdaletrec-term. Alternatively, we could employ a rule as follows:
	\newcommand\free{\text{free}}
	\[
		\letin{f_1 = \L_1 \dots f_n = \L_n}{\L}
		\red_\free \L \hspace{3ex} \sideCondition{if \m{f_1,\dots,f_n} do not occur in \L}
	\]
	Note that a \m{\red_\free}-step can be simulated by a \m{\red_\reduce}-step followed by a \m{\red_\nil}-step.
	We will, however, at a later point embed the unfolding rules into other rewriting systems of which we wish to perform unfolding in a lazy way such that the number of derivable subterms is bounded. Using the \m{\red_\free}-rule instead of \m{\red_\reduce} and \m{\red_\nil} the size of the bindings groups in \cref{fig:unbounded:growth:bindgroup} keeps growing and we obtain an infinite number of different subterms.
\end{para}

\begin{para}[informal notation]\label{para:informal_notation}
	We will be using the informal notation as above throughout most of the thesis. For reasoning however, we lean on the theory of higher-order rewriting. In this way we can avoid the ado of an explicit substitution calculus, which would be required for sound reasoning in a first-order formulation.
\end{para}

\section{Preliminaries}\label{sec:prelims}

\begin{notation}[ℕ, natural numbers]
	By ℕ we denote the natural numbers including zero. We let \m{ℕ = \set{0,1,\dots}}.
\end{notation}

\begin{notation}[functions, domain, image]
	For a total function \m{f : A → B} we denote by \domain{f} the \emph{domain} \A, and by \image{f} the \emph{image} \B of \f. \par For a partial function \m{f : A ⇀ B}, and \m{a ∈ A} we denote by \m{\defd{f(a)}} that \f is defined for \a. The \emph{domain} of \f is the set \m{\domain{f} := \setcompr{a ∈ A}{\defd{f(a)}}}.
\end{notation}

\begin{notation}
	We denote by \srestrictto{f}{D} the restriction of function \f to domain \m{D}.
\end{notation}

\begin{notation}[rewrite relations]
	Let \m{\red\ \subseteq A \times A} be a rewrite relation. We denote by \mred the \emph{many-step} rewrite relation induced by \red, by which we mean the reflexive and transitive closure of \red. By \morestepred we denote the \emph{one-or-more-step} rewrite relation of \red, the transitive closure of \red. By \eqred we mean the \emph{zero-or-one-step} rewrite relation of \red, the reflexive closure of \red. By \infred we denote the infinite rewrite relation of finitely of infinitely many \red-steps (see \cref{infred} for more details).
	By a normal form of \red we mean an \m{a ∈ A} such that there is no \m{a' ∈ A} with \m{a \red a'}. By \nfred we mean the \emph{reduction to normal form} rewrite relation induced by \red. It is equivalent to the restriction of \mred to a relation with the normal forms of \red as codomain: \m{\nfred\ = \setcompr{\pair{a}{a'}}{a \mred a', \text{\m{a'} normal form of \red}}}.
\end{notation}

\section{λ-terms and \lambdaletrec-terms -- CRS formalisation}\label{sec:terms_formal}

\begin{para}[Combinatory Reduction Systems]
	Many of the formalisations we introduce are based on the framework of Combinatory Reduction Systems (CRSs) \cite{klop:1980}, \cite{klop:oost:raam:1993} \cite[Section~11.3]{terese:2003}, and, in particular, on infinitary Combinatory Reduction Systems (iCRSs) \cite{kete:simo:2011}.
	CRSs are a higher-order term rewriting framework tailor-made for formalising and manipulating expressions in higher-order languages (i.e.\ languages with binding constructs like λ-abstractions and function bindings). They provide a sound basis for defining our language and for reasoning with \letrec-expressions. By formalising a system of unfolding rules as a CRS we conveniently externalise issues like name capturing and α-renaming, which otherwise would have to be handled by a calculus of explicit substitution. Also, we can lean on the rewriting theory of CRSs for the proofs.
\end{para}

\begin{remark}[infinitary rewriting]
	We rely on CRSs (and not for instance Higher-Order Rewriting Systems (HRSs) \cite{terese:2003}) as a rewriting framework since to date infinitary rewriting theory has only been developed for CRSs \cite{kete:simo:2009, kete:simo:2010, kete:simo:2011}.
\end{remark}


\begin{para}[the `calculi' \lambdacal, \lambdabhcal, and \lambdaletreccal]
	We will use the symbols \lambdacal, \lambdabhcal, and \lambdaletreccal to refer to the λ-calculus, the λ-calculus with black holes, and the \lambdaletrec-calculus with \letrec. However, we consider the former two calculi not to have any rewriting rules at all, since we concern ourselves only with unfolding, not with β-reduction.
\end{para}

For formulating the rules from \cref{sec:unfolding_informal} as a CRS, we provide CRS signatures for \lambdacal, \lambdabhcal and \lambdaletreccal.

\begin{definition}[CRS signatures for \lambdacal, \lambdabhcal, and \lambdaletreccal]\label{def:sigs:lambdacal:lambdaletrec:CRS}
	The CRS signature for \lambdacal consists of the set \m{\sigCRS = \set{\sappCRS, \sabsCRS}} where \sappCRS is a binary (\m{\arity{\sappCRS} = 2}) and \sabsCRS is a unary (\m{\arity{\sabsCRS} = 1}) function symbol. The CRS signature for \lambdabhcal is \m{\sigCRSbh = \sigCRS ∪ \set{\bh}} where \bh is a nullary (\m{\arity{\bh} = 0}) function symbol. The CRS signature \sigCRSletrec consists of the countably infinite set \m{\sigCRSletrec = \sigCRSbh ∪ \set{\LetCRS} ∪ \setcompr{\InCRS{n}}{n ∈ ℕ}} of function symbols, with \m{\arity{\LetCRS}=1} and \m{\arity{\InCRS{n}}=n+1} for all \m{n ∈ ℕ}.
\end{definition}

\begin{definition}[set of λ-terms]\label{def:set_of_lambda_terms}
	By \iTer{\lambdacal} we denote the set of closed iCRS terms (terms in an infinitary CRS \cite{kete:simo:2011}) over \sigCRS. Likewise, we denote by \iTer{\lambdabhcal} the set of closed iCRS terms over \sigCRSbh.
	Note that the set includes finite as well as infinite terms, thus whenever we speak of λ-terms we refer to λ-terms that are either finite or infinite.
	\par We will look more closely at iCRS-terms in \cref{iCRS-terms}
\end{definition}

\begin{definition}[set of \lambdaletrec-terms]\label{def:set_of_lambdaletrec_terms}
	By \Ter{\lambdaletreccal} we denote the set of closed CRS terms over \sigCRSletrec, with the restrictions
	\begin{itemize}
		\item that there is no occurrence of the \bh-symbol
		\item that \LetCRS and \InCRS{n} can only occur as patterns of the form \[\letCRS{n}{f_1 \dots f_n}{\dots}\]
		\item and that otherwise a CRS abstraction can only occur directly beneath an \sabsCRS-symbol.
	\end{itemize}
\end{definition}

\begin{example}[\app{\fun{fix}}{\fun{id}}]\label{ex:crs_fix_id}
	The naive version of \fun{fix} applied to the identity function as in \cref{ex:unf_fix_id} in CRS notation:
	\[\appCRS{\letCRS{1}{\fun{fix}}{\absCRS{f}{\appCRS{f}{\appCRS{\fun{fix}}{f}}}, \fun{fix}}}{\absCRS{x}{x}}\]
	The unfolding of \fun{fix} in CRS notation:
	\[\absCRS{f}{\appCRS{f}{\appCRS{\absCRS{f}{\appCRS{f}{\appCRS{\dots}{f}}}}{f}}}\]
\end{example}

\begin{example}[\fun{fix}]\label{ex:crs-notation:fix}
	The (not so naively implemented) \fun{fix}-function from \cref{ex:unfolding_deriv_fix} in CRS notation:
	\[\absCRS{f}{(\letCRS{1}{r}{\appCRS{f}{r}, r})}\]
\end{example}


\section{Unfolding Rules -- CRS formalisation}\label{sec:unfolding_formal}

Here we give a CRS formalisation of the rules for unfolding \lambdaletrec-terms, corresponding to unfolding as described informally in \cref{sec:unfolding_informal}.

\begin{definition}[CRSs \unfCRS and \unfbhCRS for unfolding \lambdaletrec-terms]\label{def:unfCRS}
	\unfCRS and \unfbhCRS for unfolding \lambdaletrec-terms are CRSs over the signature \sigCRSletrec. The rules of \unfbhCRS consist of all the rule schemes below, while in \unfCRS the last two rules schemes are excluded (\unfrule{\tighten} and \unfrule{\bh}).
	We use vector notation to denote sequences of CRS abstractions (\vec{f} instead of \m{f_1,\dots,f_n}) and metaterms (\m{\vec{B}(\vec{f})} instead of \m{B_1(\vec{f}),\dots,B_n(\vec{f})}).

	\newcommand\rulename[1]{&\unfrule{#1}:~}
	\newcommand\rewritestoBreak{\\&\hspace{5mm}→}
	\newcommand\rewritesto{\red}
	\newcommand\nextRule{\\[0.5em]}
	\newcommand\annotation[1]{\\&\hspace{1cm}{#1}}
	\newcommand\annotationHigher[1]{\\[-1em]&\hspace{1cm}{#1}}

	\begin{align*}
		\rulename{λ}
			\letCRS{n}{\vec{f}}{\vec{B}(\vec{f}), \absCRS{x}{M(\vec{f},x)}}
			\rewritestoBreak
			\absCRS{x}{\letCRS{n}{\vec{f}}{\vec{B}(\vec{f}), M(\vec{f},x)}}
			\nextRule
		\rulename{@}
			\letCRS{n}{\vec{f}}{\vec{B}(\vec{f}), \appCRS{M(\vec{f})}{N(\vec{f})}}
			\rewritestoBreak
			\appCRSbreak{\letCRS{n}{\vec{f}}{\vec{B}(\vec{f}),M(\vec{f})}}{\letCRS{n}{\vec{f}}{\vec{B}(\vec{f}),N(\vec{f})}}
			\nextRule
		\rulename{\merg}
			\letCRS{n}{\vec{f}}{\vec{B}(\vec{f}), \letCRS{m}{\vec{g}}{\vec{C}(\vec{f},\vec{g}), M(\vec{f},\vec{g})}}
			\rewritestoBreak
			\letCRS{n+m}{\vec{f}\vec{g}}{\vec{B}(\vec{f}), \vec{C}(\vec{f},\vec{g}), M(\vec{f},\vec{g})}
			\nextRule
		\rulename{\rec}
			\letCRS{n}{\vec{f}}{\vec{B}(\vec{f}), f_i}
			\rewritesto
			\letCRS{n}{\vec{f}}{\vec{B}(\vec{f}), {B}_i(\vec{f})}
			\nextRule
		\rulename{\nil}
			\letCRS{0}{}{M}
			\rewritesto
			M
			\nextRule
		\rulename{\reduce}
		\letCRS{n}{\vec{f}}{\listcompr{i ∈ \set{1,\dots,n}}{\hspace{-2.5ex}B_i(\vec{f}_i)}~,~~M(\vec{f'})}
			\rewritestoBreak
			\letCRS{|I|}{\vec{f}'}{\listcompr{i∈I}{B_i(\vec{f}')}~,~~M(\vec{f'})}
			\annotation{\text{for some \m{I \subset \set{1,\dots,n}}}}
			\annotationHigher{\text{with } \m{\vec{f}' = \listcompr{i ∈ I}{f_i}} \text{~~~and } \m{\vec{f}_i = \begin{cases}\vec{f}' & i∈I \\ \vec{f} & i∉I\end{cases}}}
			\nextRule
		\rulename{\tighten}
			\letCRS{n}{\vec{f}}{\vec{B}(\vec{f}), M(\vec{f})} \text{ where } B_i(\vec{f}) = f_j
			\rewritestoBreak
			\letCRS{n-1}{\vec{g}}{
				\begin{array}{l}
					B_1(\vec{g}'),\dots,{B_{i-1}}(\vec{g}'),\\
					\indent B_{i+1}(\vec{g}'),\dots,{B_n}(\vec{g}'), M(\vec{g}')
				\end{array}
			}
			\annotation{
				\begin{array}{l}
					\text{for some } i,j ∈ \set{1,\dots,n} \text{ with } i ≠ j\\
					\text{and } \vec{g}' = \tuple{g_1,\dots,g_{i-1},g_k,g_{i+1},\dots,g_{n-1}}\\
					\text{where \m{k=j} if \m{j<i} and \m{k=j-1} if \m{j>i}}
				\end{array}
			}
			\nextRule
		\rulename{\bh}
			\letCRS{n}{\vec{f}}{\vec{B}(\vec{f}), M(\vec{f})} \text{ where } B_i(\vec{f}) = f_i
			\rewritestoBreak
			\letCRS{n-1}{\vec{g}}{
				\begin{array}{l}
					B_1(\vec{g}'),\dots,{B_{i-1}}(\vec{g}'),\\
					\indent B_{i+1}(\vec{g}'),\dots,{B_n}(\vec{g}'), M(\vec{g}')
				\end{array}
			}
			\annotation{
				\begin{array}{l}
					\text{for some } i ∈ \set{1,\dots,n}\\
					\text{and } \vec{g}' = \tuple{g_1,\dots,g_{i-1},\bh,g_{i+1},\dots,g_{n-1}}
				\end{array}
			}
	\end{align*}

	The first four rule schemes require little explanation. The rules are generated from having \n and \m{m} range over ℕ. The length of the vectors \vec{f}, \vec{g}, \vec{B}, and \vec{C} are stipulated by the subscript of the \Let-symbol in accordance to \cref{def:set_of_lambdaletrec_terms}. \unfrule{\nil} is actually a rule, not a rule scheme.
	\par The formulation of \unfrule{\reduce} uses, as an abbreviation, an ad-hoc list-builder notation, which works just as customary mathematical notation for, say, the union by indexing over an ordered set. Compare:
	\begin{align*}
		\bigcup_{i∈\set{i_1,\dots,i_n}} \hspace{-2.5ex} A(i)    &= A(i_1) ∪ \dots ∪ A(i_n) \\[-1.3em]
		&&& \text{where } i_1 < \dots < i_n \hfill \\[-0.3em]
		\listcompr{i∈\set{i_1,\dots,i_n}}{\hspace{-2.5ex} A(i)} &= A(i_1), \dots, A(i_n)
	\end{align*}
	Moreover, in the rule scheme \unfrule{\reduce} not only does \n range over ℕ; also the index set \m{I} ranges over all subsets of \set{1,\dots,n}. The purpose of the rule scheme is to remove all bindings that are not \emph{required}. A binding is considered required if it is used directly or indirectly by \M. \m{I} is (by formulation of the rule scheme) a superset of all required bindings, or to be more precise: for a given term, only if \m{I} is chosen as a superset of the required bindings, the rule scheme yields an CRS rule applicable to the term. From the set of `valid' choices we consider the single minimal choice for \m{I} in order to remove \emph{all} unrequired bindings at once. An implementation of this rule scheme would entail a reachability analysis, which is implicit here.
\end{definition}

\begin{example}[\unfrule{\reduce}]
	To understand the rule scheme \unfrule{\reduce}, consider the term \m{\L=\letin{f_1=f_2 \eqsep f_2=f_1}{f_1}} or in CRS notation \m{\L=\letCRS{2}{f_1 f_2}{f_2, f_1, f_1}}. Considering only the four possibilities we have for \m{I}, \unfrule{\reduce} induces the following four CRS rules:
	\providecommand\CRSline[3]{\set{#1}:~~ & #2 \\ & ~~~~ → ~~ #3 \\}
	\begin{align*}
		\CRSline{   }{\letCRS{2}{f_1 f_2}{B_1(f_1,f_2), B_2(f_1,f_2), M()}}        {\letCRS{0}{}{M()}}
		\CRSline{ 1 }{\letCRS{2}{f_1 f_2}{B_1(f_1), B_2(f_1,f_2), M(f_1)}}         {\letCRS{1}{f_1}{B_1(f_1), M(f_1)}}
		\CRSline{ 2 }{\letCRS{2}{f_1 f_2}{B_1(f_1,f_2), B_2(f_2), M(f_2)}}         {\letCRS{1}{f_2}{B_2(f_2), M(f_2)}}
		\CRSline{1,2}{\letCRS{2}{f_1 f_2}{B_1(f_1,f_2), B_2(f_1,f_2), M(f_1,f_2)}} {\letCRS{2}{f_1,f_2}{B_2(f_1,f_2), M(f_1,f_2)}}
	\end{align*}
	The first and the third rule are not applicable to \L, since \M has an occurrence of \m{f_1}. The rules induced by \m{I=\set{1}} and \m{I=\set{2}} are not applicable to \L, because \m{B_1} has an occurrence of \m{f_2} and \m{B_2} has an occurrence of \m{f_1}. All the bindings here are used by \M, therefore the only applicable rule is the last one, which does not alter \L at all.
	\par Now let us consider the term \m{\L'=\letin{f_1=f_2 \eqsep f_2=f_1}{x}} or in CRS notation \m{\L'=\letCRS{2}{f_1 f_2}{f_2, f_1, x}}. As before the rules induced by \m{I=\set{1}} and \m{I=\set{2}} are not applicable. However this time not only the last but also the first rule is applicable, indicating that none of the bindings are used by \M and thus all of them can be removed in one \m{\red_\reduce}-step.
\end{example}

\begin{definition}[garbage free]\label{def:garbage-free}
	We call a \lambdaletrec-term that is a normal form w.r.t.\ to the rules \unfrule{\reduce} and \unfrule{\nil} \emph{garbage free}.
\end{definition}

\begin{notation}[\rulep{\arule}, \rulebp{s}{\arule}]\label{notation:rules}
	Above we used the notation \rulebp{s}{\arule} to denote a rule named \arule that belongs to the rewriting system \m{s}. We will omit the \m{s} and write \rulep{\arule} to denote a rule \arule if it is unambiguous to which rewriting system \arule belongs to.
\end{notation}

\begin{notation}[rewriting relations for \unfCRS and \unfbhCRS]
	We write \m{\red_\unfold} (\m{\red_\unfoldbh}) for the rewrite relation induced by \unfCRS (\unfbhCRS). And by
	\m{\red_λ},
	\m{\red_@},
	\m{\red_\merg},
	\m{\red_\rec},
	\m{\red_\nil},
	\m{\red_\reduce},
	\m{\red_\tighten}, and
	\m{\red_\bh},
	we denote the rewrite relations of both the CRSs that are induced by the rules
	\rulep{λ},
	\rulep{@},
	\rulep{\merg},
	\rulep{\rec},
	\rulep{\nil},
	\rulep{\reduce},
	\rulep{\tighten}, and
	\rulep{\bh},
	respectively.
\end{notation}

\begin{remark}[alternative formalisation with permutations]
	Note that in the rule patterns in \unfCRS we have to ensure that the function binding we want to refer to may be at an arbitrary position among the function bindings of a \Let-expression. An alternative approach would be to adding a permutation rule by which two adjacent function bindings can be swapped. Then the remaining rules which refer to a function binding could be written such that they always make use of the first function binding of the \Let-expression. This approach is for instance used in \cite{grab:roch:letfloating}.
\end{remark}


\begin{para}[informal notation]
	As mentioned in \cref{para:informal_notation} when studying examples we will mostly rely on the easier-to-read informal notation from \cref{sec:unfolding_informal} instead of the more cumbersome CRS notation.
\end{para}

\section{Unfolding Semantics of \lambdaletrec}\label{sec:unfolding_semantics}

Based on the CRSs \unfCRS and \unfbhCRS we define the unfolding semantics of \lambdaletrec-terms as its infinite unique normal form. The well-definedness of the functions below is witnessed by \cref{thm:unfsem_well-defined} below.

\begin{definition}[partial and total unfolding semantics]\label{def:unf-mapping}
	\unfCRS induces the unfolding semantics
	\begin{align*}
		\unfsem{} : \Ter{\lambdaletreccal} &⇀ \iTer{\lambdacal}\\
		            \L                &↦ \M ~~~ \text{if}~\L \infnfred_\unfold \M
	\end{align*}
	which is partial, because of meaningless bindings. (see \cref{meaningless_bindings})\\

	\unfbhCRS induces the unfolding semantics
	\begin{align*}
		\unfbhsem{} : \Ter{\lambdaletreccal} &→ \iTer{\lambdabhcal}\\
		              \L                &↦ \M ~~~ \text{if}~\L \infnfred_\unfoldbh \M
	\end{align*}
	which is total, as it maps meaningless terms to \bh.
\end{definition}

\begin{theorem}\label{thm:unfsem_well-defined}
	\unfsem{} and \unfbhsem{} are well-defined functions.
\end{theorem}

\begin{proof}
	Well-definedness of these mappings is guaranteed by the following properties of \unfCRS and \unfbhCRS:
	\begin{description}
		\item[\Cref{lem:unf-normalisation} (infinite normalisation):] the existence of a normal form (only required for \unfbhCRS)
		\item[\Cref{lem:unf-nf_well-formed} (well-formedness of the normal forms):] normal forms do indeed adhere to the signatures \sigCRS, and \sigCRSbh, respectively.
		\item[\Cref{lem:unique_normal_forms} (uniqueness of normal forms)]
	\end{description}
\end{proof}

\begin{proposition}[\unfsem{} is a specialisation \unfbhsem{}]
	\[ ∀ \L ∈ \domain{\unfsem{}} ~ \unfsem{\L} = \unfbhsem{\L} \]
	\begin{proof}
		For \unfCRS as well as \unfbhCRS a \lambdaletrec-term is a normal form if and only if it is a \Let-expression. Therefore every normal form of \unfsem{} is a normal form of \unfbhsem{}. Reachability of the normal forms is guaranteed by the fact that the rules of \unfbhCRS are a superset of the rules of \unfCRS.
	\end{proof}
\end{proposition}

\begin{terminology}
	We say that a \lambdaletrec-term \L\ \emph{unfolds to} \M if either \m{\unfsem{\L} = \M} or \m{\unfbhsem{\L} = \M}. Which is meant, should be clear from the context.
\end{terminology}

We finish the section with the lemmas required for \cref{thm:unfsem_well-defined}.

\begin{lemma}[well-formedness of the normal forms]\label{lem:unf-nf_well-formed}
	\unfCRS (\unfbhCRS) has λ-terms (\lambdabh-terms) as normal forms.
\end{lemma}

\begin{proof}[Proof sketch]
	The names of the first four rules are chosen to reflect the kind of term contained by the body of the \letrec-expression, which helps to see that the rules are complete in the sense that every \Let-expression is a redex, thus normal forms do not contain \Let-expressions. Terms over \sigCRSletrec without \Let-expressions are \lambdabh-terms. The arguments also holds analogously for \unfCRS, which contains no rule to give rise to a \bh-symbol.
\end{proof}

\begin{lemma}[infinite normalisation of \unfbhCRS]\label{lem:unf-normalisation}
	Every \lambdaletrec-term is either weakly normalising w.r.t.\ \m{\red_\unfoldbh} or admits a strongly convergent outermost-fair \m{\red_\unfoldbh}-rewriting-sequence.
\end{lemma}

For definitions of `outermost fair' and `strongly convergent' we refer to \cite{kete:simo:2010}.

\begin{proof}[Proof sketch]
	We consider outermost-fair \m{\red_\unfoldbh}-rewrite-sequences on a \lambdaletrec-term \L, in which \rulep{\reduce}, \rulep{\nil}, \rulep{\tighten}, and \rulep{\bh} are applied eagerly. We show that every such sequence \arewseq is either finite, or that otherwise its rewrite activity tends to infinity. We argue by contradiction: We assume \arewseq performs infinitely many rewriting steps on position \p. Then there is an infinite subsequence \brewseq of \arewseq which contracts only redexes at \p. We show that \brewseq cannot exist.
	\par First note that a \lambdaletrec-term is a \m{\red_\unfoldbh}-redex if and only if it is a \Let-expression. Also an outermost \m{\red_\unfoldbh}-rewriting-sequence cannot create redexes upwards. Therefore, for \brewseq to be infinite all terms in \brewseq must have a \Let-expression at \p.
	\par \brewseq can contain neither \m{\red_λ}-steps nor \m{\red_@}-steps, since they would generate a function symbol at \p (and thus yield a term that is not a \Let-expression).
	\par For the rest of the proof we argue using the concept of \emph{\letrec-depth}, the number of \Let-symbols occurring under \p.
	\par \brewseq cannot contain an infinite number of \m{\red_\nil}-steps because it reduces \letrec-depth and no other rule increases \letrec-depth.
	\par Thus, \brewseq must have an infinite suffix \crewseq which contains no \m{\red_λ}-, \m{\red_@}-, and \m{\red_\nil}-steps, but only steps due to \m{\red_\merg}, \m{\red_\rec}, \m{\red_\reduce}, \m{\red_\tighten}, and \m{\red_\bh}.
	\par For \crewseq to be infinite it must contain infinitely many \m{\red_\rec}-steps since the remaining four rules are only finitely often applicable, because they manipulate or remove bindings (which can happen only once per binding), or in the case of \m{\red_\merg} decrease the \letrec-depth.
	\par We will now go on to show that \crewseq cannot have infinitely many \m{\red_\rec}-steps, if \m{\red_\tighten}-steps and \m{\red_\bh}-steps take precedence. Every binding that is employed by the \m{\red_\rec}-step can only be of the form \m{f = g} or \m{f = \letin{\dots}{\dots}}, otherwise a \m{\red_λ}-step or a \m{\red_@}-step would ensue, which we already excluded. A term with only these two forms of bindings can be reduced to a term with \letrec-depth \m{1} with only bindings of the form \m{f = g} by applications of \rulep{\tighten}, \rulep{\reduce}, \rulep{\nil}, \rulep{\rec}, \rulep{\merg}. If the bindings are cyclic, the \crewseq terminates with an application of \rulep{\bh}, otherwise the resulting term is a free variable.
\end{proof}

\begin{lemma}[uniqueness of normal forms]\label{lem:unique_normal_forms}
	If for a term in \Ter{\lambdaletreccal} a normal form exists with respect to \unfCRS or \unfbhCRS, it is unique.
\end{lemma}

\begin{proof}
	This follows from finitary confluence (\cref{prop:unf-confluence})
\end{proof}



\begin{proposition}\label{prop:unf-confluence}
	\unfCRS and \unfbhCRS are confluent.
\end{proposition}

\begin{proof}
	Unique infinite normalisation of \unfCRS and \unfbhCRS follows from finitary confluence of \unfCRS and \unfCRS. In previous work \cite{grab:roch:confunf} we proved confluence for a CRS which is very similar to \unfCRS. In \cref{app:conf_proof} an adapted version of that proof is provided, also extended by the rules \rulep{\tighten} and \rulep{\bh}. The proof is based on decreasing diagrams \cite[Section~14.2]{terese:2003} and involves a comprehensive critical-pair analysis.
\end{proof}


\begin{remark}
	Note that this confluence result concerns a rewriting system only for unfolding \lambdaletrec-terms, and therefore does not conflict with non-confluence observations concerning versions of cyclic λ-calculi which include unfolding rules as well as β-reduction \cite{ario:klop:1997}.
\end{remark}

\begin{para}[\domain{\unfsem{}}]
	Finally we are going to characterise the domain of \unfsem{}, thus identify those \lambdaletrec-terms that are not meaningless but unfold to some λ-term. Meaningless \lambdaletrec-terms are unproductive in the sense that during all outermost-fair \unfCRS-rewrite-sequences the production of symbols stagnates due to an unproductive cycle.
\end{para}

\begin{lemma}\label{lem:outermost-fair-sequences}
	Let \L be a \lambdaletrec-term be a \Let-expression. Then exactly one of the following statements hold:
	\begin{itemize}
		\item All maximal outermost-fair \unfCRS-rewriting-sequences on \L solely contain \Let-expressions.
		\item All maximal outermost-fair \unfCRS-rewriting-sequences on \L contain finitely many \Let-expressions.
	\end{itemize}
\end{lemma}

\begin{definition}[\unfCRS-productivity]
	We say that a \lambdaletrec-term \L is \emph{\unfCRS-productive} if the following statement holds:
	\begin{itemize}
		\item \L does not have a \m{\red_\unfold}-reduct that is the source of an infinite \m{\red_\unfold}-rewrite-sequence consisting exclusively of outermost steps with respect to \m{\red_\rec}, \m{\red_\nil}, \m{\red_\merg}, or \m{\red_\reduce}.
\end{itemize}
\end{definition}

\begin{lemma}\label{lem:unfolding}
	For every \lambdaletrec-terms \L the following statements are equivalent:
	\begin{enumerate}[(i)]
		\item\label{lem:unfolding:i} \m{\L \m{\omegared_\unfold} \M} for some infinite λ-term \M.
		\item\label{lem:unfolding:ii} \L is \unfCRS-productive.
		\item\label{lem:unfolding:iii} Every maximal outermost-fair \m{\red_\unfold}-rewrite-sequence on \L is strongly convergent.
	\end{enumerate}
\end{lemma}

\begin{proof}
	\cref{lem:unfolding:ii} ⇒ \cref{lem:unfolding:iii}, because if \L is \unfCRS-productive then every outermost occurrence of a \letrec in every \m{\red_\unfold}-reduct will be eventually pushed down to a higher position by either a \m{\red_λ}-step or a \m{\red_@}-step of any maximal outermost-fair \m{\red_\unfold}-sequence. Since only \Let-expressions are \m{\red_\unfold}-redexes any maximal outermost-fair rewrite sequence starting from \L converges to an infinite normal form. (i) follows directly from (iii). \Cref{lem:unfolding:i} ⇒ \cref{lem:unfolding:ii} follows from \cref{lem:outermost-fair-sequences} by contradiction. If \L is not \unfCRS-productive then it has by definition a \m{\red_\unfold}-reduct with at least one occurrence of a \letrec which cannot be pushed further down by any outermost application of any \unfCRS-rule. By \cref{lem:outermost-fair-sequences} the same holds for every other maximal outermost-fair rewrite sequence. Therefore \L cannot unfold to a λ-term \M because \M may not contain any \letrec{s}.
\end{proof}

\begin{proposition}[\domain{\unfsem{}}]
	A \lambdaletrec-term is in \domain{\unfsem{}} if and only if it admits no cyclic \unfCRS-rewriting-sequence without \rulebp{\unfold}{@}-steps and \rulebp\unfold{λ}-steps.
\end{proposition}

\begin{proof}
	``⇒'' follows from confluence of \unfCRS (\cref{prop:unf-confluence}). For ``⇐'' we assume that \L is not infinitarily normalising in \unfCRS. That means that any fair strategy rewrites \M eventually to a reduct with root-active subterm \O. The infinite rewriting sequence that only applies rules to the root of \O can never contain any \rulebp{\unfold}{@}-steps or \rulebp{\unfold}{λ}-steps because the reduct would not be a redex. The cyclicity of this rewriting sequence follows from \cref{prop:Regletrec:stRegletrec}.
\end{proof}

\begin{para}[a single-rule unfolding system in the likeness of μ-unfolding]
	In \cref{app:single_rule} we present an alternative rewriting system for unfolding \lambdaletrec-terms. It has only a single rule, like the canonical rewriting system for unfolding μ-terms.
\end{para}

\begin{para}[outlook]
	Now that we have established what we mean by unfolding, we can phrase the problems we tackle in the three following chapters of this thesis:
	\begin{description}
		\item[\cref{chap:expressibility}:] Here we characterise the set of λ-terms that can be expressed finitely as \lambdaletrec-terms. To this end we introduce rewriting systems for deconstructing λ-terms. These rewriting systems induce a notion of regularity on a decomposed λ-term: if the set of its components is finite it is regular.
		\item[\cref{chap:representations}:] From the decompositions systems arise graphical representations for λ-terms in a natural way (essentially the reduction graph of a term w.r.t.\ a decomposition system). We study these graph representations in order to use them to reason about their respective λ-terms.
		\item[\cref{chap:maxsharing}:] Here we relate these graph formalisms back to \lambdaletreccal and unfolding, by which we obtain concrete practical methods to analyse and manipulate \lambdaletrec-terms.
	\end{description}
\end{para}

\chapter{Expressibility in \lambdaletrec}\label{chap:expressibility}

\section{Overview}\label{sec:expressibility:overview}

\setcounter{theorem}{-1}
\begin{para}[teaser]
	Why cannot all infinite \lambdaletrec-terms with a simple repetitive structure like \abs{a}{\abs{b}{\app{(\abs{a}{\app{(\abs{b}{\app{\dots}{a}})}{b}})}{a}}} be expressed in \lambdaletreccal?
\end{para}

\begin{para}[subject matter]
	In this chapter we study the relationship between finite terms in \lambdaletreccal and the λ-terms they express. We investigate λ-terms that are not unfoldings of any \lambdaletrec-term and we consider the question: Which are the infinite lambda terms that are \lambdaletrec-expressible in the sense that they can be obtained as infinite unfoldings of finite \lambdaletrec-terms? Or in other words: how expressive is the language \lambdaletreccal?
\end{para}


\begin{para}[methods and formalisms]
	We introduce a rewrite system for observing λ-terms through repeated experiments carried out at the head of the term, thereby decomposing it into `generated subterms'. There are four sorts of decomposition steps: \m{\red_λ}-steps (decomposing a λ-abstraction), \m{\red_{@_0}}-steps and \m{\red_{@_1}}-steps (decomposing an application into its function and argument), and a scope-delimiting step. The scope-delimiting step comes in two variants, \m{\red_\del} and \m{\red_\S}, defining two rewriting systems with rewrite relations \m{\red_\reg} and \m{\red_\streg}. These rewrite relations each induce a notion of `λ-transition-graph', a sort of graphical representation of λ-terms, which can be compared by means of bisimulation. We call a λ-term `regular' (`strongly regular') if its set of \m{\red_\reg}-reachable (\m{\red_\streg}-reachable) generated subterms (and therefore its λ-transition-graph) is finite. Furthermore, we analyse the binding structure of λ-terms with the concept of `binding--capturing chains'.
\end{para}

\begin{para}[result]
	Utilising these concepts, we answer the above question by providing two characterisations of \lambdaletrec-expressibility. For all λ-terms \M, the following statements are equivalent:
	(i): \M is \lambdaletrec-expressible;
	(ii): \M is strongly regular;
	(iii): \M is regular, and it only has finite binding--capturing chains.
\end{para}

\section{Introduction}\label{sec:expressibility:introduction}


In \cref{chap:introduction} we have established how \lambdaletrec-terms serve as a finite representation of (potentially) infinite λ-terms. It is quite obvious, that not every infinite λ-term can be represented finitely, only those λ-terms come into consideration that have some kind of repetitive, or regular, structure. It turns out, however, that not even λ-terms with a regular syntax tree can always be expressed as the infinite unfolding of a term in \lambdaletreccal (see \cref{ex:entangled} below).

\begin{terminology}[\lambdaletrec-expressible]
	We say that a λ-term \M is \emph{\lambdaletrec-expressible} if it has a representation as a (finite)\footnote{Recall that \lambdaletrec-terms are finite by \cref{def:set_of_lambdaletrec_terms}.} term \L in \lambdaletreccal:
	\[ ∃ \L ∈ \Ter{\lambdaletreccal} ~~ \M = \unfsem{\L} \]
	We then also say that \L\ \emph{expresses} \M.
\end{terminology}

Note that λ-terms have infinitely many different representations as \lambdaletrec-terms (see e.g.\ \cref{ex:expressibility:fix} below).

\begin{figure}
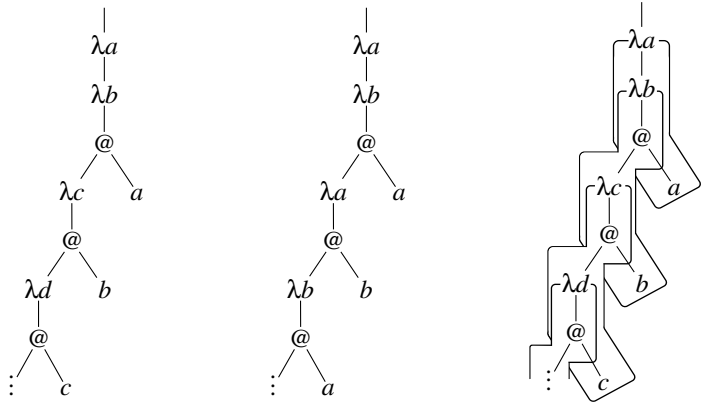

	\begin{hspread}
		\vcentered{\fig{pstricks/named/entangled-st}} & \vcentered{\fig{pstricks/named/entangled-ab-st}} & \vcentered{\fig{entangled-scopes}}
	\end{hspread}
	\caption{Two (α-equivalent) syntax trees of \M from \cref{ex:entangled}. The second syntax tree is cleary regular. The syntax tree on the right is a version of the first syntax tree with scopes depicted as box-like structures.}
	\label{fig:ex:entangled}
\end{figure}

\begin{example}[not \lambdaletrec-expressible]\label{ex:entangled}
	Consider the infinite λ-term of the form \m{\M = \abs{a}{\abs{b}{\app{(\abs{c}{\app{(\abs{d}{\app{\dots}{c}})}{b}})}{a}}}} with syntax trees as shown in \cref{fig:ex:entangled}. Even though it has a syntax tree with a regular structure, \M is not \lambdaletrec-expressible.
\end{example}

\begin{para}[the relevance of scopes]
	To understand why there is no \lambdaletrec-term that unfolds to \M, it helps to consider the scopes of the abstractions in \M. Informally speaking, the scope of an abstraction denotes the minimal connected portion of the term which includes the abstraction itself as well as all occurrences of the bound variable that it binds. A precise definition is given later in \cref{def:scope:extscope}. The scopes of \M as shown in \cref{fig:ex:entangled} are infinitely entangled: the scope of \abs{a}{} reaches into the scope of \abs{b}{}, the scope of \abs{b}{} into the scope of \abs{c}{}, and so on. This trait of \M suggests \M cannot be the result of `unrolling' a \lambdaletrec-term \L, since such a process would map \L's scoping structure onto \M in a regular manner, and result in a term which is `tiled' into finite, non-overlapping scopes with bounded scope-nesting depth. This excludes, intuitively, the formation of the infinite entanglement of successively overlapping scopes that can be observed in \M.
\end{para}

\begin{example}[\lambdaletrec-expressible]\label{ex:expressible:simpleletrec}
	Consider the infinite λ-term of the form \m{\M = \abs{xy}{\app{\app{\M}{y}}{x}}}.
	It is \lambdaletrec-expressible as it arises as the unfolding of \m{\letin{f = \abs{xy}{\app{\app{f}{y}}{x}}}{f}} as witnessed by the following rewrite sequence:
	\rewritingSequence{
		& \letin{f = \abs{xy}{\app{\app{f}{y}}{x}}}{f}
		\alignbreak\m{\red_\rec} ~& \letin{f = \abs{xy}{\app{\app{f}{y}}{x}}}{\abs{xy}{\app{\app{f}{y}}{x}}}
		\alignbreak\m{\red_λ} ~& \abs{x}{\letin{f = \abs{xy}{\app{\app{f}{y}}{x}}}{\abs{y}{\app{\app{f}{y}}{x}}}}
		\alignbreak\m{\red_λ} ~& \abs{xy}{\letin{f = \abs{xy}{\app{\app{f}{y}}{x}}}{\app{\app{f}{y}}{x}}}
		\alignbreak\m{\red_@} ~& \abs{xy}{\app{(\letin{f = \abs{xy}{\app{\app{f}{y}}{x}}}{\app{f}{y}})}{(\letin{f = \abs{xy}{\app{\app{f}{y}}{x}}}{x})}}
		\alignbreak\m{\red_\reduce} ~& \abs{xy}{\app{(\letin{f = \abs{xy}{\app{\app{f}{y}}{x}}}{\app{f}{y}})}{(\letin{}{x})}}
		\alignbreak\m{\red_\nil} ~& \abs{xy}{\app{(\letin{f = \abs{xy}{\app{\app{f}{y}}{x}}}{\app{f}{y}})}{x}}
		\alignbreak\m{\red_@} ~& \abs{xy}{\app{\app{(\letin{f = \abs{xy}{\app{\app{f}{y}}{x}}}{f})}{(\letin{f = \abs{xy}{\app{\app{f}{y}}{x}}}{y})}}{x}}
		\alignbreak\m{\red_\reduce} ~& \abs{xy}{\app{\app{(\letin{f = \abs{xy}{\app{\app{f}{y}}{x}}}{f})}{(\letin{}{y})}}{x}}
		\alignbreak\m{\red_\nil} ~& \abs{xy}{\app{\app{(\letin{f = \abs{xy}{\app{\app{f}{y}}{x}}}{f})}{y}}{x}}
		\alignbreak\m{\red_\rec} ~& \dots
	}
	As shown in \cref{fig:expressible:simpleletrec} the syntax tree of \M has some entanglement (the scopes of \abs{x}{} and \abs{y}{} do overlap) but the entanglement is finite as opposed to the entanglement in \cref{ex:entangled}.
\end{example}

\begin{figure}
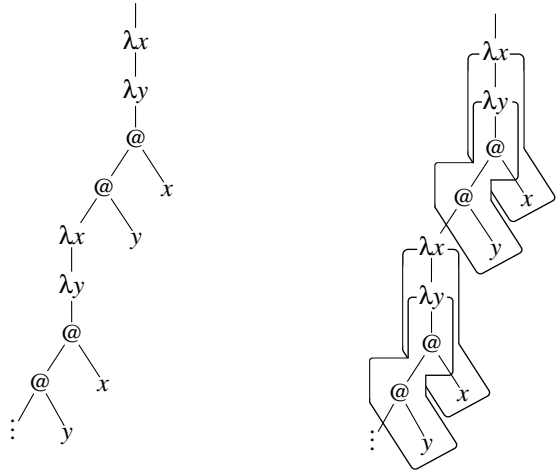

	\begin{hspread}
		\vcentered{\fig{simpleletrec-unf}} & \vcentered{\fig{simpleletrec-scopes}}
	\end{hspread}
	\caption{Syntax tree of \M from \cref{ex:expressible:simpleletrec}, and a version annotated with scopes.}
	\label{fig:expressible:simpleletrec}
\end{figure}

\begin{example}[\lambdaletrec-expressible]\label{ex:expressibility:fix}
	The infinite λ-term \abs{f}{\app{f}{(\app{f}{(\app{f}{\dots})})}} from \cref{ex:unfolding_deriv_fix} can be expressed by
	\L as well as by \P defined as follows:
	\[
		L := \abs{f}{\letin{r=\app{f}{r}}{r}} ~~~~~~~ P := \abs{f}{\letin{r=\app{f}{(\app{f}{r})}}{r}}
	\]
	It holds that \m{\unfsem{L} = \abs{f}{\app{f}{(\app{f}{(\app{f}{\dots})})}} = \unfsem{P}}.
	See \cref{fig:ex:expressibility:fix} for the corresponding `syntax graphs'. There is only one abstraction, so trivially there are no overlapping scopes.
\end{example}

\begin{figure}
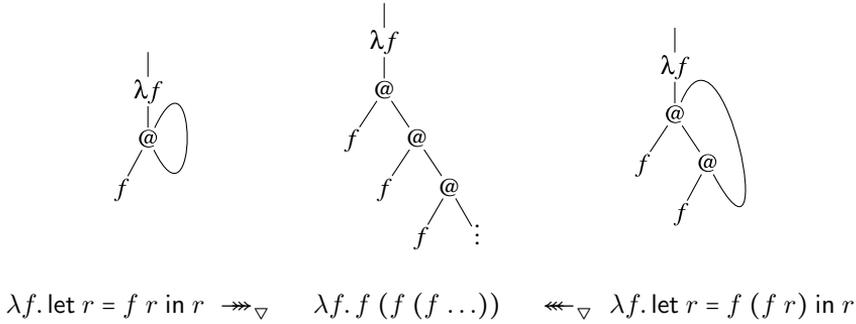

	\begin{hspread}
		\vcentered{\fig{fix-eff-sg}} & \vcentered{\fig{pstricks/named/fix-unf-st}} & \vcentered{\fig{fix-big-sg}} \\
		\m{\abs{f}{\letin{r=\app{f}{r}}{r}} ~ \infred_\unfold} &
		\m{\abs{f}{\app{f}{(\app{f}{(\app{f}{\dots})})}}} &
		 \m{\convinfred_\unfold ~ \abs{f}{\letin{r=\app{f}{(\app{f}{r})}}{r}}}
	\end{hspread}
	\caption{The `syntax graphs' of the \lambdaletrec-terms from \cref{ex:expressibility:fix} and the syntax tree of their unfolding.}
	\label{fig:ex:expressibility:fix}
\end{figure}

\begin{para}[regularity of first-order terms]
	Over a first order signature, the simplest kind of infinite terms are those that are regular in the sense that they possess only a finite number of different subterms. They correspond to trees over ranked alphabets that are regular \cite{cour:1983}. Like regular trees, also regular terms can be expressed finitely by systems of recursion equations \cite{cour:1983} or by `rational expressions' \cite[Definition\hspace*{1pt}4.5.3]{cour:1983}, which correspond to μ-terms (see e.g.\ \cite{endr:grab:klop:oost:2011}). Hereby finite expressions denote infinite terms either via a mathematical definition (a fixed-point construction, or induction on paths) or as the limit of a rewrite sequence consisting of unfolding steps.
\end{para}

\begin{para}[regularity of higher-order terms]
	For higher-order terms such as λ-terms the concept of regularity is less clear-cut from the outset, due to the presence of variable bindings. Frequently, regularity has been used to denote the existence of a first-order representation with named variables that is regular (e.g.\ in \cite{ario:klop:1997,ario:blom:1997}). According to that definition \M from \cref{ex:entangled} would be regular: while the syntax tree on the left in \cref{fig:ex:entangled} of the term \M contains infinitely many variables (and therefore is not regular), \M has another α-equivalent regular syntax tree (on the right) that uses only two variable names.
	However, such a definition of regularity has two drawbacks: it does not (as desired) correspond to \lambdaletrec-expressibility; and it relies on a property of a first-order representation.
	We will attempt to find a stronger, more direct definition of regularity for λ-terms that arises as an adaptation of the first-order notion of regularity.
\end{para}

\begin{para}[subterms of higher-order terms]
	The first-order notion of regularity relies on the finiteness of the number of subterms. When adapting this notion to higher-order terms, the problem is that there is no immediate, self-evident notion of subterm. Should for instance \y be considered a subterm of \abs{x}{x} which is after all equivalent to \abs{y}{y}? In order to arrive at a viable notion of higher-order subterms we define a CRS for decomposing higher-order terms. The decomposition `remembers' the abstractions encountered and `saves' them as part of the such obtained components which we call `generated subterms'.
\end{para}

\begin{para}[prefixed subterms]\label{prefixed-subterms}
	Viable notions of subterms for λ-terms in a higher-order formalisation require a stipulation on how to treat variable binding when descending into the body of a λ-abstraction. For this purpose we enrich the syntax of λ-terms with a parenthesised prefix of abstractions (similar to a proof system for weak μ-equality in \cite[Figure\hspace*{1.5pt}12]{endr:grab:klop:oost:2011}). An expression \prefixed{x_1 \dots x_n}{\M} represents a partially decomposed λ-term: the body \M typically contains free occurrences of variables which in the original λ-term were bound by λ-abstractions that have been split off by decomposition steps. The role of such abstractions has then been taken over by abstractions in the prefix \m{\absprefix{x_1 \dots x_n}}. In this way expressions with abstraction prefixes can be kept closed under decomposition steps. 
\end{para}

\begin{para}[decomposition CRSs]
	On these prefixed λ-terms, we will define two closely related rewrite systems \RegCRS and \stRegCRS. Rewrite sequences in \RegCRS and in \stRegCRS deconstruct λ-terms by steps that typically decompose applications and λ-abstractions that occur just below the abstraction prefix. \RegCRS and \stRegCRS differ with respect to the steps for removing vacuous prefix bindings they facilitate: while such bindings can always be removed by pertinent steps in \RegCRS, the system \stRegCRS only enables steps that remove vacuous bindings at the end of the abstraction prefix.
\end{para}

\begin{para}[scope-delimiting strategies]
	For each of these systems we consider a family of strategies that make deterministic choices concerning the application of the steps for removing vacuous prefix bindings: we call these strategies `scope-delimiting strategies' for \RegCRS, and `\extscope-delimiting strategies'\footnote{We usually say `extended scope' when pronouncing \extscope, see \cref{Reg:scope-delimiters}} for \stRegCRS. Scope-delimiting strategies \astrat for \RegCRS and \extscope-delimiting strategies \astratplus for \stRegCRS induce rewrite relations \m{\red_{\astrat}} and \m{\red_{\astratplus}}, respectively. These families of rewrite strategies define respective notions of `generated subterm', and they give rise to differently strong concepts of regularity: a λ-term \M is called regular (strongly regular) if there is a rewrite strategy \astrat for \RegCRS (a rewrite strategy \astratplus for \stRegCRS) such that the set of from \M \m{\red_{\astrat}}-reachable (\m{\red_{\astratplus}}-reachable) generated subterms is finite.
\end{para}

\begin{figure}
	\begin{hspread}
		\vcentered{\fig{simpleletrec-unf}} & \vcentered{\fig{simpleletrec-chains}}
		\\
		\vcentered{syntax tree} & \vcentered{binding--capturing chains}
		\\
		\vcentered{\fig{simpleletrec-scopes}} & \vcentered{\fig{simpleletrec-extscopes}} & \vcentered{\fig{simpleletrec-ltg}}
		\\
		\vcentered{scopes} & \vcentered{\extscope} & \vcentered{\m{\red_{\eagStratPlus}}-generated subterms}
	\end{hspread}
	\caption{Various depictions of the strongly regular infinite λ-term that is expressed by the \lambdaletrec-term \m{\letin{f = \abs{xy}{\app{\app{f}{y}}{x}}}{f}}.}
	\label{fig:ll:expressible}
\end{figure}

\begin{example}[\stRegCRS-decomposition]
	Before giving definitions of the decomposition CRSs (introduced in \cref{sec:regular}) let us consider an example and decompose\footnote{The term is decomposed with respect to the `eager scope-delimiting' strategy \eagStratPlus for \stRegCRS; see \cref{def:eagerStrats}.} the λ-term \M of the form \m{\M = \abs{xy}{\app{\app{\M}{y}}{x}}} from \cref{ex:expressible:simpleletrec}. The set of generated subterms consists of these prefixed terms (in no particular order):\\
		\prefixed{}{\M} ~~
		\prefixed{x}{\abs{y}{\app{\app{\M}{y}}{x}}} ~~
		\prefixed{xy}{\app{\app{\M}{y}}{x}} ~~
		\prefixed{xy}{{{\app{\M}{y}}}} ~~
		\prefixed{xy}{\M} ~~
		\prefixed{x}{\M} ~~
		\prefixed{xy}{y} ~~
		\prefixed{xy}{x} ~~
		\prefixed{x}{x}.\\
		Here, \prefixed{xy}{{{\app{\M}{y}}}} for instance is a witness for \app{\M}{y} being a (generated) subterm of \M, with the variables \x and \y being bound `somewhere above'. We could of course also have written \prefixed{xz}{{{\app{\M}{z}}}}, which is equivalent. So why is the subterm \prefixed{xy}{x} in the list along with \prefixed{x}{x}, if \x occurs at a position where both \x and \y have been bound? It is, because the decomposition rewrite systems include rules to remove variables from the prefix when no longer required (for leaving the scope).
	See \cref{fig:ll:expressible} for the reduction graph (\cref{def:reduction_graph}) of \M, which illustrates how the decomposition deconstructs \M into the generated subterms above. Since \prefixed{}{\M} has only 9 different \m{\red_{\eagStratPlus}}-reducts, the term \M is strongly regular.
\end{example}

\begin{para}[result]
	The generalisations of the concept of regularity to λ-terms suggest the question: do the expressibility results in \cite{cour:1983} for regular first-order trees with respect to systems of recursion equations, rational expressions, or μ-terms also generalise in an appropriate way? We tackle only the case of strong regularity here, and obtain an expressibility result with respect to the \lambdaletreccal, the λ-calculus with \letrec. We show that an infinite unfolding is unique if it exists, and it can be obtained as the limit of an infinite rewrite sequence of unfolding steps. We prove that a λ-term is \lambdaletrec-expressible if and only if it is strongly regular. This result confirms a conjecture\footnote{Cf.\ the last sentence of \cite[Section~1.2.4]{blom:2001}: `We conjecture that the set of regular lambda-trees is precisely the set of lambda-trees that can be obtained as the unwinding of terms with \letrec'. Mind that `regular lambda-trees' there correspond to strongly regular λ-terms in our sense, and that the notion of `sub-tree' of a `lambda-tree' corresponds to our notion of \m{\red_{\astratplus}}-generated subterm with respect to a \extscope-delimiting strategy \astratplus for \stRegCRS.} by Blom in \cite[Section 1.2.4]{blom:2001}.
\end{para}

\begin{para}[chapter overview]
	\cref{sec:expressibility:prelims} introduces formalisms used in this chapter, predominantly concerning abstract rewriting systems and rewriting strategies.
	In \cref{sec:regular} we introduce rewriting systems for decomposing λ-terms into `generated subterms', and we show some properties of these systems in connection to so-called scope-delimiting strategies. Also in this section, we define regularity and strong regularity for λ-terms employing the concepts of generated subterms and scope-delimiting strategies.
	In \cref{sec:letrec2} we adapt the rewrite systems for decomposing λ-terms and the notions of scope-delimiting strategies to the λ-calculus with \letrec.
	In \cref{sec:proofs}, we develop proof systems for the notions of regularity and strong regularity, for equality of strongly regular λ-terms, and for the property of a \lambdaletrec-term to unfold to a λ-term.
	In \cref{sec:chains} we examine the binding structure of λ-terms (binding--capturing chains) and connect to the concepts introduced so far.
	In \cref{sec:express} we establish the correspondence between strong regularity and \lambdaletrec-expressibility for λ-terms.
	In \cref{sec:lambda-transition-graphs} we introduce `λ-transition-graphs' of λ-terms and of \lambdaletrec-terms as
	labelled transition graphs in which the edges carry one of the four different labels \m{@_0}, \m{@_1}, λ, and \S.
	\Cref{sec:expressibility:conclusion} summarises the above results.
\end{para}

\section{Preliminaries}\label{sec:expressibility:prelims}
This section gathers known concepts vital to this chapter. Some notions concerning rewriting are recapitulated from \cite{terese:2003}, for others references are given. Some definitions of known concepts are simplified or tailored to our purposes.

\begin{notation}[\length{w}, length of a word \w]
	For words \w over some alphabet we denote the length of \w by \length{w}.
\end{notation}

We will use the following specific version of Kőnig's Lemma.

\begin{para}[Kőnig's Lemma]\label{koenigs_lemma}
	Let \m{G = \pair{V}{E}} be an undirected graph with set \V of vertices and set \E of edges. Suppose that \G has infinitely many vertices (\V is infinite), that it is connected (for all vertices \m{\v,\w ∈ \V} there exists a path in \G from \v to \w) and that every vertex has finite degree (it is adjacent to only finitely many other vertices in \G). Then for every vertex \m{\v ∈ \V}, \G contains an infinitely long simple path from \v, that is, a path starting at \v without repetition of vertices.\footnote{\label{footnote:Koenigs:Lemma}This formulation corresponds to the following original formulation by D\'{e}nes Kőnig \cite[page~80]{koen:2001} ``\emph{Satz 3: Jeder unendliche zusammenhängende Graph \G endlichen Grades besitzt einen einseitig unendlichen Weg, wobei der Anfangspunkt \m{P_0} dieses Weges beliebig vorgeschrieben werden kann.}'' in connection with the definition \cite[page~10 ]{koen:2001} ``Eine unendliche Menge von Kanten \m{P_i P_{i+1}} (\m{i=0,1,\dots} \emph{in inf.}), bzw.\ der durch sie gebildete Graph, heißt ein \emph{einseitig unendlicher Weg}, falls für \m{i ≠ j} stets \m{P_i ≠ P_j} ist.''}
\end{para}

\subsection{Rewriting Relations}

\begin{notation}[\m{\arel ⋅ \brel}, composition of relations \arel and \brel]
	For relations \m{\arel\subseteq\A\times\B} and \m{\brel\subseteq\B\times\C} we denote by \m{\arel ⋅ \brel} the \emph{composition of \arel with \brel} defined by \m{\arel ⋅ \brel := \setcompr{\pair{x}{z}}{(∃ y ∈ \B) \pair{x}{y} ∈ \arel ∧ \pair{y}{z} ∈ \brel} \subseteq \A\times\C}.
\end{notation}

\begin{notation}[\m{\arel^*}, reflexive transitive closure of a relation \arel]
	For a relation \m{\arel\subseteq\A\times\B} we denote by \m{\arel^*} the reflexive transitive closure of \arel under composition, for which it holds that \m{\arel^* := \bigcup_{i ∈ ℕ} \arel^i} where \m{\arel^0 := \idon{\A}{} := \setcompr{\pair{x}{x}}{x ∈ \A}} and, for all \m{i ∈ ℕ}, \m{\arel^{i+1} := \arel ⋅ \arel^i}.
\end{notation}

\subsection{Abstract Rewriting Systems}

We use abstract rewriting systems (ARSs) to reason about CRSs (from which we derive the ARSs) and more specifically about the set of terms that a term can be reduced to in an arbitrary number of reductions. ARSs are essentially binary relations on sets. They are in a sense graph-like structures, and more tangible in comparison to CRSs and therefore easier to manipulate.

\begin{para}[abstract rewriting systems]\label{def:ARS}
	An \emph{abstract rewriting system (ARS)} is a quadruple \quadruple{O}{\steps}{\src{}}{\tgt{}} consisting of a set \O of \emph{objects}, a set \steps of \emph{steps}, and \m{\src{}, \tgt{} : \steps → O}, the \emph{source} and \emph{target} functions. For objects \m{o ∈ O} we denote by \outgoingst{\steps}{o} and by \incomingst{\steps}{o} the set of steps in \steps that depart (are outgoing steps) from \o, and that arrive (are incoming steps) at \o, respectively. We say that an ARS is \emph{finite} if \steps is finite.
\end{para}

\begin{para}[ARS induced by a CRS/iCRS]\label{def:inducedARS}
	Let \aCRS be a CRS/iCRS over signature \sig with rules \m{R} and let \Ter{\sig}/\iTer{\sig} be the set of terms over \sig. We call the ARS \m{\aARS = \tuple{O,\steps,\src{},\tgt{}}} the \emph{ARS induced by} \aCRS where:
	\begin{gather*}
		O = \Ter{\sig} ~~/~~ O = \iTer{\sig}
		\\
		\steps = \setcompr{\triple{s}{\arule}{t}}{s,t ∈ O, \arule ∈ R, s \red_\arule t}
		\\
		\src{} : \triple{s}{\arule}{t} ↦ s
		\\
		\tgt{} : \triple{s}{\arule}{t} ↦ t
	\end{gather*}
\end{para}

\begin{notation}[induced ARS steps]\label{notation:inducedSteps}
	While the ARS formalism makes no qualitative distinction between different kinds of steps (\steps being an unstructured set in the definition of ARS), in \cref{def:inducedARS} we retain the original information from the CRS \aCRS as to which rule a step stems from (\steps being a set of triples). This allows us to do more fine-grained reasoning and we will continue to write \m{o_1 \red_\arule o_2} to indicate that there is a step \m{\astep ∈ \steps} with \m{\src{\astep} = o_1} and \m{\tgt{\astep} = o_2} which is due to the application of a CRS rule \rulebp{\aCRS}{\arule}, or even \m{o_1 \red_{\aARS.\arule} o_2} to indicate which ARS we mean specifically. We sometimes also name the step specifically and write \m{\astep : o_1 \red_\arule o_2} or \m{\astep : o_1 \red_{\aARS.\arule} o_2}.
\end{notation}

\begin{para}[sub-ARS]\label{def:sub-ARS}
	Let \m{\aARS_1 = \tuple{O_1,\steps_1,\srci{1}{},\tgti{1}{}}} and \m{\aARS_2 = \tuple{O_2,\steps_2,\srci{2}{},\tgti{2}{}}} be ARSs. We say that \m{\aARS_1} is a \emph{sub-ARS} of \m{\aARS_2} if \m{O_1 \subseteq O_2}, \m{\steps_1 \subseteq \steps_2}, and \srci{1}{}, \tgti{1}{} are the restrictions of \srci{2}{} and \tgti{2}{}, respectively, to \m{\steps_1}, which are required to be total functions. This implies that, for all \m{\astep ∈ \steps_1}, it holds that \m{\srci{2}{\astep} = \srci{1}{\astep} ∈ O_1}, and \m{\tgti{2}{\astep} = \tgti{1}{\astep} ∈ O_1}.
\end{para}

\begin{para}[generated sub-ARS]\label{def:generatedSubARS}
	For an object \m{o ∈ O} of an ARS \m{\aARS = \tuple{O,\steps,\src{},\tgt{}}} we denote by \m{\GeneratedSubARS{o} := \tuple{O',\steps',\src{}',\tgt{}'}} the \emph{sub-ARS of \aARS generated by} \o, where \m{O'} comprises only the objects from \O that are reachable from \o by an arbitrary number of steps (or no steps) and with \m{\steps'}, \m{\src{}'}, \m{\tgt{}'} being the restrictions of \steps, \src{}, \tgt{} to the objects in \m{O'} and to steps between objects in \m{O'}. We write \GeneratedSubARSi{\aARS}{o} to explicitly refer to a specific ARS.
\end{para}

\begin{terminology}[reduction graph]\label{def:reduction_graph}
	We call \GeneratedSubARS{o} the \emph{reduction graph} of \o, and \GeneratedSubARSi{\aARS}{o} the reduction graph of \o with respect to \aARS.
\end{terminology}

\subsection{ARS Strategies}

Next we will give definitions for history-free and history-aware strategies for ARSs. The latter is based on the notion of ARS labelling which in turn is based on the notion of ARS bisimulations. Note that ARS bisimulation is solely introduced for the sake of defining ARS labellings and is not to be confused with bisimulation between LTSs or TRSs.

\begin{para}[bisimulation between ARSs] 
	Let \m{\aARS_i = \quadruple{O_i}{\steps_i}{\srci{i}{}}{\tgti{i}{}}} for \m{i ∈ \set{1,2}} be ARSs. A relation \m{\aARSbisim~\subseteq (O_1 \times O_2) ∪ (\steps_1 \times \steps_2)}, which relates objects with objects and steps with steps, is called an \emph{ARS bisimulation} between \m{\aARS_1} and \m{\aARS_2} if the following holds:
	\begin{itemize}
		\item if \aARSbisim relates two objects, then \aARSbisim also relates their outgoing steps:
			\begin{align*}
				&∀ \pair{o_1}{o_2} ∈ O_1 \times O_2~~ o_1 \aARSbisim o_2 \implies \\
				&\hspace{1cm}∀ \astep_1 ∈ \outgoingst{\steps_1}{o_1} ~ ∃ \astep_2 ∈ \outgoingst{\steps_2}{o_2} ~~ \astep_1 \aARSbisim \astep_2~~∧\\
				&\hspace{1cm}∀ \astep_2 ∈ \outgoingst{\steps_2}{o_2} ~ ∃ \astep_1 ∈ \outgoingst{\steps_1}{o_1} ~~ \astep_2 \aARSbisim \astep_1
			\end{align*}
		\item if \aARSbisim relates two steps, then \aARSbisim also relates their sources and targets:
			\begin{align*}
				∀ \pair{\astep_1}{\astep_2} ∈ \steps_1 \times \steps_2~~ \astep_1 \aARSbisim \astep_2 \implies~
				&\srci{1}{\astep_1} \:\aARSbisim\: \srci{2}{\astep_2} ~∧\\
				&\tgti{1}{\astep_1} \:\aARSbisim\: \tgti{2}{\astep_2}
			\end{align*}
	\end{itemize}
\end{para}

In this work we need ARS-bisimulation only to define ARS-labellings:

\begin{para}[labellings of ARSs]
	Let \m{\aARS = \tuple{O,\steps,\src{},\tgt{}}} and \m{\aARS' = \tuple{O',\steps',\src{}',\tgt{}'}} be ARSs.
	\begin{enumerate}[(i)]
		\item An ARS bisimulation \alabelling between \aARS and \m{\aARS'} is called a \emph{labelling of \aARS to \m{\aARS'}}, and \m{\aARS'} \emph{the \alabelling-labelled version of} \aARS, if the converse \converse{\alabelling} of \alabelling is a function \m{\converse{\alabelling} : O' ∪\steps' → O ∪\steps}, and if additionally, for all \m{o' ∈ O'} and \m{o ∈ O} with \m{o \mathrel{\alabelling} o'}, the restriction \m{\srestrictto{\converse{\alabelling}}{\soutgoingst{\steps}'(o')} : \soutgoingst{\steps}'(o') → \soutgoingst{\steps}(o)} of \converse{\alabelling} to the steps departing from \m{o'} is bijective.
		\item A \emph{rewrite labelling} \L of \aARS to \m{\aARS'} is a pair \pair{\alabelling}{l} consisting of a labelling \alabelling of \aARS to \m{\aARS'} together with an \emph{initial} labelling function \m{l} mapping objects of \aARS to bisimilar objects of \m{\aARS'}.
	\end{enumerate}
\end{para}

\begin{para}[history-free strategy]
	A \emph{history-free strategy} for an abstract rewriting system \aARS is a sub-ARS of \aARS that has the same objects, and the same normal forms as \aARS.
\end{para}

\begin{para}[history-aware strategy]
	A \emph{history-aware strategy} for an abstract rewriting system \aARS is a history-free strategy for the \alabelling-labelled version of \aARS with respect to, and together with, a rewrite labelling \pair{\alabelling}{l} of \aARS.
\end{para}

\begin{para}[history-aware and history-free strategies]
	While strategies are simply defined as sub-ARSs of an ARS, here one may think of them as restrictions on the applicability of CRS rules. This is because in this work we consider ARSs that are induced by a CRS. A history-free strategy cannot restrict the applicability of a CRS rule to some term differently depending on the history (i.e.\ the previously applied rules) of that term. In history-aware strategy, on the other hands, terms are coupled with additional information that can be used to record the history of a term. Thus, how the applicability of CRS rules are restricted can differ depending on that record.
 	\par By a \emph{strategy} for \aARS we will mean a history-free strategy or a history-aware strategy for \aARS.
\end{para}

\begin{para}[projecting history-aware strategies to history-free strategies]
	Let \astrat be a history-aware strategy for \aARS, and let \m{\aARS'} be that \alabelling-labelled version of \aARS which \astrat is a history-free strategy for. Then \astrat projects to a history-free strategy \check{\astrat} of \aARS. The projection is defined by \alabelling, which induces a local bijective correspondence between outgoing steps of related sources of \aARS and \m{\aARS'}. Mind that for deterministic \astrat, \check{\astrat} may become non-deterministic. Furthermore, every rewrite sequence according to \astrat in \m{\aARS'} projects to a unique rewrite sequence in \aARS (which is a rewrite sequence according to \check{\astrat}).
	\par The last mentioned fact makes it possible to speak, for a given rewrite labelling, of rewrite sequences of a history-aware strategy for the objects of the original ARS.
\end{para}



Let \astrat be a history-aware strategy for an ARS \aARS, and \o an object of \aARS. Suppose that \astrat is a sub-ARS of the \alabelling-labelled version \m{\aARS'} of \aARS for some rewrite labelling \pair{\alabelling}{l} of \aARS. Then by a \emph{rewrite sequence of \astrat on \o} (in \aARS) we will mean the projection to \aARS of a rewrite sequence of \astrat (in \m{\aARS'}) on the result \m{l(o)} of the initial labelling applied to \o.


\subsection{Labelled Transition Systems}

\begin{para}[labelled transition systems]\label{def:labelled_transition_systems}
	A \emph{labelled transition system (LTS)} is a triple \m{\aLTS = \triple{\states}{L}{T}} consisting of a set \states of \emph{states}, a set \L of \emph{labels}, and a set \m{T \subseteq \states \times L \times \states} of \L-labelled transitions. We write \m{s_1 \red_\alab s_2} to denote labelled transitions, indicating that \m{\triple{s_1}{\alab}{s_2} ∈ T}.
\end{para}

\begin{para}[bisimulation between LTSs]
	Let \m{\aLTS_1 = \triple{\states_1}{L}{T_1}} and \m{\aLTS_2 = \triple{\states_2}{L}{T_2}} be a LTSs over a common set of labels. A \emph{bisimulation on \aLTS} is a binary relation \m{\abisim \subseteq \states_1\times\states_2} that satisfies, for all \m{\m{s_1} ∈ \states_1} and \m{s_2 ∈ \states_2}:
	\begin{enumerate}[(i)]
		\item if \m{\m{s_1} \mathrel{\abisim} s_2} and \m{\m{s_1} \red_\alab \m{s_1}'}, then there exists \m{s_2' ∈ \states_2} such that \m{s_2 \red_\alab s_2'} and \m{\m{s_1}'\mathrel{\abisim} s_2'};
		\item if \m{\m{s_1} \mathrel{\abisim} s_2} and \m{s_2 \red_\alab s_2'}, then there exists \m{\m{s_1}' ∈ \states_1} such that \m{\m{s_1} \red_\alab \m{s_1}'} and \m{\m{s_1}'\mathrel{\abisim} s_2'}.
	\end{enumerate}
	Two states \m{\m{s_1} ∈ \states_1} and \m{s_2 ∈ \states_2} are \emph{bisimilar}, denoted by \m{\m{s_1} \bisim s_2}, if there exists a bisimulation \abisim such that \m{\m{s_1} \mathrel{\abisim} s_2}.
\end{para}

\begin{remark}[ARSs versus LTSs]\label{rem:diff-ARS-LTS}
	Note that labelled transition systems (LTSs) are essentially the same as (indexed) ARSs \cite[1.1]{terese:2003}. Traditionally, research about LTSs is more concerned with the transitions and their labels while research about ARSs is more concerned with set of objects and their reachability. This is reflected in a different definition of bisimulation for the two systems. Only on LTSs the bisimulation is sensitive to transition labels; ARSs do not have explicit labels on their steps. In this work we adhere to the tradition. We use ARSs for the definition of (strong) regularity which depends on whether the set of reachable terms is finite, and we use LTSs to define a graph representation for λ-terms on which we study properties revolving around (functional) bisimulation.
\end{remark}

\begin{para}[labelled transition graphs]\label{def:labelled_transition_graphs}
	A \emph{labelled transition graph (LTG)} \G is a pointed LTS, that is, \m{G = \quadruple{\states}{L}{i}{T}} where \triple{\states}{L}{T} is an LTS, and \m{i ∈ \states}, which is called the \emph{initial state}.
\end{para}

\begin{para}[bisimulation between LTGs]
	Two LTGs \m{G_1 = \triple{\states_1,L}{i_1}{\red_1}} and \m{G_1 = \triple{\states_2,L}{i_2}{\red_2}} are \emph{bisimilar} if there is a bisimulation on the underlying LTSs that relates the initial states \m{i_1} and \m{i_2} with one another.
\end{para}

\subsection{iCRS Terms}

\begin{para}[iCRS preterms, iCRS terms]\label{iCRS-terms}
	When speaking of `infinite terms' for CRSs over some signature we draw on \cite[12.4]{terese:2003} and \cite{kete:simo:2011} where metaterms of iCRSs are defined by means of metric completion. The metric is defined on α-equivalence classes of finite metaterms dependent on the minimal depth at which two finite preterms belonging to the equivalence classes have a `conflict'. \emph{iCRS terms} are defined as the objects formed by the metric completion process. They can be represented as equivalence classes of infinite preterms, which we call \emph{iCRS preterms}, with respect to a notion of α-equivalence that again is based on the notion of `conflict' (see also \cite[Definition~12.4.1]{terese:2003}). iCRS preterms are infinite ordered dyadic trees in which each node is either labelled by a variable name, and then the node does not have a successor, or by named abstractions \abs{x}{} (with some variable name \x), and then the node has a single successor node, or by an application symbol, and then the node has a right and a left successor node.
\end{para}

\begin{definition}[α-equivalence for iCRS preterms, Schroer-style proof system]\label{def:iCRSPretermAlphaSchroer}
	The notion of α-equivalence on iCRS preterms based on the absence of conflicts can be described by provability in the proof system \iCRSPretermAlpha{\Schroer}{\sig} in \cref{fig:alpha-infpreterm:Schroer} which is an adaptation of a proof system due to Schroer \cite{hend:oost:2003}. The proof system \iCRSPretermAlpha{\Schroer}{\sig} over signature \sig consists of the axioms and rules displayed in \cref{fig:alpha-infpreterm:Schroer} and contains a rule \f for every \m{f ∈ \sig}.
	Provability in \iCRSPretermAlpha{\Schroer}{\sig} of an equation between preterms is defined as the existence of a possibly infinite, \emph{completed} derivation: for example, by \m{\infderivablein{\siCRSPretermAlpha{\Schroer}}{\apreter = \bpreter}} we mean the existence of a possibly infinite proof tree \infDeriv with conclusion \m{\apreter = \bpreter} such that maximal threads from the conclusion upwards either have length \omega, or have finite length and end at a leaf that carries an axiom. (We will generally use the decorated turnstyle symbol \sinfderivable to indicate provability by a completed, possibly infinite derivation.)
\end{definition}

\begin{figure}
	\proofsystem{
		\begin{bprooftree}
			\emptyAxiom
			\infLabel{\text{const}}
			\unaryInf{c = c}
		\end{bprooftree}
		\hsep
		\begin{bprooftree}
			\axiom{\subst\apreter{x}{c} = \subst\bpreter{y}{c}}
			\infLabel{\CRSabs{\hspace*{1pt}}{}}
			\unaryInf{\CRSabs{x}{\apreter} = \CRSabs{y}{\bpreter}}
		\end{bprooftree}
		\\
		\begin{bprooftree}
			\axiom{\apreter_1 = \bpreter_1}
			\axiom{\dots\dots}
			\axiom{\apreter_n = \bpreter_n}
			\infLabel{f}
			\ternaryInf{f(\apreter_1,\dots,\apreter_n) = f(\bpreter_1,\dots,\bpreter_n)}
		\end{bprooftree}
	}
	\caption{Schroer-style proof system \iCRSPretermAlpha{\Schroer}{\sig} for α-equivalence of iCRS preterms over signature \sig: for every \m{f ∈ \sig} with arity \n, \iCRSPretermAlpha{\Schroer}{\sig} contains a rule \f. In instances of the rule \m{\CRSabs{\hspace*{1pt}}{}}, the constant c is chosen fresh for \apreter and \bpreter. Substitution which occurs in the assumption of \m{\CRSabs{\hspace*{1pt}}{}} denotes substitution by variable replacement on iCRS preterms. It needs not to be capture-avoiding because of the freshness of the substituant.}
	\label{fig:alpha-infpreterm:Schroer}
\end{figure}

However, closer to coinductive proof systems for λ-terms that we develop is the following different, but equivalent characterisation of α-equivalence for infinite iCRS preterms, a variant for iCRS terms of a proof system for α-equivalence between finite λ-terms due to Kahrs (see \cite{hend:oost:2003}).

\begin{definition}[α-equivalence for iCRS preterms, Kahrs-style proof system]\label{def:iCRSPretermAlphaKahrs}
	The proof system \iCRSPretermAlpha{\Kahrs}{\sig} for α-equivalence on iCRS preterms over signature \sig consists of the axioms and the rules in \cref{fig:alpha-infpreterms:Kahrs} with, for every \m{f ∈ \sig}, a rule \f.
	Provability of an equation between preterms in \iCRSPretermAlpha{\Kahrs}{\sig} is defined, analogously as in \cref{fig:alpha-infpreterm:Schroer}, as the existence of a possibly infinite, completed derivation.
\end{definition}

\begin{figure}
	\proofsystem{
		\begin{bprooftree}
			\emptyAxiom
			\infLabel{\0}
			\unaryInf{\KahrsCRSabs{\vec{x}y}{y} = \KahrsCRSabs{\vec{z}u}{u}}
		\end{bprooftree}
		\\
		\begin{bprooftree}
			\axiom{\KahrsCRSabs{\vec{x}}{\apreter} = \KahrsCRSabs{\vec{z}}{\bpreter}}
			\infLabel{\S~~\mlSideCondition{if \y does not occur in \apreter,\\ and \w does not occur in \bpreter}}
			\unaryInf{\KahrsCRSabs{\vec{x}y}{\apreter} = \KahrsCRSabs{\vec{z}w}{\bpreter}}
		\end{bprooftree}
		\\
		\begin{bprooftree}
			\axiom{\KahrsCRSabs{\vec{x}y}{\apreter} = \KahrsCRSabs{\vec{z}u}{\bpreter}}
			\infLabel{\CRSabs{\hspace*{1pt}}{}}
			\unaryInf{\KahrsCRSabs{\vec{x}}{\CRSabs{y}{\apreter}} = \KahrsCRSabs{\vec{z}}{\CRSabs{u}{\bpreter}}}
		\end{bprooftree}
		\\
		\begin{bprooftree}
			\axiom{\KahrsCRSabs{\vec{x}}{\apreter_1} = \KahrsCRSabs{\vec{y}}{\bpreter_1}}
			\axiom{\dots\dots}
			\axiom{\KahrsCRSabs{\vec{x}}{\apreter_n} = \KahrsCRSabs{\vec{y}}{\bpreter_n}}
			\infLabel{f}
			\ternaryInf{\KahrsCRSabs{\vec{x}}{f(\apreter_1,\dots,\apreter_n)} = \KahrsCRSabs{\vec{y}}{f(\bpreter_1,\dots,\bpreter_n)}}
		\end{bprooftree}
	}
	\caption{Kahrs-style proof system \iCRSPretermAlpha{\Kahrs}{\sig} for α-equivalence on iCRS preterms over signature \sig: for every \m{f ∈ \sig} with arity \n, \iCRSPretermAlpha{\Kahrs}{\sig} contains a rule \f.}
	\label{fig:alpha-infpreterms:Kahrs}
\end{figure}

This formulation of α-equivalence for λ-terms will be the key to our formulation of a `coinduction principle' for λ-terms in \cref{thm:coinduction-principle}.

\begin{para}[infinite rewrite relation, \infred]\label{infred}
	For an iCRS with rewrite relation \red, we denote by \infred the infinite rewrite relation induced by strongly convergent and continuous rewrite sequences of arbitrary (countable) ordinal length. Hereby strong convergence means that, at every limit ordinal, the depth of the rewrite activity in the rewrite sequence's terms tends to infinity. Continuity means that the terms of the rewrite sequence converge, in the metric space of infinite terms, at every limit ordinal. By \omegared we will denote the rewrite relation induced by strongly continuous \red-rewrite-sequences of length \omega.
\end{para}

\section{Regular and strongly regular λ-terms}\label{sec:regular}

\begin{para}[regularity]\label{para:regularity}
	For infinite first-order trees the concept of regularity is well-known and well-studied \cite{cour:1983}. Regularity of a labelled tree\footnote{By a `labelled tree' we here mean a finite or infinite tree whose nodes are labelled by function symbols with a fixed arity. The arity determines the number of (ordered) successor nodes each node has.} is defined as the existence of only finitely many subtrees and implies the existence of a finite graph that unfolds to that tree. In this section we generalise the notion of regularity to trees with a binding mechanism, i.e.\ to λ-terms specifically. We give a definition for regularity which corresponds to regularity of a term when regarded as a first-order tree, and for strong regularity, which will be shown in the following sections to coincide with \lambdaletrec-expressibility.
\end{para}

\begin{para}[decomposition CRSs \RegCRS and \stRegCRS]
	We define regularity and strong regularity in terms of rewriting systems that will be called \RegCRS and \stRegCRS. Rewrite sequences in these systems \emph{inspect} a given term coinductively in the sense that a rewrite sequence corresponds to a decomposition of the term along one of its paths from the root. Both \RegCRS and \stRegCRS extend a kernel system \RegzeroCRS comprising three rewrite rules which denote whether the position just passed in the tree is an abstraction or an application and in the second case whether the application is being followed to the left or to the right.
\end{para}

\begin{para}[prefixed terms]\label{prefixed-terms}
	The rewriting systems are defined on λ-terms enriched by what we call an \emph{abstraction prefix}, by which the terms can be kept closed during the deconstruction. This is crucial for the definition of the rewriting system as a CRS. While intuitively it is clear that the λ-term \app\M\N is composed of the subterms \M and \N, abstractions are more problematic. In a first-order setting one could say that \abs{x}\M contains \M as a subterm, but if \x occurs freely in \M then \M would be an open term. That means that the scrutinisation of an abstraction would be able to go from a closed term to an open term, which would run counter to the interpretation of a higher-order term as an α-equivalence class. In the definition of the rewriting systems below this issue is resolved as follows. When inspecting an abstraction, the binder is eliminated but moved from the scrutinised subterm into the prefix. That guarantees that the term as a whole remains closed.
\end{para}

\begin{example}
	In \m{\abs{x}{\abs{y}{\app{\app{x}{x}}{y}}}} for instance the path from the root to the second occurrence of \x then corresponds to the rewrite sequence:
	\m{
		\prefixed{}{\abs{x}{\abs{y}{\app{\app{x}{x}}{y}}}}
		\red_λ
		\prefixed{x}{\abs{y}{\app{\app{x}{x}}{y}}}
		\red_λ
		\prefixed{xy}{\app{\app{x}{x}}{y}}
		\red_{@_0}
		\prefixed{xy}{\app{x}{x}}
		\red_{@_1}
		\prefixed{xy}{x}
	}.
\end{example}

\begin{para}[scope and \extscope]\label{Reg:scope-delimiters}
	The \RegCRS and the \stRegCRS system extend \RegzeroCRS by a \emph{scope-delimiting} rule, which signifies that the scope of an abstraction has ended, with both systems being based on different notions of scope.
	\RegCRS relies on what we simply call \emph{scope} of an abstraction: the range from the abstraction up to the positions under which the bound variable does not occur anymore.
	We base \stRegCRS on a different notion of scope, called \extscope (pronounced `extended scope'), which is strictly nested. The \extscopes of an abstraction extends its scope by encompassing all \extscopes that are opened within its range. As a consequence, \extscopes do no overlap partially (see \cref{fig:extscope}) but are strictly nested. In a sense \extscope is the transitive closure of scope; this becomes obvious in \cref{def:scope:extscope} where a precise definition of scope and \extscope is given. In \cite{bala:2012} \extscopes are called `skeletons'.
\end{para}

\begin{figure}
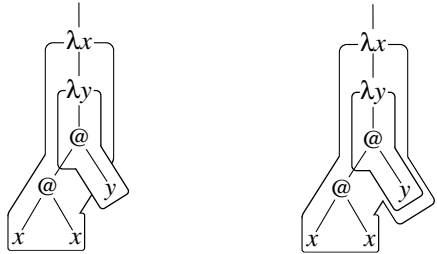

	\fig{scope} \hspace{2cm} \fig{extscope}
	\caption{The difference between scope and \extscope}
	\label{fig:extscope}
\end{figure}

\begin{para}[\extscope-delimiters identify abstractions]
	When every \extscope is closed by the \extscope-delimiting rule then the sequence of rewrite steps alone (i.e.\ without the terms themselves) unambiguously determines which abstraction a variable occurrence belongs to.
\end{para}

\begin{example}\label{ex:delimiter_x}
	The rewrite sequence from above would then have one additional \extscope-delimiting step asserting that the variable at the end of the path is indeed \x and not \y:
	\m{
		\prefixed{}{\abs{x}{\abs{y}{\app{\app{x}{x}}{y}}}}
		\red_λ
		\prefixed{x}{\abs{y}{\app{\app{x}{x}}{y}}}
		\red_λ
		\prefixed{xy}{\app{\app{x}{x}}{y}}
		\red_{@_0}
		\prefixed{xy}{\app{x}{x}}
		\red_{@_1}
		\prefixed{xy}{x}
		\red_\S
		\prefixed{x}{x}
	}
\end{example}

\begin{para}[abstraction prefix and \extscope]
	The abstraction prefix not only keeps the term closed but also denotes which \extscope is still open, which provides the information to decide applicability of the \extscope-delimiting rule.
\end{para}

\begin{example}
	The last step closes the \extscope of \y, therefore that variable is removed from the prefix.
	The rewrite sequence for the path to the occurrence of \y does not include a \extscope-delimiting step:
	\m{
		\prefixed{}{\abs{x}{\abs{y}{\app{\app{x}{x}}{y}}}}
		\red_λ
		\prefixed{x}{\abs{y}{\app{\app{x}{x}}{y}}}
		\red_λ
		\prefixed{xy}{\app{\app{x}{x}}{y}}
		\red_{@_1}
		\prefixed{xy}{y}
	}
\end{example}

\begin{para}[nameless representation for λ-terms]
Ultimately, the \stRegCRS rewriting system defines nameless representations for λ-terms related to the de-Bruijn notation for λ-terms. Considering the de-Bruijn representation of the above term \m{\abs{}{\abs{}{\app{\app{\S(\0)}{\S(\0)}}{\0}}}} we find that the position of the \m{\red_\S}-steps indeed coincides with the position of the \S markers. However, the rewrite system \stRegCRS permits more flexibility for the placement of \m{\red_\S}-steps.
\end{para}

\begin{example}
	For example the path from \cref{ex:delimiter_x} to the second occurrence of \x can also be witnessed by the following rewrite sequence:
	\m{
		\prefixed{}{\abs{x}{\abs{y}{\app{\app{x}{x}}{y}}}}
		\m{\red_λ}
		\prefixed{x}{\abs{y}{\app{\app{x}{x}}{y}}}
		\m{\red_λ}
		\prefixed{xy}{\app{\app{x}{x}}{y}}
		\red_{@_0}
		\prefixed{xy}{\app{x}{x}}
		\m{\red_\S}
		\prefixed{x}{\app{x}{x}}
		\red_{@_1}
		\prefixed{x}{x}
	}.
	Here the \extscope of \y is closed earlier. This would correspond to \m{\abs{}{\abs{}{\app{\S(\app{\0}{\0})}{\0}}}} in de-Bruijn notation, or more precisely, in a variant of the de-Bruijn notation which permits the scope/\extscope-delimiter \S to occur anywhere between a variable occurrence and its binding abstraction.
\end{example}

\begin{para}[scope delimiters in the literature]\label{S}
	This variation of the de-Bruijn notation \cite{de-bruijn:72}, in which \S-symbols can be used anywhere in the term signify the end of a scope, is due to Paterson \cite{bird:pat:1999}. The idea is also used in \cite{oost:looi:zwit:2004} and related to an even more flexible end-of-scope symbol \adbmal \cite{hend:oost:2003}.
\end{para}

\begin{definition}[CRS terms with abstraction prefixes]\label{def:sig:lambdaprefixcal:CRS}
	The CRS signature for \lambdaprefixcal, the λ-calculus with abstraction prefixes, extends the CRS signature \sigCRS for \lambdacal (see \cref{def:sigs:lambdacal:lambdaletrec:CRS}) and consists of the set \m{\sigCRSPrefixed = \sigCRS ∪ \setcompr{ \sPrefixedCRS{n}}{n ∈ ℕ}} of function symbols, where 
	for \m{n ∈ ℕ} the function symbols \sPrefixedCRS{n} for \emph{prefix λ-abstractions} of length \n are unary (have arity one). CRS terms with leading prefixes \m{\prefixedCRS{n}{x_1}{\dots \CRSabs{x_n}{\M}}} will informally be denoted by \prefixed{x_1 \dots x_n}{\M}, abbreviated as \prefixed{\vec{x}}{\M}.
\end{definition}

\begin{definition}[the CRSs \RegzeroCRS, \RegCRS, \stRegCRS for decomposing λ-terms]\label{def:RegCRS:stRegCRS}
	Consider the following CRS rule schemes over \sigCRSPrefixed:
	\newcommand\rulename[1]{&\decrule{#1}\hspace*{-2ex}&:~~&}
	\newcommand\lhs[1]{#1 \red}
	\newcommand\rhs[1]{#1 \\}
	\newcommand\rhsBr[1]{\\ &&& \indent #1 \\}
	\begin{align*}
		\rulename{@_0}
			\lhs{\prefixedCRS{n}{x_1 \dots x_n}{\appCRS{Z_0(\vec{x})}{Z_1(\vec{x})}}}
			\rhs{\prefixedCRS{n}{x_1 \dots x_n}{Z_0(\vec{x})}}
		\rulename{@_1}
			\lhs{\prefixedCRS{n}{x_1 \dots x_n}{\appCRS{Z_0(\vec{x})}{Z_1(\vec{x})}}}
			\rhs{\prefixedCRS{n}{x_1 \dots x_n}{Z_1(\vec{x})}}
		\rulename{λ}
			\lhs{\prefixedCRS{n}{x_1 \dots x_n}{\absCRS{x_{n+1}}{Z(\vec{x})}}}
			\rhs{\prefixedCRS{n+1}{x_1 \dots x_{n+1}}{Z(\vec{x})}}
		\rulename{\S}
			\lhs{\prefixedCRS{n+1}{x_1 \dots x_{n+1}}{Z(x_1,\dots,x_n)}}
			\rhs{\prefixedCRS{n}{x_1 \dots x_n}{Z(x_1,\dots,x_n)}}
		\rulename{\del}
			\lhs{\prefixedCRS{n+1}{x_1 \dots x_{n+1}}{Z(x_1,\dots,x_{i-1},x_{i+1},\dots,x_{n+1})}}
			\rhsBr{\prefixedCRS{n}{x_1 \dots x_{i-1}x_{i+1}\dots x_n}{Z(x_1,\dots,x_{i-1},x_{i+1},\dots,x_{n+1})}}
	\end{align*}
	By \RegzeroCRS we denote the CRS with rules \decrule{@_0}, \decrule{@_1}, and \decrule{λ}.
	By \RegCRS (by \stRegCRS) we denote the CRS consisting of all of the above rules
	\emph{except} the rule \decrule{\S} (\emph{except} the rule \decrule{\del}).
	The rewrite relations of \RegzeroCRS, \RegCRS, and \stRegCRS are denoted by \m{\red_\regzero}, \m{\red_\reg} and \m{\red_\streg}, respectively.
	And by \m{\red_{@_0}}, \m{\red_{@_1}},
	\m{\red_λ}, \m{\red_\S}, \m{\red_\del},
	we respectively denote the rewrite relations
	induced by each of the single rules
	\decrule{@_0}, \decrule{@_1}, \decrule{λ},
	\decrule{\S}, and \decrule{\del}. 
\end{definition}

\begin{para}[\RegzeroCRS, \RegCRS, and \stRegCRS in informal notation]
	\newcommand\rulename[1]{&\decrule{#1}\hspace*{-7ex}&:~~&}
	\newcommand\lhs[1]{#1 \red}
	\newcommand\rhs[1]{#1 \\}
	\newcommand\annotation[1]{&&& \indent \sideCondition{#1} \\}
	For better readability we will from now on rely on the informal notation corresponding to \cref{def:RegCRS:stRegCRS} which is as follows:
	\begin{align*}
		\rulename{@_0}
			\lhs{\prefixed{x_1 \dots x_n}{\app{\M_0}{\M_1}}}
			\rhs{\prefixed{x_1 \dots x_n}{\M_0}}
		\rulename{@_1}
			\lhs{\prefixed{x_1 \dots x_n}{\app{\M_0}{\M_1}}}
			\rhs{\prefixed{x_1 \dots x_n}{\M_1}}
		\rulename{λ}
			\lhs{\prefixed{x_1 \dots x_n}{\abs{x_{n+1}}{\M_0}}}
			\rhs{\prefixed{x_1 \dots x_{n+1}}{\M_0}}
		\rulename{\S}
			\lhs{\prefixed{x_1 \dots x_{n+1}}{\M_0}}
			\rhs{\prefixed{x_1 \dots x_n}{\M_0}}
			\annotation{if the binding \m{λx_{n+1}} is vacuous}
		\rulename{\del}
			\lhs{\prefixed{x_1 \dots x_{n+1}}{\M_0}}
			\rhs{\prefixed{x_1 \dots x_{i-1}x_{i+1}\dots x_{n+1}}{\M_0}}
			\annotation{if the binding \m{λx_i} is vacuous}
	\end{align*}
\end{para}

\begin{remark}
	Be reminded that as mentioned in \cref{notation:rules} we will at times omit the \decompose and write \rulep{\arule} instead of \decrule{\arule}.
\end{remark}

\begin{figure}
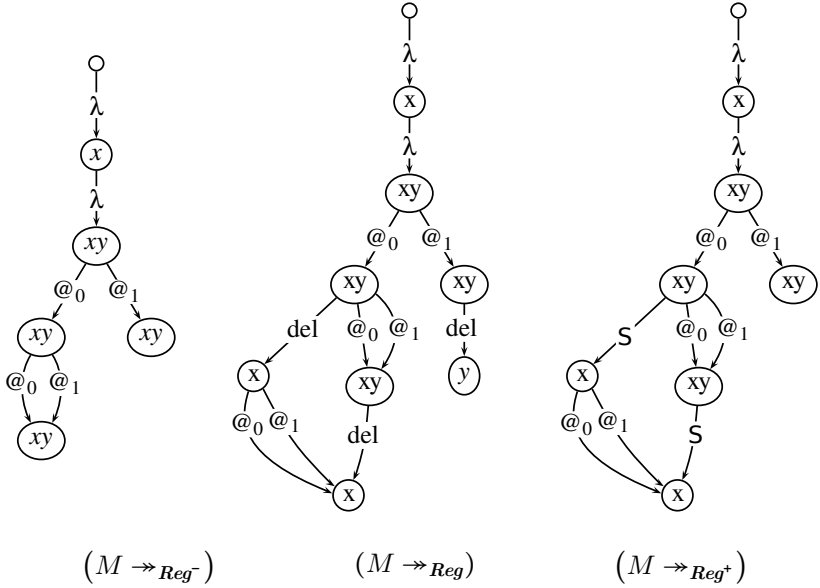

	\begin{hspread}
		\vcentered{\fig{lxy-xxy-regmin-ltg-l}} &
		\hspace{-1em}\vcentered{\fig{lxy-xxy-cprgraph-ltg-l}} &
		\hspace{-1em}\vcentered{\fig{lxy-xxy-redgraph-ltg-l}}
		\\
		\vcentered{\GeneratedSubARSi{\RegzeroCRS}{M}} &
		\vcentered{\GeneratedSubARSi{\RegCRS}{M}} &
		\vcentered{\GeneratedSubARSi{\stRegCRS}{M}}
	\end{hspread}
	\caption{The reduction graphs of \m{M = \prefixed{}{\abs{x}{\abs{y}{\app{\app{x}{x}}{y}}}}} with respect to \RegzeroCRS/\RegCRS/\stRegCRS. Note, that in the vertices not the entire \lambdaprefixcal-terms are displayed but only the respective prefixes.}
	\label{fig:three_redgraphs}
\end{figure}

\begin{para}[\stRegCRS defines nameless representations]
	Considering the reduction graphs from \cref{fig:three_redgraphs} without labels on their nodes we see that only from the \stRegCRS graph the original term can be reconstructed unambiguously. For example, the path to the rightmost occurrence of \x has the rewrite sequence in \RegCRS \[\red_λ . \red_λ . \red_{@_0} . \red_\del . \red_{@_1}\] which witnesses an occurrence of \y in place of \x at the same position. This ambiguity plays a role for the definition of λ-transition-graphs in \cref{sec:lambda-transition-graphs}, and is discussed in that context in \cref{rem:nameless-repr}.
\end{para}

\begin{para}[\m{\red_\streg ~\subseteq~ \red_\reg ~\subseteq~ \red_\regzero}]
	Note that \m{\red_\reg} is contained in \m{\red_\regzero}, and since the rule \rulep{\del} generalises the rule \rulep{\S}, \m{\red_\streg} is contained in \m{\red_\reg}.
\end{para}

We will specifically need the following statement later on:
\begin{proposition}\label{prop:rewseqs:stRegCRS:2:RegCRS}
	Every rewrite sequence in \stRegCRS corresponds directly to a rewrite sequence in \RegCRS by exchanging \m{\red_\S}-steps with \m{\red_\del}-steps.
\end{proposition}

\begin{para}[abstract prefix symbol]
	We are only interested in terms which have a single occurrence of the abstraction prefix symbol at the outermost position. Note that the rules in \RegCRS and \stRegCRS guarantee that every reduct of a term of the form \prefixed{\vec{x}}{\M} is again a term of this form. Therefore we define:
\end{para}

\begin{definition}[prefixed λ-terms]\label{def:CRSterms}
	By \iTer{\lambdaprefixcal} we denote the subset of closed iCRS terms over \sigCRSPrefixed with the restrictions:
	\begin{itemize}
		\item Every term \m{\M ∈ \iTer{\lambdaprefixcal}} has a prefix at its root and nowhere else: \M is of the form \m{\prefixedCRS{n}{x_1}{\dots \CRSabs{x_n}\M}} and \M has no occurrences of function symbols \sPrefixedCRS{i} for any \m{i ∈ ℕ}.
		\item Otherwise a CRS abstraction can only occur directly beneath an \sabsCRS-symbol.
	\end{itemize}
	\iTer{\lambdaprefixcal} is more formally specified in \cref{def:Ter-lambdaletrec}.
\end{definition}

\begin{proposition}\label{prop:oneRegARS}
	\iTer{\lambdaprefixcal} is closed under \m{\red_\regzero}, \m{\red_\reg}, and \m{\red_\streg}.
\end{proposition}

\begin{definition}[the ARSs \RegzeroARS, \RegARS, \stRegARS]\label{RegARS:stRegARS:oneRegARS:onestRegARS}
	We denote by \RegzeroARS, \RegARS and \stRegARS the ARSs induced by the iCRSs derived from \RegzeroCRS, \RegCRS, \stRegCRS, restricted to terms in \iTer{\lambdaprefixcal}.
\end{definition}

\begin{proposition}\label{prop:rewprops:RegCRS:stRegCRS}
	The restrictions of the rewrite relations of \RegzeroARS, \RegARS and \stRegARS to \iTer{\lambdaprefixcal}, the set of objects of \RegzeroARS, \RegARS, and \stRegARS, have the following properties:
	\begin{enumerate}[(i)]
		\item\label{prop:rewprops:RegCRS:stRegCRS:i} \m{\red_\del} is confluent, and terminating.
		\item\label{prop:rewprops:RegCRS:stRegCRS:ii} \m{\red_\del} one-step commutes with \m{\red_λ}, \m{\red_{@_0}}, \m{\red_{@_1}}, \m{\red_\S}:
			\[\convred_\del ⋅ \red_\arule ~\subseteq~ \red_\arule ⋅ \convred_\del ~~~ ∀\arule∈\set{λ,@_0,@_1,\S}\]
		\item\label{prop:rewprops:RegCRS:stRegCRS:ii-1} \m{\red_\del} can be postponed:
			\[\red_\del ⋅ \red_\arule ~\subseteq~ \red_\arule ⋅ \red_\del ~~~ ∀\arule∈\set{λ,@_0,@_1,\S}\]
		\item\label{prop:rewprops:RegCRS:stRegCRS:iii-0} \m{\red_\S ~\subseteq~ \m{\red_\del}} and thus \m{\red_\streg ~\subseteq~ \m{\red_\reg}}. Furthermore \m{\red_\regzero ~\subset~ \m{\red_\streg} ~\subset~ \m{\red_\reg}}.
		\item\label{prop:rewprops:RegCRS:stRegCRS:iii} \m{\red_\S} is deterministic, hence confluent, and terminating.
		\item\label{prop:rewprops:RegCRS:stRegCRS:iv} \m{\red_\S} one-step commutes with \m{\red_{@_0}}, and \m{\red_{@_1}}:
			\[\convred_\S ⋅ \red_{@_i} ~\subseteq~ \red_{@_i} ⋅ \convred_\S ~~~ ∀i ∈ \set{0,1}\]
		\item\label{prop:rewprops:RegCRS:stRegCRS:iv-1} \prefixed{x}{x} is the sole term in \m{\red_\reg}-normal-form. \m{\red_\streg}-normal-forms are of the form \prefixed{x_1 \dots x_n}{x_n}.
		\item\label{prop:rewprops:RegCRS:stRegCRS:v} \m{\red_\reg} and \m{\red_\streg} are finitely branching, and -- on finite terms -- terminating.
	\end{enumerate}
\end{proposition}

\begin{proof}
	Most properties, including those concerning commutation of steps, are easy to verify by analysing the behaviour of the rewrite rules in \RegARS on terms of \iTer{\lambdaprefixcal}.
	\par For \m{\red_\streg \subset \red_\reg} in \cref{prop:rewprops:RegCRS:stRegCRS:iii-0} note that, for example, \prefixed{xy}{y} \m{\red_\reg} \prefixed{y}{y} by a \m{\red_\del}-step, but that \prefixed{xy}{y} is a \m{\red_\S}-normal-form, and hence also a \m{\red_\streg}-normal-form.
	\par Concerning \cref{prop:rewprops:RegCRS:stRegCRS:v} we first argue for finite branchingness of \m{\red_\reg} and \m{\red_\streg} on \iTer{\lambdacal}: this property follows from the fact that, on a term \prefixed{\vec{x}}{\M} with just one abstraction in its prefix, of the constituent rewrite relations \m{\red_{@_0}}, \m{\red_{@_1}}, \m{\red_λ}, \m{\red_\S}, \m{\red_\del} of \m{\red_\reg} and \m{\red_\streg} only \m{\red_\del} can have branching degree greater than one, which in this case then also is bounded by the length \length{\vec{x}} of the abstraction prefix.
	For termination of \m{\red_\reg} and \m{\red_\streg} on finite terms with just a leading abstraction prefix we can restrict to \m{\red_\reg}, due to \cref{prop:rewprops:RegCRS:stRegCRS:iii-0}, and argue as follows: on finite terms in \iTer{\lambdaprefixcal}, in every \m{\red_\reg}-rewrite-step either the size of the body of the term decreases strictly, or the size of the body stays the same, but the length of the prefix decreases by one. Hence in every rewrite step the measure \m{\pair{\text{body size}}{\text{prefix length}}} on terms decreases strictly in the (well-founded) lexicographic ordering on \m{ℕ\timesℕ}.
\end{proof}


\begin{corollary}[\m{\red_\del} and \m{\red_\S} are normalising]
	Note that as a consequence of \cref{prop:rewprops:RegCRS:stRegCRS}~\cref{prop:rewprops:RegCRS:stRegCRS:i} and \cref{prop:rewprops:RegCRS:stRegCRS:iii}, the rewrite relations \m{\red_\del} and \m{\red_\S} are normalising on \iTer{\lambdacal}.
\end{corollary}

\begin{proposition}\label{prop:compress:prefix:RegARS}
	The following statements hold:
	\begin{enumerate}[(i)]
		\item\label{prop:compress:prefix:RegARS:i} Let \prefixed{\vec{x}}{\M} a term in \RegARS with \m{\length{\vec{x}} = n ∈ ℕ}. Then the number of terms \prefixed{\vec{y}}{\N} in \RegARS with \prefixed{\vec{y}}{\N} \m{\mred_\del} \prefixed{\vec{x}}{\M} and \m{\length{\vec{y}} = n+k ∈ ℕ} is \m{\binom{n+k}{n}}.
		\item\label{prop:compress:prefix:RegARS:ii} Let \m{T} be a finite set of terms in \RegARS, and \m{k ∈ ℕ}. Then also the set of terms in \RegARS that are the form \prefixed{\vec{y}}{\N} with \m{\length{\vec{y}} ≤ k} and that have a \m{\mred_\del}-reduct in \m{T} is finite.
	\end{enumerate}
\end{proposition}

\begin{proof}
	From \m{\prefixed{\vec{y}}{{\M}(y)} = \prefixed{y_1 \dots y_{n+k}}{{\N}(y_1,\dots,y_{n+k})} \mred_\del \prefixed{x_1 \dots x_n}{{\M}(x_1,\dots,x_n)} = \prefixed{\vec{x}}{{\M}(\vec{x})}}
	it follows that there are
	\m{i_1,\dots,i_n ∈ \set{1,\dots,n+k}} with \m{i_1 < i_2 < \dots < i_n}
	such that the term \prefixed{\vec{y}}{{\M}(y)} is actually of the form
	\prefixed{y_1 \dots y_{n+k}}{{\N}(y_{i_1},\dots,y_{i_n})}
	and furthermore
	\m{\prefixed{y_{i_1}\dots y_{i_n}}{{\N}(y_{i_1},\dots,y_{i_n})} = \prefixed{x_1 \dots x_n}{{\M}(x_1,\dots,x_n)}}.
	Hence the number of terms
	\prefixed{\vec{y}}{{\M}(y)}
	with 
	\m{\mred_\del}-reduct
	\prefixed{\vec{x}}{{\M}(\vec{x})}
	is equal to the number of choices \m{i_1,\dots,i_n ∈ \set{1,\dots,n+k}} such that \m{i_1 < i_2 < \dots < i_n}.
	This establishes statement \cref{prop:compress:prefix:RegARS:i}.
	Statement \cref{prop:compress:prefix:RegARS:ii} is an easy consequence.
\end{proof}

\begin{lemma}\label{lem:rewprops:projection:RegCRS:stRegCRS}
	On \iTer{\lambdaprefixcal}, the rewrite relations \m{\red_\reg} and \m{\red_\streg} have the following further properties with respect to \m{\red_\del}, \m{\red_\S}, \m{\nfred_\del}, and \m{\nfred_\S}:
	\begin{align*}
		{\m{\convmred_\del} ⋅ \m{\red_\reg}}
		\; & \subseteq\;
		{(\m{\nfred_\del} ⋅ \m{\eqred_\regzero}) ⋅ \m{\convmred_\del}}
		\\
		{\m{\convmred_\del} ⋅ \m{\mred_\reg}}
		\; & \subseteq\;
		{(\m{\nfred_\del} ⋅ \m{\eqred_\regzero})^* ⋅ \m{\convmred_\del}}
		\\[2ex]
		{\m{\convmred_\del} ⋅ \m{\red_\streg}}
		\; & \subseteq\;
		{(\m{\nfred_\S} ⋅ \m{\eqred_\regzero}) ⋅ \m{\convmred_\del}}
		\\
		{\m{\convmred_\del} ⋅ \m{\mred_\streg}}
		\; & \subseteq\;
		{(\m{\nfred_\S} ⋅ \m{\eqred_\regzero})^* ⋅ \m{\convmred_\del}}
	\end{align*}
\end{lemma}

\begin{proof}
	These commutation properties, which can be viewed as projection properties, can be shown by arguments with diagrams using the commutation properties in \cref{prop:rewprops:RegCRS:stRegCRS}.
\end{proof}

\begin{remark}
	The commutation properties in \cref{lem:rewprops:projection:RegCRS:stRegCRS} can be refined to state that \m{\red_λ}-steps project to \m{\red_λ}-steps, and \m{\red_{@_0}}-steps and \m{\red_{@_1}}-steps project to \m{\red_{@_0}}-steps and \m{\red_{@_1}}-steps, accordingly.
\end{remark}

As an immediate consequence of \cref{lem:rewprops:projection:RegCRS:stRegCRS} we obtain the following lemma, which formulates a connection via projection between rewrite sequences in \RegARS (in \stRegARS) and \emph{\m{\red_\del}-eager} (\emph{\m{\red_\S}-eager}) rewrite sequences in \RegARS (in \stRegARS) that do not contain \m{\red_λ}-steps, \m{\red_{@_0}}-steps, or \m{\red_{@_1}}-steps on terms that allow \m{\red_\del}-steps (\m{\red_\S}-steps).

\begin{lemma}\label{lem:projection:RegCRS:stRegCRS}
	The following statements hold:
	\begin{enumerate}[(i)]
		\item\label{lem:projection:RegCRS:stRegCRS:reg}
			Every (finite or infinite) rewrite sequence in \RegARS of the form \[\arewseq : \prefixed{\vec{x}_0}{\M_0} \m{\red_\reg} \prefixed{\vec{x}_1}{\M_1} \m{\red_\reg} \dots \m{\red_\reg} \prefixed{\vec{x}_k}{\M_k} \m{\red_\reg} \dots\] projects over a rewrite sequence \m{\crewseq : \prefixed{\vec{x}_0}{\M_0} \m{\mred_\del} \prefixed{\vec{x}'_0}{\M_0}} to a \m{\red_\del}-eager rewrite sequence in \RegARS of the form
			\begin{align*}
				\check{\arewseq} : \prefixed{\vec{x}'_0}{\M_0} & \m{\nfred_\del} ⋅ \m{\eqred_\regzero} \prefixed{\vec{x}'_1}{\M_1} \m{\nfred_\del} ⋅ \m{\eqred_\regzero} ~\dots \\
				\dots~ & \m{\nfred_\del} ⋅ \m{\eqred_\regzero} \prefixed{\vec{x}'_k}{\M_k} \m{\nfred_\del} ⋅ \m{\eqred_\regzero} ~\dots
			\end{align*}
			in the sense that \prefixed{\vec{x}_i}{\M_i} \m{\mred_\del} \prefixed{\vec{x}'_i}{\M_i}\, for all \m{i ∈ ℕ} less or equal to the length of \arewseq.
		\item\label{lem:projection:RegCRS:stRegCRS:streg}
			Every (finite or infinite) rewrite sequence in \stRegARS of the form \[\arewseq : \prefixed{\vec{y}_0}{\N_0} \m{\red_\streg} \prefixed{\vec{y}_1}{\N_1} \m{\red_\streg} \dots \m{\red_\streg} \prefixed{\vec{y}_k}{\N_k} \m{\red_\streg} \dots\] projects over a rewrite sequence \m{\crewseq : \prefixed{\vec{y}_0}{\N_0} \m{\mred_\del} \prefixed{\vec{y}'_0}{\N_0}} to a \m{\red_\S}-eager rewrite sequence in \stRegARS:
			\begin{align*}
				\check{\arewseq} : \prefixed{\vec{y}'_0}{\N_0} & \m{\nfred_\S} ⋅ \m{\eqred_\regzero} \prefixed{\vec{y}'_1}{\N_1} \m{\nfred_\S} ⋅ \m{\eqred_\regzero} ~\dots \\
				\dots~ & \m{\nfred_\S} ⋅ \m{\eqred_\regzero} \prefixed{\vec{y}'_k}{\N_k} \m{\nfred_\S} ⋅ \m{\eqred_\regzero} ~\dots
			\end{align*}
			in the sense that \prefixed{\vec{y}_i}{\N_i} \m{\mred_\S} \prefixed{\vec{y}'_i}{\N_i} for all \m{i ∈ ℕ} less or equal to the length of \arewseq.
	\end{enumerate}
\end{lemma}

\begin{para}[non-determinism of \m{\red_\reg} and \m{\red_\streg}]\label{rem:nondeterminism_of_regred_and_stregred}
	On terms in \iTer{\lambdaprefixcal}, which have just one prefix at the top of the term, there are two different causes for non-determinism of the rewrite relations in \Reg and \stReg: First, since the left-hand sides of the rules \rulep{@_0} and \rulep{@_1} coincide, these rules enable different steps on the same term, producing the left- and respectively the right subterm of the application immediately below the prefix. Second, the rules \rulep{\del} and \rulep{\S} can be applicable in situations where also one of the rules \rulep{@_0}, \rulep{@_1}, or \rulep{λ} is applicable (see for instance \cref{fig:three_redgraphs}, in the middle). Whereas the first kind of non-determinism is due to the `observer' having to observe the two different subterms of an application in a λ-term, the second is due to a freedom of the observer as to when to attest the end of a scope/\extscope in the analysed λ-term.
\end{para}

\begin{para}[outlook: λ-transition-graphs]
	In the definition below we define strategies for \RegARS and \stRegARS that only allow the former source of non-determinism while forbidding the second kind. The intention is that the reduction graph of a term with respect to such a strategy corresponds to the term's syntax tree. These reduction graphs can be seen as a nameless graph representation for λ-terms. We will introduce and study such representations (called λ-transition-graphs) in \cref{sec:lambda-transition-graphs} and a further adaptation (λ-term-graphs) in \cref{chap:representations}.
\end{para}

\begin{definition}[scope/\extscope-delimiting strategy]\label{def:scope:delimiting:strat:Reg:stReg}
	We call a strategy \astrat for \RegARS (for \stRegARS) a \emph{scope-delimiting strategy} (\emph{\extscope-delimiting strategy}) if the source of a step is non-deterministic (that is, it is the source of more than one step) if and only if it is the source of precisely a \m{\red_{@_0}}-step and a \m{\red_{@_1}}-step.
\end{definition}

\begin{para}[alternative formulation of \cref{def:scope:delimiting:strat:Reg:stReg}]
	We can define a bit more verbosely: a strategy \astrat for \RegARS (for \stRegARS) is called a scope-delimiting (\extscope-delimiting) strategy if
	\begin{itemize}
		\item every source of a step is one of three kinds: the source of a \m{\red_λ}-step, the source of a \m{\red_\del}-step (a \m{\red_\S}-step), or the source of both a \m{\red_{@_0}}-step and a \m{\red_{@_1}}-step with the restriction that (in all three cases) it is not the source of any other step.
	\end{itemize}
	Heeding the fact that sources of \m{\red_λ}-steps are never sources of \m{\red_{@_i}}-steps, and vice versa, this condition can be relaxed to:
	\begin{itemize}
		\item no source of a \m{\red_λ}-step or a \m{\red_{@_i}}-step for \m{i ∈ \set{0,1}} is also the source of a \m{\red_\del}-step (\m{\red_\S}-step), and every source of a \m{\red_{@_i}}-step for \m{i ∈ \set{0,1}} is the source of both a \m{\red_{@_0}}-step and a \m{\red_{@_1}}-step, but not of any other step.
	\end{itemize}
\end{para}

\begin{definition}[eager scope/\extscope-delimiting strategy]\label{def:eagerStrats}
	The \emph{eager scope-delimiting strategy\/ \eagStrat for \RegARS} is defined as the restriction of rewrite steps in \RegARS to eager application of the rule \rulep{\del}: applications rules other than \rulep{\del} are only allowed if \rulep{\del} is not applicable. Analogously, the \emph{eager \extscope-delimiting strategy\/ \eagStratPlus for \stRegARS} is defined as the restriction of rewrite steps in \stRegARS to eager application of the rule \rulep{\S}.
\end{definition}

\begin{definition}[lazy scope/\extscope-delimiting strategy]\label{def:lazyStrats}
	The \emph{lazy scope-delimiting strategy \lazyStrat for \RegARS} is defined as the restriction of rewrite steps in \RegARS to lazy application of the rule \rulep{\del}: applications of \rulep{\del} are only allowed when other rules are not applicable. Analogously, the \emph{lazy \extscope-delimiting strategy \lazyStratPlus for \stRegARS} is defined as the restriction of rewrite steps in \stRegARS to lazy application of the rule \rulep{\S}.
\end{definition}

\begin{figure}
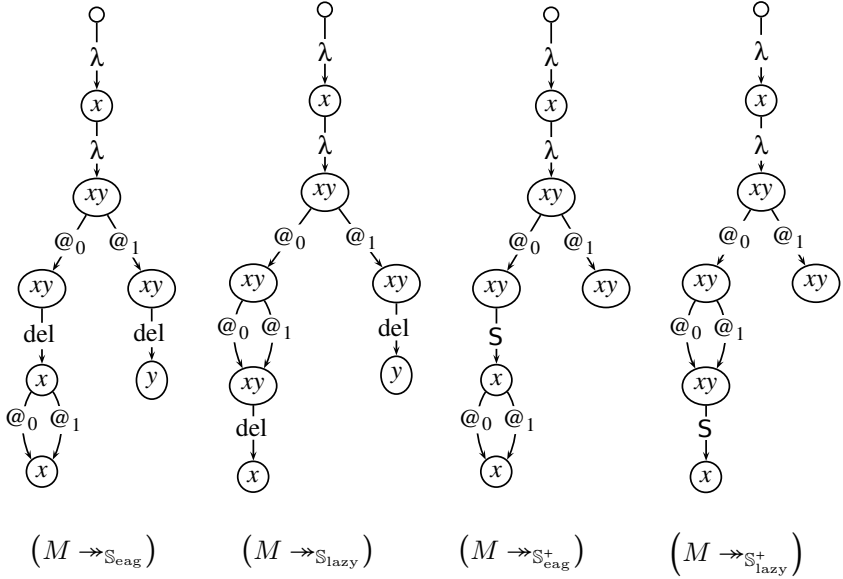

	\begin{hspread}
		\vcentered{\fig{lxy-xxy-cpreager-ltg-l}} &
		\hspace{-0.5em}\vcentered{\fig{lxy-xxy-cprlazy-ltg-l}} &
		\hspace{-0.5em}\vcentered{\fig{lxy-xxy-eager-ltg-l}} &
		\hspace{-0.5em}\vcentered{\fig{lxy-xxy-lazy-ltg-l}} \\
		\vcentered{\GeneratedSubARSi{\eagStrat}{M}} &
		\vcentered{\GeneratedSubARSi{\lazyStrat}{M}} &
		\vcentered{\GeneratedSubARSi{\eagStratPlus}{M}} &
		\vcentered{\GeneratedSubARSi{\lazyStratPlus}{M}}
	\end{hspread}
	\caption{Reduction graphs of \m{M = \abs{x}{\abs{y}{\app{\app{x}{x}}{y}}}} with respect to different \RegARS and \stRegARS strategies; compare: \cref{fig:three_redgraphs}. Again note, that the labels do only show the prefixes associated with each term.}
	\label{fig:four-strategies}
\end{figure}

\begin{para}[history-aware versus history-free scope-delimiting strategies]
	The history-free strategy obtained by projection from a history-aware scope-delimiting strategy is not in general a scope-delimiting strategy. This is due to the non-determinism which may be introduced by the projection. Consider for example the term \app{M}{M} with \m{M=\abs{x}{\abs{y}{\app{\app{x}{x}}{y}}}} and the history-aware strategy constructed by using \eagStrat on the left component and \lazyStrat on the right component of \app{M}{M}. See \cref{fig:MM_leager_rlazy} for the reduction graph. The reduction graph w.r.t.\ the history-free strategy obtained by projection, however, bears the kind of non-determinism that is not permitted for a scope-delimiting strategy, in the form of a vertex that is the source of both a \m{\red_\del}- and a \m{\red_{@_i}}-step.
\end{para}

\begin{figure}
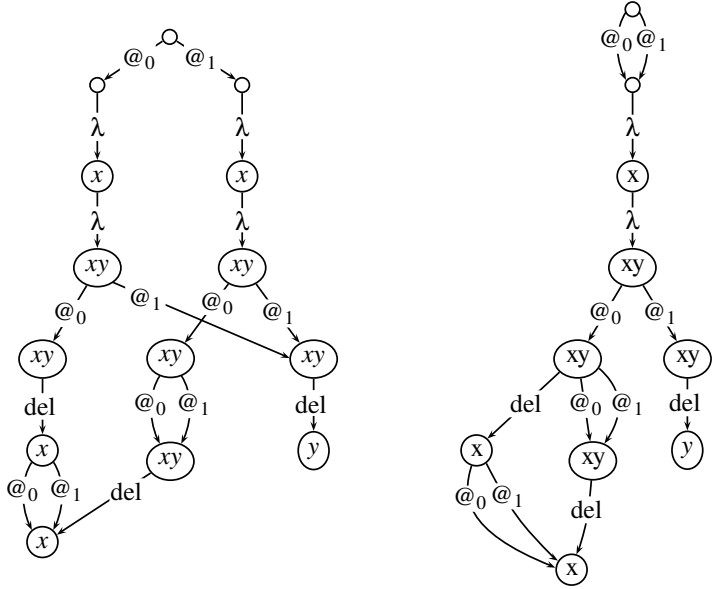

	\begin{hspread}
		\vcentered{\fig{MM-leager-rlazy-ltg-l}} & \vcentered{\fig{MM-leager-rlazy-proj-ltg-l}}
	\end{hspread}
	\caption{On the left: reduction graph of \app{M}{M} with \m{M = \abs{x}{\abs{y}{\app{\app{x}{x}}{y}}}} with respect to the history-aware strategy for \RegARS constructed by using \eagStrat on the left component and \lazyStrat on the right component of \app{M}{M}. On the right: reduction graph with respect to the history-free strategy obtained by projection. Note: vertex labels only show prefixes.}
	\label{fig:MM_leager_rlazy}
\end{figure}

The following proposition formulates a property of the eager scope-delimiting (\extscope-delimiting) strategy for \RegARS (in \stRegARS) that assigns it a special status: the target of every rewrite sequence with respect to \eagStrat (with respect to \eagStratPlus) can be reached, modulo some final \m{\red_\del}-steps (\m{\red_\S}-steps), also by a rewrite sequence with respect to an arbitrary scope-delimiting (\extscope-delimiting) strategy. Furthermore, rewrite sequences with respect to \eagStrat (with respect to \eagStratPlus) are able to mimic rewrite sequences with respect to an arbitrary scope-delimiting (\extscope-delimiting) strategy, up to trailing \m{\red_\del}-steps (\m{\red_\S}-steps) applied to the latter.

\begin{definition}
	Let \m{\red_1}, \m{\red_2}, \m{\red_3} be rewrite relations.
	The rewrite relation \m{\red_1} is called \emph{cofinal for} \m{\red_2} if \m{\mred_2~\subseteq {\mred_1 ⋅ \convmred_2}}.
	We say that \m{\red_1} is \emph{cofinal for \m{\red_2} with trailing \m{\red_3}-steps} if \m{\mred_2~\subseteq {\mred_1 ⋅ \convmred_3}}.
	Furthermore we say that \m{\red_1} \emph{factors into} \m{\red_2} and \m{\red_3} if \m{\red_1~\subseteq {\red_2 ⋅ \red_3}}.
\end{definition}

\begin{proposition}\label{prop:eagstrat}
	For all scope-delimiting strategies \astrat on \RegARS and all \extscope-delimiting strategies \astratplus on \stRegARS the following holds:
	\begin{enumerate}[(i)]
		\item\label{prop:eagstrat:i}
			\m{\mred_{\eagStrat}} factors into \m{\mred_{\astrat}} and \m{\mred_\del}.\\
			\m{\mred_{\eagStratPlus}} factors into \m{\mred_{\astratplus}} and \m{\mred_\S}.
		\item\label{prop:eagstrat:ii}
			\m{\red_{\eagStrat}} is cofinal for \m{\red_{\astrat}} with trailing \m{\red_\del}-steps.\\
			\m{\red_{\eagStratPlus}} is cofinal for \m{\red_{\astratplus}} with trailing \m{\red_\S}-steps.
	\end{enumerate}
\end{proposition}


\begin{figure}
	\begin{hspread}
		\vcentered{\fig{entangled-cpreager-ltg-l}} & \vcentered{\fig{pstricks/named/entangled-ltg-l}} & \vcentered{\fig{fix-eff-ltg-l}}
	\end{hspread}
	\captionsetup{singlelinecheck=off}
	\caption[]{
		Reduction graphs (with prefixes as vertex labels) of:
		\begin{description}
			\item[left:] \cref{ex:entangled} w.r.t.\ \eagStrat. The dotted line denotes node equality with the connected nodes representing identical (α-equivalent) terms. They are drawn as two separate nodes instead of one to avoid confusion as to which variable (it is \a) is removed from the prefix by the incoming \del-edge.
			\item[middle:] \cref{ex:entangled} w.r.t.\ an arbitrary \extscope-delimiting strategy for \stRegARS (and also with respect to \lazyStrat). The vertical dots denote an infinite growth of the graph.
			\item[right:] \cref{ex:expressibility:fix} w.r.t.\ an arbitrary scope/\extscope-delimiting strategy for \stRegARS/\RegARS.
		\end{description}
	}
	\label{fig:entangled-fix-ltg}
\end{figure}

\begin{definition}[generated subterms]\label{def:ST}
	Let \astrat be a scope-delimiting strategy for \RegARS/\stRegARS. For every \m{M ∈ \iTer{\lambdacal}}, the terms in the set \ST{\astrat}{M} are called the \emph{generated subterms of \M with respect to \astrat}, where \ST{\astrat}{M} is the set of objects of the generated sub-ARS \GeneratedSubARSi{\astrat}{\prefixed{}{M}}, or in other words, the set of \m{\mred_{\astrat}}-reducts of \prefixed{}{\M}:
	\begin{align*}
		\ST{\astrat}{} ~:~ \iTer{\lambdacal}
		& ~→~ \powerset(\iTer{\lambdaprefixcal}) \\
		M & ~↦~ O ~~~ \text{where \m{\tuple{O,\steps,\src{},\tgt{}} = \GeneratedSubARSi{\astrat}{\prefixed{}{M}}}}
	\end{align*}
\end{definition}

\begin{definition}[\astrat-regularity]\label{def:regularity:ARS}
	Let \aARS be an abstract rewriting system, and \astrat a strategy for \aARS. We say that an object \o in \aARS is \emph{\astrat-regular in \aARS} if the set of generated subterms \ST{\astrat}{o} of \o is finite.
\end{definition}

\begin{definition}[regular and strongly regular λ-terms]\label{def:reg:streg}
	A λ-term \M is called \emph{regular} (\emph{strongly regular}) if there exists a scope-delimiting strategy \astrat for \RegARS (a \extscope-delimiting strategy \astratplus for \stRegARS) such that \M is \astrat-regular (is \astratplus-regular).
\end{definition}

\begin{corollary}[regularity and strong regularity]
	A λ-term \M is regular (strongly regular) if and only if the set \ST{\astrat}{\M} of generated subterms of \M with respect to some scope-delimiting strategy (\extscope-delimiting strategy) \astrat is finite.
\end{corollary}

\begin{proposition}[\m{\text{finiteness} \implies \text{strong regularity} \implies \text{regularity}}]\label{prop:def:reg:streg}
	\ 
	\begin{enumerate}[(i)]
		\item\label{prop:def:reg:streg:i} Every strongly regular λ-term is also regular.
		\item\label{prop:def:reg:streg:ii} Finite λ-terms are strongly regular.
	\end{enumerate}
\end{proposition}

\begin{proof}
	For statement \cref{prop:def:reg:streg:i}, let \M be a λ-term, and let \astratplus be a \extscope-delimiting strategy for \stRegARS by which \ST{\astratplus}{\M} is finite. Due to \cref{prop:rewseqs:stRegCRS:2:RegCRS} \astratplus can be modified to yield a scope-delimiting strategy \astrat for \RegARS by replacing \m{\red_\S}-steps with \m{\red_\del}-steps. Then there is a stepwise correspondence between \m{\red_{\astratplus}}-rewrite-sequences and \m{\red_{\astrat}}-rewrite-sequences that pass through the same terms. 
	Consequently, the sets of \m{\mred_{\astratplus}}-reducts and \m{\mred_{\astrat}}-reducts of \prefixed{}{\M} coincide: \m{\ST{\astrat}{\M} = \ST{\astratplus}{\M}}. It follows that \ST{\astrat}{\M} is finite.
	\par For statement \cref{prop:def:reg:streg:ii}, note that by \cref{prop:rewprops:RegCRS:stRegCRS}~\cref{prop:rewprops:RegCRS:stRegCRS:v}, and Kőnig's Lemma every finite term in \iTer{\lambdaprefixcal} has only finitely many reducts with respect to \m{\mred_\reg} and \m{\mred_\streg}. It follows that for every finite λ-term \M and every scope-delimiting strategy \astrat on \RegARS, or on \stRegARS, the number of \m{\mred_{\astrat}}-reducts of \prefixed{}{\M} is finite.
\end{proof}

\begin{proposition}[eager scope-closure and regularity]\label{prop:eager:strat:in:def:reg:streg}
	For all λ-terms \M the following statements hold:
	\begin{enumerate}[(i)]
		\item\label{prop:eager:strat:in:def:reg:streg:i} \M is regular if and only if \M is \eagStrat-regular.
		\item\label{prop:eager:strat:in:def:reg:streg:ii} \M is strongly regular if and only if \M is \eagStratPlus-regular.
	\end{enumerate}
\end{proposition}

\begin{proof}
	We only prove \cref{prop:eager:strat:in:def:reg:streg:i}, because \cref{prop:eager:strat:in:def:reg:streg:ii} can be established analogously.
	The implication ``⇐'' follows from the definition of regularity.
	For ``⇒'', let \M be a regular λ-term. Then there exists a scope-delimiting strategy \astrat so that \ST{\astrat}{\M} is finite. Since by \cref{prop:eagstrat}~\cref{prop:eagstrat:i} every \m{\mred_{\eagStrat}}-rewrite-sequence factors into an \m{(\mred_{\astrat} ⋅ \m{\mred_\del})}-rewrite-sequence, it follows that every term in \ST{\eagStrat}{\M} is the \m{\mred_\del}-reduct of a term in \ST{\astrat}{\M}. As every term in \iTer{\lambdaprefixcal} has only finitely many \m{\mred_\del}-reducts, it follows that also \ST{\eagStrat}{\M} is finite.
\end{proof}

\begin{example}[regular and strongly regular terms]\label{ex:reg-and-streg-terms}
	The following examples demonstrate the connection between (strong) regularity and, as illustrated in \cref{fig:entangled-fix-ltg}, the finiteness of the ARSs generated by different \RegARS/\stRegARS strategies.
	\begin{itemize}
		\item \cref{ex:entangled} is regular but not strongly regular.
		\item \cref{ex:expressibility:fix} is strongly regular.
	\end{itemize}
\end{example}

To further illustrate the statements made in \cref{ex:reg-and-streg-terms} let us consider various \RegARS and \stRegARS rewrite sequences corresponding to infinite paths through the terms.

\begin{example}\label{ex:entangled-infinite-path}
	\newfunsymbol\recM{rec_{\mathit{M}}}
	For the term \M from \cref{ex:entangled}, we first introduce a finite CRS based notation, as a `higher-order recursive program scheme'. We can represent \M by \m{\abs{a}{\recM(a)}} together with the CRS rule \m{\recM(X) \red_\recM \abs{x}{\app{\recM(x)}{X}}}. It holds that \m{\abs{a}{\recM(a)} \infred_\recM \M}. Using this notation we can finitely describe the infinite path down the spine of \cref{ex:entangled} by the cyclic \eagStrat-rewrite-sequence:
	\[\begin{array}{ll}
		                & \prefixed{}{\abs{a}{\recM(a)}}
		\\\m{\red_λ}    & \prefixed{a}{\recM(a)}
		\\\red_\recM    & \prefixed{a}{\abs{b}{\app{\recM(b)}a}}
		\\\m{\red_λ}    & \prefixed{ab}{\app{\recM(b)}a}
		\\\red_{@_0}    & \prefixed{ab}{\recM(b)}
		\\\m{\red_\del} & \prefixed{b}{\recM(b)}
		\\=             & \prefixed{a}{\recM(a)}
	\end{array}\]
	In \stRegARS the rewriting sequence for the same path is invariant over all
	\extscope-delimiting strategies and necessarily infinite:
	\[\hspace{-0.5em}\begin{array}{llll}
		\prefixed{}{\abs{a}{\recM(a)}}
		& \hspace{-0.6em}\m{\red_λ}    \prefixed{a}{\recM(a)}
		& \hspace{-0.6em}\red_\recM \\ \prefixed{a}{\abs{b}{\app{\recM(b)}a}}
		& \hspace{-0.6em}\m{\red_λ}    \prefixed{ab}{\app{\recM(b)}a}
		& \hspace{-0.6em}\red_{@_0}    \prefixed{ab}{\recM(b)}
		& \hspace{-0.6em}\red_\recM \\ \prefixed{ab}{\abs{c}{\app{\recM(c)}b}}
		& \hspace{-0.6em}\m{\red_λ}    \prefixed{abc}{\app{\recM(c)}b}
		& \hspace{-0.6em}\red_{@_0}    \prefixed{abc}{\recM(c)}
		& \hspace{-0.6em}\red_\recM \\ \prefixed{abc}{\abs{d}{\app{\recM(d)}c}}
		& \hspace{-0.6em}\m{\red_λ}    \prefixed{abcd}{\app{\recM(d)}c}
		& \hspace{-0.6em}\red_{@_0}    \dots
	\end{array}\]
\end{example}

\begin{example}\label{ex:simpleletrec-infinite-path}
	As an illustration of a regular term we study \M defined as the unfolding of \m{\letin{f = \abs{xy}{\app{\app{f}{y}}{x}}}{f}} from \cref{ex:expressible:simpleletrec}. It is strongly regular since the infinite path through the term can be witnessed by this cyclic \stRegARS rewriting sequence:
	\rewritingSequence{
		\init{\prefixed{}\M}
		\step{=}         {\prefixed{}{\abs{xy}{\app{\app\M{y}}{x}}}}
		\step{\red_λ}    {\prefixed{x}{\abs{y}{\app{\app\M{y}}{x}}}}
		\step{\red_λ}    {\prefixed{xy}{\app{\app\M{y}}{x}}}
		\step{\red_{@_0}}{\prefixed{xy}{\app\M{y}}}
		\step{\red_{@_0}}{\prefixed{xy}\M}
		\step{\red_\S}   {\prefixed{x}\M}
		\step{\red_\S}   {\prefixed{}\M}
		\step{}          {\dots}
	}
	See also \cref{fig:distance} for a graphical illustration of the reduction graph.
\end{example}

\begin{figure}
	\begin{hspread}
		\vcentered{\fig{simpleletrec-eager-ltg-l}} &
		\vcentered{\fig{pstricks/named/simpleletrec-lazy-ltg-l}}
	\end{hspread}
	\caption{\stRegCRS-reduction-graphs of \m{\unfsem{\letin{f = \abs{xy}{\app{\app{f}{y}}{x}}}{f}}} with respect to the \extscope-delimiting strategies \eagStratPlus (left) and \lazyStratPlus (right).}
	\label{fig:distance}
\end{figure}

\begin{para}[eager scope-closure is necessary]
	The restriction of \cref{prop:eager:strat:in:def:reg:streg} to the eager scope-delimiting strategy cannot be relaxed to arbitrary scope-delimiting strategies. The term in \cref{ex:simpleletrec-infinite-path} for instance is \eagStratPlus-regular but not \lazyStratPlus-regular (see \cref{fig:distance}).
\end{para}

\begin{definition}[grounded cycles in \RegARS, \stRegARS]\label{def:grounded:cycle}
	Let
	\m{\arewseq :
		\prefixed{\vecsub{x}{0}}{\M_0}
		\red
		\prefixed{\vecsub{x}{1}}{\M_1}
		\red
		\dots}
	be a finite or infinite rewrite sequence with respect to \m{\red_\reg} or \m{\red_\streg}.
	By a \emph{grounded cycle} in \arewseq we mean
	a cycle
	\m{\prefixed{\vecsub{x}{i}}{\M_i}
		\red
		\prefixed{\vecsub{x}{i+1}}{\M_{i+1}}
		\red
		\dots
		\red
		\prefixed{\vecsub{x}{i+k}}{\M_{i+k}} = \prefixed{\vecsub{x}{i}}{\M_i}},
	in \arewseq, where \m{i ∈ ℕ} and \m{k ≥ 1}, with the additional property that \m{\length{\vecsub{x}{i+j}} ≥ \length{\vecsub{x}{i}}} for all \m{j ∈ \set{0,\dots,k}} (i.e.\ the lengths of the abstraction prefixes in the terms of the cycle is greater or equal to the length of the abstraction prefix at the first and final term of the cycle).
\end{definition}

\begin{proposition}[infinite \astrat-rewrite-sequences contain grounded cycles]\label{prop:grounded:cycle}
	Let \M be a λ-term that is \astrat-regular for some scope/\extscope-delimiting strategy \astrat. Then every infinite rewrite sequence with respect to \astrat contains a grounded cycle.
\end{proposition}

\begin{proof}
	Since the argument is analogous in both cases, we only treat strongly regular terms. Let \M be a λ-term that is \astratplus-regular for some \extscope-delimiting strategy \astratplus, and let
	\m{\arewseq :
		\M
		=
		\prefixed{\vecsub{x}{0}}{\M_0}
		\red_{\astratplus}
		\prefixed{\vecsub{x}{1}}{\M_1}
		\red_{\astratplus}
		\dots}
	be an infinite rewrite sequence.
	As \M is \astratplus-regular, 
	the sequence \m{\sequence{\prefixed{\vecsub{x}{i}}{\M_i}}{i ∈ ℕ}} of terms on \arewseq contains only finitely many different terms.
	Let \m{l := \liminf \sequence{\length{\vecsub{x}{i}}}{i ∈ ℕ}}, that is, the minimum of abstraction prefix lengths that appears infinitely often on \arewseq. Let \m{\sequence{\prefixed{\vecsub{x}{i_j}}{\M_{i_j}}}{j ∈ ℕ}} be the subsequence of \m{\sequence{\prefixed{\vecsub{x}{i}}{\M_i}}{i ∈ ℕ}} consisting of terms with prefix length \m{l}, and such that, for all \m{k ≥ i_0}, \m{\length{\vecsub{x}{k}} ≥ l}.
	Since also this subsequence contains only finitely many terms, there exist \m{j_1,j_2 ∈ ℕ}, \m{j_1 < j_2} such that \prefixed{\vecsub{x}{i_{j_1}}}{\M_{i_{j_1}}} = \prefixed{\vecsub{x}{i_{j_2}}}{\M_{i_{j_2}}}. By the choice of the subsequence it follows that
	\m{\prefixed{\vecsub{x}{i_{j_1}}}{\M_{i_{j_1}}}
		\red_{\astratplus}
		\dots
		\red_{\astratplus}
		\prefixed{\vecsub{x}{i_{j_2}}}{\M_{i_{j_2}}}}
	is a grounded cycle in \arewseq.
\end{proof}

\begin{para}
	We round off this section by providing a motivation for the system \stRegARS in terms of an operation `parse' that when applied to a λ-term (i) decomposes it into its generated subterms, and (ii) recombines the generated subterms encountered in the decomposition analysis with, in the limit, the original term as the result. With this purpose in mind, we define a CRS \stParseCRS.
\end{para}

\begin{definition}[\stParseCRS]\label{def:stParseCRS}
	Let \m{\sigCRSParse = \sigCRSPrefixed ∪ \setcompr{\parsei{n}}{n ∈ ℕ}} be the extension of the CRS signature for \lambdaprefixcal, where for \m{n ∈ ℕ}, the symbols \parsei{n} have arity \n. By \stParseCRS we denote the CRS over \sigCRSParse with the following rules:
	\newcommand\rulename[1]{\rulebp{\sparse}{#1}\hspace*{-.7ex}:~~&}
	\newcommand\lhs[1]{#1 \red}
	\newcommand\rhs[1]{\\ & \hspace{-8mm} #1 \\[0.5em]}
	\begin{align*}
		\rulename{@}
			\lhs{\parse{n}{\vecsub{X}{n}, \prefixedCRS{n}{\vecsub{x}{n}}{\appCRS{Z_0(\vecsub{x}{n})}{Z_1(\vecsub{x}{n})}}}}
			\rhs{\appCRS{\parse{n}{\vecsub{X}{n}, \prefixedCRS{n}{\vecsub{x}{n}}{Z_0(\vecsub{x}{n})}}}{\parse{n}{\vecsub{X}{n}, \prefixedCRS{n}{\vecsub{x}{n}}{Z_1(\vecsub{x}{n})}}}}
		\rulename{λ}
			\lhs{\parse{n}{\vecsub{X}{n}, \prefixedCRS{n}{\vecsub{x}{n}}{\absCRS{y}{Z(\vecsub{x}{n},y)}}}}
			\rhs{\absCRS{y}{\parse{n+1}{\vecsub{X}{n},y, \prefixedCRS{n+1}{\vecsub{x}{n}}{\CRSabs{y}{Z(\vecsub{x}{n},y)}}}}}
		\rulename{\S}
			\lhs{\parse{n+1})X_1,\dots,X_{n+1}(\sPrefixedCRS{n+1}(\CRSabs{\vecsub{x}{n+1}}{Z(\vecsub{x}{n})})}
			\rhs{\parse{n}{X_1,\dots,X_n, {\sPrefixedCRS{n}}(\CRSabs{\vecsub{x}{n}}{Z(\vecsub{x}{n})})}}
		\rulename{\0}
			\lhs{\parse{n}{X_1,\dots,X_n, \prefixedCRS{n}{\vecsub{x}{n}}{x_n}}}
			X_n 
	\end{align*}
	We denote by \m{\red_\sparse} the rewrite relation induced by this CRS.
\end{definition}

\begin{para}[\stParseCRS contains \stRegCRS]
	Observe that the rules \rulep{@_i} for \m{i ∈ \set{0,1}}, \rulep{λ}, and \rulep{\S} of the CRS \stRegCRS are contained within the rules \rulebp{\sparse}{@}, \rulebp{\sparse}{λ}, \rulebp{\sparse}{\S}, respectively, of the CRS \stParseCRS, in the sense that applications of the latter rules include applications of the former. This has as a consequence that repeated \m{\red_\sparse}-steps on a term \prefixed{}{\M} lead to terms that contain generated subterms of \M as closed subexpressions. Furthermore \m{\red_\sparse}-rewrite-sequences on \prefixed{}{\M} are possible that move redexes simultaneously deeper and deeper, analysing ever larger parts of \M, and at the same time recreating a larger and larger λ-term parts (stable prefix contexts) of \M, the original term.
\end{para}

\begin{proposition}\label{prop:stParseCRS}
	For every term \m{\M ∈ \iTer{\lambdacal}} it holds:
	\begin{enumerate}[(i)]
		\item\label{prop:stParseCRS:i} \stParseCRS analyses \M into its generated subterms: If \parse{0}{\prefixed{}{\M}} \m{\mred_\sparse} \m{\M'}, then all subexpressions starting with \sPrefixedCRS{n} (for some \m{n ∈ ℕ}) in \m{\M'} are generated subterms of \M. Moreover, for every generated subterm \prefixed{\vec{y}}{\N} of \M, there exists a \m{\mred_\sparse}-reduction \m{\M''} of \M such that \prefixed{\vec{y}}{\N} is contained in \m{\M''}.
		\item\label{prop:stParseCRS:ii} \stParseCRS reconstructs \M: \m{\parse{0}{\prefixed{}{\M}} \m{\infred_\sparse} \M}.
	\end{enumerate}
\end{proposition}

\begin{example}\label{ex:stParseCRS}
	Let \M be the infinite unfolding of \m{\letin{f = \abs{xy}{\app{\app{f}{y}}{x}}}{f}} for which we use as a finite representation the equation \m{\M = \abs{xy}{\app{\app{\M}{y}}{x}}}. In \stParseCRS, \M is decomposed, and composed again, by the infinite rewrite sequence (see also \cref{ex:expressible:simpleletrec}):
	\rewritingSequence{
		\init{\parse{0}{\prefixed{}{\M}}}
		\step{=}{\parse{0}{\prefixed{}{\abs{x}{\abs{y}{\app{\app{\M}{y}}{x}}}}}}
		\step{\red_{\sparse.λ}}
			{\abs{x'}{\parse{1}{x', \prefixed{x'}{\abs{y}{\app{\app{\M}{y}}{x'}}}}}}
		\step{\red_{\sparse.λ}}
			{\abs{x'}{\abs{y'}{\parse{2}{x',y', \prefixed{x'y'}{\app{\app{\M}{y'}}{x'}}}}}}
		\step{\red_{\sparse.@}}{
			\abs{x'}{\abs{y'}{\app
				{\parse{2}{x',y', \prefixed{x'y'}{\app{\M}{y'}}}}
				{\parse{2}{x',y', \prefixed{x'y'}{x'}}}}}}
		\step{\red_{\sparse.\S}}{
			\abs{x'}{\abs{y'}{\app
				{\parse{2}{x',y', \prefixed{x'y'}{\app{\M}{y'}}}}
				{\parse{1}{x', \prefixed{x'}{x'}}}}}}
		\step{\red_{\sparse.\0}}{
			\abs{x'}{\abs{y'}{\app
				{\parse{2}{x',y', \prefixed{x'y'}{\app{\M}{y'}}}}
				{x'}}}}
		\step{\red_{\sparse.@}}{
			\abs{x'}{\abs{y'}{\app
				{\app
					{\parse{2}{x',y', \prefixed{x'y'}{\M}}}
					{\parse{2}{x',y', \prefixed{x'y'}{y'}}}}
				{x'}}}}
		\step{\red_{\sparse.\0}}{
			\abs{x'}{\abs{y'}{\app
				{\app
					{\parse{2}{x',y', \prefixed{x'y'}{\M}}}
					{y'}}
				{x'}}}}
		\step{\red_{\sparse.\S}}{
			\abs{x'}{\abs{y'}{\app
				{\app
					{\parse{1}{x', \prefixed{x'}{\M}}}
					{y'}}
				{x'}}}}
		\step{\red_{\sparse.\S}}{
			\abs{x'}{\abs{y'}{\app
				{\app
					{\parse{0}{\prefixed{}{\M}}}
					{y'}}
				{x'}}}}
		\step{=}{
			\abs{x'}{\abs{y'}{\app
				{\app
					{\parse{0}{\prefixed{}{\abs{x}{\abs{y}{\app{\app{\M}{y}}{x}}}}}}
					{y'}}
				{x'}}}}
		\step{\red_{\sparse.λ}}{\dots}
	}
	Note that the generated subterms of \M appear as the last arguments of the \parsei{i}-symbols in this rewrite sequence.
\end{example}

\section{Observing \lambdaletrec-terms by their generated subterms}\label{sec:letrec2}

\begin{para}[overview]
	In this section we adapt the concepts developed so far for the infinitary λ-calculus to \lambdaletrec{}. By combining the rules of \unfCRS with those of \RegCRS and \stRegCRS, respectively, we obtain the CRSs \RegletrecCRS and \stRegletrecCRS for the deconstruction of \lambdaletrec-terms furnished with an abstraction prefix. We define scope-delimiting and \extscope-delimiting strategies for \RegletrecCRS and \stRegletrecCRS as before by excluding all non-determinism except for sources of \m{\red_{@_0}}-steps and \m{\red_{@_1}}-steps.
\end{para}

\begin{definition}[the CRSs \RegletrecCRS and \stRegletrecCRS for decomposing \lambdaletrec-terms]\label{def:RegletrecCRS:stRegletrecCRS}
	We extend \sigCRSletrec (see \cref{def:sigs:lambdacal:lambdaletrec:CRS}) by function symbols \sPrefixedCRS{n} with arity one to obtain the signature \m{\sigCRSLetrecPrefixed = \sigCRSletrec ∪ \setcompr{\sPrefixedCRS{n}}{n ∈ ℕ}}. We denote the induced set of (finite) prefixed \lambdaletrec-terms by \Ter{\lambdaletrecprefixcal} and we adopt the same informal notation for \lambdaletrecprefixcal as for \lambdaprefixcal{} (see \cref{def:sig:lambdaprefixcal:CRS}).
	On the signature \sigCRSletrec we define the CRS \RegletrecCRS (the CRS \stRegletrecCRS) with the rules as the union of the rules of \unfCRS and \RegCRS (the union of the rules of \unfCRS and \stRegCRS).
\end{definition}

\begin{para}[list of rules for \RegletrecCRS/\stRegletrecCRS]
	The rules of \RegletrecCRS/\stRegletrecCRS consist of the set of rules arising from the rule schemes
	\unfrule{λ}, \unfrule{@}, \unfrule{\merg}, \unfrule{\rec}, \unfrule{\nil}, \unfrule{\reduce}
	from the unfolding CRS \unfCRS, joined with the rules arising from the rule schemes
	\decrule{@_0}, \decrule{@_1}, \decrule{λ}, \decrule{\del}/\decrule{\S}
	from the decomposition CRS \RegCRS/\stRegCRS.
\end{para}

\begin{notation}[rewrite relations for \RegletrecCRS/\stRegletrecCRS]
	To avoid ambiguity when referring to the rewriting relations induced by the rules of \RegletrecCRS/\stRegletrecCRS, we will prefix the rule names of \unfCRS with \unfold. Thus, we will denote the induced rewrite relations for \RegletrecCRS/\stRegletrecCRS as:
	\m{\red_{\unfoldpre{λ}}}, \m{\red_{\unfoldpre{@}}}, \m{\red_{\unfoldpre{\merg}}}, \m{\red_{\unfoldpre{\rec}}}, \m{\red_{\unfoldpre{\nil}}}, \m{\red_{\unfoldpre{\reduce}}},
	\m{\red_{@_0}}, \m{\red_{@_1}}, \m{\red_{λ}}, \m{\red_{\del}}/\m{\red_{\S}}.
\end{notation}

\begin{definition}[the ARSs \RegletrecARS, \stRegletrecARS]
	By \RegletrecARS, and \stRegletrecARS we denote the ARSs that result by restricting the ARSs induced by \RegletrecCRS, and \stRegletrecCRS, respectively, to the subset \Ter{\lambdaletreccal} of terms.
\end{definition}


The lemmas below state a number of simple rewrite properties for \RegletrecARS and \stRegletrecARS concerning the interplay between unfolding and decomposition steps.

\begin{lemma}\label{prop:rewprops:RegletrecCRS:stRegletrecCRS}
	On \Ter{\lambdaletrecprefixcal}, the rewrite relations in \RegletrecCRS and \stRegletrecCRS have the following properties:
	\begin{enumerate}[(i)]
		\item\label{prop:rewprops:RegletrecCRS:stRegletrecCRS:i}
			\m{\red_\unfold} one-step commutes with \m{\red_λ}, \m{\red_{@_0}}, \m{\red_{@_1}}, \m{\red_\S}, and \m{\red_\del}:
			\begin{align*}
				\m{\convred_\unfold} &⋅ \m{\red_λ} && \subseteq& \m{\red_λ} &⋅ \m{\convred_\unfold}
				\\
				\m{\convred_\unfold} &⋅ \red_{@_i} && \subseteq& \red_{@_i} &⋅ \m{\convred_\unfold} ~~~ (i ∈ \set{0,1})
				\\
				\m{\convred_\unfold} &⋅ \m{\red_\S} && \subseteq& \m{\red_\S} &⋅ \m{\convred_\unfold}
				\\
				\m{\convred_\unfold} &⋅ \m{\red_\del} && \subseteq& \m{\red_\del} &⋅ \m{\convred_\unfold}
			\end{align*}
		\item\label{prop:rewprops:RegletrecCRS:stRegletrecCRS:ii}
			\RegletrecARS and \stRegletrecARS have the same normal forms as \RegARS and \stRegARS, respectively:
			\prefixed{x}{x} is the only term in \Ter{\lambdaletrecprefixcal} in \m{\red_\reg}-normal-form.
			Every \m{\red_\streg}-normal-form in \Ter{\lambdaletrecprefixcal} is of the form
			\prefixed{x_1 \dots x_n}{x_n}.
	\end{enumerate}
\end{lemma}

\begin{proof}
	The commutation properties in \cref{prop:rewprops:RegletrecCRS:stRegletrecCRS:i} are easy to verify by analysing the behaviour of the rewrite rules in \RegletrecARS and in \stRegletrecARS on the terms of \Ter{\lambdaletrecprefixcal}.
	\par The statement in \cref{prop:rewprops:RegletrecCRS:stRegletrecCRS:ii} follows from \cref{prop:rewprops:RegCRS:stRegCRS}~\cref{prop:rewprops:RegCRS:stRegCRS:i}, and \cref{prop:rewprops:RegCRS:stRegCRS:v}: Normal forms with respect to \RegletrecARS and \stRegletrecARS can only be λ-terms without occurrences of \letrec, since every occurrence of \letrec in a \lambdaletrec-term gives rise to a \m{\red_\unfold}-redex.
\end{proof}

\begin{lemma}\label{lem:commute:unfoldomegared:stregred}
	The rewrite relation \m{\omeganfred_\unfold} (to \m{\red_\unfold}-normal-form in at most \omega steps) one-step commutes with \m{\red_λ}, \m{\red_{@_0}}, \m{\red_{@_1}}, \m{\red_\S}, and \m{\red_\del}:
	\begin{align*}
		\m{\convomeganfred_\unfold} ⋅ \m{\red_λ}
			& \;\subseteq\;
		\m{\red_λ} ⋅ \m{\convomeganfred_\unfold}
		&
		\m{\convomeganfred_\unfold} ⋅ \red_{@_i}
			& \;\subseteq\;
		\red_{@_i} ⋅ \m{\convomeganfred_\unfold}
			 ~~\; (i ∈ \set{0,1})
		\\
		\m{\convomeganfred_\unfold} ⋅ \m{\red_\S}
			& \;\subseteq\;
		\m{\red_\del} ⋅ \m{\convomeganfred_\unfold}
		&
		\m{\convomeganfred_\unfold} ⋅ \m{\red_\del}
			& \;\subseteq\;
		\m{\red_\del} ⋅ \m{\convomeganfred_\unfold}
	\end{align*}
	This implies, for prefixed terms that have unfoldings that:
	\begin{gather*}
		\unfsem{\prefixed{\vec{x}}{\abs{y}{\L_0}}}
		\m{\red_λ}
		\unfsem{\prefixed{\vec{x}y}{\L_0}}
		\\
		\unfsem{\prefixed{\vec{x}}{\app{\L_0}{\L_1}}}
		\m{\red_{@_i}}
		\unfsem{\prefixed{\vec{x}}{\L_i}}
		\m{~~\; (i ∈ \set{0,1})}
		\\
		\prefixed{\vec{x}}{\M}
		\m{\red_\S}
		\prefixed{\vec{x}'}{\M}
		~⇒~
		\unfsem{\prefixed{\vec{x}}{\M}}
		\m{\red_\S}
		\unfsem{\prefixed{\vec{x}'}{\M}}
		\\
		\prefixed{\vec{x}}{\M}
		\m{\red_\del}
		\prefixed{\vec{x}'}{\M}
		~⇒~
		\unfsem{\prefixed{\vec{x}}{\M}}
		\m{\red_\del}
		\unfsem{\prefixed{\vec{x}'}{\M}}
	\end{gather*}
	Furthermore it holds: \m{\convomeganfred_\unfold ⋅ \red_\unfold ~\subseteq~ \convomeganfred_\unfold}.
\end{lemma}

\begin{proof}
	\newcommand\sunfolddepthconvred[1]{\m{\convred^{#1}_{\unfold}}}
	\newcommand\sunfolddepthred[1]{\m{\red^{#1}_{\unfold}}}
	The commutation properties with the rewrite relation \m{\omeganfred_\unfold} can be shown by using refined versions of the commutation properties in \cref{prop:rewprops:RegletrecCRS:stRegletrecCRS}~\cref{prop:rewprops:RegletrecCRS:stRegletrecCRS:i}, in which the minimal depth of unfolding steps is taken account of. When denoting by \m{\sunfolddepthred{≥ n}} the rewrite relation that is generated by \m{\red_\unfold}-steps of depth \m{≥ n}, then the following properties hold:
	\begin{align*}
		\sunfolddepthconvred{≥ n} ⋅ \m{\red_λ} & \;\subseteq\; \m{\red_λ} ⋅ \sunfolddepthconvred{≥ n} &
		\sunfolddepthconvred{≥ n+1} ⋅ \red_{@_i} & \;\subseteq\; \red_{@_i} ⋅ \sunfolddepthconvred{≥ n} \\
		\sunfolddepthconvred{≥ n+1} ⋅ \m{\red_\S} & \;\subseteq\; \m{\red_\S} ⋅ \sunfolddepthconvred{≥ n} &
		\sunfolddepthconvred{≥ n+1} ⋅ \m{\red_\del} & \;\subseteq\; \m{\red_\del} ⋅ \sunfolddepthconvred{≥ n}
	\end{align*}
	Using these properties, strongly convergent \m{\red_\unfold}-rewrite-sequences can be shown to project, via \m{\red_λ}-, \m{\red_{@_i}}-, \m{\red_\del}-, and \m{\red_\S}-steps, to strongly convergent \m{\red_\unfold}-rewrite-sequences.
	\par The property \m{\convomeganfred_\unfold ⋅ \red_\unfold ~\subseteq~ \convomeganfred_\unfold} can be shown by using refined versions of the elementary diagrams from the confluence proof that take the minimal depths of steps into account.
\end{proof}

\begin{para}[scope/\extscope-delimiting strategy for \RegletrecARS/\stRegletrecARS]
	As for \RegARS/\stRegARS we require of scope/\extscope-delimiting strategies to have deterministic \m{\red_\del}/\m{\red_\S}-steps. As we did for \cref{def:scope:delimiting:strat:Reg:stReg}, we will also fix all non-determinism except for the choice between \m{@_0} and \m{@_1}.
\end{para}

\begin{definition}[scope/\extscope-delimiting strategy for \RegletrecARS/\stRegletrecARS]\label{def:scope:delimiting:strat:Regletrec:stRegletrec}
	A strategy \astrat for \RegletrecARS (\stRegletrecARS) will be called a \emph{scope-delimiting} (\emph{\extscope-delimiting}) \emph{strategy} if:
	\begin{itemize}
		\item \astrat is deterministic for sources of \m{\red_λ}-steps, \m{\red_\del}-steps (\m{\red_\S}-steps), and all \letrec-unfolding steps (i.e.\ all \m{\red_{\unfoldpre{\arule}}}-steps for every rule \unfrule{\arule} of \unfCRS).
		\item \astrat enforces eager application of \unfrule{\reduce}: every source of a step in \astrat according to an application of a rule different from \unfrule{\reduce} is not the source of a \m{\red_\reduce}-step in the underlying ARS.
	\end{itemize}
	We say that such a strategy \astrat is a \emph{lazy-unfolding} scope-delimiting strategy (a \emph{lazy-unfolding} \extscope-delimiting strategy) if furthermore:
	\begin{itemize}
	\item \astrat applies the rules of \unfCRS except for \unfrule{\reduce} only at the root of the term, i.e.\ directly beneath the abstraction prefix.
		\item \astrat uses the rules of \unfCRS other than \unfrule{\reduce} in a lazy way: every source of a step in \astrat with respect to a rule of \unfCRS other than \unfrule{\reduce} is not also the source of a step in the underlying ARS, with respect to one of the rules of \RegARS (of \m{\stRegARS)}.
	\end{itemize}
\end{definition}

\begin{notation}\label{not:dot}
	For every scope-delimiting strategy \astrat on \RegletrecARS (on \stRegletrecARS), we denote by \m{\red_{\astrat.\arule}} the rewrite relation that is induced by those steps according to \astrat due to applications of the rule \arule of \RegletrecARS (of \stRegletrecARS).
\end{notation}

\begin{para}[deterministic unfolding]\label{rem:nondet:unfolding}
	Note that in \cref{def:scope:delimiting:strat:Regletrec:stRegletrec} we do not only require of a strategy to eliminate the non-determinism with respect to \decrule{\del}-steps (\decrule{\S}-steps) but all non-determinism except for the \decrule{@_0}/\decrule{@_1}-non-determinism. This restriction will play a role later for the definition of λ-transition-graphs in \cref{sec:lambda-transition-graphs}, and aso in this section for defining projections of scope/\extscope-delimiting strategies for \lambdaletrec-terms to scope/\extscope-delimiting strategies for λ-terms.
\end{para}

\begin{para}[eager application of \unfrule{\reduce}]
	By requiring scope/\extscope delimiting strategies to apply \unfrule{\reduce} eagerly we can exploit a useful property with respect to free variables of a term: if \prefixed{\vec{x}}{\L ∈ \Ter{\lambdaletrecprefixcal}} is in \m{\red_{\unfoldpre{\reduce}}}-normal-form, 
	then the free variables occurring in \M correspond to the free variables of \unfsem{\L}.
\end{para}

\begin{para}[unfolding of prefixed terms]
	Note that we have just applied \unfsem{} to a prefixed term while in \cref{def:unf-mapping} \unfsem{} is only defined for `normal', unprefixed terms from \Ter{\lambdaletreccal}. We will continue to do so in the rest of this section and assume that the domain of \unfsem{} is extended in the obvious way to terms in \Ter{\lambdaletrecprefixcal}.
\end{para}


\begin{definition}[\astrat-productive terms]\label{def:productive:strats}
	Let \L be a \lambdaletrec-term, and \astrat a scope-delimiting strategy for \RegletrecARS or a \extscope-delimiting strategy for \stRegletrecARS. We say that \L is \emph{\astrat-productive} if every infinite rewrite sequence on \L with respect to \astrat contains infinitely many steps according to \m{\red_{{\astrat}.@_0}}, \m{\red_{{\astrat}.@_1}}, or \m{\red_{{\astrat}.λ}}.
\end{definition}

\begin{definition}[generated subterms of \lambdaletrec-terms]\label{def:ST:letrec}
	We extend the domain of \ST{}{} from \cref{def:ST} to \lambdaletrec-terms. 
	Let \astrat be a scope-delimiting strategy for \RegletrecARS/\stRegletrecARS. For every \m{\L ∈ \Ter{\lambdaletreccal}}, \ST{\astrat}{\L} is defined as the set of objects of the sub-ARS \GeneratedSubARSi{\astrat}{\prefixed{}{\L}}, or in other words, the set of \m{\mred_{\astrat}}-reducts of \prefixed{}{\L}:
	\begin{align*}
		\ST{\astrat}{} ~:~ \Ter{\lambdaletreccal}
		  & ~→~ \powerset(\Ter{\lambdaletrecprefixcal}) \\
		L & ~↦~ O ~~~ \text{where \m{\tuple{O,\steps,\src{},\tgt{}} = \GeneratedSubARSi{\astrat}{\prefixed{}{L}}}}
	\end{align*}
\end{definition}

\begin{para}
	The following lemma states that every scope/\extscope-delimiting strategy for \stRegletrecARS (on \lambdaletrec-terms), when restricted to the reducts of a \lambdaletrec-term \L that expresses a λ-term \M, projects to the restriction of a \extscope-delimiting strategy for \stRegARS (on λ-terms) to reducts of \M. And it asserts a similar statement for scope-delimiting strategies. The proof uses the commutation properties described in \cref{lem:commute:unfoldomegared:stregred} between the infinite unfolding \m{\omeganfred_\unfold} and the decomposition rewrite relations \m{\red_λ}, \m{\red_{@_0}}, \m{\red_{@_1}}, \m{\red_\S}, and \m{\red_\del}.
\end{para}

\begin{lemma}[projection of scope-delimiting strategies]\label{lem:proj:strat:letrec:lambda}
	Let \astrat be a scope/\extscope-delimiting strategy \astrat for \RegletrecARS/\stRegletrecARS, and let \L be a \lambdaletrec-term that is \astrat-productive. Then there exists a (history-aware) scope/\extscope-delimiting strategy \Check{\astrat} for \RegARS/\stRegARS such that the generated sub-ARS \GeneratedSubARSi{\Check{\astrat}}{\unfsem{\L}} of \unfsem{\L} is the projection (under the unfolding mapping \unfsem{}) of the generated sub-ARS \GeneratedSubARSi{\astrat}{\L} of \L, in the sense that for all \m{\L'} in \GeneratedSubARSi{\astrat}{\L} it holds:
	\begin{gather*}
		\L' \mred_{\astrat.\unfold} ⋅ \red_{\astrat.\decompose} \L''
		~\Longrightarrow~
		\unfsem{\L'} \red_{\Check{\astrat}} \unfsem{\L''}
		\\
		\unfsem{\L'} \red_{\Check{\astrat}} \M''
		~\Longrightarrow~
		(∃ \L'')~
		\L' \mred_{\astrat.\unfold} ⋅ \red_{\astrat.\decompose} \L''
		\;∧\;
		\M'' = \unfsem{\L''}
	\end{gather*}
	As a consequence, \unfsem{\L} is \Check{\astrat}-regular if \L is \astrat-regular.
\end{lemma}

\begin{proof}[Proof sketch]
	We can utilise \cref{lem:commute:unfoldomegared:stregred} to make commuting diagrams out of the two formulas above (for any given \m{\L'}), which allows us to determine \Check{\astrat} with respect to all terms in \GeneratedSubARSi{\astrat}{\L}. This freedom in the definition of \Check{\astrat} also guarantees the property in the second implication in the lemma.
\end{proof}

\begin{definition}[\stParseUnfCRS]\label{def:stParseUnfCRS}
	By \stParseUnfCRS we denote the CRS comprising the rules of \stParseCRS as well as the unfolding rules from \unfCRS.
\end{definition}

\begin{example}\label{ex:stParseUnfCRS}
	When applied to \m{\letin{f = \abs{xy}{\app{\app{f}{y}}{x}}}{f}}, the rewrite relation in \stParseUnfCRS unfolds and decomposes, but at the same time recreates the corresponding λ-term (see also \cref{ex:expressible:simpleletrec} and \cref{ex:stParseCRS}). In the rewriting sequence below, we prefix the rewrite relations that are due to rules from \unfCRS with \unfold, while the rewrite relations due to rules from \stParseUnfCRS are not prefixed.
	\rewritingSequence{
		\init{\parse{0}{\prefixed{}{\letin{f = \abs{xy}{\app{\app{f}{y}}{x}}}{f}}}}
		\step{\m{\red_{\unfoldpre{\rec}}}}
		{\parse{0}{\prefixed{}{\letin{f = \abs{xy}{\app{\app{f}{y}}{x}}}{\abs{xy}{\app{\app{f}{y}}{x}}}}}}
		\step{\m{\red_{\unfoldpre{λ}}}}
		{\parse{0}{\prefixed{}{\abs{x}{\letin{f = \abs{xy}{\app{\app{f}{y}}{x}}}{\abs{y}{\app{\app{f}{y}}{x}}}}}}}
		\step{\m{\red_{λ}}}
		{\abs{x}{\parse{1}{x, \prefixed{x}{\letin{f = \abs{xy}{\app{\app{f}{y}}{x}}}{\abs{y}{\app{\app{f}{y}}{x}}}}}}}
		\step{\m{\red_{\unfoldpre{λ}}}}
		{\abs{x}{\parse{1}{x, \prefixed{x}{\abs{y}{\letin{f = \abs{xy}{\app{\app{f}{y}}{x}}}{\app{\app{f}{y}}{x}}}}}}}
		\step{\m{\red_{λ}}}
		{\abs{xy}{\parse{2}{x,y, \prefixed{xy}{\letin{f = \abs{xy}{\app{\app{f}{y}}{x}}}{\app{\app{f}{y}}{x}}}}}}
		\step{\m{\red_{\unfoldpre{@}}}}
		{\simpleabs{xy}{\parse{2}{x,y, \prefixed{xy}{\appbreak
			{(\letin{f = \abs{xy}{\app{\app{f}{y}}{x}}}{\app{f}{y}})}
			{(\letin{f = \abs{xy}{\app{\app{f}{y}}{x}}}{x})}
		}}}}
		\step{\m{\red_{\unfoldpre{\reduce}}}}
		{\abs{xy}{\parse{2}{x,y, \prefixed{xy}{\app
			{(\letin{f = \abs{xy}{\app{\app{f}{y}}{x}}}{\app{f}{y}})}
			{(\letin{}{x})}
		}}}}
		\step{\m{\red_{\unfoldpre{\nil}}}}
		{\abs{xy}{\parse{2}{x,y, \prefixed{xy}{\app
			{(\letin{f = \abs{xy}{\app{\app{f}{y}}{x}}}{\app{f}{y}})}
			{x}
		}}}}
		\step{\m{\red_{@}}}
		{\simpleabs{xy}{\appbreak
			{(\parse{2}{x,y, \prefixed{xy}{\letin{f = \abs{xy}{\app{\app{f}{y}}{x}}}{\app{f}{y}}}})}
			{(\parse{2}{x,y, \prefixed{xy}{x}})}
		}}
		\step{\m{\red_{\S}}}
		{\simpleabs{xy}{\appbreak
			{(\parse{2}{x,y, \prefixed{xy}{\letin{f = \abs{xy}{\app{\app{f}{y}}{x}}}{\app{f}{y}}}})}
			{(\parse{1}{x, \prefixed{x}{x}})}
		}}
		\step{\m{\red_{\0}}}
		{\abs{xy}{\app
			{(\parse{2}{x,y, \prefixed{xy}{\letin{f = \abs{xy}{\app{\app{f}{y}}{x}}}{\app{f}{y}}}})}
			{x}
		}}
		\step{\m{\red_{\unfoldpre{@}}}}
		{\simpleabs{xy}{\app
			{(\parse{2}{x,y, \prefixed{xy}{\appbreak
				{(\letin{f = \abs{xy}{\app{\app{f}{y}}{x}}}{f})}
				{(\letin{f = \abs{xy}{\app{\app{f}{y}}{x}}}{y})}
			}})}
			{x}
		}}
		\step{\m{\red_{\unfoldpre{\reduce}}}}
		{\abs{xy}{\app
			{(\parse{2}{x,y, \prefixed{xy}{\app
				{(\letin{f = \abs{xy}{\app{\app{f}{y}}{x}}}{f})}
				{(\letin{}{y})}
			}})}
			{x}
		}}
		\step{\m{\red_{\unfoldpre{\nil}}}}
		{\abs{xy}{\app
			{(\parse{2}{x,y, \prefixed{xy}{\app
				{(\letin{f = \abs{xy}{\app{\app{f}{y}}{x}}}{f})}
				{y}
			}})}
			{x}
		}}
		\step{\m{\red_{@}}}
		{\simpleabs{xy}{\app{\appbreak
				{(\parse{2}{x,y, \prefixed{xy}{\letin{f = \abs{xy}{\app{\app{f}{y}}{x}}}{f}}})}
					{(\parse{2}{x,y, \prefixed{xy}{y}})}}
				{x}
		}}
		\step{\m{\red_{\0}}}
		{\abs{xy}{\app{\app
			{(\parse{2}{x,y, \prefixed{xy}{\letin{f = \abs{xy}{\app{\app{f}{y}}{x}}}{f}}})}
			{y}}
			{x}
		}}
		\step{\m{\red_{\S}}}
		{\abs{xy}{\app{\app
				{(\parse{1}{x, \prefixed{x}{\letin{f = \abs{xy}{\app{\app{f}{y}}{x}}}{f}}})}
				{y}}
			{x}
		}}
		\step{\m{\red_{\S}}}
		{\abs{xy}{\app{\app
				{(\parse{0}{\prefixed{}{\letin{f = \abs{xy}{\app{\app{f}{y}}{x}}}{f}}})}
				{y}}
			{x}
		}}
		\step{\m{\red_{\unfoldpre{\rec}}}}
		{\abs{xy}{\app{\app{~\dots~}{y}}{x}}}
	}
\end{example}

\begin{lemma}\label{lem:unfolding:versus:strats}
	For all closed \m{\L ∈ \Ter{\lambdaletreccal}} the following statements are equivalent:
	\begin{enumerate}[(i)]
		\item\label{lem:unfolding:versus:strats:i}
			\L expresses an infinite λ-term \M, that is, \m{\L \m{\omegared_\unfold} \M}.
		\item\label{lem:unfolding:versus:strats:ii}
			\m{\parse{0}{\prefixed{}{\L}} \m{\omegared_\sparseunf} \M}, for some infinite λ-term \M.
		\item\label{lem:unfolding:versus:strats:iii}
			\L is \astratplus-productive for some \extscope-delimiting strategy \astratplus.
		\item\label{lem:unfolding:versus:strats:iv}
			\L is \astratplus-productive for every \extscope-delimiting strategy \astratplus.
	\end{enumerate}
\end{lemma}

\begin{proof}
	Let \m{\L ∈ \Ter{\lambdaletreccal}}. We show the lemma by establishing the implications in the following order:
	``\cref{lem:unfolding:versus:strats:iv} ⇒ \cref{lem:unfolding:versus:strats:iii} ⇒ \cref{lem:unfolding:versus:strats:ii} ⇒ \cref{lem:unfolding:versus:strats:i} ⇒ \cref{lem:unfolding:versus:strats:iv}''.
	\par The implication ``\cref{lem:unfolding:versus:strats:iv} ⇒ \cref{lem:unfolding:versus:strats:iii}'' is clear: \cref{lem:unfolding:versus:strats:iv} implies that \L is productive for e.g.\ the lazy-unfolding, eager \extscope-delimiting strategy for \stRegletrecARS.
	\par For showing the implication ``\cref{lem:unfolding:versus:strats:iii} ⇒ \cref{lem:unfolding:versus:strats:ii}'', let \astrat be a \extscope-delimiting strategy for \stRegARS such that \L is \astrat-productive. Then the strategy \astrat defines a \m{\red_\sparseunf}-rewrite-sequence \arewseq on \parse{0}{\prefixed{}{\L}} by using \astrat to define next steps on subexpressions that are of the form \m{\parse{n}{\dots, \prefixedCRS{n}{x_1,\dots,x_n}{P}}} in already obtained reducts: if on a term \prefixed{x_1,\dots,x_n}{P} the strategy \astrat prescribes a \m{\red_\unfold}-step, then this step is adopted in \arewseq; if \astrat prescribes a \m{\red_{λ}}-step, then \arewseq can continue with a \m{{\red}_{\sparse.λ}}-step; if \astrat prescribes a \m{\red_{@_0}}- and a \m{\red_{@_1}}-step, then \arewseq can continue with a \m{{\red}_{\sparse.@}}-step. For the construction of \arewseq, possible steps in subexpressions \m{\parse{n}{\dots, \prefixedCRS{n}{x_1,\dots,x_n}{P}}} at parallel positions have to be interleaved to ensure that the reduction work is done in an outermost-fair way. Productivity of \astrat on \L ensures that always after finitely many steps inside a subexpression \m{\parse{n}{\dots, \prefixedCRS{n}{x_1,\dots,x_n}{P}}} the function symbol \parsei{n} disappears at this position (either entirely, or it is moved deeper over a λ-abstraction or an application). In the terms of the rewrite sequence \arewseq larger and larger λ-term contexts appear at the head. Hence \arewseq is strongly convergent, and it obtains, in the limit, an infinite λ-term; thus it witnesses \m{\arewseq : \parse{0}{\prefixed{}{\L}} \m{\omegared_{\sparseunf}} \M}.
	\par For the implication ``\cref{lem:unfolding:versus:strats:ii} ⇒ \cref{lem:unfolding:versus:strats:i}'', suppose that \arewseq is a rewrite sequence that witnesses \m{\parse{0}{\prefixed{}{\L}} \m{\omegared_{\sparseunf}} \M} for some infinite λ-term \M. Since the \m{\red_\sparse}-steps require already unfolded parts of the term, they have to `shadow' unfolding steps. All \m{\red_\unfold}-steps in \arewseq take place beneath symbols \parsei{n}. So the possibility of \m{\red_\sparse}-steps during \arewseq depends on the unfolding steps during \arewseq, but not vice versa. Hence a rewrite sequence \m{\arewseq'} on \parse{0}{\prefixed{}{\L}} can be constructed that only adopts the \m{\red_\unfold}-steps from \arewseq. Since \arewseq is strongly convergent and converges to \M, \m{\arewseq'} witnesses \m{\parse{0}{\prefixed{}{\L}} \m{\omegared_\unfold} \parse{0}{\prefixed{}{\M}}}. By dropping the `non-participant' prefix context \parse{0}{\prefixed{}{\acxthole}} from all terms in \m{\arewseq'}, and adapting the steps accordingly, a rewrite sequence \m{\arewseq''} is obtained that witnesses \m{\arewseq'' : \L \m{\omegared_\unfold} \M}.
	\par We show the implication ``\cref{lem:unfolding:versus:strats:i} ⇒ \cref{lem:unfolding:versus:strats:iv}'' indirectly. So we assume that there is a \extscope-delimiting strategy \astrat such that \L is not \astrat-productive. As in the proof above of ``\cref{lem:unfolding:versus:strats:iii} ⇒ \cref{lem:unfolding:versus:strats:ii}'', \astrat defines an outermost-fair \m{\red_{\sparseunf}}-rewrite-sequence \arewseq on \parse{0}{\prefixed{}{\L}}. But since \astrat here is a strategy that is not productive for \L, it follows that, due to its construction, \arewseq does not succeed in `pushing' all function symbols \letrec to deeper and deeper depth, and thereby building up an infinite λ-term. Instead, this outermost-fair \m{\red_{\sparseunf}}-rewrite-sequence contains infinitely many steps at the position of an outermost occurrence of \letrec. Since, other than the \m{\red_\unfold}-steps, the \m{\red_\sparse}-steps (which always take place above outermost occurrences of \letrec-symbols) cannot be the reason for this, the same stagnation of an outermost-fair unfolding process takes place if the \m{\red_\sparse}-steps are postponed, that is dropped from \arewseq. In this way, by again dropping the `non-participant' prefix context \parse{0}{\prefixed{}{\acxthole}} from the terms of \arewseq, and adapting the steps accordingly, we obtain an outermost-fair \m{\red_\unfold}-rewrite-sequence starting on \prefixed{}{\L} that does not converge to an infinite λ-term. But then \cref{lem:unfolding} implies that \L does not unfold to an infinite λ-term.
\end{proof}

\section{Proving regularity and strong regularity}\label{sec:proofs}

\begin{para}[overview]
	In this section we introduce proof systems that are sound and complete for the notions of regular, and strongly regular λ-terms. In order to prove soundness and completeness, we establish, as auxiliary results, a correspondence between scope/\extscope-delimiting strategies for \RegARS/\stRegARS and closed derivations in the corresponding proof systems. Then we introduce a proof system that is sound and complete for equality between strongly regular λ-terms. Furthermore, we give two proof systems that are sound and complete for the property of \lambdaletrec-terms to unfold to λ-terms. And finally, we show the following part of our characterisation result: λ-terms that are unfoldings of \lambdaletrec-terms are strongly regular.
\end{para}

We start with a more formal definition of λ-terms and \lambdaletrec-terms than \cref{def:CRSterms}, by means of derivability in a proof system that formalises term decomposition.

\begin{figure}
	\proofsystem{
		\begin{bprooftree}
			\emptyAxiom
			\infLabel{\0}
			\unaryInf{\prefixed{\vec{x}y}{y}}
		\end{bprooftree}
		\hsep
		\begin{bprooftree}
			\axiom{\prefixed{\vec{x}y}{\M_0}}
			\infLabel{λ}
			\unaryInf{\prefixed{\vec{x}}{\abs{y}{\M_0}}}
		\end{bprooftree}
		\hsep
		\begin{bprooftree}
			\axiom{\prefixed{\vec{x}}{\M_0}}
			\axiom{\prefixed{\vec{x}}{\M_1}}
			\infLabel{@}
			\binaryInf{\prefixed{\vec{x}}{\app{\M_0}{\M_1}}}
		\end{bprooftree}
		\\
		\begin{bprooftree}
			\axiom{\prefixed{x_1 \dots x_{n-1}}{\M}}
			\infLabel{\S~~\sideCondition{if the binding \m{λx_n} is vacuous}}
			\unaryInf{\prefixed{x_1 \dots x_n}{\M}}
		\end{bprooftree}
	}
	\caption{Proof system \iTer{\lambdacal} for defining the set of λ-terms.}
	\label{fig:iTer-lambda}
\end{figure}

\begin{definition}[λ-terms]\label{def:iTer-lambda}
	We define the set of prefixed λ-terms as those terms in \Ter{\sigCRSPrefixed} for which there exists a possibly infinite, completed (see \cref{def:iCRSPretermAlphaKahrs}) derivation in the proof system \iTer{\lambdacal} with axioms and rules as shown in \cref{fig:iTer-lambda}:
	\[\iTer\lambdaprefixcal := \setcompr{\M ∈ \Ter\sigCRSPrefixed}{\infderivablein{\iTer{\lambdacal}}\M}\]
	The set of plain λ-terms are those terms that comply with the previous definition when equipped with an empty prefix:
	\[\iTer\lambdacal := \setcompr{\M ∈ \Ter\sigCRS}{\prefixed{}{\M} ∈ \iTer\lambdaprefixcal}\]
\end{definition}

\begin{figure}
	\proofsystem{
		\begin{bprooftree}
			\axiom{\prefixed{\vec{x} f_1 \dots f_n}{\M_0}}
			\axiom{\dots}
			\axiom{\prefixed{\vec{x} f_1 \dots f_n}{\M_n}}
			\infLabel{\letrec}
			\ternaryInf{\prefixed{\vec{x}}{\letin{f_1=\M_1,\dots,f_n=\M_n}{\M_0}}}
		\end{bprooftree}
	}
	\caption{Proof system \Ter{\lambdaletreccal} for defining the set of \lambdaletrec-terms defined as an extension of \iTer{\lambdacal} in \cref{fig:iTer-lambda} by an additional rule \ruleref{\letrec}}
	\label{fig:Terlambdaletrec}
\end{figure}

\begin{definition}[\lambdaletrec-terms]\label{def:Ter-lambdaletrec}
	The set of prefixed \lambdaletrec-terms comprises those terms out of \Ter{\sigCRSLetrecPrefixed} for which there exists a finite derivation in the proof system \Ter{\lambdaletreccal} (\cref{fig:Terlambdaletrec}): \[\Ter\lambdaletrecprefixcal := \setcompr{\M ∈ \Ter\sigCRSLetrecPrefixed}{\derivablein{\Ter{\lambdaletreccal}}{\M}}\]
	The set of plain \lambdaletrec-terms are those terms that comply with the previous definition when equipped with an empty prefix: \[\Ter\lambdaletreccal := \setcompr{\M ∈ \Ter\sigCRSletrec}{\prefixed{}\M ∈ \Ter{\lambdaletrecprefixcal}}\]
\end{definition}

Building on rules already used in the proof systems for term formation in
\lambdacal and \lambdaprefixcal from the definition above,
we now introduce proof systems for regularity and strong regularity of λ-terms in \lambdacal.

\begin{definition}[proof systems \Reg, and \stReg, \stRegzero]\label{def:Reg:stReg:stRegzero}
	The natural-deduction style proof system \stReg for recognising strongly regular λ-terms contains the axioms and rules as shown in \cref{fig:stReg:stRegzero}. In particular, the rule \FIX is a natural-deduction style derivation rule in which marked assumptions from the top of the proof tree can be discharged. Instances of this rule carry the side-condition that the depth \depth{\Deriv_0} of the immediate subderivation \m{\Deriv_0} of its premise is greater or equal to 1 (hence this subderivation contains at least one rule instance, and, importantly, for a topmost occurrence of \FIX, \m{\Deriv_0} must have a bottommost instance of one of the rules \ruleref{λ}, \ruleref{@}, or \ruleref{\S}).
	\par The variant \stRegzero of \Reg contains the same axioms and rules as \stReg, but in it instances of \FIX are subject to the additional side-condition: for all \prefixed{\vec{y}}{\N} on threads in \m{\Deriv_0} from open marked assumptions \m{(\prefixed{\vec{x}}{\M})^u} downwards it holds that \m{\length{\vec{y}} ≥ \length{\vec{x}}}.
	\par The natural-deduction style proof system \Reg for recognising regular λ-terms differs from \stReg by the absence of the rule \ruleref{\S}, and the presence instead of the rule \ruleref{\del} in \cref{fig:Reg}, and by the restriction of the axiom scheme (\0) to the more restricted version displayed in \cref{fig:Reg}.
	\par Provability of a term in \lambdaprefixcal in one of these proof systems is defined as the existence of a \emph{closed} derivation: for \m{\boldsymbol{R} ∈ \set{\Reg, \stReg, \stRegzero}} we denote by \derivablein{\boldsymbol{R}}{\prefixed{\vec{x}}{\M}} the existence of a proof tree \Deriv with conclusion \M and with rule instances of \boldsymbol{R} such that all marked assumptions at the top of the \Deriv are discharged at some instance of the rule \FIX.
\end{definition}

\begin{figure}
	\proofsystem{
		\begin{bprooftree}
			\axiom{\prefixed{\vec{x}y}{\M_0}}
			\infLabel{λ}
			\unaryInf{\prefixed{\vec{x}}{\abs{y}{\M_0}}}
		\end{bprooftree}
		\hsep
		\begin{bprooftree}
			\axiom{\prefixed{\vec{x}}{\M_0}}
			\axiom{\prefixed{\vec{x}}{\M_1}}
			\infLabel{@}
			\binaryInf{\prefixed{\vec{x}}{\app{\M_0}{\M_1}}}
		\end{bprooftree}
		\\
		\begin{bprooftree}
			\emptyAxiom
			\infLabel{\0}
			\unaryInf{\prefixed{\vec{x}y}{y}}
		\end{bprooftree}
		\hsep
		\begin{bprooftree}
			\axiom{\prefixed{x_1 \dots x_{n-1}}{\M}}
			\infLabel{\S~~\mlSideCondition{if the binding \\ \m{λx_n} is vacuous}}
			\unaryInf{\prefixed{x_1 \dots x_n}{\M}}
		\end{bprooftree}
		\\
		\begin{bprooftree}
			\axiom{[\prefixed{\vec{x}}{\M}]^u}
			\noLine
			\unaryInf{\Deriv_0}
			\noLine
			\unaryInf{\prefixed{\vec{x}}{\M}}
			\infLabel{\FIX, u~~\sideCondition{if \m{\depth{\Deriv_0} ≥ 1}}}
			\unaryInf{\prefixed{\vec{x}}{\M}}
		\end{bprooftree}
	}
	\caption{The natural-deduction style proof system \stReg for strongly regular λ-terms is an extension of \iTer{\lambdacal} by one additional rule \FIX. In the variant system \stRegzero, instances of \FIX are subject to the following side-condition: for all \prefixed{\vec{y}}{\N} on threads in \m{\Deriv_0} from open marked assumptions \m{(\prefixed{\vec{x}}{\M})^u} downwards it holds that \m{\length{\vec{y}} ≥ \length{\vec{x}}}.}
	\label{fig:stReg:stRegzero}
\end{figure}

\begin{figure}
	\proofsystem{
		\begin{bprooftree}
			\axiom{\prefixed{\vec{x}y}{\M_0}}
			\infLabel{λ}
			\unaryInf{\prefixed{\vec{x}}{\abs{y}{\M_0}}}
		\end{bprooftree}
		\hsep
		\begin{bprooftree}
			\axiom{\prefixed{\vec{x}}{\M_0}}
			\axiom{\prefixed{\vec{x}}{\M_1}}
			\infLabel{@}
			\binaryInf{\prefixed{\vec{x}}{\app{\M_0}{\M_1}}}
		\end{bprooftree}
		\\
		\begin{bprooftree}
			\emptyAxiom
			\infLabel{\0}
			\unaryInf{\prefixed{y}{y}}
		\end{bprooftree}
		\hsep
		\begin{bprooftree}
			\axiom{\prefixed{x_1 \dots x_{i-1}x_{i+1}\dots x_n}{\M}}
			\infLabel{\del~~\mlSideCondition{if the binding \\ \m{λx_i} is vacuous}}
			\unaryInf{\prefixed{x_1 \dots x_n}{\M}}
		\end{bprooftree}
		\\
		\begin{bprooftree}
			\axiom{[\prefixed{\vec{x}}{\M}]^u}
			\noLine
			\unaryInf{\Deriv_0}
			\noLine
			\unaryInf{\prefixed{\vec{x}}{\M}}
			\infLabel{\FIX, u~~\sideCondition{if \m{\depth{\Deriv_0} ≥ 1}}}
			\unaryInf{\prefixed{\vec{x}}{\M}}
		\end{bprooftree}
	}
	\caption[natural-deduction style proof system \Reg for regular λ-terms]{The natural-deduction style proof system \Reg for regular λ-terms arises from the proof system \stReg by replacing the rule \ruleref{\S} with the rule \ruleref{\del} for the introduction of vacuous bindings in the λ-abstraction prefixes, and by replacing the axiom scheme (\0) of \stReg by the more restricted version here.}
	\label{fig:Reg}
\end{figure}

\begin{remark}[\stReg versus \stRegzero]
	While it will be established in \cref{prop:stReg:stRegzero} that provability in \stReg and \stRegzero coincides, the difference between these systems will come to the fore in annotated versions that are purpose-built for the extraction of \lambdaletrec-terms that express λ-terms. This will be explained and illustrated later in \cref{example:annstRegzero}, using annotated versions of the two derivations in \cref{example:stReg} above.
\end{remark}

\begin{remark}
	The proof system \stReg is related to a proof system for nameless, finite terms in the λ-calculus that is used in \cite[sec.~2]{oost:looi:zwit:2004} as part of a translation of λ-terms into `Lambdascope' interaction nets, which are used for optimal evaluation (in the sense of L\'{e}vy) of λ-terms.
\end{remark}

The proposition below explains that the side-condition on instances of \FIX from the proof systems above to have immediate subderivations \m{\Deriv_0} with \m{\depth{\Deriv_0} ≥ 1} entails a `guardedness' property for threads from such instances upwards to discharged instances.

\begin{proposition}[cycles are guarded]\label{prop:Reg:stReg}
	Let \Deriv be a derivation in \Reg, in \stReg or in \stRegzero possibly with open marked assumptions. Then for all instances \ainst of the rule \FIX in \Deriv it holds: every thread from \ainst upwards to a marked assumption that is discharged at \ainst passes at least one instance of a rule \ruleref{λ} or \ruleref{@}.
\end{proposition}

\begin{proof}
	Since for \stReg and \stRegzero the argument is analogous, we only consider derivations in \Reg. So, let \Deriv be a derivation in \Reg. Furthermore, let \ainst be an instance of the rule \FIX in \Deriv with conclusion \prefixed{\vec{x}}{\M}, and let \apath be a thread from the conclusion of \ainst upwards to a marked assumption \m{(\prefixed{\vec{x}}{\M})^{\amarker}}. Let \binst be the topmost instance of \FIX in \Deriv that is passed on \apath. By its side-condition, the immediate subderivation of \binst has depth greater or equal to 1, and hence there is at least one instance of a rule \ruleref{λ}, \ruleref{@}, or \ruleref{\del} passed on \apath above \binst. If there is an instance of \ruleref{λ} or \ruleref{@} on this part of \apath, we are done. Otherwise only rules \ruleref{\del} are passed on \apath above \binst. But since the rule \ruleref{\del} decreases the length of the abstraction prefix in the term occurrences in a pass from the conclusion to the premise, and since the length of the abstraction prefix at the start of \apath is the same as at the end of \apath, namely \length{\vec{x}}, it follows that at least one occurrence of a rule that increases the length of the abstraction prefix must also have been passed on \apath, on the segment from \ainst to \binst. Since the only rule of \Reg that increases the length of an abstraction prefix in a pass from conclusion to a premise is the rule \ruleref{λ}, we have also in this case found a desired rule instance on \apath.
\end{proof}

\begin{example}[difference between \stReg and \stRegzero]\label{example:stReg}
	Let \M be the infinite unfolding of \m{\letin{f = \abs{xy}{\app{\app{f}{y}}{x}}}{f}} for which we use as a finite representation the equation \m{\M = \abs{xy}{\app{\app{\M}{y}}{x}}}. This term admits the following two derivations in \stReg, with the latter having some redundancy:
	\begin{prooftree}
		\axiom{(\prefixed{}{\M})^{\amarker}}
		\infLabel{\S}
		\unaryInf{\prefixed{x}{\M}}
		\infLabel{\S}
		\unaryInf{\prefixed{xy}{\M}}
		\emptyAxiom
		\infLabel{\0}
		\unaryInf{\prefixed{xy}{y}}
		\infLabel{@}
		\binaryInf{\prefixed{xy}{\app{\M}{y}}}
		\emptyAxiom
		\infLabel{\0}
		\unaryInf{\prefixed{x}{x}}
		\infLabel{\S}
		\unaryInf{\prefixed{xy}{x}}
		\infLabel{@}
		\binaryInf{\prefixed{xy}{\app{\app{\M}{y}}{x}}}
		\infLabel{λ}
		\unaryInf{\prefixed{x}{\abs{y}{\app{\app{\M}{y}}{x}}}}
		\infLabel{λ}
		\unaryInf{\prefixed{}{\abs{xy}{\app{\app{\M}{y}}{x}}}}
		\infLabel{\FIX,\amarker}
		\unaryInf{\prefixed{}{\M}}
	\end{prooftree}
	\begin{prooftree}
		\axiom{(\prefixed{x}{\abs{y}{\app{\app{\M}{y}}{x}}})^{\amarker}}
		\infLabel{λ}
		\unaryInf{\prefixed{}{\M}}
		\infLabel{\S}
		\unaryInf{\prefixed{x}{\M}}
		\infLabel{\S}
		\unaryInf{\prefixed{xy}{\M}}
		\emptyAxiom
		\infLabel{\0}
		\unaryInf{\prefixed{xy}{y}}
		\infLabel{@}
		\binaryInf{\prefixed{xy}{\app{\M}{y}}}
		\emptyAxiom
		\infLabel{\0}
		\unaryInf{\prefixed{x}{x}}
		\infLabel{\S}
		\unaryInf{\prefixed{xy}{x}}
		\infLabel{@}
		\binaryInf{\prefixed{xy}{\app{\app{\M}{y}}{x}}}
		\infLabel{λ}
		\unaryInf{\prefixed{x}{\abs{y}{\app{\app{\M}{y}}{x}}}}
		\infLabel{\FIX,\amarker}
		\unaryInf{\prefixed{x}{\abs{y}{\app{\app{\M}{y}}{x}}}}
		\infLabel{λ}
		\unaryInf{\prefixed{}{\M}}
	\end{prooftree}
	Note that the first derivation is also a derivation in \stRegzero, but that this is not the case for the secord derivation, as for the occurrence of \FIX the side-condition in the system \stRegzero is violated: on the path from the marked assumption \m{(\prefixed{x}{\abs{y}{\app{\app{\M}{y}}{x}}})^{\amarker}} down to the instance of \FIX there is the occurrence \prefixed{}{\M} of a term with shorter prefix than the term in the assumption and in the conclusion.
	\par See \cref{ex:simpleletrec-infinite-path} for a rewriting sequence in \stRegARS corresponding to the leftmost path in both derivations, and also \cref{fig:distance} for the corresponding transition graph.
\end{example}

\begin{example}[difference between \RegCRS and \stRegCRS]
	The infinite λ-term from \cref{ex:entangled} with the \RegCRS/\stRegCRS-transition-graphs shown in \cref{fig:entangled-fix-ltg} is derivable in \Reg by the following closed derivation using the notation from \cref{ex:entangled-infinite-path}:
	\begin{prooftree}
		\axiom{(\overbrace{\prefixed{b}{{rec_{\M}}(b)}}^{\prefixed{a}{{rec_{\M}}(a)}})^{\amarker}}
		\infLabel{\del}
		\unaryInf{\prefixed{ab}{{rec_{\M}}(b)}}
		\emptyAxiom
		\infLabel{\0}
		\unaryInf{\prefixed{a}{a}}
		\infLabel{\del}
		\unaryInf{\prefixed{ab}{a}}
		\infLabel{@}
		\binaryInf{\prefixed{ab}{\app{{rec_{\M}}(b)}{a}}}
		\infLabel{λ}
		\unaryInf{\prefixed{a}{\abs{b}{\app{{rec_{\M}}(b)}{a}}}}
		\infLabel{\FIX, \amarker}
		\unaryInf{\prefixed{a}{{rec_{\M}}(a)}}
		\infLabel{λ}
		\unaryInf{\prefixed{}{\abs{a}{{rec_{\M}}(a)}}}
	\end{prooftree}
	When trying to construct a derivation for this term in \stReg from the bottom upwards, the rules of \stReg apart from \FIX offer only deterministic choices, resulting in an infinite proof tree of the form:
	\begin{prooftree}
		\axiom{\vdots}
		\noLine
		\unaryInf{\prefixed{abcde}{{rec_{\M}}(e)}}
		\infLabel{λ}
		\unaryInf{\prefixed{abcd}{\abs{e}{{rec_{\M}}(e)}}}
		\emptyAxiom
		\infLabel{\0}
		\unaryInf{\prefixed{abc}{c}}
		\infLabel{\S}
		\unaryInf{\prefixed{abcd}{c}}
		\infLabel{@}
		\binaryInf{\prefixed{abcd}{\app{{rec_{\M}}(d)}{c}}}
		\infLabel{λ}
		\unaryInf{\prefixed{abc}{\abs{d}{\app{{rec_{\M}}(d)}{c}}}}
		\emptyAxiom
		\infLabel{\0}
		\unaryInf{\prefixed{ab}{b}}
		\infLabel{\S}
		\unaryInf{\prefixed{abc}{b}}
		\infLabel{@}
		\binaryInf{\prefixed{abc}{\app{{rec_{\M}}(c)}{b}}}
		\infLabel{λ}
		\unaryInf{\prefixed{ab}{ \abs{c}{\app{{rec_{\M}}(c)}{b}}}}
		\emptyAxiom
		\infLabel{\0}
		\unaryInf{\prefixed{a}{a}}
		\infLabel{\S}
		\unaryInf{\prefixed{ab}{a}}
		\infLabel{@}
		\binaryInf{\prefixed{ab}{\app{{rec_{\M}}(b)}{a}}}
		\infLabel{λ}
		\unaryInf{\prefixed{a}{\abs{b}{\app{{rec_{\M}}(b)}{a}}}}
		\infLabel{λ}
		\unaryInf{\prefixed{}{\abs{a}{{rec_{\M}}(a)}}}
	\end{prooftree}
	But then, since this proof tree does not contain repetitions, the use of the rule \FIX in order to discharge assumptions is impossible. Consequently, the term is not derivable in \stReg.
	\par For the \RegARS and \stRegARS rewriting sequences corresponding to the leftmost paths through the two proofs above, see \cref{ex:entangled-infinite-path}. The corresponding transition graphs are displayed in \cref{fig:entangled-fix-ltg}.
\end{example}




\begin{proposition}[correspondence between proof systems and decomposition CRSs]\label{prop:derivationpaths:2:rewritesequences:Reg:stReg}
	Let \Deriv be a derivation in \Reg/\stReg with conclusion \prefixed{\vec{x}}{\M}.
	Then every path in \Deriv from the conclusion upwards corresponds to a \m{\red_\reg}/\m{\red_\streg}-rewrite-sequence from \prefixed{\vec{x}}{\M}: passes over instances of \FIX correspond to empty rewrite steps; passes over instances of \ruleref{@} to the left and to the right correspond to \m{\red_{@_0}}- and \m{\red_{@_1}}-steps, respectively; passes over instances of \ruleref{λ} correspond to \m{\red_λ}-steps; passes over instances of \ruleref{\del}/\ruleref{\S} correspond to \m{\red_\del}/\m{\red_\S}-steps.
	\par The same holds for (finite or infinite) cyclic paths in \Deriv that return, possibly repeatedly, from a marked assumption at the top down to the conclusion of the instance of \FIX at which the respective assumption is discharged.
\end{proposition}

\begin{proof}
	The proposition is an easy consequence of the following facts: passes from a term in the conclusion of an instance \ainst of one of the rules \ruleref{λ}, \ruleref{\del}, \ruleref{\S} to the term in the premise of \ainst correspond to \m{\red_λ}-, \m{\red_\del}-, and \m{\red_\S}-steps, respectively; passes from a term in the conclusion of an instance of \ruleref{@} to the left and the right premise correspond to \m{\red_{@_0}}-steps and \m{\red_{@_1}}-steps, respectively.
\end{proof}

\begin{para}[proofs and scope/\extscope-delimiting strategies]
	Observe that, for the derivation \Deriv in \cref{example:stReg}, the \m{\red_\streg}-rewrite-sequences that correspond to paths in \Deriv as described in \cref{prop:derivationpaths:2:rewritesequences:Reg:stReg} are actually rewrite sequences with respect to the \emph{eager} \extscope-delimiting strategy \eagStratPlus for \stRegARS. This illustrates the general situation, formulated by the lemma below: paths in a derivation \Deriv in \Reg/\stReg from the conclusion upwards correspond to rewrite sequences according to some -- usually history-aware -- scope/\extscope-delimiting strategy \astrat, which can be extracted from \Deriv.
\end{para}

\begin{lemma}[from \Reg/\stReg-derivations to scope/\extscope-delimiting strategies]\label{lem:derivations:Reg:stReg:2:strategies}
	Let \m{\M ∈ \iTer{\lambdacal}}, and let \Deriv be a closed derivation in \Reg (in \stReg) with conclusion \prefixed{}{\M}. Then there exists an, in general history-aware, scope-delimiting strategy \m{\astrat_{\Deriv}} for \RegARS (\extscope-delimiting strategy \m{\astrat_{\Deriv}} for \stRegARS) with the following properties:
	\begin{enumerate}[(i)]
		\item\label{lem:derivations:Reg:stReg:2:strategies:i} Every (possibly cyclic) path in \Deriv from the conclusion upwards corresponds to a rewrite sequence with respect to \m{\astrat_{\Deriv}} starting on \prefixed{}{\M} in the sense of \cref{prop:derivationpaths:2:rewritesequences:Reg:stReg} where passes over instances of the rules \ruleref{@} to the left and to the right correspond to \m{\red_{{\astrat_{\Deriv}}.@_0}}-steps and \m{\red_{{\astrat_{\Deriv}}.@_1}}-steps, respectively, and passes over instances of \ruleref{λ} and of \ruleref{\del} (of \ruleref{\S}) correspond to \m{\red_{{\astrat_\Deriv}.λ}}-, and \m{\red_{{\astrat_\Deriv}.\del}}-steps (\m{\red_{{\astrat_\Deriv}.\S}}-steps).
		\item\label{lem:derivations:Reg:stReg:2:strategies:ii} Every rewrite sequence that starts on \prefixed{}{\M} and proceeds according to \m{\astrat_\Deriv} corresponds to a (possibly cyclic) path in \Deriv starting at the conclusion in upwards direction: thereby a \m{\red_{{\astrat_\Deriv}.@_0}}-step and a \m{\red_{{\astrat_\Deriv}.@_1}}-step corresponds to a pass over (possibly successive \FIX-instances, or from a marked assumption to the instance of \FIX that binds it, followed by) an instance of \ruleref{@} in direction left and right, respectively; a \m{\red_{{\astrat_\Deriv}.λ}}-step or \m{\red_{{\astrat_\Deriv}.\del}}-step (\m{\red_{{\astrat_\Deriv}.\S}}-step) corresponds to a pass over (possibly \FIX-instances and assumption bindings to \FIX-instances) an instance of \ruleref{λ} or \ruleref{\del} (of \ruleref{\S}), respectively.
		\item\label{lem:derivations:Reg:stReg:2:strategies:iii} \m{\ST{\astrat_\Deriv}{\M} = \setcompr{ \prefixed{\vec{y}}{\N}}{\text{the term \prefixed{\vec{y}}{\N} occurs in \Deriv}}}.
	\end{enumerate}
\end{lemma}

\begin{figure}
	\proofsystem{
		\begin{bprooftree}
			\axiom{\labprefixed{l}{\vec{x}y}{\M_0}}
			\infLabel{λ}
			\unaryInf{\labprefixed{l}{\vec{x}}{\abs{y}{\M_0}}}
		\end{bprooftree}
		\hsep
		\begin{bprooftree}
			\axiom{\labprefixed{{l}\, 0}{\vec{x}}{\M_0}}
			\axiom{\labprefixed{{l}\, 1}{\vec{x}}{\M_1}}
			\infLabel{@}
			\binaryInf{\labprefixed{l}{\vec{x}}{\app{\M_0}{\M_1}}}
		\end{bprooftree}
		\\
		\begin{bprooftree}
			\emptyAxiom
			\infLabel{\0}
			\unaryInf{\labprefixed{l}{\vec{x}y}{y}}
		\end{bprooftree}
		\hsep
		\begin{bprooftree}
			\axiom{\labprefixed{l}{x_1 \dots x_{n-1}}{\M}}
			\infLabel{\S~~\mlSideCondition{if \m{x_n} does not \\ occur in \M}}
			\unaryInf{\labprefixed{l}{x_1 \dots x_n}{\M}}
		\end{bprooftree}
		\\
		\begin{bprooftree}
			\axiom{[\labprefixed{l}{\vec{x}}{\M}]^u}
			\noLine
			\unaryInf{\Deriv_0}
			\noLine
			\unaryInf{\labprefixed{l}{\vec{x}}{\M}}
			\infLabel{\FIX,u~~\sideCondition{if \m{\depth{\Deriv_0} ≥ 1}}}
			\unaryInf{\labprefixed{l}{\vec{x}}{\M}}
		\end{bprooftree}
	}
	\caption{Proof system \stReglab for decorating \stReg-derivations with labels in \m{\set{0,1}^*}.}
	\label{fig:stReglab}
\end{figure}

\begin{proof}\label{prf:lem:derivations:Reg:stReg:2:strategies}
	\newcommand\stepslabon[1]{\m{{\stepslab}_{\text{on-}{#1}}}}
	\newcommand\stepslabnoton[1]{\m{{\stepslab}_{\text{not-on-}{#1}}}}
	\newcommand\stepslab{\widehat{\steps}}
	The proof defines a history-aware strategy \m{\astrat_\Deriv} for \stRegARS as a modification of an arbitrary (history-free) strategy for \stRegARS lifted to a labelled version of \stRegARS. Thereby the modification is performed according to a given derivation \Deriv, and the construction will guarantee that \cref{lem:derivations:Reg:stReg:2:strategies:i}, \cref{lem:derivations:Reg:stReg:2:strategies:ii}, and \cref{lem:derivations:Reg:stReg:2:strategies:iii} hold.
	\par We establish the lemma only for the case of derivations in \stReg, since the case of derivations in \Reg can be treated analogously. So, let \Deriv be a derivation in \stReg with conclusion \prefixed{}{\M}.
	\par In a first step we decorate \Deriv with position labels such that a derivation \Derivlab with conclusion \labprefixed{}{\rootpos}{\M} in the variant proof system \stReglab in \cref{fig:stReglab} is obtained. Note that the decoration process can be carried out in a bottom-up manner, where the label in the conclusion of a rule instance determines the label in the premise(s) if that is not an already labelled term, and where in the case of instances of the rule of \FIX also the labels in marked assumptions are determined.
	\par In a second step we use the decorated version \Derivlab of \Deriv to define a history-aware strategy \m{\astrat_{\Deriv}} according to which the term \prefixed{}{\M} can be reduced as `prescribed' by \Derivlab. Since the derivations can only determine the strategy \m{\astrat_{\Deriv}} on terms occurring in \Deriv, we also have to define \m{\astrat_{\Deriv}} on other terms. This will be done by choosing an arbitrary (but here: history-free) \extscope-delimiting strategy \astrat for \RegARS, and basing the definition of \m{\astrat_{\Deriv}} on it.
	\par
	\newfunction\Terlab{\widehat{Ter}}
	\newcommand\srclab{\widehat{\src{}}}
	\newcommand\tgtlab{\widehat{\tgt{}}}
	We start by defining a labelling of \stRegARS as the ARS for which \m{\astrat_{\Deriv}} will be defined as a history-free strategy, which together with an initial labelling \m{l} then yields a history-aware strategy for \stRegARS. Assuming \m{\stRegARS = \tuple{\iTer{\lambdaprefixcal}, \steps, \src{}, \tgt{}}} as the formal representation of \stRegARS, we define the ARS
	\m{\stRegARSlab := \tuple{\Terlab{\lambdaprefixcal}, \stepslab, \srclab, \tgtlab}}
	where
	\begin{align*}
		&
		\Terlab{\lambdaprefixcal} := \setcompr{\labprefixed{l}{\vec{y}}{\N}}{\prefixed{\vec{y}}{\N} ∈ \iTer{\lambdaprefixcal}, {l} ∈ \set{0,1}^*} \\
		&
		\stepslab := \setcompr{\triple{\labprefixed{l}{\vec{y}}{\N}}{\astep}{\labprefixed{l'}{\vec{y}'}{\N'}}}{\text{the following holds}} \\
		& \hspace{5ex}
		\parbox{275pt}{there is an instance of \ruleref{λ}, \ruleref{@}, or \ruleref{\S} in \stReglab with \labprefixed{l}{\vec{y}}{\N} in the conclusion and the term \m{\labprefixed{l'}{\vec{y}'}{\N'}} in the premise, and with \m{\astep : \prefixed{\vec{y}}{\N} \m{\red_\streg} \prefixed{\vec{y}'}{\N'}} (one of) the corresponding step(s) in \stRegARS}
	\end{align*}
	and where \m{\srclab, \tgtlab : \stepslab → \Terlab{\lambdaprefixcal}} are defined as projections on the first, and respectively, the third component of the triples that constitute steps in \stepslab. Then the relation
	\begin{align*}
		\alabelling :=
		& \setcompr{\pair{\prefixed{y}{\N}}{\labprefixed{l}{y}{\N}}}{\prefixed{y}{\N} ∈ \iTer{\lambdaprefixcal}, {l} ∈ \set{0,1}^*} \\
		& {} ∪ \setcompr{\pair{\astep}{\triple{\prefixed{y}{\N}}{\astep}{\prefixed{y}{\N}}}}{\triple{\prefixed{y}{\N}}{\astep}{\prefixed{y}{\N}} ∈ \stepslab}
	\end{align*}
	is a labelling of \stRegARS to \stRegARSlab.
	As initial labelling we choose the function \m{l} that is defined by
	\m{l : \iTer{\lambdaprefixcal} → \Terlab{\lambdaprefixcal}, \prefixed{\vec{x}}{\M} ↦ \labprefixed{\rootpos}{\vec{x}}{\M}}.
	and which adds the label `\rootpos'.
	\par Now we define the strategy \m{\astrat_\Deriv} with:
	\begin{align*}
		& \astrat_\Deriv := \tuple{\Terlab{\lambdaprefixcal}, \stepslabon{\Derivlab} ∪ \stepslabnoton{\Derivlab}, \srclab', \tgtlab'} \\
		& \stepslabon{\Derivlab} := \setcompr{\triple{\labprefixed{l}{\vec{y}}{\N}}{\astep}{\labprefixed{l'}{\vec{y}'}{\N'}} ∈ \stepslab}{\text{the following holds}} \\
		& \hspace*{5ex} \parbox{275pt}{there is an instance of \ruleref{λ}, \ruleref{@}, or \ruleref{\S} in \Derivlab with \labprefixed{l}{\vec{y}}{\N} in the conclusion and the term \m{\labprefixed{l'}{\vec{y}'}{\N'}} in the premise, and with \m{\astep : \prefixed{\vec{y}}{\N} \m{\red_\streg} \prefixed{\vec{y}'}{\N'}} (one of) the corresponding step(s) in \stRegARS} \\
		& \stepslabnoton{\Derivlab} := \setcompr{\triple{\labprefixed{l}{\vec{y}}{\N}}{\astep}{\labprefixed{l'}{\vec{y}'}{\N'}} ∈ \stepslab}{\text{the following holds}} \\
		& \hspace*{5ex} \parbox{275pt}{\labprefixed{l}{\vec{y}}{\N} does not occur in \Derivlab, \astep is a step according to \astrat}
	\end{align*}
	where \m{\srclab'}, \m{\tgtlab'} are the appropriate restrictions of \srclab and \tgtlab.
	\par Note that, by its definition, \m{\astrat_\Deriv} is a sub-ARS of \stRegARSlab. Now for showing that \m{\astrat_\Deriv} is a (history-aware) strategy for \stRegARSlab, it has to be established that \m{\astrat_\Deriv} is a history-free strategy for the lifted version \stRegARSlab of \stRegARS. For this it remains to show that every normal form of \m{\astrat_\Deriv} is also a normal form of \stRegARSlab. So, let \m{\labprefixed{l}{\vec{y}}{\N} ∈ \iTer{\lambdaprefixcal}} be such that it is not a normal form of \stRegARSlab. Then also \prefixed{\vec{y}}{\N} is not a normal form of \stRegARS. For showing that there is a step in \m{\astrat_\Deriv} with this labelled term as a source we will distinguish two cases, i.e.\ whether \labprefixed{l}{\vec{y}}{\N} occurs on \Derivlab or not.
	\par For the second case, we assume that \labprefixed{l}{\vec{y}}{\N} does not occur in \Derivlab. Then there is a step \m{\astep : (\prefixed{\vec{y}}{\N}) \red_{\astrat} (\prefixed{\vec{y}'}{\N'})} in the scope-delimiting strategy \astrat in \stRegARS, which gives rise to the step \m{\astep : (\labprefixed{l}{\vec{y}}{\N}) \red (\labprefixed{l'}{\vec{y}'}{\N'})} in \stRegARSlab and in \m{\astrat_\Deriv}.
	\par For the first case, we assume that \labprefixed{l}{\vec{y}}{\N} occurs in \Derivlab, and we fix an occurrence \o. Since by assumption \labprefixed{l}{\vec{y}}{\N} is not a normal form of \stRegARSlab, \o cannot be the occurrence of an axiom (\0), and hence it is either an occurrence as the conclusion of an instance of one of the rules \ruleref{λ}, \ruleref{@}, \ruleref{\S} in \Derivlab, or as a marked assumption in \Derivlab. If \o is the conclusion of an instance \iota of \ruleref{λ}, \ruleref{@}, or \ruleref{\S}, then \iota defines a step on \labprefixed{l}{\vec{y}}{\N} which also is a step in \m{\astrat_\Deriv}. If \o is the conclusion of an instance of \FIX in \Deriv, then we consider an arbitrary path \apath in \Derivlab from \o upwards towards a leaf of \Derivlab. Since, due to the side-condition of the rule \FIX, immediate subderivations of instances of \FIX consist of at least one rule application, \apath cannot consist merely of applications of \FIX. Hence by following \apath from \o upwards, after a number of successive instances of \FIX, each of which have \labprefixed{l}{\vec{y}}{\N} as conclusion and premise, an instance of one of the rules \ruleref{λ}, \ruleref{@}, \ruleref{\S} follows, which witnesses a step with source \ruleref{λ}, \ruleref{@}, \ruleref{\S} in \stRegARSlab and in \m{\astrat_\Deriv}. Finally, if \o is an occurrence in a marked assumption at the top of the proof tree \Derivlab, then, since \Derivlab is a closed derivation and due to the form of instances of the assumption-discharging rule \FIX, there is also an occurrence \m{o'} of \labprefixed{l}{\vec{y}}{\N} as the conclusion of an instance of \FIX in \Derivlab. Now the argument above can be applied to the occurrence \m{o'} to obtain a step of \m{\astrat_\Deriv} on \labprefixed{l}{\vec{y}}{\N}.
	\par By construction \m{\astrat_\Deriv} conforms to \cref{lem:derivations:Reg:stReg:2:strategies:i} and \cref{lem:derivations:Reg:stReg:2:strategies:ii} because of the inclusion of \stepslabon{\Derivlab} and \stepslabnoton{\Derivlab} respectively; \cref{lem:derivations:Reg:stReg:2:strategies:iii} follows from \cref{lem:derivations:Reg:stReg:2:strategies:ii}.
\end{proof}

\begin{lemma}[from scope/\extscope-delimiting strategies to \Reg/\stReg-derivations]\label{lem:strategies:2:derivations:Reg:stReg}
	Let \m{\M ∈ \iTer{\lambdacal}}, and let \astrat be a scope-delimiting strategy for \RegARS (a \extscope-delimiting strategy for \stRegARS) such that \ST{\astrat}{\M} is finite. Then there exists a closed derivation \Deriv in \Reg (in \stRegzero, and hence in \stReg) with conclusion \prefixed{}{\M} such that the following properties hold (note the minor differences with the items \cref{lem:derivations:Reg:stReg:2:strategies:i}, \cref{lem:derivations:Reg:stReg:2:strategies:ii}, and \cref{lem:derivations:Reg:stReg:2:strategies:iii} in \cref{lem:derivations:Reg:stReg:2:strategies}):
	\begin{enumerate}[(i)]
		\item\label{lem:strategies:2:derivations:Reg:stReg:i} Every (non-cyclic) path in \Deriv from the conclusion upwards to a leaf of the proof tree \Deriv corresponds to a \m{\red_{\astrat}}-rewrite-sequence starting on \prefixed{}{\M} where passes over instances of the rules \ruleref{@} to the left and to the right correspond to \m{\red_{{\astrat}.@_0}}- and \m{\red_{{\astrat}.@_1}}-steps, respectively, and passes over instances of \ruleref{λ} and of \ruleref{\del} (of \ruleref{\S}) correspond to \m{\red_{{\astrat}.λ}}-, and \m{\red_{{\astrat}.\del}}-steps (\m{\red_{{\astrat}.\S}}-steps); passes from the conclusion to the premise of instances of \FIX correspond to empty steps.
		\item\label{lem:strategies:2:derivations:Reg:stReg:ii} Every sufficiently long \m{\red_{\astrat}}-rewrite-sequence on \prefixed{\vec{x}}{\M} has an initial segment that corresponds to a (non-cyclic) path in \Deriv from the conclusion upwards to a leaf of the proof tree: thereby a \m{\red_{{\astrat}.@_0}}-step or \m{\red_{{\astrat}.@_1}}-step corresponds to a pass over (possibly some \FIX-instances followed by) an instance of \ruleref{@} in direction left and right, respectively; a \m{\red_{{\astrat}.λ}}-step or \m{\red_{{\astrat}.\del}}-step (\m{\red_{{\astrat}.\S}}-step) corresponds to a pass over (possibly some \FIX-instances followed by) an instance of \ruleref{λ} or \ruleref{\del} (of \ruleref{\S}), respectively.
		\item\label{lem:strategies:2:derivations:Reg:stReg:iii} \m{\ST{\astrat}{\M} \subseteq \setcompr{ \prefixed{\vec{y}}{\N}}{\text{the term }\prefixed{\vec{y}}{\N}\text{ occurs in }\Deriv}}.
	\end{enumerate}
\end{lemma}

\begin{proof}{\label{prf:lem:strategies:2:derivations:Reg:stReg}}
	We will argue only for the part of the statement of the lemma concerning a \extscope-delimiting strategy for \stRegARS, since the case with a scope-delimiting strategy for \RegARS can be established analogously.
	\par Let \M be an λ-term, and let \astrat be a \extscope-delimiting strategy for \RegARS such that \ST{\astrat}{\M} is finite. Now let \m{\Deriv_0} be the (trivial) derivation with conclusion \prefixed{}{\M}, which, in case that this is not an axiom of \stRegzero (and \stReg), is also an assumption, and then is of the form \m{(\prefixed{}{\M})^u}, carrying an assumption marker \u. If \m{\Deriv_0} is an axiom, then it is easy to verify that the statements \cref{lem:strategies:2:derivations:Reg:stReg:i}, \cref{lem:strategies:2:derivations:Reg:stReg:ii}, and \cref{lem:strategies:2:derivations:Reg:stReg:iii} hold.
	\par Otherwise we construct a sequence \m{\Deriv_1, \Deriv_2, \dots} of derivations where each \m{\Deriv_n} satisfies the properties \cref{lem:strategies:2:derivations:Reg:stReg:i}, \cref{lem:strategies:2:derivations:Reg:stReg:ii}, and \cref{lem:strategies:2:derivations:Reg:stReg:iii}, where terms in marked assumptions are not also terms in axioms \0, and where \m{\Deriv_{n+1}} extends \m{\Deriv_n} by one additional rule instance above a marked assumption in \m{\Deriv_n}: For the extension step on a derivation \m{\Deriv_n}, a marked assumption \m{(\prefixed{\vec{y}}{\N})^{\amarker}} in \m{\Deriv_n} is picked with the property that the term \prefixed{\vec{y}}{\N} does not appear in the thread down to the conclusion of \m{\Deriv_n}.
	\par Suppose that the \m{\red_{\astrat}}-rewrite-sequence from the conclusion of \m{\Deriv_n} up to the marked assumption is of the form:
	\[\arewseq : \prefixed{}{\M} = \prefixed{\vec{x}_0}{\M_0} \red_{\astrat} \prefixed{\vec{x}_1}{\M_1} \red_{\astrat} \dots \red_{\astrat} \prefixed{\vec{x}_m}{\M_m} = \prefixed{\vec{y}}{\N}\]
	Note that, since by assumption \prefixed{\vec{y}}{\N} is not a term in an axiom \0 of \stRegzero, it follows by \cref{prop:rewprops:RegCRS:stRegCRS}~\cref{prop:rewprops:RegCRS:stRegCRS:v}, that it is not a \m{\red_\streg}-normal-form. Then depending on whether the possible next step(s) in an \m{\red_{\astrat}}-rewrite that extends \arewseq by one step is a \m{\red_{{\astrat}.λ}}-, \m{\red_{{\astrat}.\del}}-step, or either a \m{\red_{{\astrat}.@_0}}-steps or a \m{\red_{{\astrat}.@_1}}-steps, the derivation \m{\Deriv_n} is extended above the marked assumption \m{(\prefixed{\vec{y}}{\N})^{\amarker}} by an application of λ, \S, or @, respectively. For example in the case that \arewseq extends by one additional step to either of the two rewrite sequences:
	\begin{align*}
		\arewseq_i : \prefixed{}{\M} = \prefixed{\vec{x}_0}{\M_0} \red_{\astrat} \dots \red_{\astrat} & \prefixed{\vec{x}_m}{\M_m} \\ = & \prefixed{\vec{x}_m}{\app{\M_{m,0}}{\M_{m,1}}} \red_{{\astrat}.@_i} \prefixed{\vec{x}_m}{\M_{m,i}}
	\end{align*}
	with \m{i ∈ \set{0,1}}, the derivation \m{\Deriv_n} of the form:
	\begin{prooftree}
		\axiom{\langle (\prefixed{\vec{x}_m}{\app{\M_{m,0}}{\M_{m,1}}})^{\amarker} \rangle}
		\noLine
		\unaryInf{\Deriv_n}
		\noLine
		\unaryInf{\prefixed{}{\M}}
	\end{prooftree}
	is extended to \m{\Deriv_{n+1}}:
	\begin{prooftree}
		\axiom{(\prefixed{\vec{x}_m}{\M_{m,0}})^{\amarker_0}}
		\axiom{(\prefixed{\vec{x}_m}{\M_{m,1}})^{\amarker_1}}
		\insertBetweenHyps{\hspace*{1em}}
		\infLabel{@}
		\binaryInf{\langle \prefixed{\vec{x}_m}{\app{\M_{m,0}}{\M_{m,1}}} \rangle}
		\noLine
		\unaryInf{\Deriv_n}
		\noLine
		\unaryInf{\prefixed{}{\M}}
	\end{prooftree}
	for two fresh assumption markers \m{\amarker_0} and \m{\amarker_1} (the angle brackets \m{\langle \dots \rangle} are used here to indicate just a single formula occurrence at the top of the proof tree \m{\Deriv_n}). If either of \prefixed{\vec{x}_m}{\M_{m,0}} or \prefixed{\vec{x}_m}{\M_{m,1}} is an axiom, then the assumption marker is removed and the formula is marked as an axiom \0, accordingly. Note that, if the statements \cref{lem:strategies:2:derivations:Reg:stReg:i}, \cref{lem:strategies:2:derivations:Reg:stReg:ii}, and \cref{lem:strategies:2:derivations:Reg:stReg:iii} are satisfied for \m{\Deriv = \Deriv_n}, then this is also the case for \m{\Deriv = \Deriv_{n+1}}. Furthermore, terms in marked assumptions are not terms in axioms of \stRegzero.
	\par The extension process continues as long as \m{\Deriv_n} contains a marked assumption \m{(\prefixed{\vec{y}}{\N})^{\bmarker}} without a `\stRegzero-admissible repetition' beneath it, by which we mean the occurrence \o of the term \prefixed{\vec{y}}{\N} on the thread down to the conclusion in \Deriv, but strictly beneath the marked assumption, such that furthermore all terms on the part of the thread down to \o have an abstraction prefix of length greater or equal to \length{\vec{y}}. (Note the connection to the side-condition on instances of the rule \FIX in \stRegzero, and, in particular, that marked assumptions with an \stRegzero-admissible repetition beneath it could be discharged by an appropriately introduced instance of \FIX in \stRegzero.)
	\par That the extension process terminates can be seen as follows: Suppose that, to the contrary, it continues indefinitely. Then, since the derivation size increases strictly in every step, an infinite proof tree \m{\Deriv^{\infty}} is obtained in the limit, which due to finite branchingness of the proof tree and Kőnig's Lemma possesses an infinite path \apath starting at the conclusion. Now note that due to \cref{lem:strategies:2:derivations:Reg:stReg:i}, \apath corresponds to an infinite \m{\red_{\astrat}}-rewrite-sequence. Due to \cref{prop:grounded:cycle}, this infinite rewrite sequence must contain a grounded cycle. However, the existence of such grounded cycle contradicts the termination condition of the extension process, because every grounded cycle provides an \stRegzero-admissible repetition.
	\par Let \m{\Deriv_N}, for some \m{N ∈ ℕ}, be the derivation that is reached when no further extension step, as described, is possible. By the construction the statements \cref{lem:strategies:2:derivations:Reg:stReg:i}, \cref{lem:strategies:2:derivations:Reg:stReg:ii}, and \cref{lem:strategies:2:derivations:Reg:stReg:iii} are satisfied for \m{\Deriv = \Deriv_N}. Furthermore, \m{\Deriv_N} is a derivation in \stReg and \stRegzero in which every leaf at the top is either an axiom \0 or an assumption \m{(\prefixed{\vec{y}}{\N})^{\amarker}} marked with a unique marker \amarker, and for every such marked assumption in \m{\Deriv_N}, there is a \stRegzero-admissible repetition strictly beneath it. This fact enables us to modify \m{\Deriv_N} into a closed derivation in \stRegzero by closing all open assumptions by newly introduced applications of the rule \FIX. More precisely, steps of the following kind are carried out repeatedly. A derivation with occurrences of a number of marked assumptions \m{(\prefixed{\vec{y}}{\N})^{\amarker_i}} highlighted together with a single occurrence of the term \prefixed{\vec{y}}{\N} in its interior that indicates the \stRegzero-admissible repetition for the displayed marked assumptions:
	\begin{prooftree}
		\axiom{\langle (\prefixed{\vec{y}}{\N})^{\amarker_1} \rangle}
		\insertBetweenHyps{\dots}
		\axiom{\langle (\prefixed{\vec{y}}{\N})^{\amarker_k} \rangle}
		\noLine
		\binaryInf{\Deriv_{000}}
		\noLine
		\unaryInf{\langle \prefixed{\vec{y}}{\N} \rangle}
		\noLine
		\unaryInf{\Deriv_{00}}
		\noLine
		\unaryInf{\prefixed{}{\M}}
	\end{prooftree}
	is modified into:
	\begin{prooftree}
		\axiom{\langle (\prefixed{\vec{y}}{\N})^{\cmarker} \rangle}
		\insertBetweenHyps{\dots}
		\axiom{\langle (\prefixed{\vec{y}}{\N})^{\cmarker} \rangle}
		\noLine
		\binaryInf{\Deriv_{000}}
		\noLine
		\unaryInf{\prefixed{\vec{y}}{\N}}
		\infLabel{\FIX, \cmarker}
		\unaryInf{\langle \prefixed{\vec{y}}{\N} \rangle}
		\noLine
		\unaryInf{\Deriv_{00}}
		\noLine
		\unaryInf{\prefixed{}{\M}}
	\end{prooftree}
	where \cmarker is a fresh assumption marker. In every such transformation step the number of open assumptions is strictly decreased, but the properties \cref{lem:strategies:2:derivations:Reg:stReg:i}, \cref{lem:strategies:2:derivations:Reg:stReg:ii}, and \cref{lem:strategies:2:derivations:Reg:stReg:iii} (for \Deriv the resulting derivation of such a step) is preserved. Hence after finitely many such transformation steps a derivation \Deriv in \stRegzero without open assumptions and with the properties \cref{lem:strategies:2:derivations:Reg:stReg:i}, \cref{lem:strategies:2:derivations:Reg:stReg:ii}, and \cref{lem:strategies:2:derivations:Reg:stReg:iii} is reached, and obtained as the result of this construction.
\end{proof}

As a consequence of the two lemmas above, derivability in \stReg and in \stRegzero coincides:

\begin{proposition}[\stReg ≡ \stRegzero ]\label{prop:stReg:stRegzero}
	\[∀ \M ∈ \iTer{\lambdacal}~~ \derivablein{\stReg}{\prefixed{}{\M}} ~\Longleftrightarrow~ \derivablein{\stRegzero}{\prefixed{}{\M}}\]
\end{proposition}

\begin{proof}
	The direction ``⇐'' follows by the fact that every derivation in \stRegzero is also a derivation in \stReg. For the direction ``⇒'', let \M be an infinite term such that \derivablein{\stReg}{\prefixed{}{\M}}. By \cref{lem:derivations:Reg:stReg:2:strategies} there exists a \extscope-delimiting strategy \astratplus such that \ST{\astratplus}{\M} is finite. But then it follows by \cref{lem:strategies:2:derivations:Reg:stReg} that there is also a closed derivation in \stRegzero with conclusion \prefixed{}{\M}, and hence that \derivablein{\stRegzero}{\prefixed{}{\M}}.
\end{proof}

Now we have assembled all auxiliary statements that we use for proving a theorem that tightly links derivability in the proof system \Reg with regularity, and derivability in \stReg and in \stRegzero with strong regularity, of λ-terms.

\begin{theorem}\label{lem:Reg:stReg}
	The following statements hold for the proof systems 
	\Reg, \stReg, \stRegzero:
	\begin{enumerate}[(i)]
		\item\label{lem:Reg:stReg:i}
			\Reg is sound and complete for regular λ-terms. That is, for all \m{\M ∈ \iTer{\lambdacal}} it holds:
			\begin{equation*}
			\derivablein{\Reg}{\prefixed{}{\M}}
			\hspace*{5ex}\text{if and only if}\hspace*{5ex}
			\text{\M is regular.}
			\end{equation*}
		\item\label{lem:Reg:stReg:ii}
			\stReg and \stRegzero are sound and complete for strongly regular λ-terms. That is,
		for all \m{\M ∈ \iTer{\lambdacal}} the following statements are equivalent:
		\begin{enumerate}[(a)]
			\item \M is strongly regular.
			\item \derivablein{\stReg}{\prefixed{}{\M}}.
			\item \derivablein{\stRegzero}{\prefixed{}{\M}}.
		\end{enumerate}
	\end{enumerate}
\end{theorem}

\begin{proof}
	Since the proof of statement \cref{lem:Reg:stReg:ii} of the theorem can be carried out analogously (taking into account \cref{prop:stReg:stRegzero}), we argue here only for statement \cref{lem:Reg:stReg:i}.
	\par For ``⇒'' in \cref{lem:Reg:stReg:i}, let \M be a λ-term that is regular. Then there exists a scope-delimiting strategy \astrat on \RegARS such that \ST{\astrat}{\M} is finite. By \cref{lem:strategies:2:derivations:Reg:stReg} it follows that there exists a closed derivation \Deriv in \Reg with conclusion \prefixed{}{\M}. This derivation witnesses \derivablein{\Reg}{\prefixed{}{\M}}.
	For ``⇐'' in \cref{lem:Reg:stReg:i}, suppose that \derivablein{\Reg}{\prefixed{}{\M}}. Then there exists a closed derivation \Deriv in \Reg with conclusion \prefixed{}{\M}. Now \cref{lem:derivations:Reg:stReg:2:strategies} entails the existence of a scope-delimiting strategy \astrat in \stRegARS such that, in particular, \ST{\astrat}{\M} is finite. This fact implies that \M is regular.
\end{proof}

\begin{figure}
	\proofsystem{
		\begin{bprooftree}
			\emptyAxiom
			\infLabel{\0}
			\unaryInf{\prefixedCRS{n+1}{\vecOneToN{x}y}{y} = \prefixedCRS{n+1}{\vecOneToN{z}u}{u}}
		\end{bprooftree}
		\\
		\begin{bprooftree}
			\axiom{\prefixedCRS{n}{\vecOneToN{x}}{\apreter} = \prefixedCRS{n}{\vecOneToN{z}}{\bpreter}}
			\infLabel{\S~~\mlSideCondition{\y does not occur \\ free in \apreter; \w does \\ not occur free in \bpreter}}
			\unaryInf{\prefixedCRS{n+1}{\vecOneToN{x}y}{\apreter} = \prefixedCRS{n+1}{\vecOneToN{z}w}{\bpreter}}
		\end{bprooftree}
		\\
		\begin{bprooftree}
			\axiom{\prefixedCRS{n+1}{\vecOneToN{x}y}{\apreter} = \prefixedCRS{n+1}{\vecOneToN{z}u}{\bpreter}}
			\infLabel{λ}
			\unaryInf{\prefixedCRS{n}{\vecOneToN{x}}{\absCRS{y}{\apreter}} = \prefixedCRS{n}{\vecOneToN{z}}{\absCRS{u}{\bpreter}}}
		\end{bprooftree}
		\\
		\begin{bprooftree}
			\axiom{\prefixedCRS{n}{\vecOneToN{x}}{\apreter_0} = \prefixedCRS{n}{\vecOneToN{y}}{\bpreter_0}\hspace{9ex}}
			\noLine
			\unaryInf{\hspace{9ex}\prefixedCRS{n}{\vecOneToN{x}}{\apreter_1} = \prefixedCRS{n}{\vecOneToN{y}}{\bpreter_1}}
			\infLabel{@}
			\unaryInf{\prefixedCRS{n}{\vecOneToN{x}}{\app{\apreter_0}{\apreter_1}} = \prefixedCRS{n}{\vecOneToN{y}}{\app{\bpreter_0}{\bpreter_1}}}
		\end{bprooftree}
	}
	\caption{Proof system \AlphaPreTer for equality of preterms in \lambdaprefixcal modulo \alphaequiv.}
	\label{fig:AlphaPreTer}
\end{figure}

\begin{figure}
	\proofsystem{
		\begin{bprooftree}
			\emptyAxiom
			\infLabel{\0}
			\unaryInf{\prefixed{\vec{x}y}{y} = \prefixed{\vec{z}u}{u}}
		\end{bprooftree}
		\hsep
		\begin{bprooftree}
			\axiom{\prefixed{\vec{x}y}{\M} = \prefixed{\vec{z}u}{\N}}
			\infLabel{λ}
			\unaryInf{\prefixed{\vec{x}}{\abs{y}{\M}} = \prefixed{\vec{z}}{\abs{u}{\N}}}
		\end{bprooftree}
		\\
		\begin{bprooftree}
			\axiom{\prefixed{\vec{x}}{\M} = \prefixed{\vec{z}}{\N}}
			\infLabel{\S~~\mlSideCondition{if \y does not occur in \M, \\ and \w does not occur in \N}}
			\unaryInf{\prefixed{\vec{x}y}{\M} = \prefixed{\vec{z}w}{\N}}
		\end{bprooftree}
		\\
		\begin{bprooftree}
			\axiom{\prefixed{\vec{x}}{\M_0} = \prefixed{\vec{y}}{\N_0}}
			\axiom{\prefixed{\vec{x}}{\M_1} = \prefixed{\vec{y}}{\N_1}}
			\infLabel{@}
			\binaryInf{\prefixed{\vec{x}}{\app{\M_0}{\M_1}} = \prefixed{\vec{y}}{\app{\N_0}{\N_1}}}
		\end{bprooftree}
	}
	\caption{Proof system \EqTer for equality of terms in \lambdaprefixcal in informal notation.}
	\label{fig:EqTer}
\end{figure}

For defining a proof system for equality of strongly regular λ-terms, we first give a specialised version of the proof system \siCRSPretermAlpha{\Kahrs} (for α-equivalence of iCRS preterms) for \lambdaprefixcal-preterms, and a corresponding system on \lambdaprefixcal-terms.

\begin{definition}[proof systems \AlphaPreTer, \EqTer]\label{def:AlphaPreTer:EqTer}
	The proof system \AlphaPreTer for α-equivalence of infinite preterms in \lambdaprefixcal consists of the rules displayed in \cref{fig:AlphaPreTer}.
	The proof system \EqTer for equality of infinite terms in \lambdaprefixcal consists of the rules displayed in \cref{fig:EqTer}.
	Provability in \AlphaPreTer and in \EqTer is defined, analogous to the proof system \siCRSPretermAlpha{\Kahrs} from \cref{def:iCRSPretermAlphaKahrs}, as the existence of a completed (possibly infinite) derivation, and will, as done for \siCRSPretermAlpha{\Schroer} in \cref{def:iCRSPretermAlphaSchroer}, be indicated using the symbol \sinfderivable.
\end{definition}

\begin{proposition}\label{prop:AlphaPreTer:EqTer}
	The following statements hold for the proof systems \AlphaPreTer and \EqTer.
	\begin{enumerate}[(i)]
		\item\label{prop:AlphaPreTer:EqTer:i} \AlphaPreTer is sound and complete for \alphaequiv on \lambdaprefixcal-preterms. That is, the following holds for all closed preterms \apreter and \bpreter in \lambdaprefixcal:
		\begin{equation*}
			\infderivablein{\AlphaPreTer}{\prefixedCRS{0}{}{\apreter} = \prefixedCRS{0}{}{\bpreter}}
			\hspace*{4ex}\text{if and only if}\hspace*{4ex}
			\apreter \alphaequiv \bpreter .
		\end{equation*}
		\item\label{prop:AlphaPreTer:EqTer:ii} \EqTer is sound and complete for equality between \lambdaprefixcal-terms. That is, the following holds for all terms \M and \N in \lambdaprefixcal:
		\begin{equation*}
			\infderivablein{\EqTer}{\M = \N}
			\hspace*{4ex}\text{if and only if}\hspace*{4ex}
			\M = \N .
		\end{equation*}
	\end{enumerate}
\end{proposition}

\begin{proof}
	For statement \cref{prop:AlphaPreTer:EqTer:i} it suffices to show that, for an equation \m{\prefixedCRS{0}{}{\apreter} = \prefixedCRS{0}{}{\bpreter}} between preterms of \lambdaprefixcal, derivability in \AlphaPreTer coincides with derivability of this equation in the general proof system \siCRSPretermAlpha{\Kahrs} for α-equivalence between iCRS preterms in \cref{def:iCRSPretermAlphaKahrs}. Given a derivation \infDeriv in \AlphaPreTer, a derivation \m{\infDeriv_\Kahrs} in \siCRSPretermAlpha{\Kahrs} results by replacing each formula occurrence \m{\prefixedCRS{n}{\vecOneToN{x}}{\apreter} = \prefixedCRS{n}{\vecOneToN{y}}{\bpreter}} by the formula occurrence \m{\KahrsCRSabs{\vecOneToN{x}}{\apreter} = \KahrsCRSabs{\vecOneToN{y}}{\bpreter}} and adding an instance of the rule for the function symbol \sPrefixedCRS{0} at the bottom. Then instances of the axioms and rules (\0), (\m{@}), \ruleref{λ}, and \ruleref{\S} in \infDeriv correspond to instances of axioms and rules (\0), (\sappCRS), (\m{\CRSabs{\hspace*{1pt}}{}}), and \ruleref{\S} in \m{\infDeriv_\Kahrs}, respectively. This proof transformation also has an inverse.
	\par For ``⇐'' in statement \cref{prop:AlphaPreTer:EqTer:ii} it suffices to note: Every \extscope-delimiting strategy for \stRegARS can be used to stepwise extend finite derivations in \EqTer with conclusion \prefixed{}{\M} = \prefixed{}{\M} by one additional rule application above a leaf containing a formula \prefixed{y}{\N} = \prefixed{y}{\N} that is not an axiom \0, which implies that \prefixed{y}{\N} is not a normal form of \m{\red_\streg}. If these extensions are carried out in a fair manner by extending all non-axiom leafs at depth \n in the proof tree before proceeding with leafs at depth \m{>n}, then in the limit a completed derivation in \EqTer is obtained.
	\par For ``⇒'' in statement \cref{prop:AlphaPreTer:EqTer:ii}, suppose that \infDeriv is a completed derivation in \EqTer with conclusion \prefixed{}{\M} = \prefixed{}{\N}. Let \prefixedCRS{0}{}{\apreter} and \prefixedCRS{0}{}{\bpreter} be preterm representatives of \prefixed{}{\M} and \prefixed{}{\N}, respectively. Now a completed derivation \m{\infDeriv_{\text{pter}}} in \AlphaPreTer can be found by developing it step by step from the conclusion \m{\prefixedCRS{0}{}{\apreter} = \prefixedCRS{0}{}{\bpreter}} upwards, parallel to \infDeriv, and following the rules of \AlphaPreTer, which are invertible (that is, the premises of a rule instance are uniquely determined by the conclusion). Then \m{\infDeriv_{\text{pter}}} is a preterm representative version of \infDeriv. By using \cref{prop:AlphaPreTer:EqTer:i}, it follows that \m{\apreter \alphaequiv \bpreter}. Since \apreter and \bpreter are preterm representatives of \M and \N, respectively, \m{\M = \N} follows.
\end{proof}

\begin{definition}[the proof system \stRegeq]\label{def:stRegeq}
	The natural-deduction-style proof system \stRegeq for equality of strongly regular λ-terms has all the rules of the proof system \EqTer from \cref{def:AlphaPreTer:EqTer} and \cref{fig:EqTer}, and additionally, the rule \FIX in \cref{fig:stRegeq}. But contrary to the definition in \EqTer, provability of an equation \prefixed{\vec{x}}{\M} = \prefixed{\vec{x}}{\N} in \stRegeq is defined as the existence of a finite closed derivation with conclusion \prefixed{\vec{x}}{\M} = \prefixed{\vec{x}}{\N}.
\end{definition}

\begin{figure}
	\proofsystem{
		\begin{bprooftree}
			\axiom{[\prefixed{\vec{x}}{\M} = \prefixed{\vec{y}}{\N}]^u}
			\noLine
			\unaryInf{\Deriv_0}
			\noLine
			\unaryInf{\prefixed{\vec{x}}{\M} = \prefixed{\vec{y}}{\N}}
			\infLabel{\FIX,u~~\sideCondition{if \m{\depth{\Deriv_0} ≥ 1}}}
			\unaryInf{\prefixed{\vec{x}}{\M} = \prefixed{\vec{y}}{\N}}
		\end{bprooftree}
	}
	\caption{The rule \FIX, which is added to the rules of \EqTer from \cref{fig:EqTer} in order to obtain the proof system \stRegeq for equality of strongly regular λ-terms.}
	\label{fig:stRegeq}
\end{figure}

\begin{theorem}\label{thm:stRegeq}
	\stRegeq is sound and complete for equality between strongly regular λ-terms. That is, for all strongly regular λ-terms \M and \N it holds: \[ \derivablein{\stRegeq}{\M = \N} \hspace{4ex}\text{if and only if}\hspace{4ex} \M = \N . \]
\end{theorem}

\begin{proof}[Proof sketch]\label{prf:lem:stRegeq}
	Let \M and \N be strongly regular λ-terms. In view of \cref{prop:AlphaPreTer:EqTer}~\cref{prop:AlphaPreTer:EqTer:ii}, it suffices to show:
	\begin{equation}\label{eq:prf:lem:stRegeq}
		\derivablein{\stRegeq}{\M = \N} \hspace{4ex}\text{if and only if}\hspace{4ex} \infderivablein{\EqTer}{\M = \N} .
	\end{equation}
	\par For showing ``⇐'' in \cref{eq:prf:lem:stRegeq}, let \infDeriv be a derivation in \EqTer with conclusion \prefixed{}{\M} = \prefixed{}{\N}. Since paths in \infDeriv correspond to \m{\red_\streg}-rewrite-sequences, and since the number of generated subterms of both \M and \N are finite (as a consequence of their strong regularity), on every infinite thread equation repetitions occur. These repetitions can be used to cut all infinite threads by appropriate introductions of instances of \FIX in order to obtain a finite and closed derivation in \stRegeq with the same conclusion.
	\par For showing ``⇒'' in \cref{eq:prf:lem:stRegeq}, let \Deriv be a closed derivation in \stRegeq with conclusion \prefixed{}{\M} = \prefixed{}{\N}. Now \Deriv can be unfolded into an infinite derivation \infDeriv in \EqTer by repeatedly removing a bottommost instance of \FIX and inserting its immediate subderivation above each of the marked assumptions the instance discharges. If this process is organised in a fair manner with respect to bottommost instances of \FIX, then in the limit an infinite completed proof tree with conclusion \prefixed{}{\M} = \prefixed{}{\N} in \EqTer is obtained. For productivity of this process it is decisive that the side-condition on every instance \ainst of the rule \FIX guarantees that on threads from the conclusion of \ainst to a marked assumption that is discharged by \ainst at least one instance of a rule different from \FIX is passed.
\end{proof}

\begin{figure}
	\proofsystem{
		\begin{bprooftree}
			\axiom{\set{[\prefixed{\vec{x}}{\aconstname_{f_i}} ]^{\amarker_i}}_{i=1,\dots,n}}
			\noLine
			\unaryInf{\Deriv_j}
			\noLine
			\unaryInf{\set{\dots ~ \prefixed{\vec{x}}{\subst{\L_j}{\vec{f}}{\vec{\aconstname}_{\vec{f}}}} ~ \dots}_{j=0,\dots,n}}
			\infLabel{\FIXletrec, \amarker_1,\dots,\amarker_n}
			\unaryInf{\prefixed{\vec{x}}{\letin{f_1 = \L_1 \dots f_n = \L_n}{\L_{0}}}}
		\end{bprooftree}
		\proofpar where \m{\aconstname_{f_1}, \dots,\aconstname_{f_n}} are distinct constants fresh for \m{\L_1,\dots,\L_n}, and substitutions \m{\subst{\L_j}{\vec{f}}{\vec{\aconstname}_{\vec{f}}}} stands for \m{\L_j [f_1 := \aconstname_{f_1}, \dots, f_n := \aconstname_{f_n}]}.
		\proofpar \emph{side-conditions}: \m{\length{\vec{y}} ≥ \length{\vec{x}}} holds for the prefix length of every \prefixed{\vec{y}}{\N} on a thread in \m{\Deriv_j} for \m{j ∈ \set{0,\dots,n}} from an open assumptions \m{( \prefixed{\vec{x}}{\aconstname_{f_i}})^{\amarker_i}} downwards; for bottommost instances: the arising derivation is guarded on access path cycles.
	}
	\caption{The rule \FIXletrec for the natural-deduction style proof systems \Regletrec and \stRegletrec on \lambdaletrec-terms.}
	\label{fig:Regletrec:stRegletrec}
\end{figure}

For the purpose of the following definition and the respective proof system in \cref{fig:Regletrec:stRegletrec} we extend the signature \sigCRSletrec of \Regletrec and \stRegletrec by an infinite set of constants for which we use the symbol \aconstname as syntactical variables which frequently carry index subscripts.


\begin{definition}[proof systems \Regletrec, \stRegletrec, and \ugRegletrec, \ugstRegletrec]\label{def:Regletrec:stRegletrec}
	The proof systems \stRegletrec and \Regletrec for \lambdaletrec-terms arise from the proof systems \stReg and \Reg (see \cref{def:Reg:stReg:stRegzero}, \cref{fig:stReg:stRegzero} and \cref{fig:Reg}), respectively, by replacing the terms in the axioms (\0), and rules \ruleref{λ}, \ruleref{@}, \ruleref{\S} and \ruleref{\del} through \lambdaletrec-terms with abstraction prefixes accordingly, and by replacing the rule \FIX with the rule \FIXletrec in \cref{fig:Regletrec:stRegletrec}. The side-condition concerning access path cycles on the derivation arising by an instance of \FIXletrec pertains only to bottommost occurrences of this rule, and is explained below. By \FIXletrecmin we mean the variant of the rule \FIXletrec in which the side-condition concerning guardedness of the arising derivation on access path cycles has been dropped. By \ugRegletrec/\ugstRegletrec we denote the variants of \Regletrec/\stRegletrec, respectively, in which the rule \FIXletrec is replaced by the rule \FIXletrecmin.
	\par Let \Deriv be a derivation in one of these proof systems. By an \emph{access path} of \Deriv we mean a (possibly cyclic) path \apath in \Deriv such that:
	\begin{enumerate}[(a)]
		\item\label{access:path:i} \apath starts at the conclusion and can proceed in upwards direction;
		\item\label{access:path:ii} at instances of \ruleref{@}, \apath can step from the conclusion to one of the premises;
		\item\label{access:path:iii} at instances of \FIXletrec, \apath can step from the conclusion to the rightmost premise (which corresponds to the body of the \Let-expression);
		\item\label{access:path:iv} when arriving at a marked assumption \m{(\prefixed{\vec{x}}{\aconstname_{f_i}})^{\amarker_i}} that is discharged at an application of \FIXletrec of the form as displayed in \cref{fig:Regletrec:stRegletrec}, \apath can step over to the conclusion \prefixed{\vec{x}}{\subst{\L_i}{\vec{f}}{\vec{\aconstname}_{\vec{f}}}} of the subderivation \m{\Deriv_i} of that application of \FIXletrec, and proceed from there, again in upwards direction.
	\end{enumerate}
	For every formula occurrence \o in \Deriv, by a \emph{relative access path} from \o we mean a path with the properties \cref{access:path:ii}--\cref{access:path:iv} that starts at \o and proceeds in upwards direction. An access path (or relative access path) in \Deriv is \emph{cyclic} if there is a formula occurrence in \Deriv that is visited more than once.
	\par Now we say that \Deriv is \emph{guarded on access path cycles} if every cyclic access path contains, on each of its cycles, at least one \emph{guard}, that is, an instance of a rule \ruleref{λ} or \ruleref{@}. We say that \Deriv is \emph{guarded} if every relative access path contains a guard on each of its cycles.
\end{definition}

\begin{example}
	The \lambdaletrec-term \m{\L_1 = \letin{f = f \eqsep g = \abs{x}{x}}{\abs{y}{g}}} admits the following closed derivation \m{\Deriv_1} in \Regletrec/\stRegletrec:
	\begin{prooftree}
		\axiom{(\prefixed{}{\aconstname_g})^{\amarker_2}}
		\infLabel{\del/\S}
		\unaryInf{\prefixed{y}{\aconstname_g}}
		\infLabel{λ}
		\unaryInf{\prefixed{}{\abs{y}{\aconstname_g}}}
		\axiom{(\prefixed{}{\aconstname_f})^{\amarker_1}}
		\emptyAxiom
		\infLabel{\0}
		\unaryInf{\prefixed{x}{x}}
		\infLabel{λ}
		\unaryInf{\prefixed{}{\abs{x}{x}}}
		\infLabel{\FIXletrec, \amarker_1, \amarker_2}
		\ternaryInf{\prefixed{}{\letin{f = f \eqsep g = \abs{x}{x}}{\abs{y}{g}}}}
	\end{prooftree}
	This derivation can be built in a straightforward way, from the bottom upwards. Note that \Deriv is guarded on access path cycles, and hence that the instance of \FIXletrec at the bottom is a valid one, because: \Deriv does not possess any cyclic access paths. In particular, there is no access path in \m{\Deriv_1} that reaches the first premise of the instance of \FIXletrec: this premise is the starting point of an unguarded relative access path, which entails that \m{\Deriv_1} itself is not guarded.
	\par Now consider the \lambdaletrec-term \m{\L_2 = \letin{f = f \eqsep g = \abs{x}{x}}{\abs{y}{\app{f}{g}}}}. When trying to construct a derivation in \Regletrec/\stRegletrec for this term in a bottom-up manner, one arrives at the closed derivation \m{\Deriv_2} in \ugRegletrec/\ugstRegletrec: 
	\begin{prooftree}
		\axiom{(\prefixed{}{\aconstname_f})^{\amarker_1}}
		\infLabel{\del/\S}
		\unaryInf{\prefixed{y}{\aconstname_f}}
		\axiom{(\prefixed{}{\aconstname_g})^{\amarker_2}}
		\infLabel{\del/\S}
		\unaryInf{\prefixed{y}{\aconstname_g}}
		\infLabel{@}
		\binaryInf{\prefixed{y}{\app{\aconstname_f}{\aconstname_g}}}
		\infLabel{λ}
		\unaryInf{\prefixed{}{\abs{y}{\app{\aconstname_f}{\aconstname_g}}}}
		\axiom{\hspace{-7ex}(\prefixed{}{\aconstname_f})^{\amarker_1}}
		\emptyAxiom
		\infLabel{\0}
		\unaryInf{\prefixed{x}{x}}
		\infLabel{λ}
		\unaryInf{\hspace{-1ex}\prefixed{}{\abs{x}{x}}}
		\infLabel{\FIXletrecmin, \amarker_1, \amarker_2}
		\ternaryInf{\prefixed{}{\letin{f = f \eqsep g = \abs{x}{x}}{\abs{y}{\app{f}{g}}}}}
	\end{prooftree}
	However, \m{\Deriv_2} is not a valid derivation in \Regletrec/\stRegletrec, as the inference step at the bottom is
	an instance of \FIXletrecmin, but not of \FIXletrec,
	because the side-condition on the arising derivation to be guarded on access path cycles is not satisfied:
	now there is an access path that reaches the first premise of the derivation and that continues looping on
	this an unguarded cycle.
	Since the bottom-up search procedure for derivations is deterministic in this case,
	it follows that \prefixed{}{\L_2} is not derivable in \Reg nor in \stReg.
\end{example}

Similar as the correspondence, stated by \cref{prop:derivationpaths:2:rewritesequences:Reg:stReg},
between (possibly cyclic) paths in a derivation in \Reg and \stReg starting at the conclusion
and rewrite sequences with respect to \m{\red_\reg} and \m{\red_\streg} on the infinite term in the conclusion,
there is also the following correspondence between access paths in a derivation in \Regletrec and \stRegletrec,
and rewrite sequences with respect to \m{\red_\reg} and \m{\red_\streg} on the \lambdaletrec-term in the conclusion.

\begin{proposition}\label{prop:Regletrec:stRegletrec}
	Let \Deriv be a derivation in \Regletrec or \ugRegletrec (in \stRegletrec or in \ugstRegletrec) with conclusion \prefixed{\vec{x}}{\L}.
	\par Then every access path in \Deriv to an occurrence \o of a term \prefixed{\vec{y}}{P} corresponds to a \m{\red_\reg}-rewrite-sequence \prefixed{\vec{x}}{\L} \m{\mred_\reg} \prefixed{\vec{y}}{\letin{B}{\tilde{P}}} (to a \m{\red_\streg}-rewrite-sequence \prefixed{\vec{x}}{\L} \m{\mred_\streg} \prefixed{\vec{y}}{\letin{B}{\tilde{P}}}), where \B arises as the union of all outermost binding groups in conclusions of instances of \FIXletrec below \o, and \prefixed{\vec{y}}{P} = \prefixed{\vec{y}}{\subst{\tilde{P}}{\vec{f}}{\vecsub{\aconstname}{\vec{f}}}} where \vec{f} is comprised of the function variables occurring in \B and \vecsub{\aconstname}{\vec{f}} distinct constants for \vec{f} as chosen by \Deriv. More precisely:
	\begin{enumerate}[(a)]
		\item a pass over an instance of \FIXletrec corresponds to an empty or \m{\red_{\unfoldpre{\merg}}}-step, dependent on whether the instance is the bottommost \FIXletrec-instance or not;
		\item a pass over an instance of the rule \ruleref{@} to the left/to the right corresponds to a \m{\red_{@_0}}-step/\m{\red_{@_1}}-step, which, if the application is somewhere above an instance of \FIXletrec, has to be preceded by a \m{\red_{\unfoldpre{@}}}-step;
		\item a pass over an instance of the rule \ruleref{λ} corresponds to a \m{\red_λ}-step which, if the application is above an instance of \FIXletrec, has to be preceded by a \m{\red_{\unfoldpre{λ}}}-step;
		\item a pass over an instance of the rule \ruleref{\del}/\ruleref{\S} corresponds to a \m{\red_\del}/\m{\red_\S}-step, possibly preceded by an application of \m{\red_{\unfoldpre{\reduce}}}.
		\item a step from a marked assumption to a premise of a \FIXletrec-instances, a step as described in \cref{access:path:iii} of the definition of access paths, corresponds to an \m{\red_{\unfoldpre{\rec}}}-step followed by a \m{\red_{\unfoldpre{\reduce}}}-step.
	\end{enumerate}
\end{proposition}

\begin{example}\label{ex:stRegletrec}
	The \lambdaletrec-term \m{\letin{f = \letin{g = f}{g}}{f}} from \cref{ex:non-unfoldable} does not unfold to a λ-term, as can be regognised considering this derivation:
	\begin{prooftree}
		\axiom{(\aconstname_f)^{\amarker}}
		\axiom{(\aconstname_f)^{\amarker}}
		\axiom{(\aconstname_g)^{\bmarker}}
		\insertBetweenHyps{\hspace*{5ex}}
		\infLabel{\FIXletrec, \bmarker}
		\binaryInf{\prefixed{}{\letin{g = f}{g}}}
		\infLabel{\FIXletrecmin, \amarker}
		\binaryInf{\prefixed{}{\letin{f = \letin{g = f}{g}}{f}}}
	\end{prooftree}
	The instance of \FIXletrecmin at the bottom is not an instance of \FIXletrec, since it is not guarded (has an unguarded cyclic access path that reaches and cycles on the left premise of the instance of \FIXletrec.
\end{example}

\begin{lemma}\label{lem:Regletrec:stRegletrec}
	Let \Deriv be a closed derivation in \ugRegletrec (in \ugstRegletrec) with conclusion \prefixed{}{\L}. Then there exists a scope-delimiting (\extscope-delimiting) strategy \m{\astrat_{\Deriv}} for \RegletrecARS (for \stRegletrecARS) with the following properties:
	\begin{enumerate}[(i)]
		\item\label{lem:Regletrec:stRegletrec:i}
			Every access path in \Deriv corresponds to a rewrite sequence with respect to \m{\astrat_{\Deriv}} starting on \prefixed{}{\L} in the sense of \cref{prop:Regletrec:stRegletrec}.
		\item\label{lem:Regletrec:stRegletrec:ii}
			Every rewrite sequence that starts on \prefixed{}{\L} and proceeds according to \m{\astrat_{\Deriv}} corresponds to an access path in \Deriv with correspondences as described in \cref{prop:Regletrec:stRegletrec}.
		\item\label{lem:Regletrec:stRegletrec:iii}
			\m{\ST{\astrat_{\Deriv}}{\L} = \setcompr{\prefixed{\vec{y}}{\letin{B}{\tilde{P}}}}{\multilinebox{\prefixed{\vec{y}}{\letin{B}{\tilde{P}}} arises from an \\ occurrence of \prefixed{\vec{y}}{P} on an \\ access path of \Deriv as described \\ in \cref{prop:Regletrec:stRegletrec}}}}. \\As a consequence of that \Deriv is finite, \L is \m{\astrat_{\Deriv}}-regular.
		\item\label{lem:Regletrec:stRegletrec:iv}
			\L is \m{\astrat_\Deriv}-productive ⇔ \Deriv is guarded (i.e.\ \Deriv derivation in \Regletrec (\stRegletrec)).
	\end{enumerate}
\end{lemma}

\begin{proof}
	Given a closed derivation \Deriv with conclusion \prefixed{}{\L} (for example) in \ugstRegletrec, a \extscope-delimiting strategy \m{\astrat_{\Deriv}} for \stRegletrecARS such that \cref{lem:Regletrec:stRegletrec:i}--\cref{lem:Regletrec:stRegletrec:iv} hold can be extracted from \Deriv similar as in the proof of a \extscope-delimiting strategy \m{\astrat_{\Deriv}} in \stRegARS was extracted from a closed derivation in \stReg. That the extracted strategy \m{\astrat_{\Deriv}} is productive (not productive) for \L if \Deriv is guarded (not guarded) can be seen by the fact that \m{\astrat_{\Deriv}}-rewrite-sequences correspond to access paths of \Deriv in the sense as stated by \cref{prop:Regletrec:stRegletrec}.
\end{proof}

Now we will prove that derivability in \ugRegletrec/\ugstRegletrec is guaranteed for all \lambdaletrec-terms, and that derivability in \Regletrec/\stRegletrec is a property of a \lambdaletrec-term that is decidable by an easy parsing process.

\begin{proposition}\label{prop:derivability:Regletrec:stRegletrec}
	The following statements hold:
	\begin{enumerate}[(i)]
		\item\label{prop:derivability:Regletrec:stRegletrec:i}
			For every	\lambdaletrec-term \L, \prefixed{}{\L} is derivable both in \ugRegletrec and in \ugstRegletrec.
		\item\label{prop:derivability:Regletrec:stRegletrec:ii}
			For every \lambdaletrec-term \L, derivability of \prefixed{}{\L} in \stRegletrec
		is decidable in at most quadratic time in the size of \L.
	\end{enumerate}
\end{proposition}

\begin{proof}
	For \cref{prop:derivability:Regletrec:stRegletrec:i} note that for every \lambdaletrec-term \L, a closed derivation \m{\Deriv_{\L}} with conclusion \prefixed{}{\L} in \ugstRegletrec can be produced by a bottom-up construction following the term structure of \L. Hereby use of the rules \ruleref{\S} can be restricted to instances immediately below marked assumptions such that, viewed from a (non-cyclic) path \apath from the conclusion upwards to a marked assumption, these \ruleref{\S}-instances are only introduced to shorten the frozen abstraction prefixes by all λ-abstractions that have become frozen on \apath (in order to conform to the side-condition on \FIXletrecmin-instances to have the same frozen abstraction prefix lengths in the discharged marked assumptions as in the conclusion and in the premises).
	\par Now for \cref{prop:derivability:Regletrec:stRegletrec:ii} in order to decide derivability of \prefixed{}{\L} in \stRegletrec, it suffices to decide whether the derivation \m{\Deriv_{\L}} in \ugstRegletrec obtained as described above, or its bottommost instance of \FIXletrecmin if there is any, is guarded on all of its access path cycles. (Note that in the construction of \m{\Deriv_{\L}} only the freedom in placing instances of \ruleref{\S} has been used in a certain, namely lazy, way. The specific placement of instances of these rules does not interfere with the existence or non-existence of guards, that is instances of λ or @ on cycles of access paths.) For this it remains to check whether every cycle on an access path in \m{\Deriv_{\L}} has a guard. This can be done by exploring the proof tree of \m{\Deriv_{\L}} according to all possible access paths (until for the first time a cycle is concluded) and checking for the existence of guards on cycles.
\end{proof}

We now can prove soundness and completeness of the proof system \stRegletrec for the property of \lambdaletrec-terms to unfold to λ-terms.

\begin{theorem}\label{thm:Regletrec:stRegletrec}
	\stRegletrec is sound and complete for the property of \lambdaletrec-terms to unfold to a λ-term. That is, for every term \m{\L ∈ \Ter{\lambdaletreccal}} the following statements are equivalent:
	\begin{enumerate}[(i)]
		\item\label{thm:Regletrec:stRegletrec:i}
			\L expresses a λ-term.
		\item\label{thm:Regletrec:stRegletrec:iii}
			\derivablein{\stRegletrec}{\prefixed{}{\L}}.
	\end{enumerate}
\end{theorem}

\begin{proof}
	For the proof of both directions of the equivalence, let \m{\L ∈ \Ter{\lambdaletreccal}}.
	\par For showing the implication \cref{thm:Regletrec:stRegletrec:i} ⇒ \cref{thm:Regletrec:stRegletrec:iii}, we argue indirectly, and therefore assume that \prefixed{}{\L} is not derivable in \stRegletrec. Then, while \prefixed{}{\L} is not derivable in \stRegletrec, there is, by \cref{prop:derivability:Regletrec:stRegletrec}~\cref{prop:derivability:Regletrec:stRegletrec:i}, a derivation \Deriv in \ugstRegletrec that is not guarded. It follows by \cref{lem:Regletrec:stRegletrec}, and in particular due to \cref{lem:Regletrec:stRegletrec}~\cref{lem:Regletrec:stRegletrec:iv}, that there is a \extscope-delimiting strategy \m{\astrat_{\Deriv}} for \stRegARS such that \L is not \m{\astrat_{\Deriv}}-productive. Then it follows by \cref{lem:unfolding:versus:strats}, using \cref{lem:unfolding:versus:strats:i} ⇒ \cref{lem:unfolding:versus:strats:iv} there, that \L does not unfold to a λ-term.
	\par For showing the implication \cref{thm:Regletrec:stRegletrec:iii} ⇒ \cref{thm:Regletrec:stRegletrec:i}, let \Deriv be a closed derivation in \stRegletrec with conclusion \prefixed{}{\L}. It follows by \cref{lem:Regletrec:stRegletrec} that there is a \extscope-delimiting strategy \astrat for \RegARS such that \L is \astrat-productive. Then \cref{lem:unfolding:versus:strats} implies that \L unfolds to a λ-term.
\end{proof}

\begin{para}[soundness and completeness of \Regletrec]
	Also the proof system \Regletrec can be shown to be sound and complete for the property of \lambdaletrec-terms to unfold to λ-terms. To establish this in analogy with the route of proof we pursued here, a CRS \ParseCRS similar to \stParseCRS (see \cref{def:stParseCRS}) could be defined by replacing the rule \rulebp{\sparse}{\S} by a rule \rulebp{\sparse}{\del} that can compress more abstraction prefixes, similar as the rule \rulebp{\reg}{\del} of \RegCRS can compress more abstraction prefixes than the rule \rulebp{\streg}{\S} of \stRegCRS. Then furthermore also a lemma analogous to \cref{lem:unfolding:versus:strats} can be formulated, proved, and used in a similar way.
\end{para}

We now arrive at a theorem that states one direction of our main characterisation result (\cref{thm:ll-expressible:streg} in \cref{sec:express}) that will link \lambdaletrec-expressibility to strong regularity of λ-terms.

\begin{theorem}\label{thm:ll-expressible:2:streg}
	Every \lambdaletrec-expressible λ-term is strongly regular.
\end{theorem}

\begin{proof}
	Let \M be a λ-term that is expressible by a \lambdaletrec-term \L, that is, \m{\L \m{\omegared_\unfold} \M} holds. Then by \cref{thm:Regletrec:stRegletrec} there exists a closed derivation \Deriv in \stRegletrec with conclusion \prefixed{}{\L}. Now \cref{lem:Regletrec:stRegletrec} guarantees a \extscope-delimiting strategy \m{\astrat_{\Deriv}} for \stRegletrecARS such that \L is \m{\astrat_{\Deriv}}-regular. Then \cref{lem:proj:strat:letrec:lambda} gives a \extscope-delimiting strategy \m{\Check{\astrat}_{\Deriv}} for \stRegletrecARS such that \m{\M = \unfsem{\L}} is \m{\Check{\astrat}_{\Deriv}}-regular. It follows that \M is strongly regular.
\end{proof}

\section{Binding--Capturing Chains}\label{sec:chains}

\begin{para}[overview]
	In this section we develop a characterisation for strong regularity of λ-terms by means of the `binding--capturing chains' occurring in a term. This concept is related to the notions of scope and \extscope as explained informally in \cref{para:regularity}. Binding--capturing chains occur whereever scopes overlap; and they are fully contained within \extscopes. First we give definitions for the concepts involved: binding, capturing, and binding--capturing chains. Then we show that strong regularity of regular λ-terms is equivalent to the absence of infinite binding--capturing chains.
\end{para}

\begin{para}[binding and capturing]
	We will define binding and capturing as relations on the positions of a λ-term. Binding relates an abstraction with the occurrences of the variable it binds. If \p is the position of an abstraction (\abs{x}{\dots}) that abstracts over \x and \q is the position of an occurrence of \x that is bound by the abstraction, then we will write \m{\p \binds \q} and say that \p `binds' \q. Capturing relates an abstraction with variables that occur freely underneath the abstraction. If \p is the position of an abstraction, and \m{\q > \p} is the position of a variable that is free in the entire subterm at position \p, then we will write \m{\p \captures \q} and say that \p `captures' \q. See \cref{fig:entangled_bc} for an illustration of these concepts.
\end{para}

\begin{figure}
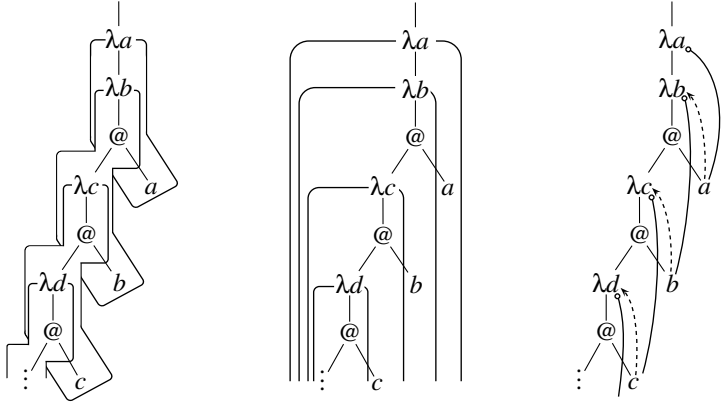

	\begin{hspread}
		\fig{entangled-scopes} & \fig{entangled-extscopes} & \fig{entangled-chains}
	\end{hspread}
	\caption{The term graph of the term in \cref{ex:entangled} with its overlapping scopes (left), its nested \extscopes (middle), and with its binding (\binds) and capturing (\capturedby) links (right).}
	\label{fig:entangled_bc}
\end{figure}

\begin{para}[positions in iCRS terms]
	When we speak of positions in λ-terms (and thus iCRS terms) we act on the assumption that positions on iCRS terms are an established concept as for example in \cite{kete:simo:2011}. Note, however, that we deviate slightly from the scheme there in addressing the arguments of an \sappCRS by \m{0} and \m{1} instead of \m{1} and \m{2}.
\end{para}

\begin{para}[binding--capturing chains in the literature]
	Binding--capturing chains have been used in \cite{endr:grab:klop:oost:2011} to study α-avoiding rewrite sequences in a rewrite calculus for μ-unfolding. They originate from the notion of `gripping' due to \cite{mell:96}, and from techniques developed in \cite{oost:97} concerning the notion of `holding' of redexes (which is shown there as being `parting' for CRSs, that is, never relating two residuals of the same redex).
\end{para}

We now define `binding' and `capturing' formally as binary relations on the set of positions of λ-terms.

\begin{definition}[binding, capturing]\label{def:bind:iscapturedby}
	For every \m{\M ∈ \iTer{\lambdacal}} we define the binary relations \binds and \capturedby on the set \m{\positions{\M} \subseteq ℕ^*} of positions of \M:
	\begin{enumerate}[(i)]
		\item\label{def:bind:iscapturedby:bind} We say that a binder at position \p\ \emph{binds} a variable occurrence at position \q, symbolically \m{\p \binds \q}, if \p is a binder position, and \q a variable position in \M, and the binder at position \p binds the variable occurrence at position \q.
		\item\label{def:bind:iscapturedby:iscapturedby} We say that a variable occurrence at position \q\ \emph{is captured by} a binder at position \p, symbolically \m{\q \capturedby \p} (and that a binder at position \p\ \emph{captures} a variable occurrence at position \q, symbolically \m{\p \captures \q}), if \q is a variable position and \m{\p < \q} a binder position in \M, and there is no binder position \m{\q_0} in \M with \m{\p ≤ \q_0} and \m{\q_0 \binds \q}.
	\end{enumerate}
\end{definition}

\begin{definition}[binding--capturing chain]\label{def:bind:capt:chain}
	Let \M be a λ-term. A finite or infinite sequence \enumsequence{\p_0,\p_1,\p_2,\dots} in \positions{\M} is called a \emph{binding--capturing chain in \M} if \m{\p_0} is the position of an abstraction in \M, and the positions in the sequence are alternatingly linked via binding and capturing, starting with a binding: \m{\p_0 \binds \p_1 \capturedby \p_2 \binds \dots}.
\end{definition}

\begin{para}
	Binding--capturing chains are closely related to the notion of scope and \extscope. In order to establish this, we first give precise definitions of the notions of scope and \extscope in terms of an `in-scope' rewrite relation on the positions of a λ-term: While the scope of a binder position \p is the set of positions between \p and variable positions bound at \p (the positions directly reachable by a single `in-scope' step), the \extscope of \p is the set of positions reachable by a finite number of successive `in-scope' steps.
\end{para}

\begin{notation}[binder positions]\label{not:binder_positions}
	In the following paragraphs for a given position \p in some λ-term, we write \bp{\p} for the proposition `\p is a binder position'.
\end{notation}

\begin{definition}[scope and \extscope]\label{def:scope:extscope}
	Let \M be a λ-term. On the set \positions{\M} of \M, the \emph{in-scope} relation \m{\red_\sscope} (for \M) is defined by:
	\[
		\p \m{\red_\sscope} \q
		~\Longleftrightarrow~
			\bp{\p} ~∧~
		∃ \p' ∈ \positions{\M} ~~ \p \bindseq \p' \:∧\: \p ≤ q ≤ p'
	\]
	where \bindseq denotes the reflexive closure of the binding relation \binds.
	
	For every binder position \m{\p ∈ \positions{\M}}, the \emph{scope of \p in \M} and the \emph{\extscope of \p in \M} are defined as the following sets of positions in \M:
	\begin{align*}
		\scopeof{\M}{\p}
		& :=
		\setcompr{\q ∈ \positions{\M}}{\p \m{\red_\sscope} \q}
		\\
		\extscopeof{\M}{\p}
		& :=
		\setcompr{\q ∈ \positions{\M}}{\p \m{\morestepred_\sscope} \q}
	\end{align*}
	(Note that the scopes and \extscopes of non-binder positions are empty sets.)
\end{definition}

The following proposition establishes that binding--capturing chains starting at a binder position \p span the positions of the \extscope of \p.

\begin{proposition}[binding, capturing, and scope/\extscope]\label{prop:scope:extscope}
	Let \M be a λ-term and \m{\p ∈ \positions{\M}} be a binder position. Then for all positions \m{\q ∈ \positions{\M}} the following statements hold:
	\begin{enumerate}[(i)]
		\item \m{\p \m{\red_\sscope} \q ~∧~ \bp{\q} ~\Longleftrightarrow~ \p = \q ~∨~ \p \binds ⋅ \capturedby \q}
		\item \m{\p \m{\morestepred_\sscope} \q ~\Longleftrightarrow~ ∃\p' ∈ \positions{\M}~~ \p~(\binds ⋅ \capturedby)^*\,⋅\,\bindseq \p' \:∧\: \p ≤ \q ≤ \p'}
		\item \m{\scopeof{\M}{\p} = \setcompr{\q ∈ \positions{\M}}{∃\p' ∈ \positions{\M}~\begin{array}{l}\p \bindseq \p' \\∧\: \p ≤ q ≤ p'\end{array}}}
		\item \m{\extscopeof{\M}{\p} = \setcompr{\q ∈ \positions{\M}}{∃ \p' ∈ \positions{\M}~\begin{array}{l}\p~(\binds ⋅ \capturedby)^*\,⋅\,\bindseq \p' \\∧\: \p ≤ q ≤ p'\end{array}}}
	\end{enumerate}
\end{proposition}

Conversely, every position that is covered by a binding--capturing chain starting at a binding position \p is in the \extscope of \p:

\begin{proposition}\label{prop:bindcaptchains}
	Let \enumsequence{\p_0,\p_1,\p_2,\dots} be a binding--capturing chain in a λ-term \M. Then it holds that \m{\p_0 < \p_2 < \dots}, and \m{\p_0 < \p_1}, \m{\p_2 < \p_3}, \dots. Furthermore, for all \q with \m{\p_0 ≤ \q ≤ \p_n} for some \m{n ∈ ℕ} with \m{\p_n} being a position on the chain it holds that \m{\q ∈ \extscopeof{\M}{\p_0}}.
\end{proposition}

\begin{para}[position-annotated variants \RegposCRS and \stRegposCRS]
	In order to study the relationship between rewrite sequences in \stRegCRS and binding--capturing chains we first introduce a position-annotated variant of \stRegCRS.
The idea is, that when a prefixed term \prefixed{y_1 \dots y_n}{\N} is obtained as a generated subterm of a λ-term \M by a \m{\red_\reg} or \m{\red_\streg} rewrite sequence \arewseq on \prefixed{}{\M}, then in the position-annotated rewrite system a prefixed term \m{\posprefixed{y_1,\dots,y_n}{\p_1 \dots \p_n}{\q}{\N}} is obtained by an annotated version \Hat{\arewseq} of the rewrite sequence \arewseq such that: the positions \m{\p_1,\dots,\p_n} are the positions in (the original λ-term) \M from which the bindings \abs{y_1}{}, \dots, \abs{y_n}{} in the abstraction prefix stem from; and \q is the position in \M of the body \N of the subterm generated by \arewseq.
\end{para}

\begin{para}[position-annotated decomposition in informal notation]
	\newcommand\rulename[1]{&\rulepos{#1}:~}
	\newcommand\rewritestoBreak{\\[-.5ex]&~\hspace{1cm}\red}
	\newcommand\rewritesto{\red}
	\newcommand\nextRule{\\}
	\newcommand\annotation[1]{\\[-.5ex]&~\hspace{2cm}{#1}}
	On \iTer{\lambdaprefixcal} we consider the following rewrite rules:
	\begin{align*}
		\rulename{@_i}
			\posprefixed{x_1 \dots x_n}{\p_1,\dots,\p_n}{\q}{\app{\M_0}{\M_1}}
			\rewritesto
			\posprefixed{x_1 \dots x_n}{\p_1,\dots,\p_n}{\q i}{\M_i}
			~~\sideCondition{\m{i ∈ \set{0,1}}}
			\nextRule
		\rulename{λ}
			\posprefixed{x_1 \dots x_n}{\p_1,\dots,\p_n}{\q}{\abs{y}{\M_0}}
			\rewritesto
			\posprefixed{x_1 \dots x_ny}{\p_1,\dots,\p_n,\q}{\q 00}{\M_0}
			\nextRule
		\rulename{\S}
			\posprefixed{x_1 \dots x_{n+1}}{\p_1,\dots,\p_{n+1}}{\q}{\M_0}
			\rewritesto
			\posprefixed{x_1 \dots x_n}{\p_1,\dots,\p_n}{\q}{\M_0}
			\annotation{\sideCondition{if the binding \abs{x_{n+1}}{} is vacuous}}
			\nextRule
		\rulename{\del}
			\posprefixed{x_1 \dots x_{n+1}}{\p_1,\dots,\p_{n+1}}{\q}{\M_0}
			\rewritestoBreak
			\posprefixed{x_1 \dots x_{i-1}x_{i+1}\dots x_{n+1}}{\p_1,\dots \p_{i-1},\p_{i+1}\dots \p_{n+1}}{\q}{\M_0}
			\annotation{\sideCondition{if the binding \abs{x_i}{} is vacuous}}
	\end{align*}
	Note that in the rule \rulepos{λ} the position changes from \q to \m{\q 00}, because of the underlying CRS notation for terms in \lambdaprefixcal: when a term \absCRS{y}{\M_0} starts at position \q, then the CRS-binding is at position \m{\q 0}, and the body \m{\M_0} starts at position \m{\q 00}.
\end{para}

\begin{definition}[position-annotated decomposition rewrite systems \RegposzeroCRS, \RegposCRS, \stRegposCRS]\label{def:RegposCRS:stRegposCRS}
	\newcommand\sigPosPrefixed{\m{\sig^{\wPrefix}_\spos}}
	\newcommand\sfabsposCRS[2]{\m{{\mathsf{pre}}^{#2}_{\enumsequence{#1}}}}
	\newcommand\fabsposCRS[4]{{\sfabsposCRS{#2}{#3}}\left(\CRSabs{#1}{#4}\right)}
	The CRS signature for \lambdaprefixposcal, the λ-calculus with position-annotated abstraction prefixes is given by \m{\sigPosPrefixed = \sigCRS ∪ \setcompr{\sfabsposCRS{\p_1,\dots,\p_n}{\q}}{ \p_1,\dots,\p_n,\q ∈ \set{0,1}^*}} where all of the function symbols \m{\sfabsposCRS{\p_1,\dots,\p_n}{\q}} are unary. We consider the following CRS rule schemes over \sigPosPrefixed:
	\newcommand\rulename[1]{&\rulepos{#1}:~}
	\newcommand\rewritestoBreak{\\[-.5ex]&~\hspace{1cm}\red}
	\newcommand\rewritesto{\red}
	\newcommand\nextRule{\\}
	\newcommand\annotation[1]{\\[-.5ex]&~\hspace{2cm}{#1}}
	\begin{align*}
		\rulename{@_i}
			\fabsposCRS{x_1 \dots x_n}{\p_1,\dots,\p_n}{\q}{\appCRS{\M_0}{\M_1}}
			\rewritestoBreak
			\fabsposCRS{x_1 \dots x_n}{\p_1,\dots,\p_n}{\q i}{\M_i}
			~~\sideCondition{\m{i ∈ \set{0,1}}}
			\nextRule
		\rulename{λ}
			\fabsposCRS{x_1 \dots x_n}{\p_1,\dots,\p_n}{\q}{\absCRS{y}{\M_0}}
			\rewritestoBreak
			\fabsposCRS{x_1 \dots x_ny}{\p_1,\dots,\p_n,\q}{\q 00}{\M_0}
			\nextRule
		\rulename{\S}
			\fabsposCRS{x_1 \dots x_{n+1}}{\p_1,\dots,\p_{n+1}}{\q}{\M_0}
			\rewritesto
			\fabsposCRS{x_1 \dots x_n}{\p_1,\dots,\p_n}{\q}{\M_0}
			\nextRule
		\rulename{\del}
			\fabsposCRS{x_1 \dots x_{n+1}}{\p_1,\dots,\p_{n+1}}{\q}{\M_0}
			\rewritestoBreak
			\fabsposCRS{x_1 \dots x_{i-1}x_{i+1}\dots x_{n+1}}{\p_1,\dots \p_{i-1},\p_{i+1}\dots \p_{n+1}}{\q}{\M_0}
	\end{align*}
	By \RegposzeroCRS we denote the CRS with the rules \rulepos{@_i} and \rulepos{λ}. By \RegposCRS/\stRegposCRS we denote the CRS consisting of the rules \rulepos{@_i}, \rulepos{λ}, and \rulepos{\del}/\rulepos{\S}.
	\par By \RegposzeroARS, \RegposARS and \stRegposARS we denote the ARSs induced by the iCRSs derived from \RegposzeroCRS, \RegposCRS, \stRegCRS, restricted to position-annotated terms in \iTer{\lambdaprefixcal}.
\end{definition}

\begin{proposition}[position-annotated rewrite sequences]\label{prop:position:lifting:projecting}
	It holds:
	\begin{enumerate}[(i)]
		\item\label{prop:position:lifting:projecting:i}
			Every rewrite sequence
			\begin{equation}\label{eq1:prop:position:lifting:projecting}
				\arewseq :~ \prefixed{\vec{x}_0}{\M_0} \red \prefixed{\vec{x}_1}{\M_1} \red \dots \red \prefixed{\vec{x}_n}{\M_n}
			\end{equation}
			in \RegzeroARS, \RegARS, or \stRegARS can be transformed (lifted) step by step, for a given \m{\q_0 ∈ ℕ^*} and \m{\vec{\p}_0 ∈ \vecpositions} with \m{\length{\vec{\p}_0} = \length{\vec{x}_0}}, by adding these and appropriate further position annotations \m{\q_1, \dots, \q_n ∈ ℕ^*} and \m{\vec{\p}_1, \dots, \vec{\p}_n ∈ \vec{ℕ^*}}, to a rewrite sequence:
			\begin{equation}\label{eq2:prop:position:lifting:projecting}
				\Hat{\arewseq} :~ \posprefixed{\vec{x}_0}{\vec{\p}_0}{\q_0}{\M_0} \red \posprefixed{\vec{x}_1}{\vec{\p}_1}{\q_1}{\M_1} \red \dots \red \posprefixed{\vec{x}_n}{\vec{\p}_n}{\q_n}{\M_n}
			\end{equation}
			in \RegposzeroARS, \RegposARS, or \stRegposARS, accordingly, such that the result of dropping the position annotations in the prefix of \Hat{\arewseq} is again \arewseq.
		\item\label{prop:position:lifting:projecting:ii}
			Conversely, every rewrite sequence \brewseq in \RegposzeroARS, \RegposARS, or \stRegposARS of the form \cref{eq2:prop:position:lifting:projecting} can be transformed step by step, by dropping the position annotations in the prefix, to a rewrite sequence \Check{\brewseq} of the form \cref{eq2:prop:position:lifting:projecting} in \RegzeroARS, \RegARS, or \stRegARS, respectively.
	\end{enumerate}
	The transformations in \cref{prop:position:lifting:projecting:i} and \cref{prop:position:lifting:projecting:ii} preserve eagerness/laziness of rewrite sequences.
\end{proposition}


The proposition below characterises the binding relation \binds and the capturing relation \capturedby on the positions of an
infinite term \M with the help of rewrite sequences with respect to \m{\red_\regzero}
on \posprefixed{}{}{\rootpos}{\M} in \RegposzeroCRS
down to `variable occurrences' \posprefixed{\vec{x}}{\vec{\p}}{\q}{x_i} in \M.

\begin{proposition}[binding, capturing, and \RegposzeroCRS rewrite sequences]\label{prop:bind:iscapturedby}
	For all \m{\M ∈ \iTer{\lambdacal}} and positions \m{\p,\q ∈ \positions{\M}} it holds:
	\begin{align*}
		\p \binds \q
		&\Leftrightarrow
		\multilinebox{there is a rewrite sequence \m{\posprefixed{}{\tuple{}}{\rootpos}{\M} \m{\mred_\regzero} \posprefixed{x_1 \dots x_n}{\p_1,\dots,\p_n}{\q}{x_i}} \\ with \m{x_1 \dots x_n} distinct, \m{i ∈ \set{1,\dots,n}}, such that \m{\p = \p_i}}
		\\
		\q \capturedby \p
		&\Leftrightarrow
		\multilinebox{there is a rewrite sequence \m{\posprefixed{}{\tuple{}}{\rootpos}{\M} \m{\mred_\regzero} \posprefixed{x_1 \dots x_n}{\p_1,\dots,\p_n}{\q}{x_i}} \\ with \m{x_1 \dots x_n} distinct, \m{i ∈ \set{1,\dots,n}}, such that \m{\p ∈ \set{\p_{i+1},\dots,\p_n}}}
	\end{align*}
\end{proposition}

The following lemmas describe the close relationship between, on the one hand,
binding--capturing chains in a λ-term \M, and on the other hand,
\m{\red_\streg}-rewrite-sequences on \posprefixed{}{}{\rootpos}{\M}
in \stRegposCRS that are guided by the eager \extscope-delimiting strategy.

\begin{lemma}[binding--capturing chains]\label{lem:bind:capt:chains:stRegpos}
	For all \m{\M ∈ \iTer{\lambdacal}} the following statements hold:
	\begin{enumerate}[(i)]
		\item\label{lem:bind:capt:chains:stRegpos:i} If there is a rewriting sequence of the form
			\begin{align*}
				\posprefixed{}{}{\rootpos}{\M} & \mred_{\eagStratPlus} \posprefixed{x_0 \dots x_{n_1}}{\p_0,\dots,\p_{n_1}}{\q}{\N}
				\\
				& \mred_{\eagStratPlus} \posprefixed{x_0 \dots x_{n_1}\dots x_{n_2}}{\p_0,\dots,\p_{n_1},\dots,\p_{n_2}}{\q'}{O}
			\end{align*}
			then there exist \m{\q_{n_1 +1},\dots,\q_{n_2} ∈ \positions{\M}} such that \m{\p_{n_1} \binds \q_{n_1 +1} \capturedby \p_{n_1 +1} \binds \dots \binds \q_{n_2} \capturedby \p_{n_2}}.
		\item\label{lem:bind:capt:chains:stRegpos:ii} If \m{\p_0 \binds \q_1 \capturedby \p_1 \binds \dots \binds \q_n \capturedby \p_n} is a binding--capturing chain in \M, then there exist \m{r_0,\dots,r_m,s ∈ \positions{\M}} with \m{m ≥ n} such that \m{\posprefixed{}{}{\rootpos}{\M} \mred_{\eagStratPlus} \posprefixed{x_0 \dots x_m}{r_0,\dots,r_m}{s}{\N}} and furthermore \m{\p_0,\dots,\p_n ∈ \set{r_0,\dots,r_m}} such that \m{\p_0 < \p_1 < \dots < \p_n = r_m}.
	\end{enumerate}
\end{lemma}

%
%

\begin{lemma}[length of binding--capturing chains]\label{lem:fin:bind:capt:chains}
	Let \M be a λ-term such that \m{\prefixed{}{\M} \mred_{\eagStratPlus} \prefixed{x_0 \dots x_n}{\N}}. Then \M contains a binding--capturing chain of length \n.
\end{lemma}

\begin{proof}
	By \cref{prop:position:lifting:projecting}~\cref{prop:position:lifting:projecting:i}, the assumed rewrite sequence \m{\prefixed{}{\M} \mred_{\eagStratPlus} \prefixed{x_0 \dots x_n}{\N}} in \stRegARS can be lifted to a rewrite sequence \m{\posprefixed{}{}{\rootpos}{\M} \red_{\eagStratPlus} \posprefixed{x_0 \dots x_n}{\p_0,\dots,\p_n}{\q}{\N}} in \stRegposARS. Then by \cref{lem:bind:capt:chains:stRegpos}~\cref{lem:bind:capt:chains:stRegpos:i}, there exists a binding--capturing chain of length \n.
\end{proof}

\begin{para}[binding--capturing chains and scope/\extscope]
	The notion of scope and \extscope helps to understand the relationship between binding--capturing chains and rewrite sequences in \stRegCRS. A binding--capturing chain corresponds to the overlap of scopes, or in other words the nesting of \extscopes. An infinite binding--capturing chain thus corresponds to a infinitely deep nesting of \extscopes and therefore to an unrestricted growth of the prefix in certain rewriting sequences in \stRegCRS.
\end{para}

\begin{lemma}[infinite binding--capturing chains]\label{lem:inf:bind:capt:chains}
	\newfunction\pl{p}
	Let \M be a λ-term, and let \arewseq be an infinite rewrite sequence in \RegARS w.r.t.\ the eager \extscope-delimiting strategy \eagStratPlus:
	\begin{equation*}
		\arewseq ~:~
		\prefixed{}{\M} = \prefixed{\vec{x}_0}{\M_0}
		\red_{\eagStratPlus}
		\prefixed{\vec{x}_1}{\M_1}
		\red_{\eagStratPlus}
		\dots
		\red_{\eagStratPlus}
		\prefixed{\vec{x}_i}{\M_i}
		\red_{\eagStratPlus}
		\dots
	\end{equation*}
	Furthermore suppose that for \m{\pl{} : ℕ→ℕ}, \m{i↦\pl{i}:= \length{\vec{x}_i}}, the prefix length function associated with \arewseq, there exists a lower bound \m{\lb{} :ℕ→ℕ} such
	that \lb{} is non-decreasing, and 
	\m{\lim_{n→\infty} \lb{n} = \infty}.
	Then there exists an infinite binding--capturing chain in \M.

	\begin{proof}
		Let \M, \arewseq, \pl{}, \lb{} as above. By \cref{prop:position:lifting:projecting}~\cref{prop:position:lifting:projecting:i}, \arewseq can be lifted to position annotated counterpart
		\begin{equation*}
			\Hat{\arewseq} :\;
			\posprefixed{}{}{\rootpos}{\M} = \posprefixed{\vec{x}_0}{}{\rootpos}{\M_0}
			\red_{\eagStratPlus}
			\posprefixed{\vec{x}_1}{\vec{\p}_1}{\q_1}{\M_1}
			\red_{\eagStratPlus}
			\dots
			\red_{\eagStratPlus}
			\posprefixed{\vec{x}_i}{\vec{\p}_i}{\q_i}{\M_i}
			\red_{\eagStratPlus}
			\dots
		\end{equation*}
		where, for all \m{i ∈ ℕ}, \m{\q_i} are positions and \m{\vec{\p}_i = \tuple{\p_1,\dots,\p_{m_i}}} vectors of positions, with \m{m_i ∈ ℕ}.
		\newfunction\st{st}
		Next we define the function:
		\[\st{} : ℕ→ℕ,\quad l ↦ \st{l} := \min{\setcompr{i}{\lb{i} ≥ l}}\]
		which is well-defined, since \m{\lim_{n→\infty}{\lb{n} = \infty}}. It describes a prefix stabilisation property: for every \m{l ∈ ℕ}, it gives the first index \m{i = \st{l}} with the property that the prefix of \prefixed{\vec{x}_i}{\M_i} contains more than \m{l} abstractions, and (since \lb{} is non-decreasing, and a lower bound for \pl{}) that from \i onward the \m{l}-th abstraction never disappears again, for \m{j ≥ i}, in terms \prefixed{\vec{x}_j}{\M_j} that follow in \arewseq as well as in \Hat{\arewseq}. Furthermore, \st{} is non-decreasing, as an easy consequence of its definition, and unbounded: if \st{} were bounded by \m{M ∈ ℕ}, then \m{∀l∈ℕ ~ ∃i∈ℕ ~ i≤M ∧ \lb{i} ≥ l} would follow, which cannot be the case since \set{\lb{0},\dots,\lb{M}} is a finite set. By non-decreasingness and unboundedness it also follows that \m{\lim_{n→\infty}{\st{n} = \infty}}.
		\par So when the rewrite sequence \Hat{\arewseq} is split into segments indicated in
		\begin{align*}
			\posprefixed{}{}{\rootpos}{\M} \mred_{\eagStratPlus} \dots
			& \mred_{\eagStratPlus} \posprefixed{\vec{x}_{\st{i}}}{\vec{\p}_{\st{i}}}{\q_{\st{i}}}{\M_{\st{i}}} \\
			& \mred_{\eagStratPlus} \posprefixed{\vec{x}_{\st{i+1}}}{\vec{\p}_{\st{i+1}}}{\q_{\st{i+1}}}{\M_{\st{i+1}}} \mred_{\eagStratPlus} \dots
		\end{align*}
		then it follows that all terms of the sequence after \m{\posprefixed{\vec{x}_{\st{i}}}{\vec{\p}_{\st{i}}}{\q_{\st{i}}}{\M_{\st{i}}}} have an abstraction prefix of length greater or equal to \i, for all \m{i ∈ ℕ}.
		\par Now note that in a step
		\[\posprefixed{\vec{x}}{\tuple{\p_1,\dots,\p_n}}{\q}{\N} \red \posprefixed{\vec{x}'}{\tuple{\p'_1,\dots,\p'_{n'}}}{\q'}{\N'}\]
		in \RegposARS that does not shorten the abstraction prefix it holds that \m{n ≤ n'} and
		\m{\p'_1 = \p_1}, \dots, \m{\p'_n = \p_n},
		that is, positions in the vector in the subscript of the abstraction prefix are preserved.
		As a consequence it follows for the rewrite sequence \Hat{\arewseq} that for all \m{i ∈ ℕ} and \m{j > i} the position vector \m{\vec{\p}_{\st{j}}} in the term \m{\posprefixed{\vec{x}_{\st{j}}}{\vec{\p}_{\st{j}}}{\q_{\st{j}}}{\M_{\st{j}}}} is of the form
		\[\vec{\p}_{\st{j}} = \tuple{\p_{1,\st{i}},\dots,\p_{i,\st{i}},\p_{j,\st{j}},\dots,\p_{j,m_j}}.\]
		This implies furthermore that for all \m{i ∈ ℕ}:
		\[\vec{\p}_{\st{i}} = \tuple{\p_{1,\st{1}},\p_{2,\st{2}},\dots,\p_{i,\st{i}},\dots,\p_{i,m_i}}.\]
		Then \cref{lem:bind:capt:chains:stRegpos}~\cref{lem:bind:capt:chains:stRegpos:i},
		implies the existence of positions \m{\q_2,\q_3,\dots} such that
		\begin{equation*}
			\p_{1,\st{1}}
			\binds \q_2 \capturedby
			\p_{2,\st{2}}
			\binds \q_3 \capturedby
			\dots
			\capturedby
			\p_{i,\st{i}}
			\binds \q_{i+1} \capturedby
			\p_{i+1,\st{i+1}}
			\binds
			\dots
		\end{equation*}
		and thereby, an infinite binding--capturing chain in \M.
	\end{proof}
\end{lemma}

Now we formulate and prove the main theorem of this section, which applies the concept of binding--capturing chain to pin down, among all λ-terms that are regular, those that are strongly regular.

\begin{theorem}[binding--capturing chains and strong regularity]\label{thm:streg:fin:bind:capt:chains}
	A regular λ-term is strongly regular if and only if it contains only finite binding--capturing chains.
\end{theorem}

And taking into account \cref{prop:def:reg:streg}~\cref{prop:def:reg:streg:i}, we obtain:

\begin{corollary}\label{cor:thm:streg:fin:bind:capt:chains}
	A λ-term is strongly regular if and only if it is regular and contains only finite binding--capturing chains.
\end{corollary}

\begin{proof}[Proof of \cref{thm:streg:fin:bind:capt:chains}.]
	Let \M be a λ-term that is regular.
	\par For showing the implication ``⇒'', we assume that \M is also strongly regular. Then there exists a \extscope-delimiting strategy \astrat such that \ST{\astrat}{\M} is finite. By \cref{prop:eager:strat:in:def:reg:streg}~\cref{prop:eager:strat:in:def:reg:streg:i} it follows that then also \ST{\eagStratPlus}{\M} is finite for the eager \extscope-delimiting strategy \eagStratPlus in \stRegARS. Now let \n be the longest abstraction prefix of a term in \ST{\eagStratPlus}{\M}. Then it follows by \cref{lem:fin:bind:capt:chains} that the length of every binding--capturing chain in \M is bounded by \m{n-1}. Hence \M only contains finite binding--capturing chains.
	\par In the rest of this proof, we establish the implication ``⇐'' in the statement of the theorem. For this we argue indirectly: assuming that \M is not strongly regular, we show the existence of an infinite binding--capturing chain in \M.
	\par So suppose that \M is not strongly regular. Then for all \extscope-delimiting strategies \astrat in \stRegARS it holds that \ST{\astrat}{\M} is infinite. This means that in particular \ST{\eagStratPlus}{\M} is infinite for the eager scope-delimiting strategy \eagStratPlus on \stRegARS. It follows that the number of \m{\mred_{\eagStratPlus}}-reducts, and hence the generated sub-ARS \GeneratedSubARSi{\eagStratPlus}{\prefixed{}{\M}} of \prefixed{}{\M} in \stRegARS is infinite. Since \m{\mred_{\eagStratPlus}} on \stRegARS has branching degree \m{≤ 2} (branching actually only happens at sources of \m{\red_{@_i}}-steps), it follows by Kőnig's Lemma that there exists an infinite rewrite sequence:
	\begin{equation}\label{eq1:prf:thm:streg:fin:bind:capt:chains}
		\arewseq :~
		\prefixed{}{\M} = \prefixed{\vec{x}_0}{\M_0}
		\red_{\eagStratPlus}
		\dots
		\red_{\eagStratPlus}
		\prefixed{\vec{x}_i}{\M_i}
		\red_{\eagStratPlus}
		\dots
	\end{equation}
	in \stRegARS that passes through distinct terms. By \cref{lem:projection:RegCRS:stRegCRS}~\cref{lem:projection:RegCRS:stRegCRS:reg}, this rewrite sequence projects to a rewrite sequence:
	\begin{equation}\label{eq2:prf:thm:streg:fin:bind:capt:chains}
		\Check{\arewseq} :~
		\prefixed{}{\M} = \prefixed{\vec{x}'_0}{\M_0}
		\mred_{\eagStrat}
		\dots
		\mred_{\eagStrat}
		\prefixed{\vec{x}'_i}{\M_i}
		\mred_{\eagStrat}
		\dots
	\end{equation}
	in \RegARS in the sense that:
	\begin{equation}\label{eq3:prf:thm:streg:fin:bind:capt:chains}
		\prefixed{\vec{x}_i}{\M_i} \m{\mred_\del} \prefixed{\vec{x}'_i}{\M_i} ~~~ \sideCondition{for all \m{i ∈ ℕ}}.
	\end{equation}
	Note that, in the terms, the projection merely shortens the length of the abstraction prefix. Since \M is regular, \ST{\eagStrat}{\M} is finite by \cref{prop:eager:strat:in:def:reg:streg}~\cref{prop:eager:strat:in:def:reg:streg:i}, and hence it follows that only finitely many terms occur in \Check{\arewseq}.
	\par Now we will use this contrast with \arewseq, and the fact that the terms of \arewseq project to terms in \Check{\arewseq} via \m{\mred_\del}-prefix compression rewrite sequences, to show that the prefix lengths in terms of \arewseq are unbounded, and stronger still, that these lengths actually tend to infinity. More precisely, we show the following:
	\begin{equation}\label{eq4:prf:thm:streg:fin:bind:capt:chains}
		∀ l ∈ ℕ ~ ∃ i_0 ∈ ℕ ~ ∀ i ≥ i_0~~ \length{\vec{x}_i} ≥ l
	\end{equation}
	Suppose that this statement does not hold. Then there exists \m{l_0 ∈ ℕ} such that \m{\length{\vec{x}_i} < l_0} for infinitely many \m{i ∈ ℕ}. This means that there is an increasing sequence \m{i_0 < i_1 < i_2 < i_3 < \dots} in ℕ such that:
	\begin{gather}
		S := \setcompr{\prefixed{\vec{x}_{i_j}}{\M_{i_j}}}{j ∈ ℕ}
		\text{ is infinite}
		\label{eq6:prf:thm:streg:fin:bind:capt:chains}
		\\
		\text{and for all \m{\prefixed{\vec{x}_{i_j}}{\M_{i_j}} ∈ S~~ \length{\vec{x}_{i_j}} < l_0}}
	\end{gather}
	(\m{S} is infinite since the terms in \arewseq are distinct). On the other hand we have:
	\begin{gather}
		T := \setcompr{\prefixed{\vec{x}'_{i_j}}{\M_{i_j}}}{j ∈ ℕ} \; \subseteq \; \ST{\eagStrat}{\M} \text{ is finite}
	\end{gather}
	because \M is regular. However, since every term in \m{S} has a \m{\mred_\del}-reduct in \m{T} due to \cref{eq3:prf:thm:streg:fin:bind:capt:chains}, as well as an abstraction prefix of a length bounded by \m{l_0}, it follows by \cref{prop:compress:prefix:RegARS}~\cref{prop:compress:prefix:RegARS:ii}, that \m{S} also has to be finite, conflicting with \cref{eq6:prf:thm:streg:fin:bind:capt:chains}. We have reached a contradiction, and thereby established \cref{eq4:prf:thm:streg:fin:bind:capt:chains}.
	\par Now we are able to define a lower bound on the lengths of the prefixes in \arewseq that fulfils the requirements of \cref{lem:inf:bind:capt:chains}. We define the function:
	\[\lb{} : ℕ → ℕ,\quad n ↦ \lb{n} := \min \setcompr{\length{\vec{x}_{n'}}}{n' ≥ n}\]
	Its definition guarantees that \lb{} is a lower bound on the prefix lengths in \arewseq, and that \lb{} is non-decreasing. Furthermore also \m{\lim_{n→\infty} \lb{n} = \infty} follows by non-decreasingness, in addition to unboundedness of \lb{}: for arbitrary \m{l ∈ ℕ}, by \cref{eq4:prf:thm:streg:fin:bind:capt:chains} there exists \m{n_0 ∈ ℕ} such that \m{\length{\vec{x}_n} ≥ l} holds for all \m{n ∈ ℕ}, \m{n ≥ n_0}; this entails \m{\lb{n_0} ≥ l}.
	\par Now since \M, \arewseq, together with \lb{} as defined above, satisfy the assumptions of \cref{lem:inf:bind:capt:chains}, this lemma can be applied, yielding an infinite binding--capturing chain in \M.
\end{proof}

\begin{example}[infinite binding--capturing chain]\label{ex:bindcaptchain}
	The infinite λ-term from \cref{ex:entangled} with a representation as a higher-order recursive program scheme in \cref{ex:entangled-infinite-path}, which was recognised there to be regular but not strongly regular, possesses an infinite binding--capturing chain as indicated on the right in \cref{fig:entangled_bc}.
\end{example}

\section{Expressibility by terms in \texorpdfstring{\lambdaletreccal}{λletrec}}\label{sec:express}

\begin{para}[overview]
	In this section we finish the proof of our main characterisation result: we prove that every strongly regular λ-term is \lambdaletrec-expressible. For this purpose we introduce an annotated variant of one of the proof systems for strongly regular λ-terms. We show that every closed derivation in \stRegzero with conclusion \prefixed{}{\M}, which witnesses that \M is strongly regular, can be annotated, by adding appropriate \lambdaletrec-terms to each prefixed term in the derivation, into a derivation in the annotated system with conclusion \annprefixed{}{\L}{\M} such that the \lambdaletrec-term annotation \L expresses the λ-term \M. We show the correctness of this construction by transforming the derivation in the annotated proof system into a derivation in the proof system \stRegeq with conclusion \prefixed{}{\unfsem{\L}} = \prefixed{}{\M}, and then drawing upon the soundness of \stRegeq with respect to equality of strongly regular λ-terms.
\end{para}

\begin{figure}
	\proofsystem{
		\begin{bprooftree}
			\emptyAxiom
			\infLabel{\0}
			\unaryInf{\annprefixed{\vec{x}y}{y}{y}}
		\end{bprooftree}
		\hsep
		\begin{bprooftree}
			\axiom{\annprefixed{\vec{x}y}{\L}{\M}}
			\infLabel{λ}
			\unaryInf{\annprefixed{\vec{x}}{\abs{y}{\L}}{\abs{y}{\M}}}
		\end{bprooftree}
		\\
		\begin{bprooftree}
			\axiom{\annprefixed{\vec{x}}{\L_0}{\M_0}}
			\axiom{\annprefixed{\vec{x}}{\L_1}{\M_1}}
			\infLabel{@}
			\binaryInf{\annprefixed{\vec{x}}{\app{\L_0}{\L_1}}{\app{\M_0}{\M_1}}}
		\end{bprooftree}
		\\
		\begin{bprooftree}
			\axiom{\annprefixed{x_1 \dots x_{n-1}}{\L}{\M}}
			\infLabel{\S~~\sideCondition{if the binding \abs{x_n}{} is vacuous}}
			\unaryInf{\annprefixed{x_1 \dots x_n}{\L}{\M}}
		\end{bprooftree}
		\\
		\begin{bprooftree}
			\axiom{[\annprefixed{\vec{x}}{\aconstname_u}{\M}]^u}
			\noLine
			\unaryInf{\Deriv_0}
			\noLine
			\unaryInf{\annprefixed{\vec{x}}{\subst{\L}{u}{\aconstname_u}}{\M}}
			\infLabel{\FIX,u~~\mlSideCondition{if \m{\depth{\Deriv_0} ≥ 1}, and \m{\length{\vec{y}} ≥ \length{\vec{x}}} for all \\ \prefixed{\vec{y}}{\N} on threads in \m{\Deriv_0} from open \\ assumptions \m{(\annprefixed{\vec{x}}{u}{\M})^u} down}}
			\unaryInf{\annprefixed{\vec{x}}{(\letin{u = \L}{u})}{\M}}
		\end{bprooftree}
	}
	\caption{Annotated natural-deduction style proof system \annstRegzero for strongly regular λ-terms, a version of \stRegzero with \lambdaletrec-terms as annotations.}
	\label{fig:annstRegzero}
\end{figure}

We start by introducing a variant of the proof system \stRegzero in which the formulas are closed, prefixed, \lambdaletrec-term-annotated λ-terms.

\begin{definition}[the proof system \annstRegzero]
	The formulas of the proof system \annstRegzero are closed expressions of the form \annprefixed{\vec{x}}{\L}{\M} with \vec{x} a variable prefix vector, \abs{\vec{x}}{\L} a \lambdaletrec-term, and \abs{\vec{x}}{\M} a λ-term. The axioms and rules of \annstRegzero are annotated versions of the axioms and rules of the proof system \stRegzero from \cref{def:Reg:stReg:stRegzero} and \cref{fig:stReg:stRegzero}, and are displayed in \cref{fig:annstRegzero}.
\end{definition}

\begin{remark}
	For an example that illustrates why we have chosen to formulate an annotated version only of the proof system \stRegzero, but not of \stReg, please see \cref{example:annstRegzero}.
\end{remark}

The following proposition is a statement that is entirely analogous to \cref{prop:Reg:stReg}. 

\begin{proposition}[cycles are guarded]\label{prop:annstRegzero:sidecondition}
	For all for all instances \ainst of the rule \FIX in a derivation \Deriv (possibly with open assumptions) in \annstRegzero it holds: every thread from \ainst upwards to a marked assumption that is discharged at \ainst passes at least one instance of a rule \ruleref{λ} or \ruleref{@}.
\end{proposition}

The lemma below states a straightforward connection between derivations in
\stRegzero and derivations in its annotated version \annstRegzero.

\begin{lemma}[from \stRegzero- to \annstRegzero-derivations, and back]\label{lem:stRegzero:2:annstRegzero}
	The following transformations are possible between derivations in \stRegzero and derivations in \annstRegzero:
	\begin{enumerate}[(i)]
		\item\label{lem:stRegzero:2:annstRegzero:i}
			Every derivation \Deriv in \stRegzero with conclusion \prefixed{\vec{x}}{\M} can be transformed into a derivation \Derivann in \annstRegzero with conclusion \annprefixed{\vec{x}}{\L}{\M} such that there is a bijective correspondence between marked assumptions \m{(\prefixed{\vec{y}}{\M})^{\amarker}} in \Deriv and marked assumptions \m{(\annprefixed{\vec{y}}{\amarker}{\M})^{\amarker}} in \Derivann. (As a consequence, \Derivann is a closed derivation if \Deriv is closed.) More precisely, \Derivann can be obtained from \Deriv by replacing every term occurrence \prefixed{\vec{y}}{\N} by an occurrence of \annprefixed{\vec{y}}{P}{\N} for a prefixed \lambdaletrec-term \prefixed{\vec{y}}{P} with the property that every prefix variable \m{y_i} bound in \P is also bound in \N. Thereby occurrences of marked assumptions and axioms \0 in \Deriv give rise to occurrences of marked assumptions and axioms \0 in \Derivann, respectively; instances of the \stRegzero-rules λ, @, \S, and \FIX in \Deriv give rise to instances of \annstRegzero-rules λ, @, \S, and \FIX in \Derivann, respectively.
		\item\label{lem:stRegzero:2:annstRegzero:ii}
			From every closed derivation \Deriv in \annstRegzero with conclusion \annprefixed{\vec{x}}{\L}{\M} a closed derivation \check{\Deriv} in \stRegzero with conclusion \prefixed{\vec{x}}{\M} can be obtained by dropping the annotations with \lambdaletrec-terms.
	\end{enumerate}
\end{lemma}

\begin{proof}
	Statement \cref{lem:stRegzero:2:annstRegzero:i} of the lemma can be established through a proof by induction on the depth \depth{\Deriv} of a derivation \Deriv in \stRegzero with possibly open assumptions. In the base case, axioms (\0) of \stRegzero are annotated to axioms (\0) of \annstRegzero, and marked assumptions \m{(\prefixed{\vec{y}}{\N})^{\amarker}} in \stRegzero to marked assumptions \m{(\annprefixed{\vec{y}}{\aconstname_{\amarker}}{\N})^{\amarker}}. In the induction step it has to be shown that a derivation \Deriv in \stRegzero with immediate subderivation \m{\Deriv_0} can be annotated to a derivation \Derivann in \annstRegzero, using the induction hypothesis which guarantees that an annotated version \m{\Derivann_0} of \m{\Deriv_0} has already been obtained.
	Then for obtaining \Derivann from \m{\Derivann_0} the fact is used that the rules in \annstRegzero uniquely determine the annotation in the conclusion of an instance once the annotation(s) in the premise(s) (and in the case of \FIX additionally the annotation markers used in the assumptions that are discharged) are given. In order to establish that instances of \S in \Deriv give rise to corresponding instances of \S in \Derivann, the part of the induction hypothesis is used which guarantees that the \lambdaletrec-term annotation in the premise contains not more variable bindings than the λ-term it annotates.

	\par Statement \cref{lem:stRegzero:2:annstRegzero:ii} of the lemma is a consequence of the fact that, by dropping the \lambdaletrec-term-annotations, every instance of a rule of \annstRegzero give rise to an instance of the corresponding rule in \stRegzero. Formally the statement can again be established by induction on the depth of derivations in \annstRegzero.
\end{proof}

\begin{example}\label{example:annstRegzero}
	The derivation \m{\Deriv_l} in \stRegzero from \cref{example:stReg} on the left can be annotated, as described by \cref{lem:stRegzero:2:annstRegzero}~\cref{lem:stRegzero:2:annstRegzero:ii}, to obtain the following derivation \m{\Hat{\Deriv}_l} in \annstRegzero:
	\begin{prooftree}
		\axiom{(\annprefixed{}{\aconstname_{\amarker}}{\M})^{\amarker}}
		\infLabel{\S}
		\unaryInf{\annprefixed{x}{\aconstname_{\amarker}}{\M}}
		\infLabel{\S}
		\unaryInf{\annprefixed{xy}{\aconstname_{\amarker}}{\M}}
		\emptyAxiom
		\infLabel{\0}
		\unaryInf{\annprefixed{xy}{y}{y}}
		\infLabel{@}
		\binaryInf{\annprefixed{xy}{\app{\aconstname_{\amarker}}{y}}{\app{\M}{y}}}
		\emptyAxiom
		\infLabel{\0}
		\unaryInf{\annprefixed{x}{x}{x}}
		\infLabel{\S}
		\unaryInf{\annprefixed{xy}{x}{x}}
		\infLabel{@}
		\binaryInf{\annprefixed{xy}{\app{\app{\aconstname_{\amarker}}{y}}{x}}{\app{\app{\M}{y}}{x}}}
		\infLabel{λ}
		\unaryInf{\annprefixed{x}{\abs{y}{\app{\app{\aconstname_{\amarker}}{y}}{x}}}{\app{\app{\M}{y}}{x}}}
		\infLabel{λ}
		\unaryInf{\annprefixed{}{\abs{xy}{\app{\app{\aconstname_{\amarker}}{y}}{x}}}{\abs{xy}{\app{\app{\M}{y}}{x}}}}
		\infLabel{\FIX, u}
		\unaryInf{\annprefixed{}{(\letin{\amarker = \abs{xy}{\app{\app{\amarker}{y}}{x}}}{f})}{\M}}
	\end{prooftree}
	Note that the term in the conclusion, which has been extracted by the annotation procedure, is actually the same as the \lambdaletrec-term \m{\letin{f = \abs{xy}{\app{\app{f}{y}}{x}}}{f}} which was used in \cref{example:stReg} to define \M as its infinite unfolding.
	\par Furthermore note that, in a variant of \annstRegzero in which the `\stRegzero-addition' (concerning abstraction prefix lengths) to the side-condition of \FIX is dropped, the derivation \m{\Deriv_r} in \cref{example:stReg} on the right could be annotated to obtain the following proof tree \m{\Hat{\Deriv}_r}:
	\begin{prooftree}
		\axiom{(\annprefixed{x}{\aconstname_{\amarker}}{\abs{y}{\app{\app{\M}{y}}{x}}})^{\amarker}}
		\infLabel{λ}
		\unaryInf{\annprefixed{}{\abs{x}{\aconstname_{\amarker}}}{\M}}
		\infLabel{\S}
		\unaryInf{\annprefixed{x}{\abs{x}{\aconstname_{\amarker}}}{\M}}
		\infLabel{\S}
		\unaryInf{\annprefixed{xy}{\abs{x}{\aconstname_{\amarker}}}{\M}}
		\emptyAxiom
		\infLabel{\0}
		\unaryInf{\annprefixed{xy}{y}{y}}
		\infLabel{@}
		\binaryInf{\annprefixed{xy}{\app{(\abs{x}{\aconstname_{\amarker}})}{y}}{\app{\M}{y}}}
		\emptyAxiom
		\infLabel{\0}
		\unaryInf{\annprefixed{x}{x}{x}}
		\infLabel{\S}
		\unaryInf{\annprefixed{xy}{x}{x}}
		\infLabel{@}
		\binaryInf{\annprefixed{xy}{\app{(\app{\abs{x}{\aconstname_{\amarker}})}{y}}{x}}{\app{\app{\M}{y}}{x}}}
		\infLabel{λ}
		\unaryInf{\annprefixed{x}{\abs{y}{\app{\app{(\abs{x}{\aconstname_{\amarker}})}{y}}{x}}}{\abs{y}{\app{\app{\M}{y}}{x}}}}
		\infLabel{\ainst, \u}
		\unaryInf{\annprefixed{x}{\letin{\amarker = \abs{y}{\app{\app{\amarker}{y}}{x}}}{\amarker}}{\abs{y}{\app{\app{\M}{y}}{x}}}}
		\infLabel{λ}
		\unaryInf{\annprefixed{}{\abs{x}{\letin{\amarker = \abs{y}{\app{\app{\amarker}{y}}{x}}}{\amarker}}}{\M}}
	\end{prooftree}
	Observe that, equally as was the case for \m{\Deriv_r}, also in \m{\Hat{\Deriv}_r} there occurs, on the thread between the marked assumption at the top and the rule instance \ainst at which this assumption is discharged, a formula, namely \annprefixed{}{\amarker}{\M}, that has a shorter abstraction prefix than the formula in the premise and conclusion of \ainst as well as in the assumption. Thus \ainst is not an instance of the rule \FIX in \annstRegzero.
	\par Furthermore note that the \lambdaletrec-term extracted by \m{\Hat{\Deriv}_r} does not unfold to \M, and hence does not express \M. This example shows that the side-condition on instances of \FIX in \annstRegzero cannot be weakened to the form used for the rule \FIX in \stReg when the aim is to extract a \lambdaletrec-term that unfolds to the infinite λ-term in the conclusion.
\end{example}

The central property of the proof system \annstRegzero still remains to be shown:
that the \lambdaletrec-terms in the conclusion of a derivation in this system
does actually unfold to the infinite λ-term in the conclusion.
This will be established below in \cref{lem:annstRegzero:2:stRegeq} and \cref{thm:annstRegzero}.
But as an intermediary proof system that will allow us to use results about the
proof system \stRegletrec from \cref{sec:proofs}, we also introduce an annotated
version of the rule \letrec in \stRegletrec, and an according annotated proof system.

\begin{figure}
	\proofsystem{
		\begin{bprooftree}
			\axiom{\set{[\annprefixed{\vec{x}}{\aconstname_{f_i}}{\M_i}]}_{i=1,\dots,n}}
			\noLine
			\unaryInf{\Deriv_j}
			\noLine
			\unaryInf{\set{\dots \annprefixed{\vec{x}}{\subst{\L_j}{\vec{f}}{\vec{\aconstname}_{\vec{f}}}}{\M_j} \dots}_{j=0,\dots,n}}
			\infLabel{\FIXletrec}
			\unaryInf{\annprefixed{\vec{x}}{(\letin{f_1 = \L_1 \dots f_n = \L_n}{\L_0})}{\M_0}}
		\end{bprooftree}
		\proofpar where \m{\aconstname_{f_1}, \dots,\aconstname_{f_n}} are distinct constants fresh for \m{\L_0,\dots,\L_n}, and substitutions \m{\subst{\L_l}{\vec{f}}{\vec{\aconstname}_{\vec{f}}}} stands for \m{\L_l [f_1 := \aconstname_{f_1}, \dots, f_n := \aconstname_{f_n}]}.
		\proofpar \emph{side-conditions}: \m{\length{\vec{y}} ≥ \length{\vec{x}}} holds for the prefix length of every \prefixed{\vec{y}}{\N} on a thread in \m{\Deriv_j} for \m{0 ≤ j ≤ n} from an open assumptions \m{(\prefixed{\vec{x}}{\aconstname_{f_i}})^{\amarker_i}} downwards; for bottommost instances: the arising derivation is guarded on access path cycles.
	}
	\caption{The proof system \annstRegletrec for \lambdaletrec-terms arises from the proof system \annstRegzero (see \cref{fig:annstRegzero}) by replacing the rule \FIX with the rule \FIXletrec.}
	\label{fig:annstRegletrec}
\end{figure}

\begin{definition}[the proof system \annstRegletrec]
	The proof system \annstRegletrec arises from \stRegzero by replacing the rule \FIX by the rule \FIXletrec in \cref{fig:annstRegletrec}, an annotated version of the rule \FIXletrec from \cref{def:Regletrec:stRegletrec} and \cref{fig:Regletrec:stRegletrec}. The side-condition on bottommost instances of \FIXletrec to be guarded on access path cycles is analogous as explained in \cref{def:Regletrec:stRegletrec}.
\end{definition}

\begin{proposition}[from \annstRegletrec- to \stRegletrec-derivations]\label{prop:annstRegletrec:2:stRegletrec}
	Let \Deriv be a closed derivation in \annstRegletrec with conclusion \annprefixed{}{\L}{\M}. Then a closed derivation \check{\Deriv} in \stRegletrec with conclusion \prefixed{}{\L} can be obtained by removing the λ-terms in \Deriv while keeping the \lambdaletrec-term-annotations.
\end{proposition}

\begin{proposition}[from \annstRegzero- to \annstRegletrec-derivations]\label{prop:annstRegzero:2:annstRegletrec}
	Every derivation \Deriv in \annstRegzero can be transformed into a derivation \m{\Deriv'} in \annstRegletrec with the same conclusion and with the same open assumption classes.
\end{proposition}

\begin{proof}
	First note that \annstRegzero and \annstRegletrec differ only by the specific version of assumption-discharging rule in the system, \FIX in \annstRegzero and \FIXletrec in \annstRegletrec.
	For showing the proposition, let \Deriv be a derivation in \annstRegzero.
	\par We define a proof tree \m{\Deriv'}, (intended to be a derivation in \annstRegletrec) by repeatedly replacing topmost occurrences of \FIX at the bottom of subderivations of the form as depicted in \cref{fig:annstRegzero}, by simulating subderivations of the form:
	\begin{prooftree}
		\axiom{(\annprefixed{\vec{x}}{\aconstname_u}{\M})^u}
		\axiom{[\annprefixed{\vec{x}}{\aconstname_u}{\M}]^u}
		\noLine
		\unaryInf{\Deriv_0}
		\noLine
		\unaryInf{\annprefixed{\vec{x}}{\subst{\L}{u}{\aconstname_u}}{\M}}
		\infLabel{\FIXletrec,u}
		\binaryInf{\annprefixed{\vec{x}}{(\letin{u = \L}{u})}{\M}}
	\end{prooftree}
	until all occurrences of instances of \FIX have been replaced by instances of \FIXletrec. The result is a proof tree with axioms and rules of \annstRegletrec \m{+} \FIXletrecmin, with the same conclusion and the same classes of open assumptions as \Deriv, but in which rule instances carrying the label \FIXletrec might actually be instances of \FIXletrecmin, unless actually proven (as will be done below) to be instances of \FIXletrec.
	\par Now first note that, due to the form of the introduced instances of \FIXletrec, every formula occurrence in \Deriv is reachable on an access path of \m{\Deriv'}. Second, note that relative access paths \m{\apath'} in \m{\Deriv'} starting at the conclusion of an instance \m{\ainst'} of \FIXletrec up to a marked assumption that is discharged at \m{\ainst'} descend from a thread \apath in \Deriv from the conclusion of an application \ainst of \FIX up to a marked assumption that is discharged at \ainst. Since by \cref{prop:annstRegzero:sidecondition} the thread \m{\apath'} passes at least one instance of a rule \ruleref{λ} or \ruleref{@}, this is also the case for \apath. As a consequence, all cycles on relative access paths are guarded. Thus \Deriv is guarded. Hence all occurrences of rule names \FIXletrec in \m{\Deriv'} rightly label occurrences of this rule, and \m{\Deriv'} is a derivation in \annstRegletrec, which moreover is guarded.
\end{proof}

\begin{example}\label{example:annstRegletrec}
	The closed derivation \m{\Hat{\Deriv}_l} in \cref{example:annstRegzero} can be transformed into the following closed derivation in \annstRegletrec:
	\begin{prooftree}
		\axiom{(\annprefixed{}{\aconstname_{\amarker}}{\M})^{\amarker}\hspace{-3em}}
		\axiom{(\annprefixed{}{\aconstname_u}{\M})^{\amarker}}
		\infLabel{\S}
		\unaryInf{\annprefixed{x}{\aconstname_u}{\M}}
		\emptyAxiom
		\infLabel{\0}
		\unaryInf{\annprefixed{x}{x}{x}}
		\infLabel{@}
		\binaryInf{\annprefixed{x}{\app{\aconstname_u}{x}}{\app{\M}{x}}}
		\infLabel{\S}
		\unaryInf{\annprefixed{xy}{\app{\aconstname_u}{x}}{\app{\M}{x}}}
		\emptyAxiom
		\infLabel{\0}
		\unaryInf{\annprefixed{xy}{y}{y}}
		\infLabel{@}
		\binaryInf{\annprefixed{xy}{\app{\app{\aconstname_u}{x}}{y}}{\app{\app{\M}{x}}{y}}}
		\infLabel{λ}
		\unaryInf{\annprefixed{x}{\abs{y}{\app{\app{\aconstname_u}{x}}{y}}}{\app{\app{\M}{x}}{y}}}
		\infLabel{λ}
		\unaryInf{\annprefixed{}{\abs{xy}{\app{\app{\aconstname_u}{x}}{y}}}{\abs{xy}{\app{\app{\M}{x}}{y}}}}
		\infLabel{\FIXletrec, \u}
		\binaryInf{\annprefixed{}{(\letin{u = \abs{xy}{\app{\app{u}{x}}{y}}}{u})}{\M}}
	\end{prooftree}
\end{example}

Now we concentrate on the remaining matter of proving that the \lambdaletrec-term
obtained by the annotation process from a closed derivation in \stRegzero to one in \annstRegzero
does indeed unfold to the λ-term it annotates.
For this, we establish a proof-theoretic transformation from derivations in \annstRegzero
to derivations in \stRegeq.

\begin{lemma}[from \annstRegzero- to \stRegeq-derivations]\label{lem:annstRegzero:2:stRegeq}
	Let \Deriv be a closed derivation in \annstRegzero with conclusion \annprefixed{}{\L}{\M}. Then \defd{\unfsem{\L}}, and \Deriv can be transformed into a closed derivation \m{\Deriv'} in \stRegeq with conclusion \prefixed{}{\unfsem{\L}} = \prefixed{}{\M} by:
	\begin{itemize}
		\item replacing each formula occurrence \o of \annprefixed{\vec{y}}{P}{\N} in \Deriv by an occurrence of the formula \m{\unfsem{\prefixed{\vec{y}}{\letin{B}{\tilde{P}}}} = \prefixed{\vec{y}}{\N}} in \m{\Deriv'}, where \B arises as the union of all outermost binding groups in conclusions of instances of \FIX at or below \o, and where \prefixed{\vec{y}}{P} = \prefixed{\vec{y}}{\subst{\tilde{P}}{\vec{f}}{\vecsub{\aconstname}{\vec{f}}}} and \vec{f} is comprised of the function variables occurring in \B and \vecsub{\aconstname}{\vec{f}} are distinct constants for \vec{f} as chosen by \Deriv; the unfoldings involved here are always defined.
	\end{itemize}
\end{lemma}

\begin{proof}
	Let \Deriv be a closed derivation in \annstRegzero with conclusion \annprefixed{}{\L}{\M}.
	\par By \cref{prop:annstRegzero:2:annstRegletrec}, \Deriv can be transformed into a closed derivation \m{\Deriv_1} in \annstRegletrec with the same conclusion. Due to \cref{prop:annstRegletrec:2:stRegletrec}, by dropping the infinite terms in \m{\Deriv_1}, a derivation \m{\Deriv_2} in \stRegletrec with conclusion \prefixed{}{\L} can be obtained. Then it follows from \cref{thm:Regletrec:stRegletrec} that \defd{\unfsem{\L}}, that is, that \L unfolds to a λ-term.
	\par We have to show that the transformation of \Deriv into \m{\Deriv'} as described in the statement of the lemma is, on the one hand, possible (that is, the unfolding of each prefixed \lambdaletrec-term is indeed defined), and on the other hand, that the proof tree \m{\Deriv'} obtained by these replacements is indeed a valid derivation in \stRegeq.
	\par We argue for the possibility of these replacements and for their correctness locally, that is by carrying out the replacements from the bottom of \Deriv upwards, thereby recognising for every replacement step that it is possible, and that it indeed produces a valid inference in \stRegzero.
	\par As a typical example of the arguments necessary to establish this fact, we consider a derivation \Deriv in \annstRegzero with in it an instance of \ruleref{λ} that immediately succeeds an instance of \FIX:
	\begin{prooftree}
		\axiom{[\annprefixed{\vec{x}y}{\aconstname_{\amarker}}{\M_0}]^{\amarker}}
		\noLine
		\unaryInf{\Deriv_0}
		\noLine
		\unaryInf{\annprefixed{\vec{x}y}{\subst{\L_0}{\amarker}{\aconstname_{\amarker}}}{\M_0}}
		\infLabel{\FIX, \amarker}
		\unaryInf{\annprefixed{\vec{x}y}{\letin{\amarker = \L_0}{\amarker}}{\M_0}}
		\infLabel{λ}
		\unaryInf{\annprefixed{\vec{x}}{(\abs{y}{\letin{\amarker = \L_0}{\amarker}})}{\abs{y}{\M_0}}}
		\noLine
		\unaryInf{\Deriv_{00}'}
		\noLine
		\unaryInf{\annprefixed{}{\L}{\M}}
	\end{prooftree}
	According to the statement of the lemma, \Deriv is transformed into the following
	\stRegeq-proof-tree:
	\begin{prooftree}
		\axiom{\dots
		(\unfsem{\prefixed{\vec{x}y}{\letin{B_0 \eqsep \amarker = \tilde\L'_0 \eqsep B'}{\amarker}}} = \prefixed{\vec{x}y}{\M_0})^{\amarker} \dots}
		\noLine
		\unaryInf{\Deriv_0'}
		\noLine
		\unaryInf{\unfsem{\prefixed{\vec{x}y}{\letin{B_0 \eqsep \amarker = \tilde\L'_0}{\tilde\L'_0}}} = \prefixed{\vec{x}y}{\M_0}}
		\infLabel{\FIX, \amarker}
		\unaryInf{\unfsem{\prefixed{\vec{x}y}{\letin{B_0 \eqsep \amarker = \tilde\L'_0}{\amarker}}} = \prefixed{\vec{x}y}{\M_0}}
		\infLabel{λ}
		\unaryInf{\unfsem{\prefixed{\vec{x}}{\letin{B_0}{\abs{y}{\letin{\amarker = \tilde\L'_0}{\amarker}}}}} = \prefixed{\vec{x}}{\abs{y}{\M_0}}}
		\noLine
		\unaryInf{\Deriv_{00}'}
		\noLine
		\unaryInf{\unfsem{\prefixed{}{\letin{}\L}} = \prefixed{}{\M}}
	\end{prooftree}
	where \m{B_0} arises as the union of all outermost binding groups in conclusions of instances of \FIX strictly below the visible instance of \FIX, and \m{B'} is the union of all outermost binding groups in conclusions of instances of \FIX strictly above the visible instance of \FIX and below the indicated marked assumptions (this binding group differs for different marked assumptions of this assumption class), and where \m{\tilde\L'_0} is the result of replacing in \m{\L_0} all occurrences of constants \m{\aconstname_f} by the function variable \f from which it originates.
	\par Now assuming that the unfolding in the conclusion of the visible instance of \ruleref{@} has been shown to exist, we want to recognise that this instance and the instance of \FIX above are valid instances in \stRegzero. For the instance of \ruleref{λ} we have to show:
	\begin{align*}
		&\defd{\unfsem{\prefixed{\vec{x}}{\letin{B_0}{\abs{y}{\letin{\amarker = \tilde\L'_0}{\amarker}}}}}}
		\\&\Longrightarrow∃\abs{\vec{x}y}{\N_0}~ 
		\begin{aligned}[t]
			& \defd{\unfsem{\prefixed{\vec{x}y}{\letin{B_0 \eqsep \amarker = \tilde\L'_0}{\tilde\L'_0}}}} = \prefixed{\vec{x}y}{\N_0}∧
			\\ & \unfsem{\prefixed{\vec{x}}{\letin{B_0}{\abs{y}{\letin{\amarker = \tilde\L'_0}{\amarker}}}}} = \prefixed{\vec{x}}{\abs{y}{\N_0}}
		\end{aligned}
	\end{align*}
	This, however, is a straightforward consequence of the following \m{\red_\unfold}-rewrite-steps:
	\rewritingSequence{
		\prefixed{\vec{x}}{\letin{B_0}{\abs{y}{\letin{\amarker = \tilde\L'_0}{\amarker}}}}
		\step{\red_{\unfoldpre{λ}}}{\prefixed{\vec{x}}{\abs{y}{\letin{B_0}{\letin{\amarker = \tilde\L'_0}{\amarker}}}}}
		\step{\red_{\unfoldpre{\merg}}}{\prefixed{\vec{x}}{\abs{y}{\letin{B_0,\amarker = \tilde\L'_0}{\amarker}}}}
	}
	in view of the fact that, by \cref{lem:unique_normal_forms}, unfolding is uniquely normalising (in at most \omega steps).
	And for the instance of \FIX we have to show:
	\begin{multline*}
		\defd{\underbrace{\unfsem{\prefixed{\vec{x}y}{\letin{B_0 \eqsep \amarker = \tilde\L'_0}{\amarker}}}}_{= (\star)}}
		\Longrightarrow
		\begin{aligned}[t]
			&
			\defd{\unfsem{\prefixed{\vec{x}y}{\letin{B_0 \eqsep \amarker = \tilde\L'_0}{\tilde\L'_0}}}}
			= (\star)
			\\
			& ∧ ~
			\defd{\unfsem{\prefixed{\vec{x}y}{\letin{B_0 \eqsep \amarker = \tilde\L'_0 \eqsep B'}{\amarker}}}}
			= (\star)
		\end{aligned}
	\end{multline*}
	(actually the statement as in the second line has to be shown for every binding-group \m{B'} that occurs for marked assumptions discharged at the instance of \FIX). This implication is a consequence of the \m{\red_\unfold}-rewrite-steps:
	\begin{gather*}
		\prefixed{\vec{x}y}{\letin{B_0 \eqsep \amarker = \tilde\L'_0}{\amarker}}
		\m{\red_{\unfoldpre{\rec}}}
		\prefixed{\vec{x}y}{\letin{B_0 \eqsep \amarker = \tilde\L'_0}{\tilde\L'_0}}
		\\
		\prefixed{\vec{x}y}{\letin{B_0 \eqsep \amarker = \tilde\L'_0 \eqsep B'}{\amarker}}
		\m{\red_{\unfoldpre{\reduce}}}
		\prefixed{\vec{x}y}{\letin{B_0 \eqsep \amarker = \tilde\L'_0}{\amarker}}
	\end{gather*}
	again in view of the statement of \cref{lem:unique_normal_forms}.
	\par The arguments used here are typical, and can be carried out similarly also for showing that axioms (\0), and instances of rules \ruleref{@} and \ruleref{\S} in \annstRegzero-derivations give rise to, under the transformation described in the statement of the lemma, valid instances of axioms (\0), and instances of \ruleref{@} and \ruleref{\S}, respectively, in \stRegeq-derivations.
\end{proof}

\begin{theorem}[\annstRegzero and unfolding semantics]\label{thm:annstRegzero}
	If \derivablein{\annstRegzero}{\annprefixed{}{\L}{\M}} holds for a \lambdaletrec-term \L and a λ-term \M, then \L unfolds to \M.
\end{theorem}

\begin{proof}
	Suppose that \Deriv is a closed derivation in \annstRegzero with conclusion \annprefixed{}{\L}{\M}. \Cref{lem:annstRegzero:2:stRegeq} entails that \L unfolds to a λ-term, and moreover, that \Deriv can be transformed into a closed derivation \m{\Deriv'} in \stRegeq with conclusion \prefixed{}{\unfsem{\L}} = \prefixed{}{\M}. Then it follows by \cref{thm:stRegeq} (applying soundness of \stRegzero with respect to the property of \lambdaletrec-terms to unfold to a λ-term that \m{\unfsem{\L} = \M}, and hence that \m{\L \m{\omegared_\unfold} \M}. In this way we have found a \lambdaletrec-term \L that expresses \M.
\end{proof}

We now arrive at our main characterisation result.

\begin{theorem}[\lambdaletrec-expressibility \sim strong regularity]\label{thm:ll-expressible:streg}
	A λ-term is \lambdaletrec-expressible if and only if it is strongly regular.
\end{theorem}

\begin{proof}
	Let \M be a λ-term.
	\par The direction ``⇒'' is the statement of \cref{thm:ll-expressible:2:streg}.
	\par For showing the direction ``⇐'' in the statement of the theorem, we assume that \M is strongly regular. Then by \cref{lem:Reg:stReg}~\cref{lem:Reg:stReg:ii}, there exists a closed derivation \Deriv in \stReg with conclusion \prefixed{}{\M}. Due to \cref{lem:stRegzero:2:annstRegzero}~\cref{lem:stRegzero:2:annstRegzero:i}, \Deriv can be transformed into a derivation \Derivann in \annstRegzero with conclusion \annprefixed{}{\L}{\M}, for some \lambdaletrec-term \L. Then it follows by \cref{thm:annstRegzero} that the \lambdaletrec-term \L expresses \M.
\end{proof}

As an immediate consequence of \cref{thm:ll-expressible:streg} and of \cref{cor:thm:streg:fin:bind:capt:chains} we obtain the following theorem, a summary of our main results:

\begin{theorem}
	For all λ-terms \M the following statements are equivalent:
	\begin{enumerate}[(i)]
		\item \M is \lambdaletrec-expressible.
		\item \M is strongly regular.
		\item \M is regular, and it only contains finite binding--capturing chains.
	\end{enumerate}
\end{theorem}

\section{λ-transition-graphs}\label{sec:lambda-transition-graphs}

\begin{para}[Overview]
	In this section we introduce the concept of λ-transition-graphs.
	λ-transition-graphs are a nameless graphical representations of λ-terms that arise naturally from the decomposition systems from this chapter. The λ-transition-graph of \M is in fact almost identical to \M's reduction graph with respect to some \extscope-delimiting strategy for \stRegCRS. The only difference is that the former is a (pointed) LTS and the latter an ARS. As explained earlier in \cref{rem:diff-ARS-LTS} these formalisms are essentially same, except for the definition of bisimulation, which for LTSs is sensitive to the transition labels. The main result of this section is a coinduction principle for λ-terms: two λ-terms are equal if and only if they have bisimilar λ-transition-graphs.
\end{para}

\begin{para}[outlook: λ-term-graphs]
	The λ-transition-graphs which we study in this section will be further developed in \cref{chap:representations} into `λ-term-graphs', which are first-order term graph and serve as a graphical representation of λ-terms and \lambdaletrec-terms.
\end{para}

\begin{definition}[transition system induced by an ARS]
	Let \m{\aARS = \tuple{O,\steps,\src{},\tgt{}}} be an ARS (or a sub-ARS (or a sub-ARS of a labelled version) of an ARS) that is induced by a CRS with rules \m{R} (see \cref{def:inducedARS}). Note that, every step in \m{\steps : O \times R \times O} carries information according to from which rule it stems from. By the \emph{LTS induced by \aARS} we mean the LTS \m{\aLTS_\aARS = \triple{O}{R}{\steps}} in which the steps in \aARS according to rule \arule are interpreted as transitions with label \arule.
\end{definition}

\begin{definition}[transition graph of a λ-term]
	Let \M be a λ-term with reduction graph \GeneratedSubARSi{\aARS}{M} w.r.t.\ to an ARS \aARS and let \triple{O}{R}{T} be the LTS induced by \GeneratedSubARSi{\aARS}{M}. We call the labelled transition graph \m{\ltg{\aARS}{M} = \quadruple{O}{R}{M}{T}} the \emph{transition graph of \M w.r.t.\ \aARS}.
\end{definition}

\begin{definition}[λ-transition-graph]\label{def:lambda-transition-graph}
	We call a labelled transition graph \m{G = \quadruple{\states}{L}{i}{T}} a \emph{λ-transition-graph} if:
	\begin{itemize}
	\item it is connected
	\item the labels are \m{L = \set{λ,\S,@_0,@_1}}
	\item there are no infinite paths in \G consisting solely of \S-transitions
	\item every state \m{s} belongs to one of the following kinds: λ-states, \S-states, and @-states, where
		\begin{itemize}
			\item a λ-state is the source of precisely one λ-transition, and no other transitions: \m{\setcompr{\pair{l}{t}}{s \red_l t} = \set{\pair{λ}{u}} \text{ for some \m{u ∈ \states}}}.
			\item a \S-state is the source of precisely one \S-transition, and no other transitions: \m{\setcompr{\pair{l}{t}}{s \red_l t} = \set{\pair{\S}{u}} \text{ for some \m{u ∈ \states}}}.
			\item a @-state is the source of precisely one \m{@_0}-transition and one \m{@_1}-transition, and no other transitions:\\
				\m{\setcompr{\pair{l}{t}}{s \red_l t} = \set{\pair{@_0}{f},\pair{@_1}{x}} \text{ for some \m{f,x ∈ \states}}}.
		\end{itemize}
	\end{itemize}
\end{definition}

\begin{proposition}[eager scope-closure yields λ-transition-graphs]\label{prop:s-is-a-lambda-ltg}
	Let \astratplus be a \extscope-delimiting strategy for \stRegARS. For every term \m{\M ∈ \iTer\lambdaprefixcal} the transition graph \ltg{\astratplus}{\M} is a λ-transition-graph. We call \ltg{\astratplus}{\M} \emph{the λ-transition-graph of \M with respect to \astratplus}.
\end{proposition}

\begin{proof}
	In the transition graph \ltg{\astratplus}{\M} there cannot be infinitely many successive \S-transitions because in the ARS that induces \ltg{\astratplus}{\M}, the rewrite relation \m{\red_\S} is terminating, due to \cref{prop:rewprops:RegCRS:stRegCRS}~\cref{prop:rewprops:RegCRS:stRegCRS:iii}.
\end{proof}

\begin{para}[λ-transition-graphs of \lambdaletrec-terms]
	Along the lines of \cref{prop:s-is-a-lambda-ltg} we can also view transition graphs of \lambdaletrec-terms as λ-transition-graphs; however, unfolding steps must not be included as transitions (let us call them \emph{silent transitions}). As hinted at in \cref{rem:nondet:unfolding}, here it is important to restrict scope/\extscope-delimiting strategies to ones that are deterministic in the application of unfolding rules.
\end{para}

\begin{definition}[LTS with silent steps]
	Let \m{\aLTS_\aARS = \triple{O}{R}{\steps}} be the LTS induced by ARS \aARS and let \m{R_0} be a subset of \m{R}. Then by \emph{the LTS induced by \aARS with silent \m{R_0}-steps} we mean the LTS \m{\SilentLTS{\aARS}{R_0} = \triple{O}{R}{T}} with
	\begin{equation*}
		T := \setcompr{\triple{o}{\arule}{o'}}{\multilinebox{if \m{o \mred_{R_0} ⋅ \red_{\arule} o'} where \m{\mred_{R_0}} are steps w.r.t.\ rules \\ in \m{R_0}, and \m{\red_{\arule}} is a step w.r.t.\ a rule \m{\arule ∈ R \setminus R_0}}}
	\end{equation*}
	in which the steps in \aARS according to rules in \m{R_0} are interpreted as silent transitions.
\end{definition}

\begin{definition}[LTG with silent steps]
	Let \o be an object of the ARS \aARS and \m{\SilentLTS{\GeneratedSubARS{o}}{R_0} = \triple{O}{R}{T}} an LTS with silent \m{R_0}-steps. We call \m{\SilentLTG{\aARS}{R_0}{o} := \quadruple{O}{R}{o}{T}} the \emph{transition graph of \o with silent \m{R_0}-steps}.
\end{definition}

\begin{proposition}
	Let \astratplus be a \extscope-delimiting strategy for \stRegletrecARS. For every term \m{\L ∈ \Ter\lambdaletrecprefixcal}, the transition graph \SilentLTG\astratplus\unfCRS\L of \L is a λ-transition-graph.
\end{proposition}

\begin{definition}[λ-transition-graph of a \lambdaletrec-term]\label{def:lambda-transition-graph-of-a-term}
	Let \m{\L ∈ \Ter\lambdaletrecprefixcal} be a (prefixed) \astratplus-productive \lambdaletrec-term. For a \extscope-delimiting strategy \astratplus of \stRegletrecARS such that \L is \astratplus-productive, we call the transition graph \SilentLTG{\astratplus}{\unfCRS}{\L} \emph{the λ-transition-graph of \L with respect to \astratplus}. We also speak of λ-transition-graphs of terms \m{\L ∈ \Ter\lambdaletreccal} or \m{\M ∈ \iTer\lambdacal} by which we refer to the λ-transition-graphs of \prefixed{}\L and \prefixed{}\M, respectively.
\end{definition}

\begin{theorem}[coinduction principle for \lambdacal]\label{thm:coinduction-principle}
	For all λ-terms \M and \N the following statements are equivalent:
	\begin{enumerate}[(i)]
		\item\label{thm:coinduction-principle:i} \m{\M = \N}.
		\item\label{thm:coinduction-principle:ii} \m{\derivablein{\EqTer}{\M = \N}}.
		\item\label{thm:coinduction-principle:iii} \M and \N have bisimilar λ-transition-graphs.
	\end{enumerate}
\end{theorem}

\begin{proof}
	In view of \cref{prop:AlphaPreTer:EqTer}~\cref{prop:AlphaPreTer:EqTer:ii}, the logical equivalence between \cref{prop:AlphaPreTer:EqTer:i} and \cref{prop:AlphaPreTer:EqTer:ii}, and the fact that \cref{thm:coinduction-principle:i} ⇒ \cref{thm:coinduction-principle:iii} clearly holds, it suffices to show that \cref{thm:coinduction-principle:iii} ⇒ \cref{thm:coinduction-principle:ii} holds.
	\par For this, suppose that \ltg{\astratplus_1}{\M} and \ltg{\astratplus_2}{\N} are bisimilar for some \extscope-delimiting strategies \m{\astratplus_1} and \m{\astratplus_2} for \RegARS. But now bisimilarity of these transition graphs guarantees that a derivation \Deriv in \EqTer with conclusion \prefixed{}{\M} = \prefixed{}{\N} can be constructed such that all threads in \Deriv correspond to \m{\red_{\astratplus_1}}-rewrite-sequences on \M and to \m{\red_{\astratplus_2}}-rewrite-sequences on \M, respectively. If the construction process is organised in a depth-fair manner (for example, all non-axiom leafs at depth \n are extended by appropriate rule instances, before extensions are carried out at depth greater than \n), then in the limit a completed derivation \infDeriv with conclusion \prefixed{}{\M} = \prefixed{}{\N} is obtained. This establishes \m{\derivablein{\EqTer}{\M = \N}}.
\end{proof}

\begin{conjecture}[coinduction principle for \lambdaletreccal]
	For all \m{\L_1,\L_2 ∈ \Ter{\lambdaletreccal}} it holds that \m{\L_1 = \L_2} if and only if \m{\L_1} and \m{\L_2} have bisimilar λ-transition-graphs.
	\begin{proof}[Proof sketch]
		In \cref{chap:maxsharing} we develop a coinduction principle for \lambdaletreccal (\cref{thm:graphrep:classltgs}) based on λ-term-graphs. λ-term-graphs are introduced in \cref{chap:representations} and are closely related to λ-transition-graphs. The conjecture can in all likelihood be validated by \cref{thm:graphrep:classltgs} after relating λ-transition-graphs to λ-term-graphs in a formal manner.
	\end{proof}
\end{conjecture}

\begin{para}[only \stRegCRS defines nameless representations]\label{rem:nameless-repr}
	The coinduction principle above holds for transition graphs derived from terms w.r.t.\ an eager \extscope-delimiting strategy for \stRegARS. It cannot be extended to \RegARS, which is witnessed by the following counterexample.
	\par Consider the transition graph \ltg{\eagStrat}{M} of the term \m{M = \abs{xy}{\app{\app{x}{x}}{y}}} w.r.t.\ the eager scope-delimiting strategy \eagStrat for \RegARS. The corresponding \RegARS-reduction graph is depicted on the left in \cref{fig:four-strategies}. Each of the following terms yields the exact same transition graph w.r.t.\ to an appropriately-chosen scope-delimiting strategy for \RegARS. For the two terms in the middle the eager scope-delimiting strategy can be chosen.
	\[
		\abs{xy}{\app{\app{x}{x}}x}
		\hspace{1cm}
		\abs{xy}{\app{\app{x}{x}}y}
		\hspace{1cm}
		\abs{xy}{\app{\app{y}{y}}x}
		\hspace{1cm}
		\abs{xy}{\app{\app{y}{y}}y}
	\]
\end{para}

\begin{para}[readback for λ-transition-graphs]
	The understanding of λ-transition-graphs as nameless representations of λ-terms implies that from a such graphs the corresponding λ-term can be extracted. We define a function for this purpose by means of a CRS which implements the assembly of a λ-term from the infinite unfolding of a λ-transition-graph. The function is closely related to the \stParseCRS in the sense that \stParseCRS does both destruct and reconstruct the scrutinised term while \readback{} only implements the reconstruction.
\end{para}

\begin{definition}[readback for λ-transition-graphs]
	\newcommand\sreadwriten[1]{{\textsf{rw}}_{#1}}
	\newcommand\readwriten[3]{{\sreadwriten{#1}}\left(#2, #3\right)}
	\newfunction\readwritezero{\m{{\textsf{rw}}_{0}}}
	\begin{align*}
		\readback{} : \iTer{\set{\0, λ, @, \S}} & {} ⇀ \iTer{\lambdacal}\\
		t & {} ↦ \readback{t} := ~ \parbox[t]{130pt}{infinite normal form of \readwritezero{t} w.r.t.\ the following CRS:}
	\end{align*}
	\begin{align*}
		\readwriten{n}{X_1,\dots,X_n}{λ(t_0)} & {} \;\red\; \absCRS{x}{\readwriten{n+1}{X_1,\dots,X_n,x}{t_0}} \\
		\readwriten{n}{\vec{X}}{@(t_0){t_1}} & {} \;\red\; \appCRS{\readwriten{n}{\vec{X}}{t_0}}{\readwriten{n}{\vec{X}}{t_1}} \\
		\readwriten{n+1}{\vec{X},x}{\S(t_0)} & {} \;\red\; \readwriten{n}{\vec{X}}{t_0} \\
		\readwriten{n}{X_1,\dots,X_n}{\0} & {} \;\red\; X_n
	\end{align*}
	The function is partial because \m{\sreadwriten n} is unproductive for infinite \S-chains. That restriction comes forth accordingly in the definition of λ-transition-graphs (\cref{def:lambda-transition-graph}). The function is thus complete on the subset of \m{\iTer{\set{\0, λ, @, \S}}} that is obtained from unfolding a λ-transition-graph.
\end{definition}

\section{Summary}\label{sec:expressibility:conclusion}

\begin{para}[A CRS for decomposing λ-terms]
	To characterise the set of \lambdaletrec-expressible λ-terms we established a framework of formalisms for `observing' λ-terms coinductively. First we introduced prefixed λ-terms that enrich λ-terms by an abstraction prefix. On the prefixed terms we defined the CRS \stRegARS in which a rewrite sequence corresponds to a deconstruction of a term along one of its paths. In that sense a prefixed term \prefixed{\vec{x}}{\M} can be understood as a `suspended decomposition' which has not advanced into subterm \M yet. Such a decomposition describes a path through the term by observations of the form \m{\red_λ}, \m{\red_{@_0}}, \m{\red_{@_1}}, \m{\red_\S}, where the \m{\red_\S} delimits the \extscope of an abstraction.
\end{para}

\begin{para}[scope/\extscope-delimiting strategies]
	Since there is some freedom as to where \extscope-delimiters can be placed, we defined \extscope-delimiting strategies to formalise specific possible choices eliminating that freedom and thereby making the observations deterministic except for the forking into the left or the right subterm of an application. By means of \extscope-delimiting strategies we formulated two important concepts: strong regularity and λ-transition graphs.
\end{para}

\begin{para}[strong regularity]
	The intuitive understanding of strong regularity is the property of a λ-term \M that from \M every `sufficiently eager' \extscope-delimiting strategy can only generate a finite number of terms. We showed that \lambdaletrec-expressibility coincides with strong regularity.
\end{para}

\begin{para}[λ-transition-graphs]
	Every \extscope-delimiting strategy defines for each term a λ-transition-graph which can be viewed as a nameless graphical representation very similar to its term graph in de-Bruijn notation with the difference that \S-nodes are not restricted to occur only near leafs but can be shared by variables (see \cref{S}). We established a coinduction principle for λ-transition-graphs that states that two λ-terms are equal if and only if their λ-transition graphs w.r.t. to a \extscope-delimiting strategy are bisimilar. The eager \extscope-delimiting strategy yields finite λ-transition-graphs for strongly regular λ-terms.
\end{para}

\begin{para}[\lambdaletreccal]
	We then adapted the concepts of the CRS for observing terms, \extscope-delimiting strategies, and λ-transition-graphs and applied them to \lambdaletreccal{} proving similar results as for \lambdacal.
\end{para}

\begin{para}[proof systems for strong regularity]
	We provided a proof system that is sound and complete for the notion of strong regularity and which admits finite proofs for strongly regular λ-terms. We define an annotated version of the proof system which not unlike an attribute-grammar definition implements the extraction of a \lambdaletrec-term \L from a proof for term \M in that system, such that \L unfolds to the \M. We show that every \extscope-delimiting strategy induces a proof and that from a proof a corresponding history-aware strategy can be deduced, which suggests a similar correspondence between λ-transition-graphs and proofs.
\end{para}

\chapter{Term Graph Representations for Strongly Regular λ-Terms}\label{chap:representations}

\section{Overview}

\setcounter{theorem}{-1}
\begin{para}[teaser]
	Can't we just do bisimulation on \lambdaletrec-terms? See also \cref{maxsharing:teaser}.
\end{para}


\begin{para}[subject matter]\label{representations-subject-matter}
	In this chapter we study various classes of higher-order and first-order term graphs (intended as representations for \lambdaletrec-terms). We focus on the relation between `λ-higher-order term graphs' (λ-ho-term-graphs), which are first-order term graphs endowed with a well-behaved scope function, and their representations as `λ-term-graphs', which are plain first-order term graphs with scope-delimiter vertices that meet certain scoping requirements. Specifically we tackle the question: Which class of first-order term graphs admits a faithful embedding of λ-ho-term-graphs in the sense that (i) the homomorphism-based sharing-order on λ-ho-term-graphs is preserved and reflected, and (ii) the image of the embedding corresponds closely to a natural class (of λ-term-graphs) that is closed under functional bisimulation?
\end{para}

\begin{para}[motivation]
	We study these graph formalisms in isolation -- that is to say without formally connecting them to \lambdaletrec-terms. But we do so with a long term goal in mind: to find a graph formalism that is suitable to adequately represent \lambdaletrec-terms. Once we have found such a graph formalism, in \cref{chap:maxsharing} we relate it back to \lambdaletreccal, by which we gain further insights concerning unfolding semantics of \lambdaletreccal and also to obtain concrete practical methods to analyse and manipulate \lambdaletrec-terms.
\end{para}

\begin{para}[methods and formalisms]
	The term graph formalims we study arise naturally from the term decomposition systems from the previous chapter. They are closely related to the λ-transition-graphs, and we derive different classes of term graphs from the term decomposition systems by a similar approach. We systematically examine which classes of λ-term-graphs satisfy the properties in \cref{representations-subject-matter}.
\end{para}

\begin{para}[results]
	We identify a particular class of first-order term graphs with these properties. Term graphs of this class are built not only from application, abstraction, and variable vertices, but also scope-delimiter vertices. They have the characteristic feature that the latter two kinds of vertices have backlinks to the corresponding abstraction. This result puts a handle on the concept of subterm sharing for higher-order term graphs, both theoretically and algorithmically: We obtain an easily implementable method for obtaining the maximally shared form of λ-ho-term-graphs. Also, we open up the possibility to transfer properties from first-order term graphs to λ-ho-term-graphs. In fact we prove in this way that the sharing-order on a set of bisimilar λ-ho-term-graphs forms a complete lattice \cref{sec:transfer}.
\end{para}

\begin{para}[outlook]
	In \cref{chap:maxsharing} we use these insights to develop practical applications w.r.t\ \lambdaletreccal:
	\begin{itemize}
		\item an efficient test for whether two \lambdaletrec-terms have the same unfolding
		\item a partial order for the amount of subterm sharing in a \lambdaletrec-term leading to
		\begin{itemize}
			\item a notion of maximal sharing for \lambdaletrec
			\item an efficient mechanism to compute the maximally shared form of a \lambdaletrec-term which generalises common subexpression elimination
		\end{itemize}
	\end{itemize}
\end{para}

\section{Preliminaries}\label{sec:representations:prelims}

\begin{para}[term graphs]
	The graph formalisms that we study in this chapter all based on term graphs (in contrast to transition graphs as in the last chapter), which is a natural choice, considering the three types of expressions (abstraction, application, variable occurrence) that make up λ-terms. Here is some notation and terminology revolving around term graphs.
\end{para}

\begin{definition}[term graph]
	Let \sig be a signature with arity function \m{\arity{} : \sig → ℕ}. A \emph{term graph over \sig} (or a \emph{\sig-term-graph}) is a tuple \m{\tuple{\V,\lab{},\args{},\r}} where \V is a set of \emph{vertices}, \m{\lab{} : \V → \sig} the \emph{(vertex) label function}, \m{\args{} : \V → \V^*} the \emph{argument function} that maps every vertex \v to the word \args{\v} consisting of the \arity{\lab{\v}} successor vertices of \v (hence it holds that \m{\length{\args{\v}} = \arity{\lab{\v}}}), and \r, the \emph{root}, is a vertex in \V. Note that term graphs may have infinitely many vertices.
\end{definition}

\begin{definition}[root connected term graphs]\label{def:root-connected}
	We say that such a term graph is \emph{root connected} if every vertex is reachable from the root by a path that arises by repeatedly going from a vertex to one of its successors. We denote by \tgsover{\sig} and by \tgsminover{\sig} the class of all root-connected term graphs over \sig, and the class of all term graphs over \sig, respectively. By `term graphs' we will from now on, always mean root-connected term graphs, except in a few situations in which we explicitly state otherwise.
\end{definition}

\begin{definition}[successor relations]
	Let \atg be a term graph over signature \sig. As useful notation for picking out the \i-th vertex, from among the ordered successors \args{\v} of a vertex \v of \atg,
	we define for each \m{i ∈ ℕ} the indexed edge relation \m{{\tgsucc_i} \subseteq {\V\times\V}}, and additionally the (not indexed) edge relation \m{{\tgsucc} \subseteq \V\times\V}, by stipulating for all \m{\v,\v' ∈ \V}:
	\begin{align*}
		\v \tgsucc_i \v' &~:⇔~ ∃ \v_0,\dots,\v_n ∈ \V ~~ \args{\v} = \v_0 \dots \v_n ∧ \v' = \v_i &
		\\
		\v \tgsucc   \v' &~:⇔~ ∃ i ∈ ℕ~~ \v \tgsucc_i \v' &
	\end{align*}
	We write \m{\v \tgsuccis{i}{l} \v'} if \m{\v \tgsucc_i \v' ∧ \lab{\v} = l} holds for \m{\v,\v' ∈ \V}, \m{i ∈ ℕ}, \m{l ∈ \sig}, to indicate the label at the source of an edge.
\end{definition}

\begin{definition}[paths]
	Let \m{\v_0,\dots,\v_n ∈ \V}. A \emph{path} in \atg is a tuple \m{\tuple{\v_0,l_1,\v_1,l_2,\v_2,l_3,\dots,l_{n-1},\v_{n-1},l_n,\v_n}} and \m{n,l_1,\dots,l_n ∈ ℕ} such that \m{\v_0 \tgsucc_{l_1} \v_1 \tgsucc_{l_2} \v_2 \tgsucc_{l_3} \dots \tgsucc_{l_n} \v_n} holds; paths will usually be denoted in the latter form, using indexed edge relations.
\end{definition}

\begin{definition}[access paths]
	An \emph{access path} of a vertex \v of a term graph \atg is a path that starts at the root of \atg, ends in \v, and does not visit any vertex twice. Note that every vertex \v has at least one access path: since every vertex in a term graph is reachable from the root (see \cref{def:root-connected}), there is a path \apath from \r to \v; then an access path of \v can be obtained from \apath by repeatedly cutting out \emph{cycles}, that is, parts of the path between one and the same vertex.
\end{definition}

In the following, let \m{\atg_1 = \tuple{\V_1,\labi{1}{},\argsi{1}{},\r_1}}, \m{\atg_2 = \tuple{\V_2,\labi{2}{},\argsi{2}{},\r_2}} be term graphs over signature \sig.

\begin{definition}[homomorphism, functional bisimulation]\label{def:ltgs:homomorphism}
	A \emph{homomorphism}, also called a \emph{functional bisimulation}, from \m{\atg_1} to \m{\atg_2} is a morphism from the structure \m{\tuple{\V_1,\labi{1}{},\argsi{1}{},\r_1}} to the structure \m{\tuple{\V_2,\labi{2}{},\argsi{2}{},\r_2}}, i.e., a function \m{h : \V_1 → \V_2} such that, for all \m{\v ∈ \V_1} it holds:
	\begin{align}
		\tag{roots}
		\label{ltgs:homomorphism:cond:roots}
		h(\r_1) & = \r_2
		\\
		\tag{labels}
		\label{ltgs:homomorphism:cond:labels}
		\labi{1}{\v} & = \labi{2}{h(\v)}
		\\
		\tag{arguments}
		\label{ltgs:homomorphism:cond:arguments}
		h^*(\argsi{1}{\v}) & = \argsi{2}{h(\v)}
	\end{align}
	where \m{h^*} is the homomorphic extension (also: pointwise lifted version) of \h to words over \m{\V_1}, i.e.\ \m{h^* : \V_1^* → \V_2^*}, \m{\v_1 \dots \v_n ↦ h(\v_1) \dots h(\v_n)}. In this case we write \m{\atg_1 \funbisim_{\h} \atg_2}, or \m{\atg_2 \convfunbisim_{\h} \atg_1}. And we write \m{\atg_1 \funbisim \atg_2}, or for that matter \m{\atg_2 \convfunbisim \atg_1}, if there is a homomorphism from \m{\atg_1} to \m{\atg_2}.
	\par Let \m{\afunsym ∈ \sig}. An \emph{\afunsym-homomorphism} from \m{\atg_1} to \m{\atg_2} is a homomorphism \h from \m{\atg_1} to \m{\atg_2} that `shares' (i.e.\ maps to the same vertex) only vertices with the label \afunsym, i.e.\ \h has the property that \m{h(v_1) = h(v_2) \;⇒\; \labi{1}{v_1} = \labi{1}{v_2} = \afunsym} holds for all \m{v_1 ≠ v_2 ∈ \V_1}. If \h is an \afunsym-homomorphism from \m{\atg_1} to \m{\atg_2}, then we write \m{\atg_1 \funbisim^{\afunsym}_{\h} \atg_2} or \m{\atg_2 \convfunbisim^{\afunsym}_{\h} \atg_1}, or dropping \h, \m{\atg_1 \funbisim^{\afunsym} \atg_2} or \m{\atg_2 \convfunbisim^{\afunsym} \atg_1}.
\end{definition}

\begin{terminology}
	The terms `homomorphism' and `functional bisimulation' will we used interchangeably throughout this document.
\end{terminology}

\begin{definition}[isomorphism]\label{def:ltgs:iso}
	An \emph{isomorphism} between \m{\atg_1} and \m{\atg_2} is a bijective homomorphism \m{i : \V_1 → \V_2} from \m{\atg_1} to \m{\atg_2} (it follows from the homomorphism conditions \cref{def:ltgs:homomorphism} that also the inverse function \m{i^{-1} : \V_1 → \V_2} is a homomorphism). If there is an isomorphism between \m{\atg_1} and \m{\atg_2}, we write \m{\atg_1 \iso \atg_2}, and say that \m{\atg_1} and \m{\atg_2} are \emph{isomorphic}. The relation \iso is an equivalence relation on \tgsover{\sig}. For every term graph \atg over \sig we denote the isomorphism equivalence class \eqcl{\atg}{\iso} by (the boldface letter) \atgiso.
\end{definition}

\begin{definition}[bisimulation, bisimilarity]\label{def:bisimulation_ltgs}
	For \m{i ∈ \set{1,2}}, let \m{\atg_i = \tuple{\V_i,\labi{i}{},\argsi{i}{},\r_i}} be term graphs over signature \sig. A \emph{bisimulation} between \m{\atg_1} and \m{\atg_2} is a relation \m{\abisim \subseteq \V_1\times\V_2} such that the following conditions hold, for all \m{\pair{v}{v'} ∈ \abisim}:
	\begin{align}
		\tag{roots}
		\pair{\r_1}{\r_2} ∈ \abisim
		\\
		\tag{labels}
		\labi{1}{v} = \labi{2}{v'}
		\\
		\tag{arguments}
		\pair{\argsi{1}{v}}{\argsi{2}{v'}} ∈ \abisim^*
	\end{align}
	where the extension \m{\abisim^* \subseteq {V_1}^* \times {V_2}^*} of \abisim to a relation between words over \m{V_1} and words over \m{V_2} is defined as:
	\begin{equation*}
		\abisim^* :=
		\setcompr{\pair{v_1 \dots v_n}{w_1 \dots w_n}}{
			\begin{array}{l}
				\m{n ∈ ℕ, ∀i ∈ \set{1,\dots,n}~ v_i ∈ V_1, w_i ∈ V_2}
				\\
				\text{such that~} \pair{v_i}{w_i} ∈ \abisim \text{ for all \m{1 ≤ i ≤ n}}
			\end{array}
		}
	\end{equation*}
	We write \m{\atg_1 \bisim \atg_2} if there is a bisimulation between \m{\atg_1} and \m{\atg_2}, and say that \m{\atg_1} and \m{\atg_2} are \emph{bisimilar}. Bisimilarity \bisim is an equivalence relation on classes \tgsover{\sig} of term graphs over a signature \sig.
	\par An \emph{\afunsym-bisimulation} between \m{\atg_1} and \m{\atg_2} is a bisimulation between \m{\atg_1} and \m{\atg_2} such that its restriction to vertices with labels different from \afunsym is a bijective function. If there is an \afunsym-bisimulation between \m{\atg_1} and \m{\atg_2} we say that \m{\atg_1} and \m{\atg_2} are \afunsym-bisimilar and write \m{\bisim^{\afunsym}} to indicate \emph{\afunsym-bisimilarity}.
\end{definition}



The following proposition is a simple but useful reformulation of the definition
of homomorphism (\cref{def:ltgs:homomorphism}).

\begin{proposition}[homomorphism]\label{prop:hom:tgs}
	Let \m{\atg_i = \tuple{\V_i,\labi{i}{},\argsi{i}{},\r_i}}, for \m{i ∈ \set{1,2}} be term graphs over signature \sig. Let \m{h : \V_1 → \V_2} be a function. Then \h is a homomorphism from \m{\atg_1} to \m{\atg_2}, if and only if, for all \m{v,v_1 ∈ \V_1}, \m{v_2 ∈ \V_2}, and \m{k ∈ ℕ}, the following four statements hold:
	\begin{gather}
		\tag{roots}
		h(\r_1) = \r_2
		\\
		\tag{labels}
		\lab{v} = \lab{h(v)}
		\\
		\tag{args-forward}\label{eq:args-forward}
		\begin{aligned}
			h(v_1) = v_2 ~∧~
			& ∀ v'_1 ∈ \V_1 ~ v_1 \tgsucc_k v'_1 \\ ⇒~
			& ∃ v'_2 ∈ \V_2 ~ v_2 \tgsucc_k v'_2 ~∧~ h(v'_1) = v'_2
		\end{aligned}
		\\
		\tag{args-backward}\label{eq:args-backward}
		\begin{aligned}
			h(v_1) = v_2 ~∧~
			& ∀ v'_2 ∈ \V_2 ~ v_2 \tgsucc_k v'_2 \\ ⇒~
			& ∃ v'_1 ∈ \V_1 ~ v_1 \tgsucc_k v'_1 ~∧~ h(v'_1) = v'_2
		\end{aligned}
	\end{gather}
\end{proposition}

\begin{proposition}[functional bisimulation and paths]\label{prop:hom:tgs:paths}
	For \m{i ∈ \set{1,2}} let \m{\atg_i = \tuple{\V_i,\labi{i}{},\argsi{i}{},\r_i}} be term graphs over signature \sig. Let \m{h : \V_1 → \V_2} be a homomorphism from \m{\atg_1} to \m{\atg_2}. Then the following statements hold:
	\begin{enumerate}[(i)]
		\item\label{prop:hom:tgs:paths:i}
			Every path \m{\apath : v_0 \tgsucc_{k_1} v_1 \tgsucc_{k_2} v_2 \tgsucc_{k_3} \dots \tgsucc_{k_{n-1}} v_{n-1} \tgsucc_{k_n} v_n} in \m{\atg_1} has an image \m{h(\apath)} under \h in \m{\atg_2} in the sense that \m{h(\apath) : h(v_0) \tgsucc_{k_1} h(v_1) \tgsucc_{k_2} h(v_2) \tgsucc_{k_3} \dots \tgsucc_{k_{n-1}} h(v_{n-1}) \tgsucc_{k_n} h(v_n)}.
		\item\label{prop:hom:tgs:paths:ii}
			For every \m{v_0 ∈ \V_0} and \m{v'_0 ∈ \V_1} with \m{h(v_0) = v'_0} it holds that every path \m{\apath' : v'_0 \tgsucc_{k_1} v'_1 \tgsucc_{k_1} v'_2 \tgsucc_{k_2} \dots \tgsucc_{k_{n-1}} v'_{n-1} \tgsucc_{k_{n-1}} v'_n} in \m{\atg_2} has a pre-image under \h in \m{\atg_1} that starts in \m{v_0}: a unique path \m{\apath : v_0 \tgsucc_{k_1} v_1 \tgsucc_{k_2} v_2 \tgsucc_{k_3} \dots \tgsucc_{k_{n-1}} v_{n-1} \tgsucc_{k_n} v_n} in \m{\atg_1} such that \m{\apath' = h(\apath)} holds in the sense of \cref{prop:hom:tgs:paths:i}.
		\item\label{prop:hom:tgs:paths:iii}
			\m{h(\V_1) = \V_2}, that is, \h is surjective.
	\end{enumerate}
\end{proposition}

\begin{proof}
	Statement \cref{prop:hom:tgs:paths:i} can be shown by induction on the length of \apath, using the \cref{eq:args-forward} property of \h from \cref{prop:hom:tgs}. Analogously, statement \cref{prop:hom:tgs:paths:ii} can be established by induction on the length of \m{\apath'}, using the \cref{eq:args-backward} property of \h from \cref{prop:hom:tgs}. Statement \cref{prop:hom:tgs:paths:iii} follows by applying, for given \m{v' ∈ \V_2}, the statement of \cref{prop:hom:tgs:paths:ii} to an access path \m{\apath'} of \m{v'} in \m{\atg_2} (which exists because term graphs are defined to be root-connected).
\end{proof}

\begin{proposition}\label{prop:funbisim:anti:symmetric}
	Let \atg, \m{\atg_1}, and \m{\atg_2} be term graphs over a signature \sig.
	\begin{enumerate}[(i)]
		\item\label{prop:funbisim:anti:symmetric:i}
			If \m{\atg \funbisim_{\h} \atg}, then \m{h = \idon{\V}{}},
			where \idon{\V}{} is the identity function on the set \V of vertices of \atg.
		\item\label{prop:funbisim:anti:symmetric:ii}
			If \m{\atg_1 \funbisim_{\h} \atg_2} and \m{\atg_2 \funbisim_g \atg_1} hold,
			then \h and \g are invertible, \m{h^{-1} = g},
			and consequently \m{\atg_1 \iso \atg_2}.
	\end{enumerate}
\end{proposition}

\begin{proof}
	For statement \cref{prop:funbisim:anti:symmetric:i}, suppose that \m{\atg \funbisim_{\h} \atg} holds for some homomorphism \h from \atg to itself. The fact that \m{h(\v) = \idon{\V}{\v}} holds for all vertices \v of \atg can be established by induction on the length of the shortest access path of \v in \atg. Note that we make use here of the root-connectedness of \atg (assumed implicitly, see \cref{sec:representations:prelims}) in the form of the assumption that every vertex can be reached by an access path. In order to show statement \cref{prop:funbisim:anti:symmetric:ii}, note that \m{\atg_1 \funbisim_{\h} \atg_2} and \m{\atg_2 \funbisim_g \atg_1} entail \m{\atg_1 \funbisim_{{\h} ∘ {g}} \atg_1} for homomorphisms \h and \g from \m{\atg_1} to \m{\atg_2}. From this \m{{\h} ∘ {g} = \idon{\V_1}{}}, where \m{\V_1} the set of vertices of \m{\atg_1}, follows by \cref{prop:funbisim:anti:symmetric:i}. Since analogously \m{{\h} ∘ {g} = \idon{\V_2}{}} follows, where \m{\V_2} is the set of vertices of \m{\atg_2}, the further claims follow.
\end{proof}

\begin{para}[sharing order]
	The homomorphism relation \funbisim is a preorder on term graphs over a given signature \sig. It induces a partial order on the isomorphism equivalence classes of term graphs over \sig, where anti-symmetry is implied by item \cref{prop:funbisim:anti:symmetric}~\cref{prop:funbisim:anti:symmetric:ii} of the following proposition. We will refer to \funbisim as the \emph{sharing preorder}, and to the induced relation on isomorphism equivalence classes as the \emph{sharing order}.
\end{para}

\begin{notation}
	Note that, deviating from some literature \cite{terese:2003}, we use the order relation \funbisim in the same direction as \m{≤\,}: if \m{\atg_1 \funbisim \atg_2}, then \m{\atg_2} is greater or equal to \m{\atg_1} with respect to the ordering \funbisim indicating that sharing is typically increased from \m{\atg_1} to \m{\atg_2}.
\end{notation}

\begin{notation}[equivalence classes of term graphs]\label{def:bisim-classes}
	Let \m{\aclass\subseteq\tgsover{\sig}} be a class of term graphs over signature \sig. For \m{\atg ∈ \tgsover{\sig}} we will use the notation
	\[\eqclin{\atgiso}{\bisim}{\aclass} := \setcompr{\atgiso'}{\altg' ∈ \aclass, \altg \bisim \altg'}\]
	to denote the bisimulation equivalence class of \atgiso (the \iso-equivalence-class of \altg) with respect to (the \iso-equivalence-classes in) \aclass. And we will write
	\[\succsofordin{\atgiso}{\funbisim}{\aclass} := \setcompr{\atgiso'}{\altg' ∈ \aclass, \altg \funbisim \altg'}\]
	to denote the class of all \iso-equivalence classes in \aclass that are reachable from \atgiso via functional bisimulation. For \m{\aclass = \tgsover{\sig}} we drop the superscript \aclass, and simply write \eqcl{\atgiso}{\bisim} and \succsoford{\atgiso}{\funbisim}.
\end{notation}

\begin{para}[complete lattice]
	A partially ordered set \m{\pair{\A}{≤}} is a \emph{complete lattice} if every subset of \A possesses a least upper bound and a greatest lower bound.
\end{para}

\begin{proposition}\label{prop:funbisim:succs:of:tgs:iso:complete:lattice}
	Let \sig be a signature, and \atg be a term graph over \sig. The bisimulation equivalence class \eqcl{\atgiso}{\bisim} of the isomorphism equivalence class \atgiso of \atg is ordered by functional bisimulation \funbisim such that \pair{\succsoford{\atgiso}{\funbisim}}{\funbisim} is a complete lattice.
\end{proposition}

\begin{remark}\label{rem:prop:funbisim:succs:of:tgs:iso:complete:lattice}
	The statement of \cref{prop:funbisim:succs:of:tgs:iso:complete:lattice} is a restriction to sets of \funbisim-successors of the statements \cite[Theorem~3.19]{ario:klop:1996} and \cite[Theorem~13.2.20]{terese:2003}, which confer the complete-lattice property for entire bisimulation equivalence classes (of \iso-equivalence classes) of term graphs.
\end{remark}

\begin{definition}[closedness under functional bisimulation]\label{def:closedness}
	Let \m{\aclass \subseteq \tgsover{\sig}} be a subclass of the term graphs over some signature \sig. We say that \aclass is \emph{closed under functional bisimulation} (\emph{closed under bisimulation}), if for all term graphs \m{\atg,\atg' ∈ \tgsover{\sig}}, whenever \m{\atg ∈ \aclass} and \m{\atg \funbisim \atg'} (\m{\atg \bisim \atg'}), then also \m{\atg' ∈ \aclass}.
\end{definition}

\section{Introduction}\label{sec:representations:intro}

\begin{para}[premise]
	In this chapter we seek to explore graph formalisms to adequately represent \lambdaletrec-terms. In developing these graph representations we draw heavily on ideas from the previous chapter such as λ-transition-graphs as a representation for λ-terms (\cref{sec:lambda-transition-graphs}) and the general concepts of including \extscope `exctended scope' in the formalisms. The inclusion of \extscope is intended to reflect the scoping rules in \lambdaletreccal, i.e.\ that variables bound by an abstraction or by \Let can only occur underneath their point of definition. Here we consider term graphs built from three kinds of vertices representating applications, abstractions, and variable occurrences, respectively.	
\end{para}


\begin{para}[three classes of graph formalisms]
	In particular we study the following three classes of term graphs:
	\begin{description}
		\item[λ-higher-order-term-graphs (\cref{sec:lambdahotgs})] are extensions of first-order term graphs by adding a scope function that assigns a set of vertices, its scope, to every abstraction vertex. There are two variants, one with and one without an edge (a \emph{backlink}) from each variable occurrence to its corresponding abstraction vertex. The class with backlinks is similar to \emph{higher-order term graphs} as defined by Blom in \cite{blom:2001}, and is in fact an adaptation of that concept to the λ-calculus.
		\item[abstraction-prefix based λ-higher-order-term-graphs (\cref{sec:lambdaaphotgs})] abbreviated as λ-ap-ho-term-graphs do not have a scope function but assign, to each vertex \v, an abstraction prefix consisting of a word of abstraction vertices that includes those abstractions for which \v is in their scope (it actually lists all abstractions for which \v is in their \emph{extended scope}, see \cref{def:scope:extscope}). Abstraction prefixes are aggregated information about the scopes entered so far (much as abstraction prefixes in the decomposition CRSs from the previous chapter; see \cref{prefixed-subterms}, \cref{prefixed-terms}, \cref{def:sig:lambdaprefixcal:CRS}).
		\item[λ-term-graphs with scope delimiters (\cref{sec:ltgs})] are plain first-order term graphs intended to represent both classes of higher-order term graphs above, and by extension \lambdaletrec-terms. Instead of relying upon added structures for describing scopes, they use scope-delimiter vertices to signify the end of scopes (much as the scope-delimiting steps in the decomposition CRSs from the previous chapter; see \cref{Reg:scope-delimiters}). Variable occurrences as well as scope delimiters may or may not have backlinks to their corresponding abstraction vertices.
	\end{description}
\end{para}

\begin{para}[desired properties]
	We develop these graph formalisms with a number of properties in mind, from which we expect to ensure that the representations are meaningful and useful:
	\begin{itemize}
		\item The term graphs should represent some \lambdaletrec-term, and in this sense are not be `meaningless'.
		\item Each of these classes induces a notion of functional bisimulation and bisimulation, which preserve the unfolding semantics of the term graphs (and therefore also the unfolding semantics of the \lambdaletrec-term they represent).
		\item Each of these classes induces a sharing order, which reflects the sharing present in the represented \lambdaletreccal-term.
		\item We are particularly interested in classes of term graphs that are closed under functional bisimulation, which ensures that when increasing sharing in a term graph it still remains meaningful.
	\end{itemize}
\end{para}

\begin{para}[relating the graph formalisms]
	It is important to stress that in this chapter we are not going to prove any of properties above that involve in any way \lambdaletrec-terms, particularly the former three. Instead we take a leap of faith and trust our intuition (and the authority of \cite{blom:2001}) that the λ-higher-order-term-graphs are indeed sound representation of \lambdaletrec-terms in that sense. As mentioned before focus on the graph formalisms themselves, and relate them amongst another. Particularly we establish a bijective correspondence between the λ-higher-order-term-graphs and λ-ap-ho-term-graphs, and a correspondence between λ-ap-ho-term-graphs and λ-term-graphs with scope delimiters that is `almost bijective' (bijective up to the sharing of scope delimiter vertices). We show that all of these correspondences preserve and reflect the sharing order.
	It is only in \cref{chap:maxsharing} that make a connection back to \lambdaletrec-terms. In particular, we supply a translation of \lambdaletrec-terms to λ-term-graphs with scope delimiters and back and show that the above desired properties do indeed hold. By the correspondences these properties extend also to the higher-order graph formalisms, which validate our conjecture about λ-higher-order-term-graphs.
\end{para}



\section{λ-higher-order-Term-Graphs}\label{sec:lambdahotgs}

\begin{definition}[signatures \sigTG, \sigTGi{0}, \sigTGi{1}]
	By \sigTG we denote the signature \set{@, λ} with \m{\arity{@} = 2}, and \m{\arity{λ} = 1}. By \sigTGi{i}, for \m{i ∈ \set{0,1}}, we denote the extension \m{\sigTG ∪ \set{\0}} of \sigTG where \m{\arity{\0} = i}. Whether the variable vertices have an outgoing edge depends on the value of \i. The intention is to consider two variants of term graphs, one with and one without variable backlinks to their corresponding abstraction vertex.
\end{definition}

\begin{definition}[term graphs classes \classtgssiglambdai{0} and \classtgssiglambdai{1}]
	The classes of term graphs over \sigTGi{0} and \sigTGi{1} are denoted by \m{\classtgssiglambdai{0} := \tgsover{\sigTGi{0}}} and \m{\classtgssiglambdai{1} := \tgsover{\sigTGi{1}}}, respectively.
\end{definition}

\begin{notation}[vertex subsets]\label{not:vertsof}
	Let \m{\atg = \tuple{\V,\lab{},\args{},\r}} be a term graph over signature \sig. For \m{l ∈ \sig} we denote by \vertsof{l} the set of \m{l}-vertices of \atg, that is, the subset of \V consisting of all vertices with label \m{l}; more formally, \m{\vertsof{l} := \setcompr{\v ∈ \V}{\lab{\v} = l}}.
\end{notation}

A `λ-higher-order-term-graph' consists of a \sigTGi{i}-term-graph together with a scope function that maps abstraction vertices to their \extscopes, which are subsets of the graph's vertices.

\begin{terminology}[scope \m{:=} \extscope]
	Henceforth when we write `scope' we mean \extscope.
\end{terminology}

\begin{definition}[scope function for \sigTGi{i}-term-graphs]\label{lhotg:scope-function}
	Let \m{i ∈ \set{0,1}} and \m{\atg = \tuple{\V,\lab{},\args{},\r}} be a \sigTGi{i}-term-graph. A function \m{\scope{} : \vertsof{λ} → \powerset(\V)} from λ-vertices of \atg to vertex sets of \atg is called a \emph{scope function} for \atg. Such a function is called \emph{correct} if for all \m{k ∈ \set{0,1}}, all vertices \m{v,w ∈ \V}, and all abstraction vertices \m{x,y ∈ \vertsof{λ}} the following holds:
	\begin{align}
		\tag{root}\label{lambdahotg:root}
		&⇒~ \r∉\scopemin{x}
		\\
		\tag{self}\label{lambdahotg:self}
		&⇒~ x ∈ \scope{x}
		\\
		\tag{nest}\label{lambdahotg:nest}
		x ∈ \scopemin{y} &⇒~ \scope{x} \subseteq \scopemin{y}
		\\
		\tag{closed}\label{lambdahotg:closed}
		v \tgsucc_k w \;∧\; w ∈ \scopemin{x} &⇒~ v ∈ \scope{x}
		\\
		\tag{scope\m{_0}}\label{lambdahotg:scope0}
		v ∈ \vertsof{\0} &⇒~ ∃ x ∈ \vertsof{λ} ~ v ∈ \scopemin{x}
		\\
		\tag{scope\m{_1}}\label{lambdahotg:scope1}
		v ∈ \vertsof{\0} \;∧\; v \tgsucc w &⇒~
		\left\{
			\begin{aligned}
				& w ∈ \vertsof{λ} \;∧ \\
				& v ∈ \scope{x} ⇔ w ∈ \scope{x}
			\end{aligned}
		\right.
	\end{align}
	where \m{\scopemin{v} := \scope{v}\setminus\set{v}}. Note that if \m{i=0}, then \cref{lambdahotg:scope1} is trivially true and hence superfluous, and if \m{i=1}, then \cref{lambdahotg:scope0} is redundant, because it follows from \cref{lambdahotg:scope1}.
	\par We say that \atg\ \emph{admits} a correct scope function if such a function exists for \atg.
\end{definition}

\begin{definition}[λ-ho-term-graph]\label{def:lambdahotg}
	Let \m{i ∈ \set{0,1}}. A \emph{λ-ho-term-graph} (short for \emph{λ-higher-order-term-graph}) over \sigTGi{i}, is a tuple \m{\ahotg = \tuple{\V,\lab{},\args{},\r,\scope{}}} where \m{\atg_{\ahotg} = \tuple{\V,\lab{},\args{},\r}} is a \sigTGi{i}-term-graph, called the term graph \emph{underlying} \ahotg, and \scope{} is a correct scope function for \m{\atg_{\ahotg}}. The classes of λ-ho-term-graphs over \sigTGi{0} and \sigTGi{1} will be denoted by \classlhotgsi{0} and \classlhotgsi{1}.
\end{definition}

\begin{example}
	Note that there is some freedom on how big the scopes are chosen. The minimal choice corresponds to \extscope in \cref{def:scope:extscope}. See \cref{fig:lambdahotg} for two different λ-ho-term-graphs over \sigTGi{i} both of which represent the same term in the λ-calculus with \letrec, namely \m{\letin{f = \abs{x}{\app{(\abs{y}{\app{y}{(\app{x}{g})}})}{(\abs{z}{\app{g}{f}})}},\; g = \abs{u}{u}}{f}}.
\end{example}

\begin{figure}
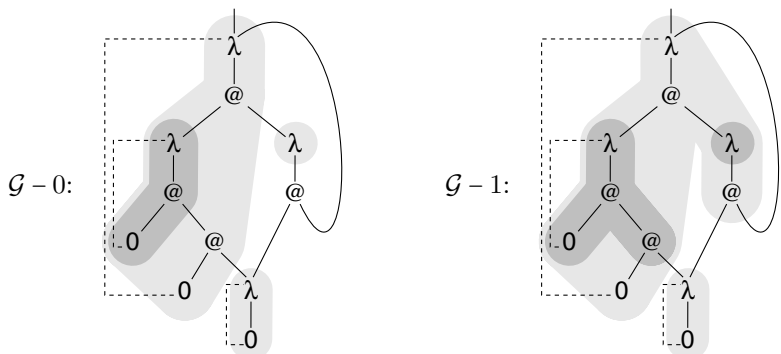

	\begin{hspread}
		\graph{\ahotg-0}{running-scopes-eager} & \graph{\ahotg-1}{running-scopes-non-eager}
	\end{hspread}
	\caption{\m{\ahotg_0} and \m{\ahotg_1} are λ-ho-term-graphs in \classlhotgsi{i} where the dotted backlink edges are present for \m{i=1}, but absent for \m{i=0}. The underlying term graphs of \m{\ahotg_0} and \m{\ahotg_1} are identical but their scope functions (signified by the shaded areas) differ. While in \m{\ahotg_0} scopes are chosen as small as possible (eager scope-closure), in \m{\ahotg_1} some scopes are closed only later in the graph.}
	\label{fig:lambdahotg}
\end{figure}

The following lemma states some basic properties of the scope function in λ-ho-term-graphs. Most importantly, scopes in λ-ho-term-graphs are properly nested, in analogy with \extscopes in finite λ-terms.

\begin{lemma}\label{lem:lhotg}
	Let \m{i ∈ \set{0,1}}, and let \m{\ahotg = \tuple{\V,\lab{},\args{},\r,\scope{}}} be a λ-ho-term-graph over \sigTGi{i}. For \m{v ∈ \V} and \m{x ∈ \vertsof{λ}} we say that \x is a \emph{binder for} \v if \m{v ∈ \scope{x}}, and we denote by \m{\binders{v} := \setcompr{x \in \vertsof{λ}}{v ∈ \scope{x}}} the set of binders of \v. Then the following statements hold for all \m{\w ∈ \V} and \m{\v,\v_1,\v_2 ∈ \vertsof{λ}}:
	\begin{enumerate}[(i)]
		\item\label{lem:lhotg:i} If \m{\w ∈ \scope{\v}}, then \v is visited on every access path of \w, and all vertices on access paths of \w after \v are in \scopemin{\v}. Hence (since \m{\atg_{\ahotg}} is a term graph, every vertex has an access path) \binders{\w} is finite.
		\item\label{lem:lhotg:ii} If \m{\scope{\v_1} ∩ \scope{\v_2} ≠ \emptyset} for \m{\v_1 ≠ \v_2}, then \m{\scope{\v_1} \subseteq \scopemin{\v_2}} or \m{\scope{\v_2} \subseteq \scopemin{\v_1}}. As a consequence, if \m{\scope{\v_1} ∩ \scope{\v_2} ≠ \emptyset}, then \m{\scope{\v_1} \subset \scope{\v_2}} or \m{\scope{\v_1} = \scope{\v_2}} or \m{\scope{\v_2} \subset \scope{\v_1}}.
		\item\label{lem:lhotg:iii} If \m{\binders{\w} ≠ \emptyset}, then \m{\binders{\w} = \set{\v_0,\dots,\v_n}} for \m{\v_0,\dots,\v_n ∈ \vertsof{λ}} and \m{\scope{\v_n} \subset \scope{\v_{n-1}} \dots \subset \scope{\v_0}}.
	\end{enumerate}
\end{lemma}

\begin{proof}
	Let \m{i ∈ \set{0,1}}, and let \m{\ahotg = \tuple{\V,\lab{},\args{},\r,\scope{}}} be a λ-ho-term-graph over \sigTGi{i}.
	\par For showing \cref{lem:lhotg:i}, let \m{\w ∈ \V} and \m{\v ∈ \vertsof{λ}} be such that \m{\w ∈ \scope{\v}}. Suppose that \m{\apath : \r = \w_0 \tgsucc_{k_1} \w_1 \tgsucc_{k_2} \w_2 \tgsucc_{k_3} \dots \tgsucc_{k_n} \w_n = \w} is an access path of \w. If \m{\w = \v}, then nothing remains to be shown. Otherwise \m{\w_n = \w ∈ \scopemin{\v}}, and, if \m{n>0}, then by \cref{lambdahotg:closed} it follows that \m{\w_{n-1} ∈ \scope{\v}}. This argument can be repeated to find subsequently smaller \i with \m{\w_i ∈ \scope{\v}} and \m{\w_{i+1},\dots,\w_n ∈ \scopemin{\v}}. We can proceed as long as \m{\w_i ∈ \scopemin{\v}}. But since, due to \cref{lambdahotg:root}, \m{\w_0 = \r∉\scopemin{\v}}, eventually we must encounter an \m{i_0} such that such that \m{\w_{i_0+1},\dots,\w_n ∈ \scopemin{\v}} and \m{\w_{i_0} ∈ \scope{\v}\setminus\scopemin{\v}}. This implies \m{\w_{i_0} = \v}, showing that \v is visited on \apath.
	\par For showing \cref{lem:lhotg:ii}, let \m{\w ∈ \V} and \m{\v_1,\v_2 ∈ \vertsof{λ}}, \m{\v_1 ≠ \v_2} be such that \m{\w ∈ \scope{\v_1} ∩ \scope{\v_2}}. Let \apath be an access path of \w. Then it follows by \cref{lem:lhotg:i} that both \m{\v_1} and \m{\v_2} are visited on \apath, and that, depending on whether \m{\v_1} or \m{\v_2} is visited first on \apath, either \m{\v_2 ∈ \scopemin{\v_1}} or \m{\v_1 ∈ \scopemin{\v_2}}. Then due to \cref{lambdahotg:nest} it follows that either
	\m{\scope{\v_2}\subseteq\scopemin{\v_1}} holds or \m{\scope{\v_1}\subseteq\scopemin{\v_2}}.
	\par Finally, statement \cref{lem:lhotg:iii} is an easy consequence of statement \cref{lem:lhotg:ii}.
\end{proof}

\begin{remark}[Comparison to `higher-order term graphs' \cite{blom:2001}]
	The notion of λ-ho-term-graph is an adaptation of the notion of `higher-order term graph' by Blom \cite[Definition\ 3.2.2]{blom:2001} for the purpose of representing finite or infinite λ-terms or cyclic (i.e. strongly regular) λ-terms. In particular, λ-ho-term-graphs over \sigTGi{1} correspond closely to higher-order term graphs over signature \sigTG. But they differ in the following respects:
	\begin{description}
		\item[abstractions:] Higher-order term graphs in \cite{blom:2001} are graph representations of finite or infinite terms in Combinatory Reduction Systems (CRSs). They typically contain abstraction vertices with label \Box that represent CRS abstractions. In contrast, λ-ho-term-graphs have abstraction vertices with label~λ that denote λ-abstractions.
		\item[signature:] Whereas higher-order term graphs in \cite{blom:2001} are based on an arbitrary CRS signature, λ-ho-term-graphs only contain the application symbol @ and the variable-occurrence symbol \0 in addition to the abstraction symbol λ.
		\item[variable backlinks and variable occurrence vertices:] In \cite{blom:2001}
		there are no explicit vertices that represent variable occurrences. Instead, variable occurrences are represented by backlink edges to abstraction vertices. Actually, in the formalisation chosen in \cite[Definition\hspace*{2pt}3.2.1]{blom:2001}, a backlink edge does not directly target the abstraction vertex \v it refers to, but ends at a special variant vertex \bar{\v} of \v. (Every such variant abstraction vertex \bar{\v} could be looked upon as a variable vertex that is shared by all edges that represent occurrences of the variable bound by the abstraction vertex \v.)
		\par In λ-ho-term-graphs over \sigTGi{1} a variable occurrence is represented by a variable vertex that as outgoing edge has a backlink to the abstraction vertex that binds the occurrence.
		\item[conditions on the scope function:] While the conditions \cref{lambdahotg:root}, \cref{lambdahotg:self}, \cref{lambdahotg:nest}, and \cref{lambdahotg:closed} on the scope function in higher-order term graphs in \cite[Definition\hspace*{2pt}3.2.2]{blom:2001} correspond directly to the respective conditions in \cref{def:lambdahotg}, the difference between the condition (scope) there and \cref{lambdahotg:scope1} in \cref{def:lambdahotg} reflects the difference described in the previous item.
		\item[free variables:] Whereas the higher-order term graphs in \cite{blom:2001} cater for the presence of free variables, free variables have been excluded from the basic format of λ-ho-term-graphs.
	\end{description}
\end{remark}

In the following, let \m{i ∈ \set{0,1}} and let \m{\ahotg_1} and \m{\ahotg_2} be λ-ho-term-graphs over \sigTGi{i} with \m{\ahotg_k = \tuple{\V_k,\labi{k}{},\argsi{k}{},\r_k,\scopei{k}{}}} for \m{k ∈ \set{1,2}}.

\begin{definition}[homomorphism]\label{lambdahotg:homo}
	\par A \emph{homomorphism} from \m{\ahotg_1} to \m{\ahotg_2} is a morphism from the structure \m{\tuple{\V_1,\labi{1}{},\argsi{1}{},\r_1,\scopei{1}{}}} to the structure \m{\tuple{\V_2,\labi{2}{},\argsi{2}{},\r_2,\scopei{2}{}}}, i.e.\ a function \m{h : \V_1 → \V_2} such that \h is a homomorphism from \m{\altg_1} to \m{\altg_2}, the term graphs underlying \m{\ahotg_1} and \m{\ahotg_2} respectively, and additionally, for all \m{\v ∈ \vertsiof{1}{λ}} 
	\begin{equation}\label{lambdahotg:homo:scope}
		\bar{h}(\scopei{1}{\v}) = \scopei{2}{h(\v)}
	\end{equation}
	where \bar{h} is the homomorphic extension of \h to sets over \m{\V_1}, i.e.\ \m{\bar{h} : \powerset(\V_1) → \powerset(\V_2)}, \m{A ↦ \setcompr{h(a)}{a ∈ A}}.
	\par If there exists a homomorphism \h from \m{\ahotg_1} to \m{\ahotg_2}, then we write \m{\ahotg_1 \funbisim_{\h} \ahotg_2} or \m{\ahotg_2 \convfunbisim_{\h} \ahotg_1}, or, dropping \h as subscript, \m{\ahotg_1 \funbisim \ahotg_2} or \m{\ahotg_2 \convfunbisim \ahotg_1}.
\end{definition}

\begin{definition}[isomorphism]\label{def:lambdahotg:iso}
	An \emph{isomorphism} from \m{\ahotg_1} to \m{\ahotg_2} is a homomorphism from \m{\ahotg_1} to \m{\ahotg_2} that, as a function from \m{\V_1} to \m{\V_2}, is bijective.
	\par If there exists an isomorphism \m{i : \V_1 → \V_2} from \m{\ahotg_1} to \m{\ahotg_2}, then we say that \m{\ahotg_1} and \m{\ahotg_2} are \emph{isomorphic}, and write \m{\ahotg_1 \iso_i \ahotg_2} or simply \m{\ahotg_1 \iso \ahotg_2}. The property of existence of an isomorphism between two λ-ho-term-graphs forms an equivalence relation. We denote by \classlhotgsisoi{0} and \classlhotgsisoi{1} the isomorphism equivalence classes of λ-ho-term-graphs over \sigTGi{0} and \sigTGi{1}, respectively.
\end{definition}

\begin{definition}[bisimulation]\label{def:lambdahotg:bisim}
	A \emph{bisimulation} between \m{\ahotg_1} and \m{\ahotg_2} is a (λ-ho-term-graph-like) structure \m{\ahotg = \tuple{\abisim,\lab{},\args{},\r,\scope{}}} where \m{\tuple{\abisim,\lab{},\args{},\r} ∈ \tgsminover{\sig}}, \m{\abisim \subseteq \V_1\times\V_2} and \m{\r = \pair{\r_1}{\r_2}} such that \m{\ahotg_1 \convfunbisim_{\aproj_1} \ahotg \funbisim_{\aproj_2} \ahotg_2} where \m{\aproj_1} and \m{\aproj_2} are projection functions, defined, for \m{i ∈ \set{1,2}}, by \m{\aproj_i : \V_1\times\V_2 → \V_i}, \m{\pair{\v_1}{\v_2} ↦ \v_i}.
	If there a bisimulation \abisim exists between \m{\ahotg_1} and \m{\ahotg_2}, then we write \m{\ahotg_1 \bisim_\abisim \ahotg_2}, or just \m{\ahotg_1 \bisim \ahotg_2}.
\end{definition}

\section{Abstraction-prefix based λ-ho-term-graphs}\label{sec:lambdaaphotgs}

\begin{para}[overview]
	By an `abstraction-prefix based λ-higher-order-term-graph' we will mean a term-graph over \sigTGi{i} for \m{i ∈ \set{0,1}} that is endowed with a correct abstraction prefix function. Such a function \abspre{} maps abstraction vertices \w to words \abspre{\w} consisting of all those abstraction vertices that have \w in their \extscope, in the order of their nesting from outermost to innermost abstraction vertex. If \m{\abspre{\w} = \v_1 \dots \v_n}, then \m{\v_1,\dots,\v_n} are the abstraction vertices that have \w in their scope, with \m{\v_1} the outermost and \m{\v_n} the innermost such abstraction vertex.
\end{para}

\begin{para}[comparison to λ-ho-term-graphs]
	So the conceptual difference between the scope functions of λ-ho-term-graphs defined in the previous section, and abstraction-prefix functions of the λ-ap-ho-term-graphs defined below is the following: A scope function \scope{} associates with every abstraction vertex \v the information on its \extscope \scope{\v} and makes it available at \v. In contrast, an abstraction-prefix function \abspre{} gathers all the scoping information that is relevant to a vertex \v (in the sense that it contains all abstraction vertices in whose scope \v is) and makes it available at \v in the form \abspre{\v}. The fact that abstraction-prefix functions make relevant scope information locally available at all vertices leads to simpler correctness conditions.
\end{para}


\begin{definition}[abstraction-prefix function for \sigTGi{i}-term-graphs]\label{def:abspre:function:siglambdai}
	Let \m{i ∈ \set{0,1}} and \m{\atg = \tuple{\V,\lab{},\args{},\r}} be a \sigTGi{i}-term-graph. A function \m{\abspre{} : \V → \V^*} from vertices of \atg to words of vertices is called an \emph{abstraction-prefix function} for \atg. Such a function is called \emph{correct} if for all \m{v,w ∈ \V} and \m{k ∈ \set{0,1}}:
	\begin{align}
		\tag{root}\label{ap-function:root} & ~ \abspre{\r} = \emptyword \\
		\tag{λ}\label{ap-function:abs} v ∈ \vertsof{λ} \;∧\; v \tgsucc w ~ ⇒ & ~ \abspre{w} ≤ \abspre{v} v \\
		\tag{@}\label{ap-function:app} v ∈ \vertsof{@} \;∧\; v \tgsucc w ~ ⇒ & ~ \abspre{w} ≤ \abspre{v} \\
		\tag{\m{\0_0}}\label{ap-function:00}
			v ∈ \vertsof{\0} ~ ⇒ & ~ \abspre{v} ≠ \emptyword \\
		\tag{\m{\0_1}}\label{ap-function:01}
			v ∈ \vertsof{\0} \;∧\; v \tgsucc w ~ ⇒ & ~ w ∈ \vertsof{λ} \;∧\; \abspre{w} \prefixcon w = \abspre{v}
	\end{align}
	Analogous to \cref{def:lambdahotg}, if \m{i=0}, then \cref{ap-function:00} is trivially true and hence superfluous, and if \m{i=1}, then \cref{ap-function:01} is redundant, because it follows from \cref{ap-function:01}.
	\par We say that \atg\ \emph{admits} a correct abstraction-prefix function if such a function exists for \atg.
\end{definition}

\begin{definition}[λ-ap-ho-term-graph]\label{def:aplambdahotg}
	Let \m{i ∈ \set{0,1}}. A \emph{λ-ap-ho-term-graph} (short for \emph{abstraction-prefix based λ-higher-order-term-graph}) over signature \sigTGi{i} is a tuple \m{\ahotg = \tuple{\V,\lab{},\args{},\r,\abspre{}}} where \m{\atg_{\ahotg} = \tuple{\V,\lab{},\args{},\r}} is a \sigTGi{i}-term-graph, called the term graph \emph{underlying} \ahotg, and \abspre{} is a correct abstraction-prefix function for \m{\atg_{\ahotg}}. The classes of λ-ap-ho-term-graphs over \sigTGi{i} will be denoted by \classaphotgsi{i}.
\end{definition}

\begin{example}
	See \cref{fig:lambdaaphotg} for two λ-ap-ho-term-graphs, which correspond (see \cref{ex:corr:lhotgs:aphotgs}) to the λ-ho-term-graphs in \cref{fig:lambdahotg}.
\end{example}

\begin{figure}
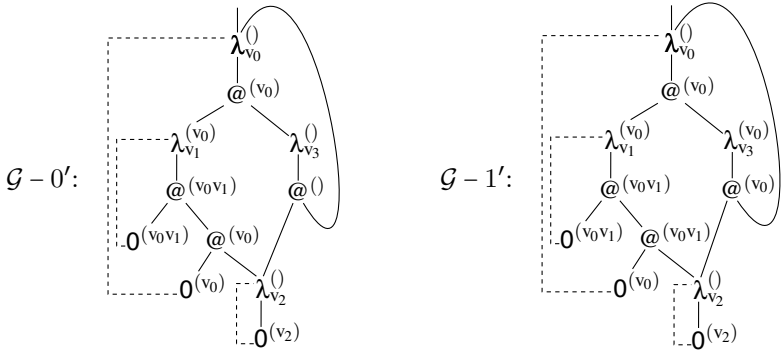

	\begin{hspread}
		\graph{\aaphotg-0'}{running-prefixed-ho-tg-eager} &
		\graph{\aaphotg-1'}{running-prefixed-ho-tg-non-eager}
	\end{hspread}
	\caption{The λ-ap-ho-term-graphs corresponding to the λ-ho-term-graphs in \cref{fig:lambdahotg}. The subscripts of abstraction vertices indicate their names. The super-scripts of vertices indicate their abstraction-prefixes.}
	\label{fig:lambdaaphotg}
\end{figure}

The following lemma states some basic properties of the scope function in λ-ap-ho-term-graphs.

\begin{lemma}\label{lem:aphotgs}
	Let \m{i ∈ \set{0,1}} and let \m{\aaphotg = \tuple{\V,\lab{},\args{},\r,\abspre{}}} be a λ-ap-ho-term-graph over \sigTGi{i}. Then the following statements hold:
	\begin{enumerate}[(i)]
		\item\label{lem:aphotg:i} Suppose that, for some \m{\v,\w ∈ \V}, \v occurs in \abspre{\w}. Then \m{\v ∈ \vertsof{λ}}, occurs in \abspre{\w} only once, and every access path of \w passes through \v, but does not end there, and thus \m{\w ≠ \v}. Furthermore it holds that \m{\abspre{\v} \prefixcon \v ≤ \abspre{\w}}. And conversely, if \m{\abspre{\w} = \apre \prefixcon \v \prefixcon \bpre} for some \m{\apre,\bpre ∈ \V^*}, then \m{\abspre{\v} = \apre}.
		\item\label{lem:aphotg:ii} Vertices in abstraction prefixes are abstraction vertices, and hence \abspre{} is of the form \m{\abspre{} : \V→(\vertsof{λ})^*}.
		\item\label{lem:aphotg:iii} For all \m{\v ∈ \vertsof{λ}} it holds that \m{\v∉\abspre{\v}}.
		\item\label{lem:aphotg:iv} While access paths might end in vertices in \vertsof{\0}, they pass only through vertices in \m{\vertsof{λ} ∪ \vertsof{@}}.
	\end{enumerate}
\end{lemma}

\begin{proof}
	Let \m{i ∈ \set{0,1}} and let \m{\aaphotg = \tuple{\V,\lab{},\args{},\r,\abspre{}}} be a λ-ap-ho-term-graph over \sigTGi{i}.
	\par For showing \cref{lem:aphotg:i}, let \m{\v,\w ∈ \V} be such that \v occurs in \abspre{\w}. Suppose further that \apath is an access path of \w. Note that when walking through \apath the abstraction prefix starts out empty (due to \cref{ap-function:root}), and is expanded only in steps from vertices \m{\v' ∈ \vertsof{λ}} (due to \cref{ap-function:abs}, \cref{ap-function:app}, and \cref{ap-function:01}) in which just \m{\v'} is added to the prefix on the right (due to \cref{ap-function:abs}). Since \v occurs in \abspre{\w}, it follows that \m{\v ∈ \vertsof{λ}}, that \v must be visited on \apath, and that \apath continues after the visit to \v. That \apath is an access path also entails that \v is not visited again on \apath, hence that \m{\w ≠ \v} and that \v occurs only once in \abspre{\w}, and that \m{\abspre{\v} \prefixcon \v}, the abstraction prefix of the successor vertex of \v on \apath, is a prefix of the abstraction prefix of every vertex that is visited on \apath after \v.
	\par Statements \cref{lem:aphotg:ii} and \cref{lem:aphotg:iii} follow directly from statement \cref{lem:aphotg:i}.
	\par For showing \cref{lem:aphotg:iv}, consider an access path \m{\apath : \r = \w_0 \tgsucc \dots \tgsucc \w_n} that leads to a vertex \m{\w_n ∈ \vertsof{\0}}. If \m{i=0}, then there is no path that extends \apath properly beyond \m{\w_n}. So suppose \m{i=1}, and let \m{\w_{n+1} ∈ \V} be such that \m{\w_n \tgsucc_0 \w_{n+1}}. Then \cref{ap-function:01} implies that \m{\abspre{\w_n} = \abspre{\w_{n+1}} \prefixcon \w_{n+1}}, from which it follows by \cref{lem:aphotg:i} that \m{\w_{n+1}} is visited already on \apath. Hence \apath does not extend to a longer path that is again an access path.
\end{proof}

In the following, let \m{i ∈ \set{0,1}} and let \m{\aaphotg_1} and \m{\aaphotg_2} be λ-ap-ho-term-graphs over \sigTGi{i} with \m{\aaphotg_k = \tuple{\V_k,\labi{k}{},\argsi{k}{},\r_k,\absprei{k}{}}} for \m{k ∈ \set{1,2}}.

\begin{definition}[homomorphism]\label{def:homom:aplambdahotg}
	A \emph{homomorphism} from \m{\aaphotg_1} to \m{\aaphotg_2} is a morphism from the structure \m{\tuple{\V_1,\labi{1}{},\argsi{1}{},\r_1,\absprei{1}{}}} to the structure \m{\tuple{\V_2,\labi{2}{},\argsi{2}{},\r_2,\absprei{2}{}}}, i.e.\ a function \m{h : \V_1 → \V_2} such that \h is a homomorphism from \m{\altg_1} to \m{\altg_2}, the term graphs underlying \m{\aaphotg_1} and \m{\aaphotg_2} respectively, and additionally, for all \m{v ∈ \V_1} \[h^*(\absprei{1}{v}) = \absprei{2}{h(v)}\] where \m{{h^*}} is the homomorphic extension of \h to words over \m{\V_1}.
	\par We then write \m{\aaphotg_1 \funbisim_{\h} \aaphotg_2}, or \m{\aaphotg_2 \convfunbisim_{\h} \aaphotg_1}. And we write \m{\aaphotg_1 \funbisim \aaphotg_2}, or for that matter \m{\aaphotg_2 \convfunbisim \aaphotg_1}, if there is a homomorphism from \m{\aaphotg_1} to \m{\aaphotg_2}.
\end{definition}

\begin{para}[isomorphism, bisimulation]
	Analogous to \cref{def:lambdahotg:iso} and \cref{def:lambdahotg:bisim}.
	We denote by \classaphotgsisoi{0} and \classaphotgsisoi{1} the isomorphism equivalence classes of λ-ho-term-graphs over \sigTGi{0} and \sigTGi{1}, respectively.
\end{para}



\begin{para}[relating λ-ho-term-graphs and λ-ap-ho-term-graphs]
	The following proposition defines mappings between λ-ho-term-graphs and λ-ap-ho-term-graphs by which we establish a bijective correspondence between the two classes. For both directions the underlying λ-term-graph remains unchanged. \lhotgstoaphotgsi{i}{} derives an abstraction-prefix function \abspre{} from a scope function by assigning to each vertex a word of its binders in the correct nesting order. \aphotgstolhotgsi{i}{} defines its scope function \scope{} by assigning to each λ-vertex \v the set of vertices that have \v in their prefix (along with \v since a vertex never has itself in its abstraction prefix).
\end{para}

\begin{proposition}[relating λ-ho-term-graphs and λ-ap-ho-term-graphs]\label{prop:mappings:lhotgs:aphotgs}
	For each \m{i ∈ \set{0,1}}, the mappings \lhotgstoaphotgsi{i}{} and \aphotgstolhotgsi{i}{} are well-defined between the class of λ-ho-term-graphs over \sigTGi{i} and the class of λ-ap-ho-term-graphs over \sigTGi{i}:
	\begin{align*}
		\lhotgstoaphotgsi{i}{} :~ & \classlhotgsi{i} → \classaphotgsi{i} \\
		& \ahotg = \tuple{\V,\lab{},\args{},\r,\scope{}} ↦ \lhotgstoaphotgsi{i}{\ahotg} := \tuple{\V,\lab{},\args{},\r,\abspre{}} \\
		& \hspace*{2ex}\text{where }
			\begin{aligned}
				&\abspre{} : \V→\V^*\\&v ↦ \v_0 \dots \v_n
			\end{aligned}
			\text{ ~with~ } \begin{aligned} & \binders{v} \setminus \set{v} = \set{\v_0,\dots,\v_n} \text{ and} \\& \scope{\v_n} \subset \scope{\v_{n-1}} \dots \subset \scope{\v_0} \end{aligned}
		\\[1em]
		\aphotgstolhotgsi{i}{} :~ & \classaphotgsi{i} → \classlhotgsi{i} \\
		& \ahotg = \tuple{\V,\lab{},\args{},\r,\abspre{}} ↦ \lhotgstoaphotgsi{i}{\ahotg} := \tuple{\V,\lab{},\args{},\r,\scope{}} \\
		& \hspace*{2ex}\text{where } \begin{aligned} \scope{} :~& \vertsof{λ}→\powerset(\V) \\& \v ↦ \setcompr{\w ∈ \V}{\v\text{ occurs in }\abspre{\w}} ∪ \set{\v} \end{aligned}
	\end{align*}
	On the isomorphism equivalence classes \classlhotgsisoi{i} of \classlhotgsi{i} , and \classaphotgsisoi{i} of \classaphotgsi{i}, the functions \lhotgstoaphotgsi{i}{} and \aphotgstolhotgsi{i}{} induce the functions \m{\slhotgsisotoaphotgsisoi{i} : \classlhotgsisoi{i} → \classaphotgsisoi{i}}, \m{\eqcl{\ahotg}{\iso} ↦ \eqcl{\lhotgstoaphotgsi{i}{\ahotg}}{\iso}} and \m{\saphotgsisotolhotgsisoi{i} : \classaphotgsisoi{i} → \classlhotgsisoi{i}}, \m{\eqcl{\aaphotg}{\iso} ↦ \eqcl{\aphotgstolhotgsi{i}{\aaphotg}}{\iso}}.
\end{proposition}

\begin{theorem}[correspondence between λ-ho-term-graphs and λ-ap-ho-term-graphs]\label{thm:corr:lhotgs:aphotgs}
	For each \m{i ∈ \set{0,1}} it holds that the mappings \lhotgstoaphotgsi{i}{} and \aphotgstolhotgsi{i}{} as defined in \cref{prop:mappings:lhotgs:aphotgs} are each other's inverse; thus they define a bijective correspondence between the class of λ-ho-term-graphs over \sigTGi{i} and the class of λ-ap-ho-term-graphs over \sigTGi{i}. Furthermore, they preserve and reflect the sharing orders on \classlhotgsi{i} and on \classaphotgsi{i}:
	\begin{align*}
		& ∀ \ahotg_1,\ahotg_2 ∈ \classlhotgsi{i} ~ &
		\ahotg_1 \funbisim \ahotg_2
		~ &\Longleftrightarrow ~
		\lhotgstoaphotgsi{i}{\ahotg_1} \funbisim \lhotgstoaphotgsi{i}{\ahotg_1}
		\\
		& ∀ \aaphotg_1,\aaphotg_2 ∈ \classaphotgsi{i} ~ &
		\aphotgstolhotgsi{i}{\aaphotg_1} \funbisim \aphotgstolhotgsi{i}{\aaphotg_1}
		~ &\Longleftrightarrow ~
		\ahotg_1 \funbisim \ahotg_2
	\end{align*}
\end{theorem}

\begin{example}\label{ex:corr:lhotgs:aphotgs}
The λ-ho-term-graphs in \cref{fig:lambdahotg} correspond to the λ-ap-ho-term-graphs in \cref{fig:lambdaaphotg}
via the mappings \lhotgstoaphotgsi{i}{} and \aphotgstolhotgsi{i}{} as follows:
	\begin{hspread}
		\m{\lhotgstoaphotgsi{i}{\ahotg_0} = \aaphotg_0'}
		&
		\m{\lhotgstoaphotgsi{i}{\ahotg_1} = \aaphotg_1'}
		&
		\m{\aphotgstolhotgsi{i}{\aaphotg_0} = \ahotg_0'}
		&
		\m{\lhotgstoaphotgsi{i}{\aaphotg_1} = \ahotg_1'}
	\end{hspread}
\end{example}

\begin{para}[higher-order term graphs]
	Due to this correspondence will henceforth sometimes say `higher-order term graph' to refer to either λ-ho-term-graphs or λ-ap-ho-term-graphs.
\end{para}

\section{λ-Term-Graphs without Scope Delimitiers}

\begin{para}[Overview]
	In this section we examine (and dismiss) a naive approach to implement of functional bisimulation on higher-order term graphs as functional bisimulation on their underlying term graph, hoping that this application extends to the higher-order term graph without further ado. We demonstrate that this approach fails, concluding that a faithful first-order implementation of functional bisimulation must not neglect the scoping information that the higher-order term graphs carry.
\end{para}

\begin{para}[variable backlinks]
	For higher-order term graphs over the signature \sigTGi{0} (i.e.\ without variable backlinks) essential binding information is lost when looking only at the underlying term graph, to the extent that λ-terms cannot be unambiguously represented anymore. For instance the higher-order term graphs that represent the λ-terms \abs{xy}{\app{x}{y}} and \abs{xy}{\app{x}{x}} have the same underlying term graph.
	\par This is different for higher-order term graphs over \sigTGi{1}, because the abstraction vertex to which a variable vertex belongs is uniquely identified by the variable backlink. This is why in this section we only consider term graphs and higher-order term graphs over \sigTGi{1}.
\end{para}

\begin{definition}[λ-term-graphs over \sigTGi{1}]\label{def:classltgs}
	A term graph over \sigTGi{1} is called a \emph{λ-term-graph} over \sigTGi{1} if it admits a correct abstraction-prefix function (or equivalently a correct scope function). By \classltgsi{1} we denote the class of λ-term-graphs over \sigTGi{1}.
\end{definition}

\begin{definition}[functional bisimulation on the underlying term graph]
	Let \ahotg be a higher-order term graph over \sigTGi{i} for \m{i ∈ \set{0,1}} with underlying term graph \altg. And suppose that there is a homomorphism \h from \altg to \m{\atg'} for some term graph \m{\atg'} over \sigTGi{i}.
	\par We say that \h\ \emph{extends to a homomorphism on} \ahotg if there is a higher-order term graph \m{\ahotg'} over \sigTGi{i} which has \m{\altg'} as its underlying term graph and \h is a homomorphism from \ahotg to \m{\ahotg'}.
	\par We say that a class \aclass of higher-order term graphs is \emph{closed under functional bisimulation on the underlying term graphs} if for every \m{\ahotg ∈ \aclass} with underlying term graph \altg, every homomorphism \h on \altg extends to a homomorphism on \ahotg.
\end{definition}

\begin{proposition}\label{prop:no:extension:funbisim:siglambda1}
	Neither the class \classlhotgsi{1} of λ-ho-term-graphs nor the class \classaphotgsi{1} of λ-ap-ho-term-graphs is closed under functional bisimulation on the underlying term graphs.
\end{proposition}

\begin{proof}
	In view of \cref{thm:corr:lhotgs:aphotgs} it suffices to show the statement for one of the two classes, say \classlhotgsi{1}. We show by example that not every homomorphism on the term graph underlying a λ-ho-term-graph over \sigTGi{1} extends to a homomorphism on \classlhotgsi{1}.
	\par Consider the λ-ho-term-graphs \m{\ahotg_1} and \m{\ahotg_0} in \cref{fig:siglambda1-not-closed} and their respective underlying term graphs \m{\atg_1} and \m{\atg_0}.
	\begin{figure}
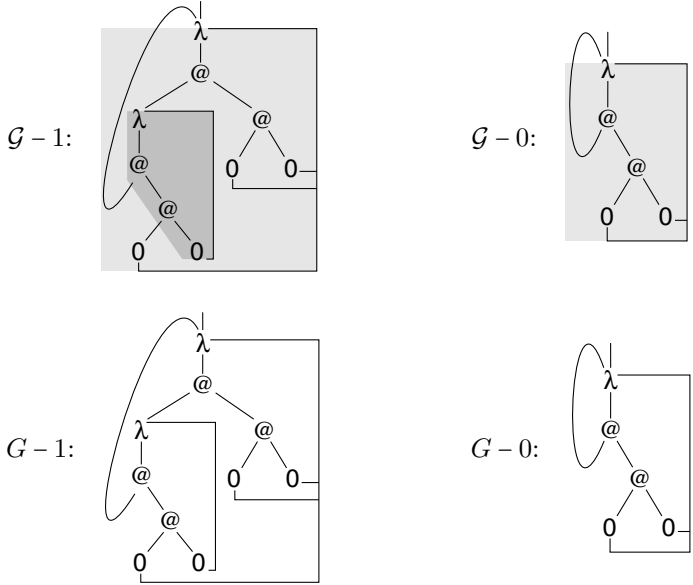

		\begin{hspread}
			\graph{\ahotg-1}{no-extension-funbisim-siglambda1-g1}
			\hsep
			\graph{\ahotg-0}{no-extension-funbisim-siglambda1-g0}
			\\
			\graph{\atg-1}{no-extension-funbisim-siglambda1-g1-tg}
			\hsep
			\graph{\atg-0}{no-extension-funbisim-siglambda1-g0-tg}
		\end{hspread}
	\caption{There is no homomorphism between the λ-ho-term-graphs \m{\ahotg_1} and \m{\ahotg_0} because of the scoping condition in \cref{lambdahotg:homo}~\cref{lambdahotg:homo:scope}. There is a homomorphism between their underlying term graphs \m{\altg_1} and \m{\altg_0} but it does not extend to \m{\ahotg_1}.}
		\label{fig:siglambda1-not-closed}
	\end{figure}
	There is an homomorphism \h from \m{\atg_1} to \m{\atg_0}. However, \h does not extend to a homomorphism on \m{\ahotg_1}, since there is no homomorphism from \m{\ahotg_1} to \m{\ahotg_0}, and \m{\ahotg_0} is the only λ-ho-term-graph with \m{\atg_0} as its underlying term graph (\m{\atg_0} admits only one scope function).
\end{proof}

The next proposition is merely a reformulation of \cref{prop:no:extension:funbisim:siglambda1}.

\begin{proposition}\label{prop:forgetful}
	\newfunction\underlying{UL}
	The scope-forgetful mapping \underlying{}{}, that maps higher-order term graphs to their underlying term graphs, preserves but does not reflect the sharing orders on the classes \classaphotgsi{1} and \classlhotgsi{1}. In particular for \m{\aclass ∈ \set{\classaphotgsi{1}, \classlhotgsi{1}}} and \m{\underlying{}{} : \aclass → \classtgssiglambdai{1}} it holds:
	\begin{align*}
		∀ \aaphotg_1,\aaphotg_2 ∈ \aclass ~ &&
		\aaphotg_1 \funbisim \aaphotg_2
		& ~⇒~
		\underlying{\aaphotg_1} \funbisim \underlying{\aaphotg_2}
		\\
		\neg \forall \aaphotg_1,\aaphotg_2 ∈ \aclass ~ &&
		\underlying{\aaphotg_1} \funbisim \underlying{\aaphotg_2}
		& ~⇒~
		\aaphotg_1 \funbisim \aaphotg_2
	\end{align*}
\end{proposition}

\begin{para}[conclusion]\label{term-graph-without-scope-delimiters-conclusion}
	As a consequence of this proposition it is not possible to faithfully implement functional bisimulation on higher-order term graphs by only considering their underlying term graphs. It seems that we cannot simply discard the scoping information that is present in the higher-order term graphs. In the next section we refine the above approach by relying on a class of first-order term graphs that accounts for scoping by means of scope delimiter vertices.
\end{para}

\section{λ-Term-Graphs with Scope Delimiters}\label{sec:ltgs}

\begin{para}[overview]
	Considering \cref{term-graph-without-scope-delimiters-conclusion} we now enrich the signature for λ-term-graphs by an scope-delimiter symbol \S (cf.\ \S-steps of the decomposition systems). Delimiter vertices signify the end of a scope. We define variations of that signature where scope-delimiter vertices and/or variable vertices can have backlinks.
\end{para}

\begin{para}[signature for λ-term-graphs with scope delimiters]
	For all \m{i ∈ \set{0,1}} and \m{j ∈ \set{1,2}} we define the signature \m{\sigTGij{i}{j} := \sigTG ∪ \set{\0, \S}} as an extension of the signature \sigTG where \m{\arity{\0} = i} and \m{\arity{\S} = j}, and we denote by \m{\classtgssiglambdaij{i}{j} := \tgsover{\sigTGij{i}{j}}} the class of term graphs over signature \sigTGij{i}{j}.
\end{para}

\begin{para}[variable and scope-delimiter backlinks]
	Analogous to the classes \classlhotgsi{i} and \classaphotgsi{i}, the index \i determines whether variable vertices have backlinks to their corresponding abstraction. Here, the additional index \j determines whether scope-delimiter vertices have such backlinks (if \m{j=2}) or not (if \m{j=1}).
\end{para}

\begin{para}[relating scope delimiters and scope]
	As the scope-delimiter vertices are meant to reproduce the scoping information in the higher-order term graphs, in order to relate the placements of \S-vertices and the scoping information we formulate abstraction-prefix conditions for λ-term-graphs (even though they do not carry abstraction prefixes). These conditions are based on \cref{def:aplambdahotg} and adapted to account for \S-vertices which are to decrease the abstraction prefix by one variable.
\end{para}

\begin{definition}[abstraction-prefix function for \sigTGij{i}{j}-term-graphs]\label{ltg:ap-function}
	Let \m{\atg = \tuple{\V,\lab{},\args{},\r}} be a \sigTGij{i}{j}-term-graph for an \m{i ∈ \set{0,1}} and an \m{j ∈ \set{1,2}}. A function \m{\abspre{} : \V → \V^*} from vertices of \atg to words of vertices is called an \emph{abstraction-prefix function} for \atg. Such a function is called \emph{correct} if for all \m{v,w ∈ \V} and \m{k ∈ \set{0,1}} the following holds:
	\newcommand\Implies{&⇒~~~&}
	\begin{align}
		\tag{root}\label{ltg:ap-function:root} &&& \abspre{\r} = \emptyword \\
		\tag{λ}\label{ltg:ap-function:abs} v ∈ \vertsof{λ} &~∧~ v \tgsucc_0 w \Implies \abspre{w} = \abspre{v} \prefixcon v \\
		\tag{@}\label{ltg:ap-function:app} v ∈ \vertsof{@} &~∧~ v \tgsucc_k w \Implies \abspre{w} = \abspre{v} \\
		\tag{\m{\0_0}}\label{ltg:ap-function:00} v ∈ \vertsof{\0} & \Implies \abspre{v} ≠ \emptyword \\
		\tag{\m{\0_1}}\label{ltg:ap-function:01} v ∈ \vertsof{\0} &~∧~ v \tgsucc_0 w \Implies w ∈ \vertsof{λ} ~∧~ \abspre{w} \prefixcon w = \abspre{v} \\
		\tag{\m{\S_0}}\label{ltg:ap-function:S0} v ∈ \vertsof{\S} &~∧~ v \tgsucc_0 w \Implies \abspre{w} \prefixcon u = \abspre{v} ~ \text{for some \m{u ∈ \V}}\\
		\tag{\m{\S_1}}\label{ltg:ap-function:S1} v ∈ \vertsof{\S} &~∧~ v \tgsucc_1 w \Implies w ∈ \vertsof{λ} ~∧~ \abspre{w} \prefixcon w = \abspre{v}
	\end{align}
	Note that analogous to \cref{def:lambdahotg} and \cref{def:aplambdahotg}, if \m{i=0}, then \cref{ltg:ap-function:01} is trivially true, and if \m{i=1} then \cref{ltg:ap-function:00} is redundant, because it follows from \cref{ltg:ap-function:01}. Additionally, if \m{j=1}, then \cref{ltg:ap-function:S1} is trivially true.
\end{definition}

\begin{definition}[{λ-term-graphs over \sigTGij{i}{j}}]\label{def:lambdatg:siglambdaij}
	Let \m{i ∈ \set{0,1}} and \m{j ∈ \set{1,2}}. A \emph{λ-term-graph} over \sigTGij{i}{j} is a \sigTGij{i}{j}-term-graph that admits a correct abstraction-prefix function. The class of λ-term-graphs over \sigTGij{i}{j} is denoted by \classltgsij{i}{j}, and the class of isomorphism equivalence classes of λ-term-graphs over \sigTGij{i}{j} by \classltgsisoij{i}{j}.
\end{definition}

\begin{example}
	See \cref{fig:lambdatg} for examples, that, as we will see, correspond to the ho-term-graphs in \cref{fig:lambdahotg} and in \cref{fig:lambdaaphotg}. A precise formulation of this correspondence is given in \cref{ex:corr:aphotgs:ltgs}.
\end{example}

\begin{figure}
	\begin{hspread}
		\graph{\atg-0}{running-snodes-eager}
		\hsep
		\graph{\atg-1}{running-snodes-non-eager}
	\end{hspread}
	\caption{The λ-term-graphs corresponding to the λ-ap-ho-term-graphs from \cref{fig:lambdaaphotg} and the λ-ho-term-graphs from \cref{fig:lambdahotg}.}
	\label{fig:lambdatg}
\end{figure}

The following lemma states some basic properties of λ-term-graphs.

\begin{lemma}\label{lem:ltg}
	Let \m{i ∈ \set{0,1}}, and \m{j ∈ \set{1,2}}. Let \m{\altg = \tuple{\V,\lab{},\args{},\r}} be a λ-term-graph over \sigTGij{i}{j}, and let \abspre{} be a correct abstraction-prefix function for \altg. Then the statements \cref{lem:aphotg:i}--\cref{lem:aphotg:iii} from \cref{lem:aphotgs} hold as items (i)--(iii) of this lemma, and additionally:
	\begin{enumerate}[(i)]\setcounter{enumi}{3}
		\item\label{lem:ltg:iv} When stepping through \altg along a path, \abspre{} behaves like a stack data structure, in the sense that in a step \m{v \tgsucc w} it holds:
		\begin{itemize}
			\item if \m{v ∈ \vertsof{λ}}, then \m{\abspre{w} = \abspre{v} \prefixcon v}, that is, \abspre{w} is obtained by adding \v to \abspre{v} on the right (\v is `pushed on the stack');
			\item if \m{v ∈ \vertsof{\0} ∪ \vertsof{\S}}, then \m{\abspre{w} \prefixcon \v = \abspre{v}} for some \m{\v ∈ \vertsof{λ}}, that is, \abspre{w} is obtained by removing the rightmost vertex \v from \abspre{v} (\v is `popped from the stack');
			\item if \m{v ∈ \vertsof{@}}, then \m{\abspre{w} = \abspre{v}}, i.e.\ the abstraction prefix remains the same.
		\end{itemize}
		\item\label{lem:ltg:v} Access paths may end in vertices in \vertsof{\0}, but only pass through vertices in \m{\vertsof{λ} ∪ \vertsof{@} ∪ \vertsof{\S}}, and depart from vertices in \vertsof{\S} only via indexed edges \tgsuccis{0}{\S}.
		\item\label{lem:ltg:vi} There is precisely one correct abstraction-prefix function for \altg.
	\end{enumerate}
\end{lemma}

\begin{proof}
	That also here \cref{lem:aphotg:i}--\cref{lem:aphotg:iii} from \cref{lem:aphotgs} hold, and that \cref{lem:ltg:v} holds, can be shown analogous to the proof of the respective items of \cref{lem:aphotgs}. Statement \cref{lem:aphotg:iv} is easy to check from the definition (see \cref{ltg:ap-function} of a correct abstraction-prefix function for a λ-term-graph). For \cref{lem:ltg:vi} it suffices to observe that if \abspre{} is a correct abstraction-prefix function for \altg, then, for all \m{v ∈ \V}, the value \abspre{v} of \abspre{} at \v can be computed by choosing an arbitrary access path \apath from \r to \v and using the conditions \cref{ltg:ap-function:abs}, \cref{ltg:ap-function:app}, and \cref{ltg:ap-function:S0} to determine in a stepwise manner the values of \abspre{} at the vertices that are visited on \apath. Hereby note that in every transition along an edge on \apath the length of the abstraction prefix only changes by at most 1.
\end{proof}

\begin{terminology}
	\Cref{lem:ltg}~\cref{lem:ltg:vi} allows us to speak of \emph{the} abstraction-prefix function of a λ-term-graph.
\end{terminology}

\begin{para}[λ-term-graphs are first-order]
	Although the requirement of the existence of a correct abstraction-prefix function restricts their possible forms, λ-term-graphs are first-order term graphs, and as such the definitions of homomorphism, isomorphism, and bisimulation for first-order term graphs from \cref{sec:representations:prelims} apply to them.
\end{para}

\begin{para}
	The question arises now how a homomorphism \h between λ-term-graphs \m{\altg_1} and \m{\altg_2} relates their abstraction-prefix functions \absprei{1}{} and \absprei{2}{}. As it turns out (\cref{prop:hom:image:absprefix:function:ltgs}) \absprei{2}{} is the `homomorphic image' of \absprei{1}{} under \h in the sense of the following definition.
\end{para}

\begin{definition}[homomorphic image]\label{def:hom:image:absprefix:function:ltgs}
	Let \m{\atg_1 = \tuple{\V_1,\labi{1}{},\argsi{1}{},\r_1}} and \m{\atg_2 = \tuple{\V_2,\labi{2}{},\argsi{2}{},\r_2}} be term graphs over \sigTGij{i}{j} for some \m{i ∈ \set{0,1}} and \m{j ∈ \set{1,2}}, and let \absprei{1}{} and \absprei{2}{} be abstraction-prefix functions (not necessarily correct) for \m{\atg_1} and \m{\atg_2}, respectively. Furthermore, let \m{h : \V_1 → \V_2} be a homomorphism from \m{\altg_1} to \m{\altg_2}.
	We say that \absprei{2}{} is the \emph{homomorphic image of} \absprei{1}{} \emph{under} \h if it holds:
	\begin{equation}\label{eq:def:hom:image:absprefix:function:ltgs}
		{h^*} ∘ {\absprei{1}{}} = {\absprei{2}{}} ∘ {\h}
	\end{equation}
	where \m{h^*} is the homomorphic extension of \h to words over \m{\V_1}.
\end{definition}

\begin{para}[homomorphisms between λ-ap-ho-term-graphs]
	Note that for λ-ap-ho-term-graphs \cref{eq:def:hom:image:absprefix:function:ltgs} (or rather an equivalent thereof) holds by definition (see \cref{def:homom:aplambdahotg}). This is a consequence of the fact that homomorphisms between λ-ap-ho-term-graphs are defined as morphisms between λ-ap-ho-term-graphs when viewed as algebraical structures. As the abstraction-prefix function is a part of λ-ap-ho-term-graphs it has to be respected by morphisms.
\end{para}

\begin{para}[homomorphisms between λ-term-graphs]
	λ-term-graphs however, do not include an abstraction-prefix function as part of their formalisation. Here the existence of an abstraction-prefix function is merely a mathematical property that distinguishes them from among the term graphs over the same signature. Homomorphisms between λ-term-graphs are defined as morphisms between the structures that underlie their formalisation, i.e.\ first-order term graphs, and therefore do not by definition respect abstraction-prefix functions.
\end{para}

However, it turns out that, that due the correctness conditions for λ-term-graphs, homomorphisms between λ-term-graphs do in fact respect their abstraction-prefix function:

\begin{proposition}\label{prop:hom:image:absprefix:function:ltgs}
	Let \m{\atg_1} and \m{\atg_2} be λ-term-graphs, and let \absprei{1}{} and \absprei{2}{} be their abstraction-prefix functions, respectively. Suppose that \h is a homomorphism from \m{\atg_1} to \m{\atg_2}. Then \absprei{2}{} is the homomorphic image of \absprei{1}{} under \h.
\end{proposition}

This proposition follows from statement \cref{lem:hom:image:absprefix:function:ltgs:ii} of the following lemma. \Cref{lem:hom:image:absprefix:function:ltgs:i} states that functional bisimulation on term graphs over \sigTGij{i}{j} preserves and reflects correctness of the abstraction-prefix functions.

\begin{lemma}\label{lem:hom:image:absprefix:function:ltgs}
	Let \m{\atg_1} and \m{\atg_2} be term graphs over \sigTGij{i}{j} for some \m{i ∈ \set{0,1}} and \m{j ∈ \set{1,2}}. Let \h be a homomorphism from \m{\atg_1} to \m{\atg_2}. Furthermore, let \absprei{1}{} and \absprei{2}{} be their abstraction-prefix functions, respectively. Then the following two statements hold:
	\begin{enumerate}[(i)]
		\item\label{lem:hom:image:absprefix:function:ltgs:ii} If \absprei{1}{} and \absprei{2}{} are correct for \m{\atg_1} and \m{\atg_2}, respectively, then \absprei{2}{} is the homomorphic image of \absprei{1}{}.
		\item\label{lem:hom:image:absprefix:function:ltgs:i} If \absprei{2}{} is the homomorphic image of \absprei{1}{}, then \absprei{1}{} is correct for \m{\atg_1} if and only if \absprei{2}{} is correct for \m{\atg_2}.
	\end{enumerate}
\end{lemma}

\begin{proof}
	Let \m{\atg_1 = \tuple{\V_1,\labi{1}{},\argsi{1}{},\r_1}} and \m{\atg_2 = \tuple{\V_2,\labi{2}{},\argsi{2}{},\r_2}}.
	In this proof we denote by \m{\tgsucc'} the directed-edge relation in \m{\atg_2}, thus if we for instance write \m{v \tgsucc'_0 w} we mean to say that \w is the first child of \v in \m{\atg_2}.
	\par For proving \cref{lem:hom:image:absprefix:function:ltgs:ii}, we assume that the abstraction-prefix functions \absprei{1}{} and \absprei{2}{} are correct for \m{\atg_1} and \m{\atg_2}, respectively. We establish that \absprei{2}{} is the homomorphic image under \h of \absprei{1}{} by showing that
	\begin{equation}\label{eq2:prf:lem:hom:image:absprefix:function:ltgs}
		∀ v ∈ \V_1~~ h^*(\absprei{1}{v}) = \absprei{2}{h(v)}
	\end{equation}
	by induction on the length of an access path \m{\apath : \r_1 = v_0 \tgsucc v_1 \tgsucc \dots \tgsucc v_n = v} of \v in \m{\atg_1}.
	\par If \m{\length{\apath} = 0}, then \m{v = \r_1}. It follows from the correctness condition \cref{ltg:ap-function:root} for abstraction-prefix functions and the condition \cref{ltgs:homomorphism:cond:roots} for homomorphisms that \m{h^*(\absprei{1}{v}) = h^*(\absprei{1}{\r_1}) = h^*(\emptyword) = \emptyword = \absprei{2}{\r_2} = \absprei{2}{h(\r_1)}}.
	\par If \m{\length{\apath} = n+1}, then \apath is of the form \m{\apath : \r_1 = v_0 \tgsucc v_1 \tgsucc \dots \tgsucc v_n \tgsucc_i v_{n+1} = v} for some \m{i ∈ \set{0,1}}. We have to show that \m{h^*(\absprei{1}{v}) = \absprei{2}{h(v)}} holds, with \m{h^*(\absprei{1}{v_n}) = \absprei{2}{h(v_n)}} as an induction hypothesis. We will do so by distinguishing the three possible labels \m{v_n} can have, namely λ, @, and \S (see \cref{lem:ltg}~\cref{lem:ltg:v}) and by applying the correctness conditions from \cref{ltg:ap-function}.
	\begin{description}
		\item[\m{v_n ∈ \vertsiof{1}{λ}:}] Since \h is a homomorphism also \m{h(v_n) ∈ \vertsiof{2}{λ}} holds. Applying the correctness condition \cref{ltg:ap-function:abs} to \m{v_n} and \m{h(v_n)} yields that \m{\absprei{1}{v} = \absprei{1}{v_n} \prefixcon v_n} and \m{\absprei{2}{h(v)} = \absprei{2}{h(v_n)} \prefixcon h(v_n)}. From this we now obtain \m{h^*(\absprei{1}{v}) = h^*(\absprei{1}{v_n} \prefixcon v_n) = h^*(\absprei{1}{v_n}) \prefixcon h(v_n) = \absprei{2}{h^*(v_n)} \prefixcon h(v_n) = \absprei{2}{h(v)}} by using the induction hypothesis.
		\item[\m{v_n ∈ \vertsiof{1}{@}:}] Since \h is a homomorphism also \m{h(v_n) ∈ \vertsiof{2}{@}} holds. Applying the correctness condition \cref{ltg:ap-function:app} to \m{v_n} and \m{h(v_n)} yields that \m{\absprei{1}{v} = \absprei{1}{v_n}} and \m{\absprei{2}{h(v)} = \absprei{2}{h(v_n)}}. Then we obtain \m{h^*(\absprei{1}{v}) = h^*(\absprei{1}{v_n}) = \absprei{2}{h(v_n)} = \absprei{2}{h(v)}} by using the induction hypothesis.
		\item[\m{v_n ∈ \vertsiof{1}{\S}:}] Since \h is a homomorphism also \m{h(v_n) ∈ \vertsiof{2}{\S}} holds. We distinguish two cases for the last step of \apath (is it via the backlink or not?):
			\begin{description}
				\item[\m{v_n \tgsucc_1 v:}] This implies \m{h(v_n) \tgsucc'_1 h(v)}. Applying the correctness condition \cref{ltg:ap-function:S1} to \m{v_n} and to \m{h(v_n)} yields \m{\absprei{1}{v_n} = \absprei{1}{v} \prefixcon v} and \m{\absprei{2}{h(v)} \prefixcon h(v) = \absprei{2}{h(v_n)}}. Applying the induction hypothesis we get \m{\absprei{2}{h(v)} \prefixcon h(v) = \absprei{2}{h(v_n)} = h^*(\absprei{1}{v_n}) = h^*(\absprei{1}{v} \prefixcon v) = h^*(\absprei{1}{v}) \prefixcon h(v)}. From this we conclude \m{\absprei{2}{h(v)} = h^*(\absprei{1}{v})}.
				\item[\m{v_n \tgsucc_0 v:}] Analogously, but we rely on \cref{ltg:ap-function:S0} instead of \cref{ltg:ap-function:S1}.
			\end{description}
	\end{description}
	\par For showing statement \cref{lem:hom:image:absprefix:function:ltgs:i}, we assume that -- as we have just shown -- \absprei{2}{} is the homomorphic image of \absprei{1}{}.
	\par For the direction ``⇒'' of the equivalence in \cref{lem:hom:image:absprefix:function:ltgs:ii} we assume that \absprei{1}{} is correct for \m{\atg_1}, and we show that \absprei{2}{} is correct for \m{\atg_2}, according to the conditions \cref{ltg:ap-function:root}, \cref{ltg:ap-function:abs}, \cref{ltg:ap-function:app}, \cref{ltg:ap-function:00}, and \cref{ltg:ap-function:01} from \cref{ltg:ap-function}.
	\par For the two conditions that do not involve transitions this is easy to show: The condition \cref{ltg:ap-function:root} for \m{\atg_2} follows from the condition \cref{ltgs:homomorphism:cond:roots} from \cref{def:ltgs:homomorphism} and \cref{eq:def:hom:image:absprefix:function:ltgs}. And the condition \cref{ltg:ap-function:00} follows similarly by using that every pre-image under \h of a variable vertex in \m{\atg'} must be a variable vertex in \altg since \h is a homomorphism.
	\par Let us only look at one of the remaining cases, and look at the condition \cref{ltg:ap-function:S0}. For this, let \m{v',v'_0 ∈ \V_2} such that \m{v' ∈ \vertsiof{2}{\S}} with \m{v' \tgsucc'_0 v'_0}. Then as a consequence of the \cref{eq:args-forward} condition for \h from \cref{prop:hom:tgs} there exist \m{v,v_0 ∈ \V_1} with \m{v ∈ \vertsiof{1}{\S}}, \m{v \tgsucc_0 v_0}, and with \m{h(v) = v'}, \m{h(v_0) = v'_0}. Since \cref{ltg:ap-function:S0} is satisfied for \m{\atg_1}, there exists a vertex \m{w ∈ \V_1} such that \m{\absprei{1}{v_0} w = \absprei{1}{v}}. From this and by \cref{eq:def:hom:image:absprefix:function:ltgs} we obtain \m{\absprei{2}{v'_0} \prefixcon h(w) = \absprei{2}{h(v_0)} \prefixcon h(w) = h^*(\absprei{1}{v_0}) \prefixcon h(w) = h^*(\absprei{1}{v_0} \prefixcon w) = h^*(\absprei{1}{v}) = \absprei{2}{h(v)} = \absprei{2}{v'}}. This shows the existence of \m{w' ∈ \V_2} (to wit \m{w' := h(w)}) with \m{\absprei{2}{v'_0} \prefixcon w' = \absprei{2}{v'}}. In this way we have established the correctness condition \cref{ltg:ap-function:S0} for \m{\atg_2}.
	\par The direction ``⇐'' of the equivalence in \cref{lem:hom:image:absprefix:function:ltgs:ii} can be established analogously by recognizing that the correctness conditions from \cref{ltg:ap-function} carry over also from \m{\atg_2} to \m{\atg_1} via \h due to the homomorphic image property \cref{eq:def:hom:image:absprefix:function:ltgs}. The arguments for the individual correctness conditions are analogous to the ones used above, but they depend on using the \cref{eq:args-backward} property from \cref{prop:hom:tgs} of \h.
\end{proof}

\begin{figure}
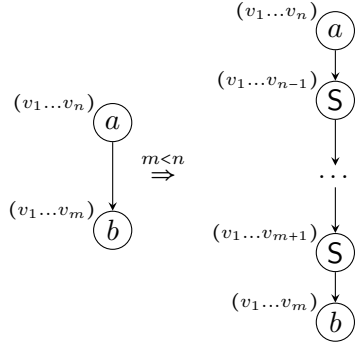

	\transpicture{
		\ltgnode{a}{\a};
		\addPrefix{a}{v_1 \dots v_n};
		\ltgnode[node distance=9mm,below=of a]{b}{\m{b}};
		\addPrefix{b}{v_1 \dots v_m};
		\draw[->](a) to (b);
	}
	\vcentered{\m{\overset{m<n}{⇒}}}
	\transpicture{
		\ltgnode{a}{\a};
		\addPrefix{a}{v_1 \dots v_n};
		\ltgnode[below=of a]{s1}{\S}; \addPrefix{s1}{v_1 \dots v_{n-1}}; \draw[->](a) to (s1);
		\node[node distance=6mm,below=of s1](si){\dots}; \draw[->](s1) to (si);
		\ltgnode[node distance=6mm,below=of si]{sn}{\S}; \addPrefix{sn}{v_1 \dots v_{m+1}}; \draw[->](si) to (sn);
		\ltgnode[below=of sn]{b}{\m{b}};
		\addPrefix{b}{v_1 \dots v_m};
		\draw[->](sn) to (b);
	}
	\caption{Definition of \aphotgstoltgsij{i}{j}{} by inserting \S-vertices, between edge-connected vertices of a λ-ap-ho-term-graph}
	\label{fig:insertS}
\end{figure}

\begin{para}[relating λ-term-graphs defined above and λ-ap-ho-term-graphs]
	We now proceed to define a precise relationship between λ-term-graphs and λ-ap-ho-term-graphs via mappings that translate between these classes:
	\begin{description}
		\item[The mapping \aphotgstoltgsij{i}{j}{}] produces a λ-term-graph for any given λ-ap-ho-term-graph by adding to the original set of vertices a number of delimiter vertices at the appropriate places. That is, at every position where the abstraction prefix decreases by \n elements, \n \S-vertices are inserted as depicted in \cref{fig:insertS}. In the image, the original abstraction prefix is retained as part of the vertices, which we use for defining the edges of the image.
		\item[The mapping \ltgstoaphotgsij{i}{j}{}] back to λ-ap-ho-term-graphs is simpler because it only has to erase the \S-vertices, and add the correct abstraction prefix that we already know to exists for the λ-term-graph.
	\end{description}
\end{para}

\begin{proposition}[from λ-ap-ho-term-graphs to λ-term-graphs]\label{prop:mappings:aphotgs:to:ltgs}
	\newcommand\snodels{\m{\#\mathrm{del}}}
	\newcommand\nodels[2]{\m{\snodels(#1, #2)}}
	Let \m{i ∈ \set{0,1}} and \m{j ∈ \set{1,2}}. The mapping \aphotgstoltgsij{i}{j}{} defined below is well-defined between the class of λ-term-graphs over \sigTGij{i}{j} and the class of λ-ap-ho-term-graphs over \sigTGi{i}:
	\[
		\aphotgstoltgsij{i}{j}{} : \classaphotgsi{i} → \classltgsij{i}{j}, ~
		\tuple{V,\lab{},\args{},r,\abspre{}} ↦ \tuple{V',\lab{}',\args{}',r'}
	\]
	where:
	\begin{align*}
		V' :=\:
		& \setcompr{\pair{v}{\abspre{v}}}{v ∈ V}\\
		∪\: & \setcompr{\tuple{v,k,v',\apre}}{v,v' ∈ V, v \tgsucc_k v',
		\begin{aligned}[c]
		& \vertsof{v} = λ ∧ \abspre{v'} < \apre ≤ \abspre{v} \prefixcon v \\
		& ∨\: \vertsof{v} = @ ∧ \abspre{v'} < \apre ≤ \abspre{v}
		\end{aligned}}
	\end{align*}
	\begin{align*}
		\newfunction\labp{lab'}
		r' := \pair{r}{\emptyword}
		\hspace*{5ex}
		\begin{gathered}[t]
		\labp{} : V' → \sigTGij{i}{j}, \;
		\begin{aligned}[c]
			\pair{v}{\abspre{v}} & ↦ \lab{v} \\
			\tuple{v,k,v',\apre} & ↦ \S
		\end{aligned}
		\end{gathered}
	\end{align*}
	and \m{\args{}' : V' → (V')^*} is defined such that for the induced indexed successor relation \m{\tgsucc'_{\smash{(\cdot)}}} it holds:
	\begin{align*}
		& v \tgsucc_k w \:∧\: \nodels{v}{k} = 0 ~\implies~ \pair{v}{\abspre{v}} \tgsucc'_k \pair{w}{\abspre{w}} \\
		& v \tgsucc_0 w \:∧\: \nodels{v}{0} > 0 \:∧\: \lab{v} = λ \:∧\: \abspre{v} = \abspre{w} \prefixcon u \prefixcon \apre \\
		&\indent\implies~ \pair{v}{\abspre{v}} \tgsucc'_0 \tuple{v,0,w,\abspre{v} \prefixcon v} \:∧\: \tuple{v,0,w,\abspre{w} \prefixcon u} \tgsucc'_0 \pair{w}{\abspre{w}} \\
		& v \tgsucc_k w \:∧\: \nodels{v}{k} > 0 \:∧\: \lab{v} = @ \:∧\: \abspre{v} = \abspre{w} \prefixcon u \prefixcon \apre \\
		&\indent\implies~ \pair{v}{\abspre{v}} \tgsucc'_k \tuple{v,k,w,\abspre{v}} \:∧\: \tuple{v,k,w,\abspre{w} \prefixcon u} \tgsucc'_0 \pair{w}{\abspre{w}} \\
		& v \tgsucc_k w \:∧\: \nodels{v}{k} > 0 \:∧\: \tuple{v,k,w,\apre \prefixcon u}, \, \tuple{v,k,w,\apre} ∈ V' \\
		&\indent\implies~ \tuple{v,k,w,\apre \prefixcon u} \tgsucc'_0 \tuple{v,k,w,\apre} \\
		&v \tgsucc_k w \:∧\: \nodels{v}{k} > 0 \:∧\: \tuple{v,k,w,\apre \prefixcon u} ∈ V' \:∧\: j = 2\\
		&\indent\implies~ \tuple{v,k,w,\apre \prefixcon u} \tgsucc'_1 \pair{w}{\abspre{w}}
	\end{align*}
	for all \m{v,w,u ∈ V}, \m{k ∈ \set{0,1}}, \m{\apre ∈ V^*},
	and where the function \snodels is defined as:
	\begin{equation*}
	\nodels{v}{k}
	:=
	\begin{cases}
	\length{\abspre{v}} - \length{\abspre{v'}}
	& \text{if } v ∈ \vertsof{@} ∧ v \tgsucc_k v'
	\\[-0.5ex]
	\length{\abspre{v}} + 1 - \length{\abspre{v'}}
	& \text{if } v ∈ \vertsof{λ} ∧ v \tgsucc_k v'
	\\[-0.5ex]
	0 & \text{otherwise}
	\end{cases}
	\end{equation*}
	\aphotgstoltgsij{i}{j}{} induces the function~
	\m{\aphotgsisotoltgsisoij{i}{j}{} : \classaphotgsisoi{i} → \classltgsisoij{i}{j}},
	\m{\aaphotgiso = \eqcl{\aaphotg}{\iso} ↦ \eqcl{\aphotgstoltgsij{i}{j}{\aaphotg}}{\iso}}
	on the respective isomorphism equivalence classes.
\end{proposition}


\begin{proposition}[from λ-term-graphs to λ-ap-ho-term-graphs]\label{prop:mappings:ltgs:to:aphotgs}
	Let \m{i ∈ \set{0,1}} and \m{j ∈ \set{1,2}}. The mapping \ltgstoaphotgsij{i}{j}{} defined below is well-defined between the class of λ-term-graphs over \sigTGij{i}{j} and the class of λ-ap-ho-term-graphs over \sigTGi{i}:
	\[
		\ltgstoaphotgsij{i}{j}{} : \classltgsij{i}{j} → \classaphotgsi{i}, ~
		\altg = \tuple{V,\lab{},\args{},r} ↦ \tuple{V',\srestrictto{\lab{}}{V'},\args{}',r,\abspre{}'}
	\]
	where
	\begin{gather*}
		V' := \vertsof{λ} ∪ \vertsof{@} ∪ \vertsof{\0}
		\\
		\args{}' : V' → (V')^*,~\multilinebox{so that for the induced indexed\\successor relation \m{\tgsucc'_{(\cdot)}} on \m{V'} it holds:}
		\\
		v \tgsucc'_k w \;⇔\; v \tgsucc_k ⋅ ~ \tgsuccisstar{0}{\S} w ~~\sideCondition{\m{∀ v,w ∈ V', k ∈ \set{0,1}}}
		\\
		\abspre{}' := \srestrictto{\abspre{}}{V'} ~\text{where \abspre{} is the abstraction-prefix function of \altg.}
	\end{gather*}
	\ltgstoaphotgsij{i}{j}{} induces the function \m{\ltgsisotoaphotgsisoij{i}{j}{} : \classltgsisoij{i}{j} → \classaphotgsisoi{i}, \atgiso = \eqcl{\altg}{\iso} ↦ \eqcl{\ltgstoaphotgsij{i}{j}{\altg}}{\iso}} on the isomorphism equivalence classes.
\end{proposition}

The mappings \ltgstoaphotgsij{i}{j}{} and \aphotgstoltgsij{i}{j}{}
now define a correspondence between the class \classaphotgsi{i} of λ-ap-ho-term-graphs
and the class \classltgsij{i}{j} of λ-term-graphs in then sense as stated by the following theorem.

\begin{theorem}[correspondence between λ-ap-ho-term-graphs and λ-term-graphs]\label{thm:corr:aphotgs:ltgs}
	Let \m{i ∈ \set{0,1}} and \m{j ∈ \set{1,2}}. The mappings \ltgstoaphotgsij{i}{j}{} from \cref{prop:mappings:ltgs:to:aphotgs} and \aphotgstoltgsij{i}{j}{} from \cref{prop:mappings:aphotgs:to:ltgs} define a correspondence between the class of λ-term-graphs over \sigTGij{i}{j} and the class of λ-ap-ho-term-graphs over \sigTGi{i} with the following properties:
	\begin{enumerate}[(i)]
		\item\label{thm:corr:aphotgs:ltgs:i} \m{{\ltgstoaphotgsij{i}{j}{}} ∘ {\aphotgstoltgsij{i}{j}{}} = \idon{\classaphotgsi{i}}{}}. 
		\item\label{thm:corr:aphotgs:ltgs:ii} For all \m{\altg ∈ \classltgsij{i}{j}}: \m{(\aphotgstoltgsij{i}{j}{} ∘ {\ltgstoaphotgsij{i}{j}{})}(\altg) \funbisim^{\S} \altg}.
		\item\label{thm:corr:aphotgs:ltgs:iii} \ltgstoaphotgsij{i}{j}{} and \aphotgstoltgsij{i}{j}{} preserve and reflect the sharing orders on \classaphotgsi{i} and on \classltgsij{i}{j}:
		\begin{align*}
			& ∀ \aaphotg_1,\aaphotg_2 ∈ \classaphotgsi{i} ~ &
			\aphotgstoltgsij{i}{j}{\aaphotg_1} \funbisim \aphotgstoltgsij{i}{j}{\aaphotg_2}
			~ & \Longleftrightarrow~
			\aaphotg_1 \funbisim \aaphotg_2
			\\
			& ∀ \altg_1,\altg_2 ∈ \classltgsij{i}{j}) ~ &
			\ltgstoaphotgsij{i}{j}{\altg_1} \funbisim \ltgstoaphotgsij{i}{j}{\altg_2}
			~ & \Longleftrightarrow~
			\altg_1 \funbisim \altg_2
		\end{align*}
	\end{enumerate}
	Furthermore, statements analogous to \cref{thm:corr:aphotgs:ltgs:i}~\cref{thm:corr:aphotgs:ltgs:ii}, and \cref{thm:corr:aphotgs:ltgs:iii} hold for the correspondences \aphotgsisotoltgsisoij{i}{j}{} and \ltgsisotoaphotgsisoij{i}{j}{}, induced by \aphotgstoltgsij{i}{j}{} and \ltgstoaphotgsij{i}{j}{}, between the classes \classaphotgsisoi{i} and \classltgsisoij{i}{j} of isomorphism equivalence classes of graphs in \classaphotgsisoi{i} and \classltgsisoij{i}{j}, respectively.
\end{theorem}

\begin{example}\label{ex:corr:aphotgs:ltgs}
	The λ-ap-ho-term-graphs in \cref{fig:lambdaaphotg} correspond to the λ-ap-ho-term-graphs in \cref{fig:lambdatg} via \aphotgstoltgsij{i}{j}{} and \ltgstoaphotgsij{i}{j}{} as follows:
	\begin{hspread}
		\m{\aphotgstoltgsij{i}{j}{\aaphotg_0} = \altg_0'} &
		\m{\aphotgstoltgsij{i}{j}{\aaphotg_1} = \altg_1'} &
		\m{\ltgstoaphotgsij{i}{j}{\altg_0} = \aaphotg_0'} &
		\m{\ltgstoaphotgsij{i}{j}{\altg_1} = \aaphotg_1'}
	\end{hspread}
\end{example}

\begin{proposition}\label{prop:ltgstoaphotgsij-not-injective}
	The mapping \ltgstoaphotgsij{i}{j}{} from λ-term-graphs to λ-ap-ho-term-graphs is not injective.
	\begin{proof}
		This is witnessed by the following example.
	\end{proof}
\end{proposition}

\begin{example}[\ltgstoaphotgsij{i}{j}{} is not injective]\label{ex:ltgstoaphotgsij-not-injective}
	Consider the λ-ap-ho-term-graph \m{\aaphotg ∈ \classaphotgsi{0}}, and the λ-term-graphs \m{\altg,\altg' ∈ \classltgsij{0}{1}} below. It holds that \m{\ltgstoaphotgsij{0}{1}{\atg} = \ahotg = \ltgstoaphotgsij{0}{1}{\atg'}}.
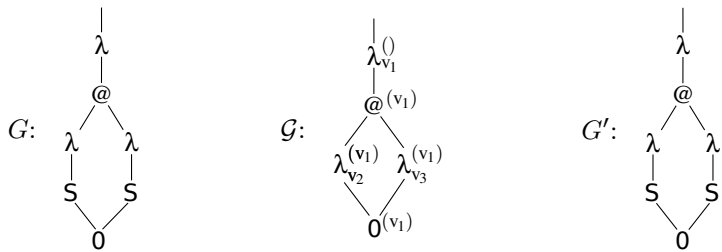
	\begin{hspread}
		\graph{\altg}{rem-corr-aphotgs-ltgs-nobij-snodes-shared} &
		\graph{\aaphotg}{rem-corr-aphotgs-ltgs-nobij-aphotg} &
		\graph{\altg'}{rem-corr-aphotgs-ltgs-nobij-snodes-unshared}
	\end{hspread}
\end{example}

\begin{para}[a weaker than bijective correspondence]
	In contrast to the correspondence between λ-ho-term-graphs and λ-ap-ho-term-graphs, which is bijective (\cref{thm:corr:lhotgs:aphotgs}), due to \cref{prop:ltgstoaphotgsij-not-injective} we have no bijective correspondence between λ-ap-ho-term-graph and λ-term-graphs. However, note that in \cref{ex:ltgstoaphotgsij-not-injective} it holds that \m{\altg' \funbisim^{\S} \altg}, and consequently \m{\altg \bisim^\S \altg'}. We can say that \atg and \m{\altg'} only differ in their `degree of \S-sharing'. Since \ltgstoaphotgsij{i}{j}{} ignores \S-vertices and their sharing, the degree of \S-sharing cannot be reflected in the corresponding λ-ap-ho-term-graph. This observation leads us to look for a weaker (than bijective) correspondence between λ-ap-ho-term-graph and λ-term-graphs: the weakening consists of equating \S-bisimilar λ-term-graphs as in the following proposition.
\end{para}

\begin{proposition}\label{prop:funbisimS:bisimS:ltgstoaphotgs}
	Let \m{i ∈ \set{0,1}} and \m{j ∈ \set{1, 2}}. The mapping \ltgstoaphotgsij{i}{j}{} in \cref{prop:mappings:ltgs:to:aphotgs} maps two λ-term-graphs that are \S-bisimilar to the same λ-ap-ho-term-graph. That is, for all λ-term-graphs \m{\atg_1,\atg_2} over \sigTGij{i}{j} it holds: \[\atg_1 \bisim^{\S} \atg_2 ~ \Longrightarrow~ \ltgstoaphotgsij{i}{j}{\atg_1} \iso \ltgstoaphotgsij{i}{j}{\atg_2}\]
\end{proposition}

\begin{remark}[λ-term-graphs over \sigTGij{i}{j} up to \S-bisimilarity]
	In the original paper \cite[5.15]{grab:roch:representations}, we use this insight to develop a variation on \classltgsij{i}{j}, where the correctness condition from \cref{ltg:ap-function} is relaxed to non-delimiter vertices. Consequently an \S-vertex is allowed to be a delimiter for two different scopes (as \G in \cref{ex:ltgstoaphotgsij-not-injective}). We show that these λ-term-graphs are always \S-bisimilar to graphs in \classltgsij{i}{j} and that there is a bijective correspondence with λ-ap-ho-term-graphs if we consider equivalence classes of these λ-term-graphs up to isomorphism and \S-bisimilarity. We omit this part this thesis. Instead we lean on the not quite bijective correspondence from \cref{prop:funbisimS:bisimS:ltgstoaphotgs} as is.
\end{remark}

\begin{para}[outlook]
	Be reminded that we seek to implement bisimulation and functional bisimulation on higher-order term graphs via (functional) bisimulation on λ-term-graphs (which are first-order term graphs). Now that we have a correspondence result that relates λ-term-graphs and higher-order term graphs, we investigate how λ-term-graphs behave under bisimulation and functional bisimulation, specifically which classes of λ-term-graphs are closed under (functional) bisimulation.
\end{para}

\section{Not closed under (functional) bisimulation}\label{sec:not:closed}

\begin{para}[overview]
	In this section we collect negative results concerning closedness under bisimulation and functional bisimulation for the classes of λ-term-graphs as introduced in the previous section.
\end{para}

\begin{proposition}\label{prop:ltgs:not:closed:under:bisim}
	None of the classes \classltgsi{1} and \classltgsij{i}{j}, for\/ \m{i ∈ \set{0,1}} and \m{j ∈ \set{1,2}}, of λ-term-graphs are closed under bisimulation.
\end{proposition}

\begin{para}
	This proposition is an immediate consequence of the following proposition, which can be viewed as a refinement, because it formulates non-closedness of classes of λ-term-graphs under specialisations of bisimulation, namely for functional bisimulation (under which some classes are not closed), and for converse functional bisimulation (under which none of the classes considered here is closed).
\end{para}

\begin{proposition}\label{ltgs-not-closed-under-funbisim}
	None of the classes \classltgsij{i}{j} of λ-term-graphs for\/ \m{i ∈ \set{0,1}} and \m{j ∈ \set{1,2}} are closed under functional bisimulation, or under converse functional bisimulation. Additionally, the class \classltgsi{1} of λ-term-graphs is not closed under converse functional bisimulation.
	\begin{proof}
		We proof the following statements by giving counterexamples:
		\begin{enumerate}[(i)]
			\item\label{ltgs-not-closed-under-funbisim:i}
				None of the classes \classltgsij{0}{j} for \m{j ∈ \set{1,2}} are closed under \funbisim, or under \convfunbisim.
			\item\label{ltgs-not-closed-under-funbisim:ii}
				None of the classes \classltgsi{1} and \classltgsij{1}{j} for \m{j ∈ \set{1,2}} of λ-term-graphs are closed under converse functional bisimulation \convfunbisim.
			\item\label{ltgs-not-closed-under-funbisim:iii}
				The class \classltgsij{1}{1} of λ-term-graphs is not closed \funbisim.
			\item\label{ltgs-not-closed-under-funbisim:iv}
				The class \classltgsij{1}{2} of λ-term-graphs is not closed under \funbisim.
		\end{enumerate}
		The counterexamples:
		\begin{enumerate}[(i)]
			\item\label{ltgs-not-closed-under-funbisim:i:proof}
				Let \Delta be one of the signatures \sigTGij{0}{j} \m{j ∈ \set{1,2}}. Consider the following term graphs over \Delta:
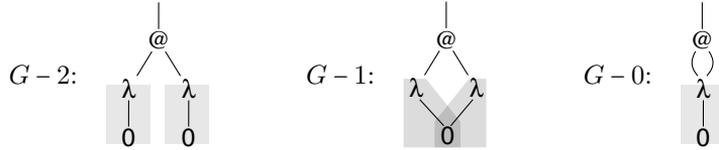
				\begin{hspread}
					\graph{\atg-2}{lambdatgs-not-closed-under-funbisim-convfunbisim-item-i-g2} &
					\graph{\atg-1}{lambdatgs-not-closed-under-funbisim-convfunbisim-item-i-g1} &
					\graph{\atg-0}{lambdatgs-not-closed-under-funbisim-convfunbisim-item-i-g0}
				\end{hspread}
				It holds that \m{\atg_2 \funbisim \atg_1 \funbisim \atg_0}. But while \m{\atg_2} and \m{\atg_0} admit correct abstraction-prefix functions over \Delta (since the implied scopes (drawn shaded) are nested), this is not the case for \m{\atg_1} (overlapping scopes). Therefore, the class of λ-term-graphs over \Delta is closed neither under functional bisimulation nor under converse functional bisimulation.
			\item\label{ltgs-not-closed-under-funbisim:ii:proof}
				Let \Delta be one of the signatures \sigTGi{1} and \sigTGij{1}{j}. Consider the following term graphs over \Delta:
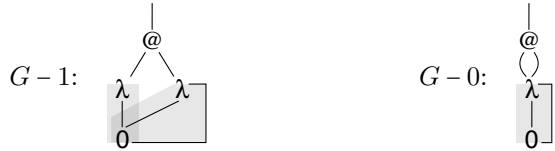
				\begin{hspread}
					\graph{\atg-1}{lambdatgs-not-closed-under-funbisim-convfunbisim-item-ii-g1} &
					\graph{\atg-0}{lambdatgs-not-closed-under-funbisim-convfunbisim-item-ii-g0}
				\end{hspread}
				It holds that \m{\atg_1 \funbisim \atg_0}. But while \m{\atg_0} admits a correct abstraction-prefix function, this is not the case for \m{\atg_1} (overlapping scopes). Therefore, the class of λ-term-graphs over \Delta is not closed under converse functional bisimulation.
			\item\label{ltgs-not-closed-under-funbisim:iii:proof}
				Consider the following term graphs over \sigTGij{1}{1}:
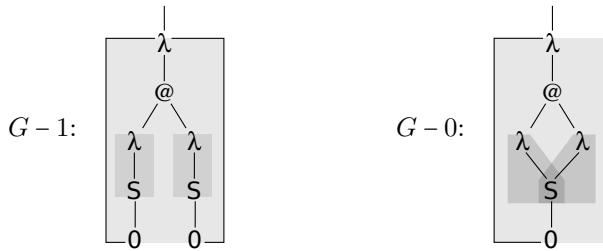
				\begin{hspread}
					\graph{\atg-1}{lambdatgs-not-closed-under-funbisim-convfunbisim-item-iii-g1} &
					\graph{\atg-0}{lambdatgs-not-closed-under-funbisim-convfunbisim-item-iii-g0}
				\end{hspread}
				It holds that \m{\atg_1 \funbisim \atg_0}. However, while \m{\atg_1} admits a correct abstraction-prefix function, this is not the case for \m{\atg_0} (overlapping scopes). Therefore the class of λ-term-graphs over \sigTGij{1}{1} is not closed under functional bisimulation.
			\item\label{ltgs-not-closed-under-funbisim:iv:proof}
				Consider the following term graphs over \sigTGij{1}{2}:
				\begin{hspread}
					\graph{\atg-1}{lambdatgs-over-siglambda12-not-closed-under-funcbisim-g1} &
					\graph{\atg-0}{lambdatgs-over-siglambda12-not-closed-under-funcbisim-g0}
				\end{hspread}
				It holds that \m{\atg_1 \funbisim \atg_0}. However, while \m{\atg_1} admits a correct abstraction-prefix function, this is not the case for \m{\atg_0} (overlapping scopes). Therefore the class of λ-term-graphs over \sigTGij{1}{2} is not closed under functional bisimulation.
		\end{enumerate}
	\end{proof}
\end{proposition}

As an easy consequence of \cref{prop:ltgs:not:closed:under:bisim}, and of \cref{ltgs-not-closed-under-funbisim}~\cref{ltgs-not-closed-under-funbisim:i} and \cref{ltgs-not-closed-under-funbisim:ii}, together with the examples used in the proof, we obtain the following two propositions.

\begin{proposition}\label{prop:not:closed:under:bisim:lambdahotgs:aplambdahotgs:siglambdai}
	Let \m{i ∈ \set{0,1}}. None of the classes \classlhotgsi{i} of λ-ho-term-graphs, or \classaphotgsi{i} of λ-ap-ho-term-graphs are closed under bisimulation on the underlying term graphs.
\end{proposition}

\begin{proposition}
	\label{prop:not:closed:under:funbisim:convfunbisim:lambdahotgs:aplambdahotgs:siglambdai}
	The following statements hold:
	\begin{enumerate}[(i)]
		\item\label{prop:not:closed:under:funbisim:convfunbisim:lambdahotgs:aplambdahotgs:siglambdai:i}
			Neither \classlhotgsi{0} nor \classaphotgsi{0} is closed under functional bisimulation or converse functional bisimulation on the underlying term graphs.
		\item\label{prop:not:closed:under:funbisim:convfunbisim:lambdahotgs:aplambdahotgs:siglambdai:ii}
			Neither \classlhotgsi{1} nor \classaphotgsi{1} is closed under converse functional bisimulation on the underlying term graphs.
	\end{enumerate}
\end{proposition}

\begin{remark}
	Note that \cref{prop:not:closed:under:funbisim:convfunbisim:lambdahotgs:aplambdahotgs:siglambdai}, \cref{prop:not:closed:under:funbisim:convfunbisim:lambdahotgs:aplambdahotgs:siglambdai:i} is a strengthening of the statement of \cref{prop:no:extension:funbisim:siglambda1} earlier.
\end{remark}

\section{Closed under functional bisimulation}\label{sec:closed}

\begin{para}[recapitulation]
	The negative results gathered in the last section leave the impression that our enterprise in a quite poor state: For the classes of λ-term-graphs we introduced, \cref{ltgs-not-closed-under-funbisim} leaves \classltgsi{1} as the only class of λ-term-graphs which still might be closed under functional bisimulation. Actually, \classltgsi{1} \emph{is} closed (we do not prove this here), but that does not help us any further, because the correspondences in \cref{thm:corr:aphotgs:ltgs} do not apply to this class, and worse still, \cref{prop:forgetful} rules out simple correspondences for \classltgsi{1}. So in this case we are left without any satisfying correspondences to higher-order term graphs that we have for the other classes of λ-term-graphs; those however are not closed under functional bisimulation.
\end{para}

\begin{para}[eager scope-closure]
	But in this section we establish that the class \classltgsij{1}{2} is very useful after all: we find that its restriction to λ-term-graphs with eager scope-closure (\cref{def:eager-scope} below) is in fact closed under functional bisimulation.
\end{para}

\begin{example}[eager-scope λ-term-graphs]\label{example:classeagltgsij12}
	Let us look at two \classltgsij{1}{2}-term-graphs from earlier and see whether we can fix closedness under functional bisimulation in that instance by making them eager-scope. Consider the λ-term-graph \m{\atg_1} from the \hyperref[ltgs-not-closed-under-funbisim:iv:proof]{proof} of \cref{ltgs-not-closed-under-funbisim}~\cref{ltgs-not-closed-under-funbisim:iv}. The scopes of the two topmost abstractions are not closed on the paths to variable occurrences belonging to the bottommost abstractions, although both scopes could have been closed immediately before the bottommost abstractions. When this is actually done, and the following variation \m{\tilde{\atg}_1} of \m{\atg_1} with eager scope-closure is obtained, then the problem disappears:
	\begin{hspread}
		\graph{\tilde{\atg}-1}{lambdatgs-over-siglambda12-closed-under-funcbisim-eager-backlinks-g1}&
		\funbisim &
		\graph{\tilde{\atg}-0}{lambdatgs-over-siglambda12-closed-under-funcbisim-eager-backlinks-g0}
	\end{hspread}
	\m{\tilde{\atg}_0} has again a correct abstraction-prefix function and is therefore a λ-term-graph.
\end{example}

For λ-term-graphs over \sigTGij{1}{1} and \sigTGij{1}{2} we define the eager-scope property.

\begin{definition}[eager-scope λ-term-graphs]\label{def:eager-scope}
	Let \m{\atg = \tuple{V,\lab{},\args{},r}} be a λ-term-graph over \sigTGij{1}{j} for \m{j ∈ \set{1,2}} with abstraction-prefix function \m{\abspre{} : V → V^*}. We say that \atg is an \emph{eager-scope λ-term-graph}, or that \atg\ \emph{is eager-scope} if:
	\begin{equation}\label{eq:def:eager:scope}
		\left.
		\begin{aligned}
			& ∀ v,w ∈ V~ ∀ p ∈ V^* \\
			& \indent\abspre{v} = p \prefixcon w ~∧~ v∉\vertsof{\S} \\
			& \indent\indent ⇒ ~
			\begin{aligned}[t]
				& ∃ n ∈ ℕ ~ ∃ v_1,\dots,v_n ∈ V \\
				& \indent v \tgsucc v_1 \tgsucc \dots \tgsucc v_n \tgsucc_0 w ~∧ \\
				& \indent v_n ∈ \vertsof{\0} ~∧~ ∀ i ∈ \set{1,\dots,n} ~ \abspre{v} ≤ \abspre{v_i}
			\end{aligned}
		\end{aligned}
		~~\right\}
	\end{equation}
	In words: from every non-delimiter vertex \v with non-empty abstraction-prefix \abspre{v} ending with \w there is a path to \w via vertices with abstraction-prefixes that extend \abspre{v} (a path within the scope of \w) and via a variable vertex just before reaching \w.
	\par By \eag{\classltgsij{1}{j}} we denote the subclass of \classltgsij{1}{j} of all eager-scope λ-term-graphs.
\end{definition}

\begin{para}[fully backlinked λ-term-graphs]
	We will find that not only the class of eager λ-term-graphs is closed under functional bisimulation, but also a super-class thereof, called `fully backlinked' λ-term-graphs.
\end{para}

\begin{definition}[fully backlinked λ-term-graphs]\label{fully-backlinked}
	Let \m{\atg = \tuple{V,\lab{},\args{},r}} be a λ-term-graph over \sigTGij{1}{j} for \m{j ∈ \set{1,2}} with abstraction-prefix function \m{\abspre{} : \V → \V^*}. We say that \atg is \emph{fully backlinked} if:
	\begin{equation}\label{eq:def:fully:backlinked}
		\left.
		\begin{aligned}
			& ∀ v,w ∈ V ~ ∀ \apre ∈ V^*\\
			& \indent \abspre{v} = \apre \prefixcon w ~⇒~
			\begin{aligned}[t]
				&∃ n ∈ ℕ ~ ∃ v_1,\dots,v_n ∈ V\\
				&\indent v \tgsucc v_1 \tgsucc \dots \tgsucc v_n \tgsucc w \\
				&\indent ∧~ ∀ i ∈ \set{1,\dots,n} ~ \abspre{v} ≤ \abspre{v_i}
			\end{aligned}
		\end{aligned}
		~~\right\}
	\end{equation}
	In words: from every vertices \v with non-empty abstraction prefix \abspre{v} that ends with \w, there is a path from \v to \w via vertices with abstraction-prefixes that extend \abspre{v} (a path within the scope of \w).
	\par By \fbl{\classltgsij{1}{j}} we denote the subclass of \classltgsij{1}{j} that consists of all fully backlinked λ-term-graphs.
\end{definition}

\begin{proposition}[~\m{\eag{\classltgsij{1}{2}} \subseteq \fbl{\classltgsij{1}{2}}}~]\label{prop:eager:scope:fully:backlinked}
	Every eager-scope λ-term-graph over \sigTGij{1}{2} is fully backlinked.
\end{proposition}

\begin{proof}
	Let \atg be an λ-term-graph over \sigTGij{1}{2} with abstraction-prefix function \abspre{} and vertex set \V. Suppose that \atg is an eager-scope λ-term-graph. This means that the condition \cref{eq:def:eager:scope} from \cref{def:eager-scope} holds for \atg. Now note that, for all \m{\w,\v ∈ \V} and \m{\apre ∈ \V^*}, if \m{\abspre{\w} = \apre \prefixcon \v} and \m{\w ∈ \vertsof{\S}} holds, then the correctness condition \cref{ltg:ap-function:S1} on the abstraction-prefix function \abspre{} entails \m{\w \tgsucc_1 \v}. It follows that \atg also satisfies the condition \cref{eq:def:fully:backlinked} from \cref{fully-backlinked}, and therefore, that \atg is fully backlinked.
\end{proof}

\begin{para}[backlinks and fully backlinkedness]\label{backlinkedness}
	The intuition for the property of a λ-term-graph \atg to be `fully backlinked' is that, for every vertex \w of \atg with a non-empty abstraction-prefix \m{\abspre{\w} = \apre \prefixcon \v}, it is possible to get back to the final abstraction vertex \v in the abstraction-prefix of \w by a directed path via vertices in the scope of \v and via a last edge that is a backlink from a variable or a delimiter vertex to the abstraction vertex \v. Therefore the presence of both sorts of backlink in λ-term-graphs is crucial for this concept. Indeed, at least backlinks for variables have to be present so that the property to be fully backlinked can make sense for a λ-term-graph: all λ-term-graphs with variable vertices but without variable backlinks (for example, consider a representation of the λ-term \abs{x}{x} by such a term graph) are not fully backlinked.
\end{para}

\begin{para}[eager-scope λ-term-graphs without variable backlinks]\label{rem:eag-scope:lambdatgs:siglambdaij:0}
	The situation is different for the eager-scope property. While the presence of backlinks was used for the definition of `eager-scope λ-term-graphs' in \cref{eq:def:eager:scope}, the assumption that all variable vertices have backlinks is not essential there. In fact the condition \cref{eq:def:eager:scope} can be generalised to apply also to λ-term-graphs over \sigTGij{0}{1} and \sigTGij{0}{2} as follows:
	\begin{equation}\label{eq:def:eager:scope:without:var:backlinks}
		\left.
		\begin{aligned}
			& ∀ v,w ∈ V~ ∀ p ∈ V^* \\
			& \indent\abspre{v} = p \prefixcon w ~∧~ v∉\vertsof{\S} \\
			& \indent\indent ⇒ ~
			\begin{aligned}[t]
				& ∃ n ∈ ℕ ~ ∃ v_1,\dots,v_n ∈ V ~ \\
				& \indent v \tgsucc v_1 \tgsucc \dots \tgsucc v_n \\
				& \indent ∧ ~ v_n ∈ \vertsof{\0} ~∧~ \abspre{v_n} = p \prefixcon w \\
				& \indent ∧ ~ ∀ i ∈ \set{1,\dots,n} ~ \abspre{v} ≤ \abspre{v_i}
			\end{aligned}
		\end{aligned}
		~~\right\}
	\end{equation}
	For λ-term-graphs over \sigTGij{1}{1} and \sigTGij{1}{2} the conditions \cref{eq:def:eager:scope} and \cref{eq:def:eager:scope:without:var:backlinks} coincide: For the implication \m{\cref{eq:def:eager:scope:without:var:backlinks} ⇒ \cref{eq:def:eager:scope}} note that the correctness condition \cref{ltg:ap-function:00} on the abstraction-prefix functions of these λ-term-graphs yields that the statements \m{v_n ∈ \vertsof{\0}} and \m{\abspre{v_n} = \abspre{v} = \apre \prefixcon w} imply that \m{v_n \tgsucc_0 w}. For the implication \m{\text{\cref{eq:def:eager:scope}} ⇒ \text{\cref{eq:def:eager:scope:without:var:backlinks}}} observe that \m{v_n ∈ \vertsof{\0}}, \m{v_n \tgsucc_0 w}, and \m{\abspre{v_n} ≤ \apre \prefixcon w} implies \m{\abspre{v_n} = \abspre{v} = \apre \prefixcon w} in view of the condition \cref{ltg:ap-function:00} and \cref{lem:ltg}.
\end{para}

\begin{proposition}\label{prop:hom:ltgs:preserve:reflect:eagscope:fb}
	Functional bisimulation on λ-term-graphs in \classltgsij{1}{j} with \m{j ∈ \set{1,2}} preserves and reflects the properties `eager scope' and `fully backlinked'. 
\end{proposition}

\begin{proof}
	We only show preservation under homomorphic image of the eager-scope property for λ-term-graphs, since preservation of the property `fully backlinked' can be shown analogously and involves less technicalities. Also, reflection of these properties under functional bisimulation can be demonstrated similarly.
	\par Let \m{j ∈ \set{1,2}}. Let \m{\altg_1 = \tuple{\V_1,\labi{1}{},\argsi{1}{},\r_1}} and \m{\altg_2 = \tuple{\V_2,\labi{2}{},\argsi{2}{},\r_2}} be λ-term-graphs over \sigTGij{1}{j} with (correct) abstraction-prefix functions \absprei{1}{} and \absprei{2}{}, respectively. Let \m{h : \V_1 → \V_2} be a homomorphism from \m{\altg_1} to \m{\altg_2}, and suppose that \m{\altg_1} is eager-scope. We show that also \m{\altg_2} is an eager-scope λ-term-graph.
	\par For this, let \m{\w' ∈ \V_2} such that \m{\w'∉\vertsof{\S}} and \m{\absprei{2}{\w'} = \apre' \prefixcon \v'} for some \m{\v' ∈ \V_2} and \m{\apre' ∈ (\V_2)^*}. Since \h is surjective by \cref{prop:hom:tgs:paths}~\cref{prop:hom:tgs:paths:iii}, there exists a vertex \m{\w ∈ \V_1} such that \m{h(\w) = \w'}. Now note that, due to \cref{prop:hom:image:absprefix:function:ltgs}, \absprei{2}{} is the homomorphic image of \absprei{1}{}. It follows that \m{h^*(\absprei{1}{\w}) = \absprei{2}{h(\w)} = \absprei{2}{\w'} = \apre' \prefixcon \v'}. Hence there exist \m{\v ∈ \V_1} and \m{\apre ∈ (\V_1)^*} such that \m{\absprei{1}{\w} = \apre \prefixcon \v} and \m{h^*(\apre) = \apre'} and \m{h(\v) = \v'}. Since \m{\atg_1} is an eager-scope λ-term-graph, there exists a path in \m{\altg_1} of the form: \[\apath :~ \w = \w_0 \tgsucc \w_1 \tgsucc \dots \tgsucc \w_n \tgsucc_0 \v\] such that \m{\w_n ∈ \vertsiof{1}{\0}} and \m{\absprei{1}{\w_i} ≥ \apre \prefixcon \v} for all \m{i ∈ \set{0,\dots,n}}. As \h is a homomorphism, it follows from \cref{prop:hom:tgs:paths}~\cref{prop:hom:tgs:paths:i} that \apath has an image \m{h(\apath)} in \m{\altg_2} of the form: \[ h(\apath) :~ \w' = h(\w) = h(\w_0) \tgsucc' h(\w_1) \tgsucc' \dots \tgsucc' h(\w_n) \tgsucc'_0 h(\v) = \v'\] where \m{\tgsucc'} is the directed-edge relation in \m{\altg_2}. Using again that \h is a homomorphism, it follows that \m{h(\w_n) ∈ \vertsiof{2}{\0}}. Due to the fact that \absprei{2}{} is the homomorphic image of \absprei{1}{} it follows that for all \m{i ∈ \set{0,\dots,n}} it holds: \m{\absprei{2}{h(\w_i)} = h^*(\absprei{1}{\w_i}) ≥ h^*(\apre \prefixcon \v) = h^*(\apre) \prefixcon h(\v) = \apre' \prefixcon \v'}. Hence we have shown that for \m{\w'_i := h(\w_i) ∈ \V_2} with \m{i ∈ \set{0,\dots,n}} it holds that \m{\w' = \w'_0 \tgsucc' \w'_1 \tgsucc' \dots \tgsucc' \w'_n \tgsucc'_0 \v'} such that \m{\w'_n ∈ \vertsiof{2}{\0}} and \m{\absprei{2}{\w'_i} ≥ \apre' \prefixcon \w'} for all \m{i ∈ \set{0,\dots,n}}. In this way we have shown that also \m{\altg_2} is eager-scope.
	\par For showing that the eager-scope property is reflected by a homomorphism \h from \m{\altg_1} to \m{\altg_2} the fact that paths in \m{\altg_2} have pre-images under \h in \altg (\cref{prop:hom:tgs:paths}~\cref{prop:hom:tgs:paths:ii}) can be used.
\end{proof}

\begin{para}[generalised conditions for eager-scope and fully backlinked]
	The following proposition states that the defining conditions for a λ-term-graph to be eager-scope (\cref{def:eager-scope}) or fully backlinked (\cref{fully-backlinked}) can be generalised. The generalised condition for eager-scopedness requires that for every non-delimiter vertex \v with \m{\abspre{}{\v} = \apre \prefixcon \w \prefixcon \bpre} there exists a path from \v to \w within the scope of \w that only transits variable-vertex backlinks, but not delimiter-vertex backlinks. In the generalised condition for fully-backlinkedness the conclusion, in which the path may also proceed via delimiter-vertex backlinks, holds for all vertices.
\end{para}

\begin{proposition}[generalised conditions for eager-scope and fully backlinked]\label{prop:eager:scope:fully:backlinked:pumped}
	Let \atg be a λ-term-graph over \sigTGij{1}{2}, and let \abspre{} be its abstraction-prefix function.
	\begin{enumerate}[(i)]
		\item\label{prop:eager:scope:fully:backlinked:pumped:eager:scope} \atg is an eager-scope λ-term-graph if and only if:
	\end{enumerate}
	\begin{equation}\label{eq:eager:scope:pumped}
		\left.
		\begin{aligned}
			& ∀ v,w ∈ V~ ∀ p,q ∈ V^*\\
			& \indent \abspre{v} = \apre \prefixcon w \prefixcon \bpre ~∧~ v∉\vertsof{\S} \\
			& \indent \indent ⇒ ~
			\begin{aligned}[t]
				& ∃ n ∈ ℕ ~ ∃ v_1,\dots,v_n ∈ V \\
				& \indent v \tgsucc v_1 \tgsucc \dots \tgsucc v_n \tgsucc w ~∧ \\
				& \indent v_n ∈ \vertsof{\0} ~∧~ ∀ 0 ≤ i ≤ n ~ \apre \prefixcon w ≤ \abspre{v_i}
			\end{aligned}
		\end{aligned}
		~~\right\}
	\end{equation}
	\begin{enumerate}[(i)]\setcounter{enumi}{1}
		\item\label{prop:eager:scope:fully:backlinked:pumped:fully:backlinked}
		The λ-term-graph \atg is fully backlinked if and only if:
	\end{enumerate}
	\begin{equation}\label{eq:fully:backlinked:pumped}
		\left.
		\begin{aligned}
			& ∀ v,w ∈ V ~ ∀ \apre,\bpre ∈ V^* \\
			& \indent \abspre{v} = \apre \prefixcon w \prefixcon \bpre
			~⇒~
			\begin{aligned}[t]
				& ∃ n ∈ ℕ ~ ∃ v_1,\dots,v_n ∈ V \\
				& \indent v \tgsucc v_1 \tgsucc \dots \tgsucc v_n \tgsucc w \\
				& \indent ∧~ ∀ 0 ≤ i ≤ n ~ \apre \prefixcon w ≤ \abspre{v_i}
			\end{aligned}
		\end{aligned}
		~~\right\}
	\end{equation}
\end{proposition}

\begin{proof}
	We will only prove the statement \cref{prop:eager:scope:fully:backlinked:pumped:fully:backlinked}, since statement \cref{prop:eager:scope:fully:backlinked:pumped:eager:scope} can be proven in much the same way.
	\par Let \atg be an arbitrary λ-term-graph over \sigTGij{1}{j} with \m{j ∈ \set{1,2}}. If \cref{eq:fully:backlinked:pumped} holds for \atg, then so does its special case \cref{eq:def:fully:backlinked}, and it follows that \atg is fully backlinked.
	\par For showing the converse we assume that \atg is fully backlinked. Then \cref{eq:def:fully:backlinked} holds. We show \cref{eq:fully:backlinked:pumped}, which is universally quantified over \m{\bpre ∈ \V^*}, by induction on the length \length{\bpre} of \bpre.
	\par In the base case we have \m{\bpre = \emptyword}, and the statement to be establish follows immediately from \cref{eq:def:fully:backlinked}.
	\par For the induction step, let \m{\apre,\bpre ∈ \V^*} and \m{\v,\w ∈ \V} be such that \m{\abspre{\v} = \apre \prefixcon \w \prefixcon \bpre} with \m{\bpre = \w' \prefixcon \bpre_0} for \m{\w' ∈ \V} and \m{\bpre_0 ∈ \V^*}. Due to \m{\length{\bpre_0} < \length{\bpre}}, the induction hypothesis can be applied to \m{\abspre{\v} = (\apre \prefixcon \w) \prefixcon \w' \prefixcon \bpre_0}, yielding a path \m{\v \pathto \v' \tgsucc \w'}, for some \m{\v' ∈ \V}, that in its part between \v and \m{\v'} visits vertices with \m{\apre \prefixcon \w \prefixcon \w'} as prefix of their abstraction prefixes. Due to \m{\apre \prefixcon \w \prefixcon \w' ≤ \abspre{\v'}} it follows that \m{\abspre{\w'} = \apre \prefixcon \w} by \cref{lem:ltg}, (i) (which entails that \m{\apre \prefixcon \w \prefixcon \w' = \abspre{\v'}} and that \m{\v' \tgsucc \w'} must be a backlink). From this the condition \cref{eq:def:fully:backlinked} yields a path \m{\w' \pathto \v'' \tgsucc \w}, for some \m{\v'' ∈ \V}, that on the way from \m{\w'} to \m{\v''} visits vertices with \m{\apre \prefixcon \w} as prefix of their abstraction prefix. Combining these two paths, we obtain a path \m{\v \pathto \v' \tgsucc \w' \pathto \v'' \tgsucc \w} that before it reaches \m{\v''} visits only vertices with \m{\apre \prefixcon \w} as prefix of their abstraction prefixes. This establishes the statement to show for the induction step.
\end{proof}

\begin{para}
	In order to show that eager-scope and fully backlinked λ-term-graphs are closed under functional bisimulation, we first establish the following lemma: if two vertices of such a λ-term-graph have the same image under a homomorphism \h, then so do their abstraction prefixes.
\end{para}

\begin{lemma}\label{lem:preserve:ltgs}
	Let \m{\atg ∈ \classltgsij{1}{2}} be a fully backlinked λ-term-graph with vertex set \V and abstraction-prefix function \abspre{}. Let \m{\atg' ∈ \classtgssiglambdaij{1}{2}} be a term graph over \sigTGij{1}{2} such that \m{\atg \funbisim_{\h} \atg'} for a homomorphism \h. Then it holds:
	\begin{equation}\label{eq:lem:preserve:ltgs}
		∀ \v_1,\v_2 ∈ \V ~ \; h(\v_1) = h(\v_2) ~⇒~ h^*(\abspre{\v_1}) = h^*(\abspre{\v_2})
	\end{equation}
	where \m{h^*} is the homomorphic extension of \h to words over \V.
\end{lemma}

To prove this lemma, as a stepping stone, we first prove the following technical lemma. To this end we use the generalised condition for fully backlinkedness from \cref{prop:eager:scope:fully:backlinked:pumped} to relate the abstraction-prefixes of vertices on paths departing from \m{v_1} and \m{v_2}.

\begin{lemma}\label{lem:preserve:paths:ltgs}
	Let \m{\atg = \tuple{V,\lab{},\args{},r} ∈ \classltgsij{1}{j}} for \m{j ∈ \set{1,2}} with abstraction-prefix function \abspre{}. Let \m{\atg' = \tuple{V',\lab{}',\args{}',r'} ∈ \classtgssiglambdaij{1}{j}} (i.e.\ \m{\atg' } is not necessarily a λ-term-graph) such that \m{\atg \funbisim_{h} \atg'} for a homomorphism \h. Let \m{v_1,v_2 ∈ V} be such that \m{h(v_1) = h(v_2)}.
	\par Suppose that \m{\abspre{v_1} = \apre_1 \prefixcon w_1 \prefixcon \bpre_1} with \m{w_1 ∈ V} and \m{\apre_1,\bpre_1 ∈ V^*} and that \m{\apath_1} is a path from \m{v_1} to \m{w_1} in \atg of the form
	\[\apath_1 : v_1 = v_{1,0} \tgsucc v_{1,1} \tgsucc \dots \tgsucc v_{1,n-1} \tgsucc v_{1,n} = w_1\]
	such that furthermore \m{\abspre{v_{1,n-1}} = \apre_1 \prefixcon w_1} holds.
	\par Then there are \m{w_2 ∈ V} and \m{\apre_2,\bpre_2 ∈ V^*} with \m{\length{\bpre_2} = \length{\bpre_1}} such that \m{\abspre{v_2} = \apre_2 \prefixcon w_2 \prefixcon \bpre_2}, and a path \m{\apath_2} in \atg from \m{v_2} to \m{w_1} of the form 
	\[\apath_2 : v_2 = v_{2,0} \tgsucc v_{2,1} \tgsucc \dots \tgsucc v_{2,n-1} \tgsucc v_{2,n} = w_2\]
	such that \m{\abspre{v_{2,n-1}} = \apre_2 \prefixcon w_2}, and that \m{h(v_{1,j}) = h(v_{2,j})} for all \m{j ∈ \set{0,\dots,n}}, and in particular \m{h(w_1) = h(w_2)}.
\end{lemma}

\begin{proof}
	\par By \cref{lem:ltg}~(i), from \m{\abspre{\v_{1,n-1}} = \apre_1 \prefixcon \w_1} it follows that \m{\abspre{\w_1} = \apre_1}. Hence the final transition in \m{\apath_1} must be either via the backlink of a variable vertex or a delimiter vertex and is therefore either of the form ~\m{v_{1,n-1} \tgsuccis{0}{\0} v_{1,n}}~ or of the form ~\m{v_{1,n-1} \tgsuccis{1}{\S} v_{1,n}}.
	\par Furthermore by \cref{prop:hom:tgs:paths}~\cref{prop:hom:tgs:paths:i}, it follows that there is a path \m{h(\apath_1)} in \m{\atg'} from \m{h(v_1)} to \m{h(w_1)} of the form:
	\[h(\apath_1) : h(v_1) = h(v_{1,0}) \tgsucc h(v_{1,1}) \tgsucc \dots \tgsucc h(v_{1,n-1}) \tgsucc h(v_{1,n}) = h(w_1)\]
	From this path, \cref{prop:hom:tgs:paths}~\cref{prop:hom:tgs:paths:ii}, yields a path \m{\apath_2} in \atg from \m{v_2} of the form:
	\[\apath_2 : v_2 = v_{2,0} \tgsucc v_{2,1} \tgsucc \dots \tgsucc v_{2,n-1} \tgsucc v_{2,n} = w_2\]
	such that \m{h(\apath_2) = h(\apath_1)}. Therefore it holds \m{h(v_{1,j}) = h(v_{2,j})} for all \m{j ∈ \set{1,\dots,n}}, and in particular \m{h(w_1) = h(w_2)}. Since \h is a homomorphism also \m{\lab{v_{1,j}} = \lab{v_{2,j}}} follows for all \m{j ∈ \set{0,\dots,n}}. Therefore, and due to \cref{lem:ltg}~\cref{lem:ltg:iv}, the abstraction-prefix function \abspre{} quantitatively behaves the same when stepping through \m{\apath_2} as when stepping through \m{\apath_1}. It follows that \m{\abspre{v_2} = \apre_2 \prefixcon w \prefixcon \bpre_2}, \m{\abspre{v_{2,n-1}} = \apre_2 \prefixcon w}, and \m{\abspre{w_2} = \apre_2} for some \m{w ∈ V}, and \m{\apre_2,\bpre_2 ∈ V^*} with \m{\length{\bpre_2} = \length{\bpre_1}}. Due to \m{\lab{v_{2,n-1}} = \lab{v_{1,n-1}}}, and since the final transition \m{v_{1,n-1} \tgsucc w_1} in \m{\apath_1} is one along a backlink of a variable or delimiter vertex, this also holds for the final transition \m{v_{2,n-1} \tgsucc w_2} in \m{\apath_2}. Then it follows by the correctness condition \cref{ltg:ap-function:01} or \cref{ltg:ap-function:S1}, respectively, that \m{w = w_2}, and therefore that \m{\abspre{v_2} = \apre_2 \prefixcon w_2 \prefixcon \bpre_2}, and \m{\abspre{v_{2,n-1}} = \apre_2 \prefixcon w_2}.
\end{proof}

Relying on this lemma, we can now give a rather straightforward proof of \cref{lem:preserve:ltgs}.

\begin{proof}[Proof of \cref{lem:preserve:ltgs}]
	Let \atg, \m{\atg'} be as assumed in the lemma, and let \h be a homomorphism that witnesses \m{\atg \funbisim_{\h} \atg'}. We first show:
	\begin{equation}\label{eq1:prf:lem:preserve:ltgs}
	\left.
	\begin{aligned}
		& ∀ \v_1,\v_2 ∈ \V ~ ∀ \w_1 ∈ \V ~ ∀ \apre_1,\bpre_1 ∈ \V^* \\
		& \indent h(\v_1) = h(\v_2) ~∧~ \abspre{\v_1} = \apre_1 \prefixcon \w_1 \prefixcon \bpre_1 ~\Longrightarrow \\
		& \indent \indent ∃ \w_2 ∈ \V ~ ∃ \apre_2,\bpre_2 ∈ \V^* ~
		\begin{aligned}[t]
			& \abspre{\v_2} = \apre_2 \prefixcon \w_2 \prefixcon \bpre_2 ~∧ \\
			& h(\w_1) = h(\w_2) ~∧~ \length{\bpre_2} = \length{\bpre_1}
		\end{aligned}
	\end{aligned}
	~~\right\}
	\end{equation}
	For this, let \m{\v_1,\v_2 ∈ \V} be such that \m{h(\v_1) = h(\v_2)} and \m{\abspre{\v_1} = \apre_1 \prefixcon \w_1 \prefixcon \bpre_1} for \m{\apre_1,\bpre_1 ∈ \V^*} and \m{\w_1 ∈ \V}. Now since \atg is fully backlinked, there is a path
	\m{\apath_1 : \v_1 = \v_{1,0} \tgsuccstar \v_{1,n-1} \tgsucc \v_{1,n} = \w_1} in \atg with \m{\abspre{\v_{1,n-1}} = \apre_1 \prefixcon \w_1}.
	Then \cref{lem:preserve:paths:ltgs} yields the existence of \m{\apre_1,\bpre_1 ∈ \V^*} and \m{\w_1 ∈ \V} with \m{\length{\bpre_2} = \length{\bpre_1}} such that \m{h(\w_1) = h(\w_2)}, and
	\m{\abspre{\v_2} = \apre_2 \prefixcon \w_2 \prefixcon \bpre_2}.
	This establishes \cref{eq1:prf:lem:preserve:ltgs}.
	\par As an direct consequence of \cref{eq1:prf:lem:preserve:ltgs} we obtain:
	\begin{equation*}\label{eq2:prf:lem:preserve:ltgs}
		\begin{aligned}
			∀ \v_1,\v_2 ∈ \V ~ ∀ \cpre_1 ∈ \V^* ~
			& h(\v_1) = h(\v_2) ~∧~ \abspre{\v_1} = \cpre_1 ~\Longrightarrow \\
			& \indent ∃ \dpre_2,\cpre_2 ∈ \V^* ~ \abspre{\v_2} = \dpre_2 \prefixcon \cpre_2 \;∧\; h^*(\cpre_1) = h^*(\cpre_2)
		\end{aligned}
	\end{equation*}
	But since in this statement the roles of \m{\v_1} and \m{\v_2} can be exchanged,
	it follows that
	\m{h(\v_1) = h(\v_2)}
	always entails
	\m{\length{\abspre{\v_1}} = \length{\abspre{\v_2}}},
	and hence
	\m{h^*(\abspre{\v_1}) = h^*(\abspre{\v_2})}.
	This establishes \cref{eq:lem:preserve:ltgs}.
\end{proof}

\Cref{lem:preserve:ltgs} is the crucial stepping stone for the proof of the main theorem of this chapter.

\begin{theorem}[preservation of the properties `eager scope' and `fully backlinked' under functional bisimulation on \classltgsij{1}{2}]\label{thm:preserve:ltgs}
	Let \altg be a λ-term-graph over \sigTGij{1}{2} and suppose that \h is a homomorphism from \altg to a term graph \m{\atg' ∈ \classtgssiglambdaij{1}{2}}. Then the following two statements hold:
	\begin{enumerate}[(i)]
		\item\label{thm:preserve:ltgs:i} If \atg is fully backlinked, then also \m{\atg'} is a λ-term-graph, which is fully backlinked.
		\item\label{thm:preserve:ltgs:ii} If \atg is eager-scope, then also \m{\atg'} is a λ-term-graph, which is eager scope.
	\end{enumerate}
\end{theorem}

\begin{proof}
	\newfunction\absprep{P'}
	Let \m{\altg = \tuple{\V,\lab{},\args{},\r} ∈ \classltgsij{1}{2}} and \m{\atg' = \tuple{\V',\lab{}',\args{}',\r'} ∈ \classtgssiglambdaij{1}{2}}, and \h a homomorphism from \altg to \m{\atg'}.
	\par For showing statement \cref{thm:preserve:ltgs:i} of the theorem, we assume that \altg is fully backlinked. On \m{\atg'} we define the following abstraction-prefix function:
	\begin{equation}\label{eq1:prf:thm:preserve:ltgs}
		\begin{aligned}[c]
			& \absprep{}~: & \V' &→ (\V')^*, \\
			&              & \w' &↦ h^*(\abspre{\w}) ~~ \text{for some \m{\w ∈ \V} with \m{h(\w) = \w'}}
		\end{aligned}
	\end{equation}
	This function is well-defined because: for every \m{\w' ∈ \V'} there exists a \m{\w ∈ \V} with \m{\w' = h(\w)}, due to \cref{prop:hom:tgs:paths}~\cref{prop:hom:tgs:paths:iii}; and due to \cref{lem:preserve:ltgs} the \m{\absprep{\w'}} is the same for all \m{\w ∈ \V} with \m{\w' = h(\w)}.
	\par For the thus defined abstraction-prefix function \absprep{} it holds:
	\[∀ \w ∈ \V ~ \;\, h^*(\abspre{\w}) = \absprep{h(\w)}\] and hence \absprep{} is the homomorphic image of \abspre{} under \h in the sense of \cref{def:hom:image:absprefix:function:ltgs}.
	\par It remains to show that \absprep{} is a correct abstraction-prefix function for \m{\atg'}, and that \m{\atg'} is fully backlinked. Both properties follow from statements established earlier by using that \m{\atg'} and \absprep{} are the homomorphic images under \h of \atg and \abspre{}, respectively: That \absprep{} is a correct abstraction-prefix function for \m{\atg'} follows from \cref{lem:hom:image:absprefix:function:ltgs}~\cref{lem:hom:image:absprefix:function:ltgs:i} which entails that the homomorphic image of a correct abstraction-prefix function is correct. And that \m{\atg'} is fully backlinked follows from \cref{prop:hom:ltgs:preserve:reflect:eagscope:fb}, which states that the homomorphic image of a fully backlinked λ-term-graph is again fully backlinked.
	\par In this way we have established statement \cref{thm:preserve:ltgs:i} of the theorem.
	\par For showing statement \cref{thm:preserve:ltgs:ii}, let \atg be an eager-scope λ-term-graph over \sigTGij{1}{2} . By \cref{prop:eager:scope:fully:backlinked} it follows that \atg is also fully backlinked. Therefore the just established statement \cref{thm:preserve:ltgs:i} of the theorem is applicable, and it yields that \m{\atg'} is a λ-term-graph. Since by \cref{prop:hom:ltgs:preserve:reflect:eagscope:fb} also the eager-scope property is preserved by homomorphism, it follows that \m{\atg'} is eager-scope, too. This establishes statement \cref{thm:preserve:ltgs:ii} of the theorem.
\end{proof}

\begin{corollary}\label{cor:closed:under:fun:bisim}
	\fbl{\classltgsij{1}{2}} and \eag{\classltgsij{1}{2}} are closed under functional bisimulation.
\end{corollary}

\begin{para}[backlinks are required]
	Note that statements analogous to \cref{thm:preserve:ltgs} and \cref{cor:closed:under:fun:bisim} do not hold for λ-term-graphs over \sigTGij{1}{1}: The classes \fbl{\classltgsij{1}{1}} and \eag{\classltgsij{1}{1}} are not closed under functional bisimulation. This is witnessed by the counterexample in the \hyperref[ltgs-not-closed-under-funbisim:iii:proof]{proof} of \cref{ltgs-not-closed-under-funbisim} \cref{ltgs-not-closed-under-funbisim:iii}, which maps an eager-scope (and hence fully-backlinked) λ-term-graph to a term graph that is not a λ-term-graph.
	\par Similarly the statements of \cref{thm:preserve:ltgs} and \cref{cor:closed:under:fun:bisim} do not carry over to λ-term-graphs over \sigTGij{0}{1} and \sigTGij{0}{2} that are eager-scope in the sense of \cref{rem:eag-scope:lambdatgs:siglambdaij:0}. This is witnessed by the eager-scope term graphs in the \hyperref[ltgs-not-closed-under-funbisim:i:proof]{proof} of \cref{ltgs-not-closed-under-funbisim}~\cref{ltgs-not-closed-under-funbisim:i} and \cref{ltgs-not-closed-under-funbisim:ii}.
\end{para}

Another direct consequence of \cref{thm:preserve:ltgs} is the following corollary. Recall the notations from \cref{def:ltgs:iso} and \cref{def:bisim-classes}.

\begin{corollary}\label{cor:2:thm:preserve:ltgs}
	The following statements holds:
	\begin{enumerate}[(i)]
		\item For every fully backlinked λ-term-graph \altg over \sigTGij{1}{2} it holds:
			\[\succsoford{\atgiso}{\funbisim} = \succsofordin{\atgiso}{\funbisim}{\fbl{\classltgsij{1}{2}}} = \succsofordin{\atgiso}{\funbisim}{\classltgsij{1}{2}}\]
	\item For every eager-scope λ-term-graph \altg over \sigTGij{1}{2} it holds: \[\succsoford{\atgiso}{\funbisim} = \succsofordin{\atgiso}{\funbisim}{\eag{\classltgsij{1}{2}}} = \succsofordin{\atgiso}{\funbisim}{\classltgsij{1}{2}}\]
	\end{enumerate}
\end{corollary}

This corollary will be central to proving in \cref{sec:transfer} the complete-lattice property for \funbisim-successors of (\iso-equivalence classes of) λ-term-graphs over \sigTGij{1}{2} and λ-ap-ho-term-graphs over \sigTGi{1}.

\begin{remark}[an alternative remedy]\label{rem:generalisation:to:non:eagscope}
	To recapitulate, we identify subclasses of \classltgsij{1}{2} that are closed under functional bisimulation. These subclasses are \fbl{\classltgsij{1}{2}} and the subclass \eag{\classltgsij{1}{2}} thereof. The reason why some λ-term-graphs are not fully backlinked is revealed by the counterexample of \cref{ltgs-not-closed-under-funbisim}~\cref{ltgs-not-closed-under-funbisim:iv}: in the scopes of the two topmost λ-abstractions in \m{\atg_1}, the two subgraphs representing \abs{x}{x} are `dangling' since the scopes of the topmost λ-abstractions are not closed; therefore the mentioned subgraphs can be shared in the homomorphic image \m{\atg_0}, thus leading to overlapping scopes.
	\par In the original paper \cite[7.12]{grab:roch:representations}, we explore an approach to mend this blemish of \classltgsij{1}{2} by changing its definition. We ensure that \emph{all} λ-term-graphs are fully backlinked. This is facilitated by allowing trailing \S-vertices as successors of \0-vertices. Thereby scopes that remain open at a variable occurrence are still closed afterwards. This solution is omitted in this thesis.
\end{remark}

\section{Transfer of the complete-lattice property to λ-ho-term-graphs}\label{sec:transfer}

\begin{para}[overview]
	In this section we establish that sets of \funbisim-successors of (the isomorphism equivalence class of) a given λ-term-graph over \sigTGij{1}{2} form a complete lattice under the sharing order. For this we use the fact that this is the case for first-order term graphs in general (see \cref{prop:funbisim:succs:of:tgs:iso:complete:lattice}), and we apply the results developed so far. Subsequently we transfer this complete-lattice property to the higher-order λ-ap-ho-term-graphs over \sigTGi{1} via the correspondences established in \cref{sec:ltgs}.
\end{para}


The following proposition is a specialisation \cref{prop:funbisim:succs:of:tgs:iso:complete:lattice}.

\begin{proposition}[complete-lattice property of \classtgssiglambdaij{i}{j}]\label{prop:funbisim:succs:of:ltgs:iso:complete:lattice}
	\pair{\succsoford{\atgiso}{\funbisim}}{\funbisim} is a complete lattice for every term graph \altg over \sigTGij{1}{2} with \m{i ∈ \set{0,1}} and \m{j ∈ \set{1,2}}.
\end{proposition}

Combining this proposition with the result from \cref{sec:closed} that the classes of fully backlinked and eager-scope λ-term-graphs over \sigTGij{1}{2} are closed under functional bisimulation (\cref{cor:closed:under:fun:bisim}) and with \cref{cor:2:thm:preserve:ltgs} we obtain the following theorem.

\begin{theorem}[complete-lattice property of \fbl{\classltgsij{1}{2}} and \eag{\classltgsij{1}{2}}]\label{thm:complete:lattice:fbl:eagscope:ltgs}
	\pair{\succsofordin{\atgiso}{\funbisim}{\smash{\classltgsij{1}{2}}}}{\funbisim} is a complete lattice for all \m{\altg ∈ \fbl{\classltgsij{1}{2}} ∪ \eag{\classltgsij{1}{2}}}.
\end{theorem}

\begin{para}[transfer to eager-scope higher-order term graphs]
	Considering this theorem we cannot expect all higher-order term graphs to form a complete lattice, but only those subclasses that correspond to fully backlinked or eager-scope λ-term-graphs. However, the `fully-backlinked' property does not have a natural equivalent on the higher-order term graphs, since it is based on the existence of paths that depend on backlinks departing from delimiter vertices, which the higher-order term graphs do not have. The eager-scope property on the other hand has an obvious counterpart in the higher-order term graphs. Its definition below is analogous to the eager-scope property of λ-term-graphs as described in \cref{rem:eag-scope:lambdatgs:siglambdaij:0}.
\end{para}

\begin{definition}[eager-scope λ-ap-ho-term-graphs]\label{def:eager-scope:aphotgs}
	Let \m{i ∈ \set{0,1}} and let \m{\aaphotg = \tuple{\V,\lab{},\args{},\r,\abspre{}}} be a λ-ap-ho-term-graph over \sigTGi{1}. We say that \aaphotg is an \emph{eager-scope λ-term-graph}, or that \aaphotg is \emph{eager-scope}, if:
	\begin{equation*}
		\begin{aligned}
			& ∀ w,v ∈ V ~ ∀ \apre ∈ V^* ~~ \abspre{w} = \apre \prefixcon v ~ ⇒ \\
			& \indent ∃ n ∈ ℕ ~ ∃ w_1,\dots,w_n ∈ V \\
			& \indent \indent
			\begin{aligned}[t]
				& w \tgsucc w_1 \tgsucc \dots \tgsucc w_n ~∧~ w_n ∈ \vertsof{\0} ~∧~ \abspre{w_n} = \apre \prefixcon v \\
				& ∧~ ∀ i ∈ \set{1,\dots,n-1} ~~ \apre \prefixcon v ≤ \abspre{w_i}
			\end{aligned}
		\end{aligned}
	\end{equation*}
	Or in other words, if for every vertex \w of \atg with a non-empty abstraction-prefix \abspre{\w} ending with \v there exists a path from \w via vertices with abstraction-prefixes that start with \abspre{\w} (i.e.\ a path within the scope of \v) to a variable vertex \m{\w_n} that is bound by the abstraction vertex \v. (Note that if \m{k=1}, then \v is directly reachable from \m{\w_n} via its backlink.)
	\par By \eag{\classaphotgsi{1}} we denote the subclass of \classaphotgsi{1} that consists of all eager-scope λ-ap-ho-term-graphs.
\end{definition}

\begin{proposition}[uniqueness of eager-scope λ-ap-ho-term-graphs]
	Let \m{\ahotg_i = \tuple{\V,\lab{},\args{},\r,\absprei{i}{}}} with \m{i ∈ \set{1,2}} be λ-ap-ho-term-graphs with the same underlying term graph. If  \m{\ahotg_1} is eager-scope, then \m{\length{\absprei{1}{\w}} ≤ \length{\absprei{2}{\w)}}} for all \m{\w ∈ \V}. If, in addition, also \m{\ahotg_2} is eager-scope, then \m{\absprei{1}{} = \absprei{2}{}}. Hence eager-scope λ-ap-ho-term-graphs over the same underlying term graph are unique.
\end{proposition}

\begin{proof}[Proof sketch]
	If a λ-ap-ho-term-graph is not eager-scope, then it contains a vertex \w with abstraction-prefix \m{\v_1 \dots \v_n} from which \m{\v_n} is only reachable (if at all) by leaving the scope of \m{\v_n}. It can be shown that in this case another abstraction-prefix function with shorter prefixes exists in which \m{\v_n} does not occur in the prefix of \w.
\end{proof}

As stated in the proposition below, the eager-scope property for λ-ap-ho-term-graphs and for λ-term-graphs correspond to each other via the correspondence mappings from \cref{prop:mappings:aphotgs:to:ltgs} and \cref{prop:mappings:ltgs:to:aphotgs}.

\begin{proposition}[\aphotgstoltgsij{i}{j}{} and \ltgstoaphotgsij{i}{j}{} preserve and reflect the eager-scope property]\label{prop:preserve:reflect:eagscope}
	Let \m{i ∈ \set{0,1}}, and \m{j ∈ \set{1,2}}. The correspondences \aphotgstoltgsij{i}{j}{} and \ltgstoaphotgsij{i}{j}{} between preserve and reflect the eager-scope property i.e., for all \m{\aaphotg ∈ \classaphotgsi{i}}, and for all \m{\altg ∈ \classltgsij{i}{j}} it holds:
	\begin{align*}
		\text{\aaphotg is eager-scope} ~ & ⇔~ \text{\aphotgstoltgsij{i}{j}{\aaphotg} is eager-scope} \\
		\text{\ltgstoaphotgsij{i}{j}{\altg} is eager-scope} ~ & ⇔~ \text{\altg is eager-scope}
	\end{align*}
	Consequently, the restriction of the domains of \aphotgstoltgsij{i}{j}{} and \ltgstoaphotgsij{i}{j}{} to eager-scope term graphs -- and also of \aphotgsisotoltgsisoij{i}{j}{} and \ltgsisotoaphotgsisoij{i}{j}{} -- are functions of following types:
	\begin{align*}
		\srestrictto{\aphotgstoltgsij{i}{j}{}}{\eag{\classaphotgsi{i}}} & : \eag{\classaphotgsi{i}} → \eag{\classltgsij{i}{j}} &
		\srestrictto{\ltgstoaphotgsij{i}{j}{}}{\eag{\classltgsij{i}{j}}} & : \eag{\classltgsij{i}{j}} → \eag{\classaphotgsi{i}} \\
		\srestrictto{\aphotgsisotoltgsisoij{i}{j}{}}{\eag{\classaphotgsisoi{i}}} & : \eag{\classaphotgsisoi{i}} → \eag{\classltgsisoij{i}{j}} &
		\srestrictto{\ltgsisotoaphotgsisoij{i}{j}{}}{\eag{\classltgsisoij{i}{j}}} & : \eag{\classltgsisoij{i}{j}} → \eag{\classaphotgsisoi{i}}
	\end{align*}
\end{proposition}

Furthermore, we can specialise \cref{thm:corr:aphotgs:ltgs} concerning the correspondences on the isomorphism equivalence classes (in particular \cref{thm:corr:aphotgs:ltgs:i} and \cref{thm:corr:aphotgs:ltgs:iii}) as follows. Recall \cref{def:bisim-classes}.

\begin{lemma}\label{lem:left-inverse:restrictions:to:bisim:equivclasses}
	Let \m{i ∈ \set{0,1}}, and \m{j ∈ \set{1,2}} and let \m{\aaphotg ∈ \classaphotgsi{i}} be a λ-ap-ho-term-graph. Then
	\begin{align*}
		\srestrictto{\aphotgsisotoltgsisoij{i}{j}{}}{\succsoford{\aaphotgiso}{\funbisim}} & ~:~ \succsoford{\aaphotgiso}{\funbisim} ~→~ \succsofordin{\aphotgsisotoltgsisoij{i}{j}{\aaphotgiso}}{\funbisim}{\classltgsij{i}{j}} \\
		\srestrictto{\ltgsisotoaphotgsisoij{i}{j}{}}{\succsofordin{\aphotgsisotoltgsisoij{i}{j}{\aaphotgiso}}{\funbisim}{\classltgsij{i}{j}}} & ~:~ \succsofordin{\aphotgsisotoltgsisoij{i}{j}{\aaphotgiso}}{\funbisim}{\classltgsij{i}{j}} ~→~ \succsoford{\aaphotgiso}{\funbisim}
	\end{align*}
	and it holds that \srestrictto{\ltgsisotoaphotgsisoij{i}{j}{}}{\succsofordin{\aphotgsisotoltgsisoij{i}{j}{\aaphotgiso}}{\funbisim}{\classltgsij{i}{j}}} is a left-inverse of \srestrictto{\aphotgsisotoltgsisoij{i}{j}{}}{\succsoford{\aaphotgiso}{\funbisim}}:
	\begin{equation*}
		\srestrictto{\ltgsisotoaphotgsisoij{i}{j}{}}{\succsofordin{\aphotgsisotoltgsisoij{i}{j}{\aaphotgiso}}{\funbisim}{\classltgsij{i}{j}}} ~∘~ \srestrictto{\aphotgsisotoltgsisoij{i}{j}{}}{\succsoford{\aaphotgiso}{\funbisim}} ~=~ \idon{\succsoford{\aaphotgiso}{\funbisim}}{}
	\end{equation*}
\end{lemma}

The final tool for transfering the complete-lattice property to higher-order term graphs is the following general lemma about partial orders. It states that the complete-lattice property of partial orders is reflected by order homomorphisms with left-inverses.

\begin{lemma}[reflection of the complete-lattice property under order homomorphisms with left-inverses]\label{lem:reflect:complete:lattice}
	Let \m{\pair{\A}{≤_{\A}}} and \m{\pair{\B}{≤_{\B}}} be partial orders. Suppose that \m{h : \A → \B}, and \m{i : \B → \A} are order homomorphisms such that \i is a left-inverse of \h, i.e.: \m{{\i} ∘ {\h} = \idon{\A}{}}. Then if \m{\pair{\B}{≤_{\B}}} is a complete lattice, then so is \m{\pair{\A}{≤_{\A}}}.
\end{lemma}

\begin{proof}
	Let \m{S \subseteq \A} be arbitrary. We have to show the existence of a least upper bound (l.u.b.) \lub{S} and a greatest lower bound (g.l.b.) \glb{S} of \m{S} in \m{\pair{\A}{≤_{\A}}}. We only show the existence of the l.u.b., because the argument for the g.l.b.\ is analogous. Since \m{\pair{\B}{≤_{\B}}} is a complete lattice, we know that  a l.u.b.\ \m{\lub{h(S)}} exists. We will show that \m{\lub{S} = i(\lub{h(S)})}.
	\par \m{i(\lub{h(S)})} is an upper bound of \m{S}: Let \m{s ∈ S} be arbitrary. Then \m{h(s) ∈ h(S)}, and \m{h(s) ≤_{\B} \lub{h(S)}}. Since \i is an order homomorphism, and a left-inverse of \h, it follows that \m{s = i(h(s)) ≤_{\A} i(\lub{h(S)})}.
	\par \m{i(\lub{h(S)})} is less or equal to all upper bounds of \m{S}: Let \m{u ∈ \A} be an upper bound of \m{S}. As \h is an order homomorphism, it follows that \m{h(u)} is an upper bound of \m{h(S)}. Consequently, \m{\lub{h(S)} ≤_{\B} h(u)}. Again, since \i is an order homomorphism, and a left-inverse of \h, this entails \m{i(\lub{h(S)}) ≤_{\A} i(h(u)) = u}.
\end{proof}

Now we can formulate and prove the main result of this section.

\begin{theorem}\label{hotgs-complete-lattice}
	For every \m{\ahotg ∈ \eag{\classlhotgsi{1}}} it holds that \pair{\succsoford{\aaphotgiso}{\funbisim}}{\funbisim} is a complete lattice.
\end{theorem}

\begin{proof}
	Let \m{\ahotg ∈ \eag{\classlhotgsi{1}}} Then by \cref{prop:preserve:reflect:eagscope} it follows that \m{\aphotgstoltgsij{1}{2}{\aaphotg} ∈ \eag{\classltgsij{1}{2}}} and also that \m{\aphotgsisotoltgsisoij{1}{2}{\aaphotgiso} ∈ \eag{\classltgsisoij{1}{2}}}. Now \cref{thm:complete:lattice:fbl:eagscope:ltgs} yields that \succsofordin{\aphotgsisotoltgsisoij{1}{2}{\aaphotgiso}}{\funbisim}{\classltgsij{1}{2}} w.r.t.\ \funbisim is a complete lattice. Furthermore note that \srestrictto{\ltgsisotoaphotgsisoij{1}{2}{}}{\succsofordin{\aphotgsisotoltgsisoij{1}{2}{\aaphotgiso}}{\funbisim}{\classltgsij{1}{2}}} is a left-inverse of \srestrictto{\aphotgsisotoltgsisoij{1}{2}{}}{\succsoford{\aaphotgiso}{\funbisim}} due to \cref{lem:left-inverse:restrictions:to:bisim:equivclasses}. Hence \cref{lem:reflect:complete:lattice} can be used to show that the complete-lattice property of \succsofordin{\aphotgsisotoltgsisoij{1}{2}{\aaphotgiso}}{\funbisim}{\classltgsij{1}{2}} w.r.t.\ \funbisim is reflected by \srestrictto{\aphotgsisotoltgsisoij{1}{2}{}}{\succsoford{\aaphotgiso}{\funbisim}}, yielding that \pair{\succsoford{\aaphotgiso}{\funbisim}}{\funbisim} is a complete lattice.
\end{proof}

\section{Summary}\label{sec:representations:conclusion}

\begin{para}[summary]
	We defined higher-order term graph representations for strongly regular λ-terms:
	\begin{itemize}
		\item λ-ho-term-graphs \classlhotgsi{i}, an adaptation of Blom's `higher-order term graphs' \cite{blom:2001}, which possess a scope function that maps every abstraction vertex \v to the set of vertices that are in the scope of \v.
		\item λ-ap-ho-term-graphs \classaphotgsi{i}, which instead of a scope function carry an abstraction-prefix function that assigns to every vertex information about the scoping structure. Abstraction prefixes are closely related to the notion of `generated subterms' for λ-terms (\cref{def:ST}). The correctness conditions here are simpler and more intuitive than for λ-ho-term-graphs.
	\end{itemize}
	These classes are defined for \m{i ∈ \set{0,1}}, according to whether variable occurrences have backlinks
	to abstractions or not.
	Our main statements about these classes are:
	\begin{itemize}
		\item a bijective correspondence between \classlhotgsi{i} and \classaphotgsi{i} via mappings \lhotgstoaphotgsi{i}{} and \aphotgstolhotgsi{i}{} that preserve and reflect the sharing order (\cref{thm:corr:lhotgs:aphotgs});
		\item the naive approach to implementing functional bisimulation on theses classes (ignoring all scoping information and using only the underlying first-order term graphs) fails (\cref{prop:forgetful}).
	\end{itemize}
	The latter was the motivation to consider first-order term graphs with scope delimiters:
	\begin{itemize}
		\item λ-term-graphs \classltgsij{i}{j} (with \m{i ∈ \set{0,1}} for whether there are variable backlinks and \m{j ∈ \set{1,2}} for whether there are delimiter backlinks), which are plain first-order term graphs, but which require as a correctness condition the existence of an abstraction-prefix function.
	\end{itemize}
	The most important results linking these classes with λ-ap-ho-term-graphs are:
	\begin{itemize}
		\item an `almost bijective' correspondence between \classaphotgsi{i} via mappings \aphotgstoltgsij{i}{j}{} and \ltgstoaphotgsij{i}{j}{} that preserve and reflect the sharing order (\cref{thm:corr:aphotgs:ltgs});
		\item the subclass \eag{\classltgsij{1}{2}} of eager-scope λ-term-graphs is closed under functional bisimulation (\cref{cor:closed:under:fun:bisim}).
	\end{itemize}
	\begin{figure}
		\begin{center}
			\begin{tikzpicture}
				\matrix[row sep=-5.35mm,column sep=1.5cm]{
					\node(tl){\phantom{I}};&
					\node(ml){\phantom{I}};&
					\node(bl){\phantom{I}};\\
					\node(t){\eag{\classlhotgsi{1}}};&
					\node(m){\eag{\classaphotgsi{1}}};&
					\node(b){\eag{\classltgsij{1}{2}}};\\
					\node(tr){\phantom{I}};&
					\node(mr){\phantom{I}};&
					\node(br){\phantom{I}};\\
				};
				\draw[->]($($(tl)!.5!(ml)$)!0.75cm!(tl)$) to node[above]{$\lhotgstoaphotgsi{1}{}$}  ($($(tl)!.5!(ml)$)!0.5cm!(ml)$);
				\draw[->]($($(bl)!.5!(ml)$)!0.65cm!(ml)$) to node[above]{$\aphotgstoltgsij{1}{2}{}$}  ($($(bl)!.5!(ml)$)!0.6cm!(bl)$);
				\draw[->]($($(br)!.5!(mr)$)!0.6cm!(br)$) to node[below]{$\ltgstoaphotgsij{1}{2}{}$} ($($(br)!.5!(mr)$)!0.65cm!(mr)$);
				\draw[->]($($(tr)!.5!(mr)$)!0.5cm!(mr)$) to node[below]{$\aphotgstolhotgsi{1}{}$} ($($(tr)!.5!(mr)$)!0.75cm!(tr)$);
			\end{tikzpicture}
			\\[2ex]
			\begin{tikzpicture}
				\matrix[row sep=1.4cm,column sep=1.4cm]{
					\node(tl){\ahotg};&
					\node(ml){\m{\aaphotg'}};&
					\node(bl){\altg};&\\
					\node(tr){\m{\ahotg_0}};&
					\node(mr){\m{\aaphotg_0'}};&
					\node(br){\m{\altg_0}};\\
				};
				\draw[funbisim](tl) to (tr);
				\draw[funbisim](ml) to (mr);
				\draw[funbisim](bl) to (br);
				\draw[|->](tl) to node[above]{\lhotgstoaphotgsi{1}{}} (ml);
				\draw[|->](ml) to node[above]{\aphotgstoltgsij{1}{2}{}} (bl);
				\draw[|->](br) to node[above]{\ltgstoaphotgsij{1}{2}{}} (mr);
				\draw[|->](mr) to node[above]{\aphotgstolhotgsi{1}{}} (tr);
			\end{tikzpicture}
		\end{center}
		\caption{The correspondences (top) permit an implementation of functional bisimulation on higher-order term graphs using functional bisimulation on first-order term graphs (bottom).}
		\label{fig:implementation-of-ho-fb}
	\end{figure}
	The correspondences together with the closedness result allow us to handle functional bisimulation between eager-scope higher-order term graphs in a straightforward manner by implementing them via functional bisimulation between first-order term graphs as shown in \cref{fig:implementation-of-ho-fb}.
\end{para}

\begin{para}[outlook: maximal sharing]
	The findings from this chapter are a toehold for the concept of maximal sharing, which we develop in the following chapter. In particular the complete-property of the classes \eag{\classlhotgsi{1}}, \eag{\classltgsij{1}{2}}, and \classtgssiglambdaij{1}{2} implies that for every graph in these classes there is a unique maximally compact version of that element -- its bisimulation collapse -- which we call its maximally shared form. For an eager λ-ho-term-graph \ahotg the maximally shared form can be computed as:
	\[\collC{\eag{\classlhotgsi{1}}}{\ahotg} = (\aphotgstolhotgsi{1}{} ∘ {{\ltgstoaphotgsij{1}{2}{}} ∘ {{\scollC{\classtgssiglambdaij{1}{2}}} ∘ {{\aphotgstoltgsij{1}{2}{}} ∘ {\lhotgstoaphotgsi{1}{}}}}})(\ahotg)\]
	For implementing \scollC{\classtgssiglambdaij{1}{2}} fast algorithms are available.
	\par Analogous to \cref{rem:generalisation:to:non:eagscope} this can be generalised to term graphs without eager scope-closure. For our intent of getting a grip on maximal sharing in \lambdaletrec, however, only eager scope-closure is practically relevant, because it facilitates the highest degree of sharing.
\end{para}


\chapter{Maximal Sharing in \texorpdfstring{\lambdaletreccal}{λletrec}}\label{chap:maxsharing}


\section{Overview}

\setcounter{theorem}{-1}
\begin{para}[teaser]\label{maxsharing:teaser}
	Can we transform \abs{f}{\letin{r=\app{f}{(\app{f}{r})}}{r}} into the more efficient form \abs{f}{\letin{r=\app{f}{r}}{r}} and can we prove that these terms are equivalent?
\end{para}

\begin{para}[subject matter]
	Increasing sharing in programs is generally desirable. It results in more compact code and avoids duplication of reduction work at run-time, thereby speeding up execution. We show how a maximal degree of sharing can be obtained for programs expressed as terms in the λ-calculus with \letrec. We introduce a notion of `maximal compactness' for \lambdaletrec-terms among all terms with the same unfolding. We translate \lambdaletrec-terms into a term graph representation which respects the unfolding semantics in the sense that bisimilarity preserves and reflects unfolding equivalence on the \lambdaletrec-terms. Compactness of the term graphs can then be compared and increased by means of functional bisimulation.
\end{para}

\begin{para}[methods and formalisms]
	We lean heavily on the results from the previous chapters. While in \cref{chap:expressibility} we completely focused on the λ-terms expressed by \lambdaletrec-terms, in \cref{chap:representations} we exclusively investigated possible suitable graph representations for \lambdaletrec-terms. In this chapter we put things together and connect the graph representations to the unfolding semantics. We provide a translation into higher-order and first-order term graphs and we show that it respects the unfolding semantics in the sense that bisimulation preserves and reflects the unfolding semantics.
\end{para}

\begin{para}[results]
	We obtain practical and efficient methods for the following two problems: transforming a \lambdaletrec-term into a maximally compact form; and deciding whether two \lambdaletrec-terms have the same unfolding.
	\par The transformation of a \lambdaletrec-term \L into maximally compact form \m{\L_0} proceeds in three steps:
	\begin{enumerate}[(i)]
		\item translate \L into a term graph \m{\atg = \graphsem{\L}};
		\item compute the maximally shared form of \atg as its bisimulation collapse \m{\atg_0};
		\item read back a \lambdaletrec-term \m{\L_0} from \m{\atg_0} with the property \m{\graphsem{\L_0} = \atg_0}.
	\end{enumerate}
	The transformation is sound in the sense that \m{\L_0} and \L have the same λ-term as their unfolding.
	\par The procedure for deciding whether two given \lambdaletrec-terms \m{\L_1} and \m{\L_2} are unfolding-equivalent computes their term graph interpretations \graphsem{\L_1} and \graphsem{\L_2},
	and checks whether these term graphs are bisimilar.
	\par We also provide an implementation.
\end{para}

\section{Preliminaries}





\begin{definition}[bisimulation collapse]\label{def:bisimulation_collapse}
	Let \m{\atg = \tuple{\V,\lab{},\args{},\r}} be a term graph. A \emph{bisimulation collapse} of \atg is a maximal element in the class \setcompr{G'}{\atg \funbisim G'} up to \iso, that is, a term graph \m{G'_0} with \m{\atg \funbisim G'_0} such that if \mbox{\m{G'_0 \funbisim G''_0}} for some term graph \m{G''_0}, then \mbox{\m{G''_0 \iso G'_0}}. Every two bisimulation collapses of \atg are isomorphic. This justifies the common abbreviation of saying that `the bisimulation collapse' of \atg is unique up to isomorphism.
\end{definition}

\section{Introduction}

\begin{para}[sharing by \letrec]
	Explicit sharing in pure functional programming languages is typically expressed by means of the \letrec-construct, which facilitates cyclic definitions (see also \cref{the-letrec-construct}).
	For the programmer the \letrec-construct offers the possibility to write a program more compactly by utilising subterm sharing. \letrec-expressions bind subterms to variables; these variables then denote occurrences of the respective subterms and can be used anywhere inside of the \letrec-expression (also recursively).
	In this way, instead of repeating a subterm multiple times,
	a single definition can be given which is then referenced from multiple positions.
\end{para}

\begin{example}[horizontal sharing]\label{ex-cse}
	Consider the λ-term \m{\app{(\abs{x}{x})}{(\abs{x}{x})}} with two occurrences of the subterm \abs{x}{x}. These occurrences can be shared as done in the \lambdaletrec-term \m{\letin{id=\abs{x}{x}}{\app{id}{id}}}. Obviously it holds: \m{\unfsem{\letin{id=\abs{x}{x}}{\app{id}{id}}} = \app{(\abs{x}{x})}{(\abs{x}{x})}}
\end{example}

As \Let-expressions permit cyclic definitions, sharing can not only occur horizontally but also vertically.

\begin{example}[vertical sharing]\label{ex:fix}
	Consider the \lambdaletrec-terms \L and \P and the λ-term \M from \cref{ex:expressibility:fix}:
	\begin{align*}
		&
		\begin{aligned}
			\L &:= \abs{f}{\letin{r=\app{f}{r}}{r}} \\
			\P &:= \abs{f}{\letin{r=\app{f}{(\app{f}{r})}}{r}}
		\end{aligned}
		&
		\M  &:=  \abs{f}{\app{f}{(\app{f}{(\dots)})}}
	\end{align*}
	Both \L and \P have \M as their infinite unfolding: \m{\unfsem{\L} = \unfsem{P}  = \M}. Note that \L represents \M in a more compact way than \P. It is intuitively clear that there is no \lambdaletrec-term that represents \M more compactly than \L. So \L can be called a `maximally shared form' of \P (and of \M).
\end{example}

\begin{remark}[twisted sharing]\label{horiz-vert-twisted-sharing}
	Besides horizontal and vertical sharing also a hybrid form of sharing can occur, a sort of superimposition of both kinds, which is called `twisted sharing' in \cite[Definition 4.1.7]{blom:2001}.
\end{remark}

\begin{para}[dynamic vs.\ static sharing]
	In the context of functional programming `sharing' can refer to two different -- albeit related -- notions. Here, we call them \emph{static} and \emph{dynamic} sharing, while in the literature about static and dynamic sharing, this distinction is usually not made explicitly.
	\par \emph{Static sharing} simply refers to a trait of some (or graph) languages in which a term (graph) with multiple occurrences of the same subterm (subgraph) can also be written more compactly, where the subterm (subgraph) is written out only once and referenced from multiple points. Static sharing is possible in most programing languages and most optimising compilers perform common subexpression elimination at compile time to increase sharing. \lambdaletreccal and the graph formalisms from the previous chapter with their unfolding semantics are of course typical examples of languages with static sharing. This thesis focuses (almost exclusively) on static sharing.
	\emph{Dynamic sharing} refers to the degree of `static` sharing an evaluator is able to maintain during evaluation. This is an important issue because unsharing is an integral part of evaluation. A shared function typically behaves differently in different contexts (i.e.\ with different input), therefore it is impossible to evaluate the entire function only once if the result is required for two different inputs. However some portion of the computation may very well be shared. The degree of sharing of an evaluator refers to how clever it performs unsharing, and therefore how much of the computation can be shared. Terms as `call-by-need', `full laziness' \cite{wads:1971}, `complete laziness', and `optimal evaluation' \cite{levy:1978,aspe:guer:1998} all refer to dynamic sharing. An overview can be found in \cite[3.4]{thyer}. While our maximal sharing is a priori a method for increasing static, not dynamic, sharing, we do envisage applying it as part of an evaluator, collapsing the program's graph representation periodically at run-time (see \cref{sec:applications}).
\end{para}

\begin{para}[`maximal sharing' in ATERM]
	The term `maximal sharing' stems from work on the ATERM library \cite{bran:klin:2007}. It describes a technique for minimising memory usage when representing a set of terms in a first-order term rewrite system (TRS). The terms are kept in an aggregate directed acyclic graph by which their syntax trees are shared as much as possible. Thereby terms are created only if they are entirely new; otherwise they are referenced by pointers to roots of sub-dags. Our use of the expression `maximal sharing' is inspired by that work, but our results generalise that approach in the following ways:
	\begin{itemize}
		\item Instead of first-order terms we consider terms in higher-order languages.
		\item Since \letrec can express cyclic sharing, we interpret terms as cyclic graphs instead of just dags.
		\item We increase sharing by bisimulation collapse instead of by identifying isomorphic sub-dags.
	\end{itemize}
\end{para}

\begin{para}[common subexpression elimination]
	ATERM only checks for equality of subexpressions. Therefore it only introduces horizontal sharing (for a definition see \cite{blom:2001}) and implements a form of \emph{common subexpression elimination (CSE)} \cite[14.7.2]{peyt:jone:1987}. Our approach is stronger than CSE: while \cref{ex-cse} can be handled by CSE, this is not the case for \cref{ex:fix}. In contrast to CSE, our approach increases also vertical and twisted sharing (see \cref{horiz-vert-twisted-sharing}).
\end{para}

\section{Overview: Methods and Formalisms}

Here we will quickly introduce mathematical symbols for the central formalisms (which are later properly defined) so that we can sketch a complete picture of the algorithms we develop in this chapter.

\begin{para}[graph formalisms]
	The methods that we introduce in this chapter rely heavily on the term graph formalisms developed in \cref{chap:representations}. As a main result of \cref{chap:representations} we have identified suitable classes of term graphs for representing regular λ-terms. In particular we will use eager-scope λ-ap-ho-term-graphs\footnote{Note that λ-ho-term-graphs would be just as suitable due to \cref{thm:corr:lhotgs:aphotgs}} with variable backlinks and eager-scope λ-term-graphs with variable and scope-delimiter backlinks, however over a slightly modified signature which includes black holes. In this chapter we will use an abbreviated notation for these classes. λ-ap-ho-term-graphs (originally \classaphotgsi{1}) amended with black holes are now denoted by \classlhotgs. λ-term-graphs (originally \classltgsij{1}{2}) amended with black holes are now denoted by \classltgs.
\end{para}

\begin{para}[graph semantics]
	We provide a higher-order graph semantics \graphsemC{\classlhotgs}{} for interpreting \lambdaletrec-terms as eager-scope λ-ap-ho-term-graphs. Together with the correspondence from \cref{thm:corr:aphotgs:ltgs} it induces a first-order graph semantics \graphsemC{\classltgs}{}. Specifically we use the mapping \aphotgstoltgsij{1}{2}{} from λ-ap-ho-term-graphs to λ-term-graphs from \cref{prop:mappings:aphotgs:to:ltgs}, which (amended to handling black holes) we call \lhotgstoltgs{} in this chapter.
\end{para}

\begin{para}[readback]
	In order to be able to compute the \lambdaletrec-term that a λ-term-graph stands for, we provide a readback function \readback{} from λ-term-graphs to \lambdaletrec-terms with the property that it is a right inverse of \graphsemC{\classltgs}{} up to graph isomorphism.
	A \emph{readback} function \readback{} from λ-term-graphs to \lambdaletrec-terms that, for every λ-term-graph \altg, computes a \lambdaletrec-term \L from the set of \lambdaletrec-terms that have \altg as their first-order interpretation via \graphsemC{\classlhotgs}{} and \lhotgstoltgs{} (i.e.\ a \lambdaletrec-term for which it holds that \m{\lhotgstoltgs{\graphsemC{\classlhotgs}{\L}} = \atg}).
\end{para}

\begin{figure}
	\begin{tikzpicture}
		\matrix[row sep=0.8cm,column sep=1.2cm]{
			\node(L){\L};&
			\node(G){\ahotg};&
			\node(g){\altg};
			\\
			\node[xshift=3mm](M){\M};
			\\
			\node(L0){\m{\L_0}};&
			\node(G0){\m{\ahotg_0}};&
			\node(g0){\m{\altg_0}};
			\\
		};
		\draw[|->] (L) to node[left]{\unfsem{}} (M);
		\draw[|->] (L0) to node[left]{\unfsem{}} (M);
		\draw[|->] (L) to node[above]{\graphsemC{\classlhotgs}{}} (G);
		\draw[|->] (G) to node[above]{\lhotgstoltgs{}} (g);
		\draw[bend left=40,|->] (L) to node[above]{\graphsemC{\classltgs}{}} (g);
		\draw[|->] (L0) to node[above]{\graphsemC{\classlhotgs}{}} (G0);
		\draw[|->] (G0) to node[above]{\lhotgstoltgs{}} (g0);
		\draw[bend left=32,|->] (g0) to node[below]{\readback{}} (L0);
		\draw[bend left=41,|->] (L0) to node[above]{\graphsemC{\classltgs}{}\hspace*{5ex}} (g0);
		\draw[funbisim] (G) to node[right]{\scoll} (G0);
		\draw[funbisim] (g) to node[right]{\scoll} (g0);
	\end{tikzpicture}
	\hfill
	\begin{tikzpicture}
		\matrix[row sep=0.8cm,column sep=1.2cm]{
			\node(L1){\m{\L_1}};&
			\node(G1){\m{\ahotg_1}};&
			\node(g1){\m{\altg_1}};
			\\
			\node[xshift=3mm](M){\M};
			\\
			\node(L2){\m{\L_2}};&
			\node(G2){\m{\ahotg_2}};&
			\node(g2){\m{\altg_2}};
			\\
		};
		\draw[|->] (L1) to node[left]{\unfsem{}} (M);
		\draw[|->] (L2) to node[left]{\unfsem{}} (M);
		\draw[|->] (L1) to node[above]{\graphsemC{\classlhotgs}{}} (G1);
		\draw[|->] (G1) to node[above]{\lhotgstoltgs{}} (g1);
		\draw[|->] (L2) to node[above]{\graphsemC{\classlhotgs}{}} (G2);
		\draw[|->] (G2) to node[above]{\lhotgstoltgs{}} (g2);
		\draw[bisim] (G1) to (G2);
		\draw[bisim] (g1) to (g2);
		\draw[bend left=40,|->] (L1) to node[above]{\graphsemC{\classltgs}{}} (g1);
		\draw[bend left=40,|->] (L2) to node[above]{\m{\graphsemC{\classltgs}{}\hspace*{6ex}}} (g2);
		\draw[color=white,bend left=32,|->] (g2) to node[below]{\readback{}} (L2); 
	\end{tikzpicture}
	\caption{On the left: computing the maximally shared form \m{\L_0} of a \lambdaletrec-term \L via bisimulation collapse \scoll. On the right: deciding unfolding equivalence of \lambdaletrec-terms \m{\L_1} and \m{\L_2} via bisimilarity \bisim.}
	\label{fig:methods}
\end{figure}
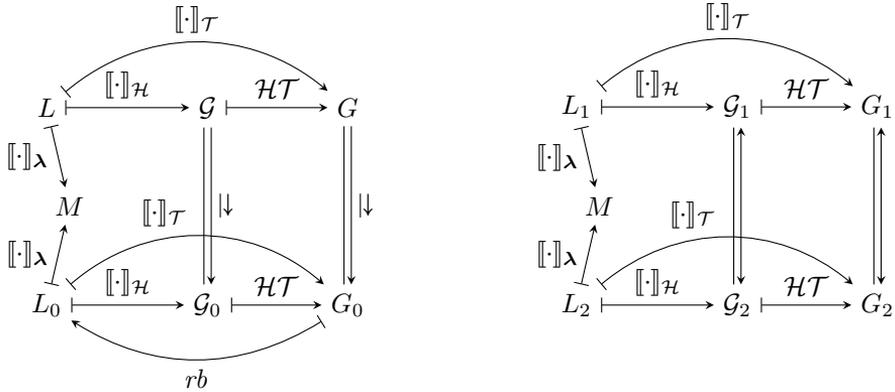

\begin{para}[methods]\label{pipeline}
	Putting the above formalisms together we obtain the followings methods, illustrated in \cref{fig:methods}):
	\begin{itemize}
		\item \emph{Maximal sharing}:
			for a given \lambdaletrec-term, a maximally shared form can be obtained by collapsing its first-order term graph interpretation, and then reading back the collapse: \m{{{\readback{}} ∘ {\scoll}} ∘ {\graphsemC{\classltgs}{}}}
		\item \emph{Unfolding equivalence}:
			for given \lambdaletrec-terms \L and \P, it can be decided whether \m{\unfsem{\L} = \unfsem{P}} by checking whether their term graph interpretations \graphsemC{\classltgs}{\L} and \graphsemC{\classltgs}{P} are bisimilar.
	\end{itemize}
	See \cref{fig:ex:compact} for an illustration of the application of the maximal sharing method to the \lambdaletrec-terms \L and \P from \cref{ex:fix}.
\end{para}

\begin{figure}
	\begin{tikzpicture}
		\matrix[row sep=-2mm,column sep=0.3cm]{
			\node(L1){\m{\abs{f}{\letin{r=\app{f}{r}}{r}}}}; &&
			\node(L2){\m{\abs{f}{\letin{r=\app{f}{(\app{f}{r})}}{r}}}};\\[5ex]
			&\node(M){\m{\abs{f}{\app{f}{(\app{f}{(\dots)})}}}}; \\
			\node(G1){\fig{fix-eff}}; && \node(G2){\fig{fix-big}};\\
		};
		\draw[funbisim] (G2) to node[above]{\scoll} (G1);
		\draw[|->] (L1) to node[above,near end]{\unfsem{}} (M);
		\draw[|->] (G1.100) to node[left]{\readback{}} (L1.235);
		\draw[|->] (L2) to node[above,near end]{\unfsem{}} (M);
		\draw[|->] (L1) to node[right]{\graphsemC{\classltgs}{}} (G1);
		\draw[|->] (L2) to node[right]{\graphsemC{\classltgs}{}} (G2);
	\end{tikzpicture}
	\caption{Computing the maximally shared version of the term \P (on the right) from \cref{ex:fix} yielding \L (on the left).}
	\label{fig:ex:compact}
\end{figure}
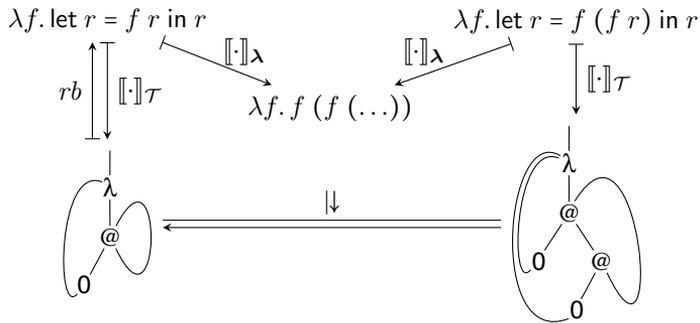

\begin{para}[correctness and practicality]
	The correctness of these methods hinges on the fact that the translation and the readback satisfy the following properties:
	\begin{enumerate}[(P1)]
		\item\label[property]{methods:properties:correctness}
			\lambdaletrec-terms \L and \P have the same infinite unfolding if and only if the term graphs \graphsemC{\classltgs}{\L} and \graphsemC{\classltgs}{P} are bisimilar. 
		\item\label[property]{methods:properties:closedness}
			The class \classltgs of λ-term-graphs is closed under functional bisimulation.
		\item\label[property]{methods:properties:readback}
			The readback \readback{} is a right inverse of \graphsemC{\classltgs}{} up to isomorphism \iso, that is, for all term graphs \m{\atg ∈ \classltgs} it holds: \m{\graphsemC{\classltgs}{\readback{\atg}} \iso \atg}.
	\end{enumerate}
	Furthermore, practicality of these methods depends on the property:
	\begin{enumerate}[(P1)]\setcounter{enumi}{3}
		\item\label[property]{methods:properties:efficiency}
			Translation \graphsemC{\classltgs}{} and readback \readback{} are efficiently computable.
	\end{enumerate}
\end{para}

\begin{para}[applications]\label{maxsharing:applications}
	Our approach holds promise for a number of practical applications:
	\begin{itemize}
		\item Increasing the efficiency of the execution of programs by transforming them into their maximally shared form at compile-time.
		\item Increasing the efficiency of the execution of programs by periodically compactifying the program at run time.
		\item Improving systems for recognising program equivalence.
		\item Providing feedback to the programmer, along the lines: `This code has identical fragments and can be written more compactly.'
	\end{itemize}
	These and a number of other potential applications are discussed in more detail in \cref{sec:applications}.
\end{para}


\section{Interpretion of \lambdaletrec-terms as λ-ap-ho-term-graphs}\label{sec:lhotgs:subsec:trans}

\begin{para}[natural first-order semantics]
	First we will quickly look at a naive translation of \lambdaletrec-terms into first-order term graphs over \sigTGi{1} which we call the \emph{natural first-order semantics} of \lambdaletreccal and see that it does not work. We will not give a formal definition of that translation but only the following informal description: consider the syntax tree of a \lambdaletrec-term and resolve variable occurrences as \0-vertices with variable backlinks; furthermore resolve occurrences of function variables as an edge to the root of the corresponding function's subgraph in the term graph translation.
\end{para}

\begin{example}[incorrectness of the natural first-order semantics]\label{ex:counternat}
	For the terms \L, \m{L_1}, \m{L_2} below, it holds that \m{\unfsem{\L_1} = \unfsem{\L} ≠ \unfsem{\L_2}}:
	\begin{equation*}
		\begin{aligned}
			\L_1 & ~=~ \letin{f = \abs{x}{\app{(\abs{y}{\app{f}{y}})}{x}}}{f} \\
			\L   & ~=~ \letin{f = \abs{x}{\app{f}{x}}}{f} \\
			\L_2 & ~=~ \letin{f = \abs{x}{\app{(\abs{y}{\app{f}{x}})}{x}}}{f}
		\end{aligned}
	\end{equation*}
	Their term graph interpretation under the natural first-order semantics \m{G_1}, \G, and \m{G_2} are however all bisimilar:
	\begin{hspread}
		\vcentered{\fig{counterex-onlyvarbl-y}} & \vcentered{\funbisim} & \vcentered{\fig{counterex-onlyvarbl-collapse}} & \vcentered{\invfunbisim} & \vcentered{\fig{counterex-onlyvarbl-x}} \\
		\m{G_1} & & \G & & \m{G_2}
	\end{hspread}
	This violates \cref{methods:properties:correctness} as under this semantics bisimilarity does not guarantee unfolding equivalence. Therefore the natural first-order semantics is incorrect.
\end{example}

Let us consider the translation of the above example into λ-ap-ho-term-graphs, before we give a formal definition thereof.

\begin{example}[the λ-ap-ho-term-graphs of the terms in \cref{ex:counternat}]\label{ex:counternat_lhotgs}
	\hfill
	\begin{hspread}
		\vcentered{\fig{counterex-onlyvarbl-y-prefixed}} & \vcentered{\funbisim} & \vcentered{\fig{counterex-onlyvarbl-collapse-prefixed}} & \vcentered{\not\invfunbisim} & \vcentered{\fig{counterex-onlyvarbl-x-prefixed}} \\
		\graphsemC{\classlhotgs}{L_1} & & \graphsemC{\classlhotgs}{\L} & & \graphsemC{\classlhotgs}{L_2}
	\end{hspread}
	We see that the abstraction prefixes prevent \graphsemC{\classlhotgs}{\L} and \graphsemC{\classlhotgs}{L_2} to be bisimilar, and that \cref{methods:properties:correctness} is satisfied in this instance.
\end{example}

\begin{para}[cf.\ bisimulation on the underlying term graphs]
	The shortcoming of the natural first-order semantics is reflected in \cref{prop:no:extension:funbisim:siglambda1} which states that higher-order term graphs are not closed under functional bisimulation on their underlying term graphs.
\end{para}

\begin{para}[λ-ap-ho-term-graphs with black holes]\label{laphotgs-blackholes}
	Before we finally define the higher-order term graph semantics \graphsemC{\classlhotgs}{} for \lambdaletrec-terms, we first have to slightly adapt the graph formalism. The reason is that in \cref{chap:representations} we introduced higher-order term graphs as a representation for strongly regular λ-terms, while here we want them to represent \lambdaletrec-terms. We know these two classes to coincide (\cref{thm:ll-expressible:2:streg}), but only if we exclude meaningless \lambdaletrec-terms (see \cref{meaningless_bindings}). Here, however, we want to be able to handle all \lambdaletrec-terms. Therefore in this chapter we need to consider an amended definition of λ-ap-ho-term-graphs over an extended signature \sigTGbh which includes black hole vertices.
\end{para}

\begin{definition}[signature for λ-ap-ho-term-graphs with black holes]\label{sigTGbh}
	By \sigTGbh we denote the signature \sigTGi{1} extended by a black-hole symbol \bh with arity \m{0}, i.e.\ \m{\sigTGbh = \set{@, λ, \0, \bh}} with \m{\arity{@} = 2}, \m{\arity{λ} = 1}, \m{\arity{\0} = 1}, and \m{\arity{\bh} = 0}.
\end{definition}

We also have to amend the correctness conditions \cref{def:abspre:function:siglambdai} for the abstraction-prefix function by a black-hole case.

\begin{definition}[abstraction-prefix function for \sigTGbh-term-graphs]\label{abs-pre-function:sigTGbh}
	As \cref{def:abspre:function:siglambdai} but over signature \sigTGbh and with the following condition added:
	\begin{align}
		\tag{black hole}\label{ap-function:blackhole}
		& ~ \abspre{\bh} = \emptyword
	\end{align}
\end{definition}

\begin{definition}[λ-ap-ho-term-graph with black holes]\label{def:aplambdahotg-bh}
	As \cref{def:aplambdahotg}, but over signature \sigTGbh and using \cref{abs-pre-function:sigTGbh}.
\end{definition}

\begin{terminology}[λ-ap-ho-term-graph := λ-ap-ho-term-graph with black holes]
	Henceforth, if we speak of λ-ap-ho-term-graphs, we refer to this specific variant with black holes and variable backlinks. All the relevant properties that hold for ordinary λ-ap-ho-term-graphs carry over; accounting for the black holes is generally easy.
\end{terminology}

\begin{figure}
	\translation{λ}{
		\node[draw,shape=transbox](lhs){\recprefixed{\vec{p}}{\abs{x}\L}};
		\draw[<-](lhs.north) to +(0mm,12mm);
	}{
		\node[draw,shape=transbox](rhs){\recprefixed{\vec{p}\prefixcon x^v\fs{}}{\L}};
		\ltgnode[node distance=3mm,above=of rhs.north]{abs}{λ}; \addPos{abs}{v};
		\addPrefix{abs}{\vs{\vec{p}}};
		\draw[<-](abs.north) to +(0mm,4mm);
		\draw[->](abs) to (rhs.north);
	}
	\hfill
	\translation{@}{
		\node[draw,shape=transbox](lhs){\recprefixed{\vec{p}}{\app{\L_0}{\L_1}}};
		\draw[<-](lhs.north) to +(0,11mm);
	}{%
		\ltgnode{app}{@};
		\node[node distance=4mm,below=of app](middle){};
		\node[node distance=1mm,draw,shape=transbox,left=of middle](l){\recprefixed{\vec{p}}{\L_0}};
		\node[node distance=1mm,draw,shape=transbox,right=of middle](r){\recprefixed{\vec{p}}{\L_1}};
		\addPrefix{app}{\vs{\vec{p}}};
		\draw[->](app) to (l.north);
		\draw[->](app) to (r.north);
		\draw[<-](app.north) to +(0,4mm);
	}
	\\[4ex]
	\translation{\f}{
		\node[draw,shape=transbox](lhs){\recprefixed{\vec{p}\prefixcon x^\v\fs{\dots,f^{\w},\dots}}{f}};
		\draw[<-](lhs.north) to +(0,5mm);
	}{
		\ltgnode[densely dashed]{indir}{\indir};
		\addPos{indir}{w};
		\addPrefix{indir}{\vs{\vec{p}}\prefixcon v};
		\draw[<-](indir.north) to +(0,5mm);
	}
	\\[4ex]
	\translation{\0}{
		\node[draw,shape=transbox](lhs){\recprefixed{\stdprefix}{x_n}};
		\draw[<-](lhs.north) to +(0,5mm);
	}{
		\ltgnode{var}{\0}; \draw[<-](var.north) to +(0,5mm);
		\addPrefix{var}{\v_1 \dots \v_n};
		\node[node distance=4mm,right=of var](left){};
		\ltgnode[node distance=3mm,above=of left,densely dashed]{abs}{λ}; \addPos{abs}{v_n};
		\draw[->](var) -| (abs);
	}
	\\[4ex]
	\translationif{\S}{
		\node[draw,shape=transbox](lhs){\recprefixed{\vec{p}\prefixcon x^\v\fs{f_1^{\v_1},\dots,f_n^{\v_n}}}\L};
		\draw[<-](lhs.north) to +(0,5mm);
	}{
		\node[draw,shape=transbox](rhs){\recprefixed{\vec{p}}\L};
		\draw[<-](rhs.north) to +(0,5mm);
	}
	{\eqAnnotation{~~~ \m{x ∉ \freevars{\L}}~~~~}}
	{\eqAnnotation{\m{f_i ∉ \freevars{\L}}}}
	\\[4.5ex]
	\translationv{\Let}{
		\node[draw,shape=transbox](lhs){\recprefixed{\stdprefix}{\letin{f_1=\L_1,\dots,f_k=\L_k}{\L_0}}};
		\draw[<-](lhs.north) to +(0,5mm);
	}{
		\node[draw,shape=transbox](in){\recprefixed{x_0^{\v_0}\fs{\vec{f}'_0} \dots x_n^{\v_n}\fs{\vec{f}'_n}}{\L_0}};
		\node[node distance=13mm,below=of in](bi){\dots};
		\node[node distance=9mm,left=of bi,draw,shape=transbox](b1){\recprefixed{x_0^{\v_0}\fs{\vec{f}'_0} \dots x_{l_1}^{\v_{l_1}}\fs{\vec{f}'_{l_1}}}{\L_1}};
		\node[node distance=9mm,right=of bi,draw,shape=transbox](bn){\recprefixed{x_0^{\v_0}\fs{\vec{f}'_0} \dots x_{l_k}^{v_{l_{\!k}}}\fs{\vec{f}'_{l_k}}}{\L_k}};
		\draw[<-](in.north) to +(0,5mm);
		\ltgnode[node distance=7mm,above=of b1.north,xshift=-5mm]{b1indir}{\indir};
		\addPrefix{b1indir}{\v_1 \dots \v_{l_{\!1}}};
		\addPos{b1indir}{w_1};
		\draw[->](b1indir) to +(0,-9.5mm);
		\ltgnode[node distance=7mm,above=of bn.north,xshift=5mm]{bnindir}{\indir};
		\addPrefix{bnindir}{\v_1 \dots \v_{l_{\!n}}};
		\addPos{bnindir}{w_k};
		\draw[->](bnindir) to +(0,-9.5mm);
		\node[node distance=5mm,above=of bi]{\m{\vec{f}'_i = \vec{f}_i ∪ \setcompr{f_j^{w_j}}{l_j = i, ~ j ∈ \set{1,\dots,k}}}};
		\node[node distance=4mm,below=of bi]{
			\multilinebox{
				\m{l_1,\dots,l_k ∈ \set{0,\dots,n}} such that \\
				~~~~\m{∀ i ∈ \set{1,\dots,k}~ \freevars{L_i} \subseteq \set{x_0,\dots,x_{l_i}} ∪ \vec{f}'_0 ∪ \dots ∪ \vec{f}'_{l_i}}
			}
		};
	}
	\caption{Translation rules \rulestranslambdaletreccaltolhotgs for interpreting \lambdaletrec-terms as λ-ap-ho-term-graphs. See \cref{sec:lhotgs:subsec:trans} for explanations.}
	\label{fig:def:graphsem:lhotgs}
\end{figure}

\begin{para}[translating \lambdaletrec-terms into λ-ap-ho-term-graphs]
	In order to interpret a \lambdaletrec-term \L as a λ-ap-ho-term-graph, the translation rules \rulestranslambdaletreccaltolhotgs from \cref{fig:def:graphsem:lhotgs} are applied to a `translation box' \adjustbox{fbox={\fboxrule} 1pt 0.5pt}{\recprefixed{\starfs{}}{\L}}. It contains \L furnished with a prefix consisting of a dummy variable \m{*} annotated with an empty set \set{} of function variables. The translation process proceeds by induction on the syntactical structure \lambdaletrec-expression. Ultimately, a term graph \atg over \sigTGbh is produced, together with a correct abstraction-prefix function for \atg.
	\par For reading the rules \rulestranslambdaletreccaltolhotgs in \cref{fig:def:graphsem:lhotgs} correctly, take notice of the explanations below. For illustration of their application, please refer to \cref{app:translation} where several \lambdaletrec-terms are translated into λ-ap-ho-term-graphs.
	\begin{itemize}
		\item A translation box \adjustbox{fbox={\fboxrule} 1pt 0.5pt}{\recprefixed{\vec{p}}{\L}} contains a prefixed, partially decomposed \lambdaletrec-term \L. The prefix contains a vector \vec{p} of annotated λ-abstractions that have already been translated and whose scope typically extends into \L. Every variable in the prefix is annotated with a set of function variables that are defined at its level. There is a special dummy variable \m{*} as the very first entry of the prefix that carries function variables for top-level function definitions, i.e.\ definitions that do not reside under any enclosing λ-abstraction. The λ-rule strips off an abstraction from the body of the expression, and pushes the abstraction variable into the prefix, which initially contains an empty set of function variables.
		\item Names of abstraction vertices are indicated to the right, and abstraction-prefixes to the left of the created vertices. In the following example the λ-abstraction vertex \m{v} has the abstraction-prefix \vec{p}:
			\vspace{-2ex}\begin{center}\begin{tikzpicture}
				\ltgnode{x}{λ};
				\addPos{x}{v};
				\addPrefix{x}{\vec{p}};
			\end{tikzpicture}\end{center}\vspace{-2ex}
		\item In order to refer to the vertices in the prefix we use the following notation: \m{\vs{\vec{\apre}} = \v_1\,\dots\,\v_n} given that \m{\vec{\apre} = \starfs{\vec{f}_0}\prefixcon x_1^{\v_1}\fs{\vec{f}_1}\prefixcon \dots\prefixcon x_n^{\v_n}\fs{\vec{f}_n}}.
		\item Vertices drawn with dashed lines have been created earlier during the translation, and are referenced by new edges in the current translation step.
		\item \freevars{\L} is the set of free variables in \L.
		\item The \Let-rule for translating \Let-expressions creates a box for the body as well as for each of its function definitions. For each function definition an \emph{indirection vertex} is created. These vertices guarantee the well-definedness of the process when it translates meaningless bindings such as \m{f = f}, or \m{g = h \eqsep h = g}, which would otherwise give rise to loops without vertices; the result would not be a term graph. Indirection vertices are eliminated by an erasure process at the end: Every indirection vertex that does not point to itself is removed, redirecting all incoming edges to its successor vertex. Finally every loop on a single indirection vertex is replaced by a \emph{black hole} vertex with an empty abstraction prefix that represents a meaningless binding.
		\item The \Let-rule is non-deterministic as there is some freedom on choosing the prefix-lengths used for the translation of each function definition. Say, a function \f does not use the rightmost variable \x in the current abstraction prefix. Then this freedom allows the translation to either remove \x from the prefix within the translation of \f's definition, or alternatively at every use site of \f outside of \f's translation. This freedom is limited by the scoping condition at the bottom of the rule: function definitions may only depend on variables and functions that occur in their respective prefix. In this context also note, that the choice of the prefix-lengths used for some function \f also determines the position of \f within the prefixes used in the translation of the other functions (and the body of the \Let-expression).
	\end{itemize}
\end{para}

\begin{definition}[\rulestranslambdaletreccaltolhotgs-generated term graphs]\label{def:rulestranslambdaletreccaltolhotgsgenerated}
	We say that a term graph \atg over \sigTGbh and an abstraction-prefix function \abspre{} is \emph{\rulestranslambdaletreccaltolhotgs-generated from} a \lambdaletrec-term \L if \atg and \abspre{} are obtained by applying the rules \rulestranslambdaletreccaltolhotgs from \cref{fig:def:graphsem:lhotgs} to \adjustbox{fbox={\fboxrule} 1pt 0.5pt}{\recprefixed{\starfs{}}{\L}}.
\end{definition}

\begin{figure}
	\proofsystem{
		\begin{bprooftree}
			\axiom{\recprefixed{\vec{p}\prefixcon x^\v\fs{}}{\L}}
			\infLabel{λ}
			\unaryInf{\recprefixed{\vec{p}}{\abs{x}{\L}}}
		\end{bprooftree}
		\hsep
		\begin{bprooftree}
			\axiom{\recprefixed{\vec{p}}{\L_0}}
			\axiom{\recprefixed{\vec{p}}{\L_1}}
			\infLabel{@}
			\binaryInf{\recprefixed{\vec{p}}{\app{\L_0}{\L_1}}}
		\end{bprooftree}
		\\
		\begin{bprooftree}
			\emptyAxiom
			\infLabel{\rec}
			\unaryInf{\recprefixed{\vec{p}\prefixcon x^v\fs{\dots,f^{w},\dots}}{f}}
		\end{bprooftree}
		\hsep
		\begin{bprooftree}
			\emptyAxiom
			\infLabel{\0}
			\unaryInf{\recprefixed{\stdprefix}{x_n}}
		\end{bprooftree}
		\\
		\begin{bprooftree}
			\axiom{\recprefixed{\vec{p}}{\L}}
			\infLabel{\S~~\sideCondition{if \m{x ∉ \freevars{\L}} and \m{f_i ∉ \freevars{\L}}}}
			\unaryInf{\recprefixed{\vec{p}\prefixcon x^\v\fs{f_1^{\v_1},\dots,f_n^{\v_n}}}{\L}}
		\end{bprooftree}
		\\
		\begin{bprooftree}
			\axiom{\recprefixed{x_0^{\v_0}\fs{\vec{f}'_0} \dots x_{l_0}^{\v_{l_0}}\fs{\vec{f}'_{l_0}}}{\L_0}}
			\insertBetweenHyps{\dots}
			\axiom{\recprefixed{x_0^{\v_0}\fs{\vec{f}'_0} \dots x_{l_k}^{\v_{l_k}}\fs{\vec{f}'_{l_k}}}{\L_k}}
			\infLabel{\Let}
			\binaryInf{\recprefixed{\stdprefix}{\letin{f_1,\dots,f_k}{\L_0}}}
			\noLine
			\unaryInf{\text{with \m{l_0 = n} and \m{l_1,\dots,l_k}, \m{\vec{f}'_0,\dots,\vec{f}'_n} as in rule \Let in \cref{fig:def:graphsem:lhotgs}}}
		\end{bprooftree}
	}
	\caption{Alternative formulation as inference rules of the translation rules in \cref{fig:def:graphsem:lhotgs} for the interpretation of \lambdaletrec-terms as λ-ap-ho-term-graphs.}
	\label{fig:trans-lambdaletreccal-lhotgs-proof-system}
\end{figure}

\begin{para}[inference rule formulation of \rulestranslambdaletreccaltolhotgs]
	See also \cref{fig:trans-lambdaletreccal-lhotgs-proof-system} for inference rules that correspond to the deconstruction of prefixed terms in \rulestranslambdaletreccaltolhotgs.
\end{para}

\begin{proposition}[λ-ap-ho-term-graph translation of \lambdaletrec-terms]\label{prop:trans_lhotgs_correct}
	Let \L be a \lambdaletrec-term. Suppose that a term graph \atg over \sigTGbh, and an abstraction-prefix function \abspre{} are \rulestranslambdaletreccaltolhotgs-generated from \L. Then \abspre{} is a correct abstraction-prefix function for \atg, and consequently, \atg and \abspre{} together form a λ-ap-ho-term-graph.
\end{proposition}

\begin{para}[non-determinism in \rulestranslambdaletreccaltolhotgs]\label{non-determinism-in-R}
	There are two sources of non-determinism in this translation: The \S-rule for shortening prefixes can be applicable at the same time as other rules. And the \Let-rule does not fix the lengths \m{l_1,\dots,l_k} of the abstraction prefixes used in the translations of the function definitions of the \Let-expression. Neither kind of non-determinism affects the underlying term graph that is produced, but induces different abstraction-prefix functions, and thus different λ-ap-ho-term-graphs.
\end{para}

\subsection{Interpretation as eager-scope λ-ap-ho-term-graphs}\label{sec:lhotgs:subsec:eagscope}

\begin{para}[eager-scope closure induces a higher degree of sharing]
	Of the different translations due to \cref{non-determinism-in-R} we are most interested in the one using the shortest possible abstraction prefixes, thus the translation yielding eager-scope λ-ap-ho-term-graphs (\cref{def:eager-scope:aphotgs}). The reason for this choice is illustrated in \cref{fig:eager_more_sharing}: eager-scope closure allows for more sharing.
\end{para}

\begin{figure}
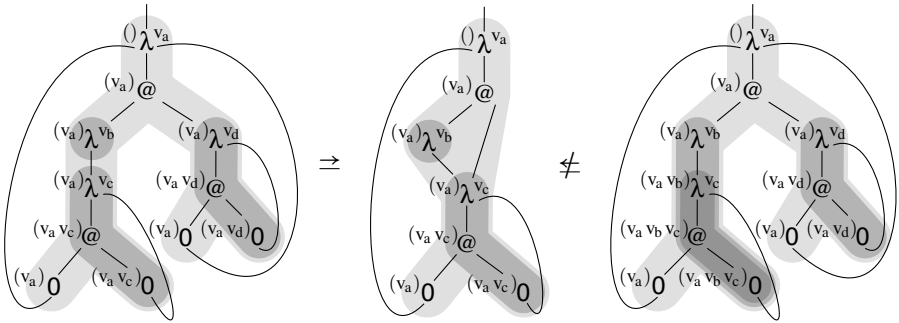

	\begin{hspread}
		\vcentered{\fig{eager-more-sharing-eager}}
		& \vcentered{\funbisim} &
		\vcentered{\fig{eager-more-sharing-collapse}}\hspace{-1ex}
		& \vcentered{\not\convfunbisim} &
		\vcentered{\fig{eager-more-sharing-lazy}}
	\end{hspread}
	\caption{Translation of \m{\abs{a}{\app{(\abs{b}{\abs{c}{\app{a}{c}}})}{(\abs{d}{\app{a}{d}})}}} with eager scope-closure (left), and with lazy scope-closure (right). While on the left four vertices can be shared, on the right only a single variable occurrence can be shared.}
	\label{fig:eager_more_sharing}
\end{figure}


Also, we will call a \emph{translation process} `eager-scope' if it resolves the non-determinism in \rulestranslambdaletreccaltolhotgs in such a way that it always yields eager-scope λ-ap-ho-term-graphs. In order to obtain an eager-scope\ translation we have to consider the following aspects.

\begin{para}[garbage removal]\label{garbage-removal}
	In the presence of garbage -- unused function bindings -- a translation process cannot be eager-scope. Consider the term \abs{x}{\abs{y}{\letin{f=x}{y}}}. The variable \x occurs solely in the unused binding \m{f=x}, which prevents the application of the \S-rule, and hence the closure of the scope of \abs{x}{}, directly below \abs{x}{}. Therefore we henceforth assume that \emph{all unused function bindings are removed} prior to applying the rules \rulestranslambdaletreccaltolhotgs.
\end{para}


\begin{para}[short enough prefix lengths in the \Let-rule]\label{short-enough-prefixes}
	For obtaining an eager-scope translation we will stipulate that the \S-rule is applied eagerly, i.e.\ it is given precedence over the other rules. This is clearly necessary for keeping the abstraction prefixes minimal. But how do we choose the prefix lengths \m{l_1,\dots,l_k} in the \Let-rule? The prefix lengths \m{l_i} determine at which position a binding \m{f_i=\L_i} is inserted into the abstraction prefixes. Therefore \m{l_i} may not be chosen too short; otherwise a function \f depending on a function \g may end up to the right of g, and hence may be removed from the prefix by the \S-rule prematurely, preventing completion of the translation. Yet simply choosing \m{l_i = n} may prevent scopes from being minimal. For example, when translating the term \m{\abs{a}{\abs{b}{\letin{f=a}{\app{\app{\app{a}{a}}{(\app{f}{a})}}{b}}}}}, it is crucial to allow shorter prefixes for the binding than for the body. This is illustrated in \cref{fig-rigid_let_rule} where the graph on the left does not have eager scope-closure even if the \S-rule is applied eagerly. Consequently the opportunity for sharing the lower application vertices is lost.
\end{para}

\begin{figure}
	\begin{hspread}
		\fig{let-prefixes-same} & \fig{let-prefixes-shorter}
	\end{hspread}
	\caption{Translation of \m{\abs{a}{\abs{b}{\letin{f=a}{\app{\app{\app{a}{a}}{(\app{f}{a})}}{b}}}}} with equal (left) and with minimal prefix lengths (right) in the \Let-rule. See also \cref{showcase:rigid_let_rule}.}
	\label{fig-rigid_let_rule}
\end{figure}

\begin{para}[required variable analysis]
	For choosing the prefixes in the \Let-rule correctly, the translation process must know for each function binding which λ-variable are `required' on the right-hand side of its definition. For this we use an analysis obtaining the required variables for positions in a \lambdaletrec-term as employed by algorithms for λ-lifting \cite{john:1985,danv:schu:2004}. The term `required variables' was coined by Morazán and Schultz \cite{mora:schu:2008}. A λ-variable \x is called \emph{required at a position \p} in a \lambdaletrec-term \L if \x is bound by an abstraction above \p, and has a free occurrence in the complete unfolding of \L below \p (also function variables from above \p are unfolded). The required variables at position \p in \L can be computed as those λ-variable with free occurrences that are reachable from \p by a downwards traversal with the stipulations: on encountering a \Let-expression the body is entered; when encountering a function variable the traversal continues at the right-hand side of the corresponding function binding (even if it is defined above \p).
\end{para}

With the result of the required variable analysis at hand, we now define properties of the translation process that can guarantee that the resulting λ-ap-ho-term-graph is eager-scope.

\begin{definition}[eager-scope and minimal-prefix generated]\label{def:eagscope:minprefix:generated}
	\par Let \L be a \lambdaletrec-term, and let \ahotg be a λ-ap-ho-term-graph.
	\par We say that \ahotg is \emph{eager-scope} \rulestranslambdaletreccaltolhotgs-generated from \L if \ahotg is \rulestranslambdaletreccaltolhotgs-generated from \L by a translation process with the following property: for every translation box reached during the process with label \prefixed{\vec{p}\prefixcon x^\v[f]}{P}, where \P is a subterm of \L at position \q, it holds that if \x is not a required variable at \q in \L, then in the next translation step performed to this box either one of the rules \f or \Let is applied, or the prefix is shortened by the \S-rule.
	\par We say that \ahotg is \rulestranslambdaletreccaltolhotgs-generated \emph{with minimal prefixes} from \L if \ahotg is \rulestranslambdaletreccaltolhotgs-generated from \L by a translation process in which minimal prefix lengths are achieved by giving applications of the \S-rule precedence over applications of all other rules, and by always choosing prefixes minimally in applications of the \Let-rule.
\end{definition}

\begin{proposition}\label{prop:eagscope:generated:minimal:prefixes}
	Let \ahotg be a λ-ap-ho-term-graph that is \rulestranslambdaletreccaltolhotgs-generated from a garbage-free \lambdaletrec-term \L. The following statements hold:
	\begin{enumerate}[(i)]
		\item\label{prop:eagscope:generated:minimal:prefixes:i}
			If \ahotg is eager-scope \rulestranslambdaletreccaltolhotgs-generated from \L, then \ahotg is eager-scope.
		\item\label{prop:eagscope:generated:minimal:prefixes:ii}
			If \ahotg is \rulestranslambdaletreccaltolhotgs-generated with minimal prefixes from \L, then \ahotg is eager-scope \rulestranslambdaletreccaltolhotgs-generated from \L, hence by \cref{prop:eagscope:generated:minimal:prefixes:i} \ahotg is eager-scope.
	\end{enumerate}
\end{proposition}

\begin{definition}[higher-order term graph semantics]\label{def:graphsem:lhotgs}
	The semantics \graphsemC{\classlhotgs}{} of \lambdaletrec-terms as λ-ap-ho-term-graphs is defined as \m{\graphsemC{\classlhotgs}{} : \Ter{\lambdaletreccal} → \classlhotgs}, \m{\L ↦ \graphsemC{\classlhotgs}{\L} :=\,} the λ-ap-ho-term-graph that is \rulestranslambdaletreccaltolhotgs-generated with minimal prefixes from a garbage-free version of \L.
\end{definition}

\begin{proposition}\label{prop:trans_eager}
	For every \lambdaletrec-term \L, \graphsemC{\classlhotgs}{\L} is eager-scope.
\end{proposition}

In preparation of proving \cref{methods:properties:correctness} in \cref{sec:trans:ltgs}, we establish that the semantics \graphsemC{\classlhotgs}{} is correct with respect to the unfolding semantics of \lambdaletreccal.

\begin{theorem}[correctness of \graphsemC{\classlhotgs}{}]\label{thm:graphrep:classlhotgs}
	\m{\unfsem{\L_1} = \unfsem{\L_2}} if and only if \m{\graphsemC{\classlhotgs}{\L_1} \bisim \graphsemC{\classlhotgs}{\L_2}}, for all \lambdaletrec-terms \m{\L_1} and \m{\L_2}.
\end{theorem}

\begin{proof}[Sketch of Proof]
	\newcommand\classlhotreetgs{\m{\mathcal{H}_T}}
	\newcommand\alhotreetg{\m{T}}
	\newfunction\Tree{Tree}
	For the proof we introduce a class of λ-ap-ho-term-graphs that have tree form, i.e.\ they only contain variable backlinks but no other backlinks. We denote that class by \m{\classlhotreetgs\subset\classlhotgs}. Every \m{\ahotg ∈ \classlhotgs} has a unique `tree unfolding' \m{\Tree{\ahotg} ∈ \classlhotreetgs}. We make use of the following statements. For all \m{\L,\L_1,\L_2 ∈ \Ter{\lambdaletreccal}}, \m{\M,\M_1,\M_2 ∈ \Ter{\lambdabhcal}}, \m{\ahotg,\ahotg_1,\ahotg_2 ∈ \classlhotgs}, and \m{\alhotreetg,\alhotreetg_1,\alhotreetg_2 ∈ \classlhotreetgs} it can be shown that:
	\begin{gather}
		\L_1 \m{\red_\unfold} \L_2 ~⇒~ \graphsemC{\classlhotgs}{\L_1} \convfunbisim \graphsemC{\classlhotgs}{\L_2}
		\label{eq:correctness:llter:unfoldred:llter:then:graphsem:funbisim:graphsem}
		\\
		\L \m{\infred_\unfold} \M ~ \text{(hence \m{\unfsem{\L} = \M})} ~⇒~ \graphsemC{\classlhotgs}{\L} \convfunbisim \graphsemC{\classlhotgs}{\M}
		\label{eq:correctness:graphsem:unfsem:llter:funbisim:graphsem:llter}
		\\
		\graphsemC{\classlhotgs}{\M} ∈ \classlhotreetgs
		\label{eq:correctness:graphsem:aiter:in:lhotreetgs}
		\\
		\graphsemC{\classlhotgs}{\M_1} \iso \graphsemC{\classlhotgs}{\M_2} ~⇒~ \M_1 = \M_2
		\label{eq:correctness:graphsem:iter:equality}
		\\
		\ahotg \convfunbisim \Tree{\ahotg}
		\label{eq:correctness:treeunfolding:funbisim}
		\\
		\alhotreetg_1 \bisim \alhotreetg_2 ~⇒~ \alhotreetg_1 \iso \alhotreetg_2
		\label{eq:correctness:tree:bisim:tree:then:tree:iso:tree}
		\\
		\ahotg_1 \bisim \ahotg_2 ~⇒~ \Tree{\ahotg_1} \iso \Tree{\ahotg_2}
		\label{eq:correctness:lhotg:bisim:lhotg:then:treeunfolding:iso:treeunfolding}
	\end{gather}
	We can use \cref{eq:correctness:llter:unfoldred:llter:then:graphsem:funbisim:graphsem} for proving \cref{eq:correctness:graphsem:unfsem:llter:funbisim:graphsem:llter}, and we can use \cref{eq:correctness:treeunfolding:funbisim} with \cref{eq:correctness:tree:bisim:tree:then:tree:iso:tree} for proving \cref{eq:correctness:lhotg:bisim:lhotg:then:treeunfolding:iso:treeunfolding}. Now for proving the theorem, let \m{\L_1} and \m{\L_2} be arbitrary \lambdaletrec-terms.
	\par For ``⇒'', suppose \m{\unfsem{\L_1} = \unfsem{\L_2}}. Let \M be the infinite unfolding of \m{\L_1} and \m{\L_2}, i.e., \m{\graphsemC{\classlhotgs}{\L_1} = \M = \graphsemC{\classlhotgs}{\L_2}}. Then by \cref{eq:correctness:graphsem:unfsem:llter:funbisim:graphsem:llter} it follows \m{\graphsemC{\classlhotgs}{\L_1} \convfunbisim \graphsemC{\classlhotgs}{\M} \funbisim \graphsemC{\classlhotgs}{\L_2}}, and hence \m{\graphsemC{\classlhotgs}{\L_1} \bisim \graphsemC{\classlhotgs}{\L_2}}.
	\par For ``⇐'', suppose \m{\graphsemC{\classlhotgs}{\L_1} \bisim \graphsemC{\classlhotgs}{\L_2}}. Then by \cref{eq:correctness:lhotg:bisim:lhotg:then:treeunfolding:iso:treeunfolding} it follows that \m{\Tree{\graphsemC{\classlhotgs}{\L_1}} \iso \Tree{\graphsemC{\classlhotgs}{\L_2}}}. Let \m{\M_1,\M_2 ∈ \Ter{\lambdabhcal}} be the infinite unfoldings of \m{\L_1} and \m{\L_2}, i.e.\ \m{\M_1 = \unfsem{\L_1}}, and \m{\M_2 = \unfsem{\L_2}}. Then \cref{eq:correctness:graphsem:unfsem:llter:funbisim:graphsem:llter} together with the assumption entails \m{\graphsemC{\classlhotgs}{\M_1} \bisim \graphsemC{\classlhotgs}{\M_2}}. Since \m{\graphsemC{\classlhotgs}{\M_1}, \graphsemC{\classlhotgs}{\M_2} ∈ \classlhotreetgs} by \cref{eq:correctness:graphsem:aiter:in:lhotreetgs}, it follows by \cref{eq:correctness:tree:bisim:tree:then:tree:iso:tree} that \m{\graphsemC{\classlhotgs}{\M_1} \iso \graphsemC{\classlhotgs}{\M_2}}. Finally, by using \cref{eq:correctness:graphsem:iter:equality} we obtain \m{\M_1 = \M_2}, and hence \m{\unfsem{\L_1} = \M_1 = \M_2 = \unfsem{\L_2}}.
\end{proof}


\section{Interpretion of \lambdaletrec-terms as λ-term-graphs}\label{sec:trans:ltgs}

\begin{para}[overview]
	We have found a way to model sharing in \lambdaletreccal by interpreting \lambdaletrec-terms as eager-scope λ-ap-ho-term-graphs. We have also identified eager-scope λ-term-graphs as a class of first-order term graphs that faithfully implements (functional) bisimulation on λ-ap-ho-term-graphs. Since we have a higher-order term graph semantics for interpreting \lambdaletrec-terms as λ-ap-ho-term-graphs, as well as translation from λ-ap-ho-term-graphs to λ-term-graphs and back, we seem to be able to turn our attention to the readback function as the last required ingredient of the methods described in \cref{pipeline}. However, we first need to add black holes to the formalisation of λ-term-graphs as done for λ-ap-ho-term-graphs. And then we also need to deal with a few intricacies of λ-term-graphs with regards to readback.
\end{para}

Analogously to \cref{laphotgs-blackholes} we extend the signature \sigTGij{1}{2} and the correctness conditions in \cref{ltg:ap-function} by a black-hole case.

\begin{definition}[signature for λ-term-graphs with black holes]\label{sigTGSbh}
	By \sigTGSbh we denote the signature \sigTGij{1}{2} extended by a black-hole symbol \bh with arity \m{0}, thus \m{\sigTGSbh = \set{@, λ, \0, \S, \bh}} with \m{\arity{@} = 2}, \m{\arity{λ} = \arity{\0} = 1}, \m{\arity{\S} = 2}, and \m{\arity{\bh} = 0}.
\end{definition}

\begin{definition}[abstraction-prefix function for \sigTGSbh-term-graphs]\label{abs-pre-function:sigTGSbh}
	As \cref{ltg:ap-function} but over signature \sigTGSbh and with the following condition added:
	\begin{align}
		\tag{black hole}\label{ltg:ap-function:blackhole}
		& ~ \abspre{\bh} = \emptyword
	\end{align}
\end{definition}

\begin{definition}[λ-term-graph with black holes]\label{def:ltg}
	As \cref{def:classltgs}, but over signature \sigTGSbh and using \cref{abs-pre-function:sigTGSbh}.
\end{definition}

\begin{terminology}[λ-term-graph := λ-term-graph with black holes]
	Henceforth, if we speak of λ-term-graphs, we refer to this specific variant with black holes and variable and scope-delimiter vertices with backlinks. All the relevant properties that hold for ordinary λ-term-graphs carry over.
\end{terminology}

\begin{remark}[relaxed black-hole condition]
	The condition from \cref{abs-pre-function:sigTGSbh} could also be relaxed to say \m{\abspre{\bh} = w} where \m{w} is an arbitrary word of vertices. In that case instead of one single shared function definition of the form \m{f=f} one would obtain multiple such definitions, but at most one for each scope.
\end{remark}

\begin{para}[\lhotgstoltgs{}]
	\lhotgstoltgs{}, the function that maps λ-ap-ho-term-graphs to their corresponding λ-term-graphs is defined as \aphotgstoltgsij{1}{2}{} from \cref{prop:mappings:aphotgs:to:ltgs} but amended to also translate black holes (in the obvious way).
\end{para}

\begin{para}[two interpretations as λ-term-graphs]
	We will consider in fact two interpretations of \lambdaletrec-terms as λ-term-graphs: first we define \graphsemCmin{\classltgs}{} as the composition of \graphsemC{\classlhotgs}{} and \lhotgstoltgs{}; then we define the semantics \graphsemC{\classltgs}{} with more fine-grained \S-sharing, which is necessary for defining a readback with \cref{methods:properties:readback}.
	\par By composing the interpretation \lhotgstoltgs{} of λ-ap-ho-term-graphs as λ-term-graphs with the λ-ap-ho-term-graph semantics \graphsemC{\classlhotgs}{}, a semantics of \lambdaletrec-terms as λ-term-graphs is obtained. There is, however, a more direct way to define this semantics: by using an adaptation of the translation rules \rulestranslambdaletreccaltolhotgs in \cref{fig:def:graphsem:lhotgs}, on which \graphsemC{\classlhotgs}{} is based. For this, let \rulestranslambdaletreccaltoltgs be the result of replacing the rule \S in \rulestranslambdaletreccaltolhotgs by the version in \cref{fig:def:graphsem:ltgs}. While applications of this variant of the \S-rule also shorten the abstraction-prefix, they additionally produce a delimiter vertex.
	\par Here, at the end of the translation process, every loop on an indirection vertex with a prefix of length \n is replaced by a chain of \n \S-vertices followed by a black hole vertex. Note that, while the system \rulestranslambdaletreccaltoltgs inherits all of the non-determinism of \rulestranslambdaletreccaltolhotgs, the possible degrees of freedom have additional impact on the result, because now they also determine the precise degree of \S-vertex sharing.
	By analogous stipulations as in \cref{def:eagscope:minprefix:generated} we define the conditions under which a λ-term-graph is called \emph{eager-scope} \rulestranslambdaletreccaltoltgs-generated, or \rulestranslambdaletreccaltoltgs-generated \emph{with minimal prefixes}, from a \lambdaletrec-term. For these notions, statements entirely analogous to \cref{prop:eagscope:generated:minimal:prefixes} hold.
\end{para}

\begin{figure}
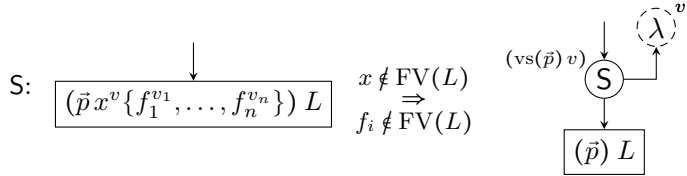

	\translationif{\S}{
		\node[draw,shape=transbox](lhs){\recprefixed{\vec{p}\prefixcon x^v\fs{f_1^{v_1},\dots,f_n^{v_n}}}\L};
		\draw[<-](lhs.north) to +(0,5mm);
	}{
		\node[draw,shape=transbox](rhs){\recprefixed{\vec{p}}\L};
		\ltgnode[node distance=4mm,above=of rhs.north]{succ}{\S};
		\addPrefix{succ}{\vs{\vec{p}}\prefixcon \v};
		\draw[->](succ) to (rhs.north);
		\draw[<-](succ.north) to +(0,5mm);
		\node[node distance=3mm,right=of succ](right){};
		\ltgnode[node distance=3mm,above=of right,densely dashed]{abs}{λ}; \addPos{abs}{v};
		\addPos{abs}{v};
		\draw[->](succ) -| (abs);
	}
	{\text{\small \m{x ∉ \freevars{\L}}}}
	{\text{\small \m{f_i ∉ \freevars{\L}}}}
	\caption{Delimiter-vertex producing version of the \S-rule from \cref{fig:def:graphsem:lhotgs}}
	\label{fig:def:graphsem:ltgs}
\end{figure}

\begin{definition}[\graphsemCmin{\classltgs}{}]\label{def:graphsem-min:ltgs}
  	The semantics \graphsemCmin{\classltgs}{} for \lambdaletrec-terms as λ-term-graphs is defined as \m{\graphsemCmin{\classltgs}{} : \Ter{\lambdaletreccal} → \eag{\classltgs}}, \m{\L ↦ \graphsemCmin{\classltgs}{\L} :=\,} the eager-scope term graph that is \rulestranslambdaletreccaltoltgs-generated with minimal prefixes from a garbage-free version \L.
\end{definition}

\begin{para}[no \S-sharing]
	For an example, see \cref{ex:translations} below. In \graphsemCmin{\classltgs}{}, `\noSsh' also indicates that λ-term-graphs obtained via this semantics exhibit minimal (in fact no) sharing (two or more incoming edges) of \S-vertices. This is substantiated by the next proposition, in the light of the fact that \lhotgstoltgs{} does not create any shared \S-vertices.
\end{para}

\begin{proposition}\label{prop:graphsem:ltgs:min:graphsem:ltgs}
  	\m{\graphsemCmin{\classltgs}{} = {\lhotgstoltgs{}} ∘ {\graphsemC{\classlhotgs}{}}}.
\end{proposition}

\begin{para}[\graphsemCmin{\classltgs}{} is partial]
	Hence \graphsemCmin{\classltgs}{} only yields λ-term-graphs without sharing of \S-vertices, and therefore its image cannot be all of \eag{\classltgs}. As a consequence, we cannot hope to define a readback function \readback{} with respect to \graphsemCmin{\classltgs}{} that adheres to \cref{methods:properties:readback}, because that requires that the image of the semantics is \eag{\classltgs} in its entirety.
\end{para}

\begin{para}[eager-scope \rulestranslambdaletreccaltolhotgs-generated with maximal prefixes]
	Therefore we provide an alternative λ-term-graph semantics \graphsemC{\classltgs}{} with \m{\image{\graphsemC{\classltgs}{}} = \eag{\classltgs}}. We achieve this by letting the \Let-binding-structure of the \lambdaletrec-term influence the degree of \S-sharing as much as possible, while still remaining eager-scope.
	\par By `letting the \Let-binding-structure determine the degree of \S-sharing' we mean to eliminate the freedom of choice with respect to closing scopes: one can either place a scope-delimiter vertex at the top of the translation of a function binding in which case it is shared, or right above every use site of that function in which case it is not. See \cref{fig:let-structure} for illustration.
	\begin{figure}
		\begin{hspread}
			\fig{let-structure1} & \fig{let-structure2} & \fig{let-structure3} \\
			\abs{x}{\letin{f=x}{\abs{y}{\app{f}{(\app{f}{y})}}}} & \abs{x}{\abs{y}{\letin{f=x}{\app{f}{(\app{f}{y})}}}} & \abs{x}{\abs{y}{\letin{f=x}{\app{\app{f}{f}}{y}}}}
		\end{hspread}
		\captionsetup{singlelinecheck=off}
		\caption[]{
			In \graphsemC{\classltgs}{} the degree of \S-sharing is determined by the \Let-structure as much as possible. Consider the above \lambdaletrec-terms and their term graph interpretation w.r.t.\ \graphsemC{\classltgs}{}.
			\begin{description}
				\item[left]: \f is bound outside of \y's scope. Therefore it is not \f's responsibility to close \y's scope but the responsibility of \f's two use sites.
				\item[middle]: \f is bound within \y's scope so it is \f's responsibility to close it.
				\item[right]: \f is bound within \y's scope. Therefore it should be \f's responsibility to close \y's scope, but putting the scope delimiter underneath the lower application would violate eager scope-closure.
			\end{description}}
		\label{fig:let-structure}
	\end{figure}
\end{para}

\begin{definition}[eager-scope \rulestranslambdaletreccaltolhotgs-generated with maximal prefixes]
	\par We say that a λ-ap-ho-term-graph \ahotg is \emph{eager-scope \rulestranslambdaletreccaltolhotgs-generated with maximal prefixes} from a \lambdaletrec-term \L if \ahotg is \rulestranslambdaletreccaltolhotgs-generated from \L by a translation process in which in applications of the \Let-rule the prefixes are chosen maximally, but so that the eager-scope property of the process is not compromised. It can be shown that this condition fixes the prefix lengths per application of the \Let-rule.
\end{definition}

\begin{definition}[\graphsemC{\classltgs}{}]\label{def:graphsem:ltgs}
  	The semantics \graphsemC{\classltgs}{} for \lambdaletrec-terms as λ-term-graphs is defined as \m{\graphsemC{\classltgs}{} : \Ter{\lambdaletreccal} → \eag{\classltgs}}, \m{\L ↦ \graphsemC{\classltgs}{\L} :=\,} the λ-term-graph that is eager-scope \rulestranslambdaletreccaltoltgs-generated with maximal prefixes from a garbage-free version of \L.
\end{definition}

\begin{figure}
	\begin{minipage}[b]{3.9cm}
		\begin{flushright}
			\graphsemC{\classltgs}{\L_1}\\
			\m{= \graphsemCmin{\classltgs}{\L_1}
			= \graphsemCmin{\classltgs}{\L_2}}\\
			\m{= \graphsemCmin{\classltgs}{\L_3}
			= \graphsemCmin{\classltgs}{\L'_3}}
			\\[0.9cm]
			\graphsemC{\classltgs}{\L_2}
			\\[1.2cm]
			\m{\graphsemC{\classltgs}{\L_3} = \graphsemC{\classltgs}{\L'_3}}
			\\[5mm]
			\mbox{}
		\end{flushright}
	\end{minipage}
	\hfill
	\fig{good-trans}
	\caption{Translation of the \lambdaletrec-terms from \cref{ex:translations} with the semantics \graphsemCmin{\classltgs}{} and \graphsemC{\classltgs}{}. For legibility some backlinks are merged.}
	\label{fig:translations}
\end{figure}

\begin{proposition}\label{prop:graphsem:ltgs:min:graphsem:ltgs:S}
  \m{\graphsemCmin{\classltgs}{\L} \funbisim^S \graphsemC{\classltgs}{\L}}
  holds for all \lambdaletrec-terms \L.
\end{proposition}

Now due to this, and due to \cref{thm:corr:aphotgs:ltgs}~\cref{thm:corr:aphotgs:ltgs:iii}, the statement of \cref{thm:graphrep:classlhotgs} can be transferred to \classltgs, yielding \cref{methods:properties:correctness} for \graphsemC{\classltgs}{}.


\begin{theorem}\label{thm:graphrep:classltgs}
  For all \lambdaletrec-terms \m{\L_1} and \m{\L_2}
  the following holds:
  \m{\unfsem{\L_1} = \unfsem{\L_2}}
    if and only if
  \m{\graphsemC{\classltgs}{\L_1} \bisim \graphsemC{\classltgs}{\L_2}}.
\end{theorem}

\begin{example}\label{ex:translations}
	Consider the following four \lambdaletrec-terms:
	\[
	\begin{aligned}
  		\L_1 &= \letin{I=\abs{z}{z}}{\abs{x}{\abs{y}{\letin{f=x}{\app{(\app{(\app{y}{I})}{(\app{I}{y})})}{(\app{f}{f})}}}}}\\
  		\L_2 &= \abs{x}{\letin{I=\abs{z}{z}}{\abs{y}{\letin{f=x}{\app{(\app{(\app{y}{I})}{(\app{I}{y})})}{(\app{f}{f})}}}}}\\
  		\L_3 &= \abs{x}{\abs{y}{\letin{I=\abs{z}{z} \eqsep f=x}{\app{(\app{(\app{y}{I})}{(\app{I}{y})})}{(\app{f}{f})}}}}\\
  		\L'_3 &= \abs{x}{\letin{I=\abs{z}{z}}{\abs{y}{\letin{f=x,\ g=I}{\app{(\app{(\app{y}{g})}{(\app{g}{y})})}{(\app{f}{f})}}}}}\\
	\end{aligned}
	\]
	The three possible fillings of the dashed area in \cref{fig:translations}
	depict the translations \graphsemC\classltgs{\L_1}, \graphsemC\classltgs{\L_2}, and \m{\graphsemC\classltgs{\L_3} = \graphsemC\classltgs{\L'_3}}. The translations of the four terms with \graphsemCmin{\classltgs}{} are identical:\\ \m{\graphsemCmin{\classltgs}{\L_1} = \graphsemCmin{\classltgs}{\L_2} = \graphsemCmin{\classltgs}{\L_3} = \graphsemCmin{\classltgs}{\L'_3} = \graphsemC{\classltgs}{\L_1}}.
	\par See \cref{translation:translations} for a fully-worked out stepwise translation of \m{L_2} and also \cref{showcase:translations}.
\end{example}

\section{Readback of λ-term-graphs}\label{sec:readback}

\begin{para}[overview]
	In this section we describe how from a given λ-term-graph \altg a \lambdaletrec-term \L that represents \altg (i.e.\ for which \m{\graphsemC{\classltgs}{\L} = \altg} holds) can be `read back'. For this purpose we define a process based on synthesis rules. It defines a readback function from λ-term-graphs to \lambdaletrec-terms. We illustrate this process by an example, formulate its most important properties, and sketch the proof of \cref{methods:properties:readback}.
\end{para}

\begin{para}[\graphsemC{\classltgs}{} is not invertible]\label{readback-inverse}
	Note that we state as the desired \cref{methods:properties:readback} of the readback \readback{} to be a right-inverse and not a `full' inverse of \graphsemC{\classltgs}{}. This is because \graphsemC{\classltgs}{} is not injective. Consider the \lambdaletrec-terms \abs{x}{\letin{f=x}{\abs{y}{f}}} and \abs{x}{\abs{y}{\letin{f=x}{f}}}. They only differ in the position of the function definition; this information is lost in the translation; \graphsemC{\classltgs}{} maps them to the same term graph. In particular, this kind of relocation of function-bindings (called \emph{let-floating} \cite{grab:roch:letfloating}) preserves the λ-term-graph interpretation.
\end{para}

\begin{para}[approach]
	The idea underlying the definition of the readback procedure is the following: For a given λ-term-graph \altg, a spanning tree \m{T} for \altg (augmented with a dedicated root node) is constructed that severs cycles due to recursive bindings and variable and scope-delimiter backlinks. Now the spanning tree \m{T} facilitates an inductive bottom-up (from the leafs upwards) synthesis process along \m{T}, which labels the edges of \atg with prefixed \lambdaletrec-terms. For this process we use local rules (see \cref{fig:readback:rules}) that synthesise labels for incoming edges of a vertex from the labels of its outgoing edges. Eventually the readback of \altg is obtained as the label for the edge that singles out the root of term graph.
\end{para}

\begin{para}[placement of function definitions]
	In the design of the readback rules, there is some freedom in where to place the function bindings in the synthesised term (\cref{readback-inverse}). 
	Here, function bindings will be put into a \Let-expression placed as high up in the term as possible: a binding arising from the term synthesised for a shared vertex \v is placed in a \Let-expression that is created at the enclosing λ-abstraction of \v (the rightmost vertex in the abstraction-prefix \abspre{\v} of \v).
\end{para}

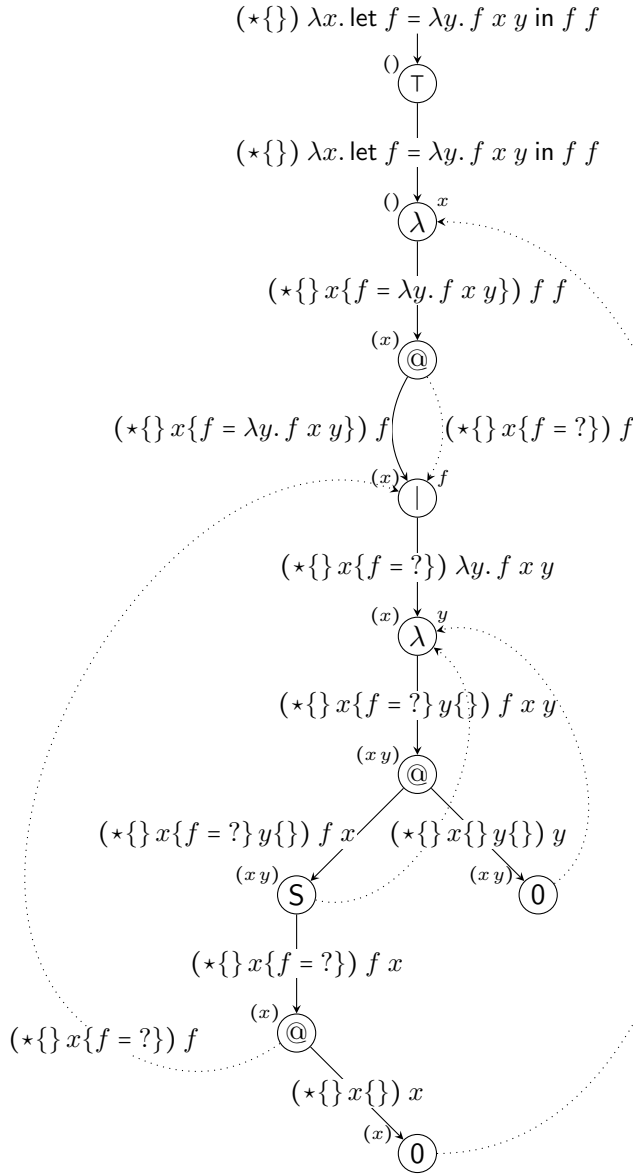
\begin{figure}
	\begin{tikzpicture}[node distance=13mm,inner sep=0.1mm]
		\node(rootrb){\recprefixed{\starfs{}}{\abs{x}{\letin{f=\abs{y}{\app{\app{f}{x}}{y}}}{\app{f}{f}}}}};
		\ltgnode[node distance=4mm,below=of rootrb]{root}{\top}; \draw[->] (rootrb) to (root);
		\ltgnode[below=of root]{absx}{λ};
		\addPos{absx}{~x};
		\rbedge{\recprefixed{\starfs{}}{\abs{x}{\letin{f=\abs{y}{\app{\app{f}{x}}{y}}}{\app{f}{f}}}}}{root}{absx};
		\addPrefix[node distance=2mm]{root}{};
		\ltgnode[below=of absx]{app1}{@};
		\rbedge{\recprefixed{\starfs{}\prefixcon x\fs{f=\abs{y}{\app{\app{f}{x}}{y}}}}{\app{f}{f}}}{absx}{app1};
		\addPrefix[node distance=2mm]{absx}{};
		\ltgnode[below=of app1]{indir}{\indir};
		\rbedge[bend right,left]{\recprefixed{\starfs{}\prefixcon x\fs{f=\abs{y}{\app{\app{f}{x}}{y}}}}{f}}{app1}{indir};
		\addPrefix[node distance=2mm]{app1}{x};
		\addPos{indir}{~f};
		\rbedge[bend left,right,dotted]{\recprefixed{\starfs{}\prefixcon x\fs{f=\nodef}}{f}}{app1}{indir};
		\ltgnode[below=of indir]{absy}{λ};
		\addPos{absy}{~y};
		\rbedge{\recprefixed{\starfs{}\prefixcon x\fs{f=\nodef}}{\abs{y}{\app{\app{f}{x}}{y}}}}{indir}{absy};
		\addPrefix[node distance=2mm]{indir}{x};
		\ltgnode[below=of absy]{app2}{@};
		\rbedge{\recprefixed{\starfs{}\prefixcon x\fs{f=\nodef}\prefixcon y\fs{}}{\app{\app{f}{x}}{y}}}{absy}{app2};
		\addPrefix[node distance=2mm]{absy}{x};
		\node[below=of app2](bapp2){};
		\ltgnode[left=of bapp2]{s}{\S};
		\rbedge[left]{\recprefixed{\starfs{}\prefixcon x\fs{f=\nodef}\prefixcon y\fs{}}{\app{f}{x}}}{app2}{s};
		\addPrefix[node distance=2mm]{app2}{x\prefixcon y};
		\ltgnode[below=of s]{app3}{@};
		\rbedge{\recprefixed{\starfs{}\prefixcon x\fs{f=\nodef}}{\app{f}{x}}}{s}{app3};
		\addPrefix[node distance=2mm]{s}{x\prefixcon y};
		\draw[dotted,->,bend right=80](s) to (absy);
		\node[below=of app3](bapp3){};
		\ltgnode[right=of bapp3]{x}{\0};
		\addPrefix[node distance=3mm]{x}{x};
		\draw[dotted,->,bend right=90](x) to (absx);
		\rbedge{\recprefixed{\starfs{}\prefixcon x\fs{}}{x}}{app3}{x};
		\addPrefix[node distance=2mm]{app3}{x};
		\ltgnode[right=of bapp2]{y}{\0};
		\addPrefix[node distance=3mm]{y}{x\prefixcon y};
		\draw[dotted,->,bend right=80](y) to (absy);
		\rbedge{\recprefixed{\starfs{}\prefixcon x\fs{}\prefixcon y\fs{}}{y}}{app2}{y};
		\node[left=of app3,yshift=-1mm,xshift=3mm](lapp3){\recprefixed{\starfs{}\prefixcon x\fs{f=\nodef}}{f}};
		\draw[dotted,bend left](app3) to (lapp3);
		\draw[dotted,->,bend left=80](lapp3) to (indir);
	\end{tikzpicture}
	\caption{Example of the readback synthesis from a λ-term-graph.}
	\label{fig:ex:readback}
\end{figure}

\begin{definition}[readback of λ-term-graphs]\label{def:readback}
	Let \m{\altg ∈ \eag{\classltgs}} be an eager-scope λ-term-graph. The process of computing the readback of \altg (a \lambdaletrec-term) consists of the following five steps, starting on \altg:
	\begin{enumerate}[(Rb-1)]
		\item\label[Rb]{readback:step:absprefix} Determine the abstraction-prefix function \abspre{} for \altg by performing a traversal over \altg, and associate with every vertex \v of \altg its abstraction-prefix \abspre{\v}.
		\item\label[Rb]{readback:step:top} Add a new vertex on top with label \top, arity~1, and empty abstraction prefix. Let \m{\altg'} be the resulting term graph, and \m{\abspre{}'} its abstraction-prefix function.
		\item\label[Rb]{readback:step:indirection} Introduce indirection vertices to organise sharing: For every vertex \v of \m{\altg'} with two or more incoming non-variable-backlink edges, add an indirection vertex \m{\v_0}, redirect the incoming edges of \v that are not variable backlinks to \m{\v_0}, and direct the outgoing edge from \m{\v_0} to \v. In the resulting term graph \m{\altg''} only indirection vertices are shared ; their names will be used. Extend \m{\abspre{}'} to an abstraction-prefix function \m{\abspre{}''} for \m{\altg''} so that every indirection vertex \m{\v_0} gets the prefix of its successor \v.
		\item\label[Rb]{readback:step:spantree} Construct a spanning tree \m{T''} of \m{\altg''} by using a depth-first search (DFS) on \m{\altg''}. Note that all variable backlinks and \S-backlinks, and some of the recursive backlinks, of \m{\altg''}, are not contained in \m{T''}, because they are back-edges of the DFS.
		\item\label[Rb]{readback:step:synthesis} Apply the readback synthesis rules from \cref{fig:readback:rules} to \m{\altg''} with respect to \m{T''}. By this a complete labelling of the edges of \m{\altg''} by prefixed \lambdaletrec-terms is constructed. The rules define how the labelling for an incoming edge (on top) of a vertex \v is synthesised under the assumption of an already determined labelling of an outgoing edge of (and below) \v. If the outgoing edge in the rule does not carry a label, then the labelling of the incoming edge can happen regardless. Note that in these rules:
		\begin{itemize}
			\item full (dotted) lines indicate spanning tree (non-spanning tree) edges; broken lines match either kind;
			\item abstraction prefixes of vertices are crucial for the \0-vertex, and the second indirection vertex rule, where the prefixes in the synthesised terms are created; in the other rules the prefix of the assumed term is used; for indicating a correspondence between a term's and a vertex's abstraction prefix we denote by \vs{\vec{p}} the word of vertices occurring in a term's prefix \vec{p};
			\item the rule for indirection vertices with incoming non-spanning tree edge introduces an unfinished function binding \m{f = \nodef} for \f; unfinished bindings are to be filled in later;
			\item the @-vertex rule applies only if \m{\vs{\vec{p}_0} = \vs{\vec{p}_1}}; the operation \pcup used in the synthesised term's prefix builds the union per prefix variable of the pertaining bindings; if the prefixed terms \recprefixed{\vec{\apre}_0}{\L_0} and \recprefixed{\vec{\apre}_1}{\L_1} assumed in this rule contain a yet unfinished function binding \m{f = \nodef} and a completed binding \m{f = P} at a λ-variable \m{z}, the synthesised term contains the completed binding \m{f = P} at \m{z};
		\end{itemize}
	\end{enumerate}
	If this process yields the label \recprefixed{\starfs{}}{\L} at the root of \m{\altg''}, we call \L the \emph{readback of \altg}.
\end{definition}

\begin{figure}
	\begin{hspread}
		\transpicture{
			\ltgnode{root}{\top};
			\addPos{root}{\star};
			\addPrefix{root}{};
			\node[above=of root](rb){\recprefixed{\starfs{}}{\L}}; \draw[->](rb) to (root);
			\node[below=of root](body){\recprefixed{\starfs{}}{\L}}; \draw[->](root) to (body);
		}
		&
		\transpicture{
			\ltgnode{root}{\top};
			\node[right=of root]{\sideCondition{\m{B≠\emptyset}}};
			\addPos{root}{\star};
			\addPrefix{root}{};
			\node[above=of root](rb){\recprefixed{\starfs{}}{\letin{B}{\L}}}; \draw[->](rb) to (root);
			\node[below=of root](body){\recprefixed{\starfs{B}}{\L}}; \draw[->](root) to (body);
		}
		&
		\transpicture{
			\ltgnode{abs}{λ};
			\addPos{abs}{\v_n};
			\addPrefix{abs}{\vs{\vec{p}}};
			\node[above=of abs](rb){\recprefixed{\vec{p}}{\abs{v_n}{L}}}; \draw[->](rb) to (abs);
			\node[below=of abs](body){\recprefixed{\vec{p}\prefixcon \v_n\fs{}}{L}}; \draw[->](abs) to (body);
		}
		&
		\transpicture{
			\ltgnode{abs}{λ};
			\node[right=of abs]{\sideCondition{\m{B≠\emptyset}}};
			\addPos{abs}{\v_n};
			\addPrefix{abs}{\vs{\vec{p}}};
			\node[above=of abs](rb){\recprefixed{\vec{p}}{\abs{v_n}{\letin{B}{\L}}}}; \draw[->](rb) to (abs);
			\node[below=of abs](body){\recprefixed{\vec{p}\prefixcon \v_n\fs{B}}{\L}}; \draw[->](abs) to (body);
		}
		\\
		\transpicture{
			\ltgnode{app}{@};
			\addPrefix{app}{\vec\v};
			\node[above=of app](rb){\recprefixed{\vec{p}_0 \pcup \vec{p}_1}{\app{\L_0}{\L_1}}}; \draw[->](rb) to (app);
			\node[below=of app](args){};
			\node[node distance=0mm,left=of args] (l){\recprefixed{\vec{p}_0}{\L_0}}; \draw[->,dashed](app) to (l);
			\node[node distance=0mm,right=of args](r){\recprefixed{\vec{p}_1}{\L_1}}; \draw[->,dashed](app) to (r);
		}
		&
		\transpicture{
			\ltgnode{var}{\0};
			\addPrefix{var}{v_1 \prefixcon\dots\prefixcon v_n};
			\node[node distance=10mm,above=of var](rb){\recprefixed{\starfs{}\prefixcon \v_1\fs{}\mcdots\v_n\fs{}}{\v_n}};
			\draw[->](rb) to (var);
			\ltgnode[right=of var,yshift=5ex]{abs}{λ}; \draw[->,dotted](var) -| (abs);
			\addPos{abs}{\v_n};
		}
		&
		\transpicture{
			\ltgnode{s}{\S};
			\addPrefix{s}{\vs{\vec{p}} \prefixcon \v_n};
			\node[above=of s](rb){\recprefixed{\vec{p}\prefixcon \v_n\fs{}}{\L}}; \draw[->](rb) to (s);
			\ltgnode[node distance=6mm,right=of s,yshift=4ex]{abs}{λ}; \draw[->,dotted](s) -| (abs);
			\addPos{abs}{\v_n};
			\node[below=of s](succ){\recprefixed{\vec{p}}\L}; \draw[->,dashed] (s) to (succ);
		}
		\\
		\transpicture{
			\ltgnode{bh}{\bh};
			\addPrefix{bh}{};
			\node[above=of bh](rb){\recprefixed{\starfs{f=f}}{f}}; \draw[->](rb) to (bh);
		}
		&
		\transpicture{
			\ltgnode{I}{\indir};
			\addPrefix{I}{\vs{\vec{p}}\prefixcon \v_{n+1}};
			\addPos{I}{f};
			\node[above=of I](rb){\recprefixed{\vec{p}\prefixcon \v_{n+1}\fs{B,f=\L}}{f}}; \draw[->](rb) to (I);
			\node[below=of I](succ){\recprefixed{\vec{p}\prefixcon \v_{n+1}\fs{B,(f=\nodef)}}{\L}}; \draw[->] (I) to (succ);
		}
		&
		\transpicture{
			\ltgnode{I}{\indir};
			\addPrefix{I}{v_1 \prefixcon\dots\prefixcon v_n \prefixcon \v_{n+1}};
			\addPos{I}{f};
			\node[above=of I](rb){\recprefixed{\starfs{}\prefixcon \v_1\fs{}\mcdots\v_n\fs{}\prefixcon \v_{n+1}\fs{f=\nodef}}{f}};
			\draw[->,dotted](rb) to (I);
			\node[below=of I](succ){\m{\phantom{\recprefixed{\vec{p}\prefixcon \v_{n+1}\fs{B}}{\L}}}}; \draw[->] (I) to (succ);
		}
	\end{hspread}
	\caption{Readback synthesis rules for computing a \lambdaletrec-term from a λ-term-graph. For explanations, see \cref{def:readback}~\cref{readback:step:synthesis}.}
	\label{fig:readback:rules}
\end{figure}


\begin{para}[readback is deterministic]
	For every edge \e the synthesis rule to be applied is uniquely determined by the label of the target vertex \v of \e together with side conditions for \m{λ}, \top, and \indir. In the last case it depends on whether \m{e} is a spanning-tree edge or not.
\end{para}

\begin{proposition}[readback function]\label{prop:readback}
	For every λ-term-graph \altg the process from \cref{def:readback} produces a complete edge labelling of the (modified) term graph, with label \recprefixed{\starfs{}}{\L} for the root edge, where \L is a \lambdaletrec-term. Hence it yields \L as the readback of \altg. Thus \cref{def:readback} defines a function \m{\readback{} : \classltgs → \Ter{\lambdaletreccal}}, the \emph{readback function}.
\end{proposition}

\begin{example}\label{ex:readback}
	See \cref{fig:ex:readback} for the illustration of the synthesis of the readback from an example λ-term-graph. Full lines depict spanning tree edges, dotted lines depict non-spanning-tree edges.
\end{example}

The following theorem validates \cref{methods:properties:readback}.

\begin{theorem}\label{thm:readback}
	For all \m{\altg ∈ \eag{\classltgs}} it holds: \m{((\graphsemC{\classltgs}{} ∘ {\readback{})})(\altg) = \graphsemC{\classltgs}{\readback{\altg}} \iso \atg}, i.e., \readback{} is a right-inverse of \graphsemC{\classltgs}{}, and \graphsemC{\classltgs}{} a left-inverse of \readback{}, up to \iso. Hence \readback{} is injective, and \graphsemC{\classltgs}{} is surjective, thus \m{\image{\graphsemC{\classltgs}{}} = \eag{\classltgs}}.
\end{theorem}

\begin{figure}
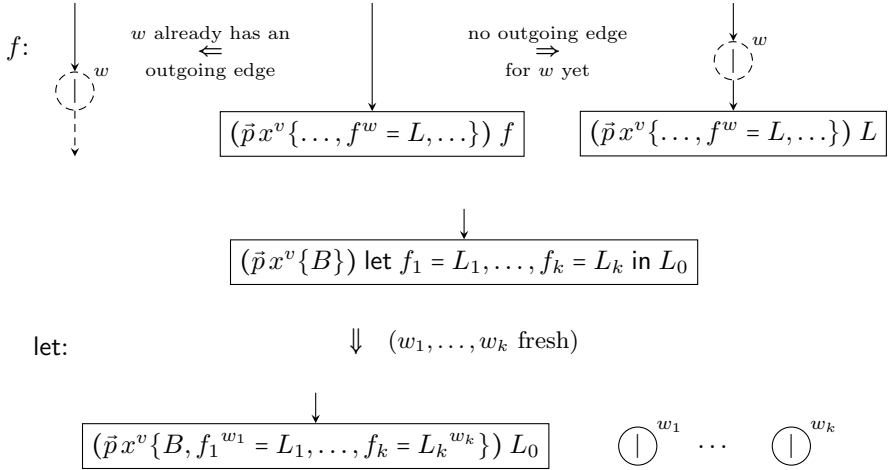

	\begin{hspread}
		\vcentered{\transname{\f}}
		&
		\vcentered{\transpicture{
			\ltgnode[densely dashed]{indir}{\indir};
			\addPos{indir}{w};
			\draw[<-](indir.north) to +(0,9mm);
			\draw[->,densely dashed](indir.south) to +(0,-6mm);
		}}
		&
		\hspace{-3mm}
		\vcentered{\m{\overset{\text{\w already has an}}{\underset{\text{outgoing edge}}{⇐}}}}
		&
		\hspace{-15mm}
		\vcentered{\transpicture{
			\node[draw,shape=transbox](lhs){\recprefixed{\vec{p}\prefixcon x^\v\fs{\dots,f^{\w}=\L,\dots}}{f}};
			\draw[<-](lhs.north) to +(0,14.2mm);
		}}
		&
		\hspace{-12mm}
		\vcentered{\m{\overset{\text{no outgoing edge}}{\underset{\text{for \w yet}}{⇒}}}}
		&
		\hspace{-12mm}
		\vcentered{\transpicture{
			\ltgnode[densely dashed]{indir}{\indir};
			\addPos{indir}{\w};
			\draw[<-](indir.north) to +(0,5mm);
			\node[node distance=4mm,draw,shape=transbox,below=of indir](rhs){\recprefixed{\vec{p}\prefixcon x^\v\fs{\dots,f^{\w}=\L,\dots}}{\L}};
			\draw[->](indir) to (rhs.north);
		}}
	\end{hspread}
	\\[3ex]
	\translationvif{\Let}{
		\node[draw,shape=transbox](lhs){\recprefixed{\vec{p}\prefixcon x^\v\fs{B}}{\letin{f_1=\L_1,\dots,f_k=\L_k}{\L_0}}};
		\draw[<-](lhs.north) to +(0,4mm);
	}
	{\m{\w_1,\dots,\w_k} fresh}
	{
		\node[draw,shape=transbox](in){\recprefixed{\vec{p}\prefixcon x^\v\fs{B,{f_1}^{\w_1}=\L_1,\dots,f_k={\L_k}^{\w_k}}}{\L_0}};
		\draw[<-](in.north) to +(0,4mm);
		\ltgnode[node distance=9mm, right=of in]{bg1}{\indir};
		\addPos{bg1}{\w_1};
		\node[right=of bg1](bgmid){\dots};
		\ltgnode[right=of bgmid]{bgk}{\indir};
		\addPos{bgk}{\w_k};
	}
	\caption{Augmented version of two of the translation rules from \cref{fig:def:graphsem:lhotgs} for an alternative definition of the λ-term-graph semantics of \lambdaletrec-terms.}
	\label{fig:trans-local}
\end{figure}

\begin{figure}
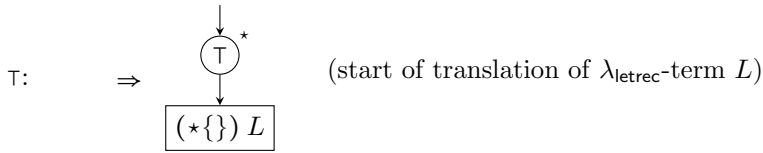

	\proofsystem{
		\begin{bprooftree}
			\emptyAxiom
			\infLabel{\rec}
			\unaryInf{\recprefixed{\vec{p}\prefixcon x^\v\fs{\dots,f^{\w}=\L,\dots}}{f}}
		\end{bprooftree}
		\\
		\begin{bprooftree}
			\axiom{\recprefixed{\vec{p}\prefixcon x^\v\fs{\dots,f^{\w}=\L,\dots}}{\L}}
			\infLabel{\rec~~\sideCondition{\w is fresh}}
			\unaryInf{\recprefixed{\vec{p}\prefixcon x^\v\fs{\dots,f=\L,\dots}}{f}}
		\end{bprooftree}
		\\
		\begin{bprooftree}
			\axiom{\recprefixed{\vec{p} \prefixcon x^\v\fs{B,f_1=\L_1,\dots,f_k=\L_k}}{\L_0}}
			\infLabel{\Let}
			\unaryInf{\recprefixed{\vec{p} \prefixcon x^\v\fs{B}}{\letin{f_1=\L_1,\dots,f_k=\L_k}{\L_0}}}
		\end{bprooftree}
	}
	\caption{Formulation of the local translation rules in \cref{fig:trans-local} in the form of inference rules.}
	\label{fig:trans-local-proof-system}
\end{figure}

\begin{proof}[Proof Idea]
	Graph translation steps can be linked with corresponding readback steps in order to establish that the former roughly reverse the latter. Roughly, because e.g.\ reversing a λ-readback step necessitates both a λ- and a \Let-translation step. Therefore, this correspondence holds only for a modification of the translation rules \rulestranslambdaletreccaltoltgs from \cref{fig:def:graphsem:lhotgs}, shown in \cref{fig:def:graphsem:ltgs} where the rules \Let (for \Let-expressions) and \f (for occurrences of function variables) are replaced by the locally-operating versions in \cref{fig:trans-local}. Moreover, the translation rules need to store entire function bindings in the prefix (not just function variables) as the readback rules do. To the augmented rules we also add a initiating rule
	\begin{center}
	\translation
	{\top}
	{\node(empty){\rule{2ex}{0pt}};}
	{
  	  \ltgnode{root}{\top};
  	  \addPos{root}{\star};
  	  \node[below=of root,shape=transbox,draw](body){\recprefixed{\starfs{}}{\L}};
  	  \draw[->](root) to (body.north);
  	  \draw[<-](root.north) to +(0mm,4mm);
	}
	\hspace*{3ex}
	(start of translation of \lambdaletrec-term \L)
	\end{center}
	for creating a top vertex. Now the translation of a \Let-expression does no longer directly spawn translations of the bindings, but the bindings will only be translated later once their calls have been reached during the translation process of the body, or of the definitions of other already translated bindings. Note that in the \Let-rule in \cref{fig:trans-local} function bindings are associated with the rightmost variable in the prefix, which corresponds to choosing \m{l_i = n} in the \Let-rule in \cref{fig:def:graphsem:lhotgs}. While such a stipulation does not guarantee the eager-scope translation of every term, it actually does so for all \lambdaletrec-terms that are obtained by the readback (on these terms the such defined translation coincides with \graphsemC{\classltgs}{} from \cref{def:graphsem:lhotgs}).
	\par Please find in \cref{fig:proof:thm:readback} a graphical argument for the stepwise reversal of readback steps through (augmented) translation steps. This establishes that graph translation steps reverse readback steps, and is the crucial step in the proof of the theorem. The proof uses induction on access paths, and an invariant that relates the eager-scope property localised for a vertex \v with the applicability of the \S-rule to the readback term synthesised at \v. Note that in most cases the sets of function bindings on the left-hand side (\m{\bar{B}_i}) and the right-hand side (\m{B_i}) differ, due to the freedoms of the function definition's positions in the translated \lambdaletrec-term (see \cref{readback-inverse}).
	\par Note that \m{\cupdot} denotes a join of two sets of function definitions where a defined function always overrules an undefined function, such that for instance the following holds: \m{\set{f=x, g=?} \cupdot \set{f=?, g=?} = \set{f=x, g=?}}.
\end{proof}

\begin{figure}
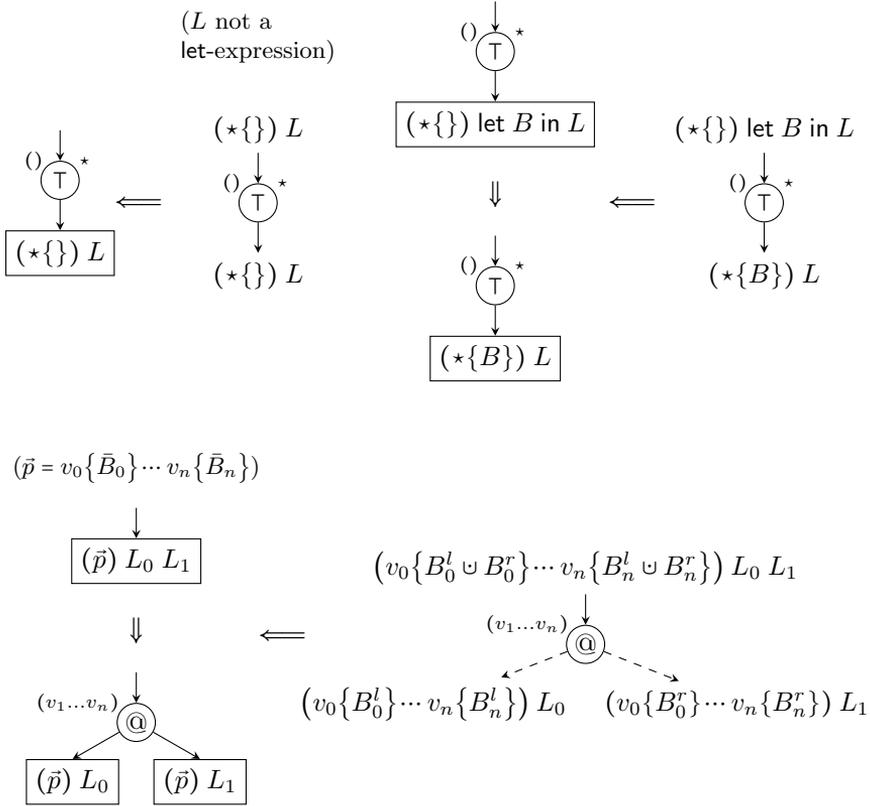

	\begin{hspread}
		\fusionstepsimple{
			\ltgnode{root}{\top};
			\addPrefix{root}{};
			\addPos{root}{\star};
			\node[below=of root,shape=transbox,draw](body){\recprefixed{\starfs{}}{L}};
			\draw[->](root) to (body.north);
			\draw[<-](root.north) to +(0mm,4mm);
		}{
			\ltgnode{root}{\top};
			\addPos{root}{\star};
			\addPrefix{root}{};
			\node[above=of root](rb){\recprefixed{\starfs{}}{L}}; \draw[->](rb) to (root);
			\node[below=of root](body){\recprefixed{\starfs{}}{L}}; \draw[->](root) to (body);
			\node[above=of rb](text){\mlSideCondition{\L not a\\\Let-expression}};
			\node[below=of body](antitext){\phantom{\mlSideCondition{\L not a\\\Let-expression}}};
		}
		&
		\fusionstep{
			\ltgnode{root}{\top};
			\addPrefix{root}{};
			\addPos{root}{\star};
			\draw[<-](root.north) to +(0mm,4mm);
			\node[below=of root,shape=transbox,draw](body){\recprefixed{\starfs{}}{\letin{B}{L}}};
			\draw[->](root.south) to (body.north);
		}{
			\ltgnode{root}{\top};
			\addPrefix{root}{};
			\addPos{root}{\star};
			\draw[<-](root.north) to +(0mm,4mm);
			\node[below=of root,shape=transbox,draw](body){\recprefixed{\starfs{B}}{L}};
			\draw[->](root.south) to (body.north);
		}{
			\ltgnode{root}{\top};
			\addPrefix{root}{};
			\addPos{root}{\star};
			\node[above=of root](rb){\recprefixed{\starfs{}}{\letin{B}{L}}}; \draw[->](rb) to (root);
			\node[below=of root](body){\recprefixed{\starfs{B}}{L}}; \draw[->](root) to (body);
		}
		\\
		\fusionstepclap{
			\node[draw,shape=transbox](lhs){\recprefixed{\vec{p}}{\app{\L_0}{\L_1}}};
			\draw[<-](lhs.north) to +(0,4mm);
			\node[node distance=6mm, above=of lhs]{\sideCondition{\m{\vec{p} = \v_0\fs{\bar{B}_0}\mcdots\v_n\fs{\bar{B}_n}}}};
		}{
			\ltgnode{app}{@}; \node[node distance=4mm,below=of app](middle){};
			\node[node distance=1mm,draw,shape=transbox,left=of middle](l){\recprefixed{\vec{p}}{\L_0}};
			\node[node distance=1mm,draw,shape=transbox,right=of middle](r){\recprefixed{\vec{p}}{\L_1}};
			\addPrefix{app}{\v_1 \dots \v_n};
			\draw[->](app) to (l.north); \draw[->](app) to (r.north); \draw[<-](app.north) to +(0,4mm);
		}{
			\ltgnode{app}{@};
			\addPrefix{app}{\v_1 \dots \v_n};
			\node[above=of app](rb){\recprefixed{\v_0\fs{B^l_0 \cupdot B^r_0}\mcdots\v_n\fs{B^l_n \cupdot B^r_n}}{\app{\L_0}{\L_1}}}; \draw[->](rb) to (app);
			\node[below=of app](args){};
			\node[node distance=0mm,left=of args] (l){\recprefixed{\v_0\fs{B^l_0}\mcdots\v_n\fs{B^l_n}}{\L_0}}; \draw[->,dashed](app) to (l);
			\node[node distance=0mm,right=of args](r){\recprefixed{\v_0\fs{B^r_0}\mcdots\v_n\fs{B^r_n}}{\L_1}}; \draw[->,dashed](app) to (r);
		}
	\end{hspread}
	\caption{Reversal of readback steps through translation steps}
	\label{fig:proof:thm:readback}
\end{figure}

\begin{figure}
	\ContinuedFloat
	\begin{hspread}
		\fusionstep
		{
			\node[draw,shape=transbox](lhs){\recprefixed{v_0\fs{\bar{B}_0}\mcdots v_{n-1}\fs{\bar{B}_{n-1}}}{\abs{v_n}{L}}};
			\draw[<-](lhs.north) to +(0mm,4mm);
		}{
			\node[draw,shape=transbox](body){\recprefixed{v_0\fs{\bar{B}_0}\mcdots v_{n-1}\fs{\bar{B}_{n-1}}\prefixcon v_n\fs{}}{L}};
			\ltgnode[node distance=3mm,above=of body.north]{abs}{λ};
			\addPos{abs}{v_n};
			\addPrefix{abs}{v_1 \dots v_{n-1}};
			\draw[<-](abs.north) to +(0mm,4mm);
			\draw[->](abs) to (body.north);
		}{
			\ltgnode{abs}{λ};
			\addPos{abs}{v_n};
			\addPrefix{abs}{v_1 \dots v_{n-1}};
			\node[above=of abs](rb){\recprefixed{v_0\fs{B_0}\mcdots v_{n-1}\fs{B_{n-1}}}{\abs{v_n}{L}}}; \draw[->](rb) to (abs);
			\node[below=of abs](body){\recprefixed{v_0\fs{B_0}\mcdots v_{n-1}\fs{B_{n-1}} \prefixcon v_n\fs{}}{L}}; \draw[->](abs) to (body);
		}
		\\
		\\
		\fusiontwostep{
			\node[draw,shape=transbox](lhs){\recprefixed{\v_0\fs{\bar{B}_0}\mcdots\v_{n-1}\fs{\bar{B}_{n-1}}}{\abs{v_n}{\letin{\bar{B}}{\L}}}};
			\draw[<-](lhs.north) to +(0mm,4mm);
		}{
			\node[draw,shape=transbox](rhs){\recprefixed{\v_0\fs{\bar{B}_0}\mcdots\v_{n-1}\fs{\bar{B}_{n-1}}\prefixcon\v_n\fs{}}{\letin{\bar{B}}{\L}}};
			\ltgnode[node distance=3mm,above=of rhs.north]{abs}{λ}; \addPos{abs}{\v_n};
			\addPrefix{abs}{\v_1 \dots\,\v_{n-1}};
			\draw[<-](abs.north) to +(0mm,4mm);
			\draw[->](abs) to (rhs.north);
		}{
			\node[draw,shape=transbox](rhs){\recprefixed{\v_0\fs{\bar{B}_0}\mcdots\v_{n-1}\fs{\bar{B}_{n-1}}\prefixcon\v_n\fs{\bar{B}}}{\L}};
			\ltgnode[node distance=3mm,above=of rhs.north]{abs}{λ}; \addPos{abs}{\v_n};
			\addPrefix{abs}{\v_1 \dots\,\v_{n-1}};
			\draw[<-](abs.north) to +(0mm,4mm);
			\draw[->](abs) to (rhs.north);
		}{
			\ltgnode{abs}{λ};
			\addPos{abs}{\v_n};
			\addPrefix{abs}{\v_1 \dots\,\v_{n-1}};
			\node[above=of abs](rb){\recprefixed{\vec{p}}{\abs{v_n}{\letin{B}{\L}}}}; \draw[->](rb) to (abs);
			\node[below=of abs](body){\recprefixed{\vec{p} \prefixcon \v_n\fs{B}}{\L}}; \draw[->](abs) to (body);
			\node[below=of body]{\sideCondition{\m{\vec{p} = \v_0\fs{B_0}\mcdots\v_{n-1}\fs{B_{n-1}}}}};
		}
	\end{hspread}
	\caption[]{(continued) Reversal of readback steps through translation steps}
\end{figure}

\begin{figure}
	\ContinuedFloat
	\begin{hspread}
		\fusionstep
		{
			\node[draw,shape=transbox](lhs){\recprefixed{\starfs{B_{0}}\:v_{1}\fs{B_{1}}\mcdots v_{n}\fs{B_{n}}}{v_{n}}};
			\draw[<-](lhs.north) to +(0,4mm); 
			} 
		{
			\ltgnode{var}{\0}; \draw[<-](var.north) to +(0,5mm);
			\addPrefix{var}{v_{1}\mcdots v_{n}};
			\node[node distance=4mm,right=of var](left){};
			\ltgnode[node distance=3mm,above=of left,densely dashed]{abs}{λ}; \addPos{abs}{v_{n}};
			\draw[->](var) -| (abs);
			}
		{
			\ltgnode{var}{\0};
			\addPrefix{var}{v_1 \dots v_n};
			\node[node distance=10mm,above=of var](rb){\m{\recprefixed{\starfs{}\prefixcon v_{1}\fs{}\mcdots v_{n}\fs{}}{v_{n}}}}; 
			\draw[->](rb) to (var);
			\ltgnode[right=of var,yshift=5ex]{abs}{λ}; \draw[->,dotted](var) -| (abs);
			\addPos{abs}{v_{n}};
		} 
		\\
		\\
		\fusionthreestep{
			\node[draw,shape=transbox](rhs){\m{\recprefixed{\starfs{\dots,f=f,\dots}}{f}}};
			\draw[<-](rhs.north) to +(0,4mm); 
		}{
			\ltgnode{indir}{\indir};
			\addPos{indir}{f};
			\addPrefix{indir}{};
			\draw[<-](indir.north) to +(0,4mm);
			\node[node distance=3mm,draw,shape=transbox,below=of indir](rhs){\m{\recprefixed{\starfs{\dots,f=f,\dots}}{f}}};
			\draw[->](indir) to (rhs.north);
		}{
			\ltgnode{indir}{\indir};
			\draw[->,bend right=120, looseness=8] (indir.south) to (indir.north);
			\phantom{\draw[->,bend left=120, looseness=8] (indir.south) to (indir.north);}
			\addPos{indir}{f};
			\addPrefix{indir}{};
			\draw[<-](indir.north) to +(0,4mm);
		}{
			\ltgnode{bh}{\bh};
			\addPrefix{bh}{};
			\addPrefixswphantom{bh}{}; 
			\draw[<-](bh.north) to +(0,4mm); 
		}{
			\ltgnode{bh}{\bh};
			\addPrefix{bh}{};
			\node[above=of bh](rb){\m{\recprefixed{\starfs{f=f}}{f}}}; \draw[->](rb) to (bh);
			\node[below=of bh](bhb){\phantom{\m{\recprefixed{\starfs{f=f}}{f}}}}; \draw[color=white](bh) to (bhb);
		}
	\end{hspread}
	\caption[]{(continued) Reversal of readback steps through translation steps}
\end{figure}

\begin{figure}
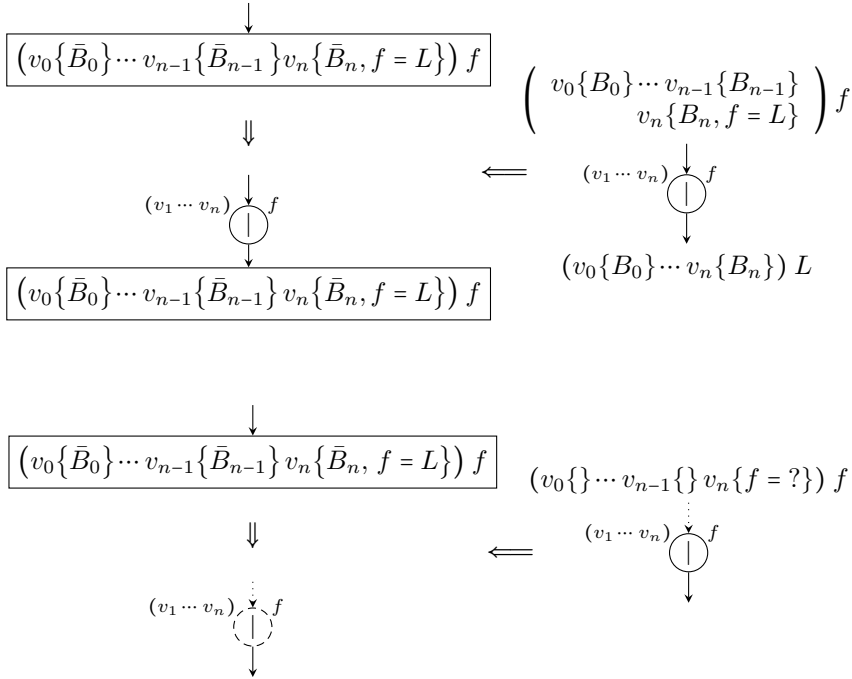

	\ContinuedFloat
	\begin{hspread}
		\fusionstepclap{
			\node[draw,shape=transbox](rhs){\m{\recprefixed{v_0\fs{\bar{B}_{0}}\mcdots v_{n-1}\fs{\bar{B}_{n-1} \prefixcon }v_n\fs{\bar{B}_{n},f=L}}{f}}};
			\draw[<-](rhs.north) to +(0,4mm);
		}{
			\ltgnode{indir}{\indir};
			\addPos{indir}{f};
			\addPrefix{indir}{v_{1}\mcdots v_{n}};
			\draw[<-](indir.north) to +(0,4mm);
			\node[node distance=3mm,draw,shape=transbox,below=of indir](rhs){\m{\recprefixed{v_0\fs{\bar{B}_{0}}\mcdots v_{n-1}\fs{\bar{B}_{n-1}}\prefixcon v_n\fs{\bar{B}_{n},f=L}}{f}}};
			\draw[->](indir) to (rhs.north);
		}{
			\ltgnode{I}{\indir}; 
			\addPrefix{I}{v_{1}\mcdots v_{n}};
			\addPos{I}{f};
			\node[above=of I](rb){\m{\recprefixed{\twolines{v_0\fs{B_{0}}\mcdots v_{n-1}\fs{B_{n-1}}}{v_{n}\fs{B_{n}, f=L}}}{f}}}; \draw[->](rb) to (I);
			\node[below=of I](succ){\m{\recprefixed{v_{0}\fs{B_{0}}\mcdots v_{n}\fs{B_{n}}}{L}}}; \draw[->] (I) to (succ);
		}
		\\
		\\
		\fusionstepclap{
			\node[draw,shape=transbox](rhs){\m{\recprefixed{v_{0}\fs{\bar{B}_{0}}\mcdots v_{n-1}\fs{\bar{B}_{n-1}}\prefixcon v_{n}\fs{\bar{B}_{n},\,f = L}}{f}}};
			\draw[<-](rhs.north) to +(0,4mm); 
		}{
			\ltgnode[densely dashed]{indir}{\indir};
			\addPos{indir}{f};
			\addPrefix{indir}{v_{1}\mcdots v_{n}};
			\addPrefixswphantom{indir}{v_{1}\mcdots v_{n}};
			\draw[<-,dotted](indir.north) to +(0,4mm);
			\draw[->](indir.south) to +(0,-4mm); 
		}{
			\ltgnode{I}{\indir};
			\addPrefix{I}{v_{1}\mcdots v_{n}};
			\addPos{I}{f};
			\node[above=of I](rb){\m{\recprefixed{v_{0}\fs{}\mcdots v_{n-1}\fs{}\prefixcon v_{n}\fs{f=\nodef}}{f}}};
			\draw[->,dotted](rb) to (I);
			\node[below=of I](succ){\phantom{\m{\recprefixed{\vec{p}\prefixcon v_{n+1}[B]}L}}}; \draw[->] (I) to (succ);
		}	
	\end{hspread}
	\caption[]{(continued) Reversal of readback steps through translation steps}
\end{figure}

\section{Complexity analysis}\label{sec:complexity}

Here we report on the complexity for the individual operations from the previous sections used form implementing maximal sharing and determining unfolding-equivalence.

In the lemma below, \cref{lem:complexity:graphsem} and \cref{lem:complexity:readback} justify \cref{methods:properties:efficiency} of our methods. Items \cref{lem:complexity:collapse} and \cref{lem:complexity:bisimilarity} detail the complexity of standard methods when used for computing bisimulation collapse and bisimilarity of λ-term-graphs. Note that first-order term graphs can be modelled by deterministic process graphs, and hence by DFAs \cite{hopc:karp:1971}. Therefore bisimilarity of term graphs can be computed via language equivalence of the corresponding DFAs (in time \bigO{n α(n)} \cite{nort:2009}, where \m{α} is the quasi-constant \emph{inverse Ackermann function}) and bisimulation collapse via state minimisation of DFAs (in time \bigO{n \log n}) \cite{hopc:1971}.

By \termsize{\L} we denote the size (number of symbols) of a \lambdaletrec-term \L. Also we denote by \graphsize{G} the size (number of vertices) of a term graph \G.

\begin{lemma}\label{lem:complexity}
  \begin{enumerate}[(i)]
    \item\label{lem:complexity:length:graphsem}
      \m{\graphsize{\graphsemC{\classltgs}{\L}} ∈ \bigO{\termsize{\L}^2}} for \m{\L ∈ \Ter{\lambdaletreccal}}.
    \item\label{lem:complexity:graphsem}
       Translating \m{\L ∈ \Ter{\lambdaletreccal}} into \m{\graphsemC{\classltgs}{\L} ∈ \classltgs} takes time \bigO{\termsize{\L}^2}.
    \item\label{lem:complexity:collapse}
       Collapsing \m{\altg ∈ \classltgs} to \coll{\altg} is in \bigO{\graphsize{\altg} \log \graphsize{\altg}}.
    \item\label{lem:complexity:bisimilarity}
      Deciding bisimilarity of \m{\altg_1,\altg_2 ∈ \classltgs}
      requires time \bigO{n α(n)} for \m{n = \max \set{\graphsize{\altg_1}, \graphsize{\altg_2}}}. 
    \item\label{lem:complexity:readback}
      Computing the readback
                \readback{\altg} for a given \m{\altg ∈ \classltgs} requires time
      \bigO{n \log n}, for \m{n = \graphsize{\altg}}.
  \end{enumerate}
\end{lemma}

See \cref{fig:ex:translation:quadratic} for an example that the size bound in item \cref{lem:complexity:length:graphsem} of the lemma is tight.

\begin{proposition}
	\m{\graphsize{\graphsemC{\classltgs}{\L}} ∈ \bigO{\termsize{\L}^2}} for \lambdaletrec-terms \L.
\end{proposition}

\begin{figure}
	Consider the finite λ-terms \m{\M_n} with \n occurrences of bindings \abs{x_2}{}:
	\begin{equation*}
		\abs{x_0}{\abs{x_1}{\app{\app{x_0}{x_1}}
			{\abs{x_2}{\app{\app{x_0}{x_1}}
				{\abs{x_1}{\app{\app{x_0}{x_2}}
					{\abs{x_2}{\app{\app{x_0}{x_1}}
						{\dots
							{\abs{x_2}{\app{\app{x_0}{x_1}}{x_2}}}}}}}}}}}}
	\end{equation*}
	\m{\termsize{\M_n} ∈ \bigO{n}}. But both the transformation of \m{\M_n} into de-Bruijn notation
	\begin{equation*}
		\abs{}{\abs{}{\app{\app{\S(\0)}{\0}}
				{\abs{}{\app{\app{\S^2(\0)}{\S(\0)}}
						{\abs{}{\app{\app{\S^3(\0)}{\S(\0)}}
								{\abs{}{\app{\app{\S^4(\0)}{\S(\0)}}
						{\dots
							{\abs{}{\app{\app{\S^{2n}(\0)}{\S(\0)}}{\0}}}}}}}}}}}}
	\end{equation*}
	and the rendering of \m{\M_n} with respect to the eager scope-delimiting strategy:
	\begin{equation*}
		\abs{}{\abs{}{\app{\app{\S(\0)}{\0}}
				{\abs{}{\app{\S(\app{\S(\0)}{\0})}
						{\abs{}{\app{\S(\app{\S^2(\0)}{\0})}
								{\abs{}{\app{\S(\app{\S^3(\0)}{\0})}
						{\dots
							{\abs{}{\app{\S(\app{\S^{2n-1}(\0)}{\0})}{\0}}}}}}}}}}}}
	\end{equation*}
	have size \bigO{n^2}.
	\caption{Example of a sequence \m{\set{\M_n}_n} of finite λ-terms \m{\M_n} whose translation into λ-term-graphs grows quadratically in the size of \m{\M_n}.}
	\label{fig:ex:translation:quadratic}
\end{figure}

Based on this lemma, and on further considerations, we obtain the following
complexity statements for our methods.

\begin{theorem}\label{thm:methods:complexity}
  \begin{enumerate}[(i)]
    \item
      The computation, for a \lambdaletrec-term \L with \m{\termsize{\L} = n},
      of a maximally compactified form
      \m{{(\readback{} \circ \scoll \circ \graphsemC{\classltgs}{})}(\L)}
      requires time \bigO{n^2 \log n}.
      \newcommand\sSunsh{\m{\textsf{unsh}_{\S}}}
      By using an \S-unsharing operation \sSunsh after the collapse,
      a (typically smaller) \lambdaletrec-term \m{((\readback{} ∘ {{\sSunsh} ∘ {\scoll}} ∘ {\graphsemC{\classltgs}{})}){\L}}
   	of size \bigO{n \log n} can be obtained,
      with the same time complexity.
    \item
      The decision of whether two \lambdaletrec-terms \m{\L_1} and \m{\L_2} are unfolding equivalent
      requires time \m{\bigO{n^2} α(n)} for \m{n = \max \set{\termsize{\L_1},\termsize{\L_2}}}.
  \end{enumerate}
\end{theorem}

\section{Implementation}\label{sec:implementation}

We have implemented our methods in Haskell using the \href{http://www.cs.uu.nl/wiki/bin/view/HUT/AttributeGrammarSystem}{Utrecht University Attribute Grammar System}. The implementation is available at \url{http://hackage.haskell.org/package/maxsharing/}. Output produced for a number of examples from this thesis can be found in \cref{app:impl:showcase}. The output includes translations into λ-term-graphs (in DFA-form) according to different semantics, complete derivations w.r.t.\ these semantics, the bisimulation collapse of the λ-term-graphs and the readback thereof.

\section{Modifications, extensions and applications}\label{sec:applications}

We conclude by describing straightforward modifications, extensions, and promising areas of application for our methods.

\subsection{Modifications}\label{sec:conclusion:subsec:modifications}

\begin{para}[implicit sharing of λ-variable]
	Our method introduces explicit sharing via \Let for multiple occurences of the same λ-variable. For instance the term \abs{x}{\app{x}{x}} is compactified into \abs{x}{\letin{f=x}{\app{f}{f}}}. Such explicit sharing of abstraction variables is excessive for many applications. This is easily resolved, by unsharing variable vertices before applying the readback, or by preventing the readback from introducing function bindings when only a variable vertex is shared.
\end{para}

\begin{para}[avoiding aliases produced by the readback]
	The readback function in \cref{sec:readback} is sensitive to the degree of sharing of \S-vertices in the given λ-term-graph: it maps two λ-term-graphs that only differ in what concerns sharing of \S-vertices to different \lambdaletrec-terms. Typically, for λ-term-graphs with maximal sharing of \S-vertices this can produce function bindings that are just `aliases', such as \g is alias for \m{I} in \m{\L'_3} from \cref{ex:translations}. This can be avoided in two ways: by slightly adapting the readback function, or by performing maximal unsharing of \S-vertices before applying the readback as defined.
\end{para}

\begin{para}[preventing disadvantageous sharing]
	Introducing sharing at compile-time can cause `space leaks', i.e.\ a needlessly high memory footprint, at run-time, because `a large data structure becomes shared [\dots], and therefore its space which before was reclaimed by garbage collection now cannot be reclaimed until its last reference is used' \cite{sant:1995}. For this reason, realisations of CSE \cite{chit:1997:CS-uncommon} prevent the introduction of such undesired sharing by suitable conditions that account for the type of potentially shared subexpressions, and their strictness in the program.
	\par As our approach generalises CSE, it inherits this weakness, and the introduction sharing needs to be restricted in a similar fashion. Technically this can easily achieved by adding additional backlinks to prevent parts of the λ-term-graph from collapsing.
\end{para}


\begin{para}[a more general notion of readback]
	\cref{methods:properties:readback} is rather rigorous in that it imposes sharing structure on \lambdaletreccal that is specific to λ-term-graphs (degrees of \S-sharing). For a weaker definition of \cref{methods:properties:readback} with \m{\bisim^\S} in place of isomorphism, a readback does not have to be injective and independently of the degree of \S-sharing a readback function always exists.
\end{para}

\begin{para}[scope-closure strategies]
	We focused on eager-scope translations since they maximise sharing. However, every scope-closure strategy \cite{grab:roch:expressibility} induces a translation and its own notion of maximal sharing.
\end{para}

\subsection{Extensions}\label{sec:conclusion:subsec:extensions}

\begin{para}[full functional languages]
	In order to support programming languages that are based on \lambdaletreccal like Haskell, additional language constructs need to be supported. Such languages can typically be desugared into a core language, which comprises only a small subset of language constructs such as constructors, case statements, and primitives. These constructs can be represented in an extension of \lambdaletreccal\ by additional function symbols. In conjunction with a desugarer our methods are applicable to full programming languages.
\end{para}

\begin{para}[other programming languages, and calculi with binding constructs]
	Most programming languages feature constructs for grouping definitions that are similar to \letrec. We therefore expect that our methods can be adapted to many imperative languages in particular, and may turn out to be fruitful for optimising compilers. Our methods for achieving maximal sharing certainly generalise to theoretical frameworks and calculi with binding constructs, such as the \pi-calculus \cite{miln:1999}, and higher-order rewrite systems (e.g.\ CRSs and HRSs, \cite{terese:2003}) as used here for the formalisation of \lambdaletreccal.
\end{para}

\begin{para}[fully-lazy λ-lifting]
	\newcommand\sreadbackC[1]{\m{\readback{}_{#1}}}
	There is a close connection between our methods and fully-lazy λ-lifting \cite{hugh:1982,peyt:jone:1987}. In particular, the required-variable and scope analysis of a \lambdaletrec-term \L on which the λ-term-graph translation \graphsemC{\classltgs}{\L} is based is analogous to the one needed for extracting from \L the supercombinators in the result \Hat{\L} of fully-lazy λ-lifting \L. Moreover, the fully-lazy λ-lifting transformation can even be implemented in a natural way on the basis of our methods. Namely as the composition \m{\sreadbackC{\textit{LL}} ∘ \graphsemC{\classltgs}{\cdot}} of the translation \graphsemC{\classltgs}{\cdot} into λ-term-graph, where \sreadbackC{\textit{LL}} is a variant readback function that, for a given λ-term-graph, synthesises the system \Hat{\L} of supercombinators, instead of the \lambdaletrec-term \readback{\L}.
\end{para}

\begin{para}[maximal sharing on supercombinator translations of \lambdaletrec-terms]
	\lambdaletrec-terms \L correspond to supercombinator systems \m{S}, the result of fully-lazy lambda-lifting \L: the combinators in \m{S} correspond with `extended scopes' in \L, and supercombinator reduction steps on \m{S} correspond with weak β-reduction steps \L. In the case of λ-calculus this has been established by Balabonski \cite{bala:2012}. Via this correspondence the maximal-sharing method for \lambdaletrec-terms can be lifted to obtain a maximal-sharing method systems of supercombinators obtained by fully-lazy lambda-lifting.
\end{para}

\begin{para}[non-eager scope-closure strategies]
	We focused on eager-scope translations, because they facilitate maximal sharing, and guarantee that interpretations of unfolding-equivalent \lambdaletrec-terms are bisimilar. Yet every scope-closure strategy induces a translation and its own notion of maximal sharing. For adapting our maximal sharing method it however necessary to modify the translation into first-order term graphs in such a way that the image class obtained is closed under functional bisimulation (\classtgssiglambdaij{1}{2} is not closed under functional bisimulation, unlike its subclass \m{\eag{\classtgssiglambdaij{1}{2}} = \classltgs}). This can be achieved by using delimiter vertices also below variable vertices to close scopes that are still open \cite[report]{grab:roch:representations}.
\end{para}


\begin{para}[weaker notions of sharing]
	The presented methods deal with sharing as expressed by \letrec that is horizontal, vertical, or twisted (see \cref{horiz-vert-twisted-sharing}). By contrast, the μ-construct \cite{blom:2001,grab:roch:expressibility:mu} expresses only vertical, and the non-recursive \Let only horizontal, sharing (see \cref{horiz-vert-twisted-sharing}). By restricting bisimulation (either artificially or by adding special backlinks), our methods can be adapted to the λ-calculus with μ \cite{grab:roch:expressibility:mu}, or with \Let.
\end{para}

\subsection{Applications}\label{sec:conclusion:subsec:applications}

\begin{para}[maximal sharing at run-time]
	Maximal sharing can be applied repeatedly at run-time in order to regain a maximally shared form, thereby speeding up evaluation. This is reminiscent of `collapsed tree rewriting' \cite{plum:1993} for evaluating first-order term graphs represented as maximally shared dags. Since the state of a program in the memory at run-time is typically represented as a supercombinator graph, compactification by bisimulation collapse can take place directly on that graph (see \cref{sec:conclusion:subsec:extensions}), no translation is needed. Compactification can be coupled with garbage collection as bisimulation collapse subsumes some of the work required for a mark and sweep garbage collector. However, a compromise needs to be found between the costs for the optimisation and the gained efficiency.
\end{para}

\begin{para}[additional prevention of disadvantageous sharing]
	While static analysis methods for preventing sharing that may be disadvantageous at run-time can be adapted from CSE to the maximal-sharing method (see \cref{sec:conclusion:subsec:modifications}), this has yet to be investigated for binding-time analysis \cite{pals:schw:1994} and a sharing analysis of partial applications \cite{gold:huda:1987}.
for fine-tuning sharing of partial applications in supercombinator translations \cite{gold:huda:1987}.
\end{para}

\begin{para}[compile-time optimisation]
	Increasing sharing facilitates potential gains in efficiency. Our method generalises common subexpression elimination, but therefore it also inherits its shortcomings: the cost of sharing (e.g.\ of very small functions) might exceed the gain. In non-strict functional languages, sharing can cause `memory leaks' \cite{chit:1997:CS-uncommon}. Therefore, similar as for CSE, additional dynamic analyses like binding-time analysis \cite{pals:schw:1994}, and heuristics to restrict sharing in cases when it is disadvantageous \cite{peyt:jone:1987,gold:huda:1987} are in order.
\end{para}

\begin{para}[code improvement]
	In programming it is generally desirable to avoid duplication of code. As an extension of CSE, our method is able to detect code duplication. The bisimulation collapse of the term graph interpretation of a program can, together with the readback, provide guidance on how code can be written more compactly with less duplication. This optimisation has to be fine-tuned to avoid excessive behaviour like the explicit sharing of variable occurrences (see \cref{sec:conclusion:subsec:modifications}). Yet for this only lightweight additional machinery is needed, such as size constraints or annotations to restrict the bisimulation collapse.
\end{para}

\begin{para}[function equivalence]
	Recognising whether two programs are equivalent in the sense that they implement the same function is undecidable. Still, this problem is tackled by proof assistants, and by automated theorem provers used in type-checkers of compilers for dependently-typed programming languages such as Agda. For such systems coinductive proofs are more difficult to find than inductive ones, and require more effort by the user. Our method for deciding unfolding-equivalence could contribute to finding coinductive proofs.
\end{para}


\bibliographystyle{thesis}
\bibliography{thesis}

\appendix
\renewcommand\thetheorem{\thechapter .\arabic {theorem}}

\chapter{Examples: Graph Translation}\label{app:translation}


\begin{para}[overview]
	For two terms from \cref{chap:maxsharing} we provide the fully worked out stepwise translation of \lambdaletrec-terms according to \rulestranslambdaletreccaltolhotgs from \cref{fig:def:graphsem:lhotgs} producing λ-ap-ho-term-graphs and for the rules \rulestranslambdaletreccaltoltgs from \cref{fig:def:graphsem:ltgs} producing λ-term-graphs. Note that in each step of the translation we apply not only one but multiple translation rules (indicated as subscripts of the steps), one for each translation box of the current graph. When no more rules are applicable, indirection vertices are erased.
\end{para}

\begin{para}
  	We start off by a simple example, namely the translation (\cref{fig:ex:compact}) of the term \P from \cref{ex:fix}. There is no application of the \S-rule, thus it yields the exact same translation sequence for \rulestranslambdaletreccaltolhotgs as for \rulestranslambdaletreccaltoltgs producing a λ-term-graph.
\end{para}

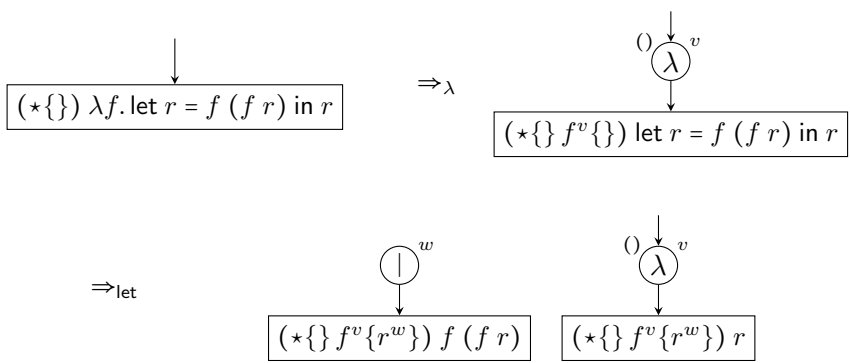
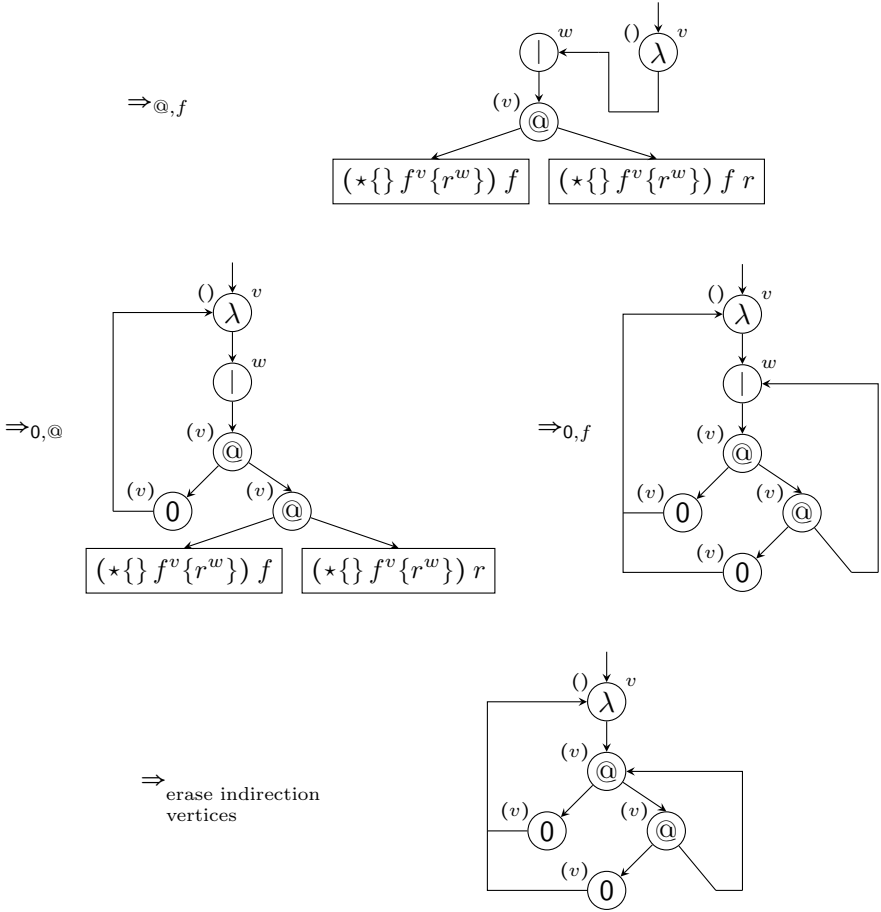
\newcommand\eraseIndirectionVertices{\multilinebox{\scriptsize erase indirection\\[-1ex]\scriptsize vertices}}

\begin{transenv}
	\transinit{
		\node[draw,shape=transbox](t){\recprefixed{\starfs{}}{\abs{f}{\letin{r=\app{f}{(\app{f}{r})}}{r}}}};
		\draw[<-](t.north) to +(0mm,6mm);
	}
	\transstep{λ}{
		\node[draw,shape=transbox](t){\recprefixed{\starfs{} \prefixcon f^v\fs{}}{\letin{r=\app{f}{(\app{f}{r})}}{r}}};
		\ltgnode[above=of t.north]{f}{λ}; \addPos{f}{v}; \addPrefix{f}{};
		\draw[<-](f.north) to +(0mm,4mm);
		\draw[->](f) to (t.north);
	}
	\transstep{\Let}{
		\node[draw,shape=transbox](in){\recprefixed{\starfs{}\prefixcon f^v\fs{r^w\commentOut{=\app{f}{(\app{f}{r})}}}}{r}};
		\ltgnode[above=of in.north]{f}{λ}; \addPos{f}{v}; \addPrefix{f}{};
		\draw[->](f) to (in.north);
		\draw[<-](f.north) to +(0mm,4mm);
		\node[draw,shape=transbox,left=of in](ffr){\recprefixed{\starfs{}\prefixcon f^v\fs{r^w\commentOut{=\app{f}{(\app{f}{r})}}}}{\app{f}{(\app{f}{r})}}};
		\ltgnode[above=of ffr]{indir}{\indir}; \addPos{indir}{w};
		\draw[->](indir) to (ffr.north);
	}
	\transstep{@,\f}{
		\ltgnode{f}{λ}; \addPos{f}{v}; \addPrefix{f}{};
		\draw[<-](f.north) to +(0mm,4mm);
		\node[below=of f](in){};
		\node[left=of in](inwest){};
		\draw(f) |- (inwest.center);
		\node[left=of f](indireast){};
		\ltgnode[left=of indireast]{indir}{\indir}; \addPos{indir}{w};
		\draw[->](indireast.center) to (indir.east);
		\draw(indireast.center) -| (inwest.center);
		\ltgnode[below=of indir]{f@fr}{@}; \addPrefix{f@fr}{v};
		\draw[->](indir) to (f@fr);
		\node[below=of f@fr](ffr){};
		\node[node distance=0mm, left=of ffr,draw,shape=transbox](l){\recprefixed{\starfs{}\prefixcon f^v\fs{r^w\commentOut{=\app{f}{(\app{f}{r})}}}}{f}};
		\node[node distance=0mm,right=of ffr,draw,shape=transbox](ffr){\recprefixed{\starfs{}\prefixcon f^v\fs{r^w\commentOut{=\app{f}{(\app{f}{r})}}}}{\app{f}{r}}};
		\draw[->](f@fr) to (l.north);
		\draw[->](f@fr) to (ffr.north);
	}
	\transstep{\0,@}{
		\ltgnode{f}{λ}; \addPos{f}{v}; \addPrefix{f}{};
		\draw[<-](f.north) to +(0mm,4mm);
		\node[below=of f](in){};
		\ltgnode[below=of f]{indir}{\indir}; \addPos{indir}{w};
		\draw[->](f) to (indir.north);
		\ltgnode[below=of indir]{f@fr}{@}; \addPrefix{f@fr}{v};
		\draw[->](indir) to (f@fr);
		\node[below=of f@fr](ffr){};
		\ltgnode[left=of ffr]{l}{\0}; \addPrefix{l}{v};
		\draw[->](f@fr) to (l);
		\node[left=of l](lwest){}; \draw(l) to (lwest.center); \draw[->](lwest.center) |- (f);
		\ltgnode[right=of ffr]{f@r}{@}; \addPrefix{f@r}{v};
		\node[below=of f@r](fr){};
		\draw[->](f@fr) to (f@r.north);
		\node[node distance=0mm, left=of fr,draw,shape=transbox](f1){\recprefixed{\starfs{}\prefixcon f^v\fs{r^w\commentOut{=\app{f}{(\app{f}{r})}}}}{f}};
		\draw[->](f@r) to (f1.north);
		\node[node distance=0mm,right=of fr,draw,shape=transbox](r){\recprefixed{\starfs{}\prefixcon f^v\fs{r^w\commentOut{=\app{f}{(\app{f}{r})}}}}{r}};
		\draw[->](f@r) to (r.north);
	}
	\transstep{\0,f}{
		\ltgnode{f}{λ}; \addPos{f}{v}; \addPrefix{f}{};
		\draw[<-](f.north) to +(0mm,4mm);
		\node[below=of f](in){};
		\ltgnode[below=of f]{indir}{\indir}; \addPos{indir}{w};
		\draw[->](f) to (indir.north);
		\ltgnode[below=of indir]{f@fr}{@}; \addPrefix{f@fr}{v};
		\draw[->](indir) to (f@fr);
		\node[below=of f@fr](ffr){};
		\ltgnode[left=of ffr]{l}{\0}; \addPrefix{l}{v};
		\draw[->](f@fr) to (l);
		\node[left=of l](lwest){}; \draw(l) to (lwest.center); \draw[->](lwest.center) |- (f);
		\ltgnode[right=of ffr]{f@r}{@}; \addPrefix{f@r}{v};
		\node[below=of f@r](fr){};
		\draw[->](f@fr) to (f@r.north);
		\ltgnode[left=of fr]{f1}{\0}; \addPrefix{f1}{v};
		\draw[->](f@r) to (f1);
		\draw(f1) -| (lwest.center);
		\node[right=of fr](r){};
		\draw(f@r) to (r.center);
		\node[node distance=1mm,right=of r](reast){};
		\draw(r.center) to (reast.center);
		\draw[->](reast.center) |- (indir.east);
	}
	\transstep{\eraseIndirectionVertices}{
		\ltgnode{f}{λ}; \addPos{f}{v}; \addPrefix{f}{};
		\draw[<-](f.north) to +(0mm,4mm);
		\node[below=of f](in){};
		\ltgnode[below=of f]{f@fr}{@}; \addPrefix{f@fr}{v};
		\draw[->](f) to (f@fr);
		\node[below=of f@fr](ffr){};
		\ltgnode[left=of ffr]{l}{\0}; \addPrefix{l}{v};
		\draw[->](f@fr) to (l);
		\node[left=of l](lwest){}; \draw(l) to (lwest.center); \draw[->](lwest.center) |- (f);
		\ltgnode[right=of ffr]{f@r}{@}; \addPrefix{f@r}{v};
		\node[below=of f@r](fr){};
		\draw[->](f@fr) to (f@r.north);
		\ltgnode[left=of fr]{f1}{\0}; \addPrefix{f1}{v};
		\draw[->](f@r) to (f1);
		\draw(f1) -| (lwest.center);
		\node[right=of fr](r){};
		\draw(f@r) to (r.center);
		\node[node distance=1mm,right=of r](reast){};
		\draw(r.center) to (reast.center);
		\draw[->](reast.center) |- (f@fr.east);
	}
\end{transenv}

\begin{para}\label{translation:translations}
	We continue with the term  \m{\L_2} from \cref{ex:translations}. We translate it in two different ways correspoding to the first-order term graph semantics \graphsemCmin{\classltgs}{} from \cref{def:graphsem-min:ltgs} and \graphsemC{\classltgs}{} from \cref{def:graphsem:ltgs}, respectively. Both sequences of translation steps yield the same λ-ap-ho-term-graph \graphsemC{\classlhotgs}{\L_2},\footnote{This is due to \cref{thm:corr:aphotgs:ltgs}, \cref{prop:graphsem:ltgs:min:graphsem:ltgs}, and \cref{prop:graphsem:ltgs:min:graphsem:ltgs:S}.} but we obtain two different λ-term-graphs \graphsemCmin{\classltgs}{\L_2} and \graphsemC{\classltgs}{\L_2} as in \cref{fig:translations}.
\end{para}

\begin{para}[dotted \S-vertices]
	Since we want to illustrate translations w.r.t.\ \rulestranslambdaletreccaltolhotgs which produces λ-ap-ho-term-graphs (which do not contain \S-vertices) and w.r.t.\ \rulestranslambdaletreccaltoltgs which produces λ-term-graphs (which contain \S-vertices) at the same time, in the graphs \S-vertices drawn with dotted lines.
\end{para}

\begin{para}[prefix lengths in the \Let-rule]
	Both translations are eager-scope (i.e. applications of \S-rules are given priority, and prefixes lengths are chosen small enough in the \Let-rule, see \cref{short-enough-prefixes}) but differ in how they resolve the non-determinism due to different possible choices for the prefix lengths \m{l_1,\dots,l_k} in the \Let-rule.
\end{para}

\begin{para}[translation 1: no \S-sharing]\label{translation-1}
	First we consider the translations of \m{\L_2} with \graphsemC{\classlhotgs}{} and \graphsemCmin{\classltgs}{}, i.e.\ the translation process in which prefix lengths are chosen minimally when applying the \Let-rule resulting in no sharing of \S-vertices.
\end{para}

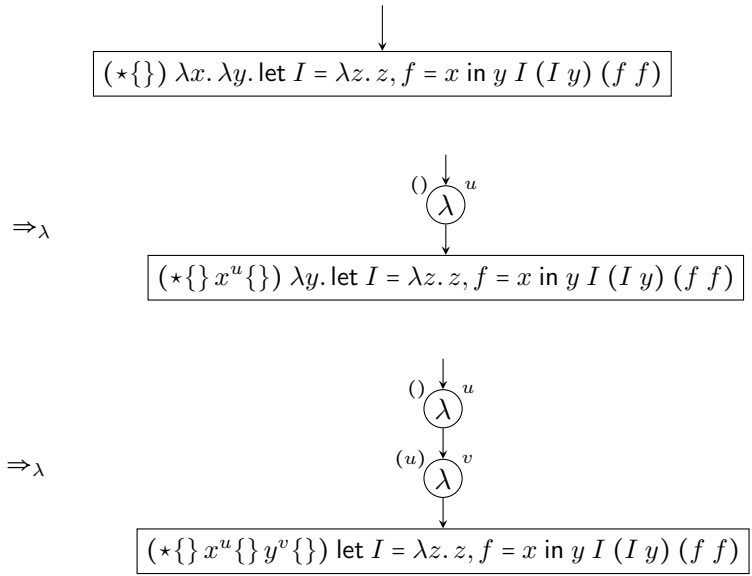
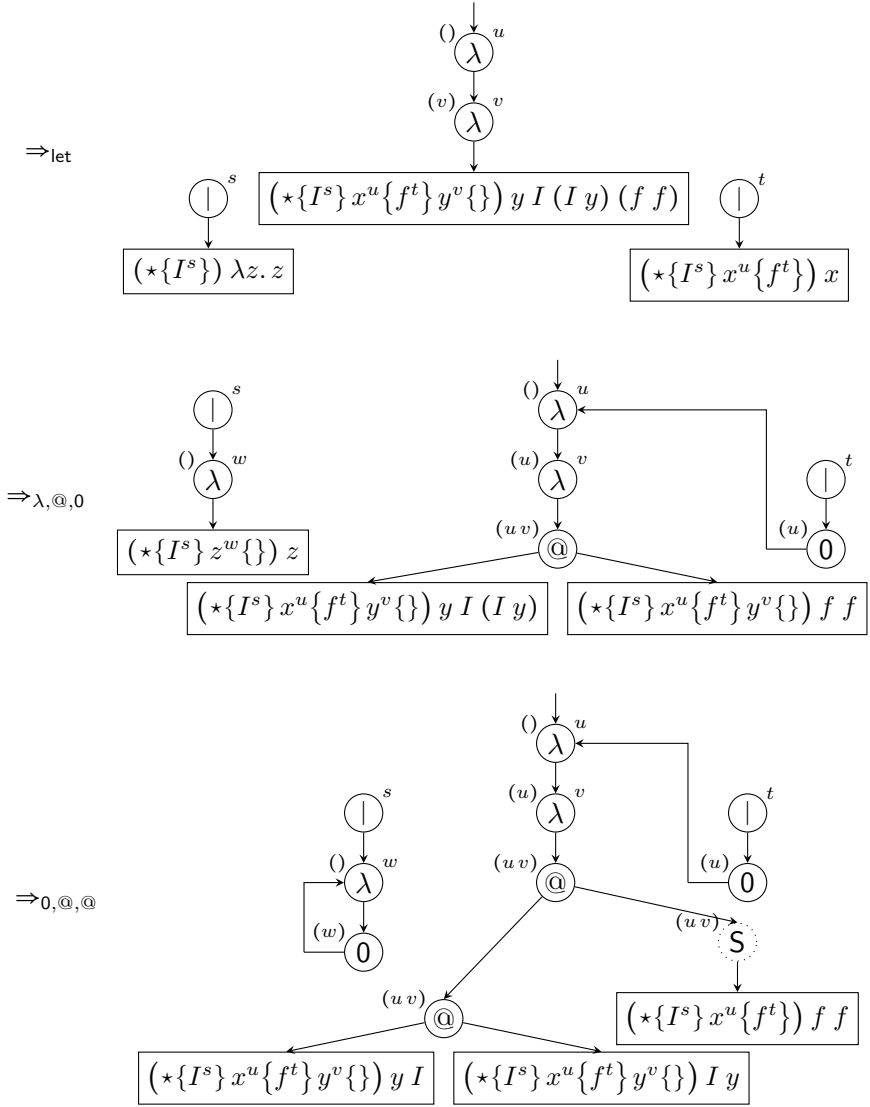
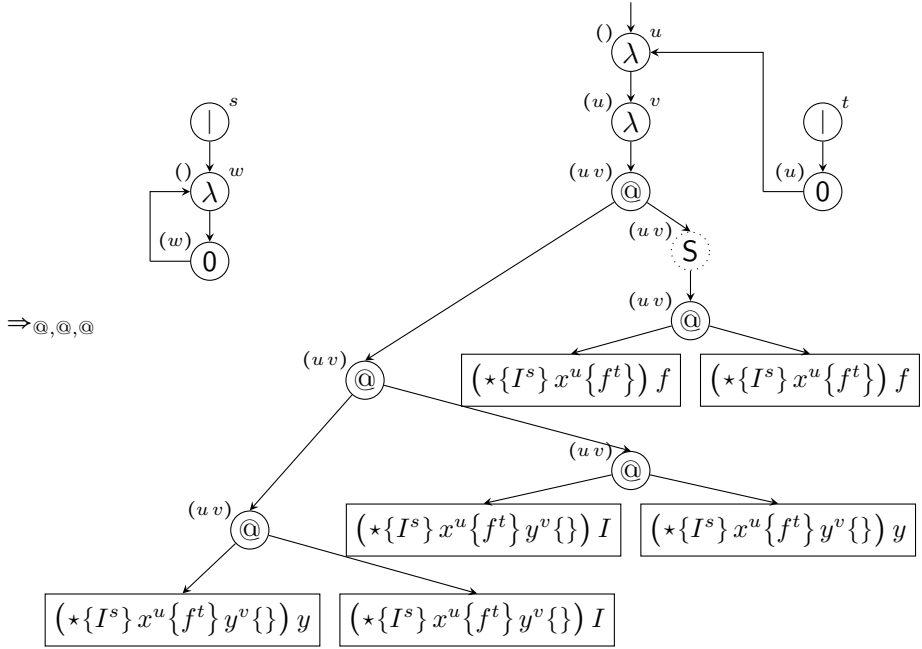
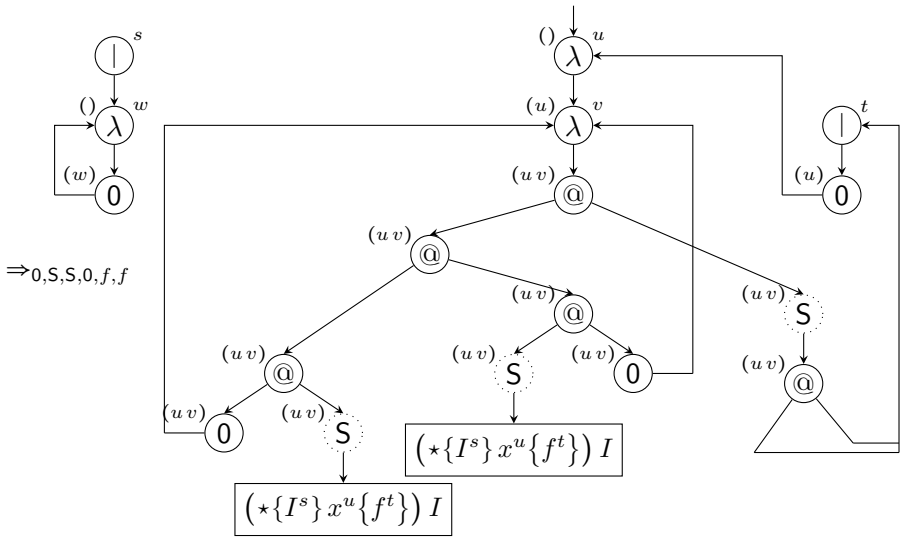
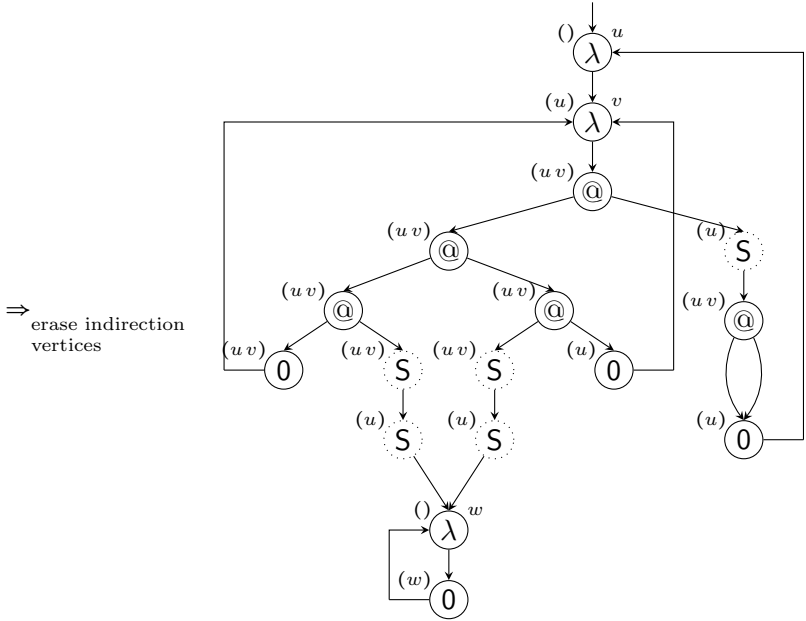
\newcommand\vx{u}
\newcommand\vy{v}
\newcommand\vz{w}
\newcommand\vI{s}
\newcommand\vf{t}

\begin{transenv}
	\transinit{
		\node[draw,shape=transbox](t){\recprefixed{\starfs{}}{\abs{x}{\abs{y}{\letin{I=\abs{z}{z},f=x}{\app{\app{\app{y}{I}}{(\app{I}{y})}}{(\app{f}{f})}}}}}};
		\draw[<-](t.north) to +(0mm,6mm);
	}
	\transstep{λ}{
		\node[draw,shape=transbox](t){\recprefixed{\starfs{}\prefixcon x^\vx\fs{}}{\abs{y}{\letin{I=\abs{z}{z},f=x}{\app{\app{\app{y}{I}}{(\app{I}{y})}}{(\app{f}{f})}}}}};
		\ltgnode[above=of t.north]{x}{λ}; \addPos{x}{\vx}; \addPrefix{x}{};
		\draw[<-](x.north) to +(0mm,4mm);
		\draw[->](x) to (t.north);
	}
	\transstep{λ}{
		\node[draw,shape=transbox](t){\recprefixed{\starfs{}\prefixcon x^\vx\fs{}\prefixcon y^\vy\fs{}}{\letin{I=\abs{z}{z},f=x}{\app{\app{\app{y}{I}}{(\app{I}{y})}}{(\app{f}{f})}}}};
		\ltgnode[above=of t.north]{y}{λ}; \addPos{y}{\vy}; \addPrefix{y}{\vx};
		\draw[->](y) to (t.north);
		\ltgnode[above=of y.north]{x}{λ}; \addPos{x}{\vx}; \addPrefix{x}{};
		\draw[->](x) to (y.north);
		\draw[<-](x.north) to +(0mm,4mm);
	}
	\transstep{\Let}{
		\node[draw,shape=transbox](yIIyff){\recprefixed{\starfs{I^\vI\commentOut{=\abs{z}{z}}}\prefixcon x^\vx\fs{f^\vf\commentOut{=x}}\prefixcon y^\vy\fs{ }}{\app{\app{\app{y}{I}}{(\app{I}{y})}}{(\app{f}{f})}}};
		\ltgnode[above=of yIIyff.north]{y}{λ}; \addPos{y}{\vy}; \addPrefix{y}{\vy};
		\draw[->](y) to (yIIyff.north);
		\ltgnode[above=of y.north]{x}{λ}; \addPos{x}{\vx}; \addPrefix{x}{};
		\draw[->](x) to (y.north);
		\draw[<-](x.north) to +(0mm,4mm);
		\ltgnode[left=of yIIyff]{Iindir}{\indir}; \addPos{Iindir}{\vI};
		\node[draw,shape=transbox,below=of Iindir](I){\recprefixed{\starfs{I^\vI\commentOut{=\abs{z}{z}}}}{\abs{z}{z}}};
		\draw[->](Iindir) to (I.north);
		\ltgnode[right=of yIIyff]{findir}{\indir}; \addPos{findir}{\vf};
		\node[draw,shape=transbox,below=of findir](f){\recprefixed{\starfs{I^\vI\commentOut{=\abs{z}{z}}}\prefixcon x^\vx\fs{f^\vf\commentOut{=x}}}{x}};
		\draw[->](findir) to (f.north);
	}
	\transstep{λ,@,\0}{
		\node(yIIyff){};
		\node[node distance=0mm,draw,shape=transbox,left=of yIIyff](yIIy){\recprefixed{\starfs{I^\vI\commentOut{=\abs{z}{z}}}\prefixcon x^\vx\fs{f^\vf\commentOut{=x}}\prefixcon y^\vy\fs{}}{\app{\app{y}{I}}{(\app{I}{y})}}};
		\node[node distance=0mm,draw,shape=transbox,right=of yIIyff](ff){\recprefixed{\starfs{I^\vI\commentOut{=\abs{z}{z}}}\prefixcon x^\vx\fs{f^\vf\commentOut{=x}}\prefixcon y^\vy\fs{}}{\app{f}{f}}};
		\ltgnode[above=of yIIyff.north]{yIIy@ff}{@}; \addPrefix{yIIy@ff}{\vx\prefixcon \vy};
		\draw[->](yIIy@ff) to (yIIy.north); \draw[->](yIIy@ff) to (ff.north);
		\ltgnode[above=of yIIy@ff.north]{y}{λ}; \addPos{y}{\vy}; \addPrefix{y}{\vx};
		\draw[->](y) to (yIIy@ff.north);
		\ltgnode[above=of y.north]{x}{λ}; \addPos{x}{\vx}; \addPrefix{x}{};
		\draw[->](x) to (y.north);
		\draw[<-](x.north) to +(0mm,4mm);
		\ltgnode[node distance=4cm,left=of y]{z}{λ}; \addPos{z}{\vz}; \addPrefix{z}{};
		\ltgnode[above=of z]{Iindir}{\indir}; \addPos{Iindir}{\vI};
		\draw[->](Iindir) to (z.north);
		\node[draw,shape=transbox,below=of z](I){\recprefixed{\starfs{I^\vI\commentOut{=\abs{z}{z}}}\prefixcon z^\vz\fs{}}{z}};
		\draw[->](z) to (I.north);
		\ltgnode[node distance=3cm,right=of yIIy@ff]{f}{\0}; \addPrefix{f}{\vx};
		\ltgnode[above=of f]{findir}{\indir}; \addPos{findir}{\vf};
		\draw[->](findir) to (f.north);
		\node[left=of f](fwest){}; \draw(f.west) to (fwest.center); \draw[->] (fwest.center) |- (x.east);
	}
	\transstep{\0, @, @}{
		\node(yIIyff){};
		\ltgnode[dotted,node distance=2cm,right=of yIIyff]{Sff}{\S}; \addPrefix{Sff}{\vx\prefixcon \vy};
		\node[draw,shape=transbox,below=of Sff](ff){\recprefixed{\starfs{I^\vI\commentOut{=\abs{z}{z}}}\prefixcon x^\vx\fs{f^\vf\commentOut{=x}}}{\app{f}{f}}};
		\draw[->](Sff) to (ff.north);
		\ltgnode[node distance=2cm,left=of ff]{yI@Iy}{@}; \addPrefix{yI@Iy}{\vx\prefixcon \vy};
		\node[below=of yI@Iy](yIIy){};
		\node[node distance=0mm,draw,shape=transbox, left=of yIIy](yI){\recprefixed{\starfs{I^\vI\commentOut{=\abs{z}{z}}}\prefixcon x^\vx\fs{f^\vf\commentOut{=x}}\prefixcon y^\vy\fs{}}{\app{y}{I}}};
		\draw[->](yI@Iy) to (yI.north);
		\node[node distance=0mm,draw,shape=transbox,right=of yIIy](Iy){\recprefixed{\starfs{I^\vI\commentOut{=\abs{z}{z}}}\prefixcon x^\vx\fs{f^\vf\commentOut{=x}}\prefixcon y^\vy\fs{}}{\app{I}{y}}};
		\draw[->](yI@Iy) to (Iy.north);
		\ltgnode[above=of yIIyff.north]{yIIy@ff}{@}; \addPrefix{yIIy@ff}{\vx\prefixcon \vy};
		\draw[->](yIIy@ff) to (yI@Iy.north); \draw[->](yIIy@ff) to (Sff.north);
		\ltgnode[above=of yIIy@ff.north]{y}{λ}; \addPos{y}{\vy}; \addPrefix{y}{\vx};
		\draw[->](y) to (yIIy@ff.north);
		\ltgnode[above=of y.north]{x}{λ}; \addPos{x}{\vx}; \addPrefix{x}{};
		\draw[->](x) to (y.north);
		\draw[<-](x.north) to +(0mm,4mm);
		\ltgnode[node distance=2cm,left=of y]{Iindir}{\indir}; \addPos{Iindir}{\vI};
		\ltgnode[below=of Iindir]{z}{λ}; \addPos{z}{\vz}; \addPrefix{z}{};
		\draw[->](Iindir) to (z.north);
		\ltgnode[below=of z]{I}{\0}; \addPrefix{I}{\vz};
		\node[left=of I](Iwest){}; \draw(I.west) to (Iwest.center); \draw[->] (Iwest.center) |- (z.west);
		\draw[->](z) to (I.north);
		\ltgnode[node distance=2cm,right=of y]{findir}{\indir}; \addPos{findir}{\vf};
		\ltgnode[below=of findir]{f}{\0}; \addPrefix{f}{\vx};
		\draw[->](findir) to (f.north);
		\node[left=of f](fwest){}; \draw(f.west) to (fwest.center); \draw[->] (fwest.center) |- (x.east);
	}
	\transstep{\mathrlap{@,@,@}}{
		\node(yIIyff){};
		\ltgnode[dotted,right=of yIIyff]{Sff}{\S}; \addPrefix{Sff}{\vx\prefixcon \vy};
		\ltgnode[below=of Sff]{f@f}{@}; \addPrefix{f@f}{\vx\prefixcon \vy}; \node[below=of f@f](ff){};
		\draw[->](Sff) to (f@f.north);
		\node[node distance=0mm,draw,shape=transbox, left=of ff](f1){\recprefixed{\starfs{I^\vI\commentOut{=\abs{z}{z}}}\prefixcon x^\vx\fs{f^\vf\commentOut{=x}}}{f}};
		\draw[->](f@f) to (f1.north);
		\node[node distance=0mm,draw,shape=transbox,right=of ff](f2){\recprefixed{\starfs{I^\vI\commentOut{=\abs{z}{z}}}\prefixcon x^\vx\fs{f^\vf\commentOut{=x}}}{f}};
		\draw[->](f@f) to (f2.north);
		\ltgnode[node distance=1cm,left=of f1]{yI@Iy}{@}; \addPrefix{yI@Iy}{\vx\prefixcon \vy};
		\ltgnode[node distance=2.5cm,below=of yIIyff]{I@y}{@}; \addPrefix{I@y}{\vx\prefixcon \vy}; \node[below=of I@y](Iy){};
		\draw[->](yI@Iy) to (I@y.north);
		\node[node distance=0mm,draw,shape=transbox, left=of Iy](I2){\recprefixed{\starfs{I^\vI\commentOut{=\abs{z}{z}}}\prefixcon x^\vx\fs{f^\vf\commentOut{=x}}\prefixcon y^\vy\fs{}}{I}};
		\draw[->](I@y) to (I2.north);
		\node[node distance=0mm,draw,shape=transbox,right=of Iy](y2){\recprefixed{\starfs{I^\vI\commentOut{=\abs{z}{z}}}\prefixcon x^\vx\fs{f^\vf\commentOut{=x}}\prefixcon y^\vy\fs{}}{y}};
		\draw[->](I@y) to (y2.north);
		\ltgnode[node distance=1cm, left=of I2]{y@I}{@}; \addPrefix{y@I}{\vx\prefixcon \vy};
		\node[node distance=8mm,below=of y@I](y@@I){}; \node[node distance=8mm,right=of y@@I](yI){};
		\draw[->](yI@Iy) to (y@I.north);
		\node[node distance=0mm,draw,shape=transbox, left=of yI](y1){\recprefixed{\starfs{I^\vI\commentOut{=\abs{z}{z}}}\prefixcon x^\vx\fs{f^\vf\commentOut{=x}}\prefixcon y^\vy\fs{}}{y}};
		\draw[->](y@I) to (y1.north);
		\node[node distance=0mm,draw,shape=transbox, right=of yI](I1){\recprefixed{\starfs{I^\vI\commentOut{=\abs{z}{z}}}\prefixcon x^\vx\fs{f^\vf\commentOut{=x}}\prefixcon y^\vy\fs{}}{I}};
		\draw[->](y@I) to (I1.north);
		\ltgnode[above=of yIIyff]{yIIy@ff}{@}; \addPrefix{yIIy@ff}{\vx\prefixcon \vy};
		\draw[->](yIIy@ff) to (yI@Iy.north); \draw[->](yIIy@ff) to (Sff.north);
		\ltgnode[above=of yIIy@ff]{y}{λ}; \addPos{y}{\vy}; \addPrefix{y}{\vx};
		\draw[->](y) to (yIIy@ff.north);
		\ltgnode[above=of y.north]{x}{λ}; \addPos{x}{\vx}; \addPrefix{x}{};
		\draw[->](x) to (y.north);
		\draw[<-](x.north) to +(0mm,4mm);
		\ltgnode[node distance=5cm,left=of y]{Iindir}{\indir}; \addPos{Iindir}{\vI};
		\ltgnode[below=of Iindir]{z}{λ}; \addPos{z}{\vz}; \addPrefix{z}{};
		\draw[->](Iindir) to (z.north);
		\ltgnode[below=of z]{I}{\0}; \addPrefix{I}{\vz};
		\node[left=of I](Iwest){}; \draw(I.west) to (Iwest.center); \draw[->] (Iwest.center) |- (z.west);
		\draw[->](z) to (I.north);
		\ltgnode[node distance=2cm,right=of y]{findir}{\indir}; \addPos{findir}{\vf};
		\ltgnode[below=of findir]{f}{\0}; \addPrefix{f}{\vx};
		\draw[->](findir) to (f.north);
		\node[left=of f](fwest){}; \draw(f.west) to (fwest.center); \draw[->] (fwest.center) |- (x.east);
	}
	\transstep{\mathrlap{\0, \S, \S, \0, f, f}}{
		\node(yIIyff){};
		\ltgnode[node distance=1.5cm,left=of yIIyff]{yI@Iy}{@}; \addPrefix{yI@Iy}{\vx\prefixcon \vy}; \node[below=of yI@Iy](yIIy){};
		\ltgnode[node distance=1.5cm,right=of yIIy]{I@y}{@}; \addPrefix{I@y}{\vx\prefixcon \vy}; \node[below=of I@y](Iy){};
		\draw[->](yI@Iy) to (I@y.north);
		\ltgnode[dotted,left=of Iy]{SI2}{\S}; \addPrefix{SI2}{\vx\prefixcon \vy};
		\draw[->](I@y) to (SI2.north);
		\node[draw,shape=transbox, below=of SI2](I2){\recprefixed{\starfs{I^\vI\commentOut{=\abs{z}{z}}}\prefixcon x^\vx\fs{f^\vf\commentOut{=x}}}{I}};
		\draw[->](SI2) to (I2.north);
		\ltgnode[right=of Iy]{y2}{\0}; \addPrefix{y2}{\vx\prefixcon \vy};
		\node[right=of y2](y2l){}; \draw(y2) to (y2l.center);
		\draw[->](I@y) to (y2.north);
		\ltgnode[node distance=2.5cm, left=of SI2]{y@I}{@}; \addPrefix{y@I}{\vx\prefixcon \vy}; \node[below=of y@I](yI){};
		\draw[->](yI@Iy) to (y@I.north);
		\ltgnode[left=of yI]{y1}{\0}; \addPrefix{y1}{\vx\prefixcon \vy};
		\node[left=of y1](y1l){}; \draw(y1) to (y1l.center);
		\draw[->](y@I) to (y1.north);
		\ltgnode[dotted,right=of yI]{SI1}{\S}; \addPrefix{SI1}{\vx\prefixcon \vy};
		\draw[->](y@I) to (SI1.north);
		\node[draw,shape=transbox, below=of SI1](I1){\recprefixed{\starfs{I^\vI\commentOut{=\abs{z}{z}}}\prefixcon x^\vx\fs{f^\vf\commentOut{=x}}}{I}};
		\draw[->](SI1) to (I1.north);
		\ltgnode[dotted,node distance=2.5cm,right=of I@y]{Sff}{\S}; \addPrefix{Sff}{\vx\prefixcon \vy};
		\ltgnode[below=of Sff]{f@f}{@}; \addPrefix{f@f}{\vx\prefixcon \vy}; \node[below=of f@f](ff){};
		\draw[->](Sff) to (f@f.north);
		\node[node distance=1cm,right=of ff](fmerge){};
		\node[ left=of ff](f1){}; \draw(f@f) to (f1.south); \draw(f1.south) |- (fmerge.south);
		\node[right=of ff](f2){}; \draw(f@f) to (f2.center); \draw(f2.center) to (fmerge.center);
		\ltgnode[above=of yIIyff.north]{yIIy@ff}{@}; \addPrefix{yIIy@ff}{\vx\prefixcon \vy};
		\draw[->](yIIy@ff) to (yI@Iy.north); \draw[->](yIIy@ff) to (Sff.north);
		\ltgnode[above=of yIIy@ff.north]{y}{λ}; \addPos{y}{\vy}; \addPrefix{y}{\vx};
		\draw[->](y2l.center) |- (y.east); \draw[->](y1l.center) |- (y.west);
		\draw[->](y) to (yIIy@ff.north);
		\ltgnode[above=of y.north]{x}{λ}; \addPos{x}{\vx}; \addPrefix{x}{};
		\draw[->](x) to (y.north);
		\draw[<-](x.north) to +(0mm,4mm);
		\ltgnode[node distance=5.5cm,left=of x]{Iindir}{\indir}; \addPos{Iindir}{\vI};
		\ltgnode[below=of Iindir]{z}{λ}; \addPos{z}{\vz}; \addPrefix{z}{};
		\draw[->](Iindir) to (z.north);
		\ltgnode[below=of z]{I}{\0}; \addPrefix{I}{\vz};
		\node[left=of I](Iwest){}; \draw(I.west) to (Iwest.center); \draw[->] (Iwest.center) |- (z.west);
		\draw[->](z) to (I.north);
		\ltgnode[node distance=3cm,right=of y]{findir}{\indir}; \addPos{findir}{\vf};
		\ltgnode[below=of findir]{f}{\0}; \addPrefix{f}{\vx};
		\draw[->](findir) to (f.north);
		\node[left=of f](fwest){}; \draw(f.west) to (fwest.center); \draw[->] (fwest.center) |- (x.east);
		\draw[->](fmerge.south) |- (findir.east);
	}
	\transstep{\S,\S}{
		\node(yIIyff){};
		\ltgnode[node distance=1cm,left=of yIIyff]{yI@Iy}{@}; \addPrefix{yI@Iy}{\vx\prefixcon \vy}; \node[below=of yI@Iy](yIIy){};
		\ltgnode[node distance=1cm,right=of yIIy]{I@y}{@}; \addPrefix{I@y}{\vx\prefixcon \vy}; \node[below=of I@y](Iy){};
		\draw[->](yI@Iy) to (I@y.north);
		\ltgnode[dotted,left=of Iy]{SI2}{\S}; \addPrefix{SI2}{\vx\prefixcon \vy};
		\draw[->](I@y) to (SI2.north);
		\ltgnode[dotted,below=of SI2]{SSI2}{\S}; \addPrefix{SSI2}{\vx};
		\draw[->](SI2) to (SSI2.north);
		\node[draw,shape=transbox, below=of SSI2](I2){\recprefixed{\starfs{I^\vI\commentOut{=\abs{z}{z}}}}{I}};
		\draw[->](SSI2) to (I2.north);
		\ltgnode[right=of Iy]{y2}{\0}; \addPrefix{y2}{\vx};
		\node[right=of y2](y2l){}; \draw(y2) to (y2l.center);
		\draw[->](I@y) to (y2.north);
		\ltgnode[node distance=2.2cm, left=of SI2]{y@I}{@}; \addPrefix{y@I}{\vx\prefixcon \vy}; \node[below=of y@I](yI){};
		\draw[->](yI@Iy) to (y@I.north);
		\ltgnode[left=of yI]{y1}{\0}; \addPrefix{y1}{\vx\prefixcon \vy};
		\node[left=of y1](y1l){}; \draw(y1) to (y1l.center);
		\draw[->](y@I) to (y1.north);
		\ltgnode[dotted,right=of yI]{SI1}{\S}; \addPrefix{SI1}{\vx\prefixcon \vy};
		\draw[->](y@I) to (SI1.north);
		\ltgnode[dotted,below=of SI1]{SSI1}{\S}; \addPrefix{SSI1}{\vx};
		\draw[->](SI1) to (SSI1.north);
		\node[draw,shape=transbox, below=of SSI1](I1){\recprefixed{\starfs{I^\vI\commentOut{=\abs{z}{z}}}}{I}};
		\draw[->](SSI1) to (I1.north);
		\ltgnode[dotted,node distance=2cm,right=of I@y]{Sff}{\S}; \addPrefix{Sff}{\vx};
		\ltgnode[below=of Sff]{f@f}{@}; \addPrefix{f@f}{\vx\prefixcon \vy}; \node[below=of f@f](ff){};
		\draw[->](Sff) to (f@f.north);
		\node[node distance=1cm,right=of ff](fmerge){};
		\node[ left=of ff](f1){}; \draw(f@f) to (f1.south); \draw(f1.south) |- (fmerge.south);
		\node[right=of ff](f2){}; \draw(f@f) to (f2.center); \draw(f2.center) to (fmerge.center);
		\ltgnode[above=of yIIyff.north]{yIIy@ff}{@}; \addPrefix{yIIy@ff}{\vx\prefixcon \vy};
		\draw[->](yIIy@ff) to (yI@Iy.north); \draw[->](yIIy@ff) to (Sff.north);
		\ltgnode[above=of yIIy@ff.north]{y}{λ}; \addPos{y}{\vy}; \addPrefix{y}{\vx};
		\draw[->](y2l.center) |- (y.east); \draw[->](y1l.center) |- (y.west);
		\draw[->](y) to (yIIy@ff.north);
		\ltgnode[above=of y.north]{x}{λ}; \addPos{x}{\vx}; \addPrefix{x}{};
		\draw[->](x) to (y.north);
		\draw[<-](x.north) to +(0mm,4mm);
		\ltgnode[node distance=5.2cm,left=of x]{Iindir}{\indir}; \addPos{Iindir}{\vI};
		\ltgnode[below=of Iindir]{z}{λ}; \addPos{z}{\vz}; \addPrefix{z}{};
		\draw[->](Iindir) to (z.north);
		\ltgnode[below=of z]{I}{\0}; \addPrefix{I}{\vz};
		\node[left=of I](Iwest){}; \draw(I.west) to (Iwest.center); \draw[->] (Iwest.center) |- (z.west);
		\draw[->](z) to (I.north);
		\ltgnode[node distance=2.3cm,right=of y]{findir}{\indir}; \addPos{findir}{\vf};
		\ltgnode[below=of findir]{f}{\0}; \addPrefix{f}{\vx};
		\draw[->](findir) to (f.north);
		\node[left=of f](fwest){}; \draw(f.west) to (fwest.center); \draw[->] (fwest.center) |- (x.east);
		\draw[->](fmerge.south) |- (findir.east);
	}
	\transstep{f,f}{
		\node(yIIyff){};
		\ltgnode[node distance=1.5cm,left=of yIIyff]{yI@Iy}{@}; \addPrefix{yI@Iy}{\vx\prefixcon \vy}; \node[below=of yI@Iy](yIIy){};
		\ltgnode[node distance=1cm,right=of yIIy]{I@y}{@}; \addPrefix{I@y}{\vx\prefixcon \vy}; \node[below=of I@y](Iy){};
		\draw[->](yI@Iy) to (I@y.north);
		\ltgnode[dotted,left=of Iy]{SI2}{\S}; \addPrefix{SI2}{\vx\prefixcon \vy};
		\draw[->](I@y) to (SI2.north);
		\ltgnode[dotted,below=of SI2]{SSI2}{\S}; \addPrefix{SSI2}{\vx};
		\draw[->](SI2) to (SSI2.north);
		\ltgnode[right=of Iy]{y2}{\0}; \addPrefix{y2}{\vx};
		\node[right=of y2](y2l){}; \draw(y2) to (y2l.center);
		\draw[->](I@y) to (y2.north);
		\ltgnode[node distance=1cm, left=of yIIy]{y@I}{@}; \addPrefix{y@I}{\vx\prefixcon \vy}; \node[below=of y@I](yI){};
		\draw[->](yI@Iy) to (y@I.north);
		\ltgnode[left=of yI]{y1}{\0}; \addPrefix{y1}{\vx\prefixcon \vy};
		\node[left=of y1](y1l){}; \draw(y1) to (y1l.center);
		\draw[->](y@I) to (y1.north);
		\ltgnode[dotted,right=of yI]{SI1}{\S}; \addPrefix{SI1}{\vx\prefixcon \vy};
		\draw[->](y@I) to (SI1.north);
		\ltgnode[dotted,below=of SI1]{SSI1}{\S}; \addPrefix{SSI1}{\vx};
		\draw[->](SI1) to (SSI1.north);
		\node[below=of SSI1](Imerge){};
		\ltgnode[dotted,node distance=2cm,right=of I@y]{Sff}{\S}; \addPrefix{Sff}{\vx};
		\ltgnode[below=of Sff]{f@f}{@}; \addPrefix{f@f}{\vx\prefixcon \vy}; \node[below=of f@f](ff){};
		\draw[->](Sff) to (f@f.north);
		\node[node distance=1cm,right=of ff](fmerge){};
		\node[ left=of ff](f1){}; \draw(f@f) to (f1.south); \draw(f1.south) |- (fmerge.south);
		\node[right=of ff](f2){}; \draw(f@f) to (f2.center); \draw(f2.center) to (fmerge.center);
		\ltgnode[above=of yIIyff.north]{yIIy@ff}{@}; \addPrefix{yIIy@ff}{\vx\prefixcon \vy};
		\draw[->](yIIy@ff) to (yI@Iy.north); \draw[->](yIIy@ff) to (Sff.north);
		\ltgnode[above=of yIIy@ff]{y}{λ}; \addPos{y}{\vy}; \addPrefix{y}{\vx};
		\draw[->](y2l.center) |- (y.east); \draw[->](y1l.center) |- (y.west);
		\draw[->](y) to (yIIy@ff.north);
		\ltgnode[above=of y]{x}{λ}; \addPos{x}{\vx}; \addPrefix{x}{};
		\draw[->](x) to (y.north);
		\draw[<-](x.north) to +(0mm,4mm);
		\ltgnode[node distance=5.5cm,left=of y]{Iindir}{\indir}; \addPos{Iindir}{\vI};
		\ltgnode[below=of Iindir]{z}{λ}; \addPos{z}{\vz}; \addPrefix{z}{};
		\draw[->](Iindir) to (z.north);
		\ltgnode[below=of z]{I}{\0}; \addPrefix{I}{\vz};
		\node[left=of I](Iwest){}; \draw(I.west) to (Iwest.center); \draw[->] (Iwest.center) |- (z.west);
		\draw[->](z) to (I.north);
		\ltgnode[node distance=1.8cm,right=of y]{findir}{\indir}; \addPos{findir}{\vf};
		\ltgnode[below=of findir]{f}{\0}; \addPrefix{f}{\vx};
		\draw[->](findir) to (f.north);
		\node[left=of f](fwest){}; \draw(f.west) to (fwest.center); \draw[->] (fwest.center) |- (x.east);
		\draw[->](fmerge.south) |- (findir.east);
		\draw(SSI1) to (Imerge.center);
		\draw(SSI2) |- (Imerge.center);
		\node[right=of Iindir](Iindireast){};
		\draw[->](Iindireast.center) to (Iindir.east);
		\draw(Imerge.center) -| (Iindireast.center);
	}
	\transstep{\eraseIndirectionVertices}{
		\node(yIIyff){};
		\ltgnode[node distance=1.5cm,left=of yIIyff]{yI@Iy}{@}; \addPrefix{yI@Iy}{\vx\prefixcon \vy}; \node[below=of yI@Iy](yIIy){};
		\ltgnode[node distance=1cm,right=of yIIy]{I@y}{@}; \addPrefix{I@y}{\vx\prefixcon \vy}; \node[below=of I@y](Iy){};
		\draw[->](yI@Iy) to (I@y.north);
		\ltgnode[dotted,left=of Iy]{SI2}{\S}; \addPrefix{SI2}{\vx\prefixcon \vy};
		\draw[->](I@y) to (SI2.north);
		\ltgnode[dotted,below=of SI2]{SSI2}{\S}; \addPrefix{SSI2}{\vx};
		\draw[->](SI2) to (SSI2.north);
		\ltgnode[right=of Iy]{y2}{\0}; \addPrefix{y2}{\vx};
		\node[right=of y2](y2l){}; \draw(y2) to (y2l.center);
		\draw[->](I@y) to (y2.north);
		\ltgnode[node distance=1cm, left=of yIIy]{y@I}{@}; \addPrefix{y@I}{\vx\prefixcon \vy}; \node[below=of y@I](yI){};
		\draw[->](yI@Iy) to (y@I.north);
		\ltgnode[left=of yI]{y1}{\0}; \addPrefix{y1}{\vx\prefixcon \vy};
		\node[left=of y1](y1l){}; \draw(y1) to (y1l.center);
		\draw[->](y@I) to (y1.north);
		\ltgnode[dotted,right=of yI]{SI1}{\S}; \addPrefix{SI1}{\vx\prefixcon \vy};
		\draw[->](y@I) to (SI1.north);
		\ltgnode[dotted,below=of SI1]{SSI1}{\S}; \addPrefix{SSI1}{\vx};
		\draw[->](SI1) to (SSI1.north);
		\ltgnode[dotted,node distance=1.6cm,right=of yIIyff]{Sff}{\S}; \addPrefix{Sff}{\vx};
		\ltgnode[below=of Sff]{f@f}{@}; \addPrefix{f@f}{\vx\prefixcon \vy}; \node[below=of f@f](ff){};
		\draw[->](Sff) to (f@f.north);
		\node[node distance=1.4cm,right=of ff](fmerge){};
		\ltgnode[above=of yIIyff]{yIIy@ff}{@}; \addPrefix{yIIy@ff}{\vx\prefixcon \vy};
		\draw[->](yIIy@ff) to (yI@Iy.north); \draw[->](yIIy@ff) to (Sff.north);
		\ltgnode[above=of yIIy@ff]{y}{λ}; \addPos{y}{\vy}; \addPrefix{y}{\vx};
		\draw[->](y2l.center) |- (y.east); \draw[->](y1l.center) |- (y.west);
		\draw[->](y) to (yIIy@ff.north);
		\ltgnode[above=of y]{x}{λ}; \addPos{x}{\vx}; \addPrefix{x}{};
		\draw[->](x) to (y.north);
		\draw[<-](x.north) to +(0mm,4mm);
		\ltgnode[node distance=2.5cm,below=of yIIy]{z}{λ}; \addPos{z}{\vz}; \addPrefix{z}{};
		\ltgnode[below=of z]{I}{\0}; \addPrefix{I}{\vz};
		\node[left=of I](Iwest){}; \draw(I.west) to (Iwest.center); \draw[->] (Iwest.center) |- (z.west);
		\draw[->](z) to (I.north);
		\ltgnode[below=of ff]{f}{\0}; \addPrefix{f}{\vx};
		\draw[->, bend right](f@f) to (f.north);
		\draw[->, bend left](f@f) to (f.north);
		\node[right=of f](feast){}; \draw(f.east) to (feast.center); \draw[->] (feast.center) |- (x.east);
		\draw[->](SSI1) to (z.north);
		\draw[->](SSI2) to (z.north);
	}
\end{transenv}

\begin{para}[translation 2: maximal \S-sharing]
	Second, we give the translation of the same term \m{\L_2} (from \cref{ex:translations}) with the process needed for the first-order term graph semantics \graphsemC{\classltgs}{}, yielding \graphsemC{\classltgs}{\L_2} in \cref{fig:def:graphsem:lhotgs}, and the corresponding λ-ap-ho-term-graph \graphsemC{\classlhotgs}{\L_2}. The obtained λ-term-graph \graphsemC{\classltgs}{\L_2} differs from the λ-term-graph \graphsemCmin{\classltgs}{\L_2} obtained in \cref{translation-1} by a higher degree of \S-sharing.
\end{para}

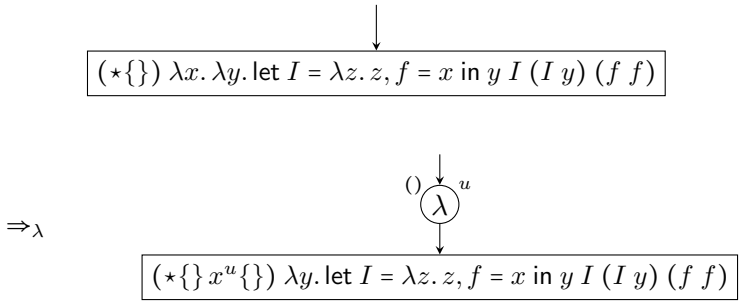
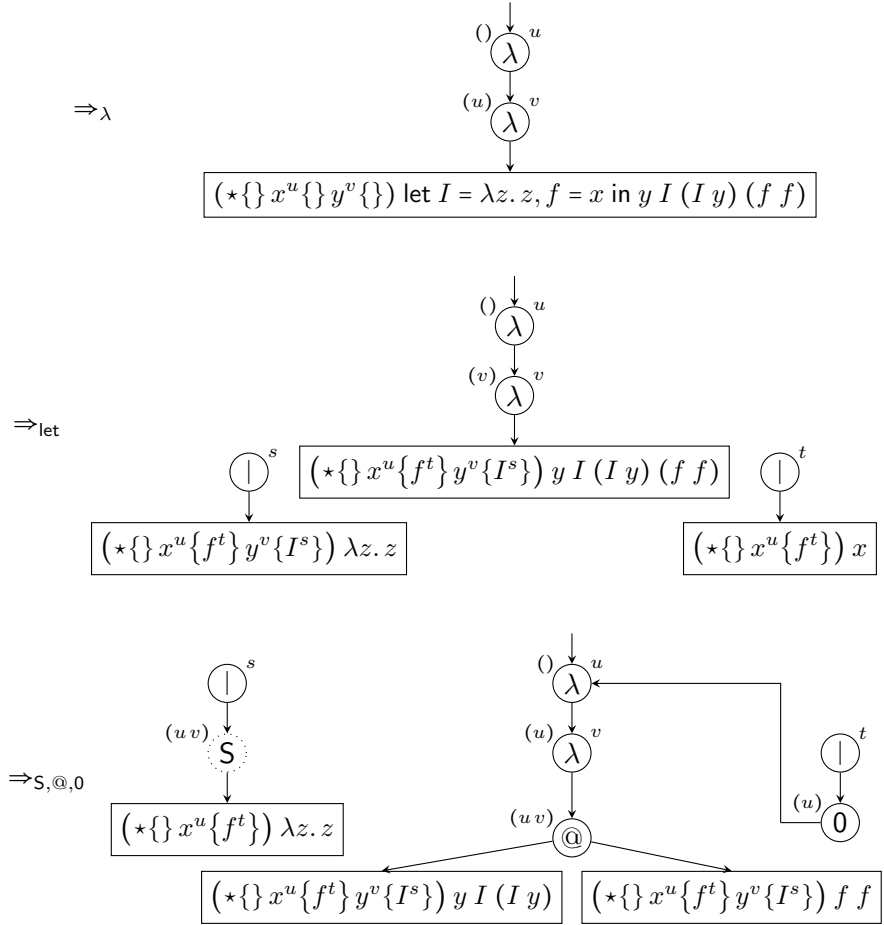
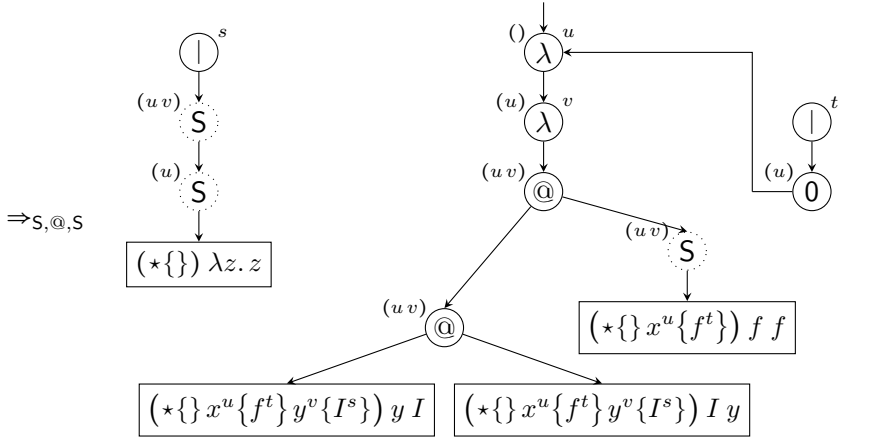
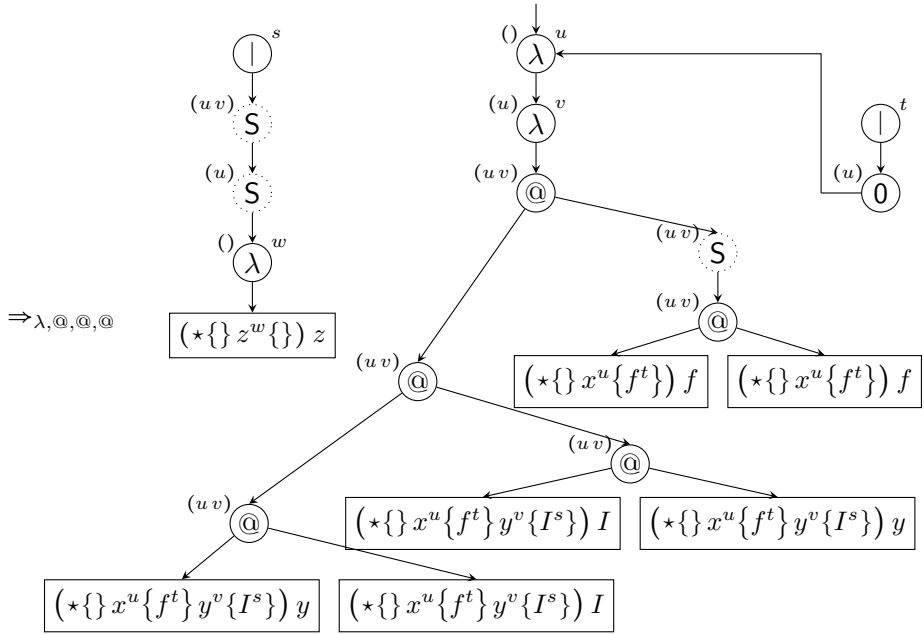
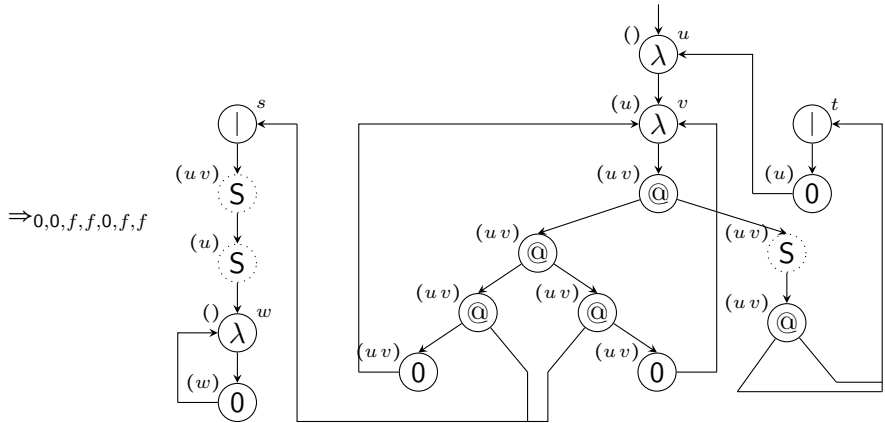
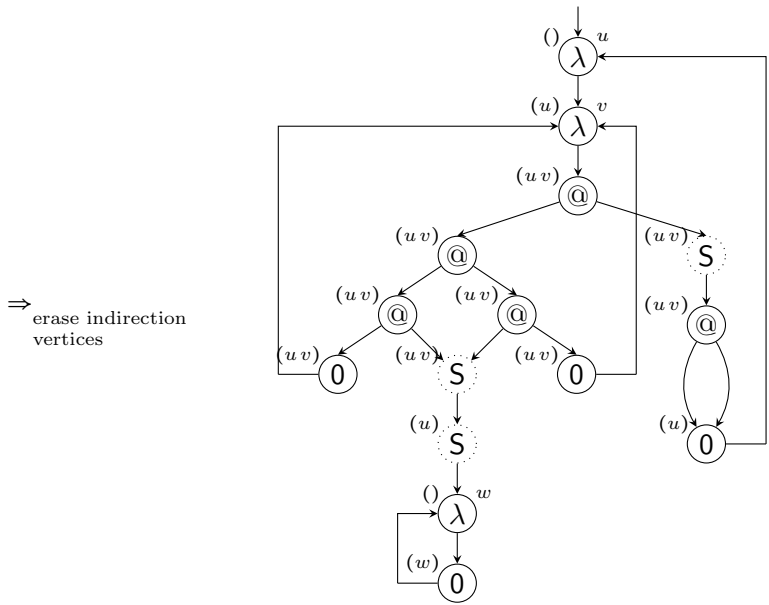
\begin{transenv}
	\transinit{
		\node[draw,shape=transbox](t){\recprefixed{\starfs{}}{\abs{x}{\abs{y}{\letin{I=\abs{z}{z},f=x}{\app{\app{\app{y}{I}}{(\app{I}{y})}}{(\app{f}{f})}}}}}};
		\draw[<-](t.north) to +(0mm,6mm);
	}
	\transstep{λ}{
		\node[draw,shape=transbox](t){\recprefixed{\starfs{}\prefixcon x^\vx\fs{}}{\abs{y}{\letin{I=\abs{z}{z},f=x}{\app{\app{\app{y}{I}}{(\app{I}{y})}}{(\app{f}{f})}}}}};
		\ltgnode[above=of t]{x}{λ}; \addPos{x}{\vx}; \addPrefix{x}{};
		\draw[<-](x.north) to +(0mm,4mm);
		\draw[->](x) to (t.north);
	}
	\transstep{λ}{
		\node[draw,shape=transbox](t){\recprefixed{\starfs{}\prefixcon x^\vx\fs{}\prefixcon y^\vy\fs{}}{\letin{I=\abs{z}{z},f=x}{\app{\app{\app{y}{I}}{(\app{I}{y})}}{(\app{f}{f})}}}};
		\ltgnode[above=of t]{y}{λ}; \addPos{y}{\vy}; \addPrefix{y}{\vx};
		\draw[->](y) to (t.north);
		\ltgnode[above=of y]{x}{λ}; \addPos{x}{\vx}; \addPrefix{x}{};
		\draw[->](x) to (y.north);
		\draw[<-](x.north) to +(0mm,4mm);
	}
	\transstep{\Let}{
		\node[draw,shape=transbox](yIIyff){\recprefixed{\starfs{}\prefixcon x^\vx\fs{f^\vf\commentOut{=x}}\prefixcon y^\vy\fs{I^\vI\commentOut{=\abs{z}{z}}}}{\app{\app{\app{y}{I}}{(\app{I}{y})}}{(\app{f}{f})}}};
		\ltgnode[above=of yIIyff]{y}{λ}; \addPos{y}{\vy}; \addPrefix{y}{\vy};
		\draw[->](y) to (yIIyff.north);
		\ltgnode[above=of y]{x}{λ}; \addPos{x}{\vx}; \addPrefix{x}{};
		\draw[->](x) to (y.north);
		\draw[<-](x.north) to +(0mm,4mm);
		\ltgnode[left=of yIIyff]{Iindir}{\indir}; \addPos{Iindir}{\vI};
		\node[draw,shape=transbox,below=of Iindir](I){\recprefixed{\starfs{}\prefixcon x^\vx\fs{f^\vf\commentOut{=x}}\prefixcon y^\vy\fs{I^\vI\commentOut{=\abs{z}{z}}}}{\abs{z}{z}}};
		\draw[->](Iindir) to (I.north);
		\ltgnode[right=of yIIyff]{findir}{\indir}; \addPos{findir}{\vf};
		\node[draw,shape=transbox,below=of findir](f){\recprefixed{\starfs{}\prefixcon x^\vx\fs{f^\vf\commentOut{=x}}}{x}};
		\draw[->](findir) to (f.north);
	}
	\transstep{\S,@,\0}{
		\node(yIIyff){};
		\node[node distance=0mm,draw,shape=transbox,left=of yIIyff](yIIy){\recprefixed{\starfs{}\prefixcon x^\vx\fs{f^\vf\commentOut{=x}}\prefixcon y^\vy\fs{I^\vI\commentOut{=\abs{z}{z}}}}{\app{\app{y}{I}}{(\app{I}{y})}}};
		\node[node distance=0mm,draw,shape=transbox,right=of yIIyff](ff){\recprefixed{\starfs{}\prefixcon x^\vx\fs{f^\vf\commentOut{=x}}\prefixcon y^\vy\fs{I^\vI\commentOut{=\abs{z}{z}}}}{\app{f}{f}}};
		\ltgnode[above=of yIIyff]{yIIy@ff}{@}; \addPrefix{yIIy@ff}{\vx\prefixcon \vy};
		\draw[->](yIIy@ff) to (yIIy.north); \draw[->](yIIy@ff) to (ff.north);
		\ltgnode[node distance=6mm,above=of yIIy@ff]{y}{λ}; \addPos{y}{\vy}; \addPrefix{y}{\vx};
		\draw[->](y) to (yIIy@ff.north);
		\ltgnode[above=of y]{x}{λ}; \addPos{x}{\vx}; \addPrefix{x}{};
		\draw[->](x) to (y.north);
		\draw[<-](x.north) to +(0mm,4mm);
		\ltgnode[dotted,node distance=4cm,left=of y]{z}{\S}; \addPrefix{z}{\vx\prefixcon \vy};
		\ltgnode[above=of z]{Iindir}{\indir}; \addPos{Iindir}{\vI};
		\draw[->](Iindir) to (z.north);
		\node[draw,shape=transbox,below=of z](I){\recprefixed{\starfs{}\prefixcon x^\vx\fs{f^\vf\commentOut{=x}}}{\abs{z}{z}}};
		\draw[->](z) to (I.north);
		\ltgnode[node distance=3cm,right=of y]{findir}{\indir}; \addPos{findir}{\vf};
		\ltgnode[below=of findir]{f}{\0}; \addPrefix{f}{\vx};
		\draw[->](findir) to (f.north);
		\node[left=of f](fwest){}; \draw(f.west) to (fwest.center); \draw[->] (fwest.center) |- (x.east);
	}
	\transstep{\S, @, \S}{
		\node(yIIyff){};
		\ltgnode[dotted,node distance=1.5cm,right=of yIIyff]{Sff}{\S}; \addPrefix{Sff}{\vx\prefixcon \vy};
		\node[draw,shape=transbox,below=of Sff](ff){\recprefixed{\starfs{}\prefixcon x^\vx\fs{f^\vf\commentOut{=x}}}{\app{f}{f}}};
		\draw[->](Sff) to (ff.north);
		\ltgnode[node distance=1.5cm,left=of ff]{yI@Iy}{@}; \addPrefix{yI@Iy}{\vx\prefixcon \vy};
		\node[node distance=7mm,below=of yI@Iy](yIIy){};
		\node[node distance=0mm,draw,shape=transbox, left=of yIIy](yI){\recprefixed{\starfs{}\prefixcon x^\vx\fs{f^\vf\commentOut{=x}}\prefixcon y^\vy\fs{I^\vI\commentOut{=\abs{z}{z}}}}{\app{y}{I}}};
		\draw[->](yI@Iy) to (yI.north);
		\node[node distance=0mm,draw,shape=transbox,right=of yIIy](Iy){\recprefixed{\starfs{}\prefixcon x^\vx\fs{f^\vf\commentOut{=x}}\prefixcon y^\vy\fs{I^\vI\commentOut{=\abs{z}{z}}}}{\app{I}{y}}};
		\draw[->](yI@Iy) to (Iy.north);
		\ltgnode[above=of yIIyff]{yIIy@ff}{@}; \addPrefix{yIIy@ff}{\vx\prefixcon \vy};
		\draw[->](yIIy@ff) to (yI@Iy.north); \draw[->](yIIy@ff) to (Sff.north);
		\ltgnode[above=of yIIy@ff]{y}{λ}; \addPos{y}{\vy}; \addPrefix{y}{\vx};
		\draw[->](y) to (yIIy@ff.north);
		\ltgnode[above=of y]{x}{λ}; \addPos{x}{\vx}; \addPrefix{x}{};
		\draw[->](x) to (y.north);
		\draw[<-](x.north) to +(0mm,4mm);
		\ltgnode[dotted,node distance=4cm,left=of y]{z}{\S}; \addPrefix{z}{\vx\prefixcon \vy};
		\ltgnode[above=of z]{Iindir}{\indir}; \addPos{Iindir}{\vI};
		\draw[->](Iindir) to (z.north);
		\ltgnode[dotted,below=of z]{zS}{\S}; \addPrefix{zS}{\vx};
		\node[draw,shape=transbox,below=of zS](I){\recprefixed{\starfs{}}{\abs{z}{z}}};
		\draw[->](z) to (zS.north);
		\draw[->](zS) to (I.north);
		\ltgnode[node distance=3cm,right=of y]{findir}{\indir}; \addPos{findir}{\vf};
		\ltgnode[below=of findir]{f}{\0}; \addPrefix{f}{\vx};
		\draw[->](findir) to (f.north);
		\node[left=of f](fwest){}; \draw(f.west) to (fwest.center); \draw[->] (fwest.center) |- (x.east);
	}
	\transstep{\mathrlap{λ,@,@,@}}{
		\node(yIIyff){};
		\ltgnode[dotted,node distance=2cm,right=of yIIyff]{Sff}{\S}; \addPrefix{Sff}{\vx\prefixcon \vy};
		\ltgnode[below=of Sff]{f@f}{@}; \addPrefix{f@f}{\vx\prefixcon \vy}; \node[below=of f@f](ff){};
		\draw[->](Sff) to (f@f.north);
		\node[node distance=0mm,draw,shape=transbox, left=of ff](f1){\recprefixed{\starfs{}\prefixcon x^\vx\fs{f^\vf\commentOut{=x}}}{f}};
		\draw[->](f@f) to (f1.north);
		\node[node distance=0mm,draw,shape=transbox,right=of ff](f2){\recprefixed{\starfs{}\prefixcon x^\vx\fs{f^\vf\commentOut{=x}}}{f}};
		\draw[->](f@f) to (f2.north);
		\ltgnode[node distance=1cm,left=of f1]{yI@Iy}{@}; \addPrefix{yI@Iy}{\vx\prefixcon \vy};
		\node[node distance=7mm,below=of yI@Iy](yIIy){};
		\ltgnode[node distance=2.4cm,right=of yIIy]{I@y}{@}; \addPrefix{I@y}{\vx\prefixcon \vy}; \node[below=of I@y](Iy){};
		\draw[->](yI@Iy) to (I@y.north);
		\node[node distance=0mm,draw,shape=transbox, left=of Iy](I2){\recprefixed{\starfs{}\prefixcon x^\vx\fs{f^\vf\commentOut{=x}}\prefixcon y^\vy\fs{I^\vI\commentOut{=\abs{z}{z}}}}{I}};
		\draw[->](I@y) to (I2.north);
		\node[node distance=0mm,draw,shape=transbox,right=of Iy](y2){\recprefixed{\starfs{}\prefixcon x^\vx\fs{f^\vf\commentOut{=x}}\prefixcon y^\vy\fs{I^\vI\commentOut{=\abs{z}{z}}}}{y}};
		\draw[->](I@y) to (y2.north);
		\ltgnode[node distance=1cm, left=of I2]{y@I}{@}; \addPrefix{y@I}{\vx\prefixcon \vy};
		\node[node distance=7mm,below=of y@I](y@@I){}; \node[node distance=8mm,right=of y@@I](yI){}; \draw[->](yI@Iy) to (y@I.north);
		\node[node distance=0mm,draw,shape=transbox, left=of yI](y1){\recprefixed{\starfs{}\prefixcon x^\vx\fs{f^\vf\commentOut{=x}}\prefixcon y^\vy\fs{I^\vI\commentOut{=\abs{z}{z}}}}{y}};
		\draw[->](y@I) to (y1.north);
		\node[node distance=0mm,draw,shape=transbox, right=of yI](I1){\recprefixed{\starfs{}\prefixcon x^\vx\fs{f^\vf\commentOut{=x}}\prefixcon y^\vy\fs{I^\vI\commentOut{=\abs{z}{z}}}}{I}};
		\draw[->](y@I) to (I1.north);
		\ltgnode[above=of yIIyff]{yIIy@ff}{@}; \addPrefix{yIIy@ff}{\vx\prefixcon \vy};
		\draw[->](yIIy@ff) to (yI@Iy.north); \draw[->](yIIy@ff) to (Sff.north);
		\ltgnode[above=of yIIy@ff]{y}{λ}; \addPos{y}{\vy}; \addPrefix{y}{\vx};
		\draw[->](y) to (yIIy@ff.north);
		\ltgnode[above=of y]{x}{λ}; \addPos{x}{\vx}; \addPrefix{x}{};
		\draw[->](x) to (y.north);
		\draw[<-](x.north) to +(0mm,4mm);
		\ltgnode[node distance=3.2cm,left=of x]{Iindir}{\indir}; \addPos{Iindir}{\vI};
		\ltgnode[dotted,below=of Iindir]{SSz}{\S}; \addPrefix{SSz}{\vx\prefixcon \vy};
		\draw[->](Iindir) to (SSz.north);
		\ltgnode[dotted,below=of SSz]{Sz}{\S}; \addPrefix{Sz}{\vx};
		\draw[->](SSz) to (Sz.north);
		\ltgnode[below=of Sz]{z}{λ}; \addPos{z}{\vz}; \addPrefix{z}{};
		\draw[->](Sz) to (z.north);
		\node[draw,shape=transbox,below=of z](I){\recprefixed{\starfs{}\prefixcon z^\vz\fs{}}{z}};
		\draw[->](z) to (I.north);
		\ltgnode[node distance=4cm,right=of y]{findir}{\indir}; \addPos{findir}{\vf};
		\ltgnode[below=of findir]{f}{\0}; \addPrefix{f}{\vx};
		\draw[->](findir) to (f.north);
		\node[left=of f](fwest){}; \draw(f.west) to (fwest.center); \draw[->] (fwest.center) |- (x.east);
	}
	\transstep{\0, \0, f, f, \0, f, f}{
		\node(yIIyff){};
		\ltgnode[node distance=1.2cm,left=of yIIyff]{yI@Iy}{@}; \addPrefix{yI@Iy}{\vx\prefixcon \vy}; \node[below=of yI@Iy](yIIy){};
		\ltgnode[right=of yIIy]{I@y}{@}; \addPrefix{I@y}{\vx\prefixcon \vy}; \node[below=of I@y](Iy){};
		\draw[->](yI@Iy) to (I@y.north);
		\node[left=of Iy](SI2){};
		\draw(I@y) to (SI2.center);
		\ltgnode[right=of Iy]{y2}{\0}; \addPrefix{y2}{\vx\prefixcon \vy};
		\node[right=of y2](y2l){}; \draw(y2) to (y2l.center);
		\draw[->](I@y) to (y2.north);
		\ltgnode[left=of yIIy]{y@I}{@}; \addPrefix{y@I}{\vx\prefixcon \vy}; \node[below=of y@I](yI){};
		\draw[->](yI@Iy) to (y@I.north);
		\ltgnode[left=of yI]{y1}{\0}; \addPrefix{y1}{\vx\prefixcon \vy};
		\node[left=of y1](y1l){}; \draw(y1) to (y1l.center);
		\draw[->](y@I) to (y1.north);
		\node[right=of yI](SI1){};
		\draw(y@I) to (SI1.center);
		\node[below=of SI1](Imerge){};
		\draw(SI1.center) to (Imerge.center);
		\draw(SI2.center) |- (Imerge.center);
		\ltgnode[dotted,node distance=1.3cm,right=of yIIyff]{Sff}{\S}; \addPrefix{Sff}{\vx\prefixcon \vy};
		\ltgnode[below=of Sff]{f@f}{@}; \addPrefix{f@f}{\vx\prefixcon \vy}; \node[below=of f@f](ff){};
		\draw[->](Sff) to (f@f.north);
		\node[node distance=1cm,right=of ff](fmerge){};
		\node[ left=of ff](f1){}; \draw(f@f) to (f1.south); \draw(f1.south) |- (fmerge.south);
		\node[right=of ff](f2){}; \draw(f@f) to (f2.center); \draw(f2.center) to (fmerge.center);
		\ltgnode[above=of yIIyff]{yIIy@ff}{@}; \addPrefix{yIIy@ff}{\vx\prefixcon \vy};
		\draw[->](yIIy@ff) to (yI@Iy.north); \draw[->](yIIy@ff) to (Sff.north);
		\ltgnode[above=of yIIy@ff]{y}{λ}; \addPos{y}{\vy}; \addPrefix{y}{\vx};
		\draw[->](y2l.center) |- (y.east); \draw[->](y1l.center) |- (y.west);
		\draw[->](y) to (yIIy@ff.north);
		\ltgnode[above=of y]{x}{λ}; \addPos{x}{\vx}; \addPrefix{x}{};
		\draw[->](x) to (y.north);
		\draw[<-](x.north) to +(0mm,4mm);
		\ltgnode[node distance=5cm,left=of y]{Iindir}{\indir}; \addPos{Iindir}{\vI};
		\node[right=of Iindir](Ieast){};
		\draw[->](Ieast.center) to (Iindir);
		\ltgnode[dotted,below=of Iindir]{SSz}{\S}; \addPrefix{SSz}{\vx\prefixcon \vy};
		\draw[->](Iindir) to (SSz.north);
		\ltgnode[dotted,below=of SSz]{Sz}{\S}; \addPrefix{Sz}{\vx};
		\draw[->](SSz) to (Sz.north);
		\ltgnode[below=of Sz]{z}{λ}; \addPos{z}{\vz}; \addPrefix{z}{};
		\draw[->](Sz) to (z.north);
		\ltgnode[below=of z]{I}{\0}; \addPrefix{I}{\vz};
		\node[left=of I](Iwest){}; \draw(I.west) to (Iwest.center); \draw[->] (Iwest.center) |- (z.west);
		\draw[->](z) to (I.north);
		\ltgnode[node distance=1.5cm,right=of y]{findir}{\indir}; \addPos{findir}{\vf};
		\ltgnode[below=of findir]{f}{\0}; \addPrefix{f}{\vx};
		\draw[->](findir) to (f.north);
		\node[left=of f](fwest){}; \draw(f.west) to (fwest.center); \draw[->] (fwest.center) |- (x.east);
		\draw[->](fmerge.south) |- (findir.east);
		\draw(Imerge.center) -| (Ieast.center);
	}
	\transstep{\eraseIndirectionVertices}{
		\node(yIIyff){};
		\ltgnode[node distance=1.2cm,left=of yIIyff]{yI@Iy}{@}; \addPrefix{yI@Iy}{\vx\prefixcon \vy}; \node[below=of yI@Iy](yIIy){};
		\ltgnode[right=of yIIy]{I@y}{@}; \addPrefix{I@y}{\vx\prefixcon \vy}; \node[below=of I@y](Iy){};
		\draw[->](yI@Iy) to (I@y.north);
		\node[left=of Iy](SI2){};
		\ltgnode[right=of Iy]{y2}{\0}; \addPrefix{y2}{\vx\prefixcon \vy};
		\node[right=of y2](y2l){}; \draw(y2) to (y2l.center);
		\draw[->](I@y) to (y2.north);
		\ltgnode[left=of yIIy]{y@I}{@}; \addPrefix{y@I}{\vx\prefixcon \vy}; \node[below=of y@I](yI){};
		\draw[->](yI@Iy) to (y@I.north);
		\ltgnode[left=of yI]{y1}{\0}; \addPrefix{y1}{\vx\prefixcon \vy};
		\node[left=of y1](y1l){}; \draw(y1) to (y1l.center);
		\draw[->](y@I) to (y1.north);
		\node[right=of yI](SI1){};
		\ltgnode[dotted,node distance=1.3cm,right=of yIIyff]{Sff}{\S}; \addPrefix{Sff}{\vx\prefixcon \vy};
		\ltgnode[below=of Sff]{f@f}{@}; \addPrefix{f@f}{\vx\prefixcon \vy}; \node[below=of f@f](ff){};
		\draw[->](Sff) to (f@f.north);
		\ltgnode[above=of yIIyff]{yIIy@ff}{@}; \addPrefix{yIIy@ff}{\vx\prefixcon \vy};
		\draw[->](yIIy@ff) to (yI@Iy.north); \draw[->](yIIy@ff) to (Sff.north);
		\ltgnode[above=of yIIy@ff]{y}{λ}; \addPos{y}{\vy}; \addPrefix{y}{\vx};
		\draw[->](y2l.center) |- (y.east); \draw[->](y1l.center) |- (y.west);
		\draw[->](y) to (yIIy@ff.north);
		\ltgnode[above=of y]{x}{λ}; \addPos{x}{\vx}; \addPrefix{x}{};
		\draw[->](x) to (y.north);
		\draw[<-](x.north) to +(0mm,4mm);
		\ltgnode[distance=2.5cm,dotted,below=of yIIy]{SSz}{\S}; \addPrefix{SSz}{\vx\prefixcon \vy};
		\draw[->](y@I) to (SSz);
		\draw[->](I@y) to (SSz);
		\ltgnode[dotted,below=of SSz]{Sz}{\S}; \addPrefix{Sz}{\vx};
		\draw[->](SSz) to (Sz.north);
		\ltgnode[below=of Sz]{z}{λ}; \addPos{z}{\vz}; \addPrefix{z}{};
		\draw[->](Sz) to (z.north);
		\ltgnode[below=of z]{I}{\0}; \addPrefix{I}{\vz};
		\node[left=of I](Iwest){}; \draw(I.west) to (Iwest.center); \draw[->] (Iwest.center) |- (z.west);
		\draw[->](z) to (I.north);
		\ltgnode[below=of ff]{f}{\0}; \addPrefix{f}{\vx};
		\node[right=of f](feast){}; \draw(f.east) to (feast.center); \draw[->] (feast.center) |- (x.east);
		\draw[->,bend right](f@f) to (f);
		\draw[->,bend left](f@f) to (f);
	}
\end{transenv}

\chapter{Implementation Showcase}\label{app:impl:showcase}

To demonstrate the realisability of our method, we have implemented the methods from \cref{chap:maxsharing}. The implementation is called \href{http://hackage.haskell.org/package/maxsharing/}{maxsharing} and is available on Hackage: \url{http://hackage.haskell.org/package/maxsharing/}. It is written in Haskell and therefore requires the \href{http://www.haskell.org/platform/}{Haskell Platform} to be installed. Then, maxsharing can be installed via \href{http://www.haskell.org/haskellwiki/Cabal-Install}{cabal-install} using the commands \texttt{cabal update} and \texttt{cabal install maxsharing} from the terminal. Invoke the executable \texttt{maxsharing} in your cabal-directory with a file as an argument that contains a \lambdaletrec-term. Run \texttt{maxsharing -h} for help on run-time flags.

For further illustration we provide the output of the \href{http://hackage.haskell.org/package/maxsharing/}{maxsharing} tool
for a number of examples from the thesis. Since the volume of the output is too large to include here it is supplied as an external appendix available at \url{http://rochel.info/thesis/}.

\chapter{Confluence of unfolding \lambdaletrec-terms}\label{app:conf_proof}

Here we present a proof of \cref{prop:unf-confluence} which is an adaptation of \cite{grab:roch:confunf} where we prove the confluence of a precursor of \unfCRS. Here the proof below is written out for \unfbhCRS, but works for \unfCRS as well if the critical pairs involving the rules \rulep{\tighten} and \rulep{\bh} are left out.

\begin{proof}[Proof of \cref{prop:unf-confluence}]
	\newcommand\TL{}
	\newcommand\TR{\m{\,\begin{picture}(-1,1)(-1,-2)\circle*{1.5}\end{picture}\ }} 
	\newcommand\BL{\m{\,\begin{picture}(-1,1)(-1,-2)\circle{3.5}\end{picture}\ }}
	\newcommand\BR{\m{\odot}}
	\newcommand\atlcxtBRTRTL[6]{\m{\begin{array}{l}\sacxt[\L_{#1}^{\BR},\dots,\L_{#2}^{\BR},\\\indent \phantom{\sacxt[}\L_{#3}^{\TR},\dots,\L_{#4}^{\TR},\\\indent\indent \phantom{\sacxt[}\L_{#5}^{\TL},\dots,\L_{#6}^{\TL}]\end{array}}}
	\newcommand\atlcxtBRBLTL[6]{\m{\begin{array}{l}\sacxt[\L_{#1}^{\BR},\dots,\L_{#2}^{\BR},\\\indent \phantom{\sacxt[}\L_{#3}^{\BL},\dots,\L_{#4}^{\BL},\\\indent\indent \phantom{\sacxt[}\L_{#5}^{\TL},\dots,\L_{#6}^{\TL}]\end{array}}}
	\newcommand\tilerowsep{1.2cm}
	\newcommand\tilecolsep{1.5cm}
	\newcommand\vpar{\m{=}}
	\newcommand\hpar{\m{||}}

	First of all, we cannot use Newman's Lemma to prove confluence, because \unfbhCRS is not terminating. To show confluence of \unfbhCRS we use the method of `decreasing diagrams' \cite[Section 2.3]{oost:1994} \cite[Section 14.2]{terese:2003}. We use it however not to prove confluence of the rewriting relation \m{\red_\unfold} induced by \unfbhCRS directly, but of the abstract reduction system \m{\aARS = (\Ter\lambdaletreccal, \setcompr{\parred_{\wDepth\arule{d}}}{(d,\arule) \in ℕ \times R})} with \m{R} as the set of rules of \unfbhCRS where \m{\parred_{\wDepth\arule{d}}} denotes the parallel rewriting relation on \Ter\lambdaletreccal induced by rule \arule at \letrec-depth \m{d}. As a precedence order we consider the order induced by the \letrec-depth:
	\[ \parred_{\wDepth\arule{d}} ~≥~ \parred_{\wDepth\brule{d'}} ~ \Longleftrightarrow ~ d \geq d'\]

	The \letrec-depth of a redex in \lambdaletrec-term denotes the number of \Let-nodes passed on the path from the root of the term tree to the corresponding position. We write \m{\red_{\wDepth\arule{d}}} to denote the relation induced by applying rule \arule contracting a redex at \letrec-depth \m{d}.

	Let us denote the rewriting relation induced by \aARS by \m{\red_\aARS}:
	\[ \red_\aARS = \bigcup\,\setcompr{\parred_{\wDepth\arule{d}}}{(d,\arule) \in ℕ \times R}\]

	If \m{\red_\aARS} is confluent then \m{\red_\unfold} is confluent because it holds: \m{\red_\unfold \subseteq \red_\aARS \subseteq \mred_\unfold} or equivalently \m{\mred_\aARS = \mred_\unfold} (see also \cite[Lemma 2.2.5]{oost:1994}).

	We use parallel steps because the preceding attempt to prove confluence of \unfbhCRS-steps themselves by decreasing diagrams was unsuccessful. As a precedent order we considered an ordering on the rules and lexicographic extensions of such orderings with the \letrec-depth of the contracted redex. We came to the conclusion that no such order could ensure decreasingness of the elementary diagrams of both the critical pairs as well as the strictly nested redexes. This was due to redex duplication induced by the diverging steps, so that joining the diagram required a many-step that disrupted decreasingness. In order to resolve this problem we considered parallel steps such that the problematic multi-step would become a single parallel step. This led to more intricate diagrams but turned out to be a viable solution.

	We will show that two diverging parallel steps in \unfbhCRS can be joined in an elementary diagram of the following form with \m{d \leq e}.

	\begin{figure}
		\begin{tikzpicture}
			\matrix[row sep=0.8cm,column sep=1.3cm]{
				\emptynode{tl};&&
				\emptynode{tr};\\
				&&\emptynode{mr};\\
				\emptynode{bl};&
				\emptynode{bm};&
				\emptynode{br};\\
			};
			\draw[->       ](tl) to node{\hpar} node[above]{\wDepth\arule{d}}   (tr);
			\draw[->       ](tl) to node{\vpar} node[left ]{\wDepth\brule{e}}   (bl);
			\draw[->,dotted](bl) to node{\hpar} node[below]{\wDepth\brule{e-1}} (bm);
			\draw[->,dotted](bm) to node{\hpar} node[below]{\wDepth\arule{d}}   (br);
			\draw[->,dotted](tr) to node{\vpar} node[right]{\wDepth\brule{e}}   (mr);
			\draw[->,dotted](mr) to node{\vpar} node[right]{\wDepth\brule{e-1}} (br);
		\end{tikzpicture}
		\caption{Elementary diagram}
		\label{fig:elem_dia}
	\end{figure}

	If we pick as the precedence order on the steps the order that is induced by their \letrec-depth, the diagram is decreasing. Note that in all the diagrams we implicitly assume the reflexive closure for all arrows. The rest of the proof is structured as follows. To justify the diagram we distinguish the cases \m{d = e} and \m{d < e}, for which we construct diagrams that are instances of the diagram in \cref{fig:elem_dia}.

	\section*{Case 1}

	For \m{d = e} we need to consider parallel diverging steps contracting redexes at the same \letrec-depth \m{d}. We construct the diagram below which is an instance of the diagram above where the diverging parallel steps are in sequentialised form. We write terms as fillings of a multihole context \sacxt with all its holes at \letrec-depth \m{d} such that the contracted \wDepth\arule{d}- and \wDepth\brule{d}-redexes are filled into these holes. In this way we can make explicit at which position a step takes place, i.e. at the root of the context hole fillings. The topmost row and the leftmost column are respective sequentialisations of the parallel diverging \wDepth\arule{d}- and \wDepth\brule{d}-steps into single steps.

	\noindent
	\begin{tikzpicture}
		\matrix[row sep=1.2cm,column sep=0.5cm]{
			\node(00){\acxt{\L_0^{\TL}, \dots, \L_n^{\TL}}};&
			\node(10){\acxt{\L_0^{\TR}, \L_1^{\TL} \dots, \L_n^{\TL}}};&
			\node(i0){\dots};&
			\node(n0){\acxt{\L_0^{\TR}, \dots, \L_n^{\TR}}};\\
			\node(01){\acxt{\L_0^{\BL}, \L_1^{\TL}, \dots, \L_n^{\TL}}};&
			\node(11){\acxt{\L_0^{\BR}, \L_1^{\TL}, \dots, \L_n^{\TL}}};&
			\node(i1){\dots};&
			\node(n1){\acxt{\L_0^{\BR}, \L_1^{\TR}, \dots, \L_n^{\TR}}};\\
			\node(0i){\vdots};&
			\node(1i){\vdots};&
			\node(ii){\ddots};&
			\node(ni){\vdots};\\
			\node(0n){\acxt{\L_0^{\BL}, \dots, \L_n^{\BL}}};&
			\node(1n){\acxt{\L_0^{\BR}, \L_1^{\BL}, \dots, \L_n^{\BL}}};&
			\node(in){\dots};&
			\node(nn){\acxt{\L_0^{\BR}, \dots, \L_n^{\BR}}};\\
		};
		\draw[->       ](00)  to            node[above]{\wDepth\arule{d}}    (10);
		\draw[->       ](10)  to            node[above]{\wDepth\arule{d}}    (i0);
		\draw[->       ](i0)  to            node[above]{\wDepth\arule{d}}    (n0);
		\draw[->,dotted](01)  to node{\hpar} node[above]{\wDepth\arule{d}}    (11);
		\draw[->,dotted](11)  to            node[above]{\wDepth\arule{d}}    (i1);
		\draw[->,dotted](i1)  to            node[above]{\wDepth\arule{d}}    (n1);
		\draw[->,dotted](0n)  to node{\hpar} node[above]{\wDepth\arule{d}}    (1n);
		\draw[->,dotted](1n)  to node{\hpar} node[above]{\wDepth\arule{d}}    (in);
		\draw[->,dotted](in)  to node{\hpar} node[above]{\wDepth\arule{d}}    (nn);
		\draw[->       ](00)  to            node[left ]{\wDepth\brule{d}}    (01);
		\draw[->,dotted](10)  to node{\vpar}  node[left ]{\wDepth\brule{d}}    (11);
		\draw[->,dotted](n0)  to node{\vpar}  node[left ]{\wDepth\brule{d}}    (n1);
		\draw[->       ](01)  to            node[left ]{\wDepth\brule{d}}    (0i);
		\draw[->,dotted](11)  to            node[left ]{\wDepth\brule{d}}    (1i);
		\draw[->,dotted](n1)  to node{\vpar}  node[left ]{\wDepth\brule{d}}    (ni);
		\draw[->       ](0i)  to            node[left ]{\wDepth\brule{d}}    (0n);
		\draw[->,dotted](1i)  to            node[left ]{\wDepth\brule{d}}    (1n);
		\draw[->,dotted](ni)  to node{\vpar}  node[left ]{\wDepth\brule{d}}    (nn);
	\end{tikzpicture}

	\noindent
	In this diagram the tiles at \m{(i,j)} for \m{i,j ∈ \set{0,\dots,n-1}} below the diagonal (\m{i<j}) look as follows:\\*
	\begin{tikzpicture}
		\matrix[row sep=\tilerowsep,column sep=\tilecolsep]{
			\node(00){\atlcxtBRBLTL{0}{i-1}{i}{j-1}{j}{n}};&
			\node(10){\atlcxtBRBLTL{0}{i}{i+1}{j-1}{j}{n}};\\
			\node(01){\atlcxtBRBLTL{0}{i-1}{i}{j}{j+1}{n}};&
			\node(11){\atlcxtBRBLTL{0}{i}{i+1}{j}{j+1}{n}};\\
		};
		\draw[->       ](00)  to node{\hpar} node[above]{\wDepth\arule{d}}    (10);
		\draw[->,dotted](01)  to node{\hpar} node[above]{\wDepth\arule{d}}    (11);
		\draw[->       ](00)  to              node[left ]{\wDepth\brule{d}}    (01);
		\draw[->,dotted](10)  to              node[left ]{\wDepth\brule{d}}    (11);
	\end{tikzpicture}

	\noindent
	The tiles at \m{(i,j)} for \m{i,j ∈ \set{0,\dots,n-1}} above the diagonal (\m{i>j}) look as follows:\\*
	\begin{tikzpicture}
		\matrix[row sep=\tilerowsep,column sep=\tilecolsep]{
			\node(00){\atlcxtBRTRTL{0}{j-1}{j}{i-1}{i}{n}};&
			\node(10){\atlcxtBRTRTL{0}{j-1}{j}{i}{i+1}{n}};\\
			\node(01){\atlcxtBRTRTL{0}{j}{j+1}{i-1}{i}{n}};&
			\node(11){\atlcxtBRTRTL{0}{j}{j+1}{i}{i+1}{n}};\\
		};
		\draw[->       ](00)  to             node[above]{\wDepth\arule{d}}    (10);
		\draw[->,dotted](01)  to             node[above]{\wDepth\arule{d}}    (11);
		\draw[->       ](00)  to node{\vpar} node[left ]{\wDepth\brule{d}}    (01);
		\draw[->,dotted](10)  to node{\vpar} node[left ]{\wDepth\brule{d}}    (11);
	\end{tikzpicture}

	\noindent
	For \m{i ∈ \set{0,\dots,n-1}} the \i-th diagonal tile looks like this:\\*
	\begin{tikzpicture}
		\matrix[row sep=\tilerowsep,column sep=\tilecolsep/2]{
			\node(00){\acxt{\L_0^{\BR}, \dots, \L_{i-1}^{\BR}, \L_i^{\TL}, \dots, \L_n^{\TL}}};&
			\node(10){\acxt{\L_0^{\BR}, \dots, \L_{i-1}^{\BR}, \L_i^{\TR}, \L_{i+1}^{\TL}, \dots, \L_n^{\TL}}};\\
			\node(01){\acxt{\L_0^{\BR}, \dots, \L_{i-1}^{\BR}, \L_i^{\BL}, \L_{i+1}^{\TL}, \dots, \L_n^{\TL}}};&
			\node(11){\acxt{\L_0^{\BR}, \dots, \L_i^{\BR}, \L_{i+1}^{\TL}, \dots, \L_n^{\TL}}};\\
		};
		\draw[->       ](00)  to            node[above]{\wDepth\arule{d}}    (10);
		\draw[->,dotted](01)  to node{\hpar} node[above]{\wDepth\arule{d}}    (11);
		\draw[->       ](00)  to            node[left ]{\wDepth\brule{d}}    (01);
		\draw[->,dotted](10)  to node{\vpar}  node[left ]{\wDepth\brule{d}}    (11);
	\end{tikzpicture}

	Only the tiles on the diagonal require closer attention because for all other tiles the vertical and horizontal steps take place in different holes of the context, therefore they are disjoint and consequently commute. In the tiles on the diagonal the diverging steps may be either due to a critical pair or to identical steps. In the latter case the diagram is easily joined. In case of a critical pair, since all steps take place at the same \letrec-depth any such critical pair must arise from a root overlap. Exhaustive scrutiny of all these critical pairs reveals that they can be joined in a way that conforms to the tiles on the diagonal. Note that the \letrec-depths of the steps have to be increased by \m{d} according to the lifting into a context with its hole at \letrec-depth \m{d}.

	\begin{remark}
		Note, that in the critical pair analysis below, we are not always faithful to the actual formulation of the rules, in the sense that we sometimes assume that a binding appears at a specific position in the list of bindings of a \Let-expression. This is merely to save page space and can be easily generalised.
	\end{remark}

	\subsection{Critical pairs due to \wDepth{λ}{0}/\wDepth{\arule}{0}}

	\noindent
	\begin{tikzpicture}
		\matrix{
			\node(tl){\letin{}{\abs{x}\L}};&
			\node(tr){\abs{x}{\letin{}\L}};\\
			\node(bl){\abs{x}\L};&
			\node(br){\abs{x}\L};\\
		};
		\draw[->       ](tl) to node[above]{\wDepth{λ}{0}}    (tr);
		\draw[->       ](tl) to node[left ]{\wDepth{\nil}{0}} (bl);
		\draw[->,dotted](tr) to node[right]{\wDepth{\nil}{0}} (br);
		\draw[double   ](bl) to                               (br);
	\end{tikzpicture}
	\hfill
	\begin{tikzpicture}
		\matrix{
			\node(tl){\letin{B}{\abs{x}\L}};&
			\node(tr){\abs{x}{\letin{B}\L}};\\
			\node(bl){\letin{B'}{\abs{x}\L}};&
			\node(br){\abs{x}{\letin{B'}\L}};\\
		};
		\draw[->       ](tl) to node[above]{\wDepth{λ}{0}}     (tr);
		\draw[->       ](tl) to node[left ]{\wDepth\reduce{0}} (bl);
		\draw[->,dotted](tr) to node[right]{\wDepth\reduce{0}} (br);
		\draw[->,dotted](bl) to node[below]{\wDepth{λ}{0}}     (br);
	\end{tikzpicture}

	\noindent
	\begin{tikzpicture}
		\matrix{
			\node(tl){\letin{B,f=g}{\abs{x}\L}};&
			\node(tr){\abs{x}{\letin{B,f=g}\L}};\\
			\node(bl){\letin{\subst{B}{f}{g}}{\abs{x}{\subst{\L}{f}{g}}}};&
			\node(br){\abs{x}{\letin{\subst{B}{f}{g}}{\subst{\L}{f}{g}}}};\\
		};
		\draw[->       ](tl) to node[above]{\wDepth{λ}{0}}      (tr);
		\draw[->       ](tl) to node[left ]{\wDepth\tighten{0}} (bl);
		\draw[->,dotted](tr) to node[right]{\wDepth\tighten{0}} (br);
		\draw[->,dotted](bl) to node[below]{\wDepth{λ}{0}}      (br);
	\end{tikzpicture}

	\noindent
	\begin{tikzpicture}
		\matrix{
			\node(tl){\letin{B,f=f}{\abs{x}\L}};&
			\node(tr){\abs{x}{\letin{B,f=f}\L}};\\
			\node(bl){\letin{\subst{B}{f}{\bh}}{\abs{x}{\subst{\L}{f}{\bh}}}};&
			\node(br){\abs{x}{\letin{\subst{B}{f}{\bh}}{\subst{\L}{f}{\bh}}}};\\
		};
		\draw[->       ](tl) to node[above]{\wDepth{λ}{0}} (tr);
		\draw[->       ](tl) to node[left ]{\wDepth\bh{0}} (bl);
		\draw[->,dotted](tr) to node[right]{\wDepth\bh{0}} (br);
		\draw[->,dotted](bl) to node[below]{\wDepth{λ}{0}} (br);
	\end{tikzpicture}

	\subsection{Critical pairs due to \wDepth{@}{0}/\wDepth{\brule}{0}}

	\noindent
	\begin{tikzpicture}
		\matrix{
			\node(tl){\letin{}{\app{L}{P}}};&
			\node(tr){\appbreak{\letin{}\L}{\letin{}{P}}};\\
			\node(bl){\app\L{P}};&
			\node(br){\app\L{P}};\\
		};
		\draw[->       ](tl) to           node[above]{\wDepth{@}{0}} (tr);
		\draw[->       ](tl) to           node[left ]{\wDepth\nil{0}}  (bl);
		\draw[->,dotted](tr) to node{\vpar} node[right]{\wDepth\nil{0}}  (br);
		\draw[double   ](bl) to                                              (br);
	\end{tikzpicture}
	\hfill
	\begin{tikzpicture}
		\matrix{
			\node(tl){\letin{B}{\app\L{P}}};&
			\node(tr){\appbreak{\letin{B}\L}{\letin{B}{P}}};\\
			\node(bl){\letin{B'}{\app\L{P}}};&
			\node(br){\appbreak{\letin{B'}\L}{\letin{B'}{P}}};\\
		};
		\draw[->       ](tl) to           node[above]{\wDepth{@}{0}}   (tr);
		\draw[->       ](tl) to           node[left ]{\wDepth\reduce{0}} (bl);
		\draw[->,dotted](tr) to node{\vpar} node[right]{\wDepth\reduce{0}} (br);
		\draw[->,dotted](bl) to           node[below]{\wDepth{@}{0}}   (br);
	\end{tikzpicture}

	\noindent
	\begin{tikzpicture}
		\matrix{
			\node(tl){\letin{B,f=g}{\app\L{P}}};&
			\node(tr){\appbreak{\letin{B,f=g}\L}{\letin{B}{P}}};\\
			\node(bl){\letin{\subst{B}{f}{g}}{\app{\subst{\L}{f}{g}}{\subst{P}{f}{g}}}};&
			\node(br){\appbreak{\letin{\subst{B}{f}{g}}{\subst{\L}{f}{g}}}{\letin{\subst{B}{f}{g}}{\subst{P}{f}{g}}}};\\
		};
		\draw[->       ](tl) to             node[above]{\wDepth{@}{0}}      (tr);
		\draw[->       ](tl) to             node[left ]{\wDepth\tighten{0}} (bl);
		\draw[->,dotted](tr) to node{\vpar} node[right]{\wDepth\tighten{0}} (br);
		\draw[->,dotted](bl) to             node[below]{\wDepth{@}{0}}      (br);
	\end{tikzpicture}

	\noindent
	\begin{tikzpicture}
		\matrix{
			\node(tl){\letin{B,f=f}{\app\L{P}}};&
			\node(tr){\appbreak{\letin{B,f=f}\L}{\letin{B}{P}}};\\
			\node(bl){\letin{\subst{B}{f}{\bh}}{\app{\subst{\L}{f}{\bh}}{\subst{P}{f}{\bh}}}};&
			\node(br){\appbreak{\letin{\subst{B}{f}{\bh}}{\subst{\L}{f}{\bh}}}{\letin{\subst{B}{f}{\bh}}{\subst{P}{f}{\bh}}}};\\
		};
		\draw[->       ](tl) to             node[above]{\wDepth{@}{0}} (tr);
		\draw[->       ](tl) to             node[left ]{\wDepth\bh{0}} (bl);
		\draw[->,dotted](tr) to node{\vpar} node[right]{\wDepth\bh{0}} (br);
		\draw[->,dotted](bl) to             node[below]{\wDepth{@}{0}} (br);
	\end{tikzpicture}

	\subsection{Critical pairs due to \wDepth{\merg}{0}/\wDepth{\brule}{0}}

	\begin{tikzpicture}
		\matrix{
			\node(tl){\letin{}{\letin{B}\L}};&
			\node(tr){\letin{B}\L};\\
			\node(bl){\letin{B}\L};&
			\node(br){\letin{B}\L};\\
		};
		\draw[->    ](tl) to node[above]{\wDepth\merg{0}} (tr);
		\draw[->    ](tl) to node[left ]{\wDepth\nil{0}}    (bl);
		\draw[double](tr) to                                      (br);
		\draw[double](bl) to                                      (br);
	\end{tikzpicture}

	\noindent
	\begin{tikzpicture}
		\matrix{
			\node(tl){\letin{B}{\letin{C}\L}};&
			\node(tr){\letin{B,C}\L};\\
			\node(bl){\letin{B'}{\letin{C}\L}};&
			\node(br){\letin{B',C}\L};\\
		};
		\draw[->       ](tl) to node[above]{\wDepth\merg{0}}   (tr);
		\draw[->       ](tl) to node[left ]{\wDepth\reduce{0}} (bl);
		\draw[->,dotted](tr) to node[right]{\wDepth\reduce{0}}   (br);
		\draw[->,dotted](bl) to node[below]{\wDepth\merg{0}}   (br);
	\end{tikzpicture}

	\noindent
	\begin{tikzpicture}
		\matrix{
			\node(tl){\letin{B,f=g}{\letin{C}\L}};&
			\node(tr){\letin{B,f=g,C}\L};\\
			\node(bl){\subst{(\letin{B}{\letin{C}{\L}})}{f}{g}};&
			\node(br){\subst{(\letin{B,C}{\L})}{f}{g}};\\
		};
		\draw[->       ](tl) to node[above]{\wDepth\merg{0}}   (tr);
		\draw[->       ](tl) to node[left ]{\wDepth\tighten{0}} (bl);
		\draw[->,dotted](tr) to node[right]{\wDepth\tighten{0}} (br);
		\draw[->,dotted](bl) to node[below]{\wDepth\merg{0}}   (br);
	\end{tikzpicture}

	\noindent
	\begin{tikzpicture}
		\matrix{
			\node(tl){\letin{B,f=f}{\letin{C}\L}};&
			\node(tr){\letin{B,f=f,C}\L};\\
			\node(bl){\subst{(\letin{B}{\letin{C}{\L}})}{f}{\bh}};&
			\node(br){\subst{(\letin{B,C}{\L})}{f}{\bh}};\\
		};
		\draw[->       ](tl) to node[above]{\wDepth\merg{0}} (tr);
		\draw[->       ](tl) to node[left ]{\wDepth\bh{0}}    (bl);
		\draw[->,dotted](tr) to node[right]{\wDepth\bh{0}}    (br);
		\draw[->,dotted](bl) to node[below]{\wDepth\merg{0}} (br);
	\end{tikzpicture}

	\subsection{Critical pairs due to \wDepth{\rec}{0}/\wDepth{\brule}{0}}

	\noindent
	\begin{tikzpicture}
		\matrix{
			\node(tl){\letin{B}{f_i}};&
			\node(tr){\letin{B}{\L_i}};\\
			\node(bl){\letin{B'}{f_i}};&
			\node(br){\letin{B'}{\L_i}};\\
		};
		\draw[->       ](tl) to            node[above]{\wDepth\rec{0}}    (tr);
		\draw[->       ](tl) to            node[left ]{\wDepth\reduce{0}} (bl);
		\draw[->,dotted](tr) to            node[right]{\wDepth\reduce{0}} (br);
		\draw[->,dotted](bl) to            node[below]{\wDepth\rec{0}}    (br);
	\end{tikzpicture}

	\noindent
	There is a \wDepth{\rec}{0}/\wDepth{\tighten}{0} critical pair. We distinguish the cases, whether the body consists of the `tightened' variable:\\*
	\begin{tikzpicture}
		\matrix{
			\node(tl){\letin{B,f=g}{f}};&
			\node(tr){\letin{B,f=g}{g}};\\
			\node(bl){\letin{\subst{B}{f}{g}}{g}};&
			\node(br){\letin{\subst{B}{f}{g}}{g}};\\
		};
		\draw[->       ](tl) to node[above]{\wDepth{\rec}{0}}     (tr);
		\draw[->       ](tl) to node[left ]{\wDepth{\tighten}{0}} (bl);
		\draw[->,dotted](tr) to node[right]{\wDepth{\tighten}{0}} (br);
		\draw[double   ](bl) to                                   (br);
	\end{tikzpicture}

	\noindent
	\begin{tikzpicture}
		\matrix{
			\node(tl){\letin{B,f=g,h=L}{h}};&
			\node(tr){\letin{B,f=g,h=L}{L}};\\
			\node(bl){\letin{\subst{B}{f}{g},h=\subst{L}{f}{g}}{h}};&
			\node(br){\subst{\letin{B,h=L}{L}}{f}{g}};\\
		};
		\draw[->       ](tl) to node[above]{\wDepth{\rec}{0}}     (tr);
		\draw[->       ](tl) to node[left ]{\wDepth{\tighten}{0}} (bl);
		\draw[->,dotted](tr) to node[right]{\wDepth{\tighten}{0}} (br);
		\draw[->,dotted](bl) to node[below]{\wDepth{\rec}{0}}     (br);
	\end{tikzpicture}

	\noindent
	Likewise, there is a \wDepth{\rec}{0}/\wDepth{\bh}{0} critical pair. We distinguish the cases, whether the body consists of the `blackholed' variable:\\*
	\begin{tikzpicture}
		\matrix{
			\node(tl){\letin{B,f=f}{f}};&
			\node(tr){\letin{B,f=f}{f}};\\
			\node(bl){\letin{\subst{B}{f}{\bh}}{\bh}};&
			\node(br){\letin{\subst{B}{f}{\bh}}{\bh}};\\
		};
		\draw[->       ](tl) to node[above]{\wDepth{\rec}{0}} (tr);
		\draw[->       ](tl) to node[left ]{\wDepth{\bh}{0}}  (bl);
		\draw[->,dotted](tr) to node[right]{\wDepth{\bh}{0}}  (br);
		\draw[double   ](bl) to                               (br);
	\end{tikzpicture}

	\noindent
	\begin{tikzpicture}
		\matrix{
			\node(tl){\letin{B,f=f,h=L}{h}};&
			\node(tr){\letin{B,f=f,h=L}{L}};\\
			\node(bl){\letin{\subst{B}{f}{\bh},h=\subst{L}{f}{\bh}}{h}};&
			\node(br){\subst{\letin{B,h=L}{L}}{f}{\bh}};\\
		};
		\draw[->       ](tl) to node[above]{\wDepth{\rec}{0}} (tr);
		\draw[->       ](tl) to node[left ]{\wDepth{\bh}{0}}  (bl);
		\draw[->,dotted](tr) to node[right]{\wDepth{\bh}{0}}  (br);
		\draw[->,dotted](bl) to node[below]{\wDepth{\rec}{0}} (br);
	\end{tikzpicture}

	\subsection{Critical pairs due to \wDepth{\nil}{0}/\wDepth{\brule}{0}}

	There are no additional \wDepth{\nil}{0}/\wDepth{\brule}{0} critical pairs on top of the ones already scrutinised above.

	\subsection{Critical pairs due to \wDepth{\reduce}{0}/\wDepth{\brule}{0}}

	There is a \wDepth{\reduce}{0}/\wDepth{\tighten}{0} critical pair, for which we distinguish the cases whether the `tightened' binding is used.\\
	\begin{tikzpicture}
		\matrix{
			\node(tl){\letin{B,f=g}{\L}};&
			\node(tr){\letin{B',f=g}{\L}};\\
			\node(bl){\letin{\subst{B}{f}{g}}{\subst{\L}{f}{g}}};&
			\node(br){\letin{\subst{B'}{f}{g}}{\subst{\L}{f}{g}}};\\
		};
		\draw[draw=none](tl) to node{\m{f ≠ g}} (br);
		\draw[->       ](tl) to node[above]{\wDepth{\reduce}{0}} (tr);
		\draw[->       ](tl) to node[left ]{\wDepth{\tighten}{0}} (bl);
		\draw[->,dotted](tr) to node[right]{\wDepth{\tighten}{0}} (br);
		\draw[->,dotted](bl) to node[below]{\wDepth{\reduce}{0}} (br);
	\end{tikzpicture}

	\begin{tikzpicture}
		\matrix{
			\node(tl){\letin{B,f=g}{\L}};&
			\node(tr){\letin{B'}{\L}};\\
			\node(bl){\letin{B}{\L}};&
			\node(br){\letin{B'}{\L}};\\
		};
		\draw[->       ](tl) to node[above]{\wDepth{\reduce}{0}} (tr);
		\draw[->       ](tl) to node[left ]{\wDepth{\tighten}{0}} (bl);
		\draw[->,dotted](tr) to node[right]{\wDepth{\tighten}{0}} (br);
		\draw[->,dotted](bl) to node[below]{\wDepth{\reduce}{0}} (br);
	\end{tikzpicture}

	Likewise, there is a \wDepth{\reduce}{0}/\wDepth{\bh}{0} critical pair, for which we distinguish the cases whether the `blackholed' binding is used.\\
	\begin{tikzpicture}
		\matrix{
			\node(tl){\letin{B,f=f}{\L}};&
			\node(tr){\letin{B',f=f}{\L}};\\
			\node(bl){\letin{\subst{B}{f}{\bh}}{\subst{\L}{f}{\bh}}};&
			\node(br){\letin{\subst{B'}{f}{\bh}}{\subst{\L}{f}{\bh}}};\\
		};
		\draw[->       ](tl) to node[above]{\wDepth{\reduce}{0}} (tr);
		\draw[->       ](tl) to node[left ]{\wDepth{\bh}{0}}     (bl);
		\draw[->,dotted](tr) to node[right]{\wDepth{\bh}{0}}     (br);
		\draw[->,dotted](bl) to node[below]{\wDepth{\reduce}{0}} (br);
	\end{tikzpicture}

	\begin{tikzpicture}
		\matrix{
			\node(tl){\letin{B,f=f}{\L}};&
			\node(tr){\letin{B'}{\L}};\\
			\node(bl){\letin{B}{\L}};&
			\node(br){\letin{B'}{\L}};\\
		};
		\draw[->       ](tl) to node[above]{\wDepth{\reduce}{0}} (tr);
		\draw[->       ](tl) to node[left ]{\wDepth{\bh}{0}} (bl);
		\draw[->,dotted](tr) to node[right]{\wDepth{\bh}{0}} (br);
		\draw[->,dotted](bl) to node[below]{\wDepth{\reduce}{0}} (br);
	\end{tikzpicture}

	\subsection{Critical pairs due to \wDepth{\tighten}{0}/\wDepth{\brule}{0}}

	\begin{tikzpicture}
		\matrix{
			\node(tl){\letin{B,f=g,h=i}{\L}};&
			\node(tr){\subst{(\letin{B,h=i}{\L})}{f}{g}};\\
			\node(bl){\subst{(\letin{B,f=g}{\L})}{h}{i}};&
			\node(br){\subst{\subst{(\letin{B}{\L})}{f}{g}}{h}{i}};\\
		};
		\draw[draw=none](tl) to node{\m{f ≠ g ~~ h ≠ i}} (br);
		\draw[->       ](tl) to node[above]{\wDepth\tighten{0}} (tr);
		\draw[->       ](tl) to node[left ]{\wDepth\tighten{0}} (bl);
		\draw[->,dotted](tr) to node[right]{\wDepth\tighten{0}} (br);
		\draw[->,dotted](bl) to node[below]{\wDepth\tighten{0}} (br);
	\end{tikzpicture}

	\noindent
	\begin{tikzpicture}
		\matrix{
			\node(tl){\letin{B,f=g,h=h}{\L}};&
			\node(tr){\subst{(\letin{B,h=h}{\L})}{f}{g}};\\
			\node(bl){\subst{(\letin{B,f=g}{\L})}{h}{\bh}};&
			\node(br){\subst{\subst{(\letin{B}{\L})}{f}{g}}{h}{\bh}};\\
		};
		\draw[draw=none](tl) to node{\m{f ≠ g}} (br);
		\draw[->       ](tl) to node[above]{\wDepth\tighten{0}} (tr);
		\draw[->       ](tl) to node[left ]{\wDepth\bh{0}}      (bl);
		\draw[->,dotted](tr) to node[right]{\wDepth\bh{0}}      (br);
		\draw[->,dotted](bl) to node[below]{\wDepth\tighten{0}} (br);
	\end{tikzpicture}

	\subsection{Critical pairs due to \wDepth{\bh}{0}/\wDepth{\brule}{0}}

	\noindent
	\begin{tikzpicture}
		\matrix{
			\node(tl){\letin{B,f=f,g=g}{\L}};&
			\node(tr){\subst{(\letin{B,g=g}{\L})}{f}{\bh}};\\
			\node(bl){\subst{(\letin{B,f=f}{\L})}{g}{\bh}};&
			\node(br){\subst{\subst{(\letin{B}{\L})}{f}{\bh}}{g}{\bh}};\\
		};
		\draw[->       ](tl) to node[above]{\wDepth\bh{0}} (tr);
		\draw[->       ](tl) to node[left ]{\wDepth\bh{0}} (bl);
		\draw[->,dotted](tr) to node[right]{\wDepth\bh{0}} (br);
		\draw[->,dotted](bl) to node[below]{\wDepth\bh{0}} (br);
	\end{tikzpicture}


	\section*{Case 2}

	For \m{d < e} we use the same approach as for \m{d = e}, the diagram is however more involved. Again, we use a context \sacxt with context holes at \letrec-depth \m{d}. But since \m{e > d}, more than one \wDepth\brule{e}-contraction may take place in one such hole. Therefore a per-hole partitioning of the vertical steps requires a sequence of parallel steps.

	The diagram below fits the scheme of the elementary diagram (\cref{fig:elem_dia}) when interleaving the \wDepth\brule{e}-steps with the \wDepth\brule{e-1}-steps in the rightmost column such that steps at depth \m{e} preceed those at depth \m{e-1}. Similarly for the bottommost row where the \wDepth\arule{e-1}-steps have to preceed the \wDepth\brule{d}-steps. These reorderings are possible since the segments represent contractions within different holes of \sacxt. As in the previous diagram the tiles which do not lie on the diagonal are unproblematic, which leaves us to complete the proof by constructing the tiles on the diagonal.

	\noindent
	
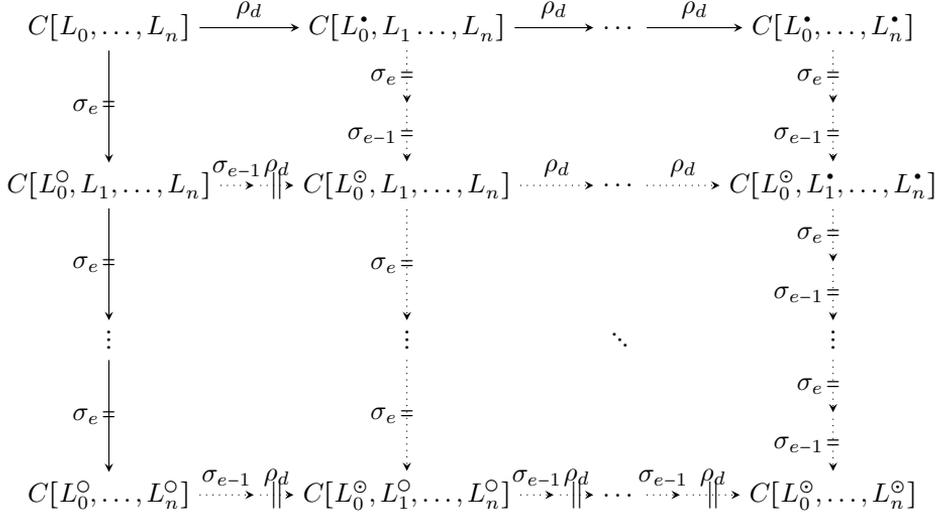
\hspace{-1.2cm}
	\begin{tikzpicture}
		\matrix[row sep=0.7cm,column sep=0.45cm]{
			\node(00){\acxt{\L_0^{\TL}, \dots, \L_n^{\TL}}};&
			&
			\node(10){\acxt{\L_0^{\TR}, \L_1^{\TL} \dots, \L_n^{\TL}}};&
			&
			\node(i0){\dots};&
			\emptynode{ih0};&
			\node(n0){\acxt{\L_0^{\TR}, \dots, \L_n^{\TR}}};\\
			\emptynode{00h};&&
			\emptynode{10h};&&
			\emptynode{i0h};&
			\emptynode{ih0h};&
			\emptynode{n0h};\\
			\node(01){\acxt{\L_0^{\BL}, \L_1^{\TL}, \dots, \L_n^{\TL}}};&
			\emptynode{0h1};&
			\node(11){\acxt{\L_0^{\BR}, \L_1^{\TL}, \dots, \L_n^{\TL}}};&
			&
			\node(i1){\dots};&
			\emptynode{ih1};&
			\node(n1){\acxt{\L_0^{\BR}, \L_1^{\TR}, \dots, \L_n^{\TR}}};\\
			\emptynode{01h};&&
			\emptynode{11h};&&
			\emptynode{i1h};&
			\emptynode{ih1h};&
			\emptynode{n1h};\\
			\node(0i){\vdots};&&
			\node(1i){\vdots};&&
			\node(ii){\ddots};&
			\emptynode{ihi};&
			\node(ni){\vdots};\\
			\emptynode{0ih};&&
			\emptynode{1ih};&&
			\emptynode{iih};&
			\emptynode{ihih};&
			\emptynode{nih};\\
			\node(0n){\acxt{\L_0^{\BL}, \dots, \L_n^{\BL}}};&
			\emptynode{0hn};&
			\node(1n){\acxt{\L_0^{\BR}, \L_1^{\BL}, \dots, \L_n^{\BL}}};&
			\emptynode{1hn};&
			\node(in){\dots};&
			\emptynode{ihn};&
			\node(nn){\acxt{\L_0^{\BR}, \dots, \L_n^{\BR}}};\\
		};
		\draw[->       ](00)  to            node[above]{\wDepth\arule{d}}    (10);
		\draw[->       ](10)  to            node[above]{\wDepth\arule{d}}    (i0);
		\draw[->       ](i0)  to            node[above]{\wDepth\arule{d}}    (n0);
		\draw[->,dotted](01)  to            node[above]{\wDepth\brule{e-1}}    (0h1);
		\draw[->,dotted](0h1) to node{\hpar} node[above]{\wDepth\arule{d}}    (11);
		\draw[->,dotted](11)  to            node[above]{\wDepth\arule{d}}    (i1);
		\draw[->,dotted](i1)  to            node[above]{\wDepth\arule{d}}    (n1);
		\draw[->,dotted](0n)  to            node[above]{\wDepth\brule{e-1}}    (0hn);
		\draw[->,dotted](0hn) to node{\hpar} node[above]{\wDepth\arule{d}}    (1n);
		\draw[->,dotted](1n)  to            node[above]{\wDepth\brule{e-1}}    (1hn);
		\draw[->,dotted](1hn) to node{\hpar} node[above]{\wDepth\arule{d}}    (in);
		\draw[->,dotted](in)  to            node[above]{\wDepth\brule{e-1}}    (ihn);
		\draw[->,dotted](ihn) to node{\hpar} node[above]{\wDepth\arule{d}}    (nn);
		\draw[->       ](00)  to node{\vpar}  node[left ]{\wDepth\brule{e}}    (01);
		\draw[->,dotted](10)  to node{\vpar}  node[left ]{\wDepth\brule{e}}    (10h);
		\draw[->,dotted](10h) to node{\vpar}  node[left ]{\wDepth\brule{e-1}}  (11);
		\draw[->,dotted](n0)  to node{\vpar}  node[left ]{\wDepth\brule{e}}    (n0h);
		\draw[->,dotted](n0h) to node{\vpar}  node[left ]{\wDepth\brule{e-1}}  (n1);
		\draw[->       ](01)  to node{\vpar}  node[left ]{\wDepth\brule{e}}    (0i);
		\draw[->,dotted](11)  to node{\vpar}  node[left ]{\wDepth\brule{e}}    (1i);
		\draw[->,dotted](n1)  to node{\vpar}  node[left ]{\wDepth\brule{e}}    (n1h);
		\draw[->,dotted](n1h) to node{\vpar}  node[left ]{\wDepth\brule{e-1}}  (ni);
		\draw[->       ](0i)  to node{\vpar}  node[left ]{\wDepth\brule{e}}    (0n);
		\draw[->,dotted](1i)  to node{\vpar}  node[left ]{\wDepth\brule{e}}    (1n);
		\draw[->,dotted](ni)  to node{\vpar}  node[left ]{\wDepth\brule{e}}    (nih);
		\draw[->,dotted](nih) to node{\vpar}  node[left ]{\wDepth\brule{e-1}}  (nn);
	\end{tikzpicture}

	\noindent
	In this diagram the tiles at \m{(i,j)} for \m{i,j ∈ \set{0,\dots,n-1}} below the diagonal (\m{i<j}) look as follows:\\*
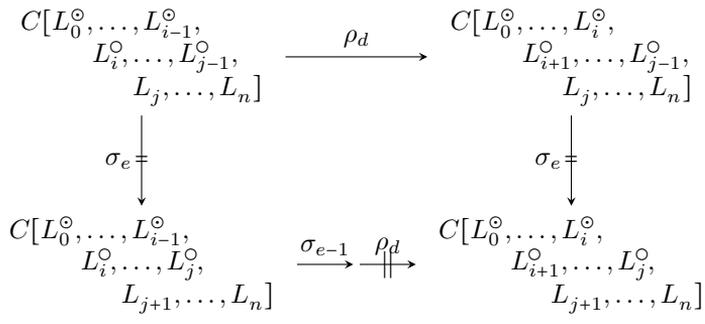
	\begin{tikzpicture}
		\matrix[row sep=\tilerowsep,column sep=\tilecolsep/2]{
			\node(00){\atlcxtBRBLTL{0}{i-1}{i}{j-1}{j}{n}};&&
			\node(10){\atlcxtBRBLTL{0}{i}{i+1}{j-1}{j}{n}};\\
			\node(01){\atlcxtBRBLTL{0}{i-1}{i}{j}{j+1}{n}};&
			\emptynode{0h1};&
			\node(11){\atlcxtBRBLTL{0}{i}{i+1}{j}{j+1}{n}};\\
		};
		\draw[->](00)  to              node[above]{\wDepth\arule{d}}    (10);
		\draw[->](00)  to node{\vpar}  node[left ]{\wDepth\brule{e}}    (01);
		\draw[->](01)  to              node[above]{\wDepth\brule{e-1}}    (0h1);
		\draw[->](0h1) to node{\hpar} node[above]{\wDepth\arule{d}}    (11);
		\draw[->](10)  to node{\vpar}  node[left ]{\wDepth\brule{e}}    (11);
	\end{tikzpicture}

	\noindent
	The tiles at \m{(i,j)} for \m{i,j ∈ \set{0,\dots,n-1}} above the diagonal (\m{i>j}) look as follows:\\*
	\begin{tikzpicture}
		\matrix[row sep=\tilerowsep/2,column sep=\tilecolsep/2]{
			\node(00){\atlcxtBRTRTL{0}{i-1}{i}{j-1}{j}{n}};&
			\node(10){\atlcxtBRTRTL{0}{i-1}{i}{j}{j+1}{n}};\\
			\emptynode{00h};&\emptynode{10h};\\
			\node(01){\atlcxtBRTRTL{0}{i}{i+1}{j-1}{j}{n}};&
			\node(11){\atlcxtBRTRTL{0}{i}{i+1}{j}{j+1}{n}};\\
		};
		\draw[->](00)  to             node[above]{\wDepth\arule{d}}    (10);
		\draw[->](00)  to node{\vpar} node[left ]{\wDepth\brule{e}}    (00h);
		\draw[->](00h) to node{\vpar} node[left ]{\wDepth\brule{e-1}}    (01);
		\draw[->](10)  to node{\vpar} node[left ]{\wDepth\brule{e}}    (10h);
		\draw[->](10h) to node{\vpar} node[left ]{\wDepth\brule{e-1}}    (11);
		\draw[->](01)  to             node[above]{\wDepth\arule{d}}    (11);
	\end{tikzpicture}

	\noindent
	For \m{i ∈ \set{0,\dots,n-1}} the \i-th diagonal tile looks like this:\\*
	\begin{tikzpicture}
		\matrix[row sep=\tilerowsep,column sep=\tilecolsep/2]{
			\node(00){\acxt{\L_0^{\BR}, \dots, \L_{i-1}^{\BR}, \L_i^{\TL}, \dots, \L_n^{\TL}}};&&
			\node(10){\acxt{\L_0^{\BR}, \dots, \L_{i-1}^{\BR}, \L_i^{\TR}, \L_{i+1}^{\TL}, \dots, \L_n^{\TL}}};\\
			&&\emptynode{10h};\\
			\node(01){\acxt{\L_0^{\BR}, \dots, \L_{i-1}^{\BR}, \L_i^{\BL}, \L_{i+1}^{\TL}, \dots, \L_n^{\TL}}};&
			\emptynode{0h1};&
			\node(11){\acxt{\L_0^{\BR}, \dots, \L_i^{\BR}, \L_{i+1}^{\TL}, \dots, \L_n^{\TL}}};\\
		};
		\draw[->](00)  to            node[above]{\wDepth\arule{d}}    (10);
		\draw[->](00)  to node{\vpar}  node[left ]{\wDepth\brule{e}}    (01);
		\draw[->](01)  to              node[above]{\wDepth\brule{e-1}}    (0h1);
		\draw[->](0h1)  to node{\hpar} node[above]{\wDepth\arule{d}}    (11);
		\draw[->](10)  to node{\vpar}  node[left ]{\wDepth\brule{e}}    (10h);
		\draw[->](10h)  to node{\vpar}  node[left ]{\wDepth\brule{e-1}}    (11);
	\end{tikzpicture}

	\vspace{2ex}
	Every hole on the diagonal is filled with at most one \wDepth\arule{d}-redex (at the root of the context hole fillings) but because of \m{d < e} with possibly many \wDepth\brule{e}-redexes (properly inside of the fillings). There may or may not be an overlap between the \wDepth\arule{d}-step and a \wDepth\brule{e}-step, but there can be at most one, which is due to the rules of \unfbhCRS.

	Therefore \wDepth\brule{e} contracts either an overlap and a number of nested redexes, or only nested redexes without an overlap. These constellations are depicted on the figure below. There is one \wDepth\arule{d}-redex and three \wDepth\brule{e}-redexes. On the left, one of the \wDepth\brule{e}-redexes overlaps with the \wDepth\arule{d}-redex while on the right all \wDepth\brule{e}-redexes are strictly nested inside the \wDepth\arule{d}-redex.

	\begin{hspread}
		\fig{overlap}
	\end{hspread}

	For the critical pairs due to a non-root overlap, and for all situations with nested redexes, we construct diagrams of the following shape, respectively:

	\noindent
	
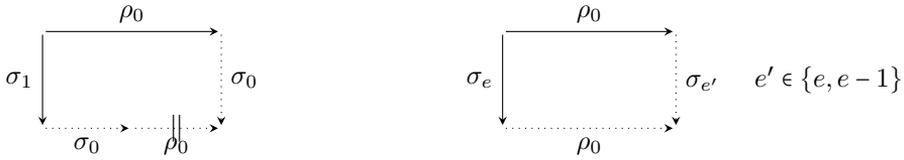
\begin{tikzpicture}
		\matrix[row sep=1.2cm,column sep=1.1cm]{
			\emptynode{tl};&&
			\emptynode{tr};\\
			\emptynode{bl};&
			\emptynode{bm};&
			\emptynode{br};\\
			\\
		};
		\draw[->       ](tl) to            node[above]{\wDepth\arule{0}}   (tr);
		\draw[->       ](tl) to            node[left ]{\wDepth\brule{1}}   (bl);
		\draw[->,dotted](bl) to            node[below]{\wDepth\brule{0}} (bm);
		\draw[->,dotted](bm) to node{\hpar} node[below]{\wDepth\arule{0}}   (br);
		\draw[->,dotted](tr) to            node[right]{\wDepth\brule{0}} (br);
	\end{tikzpicture}
	\hfill
	\begin{tikzpicture}
		\matrix[row sep=1.2cm,column sep=2.2cm]{
			\emptynode{tl};&
			\emptynode{tr};\\
			\emptynode{bl};&
			\emptynode{br};\\
			\\
		};
		\draw[->       ](tl) to            node[above]{\wDepth\arule{0}}   (tr);
		\draw[->       ](tl) to            node[left ]{\wDepth\brule{e}}   (bl);
		\draw[->,dotted](bl) to            node[below]{\wDepth\arule{0}}   (br);
		\draw[->,dotted](tr) to            node[right]{\m{\wDepth\brule{e'} ~~~~ e' \in \set{e,e-1}}} (br);
	\end{tikzpicture}

	When lifted into a context of \letrec-depth \m{d} both of the diagrams comply to the shape necessary for the diagonal tiles, but we need to be able to handle situations as on as on the left of the above figure, where both nested redexes as well as the overlapping redex are contracted. Firstly, since all \brule-redexes occur at the same \letrec-depth, it must hold that \m{d = 0} and \m{e = 1}, which is due to the rules of \unfbhCRS. Secondly, none of the involved redex contractions affect any of the nested redexes except for duplicating or erasing them, which means that the residuals of the \brule-steps after these steps are part of a parallel \wDepth\brule{e'}-step (mind that we assume the reflexive closure of all steps). Or as a diagram:

	\noindent
	
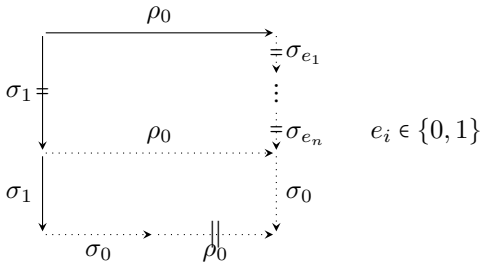
\begin{tikzpicture}
		\matrix[row sep=0.5cm,column sep=1.4cm]{
			\emptynode{tl};&&
			\emptynode{tr};\\
			&&
			\node(ir){\vdots};\\
			\emptynode{ml};&&
			\emptynode{mr};\\
			\\
			\emptynode{bl};&
			\emptynode{bm};&
			\emptynode{br};\\
		};
		\draw[->       ](tl) to            node[above]{\wDepth\arule{0}}   (tr);
		\draw[->       ](tl) to node{\vpar} node[left ]{\wDepth\brule{1}}   (ml);
		\draw[->       ](ml) to            node[left ]{\wDepth\brule{1}}   (bl);
		\draw[->,dotted](ml) to            node[above]{\wDepth\arule{0}}   (mr);
		\draw[->,dotted](mr) to            node[right]{\wDepth\brule{0}}   (br);
		\draw[->,dotted](bl) to            node[below]{\wDepth\brule{0}}   (bm);
		\draw[->,dotted](bm) to node{\hpar} node[below]{\wDepth\arule{0}}   (br);
		\draw[->,dotted](tr) to node{\vpar} node[right]{\wDepth\brule{e_1}}   (ir);
		\draw[->,dotted](ir) to node{\vpar} node[right]{\m{\wDepth\brule{e_n} ~~~~~ e_i \in \set{0,1}}} (mr);
	\end{tikzpicture}

	The diagram is composed from the previous two diagrams. A parallel version of the right one constitutes the top part, while the bottom part is an exact replica of the left one. The top part settles the portion arising from the nested redexes, the bottom part settles the portion arising from the overlapping redex.

	At last in order to fit that diagram into the scheme of the diagonal tiles the steps on the right have to be reordered such that \wDepth\brule{e_i}-steps with \m{e_i = 1} preceed \wDepth\brule{e_i}-steps with \m{e_i = 0}. The reordering is viable because every \wDepth\brule{e_i}-step takes place in its own residual of the \wDepth\brule1-step from the left.

	We conclude the proof by a comprehensive analysis all critical pairs that arise from non-root overlaps in \unfbhCRS as well as the diagrams for joining nested redexes.

	\subsection{Diagrams for joining critical pairs}

	\wDepth{\arule}{0}/\wDepth{\brule}{0} critical pairs only arise for \m{\arule=\merg}.

	\noindent
	\begin{tikzpicture}
		\matrix{
			\node(tl){\letin{B}{\letin{C}{\abs{x}{\L}}}};&
			\node(tr){\letin{B,C}{\abs{x}{\L}}};\\
			\node(ml){\letin{B}{\abs{x}{\letin{C}\L}}};\\
			\node(bl){\abs{x}{\letin{B}{\letin{C}\L}}};&
			\node(br){\abs{x}{\letin{B,C}\L}};\\
		};
		\draw[->       ](tl) to node[above]{\wDepth\merg{0}} (tr);
		\draw[->       ](tl) to node[left ]{\wDepth{λ}{1}} (ml);
		\draw[->,dotted](tr) to node[right]{\wDepth{λ}{0}}   (br);
		\draw[->,dotted](ml) to node[left ]{\wDepth{λ}{0}}   (bl);
		\draw[->,dotted](bl) to node[below]{\wDepth\merg{0}} (br);
	\end{tikzpicture}

	\noindent
	\begin{tikzpicture}
		\matrix{
			\node(tl){\letin{B}{\letin{C}{\app\L{P}}}};&
			\node(tr){\letin{B,C}{\app\L{P}}};\\
			\node(ml){\letin{B}{\appbreak{\letin{C}\L}{\letin{C}{P}}}};\\
			\node(bl){\appbreak{\letin{B}{\letin{C}\L}}{\letin{B}{\letin{C}{P}}}};&
			\node(br){\appbreak{\letin{B,C}\L}{\letin{B,C}{P}}};\\
		};
		\draw[->       ](tl) to              node[above]{\wDepth\merg{0}} (tr);
		\draw[->       ](tl) to              node[left ]{\wDepth{@}{1}} (ml);
		\draw[->,dotted](tr) to              node[right]{\wDepth{@}{0}}   (br);
		\draw[->,dotted](ml) to              node[left ]{\wDepth{@}{0}}   (bl);
		\draw[->,dotted](bl) to node{\hpar} node[below]{\wDepth\merg{0}} (br);
	\end{tikzpicture}

	\noindent
	\begin{tikzpicture}
		\matrix{
			\node(tl){\letin{B}{\letin{C}{\letin{D}\L}}};&
			\node(tr){\letin{B}{\letin{C,D}\L}};\\
			\node(bl){\letin{B,C}{\letin{D}\L}};&
			\node(br){\letin{B,C,D}\L};\\
		};
		\draw[->       ](tl) to node[above]{\wDepth\merg{0}}   (tr);
		\draw[->       ](tl) to node[left ]{\wDepth\merg{1}}   (bl);
		\draw[->,dotted](tr) to node[right]{\wDepth\merg{0}}   (br);
		\draw[->,dotted](bl) to node[below]{\wDepth\merg{0}}   (br);
	\end{tikzpicture}

	\noindent
	\begin{tikzpicture}
		\matrix{
			\node(tl){\letin{B}{\letin{C}{f_i}}};&
			\node(tr){\letin{B,C}{f_i}};\\
			\node(bl){\letin{B}{\letin{C}{\L_i}}};&
			\node(br){\letin{B,C}{\L_i}};\\
		};
		\draw[->       ](tl) to node[above]{\wDepth\merg{0}} (tr);
		\draw[->       ](tl) to node[left ]{\wDepth\rec{1}}    (bl);
		\draw[->,dotted](tr) to node[right]{\wDepth\rec{0}}    (br);
		\draw[->,dotted](bl) to node[below]{\wDepth\merg{0}} (br);
	\end{tikzpicture}

	\noindent
	\begin{tikzpicture}
		\matrix{
			\node(tl){\letin{B}{\letin{}\L}};&
			\node(tr){\letin{B}\L};\\
			\node(bl){\letin{B}\L};&
			\node(br){\letin{B}\L};\\
		};
		\draw[->       ](tl) to node[above]{\wDepth\merg{0}} (tr);
		\draw[->       ](tl) to node[left ]{\wDepth\nil{1}}   (bl);
		\draw[double   ](tr) to                                   (br);
		\draw[double   ](bl) to                                   (br);
	\end{tikzpicture}

	\noindent
	\begin{tikzpicture}
		\matrix{
			\node(tl){\letin{B}{\letin{C}\L}};&
			\node(tr){\letin{B,C}\L};\\
			\node(bl){\letin{B}{\letin{C'}\L}};&
			\node(br){\letin{B,C'}\L};\\
		};
		\draw[->       ](tl) to node[above]{\wDepth\merg{0}}  (tr);
		\draw[->       ](tl) to node[left ]{\wDepth\reduce{1}} (bl);
		\draw[->,dotted](tr) to node[right]{\wDepth\reduce{0}} (br);
		\draw[->,dotted](bl) to node[below]{\wDepth\merg{0}}  (br);
	\end{tikzpicture}

	\noindent
	\begin{tikzpicture}
		\matrix{
			\node(tl){\letin{B}{\letin{C,f=g}\L}};&
			\node(tr){\letin{B,C,f=g}\L};\\
			\node(bl){\letin{B}{\letin{\subst{C}{f}{g}}{\subst{\L}{f}{g}}}};&
			\node(br){\letin{\subst{B,C}{f}{g}}{\subst{\L}{f}{g}}};\\
		};
		\draw[draw=none](tl) to node{\m{f ≠ g}} (br);
		\draw[->       ](tl) to node[above]{\wDepth\merg{0}}   (tr);
		\draw[->       ](tl) to node[left ]{\wDepth\tighten{1}} (bl);
		\draw[->,dotted](tr) to node[right]{\wDepth\tighten{0}} (br);
		\draw[->,dotted](bl) to node[below]{\wDepth\merg{0}}   (br);
	\end{tikzpicture}

	\noindent
	\begin{tikzpicture}
		\matrix{
			\node(tl){\letin{B}{\letin{C,f=f}\L}};&
			\node(tr){\letin{B,C,f=f}\L};\\
			\node(bl){\letin{B}{\letin{\subst{C}{f}{\bh}}{\subst{\L}{f}{\bh}}}};&
			\node(br){\letin{\subst{B,C}{f}{\bh}}{\subst{\L}{f}{\bh}}};\\
		};
		\draw[->       ](tl) to node[above]{\wDepth\merg{0}} (tr);
		\draw[->       ](tl) to node[left ]{\wDepth\bh{1}}    (bl);
		\draw[->,dotted](tr) to node[right]{\wDepth\bh{0}}    (br);
		\draw[->,dotted](bl) to node[below]{\wDepth\merg{0}} (br);
	\end{tikzpicture}

	\subsection{Diagrams for joining nested redexes}

	\noindent
	\begin{tikzpicture}
		\matrix{
			\node(tl){\letin{B}{\abs{x}\L}};&
			\node(tr){\abs{x}{\letin{B}\L}};\\
			\node(bl){\letin{B'}{\abs{x}{\L'}}};&
			\node(br){\abs{x}{\letin{B'}{\L'}}};\\
		};
		\draw[->       ](tl) to            node[above]{\wDepth{λ}{0}} (tr);
		\draw[->       ](tl) to            node[left ]{\wDepth\brule{e}}    (bl);
		\draw[->,dotted](tr) to            node[right]{\wDepth\brule{e}}    (br);
		\draw[->,dotted](bl) to            node[below]{\wDepth{λ}{0}} (br);
	\end{tikzpicture}

	\noindent
	\begin{tikzpicture}
		\matrix{
			\node(tl){\letin{B}{\app{\L_0}{\L_1}}};&
			\node(tr){\appbreak{\letin{B}{\L_0}}{\letin{B}{\L_1}}};\\
			\node(bl){\letin{B'}{\app{\L_0'}{\L_1'}}};&
			\node(br){\appbreak{\letin{B'}{\L_0'}}{\letin{B'}{\L_1'}}};\\
		};
		\draw[->       ](tl) to            node[above]{\wDepth{@}{0}} (tr);
		\draw[->       ](tl) to            node[left ]{\wDepth\brule{e}}    (bl);
		\draw[->,dotted](tr) to node{\vpar} node[right]{\wDepth\brule{e}}    (br);
		\draw[->,dotted](bl) to            node[below]{\wDepth{@}{0}} (br);
	\end{tikzpicture}

	\noindent
	\begin{tikzpicture}
		\matrix{
			\node(tl){\letin{B}{\letin{C}\L}};&
			\node(tr){\letin{B,C}\L};\\
			\node(bl){\letin{B'}{\letin{C'}{\L'}}};&
			\node(br){\letin{B',C'}{\L'}};\\
		};
		\draw[->       ](tl) to            node[above]{\wDepth\merg{0}} (tr);
		\draw[->       ](tl) to            node[left ]{\wDepth\brule{e}}      (bl);
		\draw[->,dotted](tr) to            node[right]{\wDepth\brule{e' ~~~~~~ e' \in \set{e-1,e}}}      (br);
		\draw[->,dotted](bl) to            node[below]{\wDepth\merg{0}} (br);
	\end{tikzpicture}

	\noindent
	\begin{tikzpicture}
		\matrix{
			\node(tl){\letin{B}{f}};&
			\node(tr){\letin{B}{\L}};\\&
			\node(mr){\letin{B'}{\L}};\\
			\node(bl){\letin{B'}{f}};&
			\node(br){\letin{B'}{\L'}};\\
		};
		\draw[->       ](tl) to node[above]{\wDepth\rec{0}} (tr);
		\draw[->       ](tl) to node[left ]{\wDepth\brule{e}} (bl);
		\draw[->,dotted](tr) to node[right]{\wDepth\brule{e}} (mr);
		\draw[->,dotted](mr) to node[right]{\wDepth\brule{e}} (br);
		\draw[->,dotted](bl) to node[below]{\wDepth\rec{0}} (br);
	\end{tikzpicture}

	\noindent
	\begin{tikzpicture}
		\matrix{
			\node(tl){\letin{}\L};&
			\node(tr){\L};\\
			\node(bl){\letin{}{\L'}};&
			\node(br){\m{\L'}};\\
		};
		\draw[->       ](tl) to node[above]{\wDepth\nil{0}} (tr);
		\draw[->       ](tl) to node[left ]{\wDepth\brule{e}}   (bl);
		\draw[->,dotted](tr) to node[right]{\wDepth\brule{e-1}}   (br);
		\draw[->,dotted](bl) to node[below]{\wDepth\nil{0}} (br);
	\end{tikzpicture}

	\noindent
	\begin{tikzpicture}
		\matrix{
			\node(tl){\letin{B}\L};&
			\node(tr){\letin{B^{\TR}}\L};\\
			\node(bl){\letin{B^{\BL}}{\L'}};&
			\node(br){\letin{B^{\BR}}{\L'}};\\
		};
		\draw[->       ](tl) to node[above]{\wDepth\reduce{0}} (tr);
		\draw[->       ](tl) to node[left ]{\wDepth\brule{e}}      (bl);
		\draw[->,dotted](tr) to node[right]{\wDepth\brule{e}}      (br);
		\draw[->,dotted](bl) to node[below]{\wDepth\reduce{0}} (br);
	\end{tikzpicture}

	\noindent
	\begin{tikzpicture}
		\matrix{
			\node(tl){\letin{B,f=g}{\L}};&
			\node(tr){\letin{\subst{B}{f}{g}}{\subst{\L}{f}{g}}};\\
			\node(bl){\letin{B',f=g}{\L'}};&
			\node(br){\letin{\subst{B'}{f}{g}}{\subst{\L'}{f}{g}}};\\
		};
		\draw[draw=none](tl) to node{\m{f ≠ g}} (br);
		\draw[->       ](tl) to node[above]{\wDepth\tighten{0}} (tr);
		\draw[->       ](tl) to node[left ]{\wDepth\brule{e}}      (bl);
		\draw[->,dotted](tr) to node[right]{\wDepth\brule{e}}      (br);
		\draw[->,dotted](bl) to node[below]{\wDepth\tighten{0}} (br);
	\end{tikzpicture}

	\noindent
	\begin{tikzpicture}
		\matrix{
			\node(tl){\letin{B,f=f}{\L}};&
			\node(tr){\letin{\subst{B}{f}{\bh}}{\subst{\L}{f}{\bh}}};\\
			\node(bl){\letin{B',f=g}{\L'}};&
			\node(br){\letin{\subst{B'}{f}{\bh}}{\subst{\L'}{f}{\bh}}};\\
		};
		\draw[->       ](tl) to node[above]{\wDepth\tighten{0}} (tr);
		\draw[->       ](tl) to node[left ]{\wDepth\brule{e}}      (bl);
		\draw[->,dotted](tr) to node[right]{\wDepth\brule{e}}      (br);
		\draw[->,dotted](bl) to node[below]{\wDepth\tighten{0}} (br);
	\end{tikzpicture}
\end{proof}

%
%
%
%
%

\chapter{Unfolding with a Single Rule}\label{app:single_rule}

Here, we present an alternative rewriting system \unfOne for unfolding \lambdaletrec-terms, which has only a single rule, similar to μ-unfolding. We will briefly argue that the system unfolds to the same λ-terms as \unfCRS. For simplicity we will not address the problem of meaningless bindings, which can be resolved with black holes and appropriate rewriting rules, as done for \unfCRS.

\newcommand\substcompr[3]{\m{\left[{#1}:={#2}~\middle|~{#3}\right]}}
\newcommand\substicompr[2]{\m{\left[{#1}:={#2}~\middle|~1≤i≤n\right]}}
\newcommand\substjcompr[2]{\m{\left[{#1}:={#2}~\middle|~1≤j≤n\right]}}

\begin{notation}
	Hereinafter, \B stands for \m{f_1 = L_1, \dots, f_n = L_n}.
	Furthermore \substcompr{f_i}{L_i}{1≤i≤n} stands for the substitution \m{[f_1 := L_1, \dots, f_n := L_n]}.
\end{notation}

\begin{definition}[One-rule rewriting system for unfolding \lambdaletrec]
	The rewriting system \unfOne consists of the single rule:
	\[ \letin{B}{L} ~→~ L\substicompr{f_i}{L_i\substjcompr{f_j}{\letin{B}{f_j}}} \]

	\noindent
	In CRS notation:
	\begin{gather*}
		\letCRS{n}{\vec{f}}{L_1(\vec{f}),\dots,L_n(\vec{f})),L_0(\vec{f})}
		→ L_0(L_1(\vec{B}), \dots, L_n(\vec{B}))
		\\
		\text{where \vec{B} stands for~} \left(
			\begin{array}{l}
				\letCRS{n}{\vec{f}}{L_1(\vec{f}), \dots, L_n(\vec{f}), L_1(\vec{f})},\\
				~~~\dots,\\
				\letCRS{n}{\vec{f}}{L_1(\vec{f}), \dots, L_n(\vec{f}), L_n(\vec{f})}
			\end{array}\right)
	\end{gather*}
\end{definition}

\begin{para}[an even simpler approach?]\label{simpler_unfolding}
	One might wonder why the nested substitutions are necessary, why a simpler rule as follows would not be adequate:
	\[ \letin{B}{L} → L\substicompr{f_i}{\letin{B}{f_i}} \]
	This rule is not sufficient because it cannot unfold a term ~\letin{B}{L}~ correctly if \L is a variable bound in \B; it reduces a term like \letin{f=x}{f} to itself.
\end{para}

\begin{proposition}\label{prop:one_is_confluent}
	\unfOne is confluent. (\textit{Proof:} \unfOne is orthogonal.)
\end{proposition}

In order to argue that \unfOne reduces \lambdaletrec-terms to the same infinite normal forms as \unfCRS (let us write this as \unfOne ≡ \unfCRS), we introduce an intermediate system \unfTwo consisting of the rule from \cref{simpler_unfolding}, and one additional rule from \unfCRS to handle the problematic case.

\begin{definition}[Two-rule rewriting system for unfolding \lambdaletrec]
	The rewriting system \unfTwo consists of the following two rules:
	\rulearray{
		\letin{B}{L} → L\substicompr{f_i}{\letin{B}{f_i}}
		\\~~\sideCondition{if \L is not a variable bound in \B}
		\\
		\letin{B}{f_i} → \letin{B}{L_i}
	}
\end{definition}

\begin{proposition}[\unfTwo is a refinement of \unfOne] \m{\red_\unfOne ~\subseteq~ \mred_\unfTwo}\label{prop:two_refines_one}
	\begin{proof}
		Let us consider such a redex \letin{B}{L}.

		\paragraph{Case 1} (\L is not a variable bound by \B).

		\paragraph{Case 1.1} (all the \m{L_i} below are not variables bound by \B)
		\begin{align*}
			\letin{B}{L} &  \red_\unfTwo && L\substicompr{f_i}{\letin{B}{f_i}}\\
			             & \mred_\unfTwo && L\substicompr{f_i}{\letin{B}{L_i}}\\
			             & \mred_\unfTwo && L\substicompr{f_i}{L_i\substjcompr{f_j}{\letin{B}{f_j}}}\\
			             & \convred_\unfOne && \letin{B}{L}\\
		\end{align*}
		\paragraph{Case 1.2} (all the \m{L_i} below are variables bound by \B (\m{L_i = f_{j_i}}))
		\begin{align*}
			\letin{B}{L} &  \red_\unfTwo && L\substicompr{f_i}{\letin{B}{f_i}}\\
			             & \mred_\unfTwo && L\substicompr{f_i}{\letin{B}{L_i}}\\
			             & =                 && L\substicompr{f_i}{\letin{B}{f_{j_i}}}\\
			             & =                 && L\substicompr{f_i}{f_{j_i}\substjcompr{f_j}{\letin{B}{f_j}}}\\
			             & =                 && L\substicompr{f_i}{L_i\substjcompr{f_j}{\letin{B}{f_j}}}\\
			             & \convred_\unfOne && \letin{B}{L}\\
		\end{align*}

		Actually, the case distinction of \textbf{Case 1} into \textbf{Case 1.1} and \textbf{Case 2.2} is non-exhaustive. Instead, all mixtures of \textbf{Case 1.1} and \textbf{Case 2.2} have to be considered, where some of the \m{L_i} are bound by \B and others are not. To write this out would however be merely a tedious exercise.

		\paragraph{Case 2} (\L is a variable bound by \B (\m{L = f_i})).

		\paragraph{Case 2.1} (\m{L_i} below is not a variable bound by \B)
		\begin{align*}
			\letin{B}{f_i} & \red_\unfTwo     && \letin{B}{L_i}\\
			               & \red_\unfTwo     && L_i\substicompr{f_i}{\letin{B}{f_i}}\\
			               & =                    && f_i\substicompr{f_i}{L_i\substjcompr{f_j}{\letin{B}{f_j}}}\\
			               & \convred_\unfOne && \letin{B}{f_i}\\
		\end{align*}
		\paragraph{Case 2.2} (\m{L_i} below is a variable bound by \B (\m{L_i = f_j}))
		\begin{align*}
			\letin{B}{f_i} & \red_\unfTwo     && \letin{B}{L_i}\\
			               & =                    && \letin{B}{f_j}\\
			               & =                    && f_j\substicompr{f_i}{\letin{B}{f_i}}\\
			               & =                    && L_i\substicompr{f_i}{\letin{B}{f_i}}\\
			               & =                    && f_i\substicompr{f_i}{L_i\substjcompr{f_j}{\letin{B}{f_j}}}\\
			               & \convred_\unfOne && \letin{B}{f_i}\\
		\end{align*}
	\end{proof}
\end{proposition}

\begin{proposition}[\unfTwo is a compatible refinement of \unfOne]\label{prop:two_is_compatible_to_one}
	\unfTwo is a refinement of \unfOne, and additionally it holds:
	\[∀M,N∈\Ter{\lambdaletreccal}~ (M \mred_\unfTwo N ~\implies~ M \mred_\unfOne . \convmred_\unfOne N)\]
	\begin{proof}
		Let us consider a \m{\red_\unfTwo}-redex \m{M = \letin{B}{L}} and the reduction \m{M \red_\unfTwo N}. We will show that \m{M \mred_\unfOne . \mred_\unfOne N}
		\paragraph{Case 1} (\L is not a variable bound by \B).
		\[M = \letin{B}{L} \red_\unfTwo L\substicompr{f_i}{\letin{B}{f_i}} = N\]
		\begin{align*}
			M            &  \red_\unfOne && L\substicompr{f_i}{L_i\substjcompr{f_j}{\letin{B}{f_j}}}\\
			             & =                && L\substicompr{f_i}{f_i\substicompr{f_i}{L_i\substjcompr{f_j}{\letin{B}{f_j}}}}\\
			             & \convred_\unfOne && N\\
		\end{align*}
		\paragraph{Case 2} (\L is a variable bound by \B (\m{L = f_i})).
		\[ M = \letin{B}{f_i} \red_\unfTwo \letin{B}{L_i} = N\]
		\begin{align*}
			M              & \red_\unfOne     && f_i\substicompr{f_i}{L_i\substjcompr{f_j}{\letin{B}{f_j}}}\\
			               & =                && L_i\substjcompr{f_j}{\letin{B}{f_j}}\\
			               & \red_\unfOne     && L_i\substjcompr{f_j}{f_j\substicompr{f_i}{L_i\substjcompr{f_j}{\letin{B}{f_j}}}}\\
			               & =                && L_i\substicompr{f_i}{L_i\substjcompr{f_j}{\letin{B}{f_j}}}\\
			               & \convred_\unfOne && \letin{B}{L_i}\\
		\end{align*}
	\end{proof}
\end{proposition}

\begin{proposition} \unfTwo is confluent.
	\begin{proof}
		This follows from \cref{prop:one_is_confluent} and \cref{prop:two_is_compatible_to_one} \cite{staples:1975}.
	\end{proof}
\end{proposition}

\begin{proposition}[\unfOne ≡ \unfTwo]\label{prop:unfOneequnfTwo}
	\[ ∀L∈\Ter{\lambdaletrec}~ L \m{\infred_{\unfOne}^!} L' ~\iff~ L \m{\infred_\unfTwo^!} L' \]
\end{proposition}

\begin{proposition}[\m{\unfCRS ≡ \unfTwo}]\label{prop:unfCRSequnfTwo}
	\[ ∀L∈\Ter{\lambdaletrec}~ L \m{\infred_{\unfold}^!} L' ~\iff~ L \m{\infred_\unfTwo^!} L' \]
	\begin{proof}[Proof idea]
		Only a rudimentary proof idea is given here, i.e.\ that the property
		\[\red_\unfold ~\subseteq~ \mred_\unfTwo . \convmred_\unfTwo\]
		can be used to construct a coinductive proof. In order to convince ourselves that the property holds we only look at three cases here. In each of these cases the leftmost term reduces via \m{\red_\unfold} to the rightmost term.
		\paragraph{Case: \m{λ}}
		\begin{align*}
			\letin{B}{\abs{x}{L}} & \red_\unfTwo && (\abs{x}{L}){\substicompr{f_i}{\letin{B}{f_i}}}\\
			                      & =            && \abs{x}{L\substicompr{f_i}{\letin{B}{f_i}}}
			                      & \convmred_\unfTwo \abs{x}{\letin{B}{L}}
		\end{align*}
		\paragraph{Case: \m{@}}
		\begin{align*}
			\letin{B}{\app{L}{P}} &&\red_\unfTwo& (\app{L}{P}){\substicompr{f_i}{\letin{B}{f_i}}}\\
										 &&=& \appbreak{L\substicompr{f_i}{\letin{B}{f_i}}}{P\substjcompr{f_j}{\letin{B}{f_j}}} \convmred_\unfTwo \appbreak{\letin{B}{L}}{\letin{B}{P}}
		\end{align*}

		\paragraph{Case: \merg}
		\begin{align*}
			& \letin{B}{\letin{C}{L}}\\
			  \red_\unfTwo      & (\letin{C}{L})\substicompr{f_i}{\letin{B}{f_i}} \\
			  \red_\unfTwo      & (L\substicompr{g_i}{\letin{C}{g_i}})\substjcompr{f_j}{\letin{B}{f_j}} \\
			  =                 & L(\substicompr{g_i}{\letin{C}{g_i}} \cup \substjcompr{f_j}{\letin{B}{f_j}}) \\
			  \convred_\unfTwo  & \letin{B,C}{L}
		\end{align*}
	\end{proof}
\end{proposition}

\begin{proposition}[\m{\unfOne ≡ \unfCRS}]
	\[ ∀L∈\Ter{\lambdaletrec}~ L \m{\infred_{\unfOne}^!} L' ~\iff~ L \m{\infred_\unfold^!} L' \]
	\begin{proof}
		Follows from \cref{prop:unfCRSequnfTwo} and \cref{prop:unfOneequnfTwo}.
	\end{proof}
\end{proposition}

\backmatter


\chapter{Curriculum Vitae}

\newcommand\cventry[4]{\item[#1]\textbf{#2}, \textit{#3}\ifempty{#4}{}{\\#4}}

\noindent%
birth: 20 July 1984 in Bretten, Germany

\begin{description}
\cventry{1990--1994}{Elementary school}{Grundschule Nussbaum}{}
\cventry{1994--2003}{Secondary school}{Melanchthon-Gymnasium Bretten}{Major subjects: mathematics, physics}
\cventry{2002--2003}{Schülerstudium Informatik}{Universität Karlsruhe}{As a `\plaintext{bragger}{student} with exceptional learning achievements' I was given the opportunity to study computer science at the university during my last year of secondary school.}
\cventry{2003--2004}{Research assistance}{Forschungszentrum Informatik, Karlsruhe}{}
\cventry{2003--2009}{Study of computer science}{Universität Karlsruhe}{Major subjects: system architecture, compiler construction and software engineering}
\cventry{2009}{Research semester}{Chalmers University of Technology, Göteborg}{Diploma thesis under the supervision of the functional programming research group}
\cventry{2009-2013}{Doctoral research}{Departement Informatica, Universiteit Utrecht}{Research and teaching assistance in the field of functional programming}
\end{description}

\chapter{Samenvatting in het Nederlands {\normalsize(Summary in Dutch)}}
\vspace{-1ex}
\let\i\iNoDot
Dit proefschrift bevat resultaten uit onderzoek van de ongetypeerde λ-calculus (lees: ,,lambda calculus''). De λ-calculus is een formeel systeem dat in de jaren 1930 is ontworpen en is sindsdien onderzoeksonderwerp in informatica en filosofie:
\begin{itemize}
	\setlength\itemsep{-2ex}
	\newcommand\subF{\m{>}}
	\item Filosofie \subF Formele Logica \subF Herschrijven \subF λ-Calculus\\
	\item Informatica \subF Programmeertalen en Compilerconstructie \subF Functionele Programmeertalen \subF λ-Calculus
\end{itemize}
Het onderzoek is voornamelijk gemotiveerd door de rol van de λ-calculus als basis van functionele programmeertalen. Het doel van het onderzoeksproject was om de uitvoering van in functionele programeertalen geschreven programma's efficiënter te maken door het verhogen van \emph{sharing}. Het Engelse begrip sharing duidt het fenomeen aan dat een berekening \emph{gedeeld} wordt, dat wil zeggen dat een waarde die op meerderen plekken nodig is niet iedere keer opnieuw wordt uitgerekend, maar dat die waarde slechts één keer wordt uitgerekend en dat het resultaat vervolgens op deze plekken wordt hergebruikt. Om dit te bereiken worden functionele talen doorgaans als graafherschrijfsysteem geïmplementeerd. In een graaf kan een knoop meerdere inkomende kanten hebben. Op die manier wordt die knoop en zijn opvolgers gedeeld. Maar vaak moet tijdens de uitvoering een gedeelde subgraaf worden ,,ontdeeld'' of ,,ontvouwd''. Hoe later en hoe voorzichtiger het ontdelen kan worden doorgevoerd des te hoger is de graad van sharing en de hierdoor bespaarde berekeningskosten. Er zijn verschillende evaluatiemodellen met verschillende graden van sharing. In dit proefschrift focussen wij niet op het dynamische delen tijdens de uitvoering maar op het statische delen \emph{voor} de uitvoerig van het programma. Het delen tijdens de uitvoering is namelijk gebaseerd op de initiële graad van sharing van de graaf die door de compiler van een programma wordt geconstrueerd. Deze graaf wordt vaak in een taal beschreven die gebaseerd is op de \emph{λ-calculus met \letrec} (of kort: \lambdaletrec). Zo noemen we de λ-calculus waarvan de termen niet alleen normale \lambda-termen zijn maar ook termen met voorkomens van de zogenoemde \letrec-constructie. De \letrec-constructie maakt het mogelijk om een graafstructuur en dus sharing direct in de programmeercode uit te drukken.

In dit proefschrift bestuderen we sharing in \lambdaletrec. Het centrale begrip voor hoe we een \lambdaletrec-term beschouwen is de \emph{ontvouwingssemantiek}: We beschouwen een \lambdaletrec-term als een representatie van een (potientieel) oneindige term. De rechtvaardiging hiervoor is dat dit overeenkomt met hoe functionele talen geëvalueerd worden: voordat een β-reductie plaatsvindt wordt het deel van de graaf dat de redex bevat eerst ontvouwd en het ontvouwen van een \lambdaletrec-term heeft (in de limiet) een (potentieel oneindige) λ-term als resultaat. We houden ons niet verder bezig met β-reductie maar onderzoeken uitsluitend ontvouwing van \lambdaletrec-termen.

Hoofdstuk~\labelcref{chap:introduction} introduceert het \lambdaletrec-formalisme samen met een herschrijfsysteem om \lambdaletrec-termen te ontvouwen. In de volgende drie hooftstukken bestuderen we dit herschrijfsysteem en beantwoorden o.\ a.\ de volgende vragen:
\begin{description}
	\item Welke oneindige λ-termen kunnen als eindige \lambdaletrec-termen worden uitgedrukt? Met anderen woorden: hoe expressief is de taal \lambdaletreccal? In hoofdstuk~\labelcref{chap:expressibility} karakteriseren we de verzameling van \lambdaletrec-uitdrukbaren λ-termen, dus die λ-termen die de ontvouwing van eindige \lambdaletrec-termen zijn.
	\item Wat zijn goede graafrepresentaties voor \lambdaletrec-uitdrukbare termen? De \letrec-constructie dient ertoe om een graafstructuur uit te drukken. In hoofdstuk~\labelcref{chap:representations} vragen we ons af hoe precies deze grafen kunnen worden geformaliseerd en we identificeren een formalisatie die voor de volgende vragen nuttig blijkt.
	\item Hoe kunnen we bepalen of twee \lambdaletrec-termen dezelfde ontvouwing hebben? Hoe kunnen we van een \lambdaletrec-term een equivalente maar zo compact mogelijke variante berekenen? Bestaat er zoiets als een zo compact mogelijke variant? In hoofdstuk~\labelcref{chap:maxsharing} tonen we aan dat er voor iedere klasse van equivalente \lambdaletrec-termen een maximaal compacte vorm bestaat. We ontwikkelen praktische en efficiënte methoden om deze vorm te berekenen en om te bepalen of twee \lambdaletrec-termen dezelfde ontvouwing hebben.
\end{description}

\chapter{Lay Summary}

This summary is intended to give the casual reader an idea what this thesis is about.

In this thesis I study a very specific subject on the intersection of the fields of computer science and philosophy, namely a formal system called the λ-calculus (read `lambda calculus`). The λ-calculus can be placed into these two fields of science as follows:
\begin{itemize}
	\setlength\itemsep{-2ex}
	\newcommand\subF{\m{>}}
	\item Philosophy \subF Formal Logic \subF Rewriting Systems \subF λ-Calculus\\ 
	\item Computer Science \subF Programming Languages and Compiler Construction \subF Functional Programming Languages \subF λ-Calculus
\end{itemize}

In order to convey the meaning of the thesis title, let us first establish what computers and programming languages are all about.\\

A computer is a machine for processing information. It has input channels (like a touchpad or a microphone) to receive data; it processes the data and performs computations; and it has output channels (like a display or a connector to another device) to convey the result of the computation, by making the result visible, or carrying out some mechanical activity. 

The behaviour of computers is determined by their programming. They are \emph{programmable}, meaning their behaviour can be changed by a programmer. But moreover, a computer \emph{requires} programming to function; there is no intelligence inherent to the device itself; it is the programmer that imbues the machine with his ideas through programming.

The programmer's ideas (say, a method to increase the contrast in a photograph) have to be expressed such that it can be understood by the computer. The language that a computer understands is its \emph{machine language} which is usually a very simple language and which is different for each type of computer. It consists of a collection of basic instructions that, when carried out by the computer, each make very small changes to the computer's memory. It is difficult and cumbersome for a programmer to express sophisticated ideas using a machine language.

Today, programmers do not need to program directly with machine languages, but have a vast number of \emph{programming languages} at their disposal. Programming languages are languages designed to make the task of programming more efficient and convenient, by giving programmers useful, more powerful means to encode their ideas. A program called \emph{compiler} then translates the code written in a programming language into machine language, which subsequently can be executed by the machine.
A very simple compiler that uses the most straightforward way to translate a programming language into machine language produces very inefficient machine language, which leads to a slower execution of the program. This is especially the case for more complex programming languages, which provide the programmer with a higher level of convenience. Modern compilers produce more efficent machine language by analysing the original program and then translating it in a more sophisticated way. Such a compiler is called an \emph{optimising} compiler.

This thesis focusses on one specific sub-class of programming languages called \emph{functional programming languages}. They differ from \emph{imperative programming languages}, which form the mainstream of today's programming languages. Code written in an imperative programming language can be regarded as a list of instructions which are executed sequentially, one by one. Code written in a functional programming languages, on the other hand, more resembles a set of mathematical equations, by which the result of the computation is defined. 

While functional programming has for decades led a niche existence mostly limited to the academic world, its merits and its elegance are being ever more recognised, and is steadily gaining ground. Today functional programming languages are employed in commerce and in industry, and modern imperative programming languages incorporate more and more features from functional programming languages.

Functional programming has its roots in the λ-calculus, a formal system developed in the 1930s. It is a \emph{rewriting system}; that means it acts on a specific kind of formal expressions (λ-terms) which can be rewritten in a rewriting step to a different term by a fixed set of rewriting rules. It is a sort of minimalistic programming language; the algorithm and the input are expressed as a λ-term; the computation consists of repeatedly rewriting the term according to the rewriting rules of the λ-calculus until it cannot be rewritten any further; the term one obtains in the end is the result of the computation.

There has been a multitude of extensions of the λ-calculus. One such extension is the inclusion of a language construct called \letrec, which we call the `λ-calculus with \letrec' or in short \lambdaletrec. The \letrec is a syntactic element which allows for function definitions. Function definitions are equations which bind terms to names. A bound term can be referenced by its name and thus used multiple times at different places, thereby introducing `sharing'. In the course of performing a computation in \lambdaletrec the definitions need to be \emph{unfolded}. Unfolding is the process of replacing occurrences of function names by their definition. As function names may occur within their own definition unfolding may go on forever resulting an infinite term. This thesis is all about the \letrec-construct and unfolding; it tackles questions like:
\begin{itemize}
	\item Which are the infinite terms that can be obtained by unfolding of finite \lambdaletrec-terms?
	\item When do two given \lambdaletrec-terms have the same unfolding?
	\item How can we find a maximally compact form of a given \lambdaletrec-term?
	\item What is a suitable graph representation for \lambdaletrec-terms?
\end{itemize}

These results provide compilers with further opportunities for analysis and for a more optimised translation which might speed up the execution of programs written in a functional programming languages. Furthermore they may also help reasoning about programs and thus promote further research about functional programming languages.

\chapter{Acknowledgements}

First and foremost my thanks go to Clemens Grabmayer: Danke für die Zusammenarbeit und viele interessante Diskussionen; das war eine sehr lehrreiche Zeit für mich und hat meinen Horizont sehr erweitert. But of course also many thanks to my promotors Doaitse Swierstra and Vincent van Oostrom: naast heel stimulerende gesprekken en uw vakkunde kon ik vooral het vertrouwen waarderen dat in mij is gesteld om zelfstandig mijn onderzoeksdoelen te kiezen en na te gaan. I am very grateful to my research group, the Centre for Software Technology of Utrecht University for offering me a workplace beyond the duration of my contract. Furthermore I thank my parents for their substantial support and my sister, Nora Rochel for the title page.

I want to thank the reading committee and my opponents for their comments and for a challenging and pleasant defense.

\end{document}

